\documentclass[12pt,a4paper,oldfontcommands]{memoir}

\usepackage[english]{babel}
\usepackage[T1]{fontenc}

\usepackage{color}

\settypeblocksize{*}{15cm}{1.6182}
\setlrmargins{2.5cm}{*}{*}
\setulmargins{3cm}{*}{*}
\checkandfixthelayout[nearest] 

\setsecnumdepth{subsubsection}
\maxsecnumdepth{subsubsection}
\settocdepth{subsubsection}

\raggedbottomsectiontrue

\makeatletter 
\makepagestyle{bookstyle} 
\makeevenhead {bookstyle}{\small\thepage}{}{\small\leftmark}
\makeoddhead  {bookstyle}{\small\rightmark}{}{\small\thepage}
\makepsmarks  {bookstyle}{
  \createmark{chapter}      {both} {shownumber} {\@chapapp\ }{. \ }
  \createmark{section}      {right}{shownumber} {}           {. \ }
  \createmark{subsection}   {right}{shownumber} {}           {. \ }
  \createmark{subsubsection}{right}{shownumber} {}           {. \ }
  \createplainmark{toc}     {both}{\contentsname}
  \createplainmark{lof}     {both}{\listfigurename}
  \createplainmark{lot}     {both}{\listtablename}
  \createplainmark{bib}     {both}{\bibname}
  \createplainmark{index}   {both}{\indexname}
}
\nouppercaseheads 

\makeoddfoot{plain}{}{\small\thepage}{}
\makeevenfoot{plain}{}{\small\thepage}{}

\makeatother
\usepackage[shortlabels]{enumitem}

\usepackage{amsmath}
\usepackage{amssymb,bm}
\usepackage{amsthm}
\usepackage[latin1]{inputenc}

\usepackage{mathtools}
\mathtoolsset{centercolon}

\usepackage{graphicx}

\newcommand{\pr}  {\mathrm{pr}}

\newcommand{\supp}{\operatorname{supp}}

\newcommand{\rot} {\operatorname{rot}}

\DeclareMathOperator\ran{ran}
\DeclareMathOperator\rank{rank}

\newcommand{\spann}{\mathrm{span}}

\newcommand{\ad}{{\operatorname{ad}}}
\newcommand{\Opl}{{\operatorname{Op^l\!}}}
\newcommand{\Opr}{{\operatorname{Op^r\!}}}
\newcommand{\Opw}{{\operatorname{Op\!^w\!}}}
\newcommand{\N}{{\mathbb{N}}}

\newcommand{\R}{{\mathbb{R}}}

\newcommand{\C}{{\mathbb{C}}}
\renewcommand{\S}{{\mathbb{S}}}
\newcommand{\T}{{\mathbb{T}}}

\newcommand{\nfrac}[2]{\genfrac{}{}{0pt}{}{#1}{#2}}

\renewcommand\i{\mathrm{i}}

\newcommand{\y}{\mbox{\boldmath $y$}}

\renewcommand{\c}{{\mathrm c}}
\renewcommand{\r}{{\mathrm r}}

\newcommand{\sph}{{\mathrm {sph}}}

\newcommand{\bd}{{\mathrm {bd}}}

\newcommand{\e}{{\mathrm e}}
\newcommand{\ess}{{\mathrm {ess}}}

\renewcommand{\d}{{\mathrm d}}
\newcommand{\dol}{{\mathrm {dol}}}

\newcommand{\diag}{{\mathrm {diag}}}
\newcommand{\disc}{{\mathrm {disc}}}

\newcommand{\pupo}{{\mathrm {pp}}}
\renewcommand{\sc}{{\mathrm {sc}}}

\newcommand{\crt}{{\mathrm {crt}}}
\renewcommand{\Re}{\operatorname{Re}}
\renewcommand{\Im}{\operatorname{Im}}

\DeclarePairedDelimiter\inp\langle\rangle

\newcommand\parb[2][]{#1 \big ( #2#1\big )}

\newcommand\parbb[2][]{#1 \Big ( #2#1\Big )}

\newcommand{\pp}{{\mathrm {pp}}}

\renewcommand{\exp}{{\mathrm {exp}}}
\newcommand{\mand}{\text{ and }}
\newcommand{\mfor}{\text{ for }}
\newcommand{\mforall}{\text{ for all }}

\newcommand{\mif}{\text{ if }}
\newcommand{\mon}{\text{ on }}

\newcommand{\cas}{{\textrm {the Cauchy-Schwarz inequality }}}
\newcommand{\caS}{{\textrm {the Cauchy-Schwarz inequality}}}

\newcommand{\1}{\hspace{ 1cm}}
\newcommand{\w}[1]{\langle {#1} \rangle}

\DeclarePairedDelimiter\ket{\lvert}{\rangle}
\DeclarePairedDelimiter\bra{\langle}{\rvert}
\DeclareMathOperator*{\wslim}{w^\star-lim}
\DeclareMathOperator*{\slim}{s-lim}

\DeclareMathOperator*{\swslim}{s- w^\star-lim}

\DeclareFontFamily{U}{mathx}{\hyphenchar\font45}
\DeclareFontShape{U}{mathx}{m}{n}{
      <5> <6> <7> <8> <9> <10>
      <10.95> <12> <14.4> <17.28> <20.74> <24.88>
      mathx10
      }{}
\DeclareSymbolFont{mathx}{U}{mathx}{m}{n}
\DeclareFontSubstitution{U}{mathx}{m}{n}
\DeclareMathAccent{\widecheck}{0}{mathx}{"71}

\DeclarePairedDelimiter\abs\lvert\rvert
\DeclarePairedDelimiter\norm\lVert\rVert
\DeclarePairedDelimiter\set{\{}{\}}
\DeclarePairedDelimiter\comm{[}{]}

\let\pf\proof
\let\ef\endproof

\newcommand\Step[1]{
  \par\bigskip
  \noindent
  \textbf{#1}.\enspace
}

\newcommand\subStep[1]{
  \par\bigskip
  \noindent
  \underline{\textit{#1}}.\enspace
}

\newcommand{\be}{\begin{equation}}
\newcommand{\ee}{\end{equation}}
\newcommand{\bea}{\begin{eqnarray}}
\newcommand{\eea}{\end{eqnarray}}
\newcommand{\ep}{{\epsilon}}
\newcommand{\f}{\frac}

\newcommand{\gO}{{\mathfrak O}}

\def\f{\frac}

\def\la{\lambda}

\newcommand{\cB}{{\check B}}

\newcommand{\brH}{{\breve H}}

\newcommand{\bD}{{\mathbf D}}

\newcommand{\bY}{{\mathbf Y}}
\newcommand{\bX}{{\mathbf X}}

\newcommand{\vA}{{\mathcal A}}
\newcommand{\vB}{{\mathcal B}}
\newcommand{\vC}{{\mathcal C}}
\newcommand{\vD}{{\mathcal D}}
\newcommand{\vE}{{\mathcal E}}
\newcommand{\vF}{{\mathcal F}}
\newcommand{\vG}{{\mathcal G}}

\newcommand{\vH}{{\mathcal H}}
\newcommand{\vK}{{\mathcal K}}
\newcommand{\vL}{{\mathcal L}}
\newcommand{\vM}{{\mathcal M}}

\newcommand{\vO}{{\mathcal O}}

\newcommand{\vR}{{\mathcal R}}
\newcommand{\vT}{{\mathcal T}}
\newcommand{\vU}{{\mathcal U}}
\newcommand{\vV}{{\mathcal V}}

\newcommand{\vZ}{{\mathcal Z}}
\newcommand{\vS}{{\mathcal S}}
\newcommand{\vW}{{\mathcal W}}

\newtheorem{hypothesis}{Hypothesis}
     \theoremstyle{plain}
     \newtheorem{thm}{Theorem}[chapter]
     \newtheorem{prop}[thm]{Proposition}
     \newtheorem{lemma}[thm]{Lemma}
      \newtheorem{cor}[thm]{Corollary}
     \theoremstyle{definition}
     
     \newtheorem{defn}[thm]{Definition}
     \newtheorem{example}[thm]{Example}
     
     \newtheorem{cond}[thm]{Condition}

     \newtheorem{remark}[thm]{Remark}
     \newtheorem{remarks}[thm]{Remarks}
\newtheorem*{remarks*}{Remarks}
\newtheorem*{remark*}{Remark}

\begin{document}

\begin{titlingpage}
  \title{Spectral analysis of $N$-body Schr\"odinger operators at
    two-cluster thresholds }

\author{Erik Skibsted,\\ Matematisk Institut,
Aarhus Universitet\\ Ny Munkegade  8000 Aarhus C,
Denmark\\[1em] Xue Ping Wang,\\Laboratoire de Math\'ematiques Jean Leray\\
 UMR CNRS 6629, 
 Nantes Universit\'e\\
44322 Nantes Cedex, France}

  \maketitle

\date{\today}

\end{titlingpage}

\frontmatter

\tableofcontents*

\mainmatter

\chapter{Introduction, examples and  notation} \label{Notation}

\section{Scope and results}\label{Scope, goals and results}
The spectral and scattering theory for the quantum mechanical   one-body problem at zero
energy is a well studied subject. The classical theory \cite{JK, JN, JN2}
involves a   real potential $V(x)$ on $\R^3$  fastly decaying at least  like
$ \vO( |x|^{-\rho})$
for some  $ \rho >2$. The slowly  decaying case for which the decay
rate  $ \rho\in (0,2)$ requires additional conditions, roughly sign
conditions \cite{Ya1, Ya2, Na, FS, DS1, DS2}.
 The critical case is defined by $ V(x)\approx C|x|^{-2}$, possibly
 with angular dependence, and the results depend on the coupling
 constant  \cite{ Wa5,SW}. The  obtained results for the above  models
 are  highly  model and case sensitive and include  possible existence of zero
 energy bound and/or
 resonance states  as well as  zero
 energy asymptotics of   the resolvent and scattering matrix quantities
 (like scattering phase shifts). Thresholds of an $N$-body
 Schr\"odinger operator are eigenvalues of the sub-Hamiltonians. There
 exists much less literature on threshold spectral analysis for  the $N$-body problem. We only mention \cite{Wa2} on the resolvent expansion in a special case of the lowest threshold which is the bottom of the essential spectrum.

The goal of the present work is to present a systematic study  of
spectral and scattering theory for the quantum mechanical   $N$-body
problem at any negative two-cluster threshold $\lambda_0$, {\it i.e.}, $\lambda_0$ is an eigenvalue of (possibly several) sub-Hamiltonians associated with two-cluster decomposition,
 but not of those with three or more clusters. These
restrictions  on the nature
of the considered threshold exclude the  presence  of the Efimov
effect there. So for example for  the (dynamical nuclei physics) $3$-body
problem, the threshold  zero  is excluded from
our analysis, while all other
thresholds for this model are negative of two-cluster type. (We shall later in this chapter
give precise definitions.) Philosophically,  the two-cluster threshold
problem is amenable to simplification in terms of an effective one-body problem by the Feshbach-Grushin dimension reduction
method. This is indeed realized in \cite{Wa2} for  fastly  decaying pair
potentials   for the case of the lowest threshold $\lambda_0=\Sigma_2$ (assumed
non-multiple).
However  in the present work we extend the
framework considerably, so that it covers the usual atom
physics models (see Subsections \ref{First principal example} and
\ref{Second  principal example}) for which the slowly decaying
nature  of the Coulomb pair potentials requires refined analysis. Also
we include the cases where the two-cluster threshold $\lambda_0>\Sigma_2$
as well as multiple two-cluster and degenerate eigenvalue cases, which
also call for refined analysis, in particular micro-local analysis.

One main ingredient which enables us to attain the goal of spectral
analysis at any two-cluster threshold is  the Mourre's estimate for
the Hamiltonian with one threshold removed. For a given two-cluster
threshold $\la_0$, the restriction of the total Hamiltonian onto the
orthogonal complement of the associated spectral subspace is a
non-local $N$-body  Hamiltonian for which $\lambda_0$ is no longer a
threshold. We essentially prove the Mourre's estimate at $\lambda_0$ for this reduced Hamiltonian  and deduce the limiting  absorption principles and micro-local resolvent estimates. The limiting absorption principles are used to construct an appropriate Grushin problem such that we can reduce  the two-cluster problem to an effective one-body problem near an arbitrary two-cluster threshold.

To be more concrete let us now consider the dynamical nuclei physics
model from Subsection \ref{First principal example} (see
\eqref{0.1Cou}) and  assume that the particle dimension $n=3$. Let  a
two-cluster decompostion $a=(C_1,C_2)$ of $N$ particles
be given. We then
  write the full Hamiltonian as
\begin{equation*}
 H=H^1\otimes 1\otimes 1+1\otimes H^2\otimes 1+1\otimes 1 \otimes p_a^2 +I_a
\end{equation*}  where  $H^k$, $k=1,2$,  are
cluster-Hamiltonians (defined like for $H$ in their center of mass
frames), $p_a^2$ is the inter-cluster kinetic energy Hamiltonian and
$I_a$ is the \emph{inter-cluster potential}. Suppose  $\lambda_a=\lambda_0$ is  an
eigenvalue of the sub-Hamiltonian $H^a=H^1\otimes
1+1\otimes H^2$ ($\lambda_0$ being of two-cluster type, see \eqref{ass0}). Picking  a
corresponding orthonormal basis  $\varphi^a_1,\dots\varphi^a_{m}\in
\ker(H^a-\lambda_0)\subset L^2(\bX^a)$,
$m=m_a$,  the \emph{effective
    inter-cluster potential} is the $m\times m$ matrix-valued function
  in the relative position
  variable of the two clusters, viz.  $R=R_1-R_2$,
\begin{align}
  \label{eq:effe_pot1ee00}
  V(R)_{kl}:=\inp{\varphi^a_k,I_a\varphi^a_l}_{L^2(\bX^a)}=Q_1Q_2\delta_{kl}|R|^{-1}+Q_{kl}(\widehat
  R)|R|^{-2}+\vO\parb{|R|^{-3}}.
\end{align} Here $Q_1$ and
$Q_2$ are  the
total charge of the particles in the clusters $C_1$ and $C_2$,
respectively, and $\delta_{kl}$ is the Kronecker symbol.    In
addition we denote by $Q_a$ the matrix-valued homogeneous
potential $\parb{Q_{kl}}$ and $\widehat R
=R/\abs{R}$. Let $P^a$ denote the orthogonal (rank $m$) projection onto
   $\ker(H^a-\lambda_0)$ in
   $L^2(\bX^a)$. Then obviously   $\Pi^a=P^a\otimes 1$
   projects  onto
   the span of functions of the form $\varphi^a\otimes f_a$,
   $\varphi^a\in\ker(H^a-\lambda_0)$, in $L^2(\bX)$.

In terms of \eqref{eq:effe_pot1ee00} a  relevant classification reads:
\begin{description}
                                   \item[Case 1] (slowly  decaying case)\quad
                                     $Q_1Q_2\neq 0$.
                                   \item[Case 2] (critically  decaying case)\quad
                                     $Q_1Q_2= 0$\, and the
                                    function\,
                                     $Q_a\neq 0$.
                                   \item[Case 3] (fastly  decaying case)\quad
                                     $Q_1Q_2= 0$\, and the
                                    function\,
                                     $Q_a= 0$.

 \end{description}

 In general $\lambda_0$ might be a
 multiple two-cluster threshold, which might suggest  that we group
 the set  of thresholds $a$ for which $\lambda_0\in
\sigma_{\pp}(H^a)$, say denoted by $\widetilde A$,  into $\widetilde A=\vA_1\cup\vA_2\cup\vA_3$
specified  as follows.
\begin{description}
\item [$\vA_1$:] the effective
inter-cluster interaction   is to leading order attractive Coulombic,
i.e. $Q_1Q_2 < 0$.
\item [$\vA_2$:] the effective
inter-cluster interaction  is to leading order repulsive  Coulombic,
i.e. $Q_1Q_2 > 0$.
\item [$\vA_3$:] the effective
inter-cluster interaction  is $\vO(|x_a|^{-2})$, i.e. $Q_1Q_2 =0$.
\end{description} Clearly the elements of   $\vA_1\cup\vA_2$ are  classified as
Case 1, while the elements of   $\vA_3$ are  classified either as
Case 2 or Case 3. This motivates the  splitting   $\vA_3= \vA ^{\mathrm
{cd}}_3\cup \vA ^{\mathrm
{fd}}_3$ by specifying
\begin{equation*}
  \vA ^{\mathrm
{cd}}_3=\set{a\in\vA _3|\, Q_a\neq 0}\mand \vA ^{\mathrm
{fd}}_3=\set{a\in\vA _3|\, Q_a=0}
\end{equation*}
 corresponding to Case 2
and  Case 3, respectively.

For any  $a\in \vA_3$ there are  computable numbers
$s_a\geq 1$ and $d_a\in \N_0$ determined by spectral properties of the
vector-valued Schr\"odinger operator on the unit-sphere $\S^2$ with
the matrix-valued potential $Q_a$
(see Section \ref{sec:CoulRellich}) such
that in terms of standard
weighted $L^2$-space and weighted Sobolov-space notation (see
Subsection \ref{Spaces} for  definitions), referring
here to our
most general result  (see Section \ref{sec:CoulRellich}):
\begin{thm}\label{thm:physical-modelsRell00}  For any  two-cluster
  threshold $\lambda_0$:
\begin{enumerate}[1)]
  \item \label{item:PRel101} The space  of locally $H^1$ solutions  to
    $(H-\lambda_0)u=0$ in
    \begin{subequations}
    \begin{equation}\label{eq:spac1}
      \sum_{a\in \vA_1} \Pi^a L^2_{-3/4}+ \sum_{a\in \vA_2} \Pi^a
      L^2_{(-3/2)^+}+\sum_{a\in \vA_3} \Pi^a
      L^2_{(-\min\set{3/2,\,s_a})^+}+L^2_{-1/2},
    \end{equation} say  denoted by $\mathcal{E}$,
 has finite dimension.
\item \label{item:PRel202} If $\vA _3=\emptyset$, then  $\vE\subset H^1_{\infty}$.
  \item \label{item:PRel203} The dimension of the space of
    \emph{resonance states}
    \begin{equation*}
      n_{\mathrm{res}}=\dim\parb{\vE /\ker
       (H-\lambda_0)_{|{H^{1}}}}\leq \sum_{a\in \vA_3} d_a.
   \end{equation*}
 \item \label{item:PRel203aa}   The numbers  $s_a=3/2$ and $d_a=m_a$ for any $a\in \vA ^{\mathrm
{fd}}_3$. In particular if $\vA ^{\mathrm
{cd}}_3=\emptyset$, then \eqref{eq:spac1} simplifies as
\begin{equation}\label{eq:spac2}
      \sum_{a\in \vA_1} \Pi^a L^2_{-3/4}+ \sum_{a\in \vA_2\cup \vA ^{\mathrm
{fd}}_3 }\Pi^a
      L^2_{(-3/2)^+}+L^2_{-1/2},
    \end{equation}
    \end{subequations}
and
\ref{item:PRel203} reads
\begin{equation*}
  n_{\mathrm{res}}\leq \sum_{a\in \vA ^{\mathrm
{fd}}_3} \,m_a.
\end{equation*}
\end{enumerate}
\end{thm}

One may view Theorem \ref{thm:physical-modelsRell00} as a
 version of
the well-known Rellich theorem for non-threshold energies
\cite[Theorem 1.4]{AIIS} (see also Theorem \ref{thm:priori-decay-b_0})
although the above analogue at a two-cluster threshold is
considerably more complex. For example the
analogue of \eqref{eq:spac1} and \eqref{eq:spac2} in the continuous
spectrum away from thresholds
reads $L^2_{-1/2}$,
and we note that for almost  all (probably valid  for all) such energies the space of
generalized eigenfunctions in $L^2_{-1/2-\epsilon}$ (for any
$\epsilon>0$) has
infinite dimension \cite{Sk6}. In comparison  in the context of Theorem
\ref{thm:physical-modelsRell00} with $\widetilde{\vA}=\vA_1=\set{a}$ we
show  that the  space of
generalized eigenfunctions at $\lambda_0$ in $\Pi^a
L^2_{-3/4-\epsilon}+L^2_{-1/2-\epsilon}$ is infinitely dimensional  (see Theorem \ref{thm:repEigen} \ref{it:uniq2}).

The proof of Theorem \ref{thm:physical-modelsRell00} is complicated, in fact we give a full
proof only under the two technical conditions  \eqref{eq:dirCo1} and
\eqref{eq:dirCo2} (not to be elaborated on in this
introduction), treating  the general case in a somewhat  sketchy fashion.

One of the threshold phenomena indicated by Theorem
\ref{thm:physical-modelsRell00} is the possible existence of resonance
states combined freely with the possible existence of $L^2$ eigenfunctions
at the two-cluster threshold $\lambda_0$. This is completely  analogous
to the situation for the one-body problem
for  fastly  decaying  potentials  \cite{JK} (exhibiting  a somewhat similar
sophisticated Rellich theorem at zero energy), giving  rise to the
classification into a  Regular Case (where $\lambda_0$ is neither an
eigenvalue nor  a
resonance  of $H$) and Exceptional Cases 1,2 and 3 (see
Subsection \ref{sec:case-lambda_0=sub}).  The resolvent asymptotics
at zero energy for the  one-body  problem is determined by this
classification. It is a separate issue for us to obtain similar resolvent
asymptotics  at $\lambda_0$ in  the present  framework.   However
our    analysis is not complete, mainly due to  lack of
strong
decay of Coulombic potentials  hampering the   analysis. Of course the  Regular Case is the
easiest  case and we shall actually treat this with
$\widetilde{\vA}=\vA_1\cup \vA_2\cup \vA^{\rm{fd}}_3$  (see Theorem \ref{thm:resolv-asympt-physReg0}). For the Exceptional Cases 1 and
3 (defined by the presence of a resonance) we show the following
result.
\begin{thm}[Exceptional
point of $1$st or $3$rd kind]
    \label{thm:resolv-asympt-phys1st3rd00} Let  $\lambda_0$ be any
    two-cluster threshold for which $n_{\mathrm{res}}\geq 1$,
    i.e. $\lambda_0$ is a
resonance of $H$. Suppose the technical conditions
   \eqref{eq:dirCo1} and \eqref{eq:dirCo2} (referred to above),
\begin{equation}
    \label{eq:dec2000}
    \ran \Pi_H\subset L^2_t\text{  for some }t> 3/2,
  \end{equation} where $\Pi_H$ is the orthogonal projection onto $\ker
  (H-\lambda_0)$ (i.e. the eigenprojection if $\lambda_0$ is
  an eigenvalue  of $H$ and zero otherwise),  and suppose
\begin{equation}
    \label{eq:2plus300}
    \widetilde{\vA}=\vA_2\cup \vA^{\rm{fd}}_3.
  \end{equation}  Then the following asymptotics hold for $R(\lambda_0
+ z)=(H-\lambda_0
+ z)^{-1}$ as an operator from $H^{-1}_{s}$ to $H^{1}_{-s}$, $s>1$,  for
$z\to 0$ in $\vZ_{\pm}=\set{\Re z \geq 0,\, \pm\Im z>0}$
and for some $\epsilon=\epsilon(s)>0$.
\begin{equation} \label{asymRz1aa9last00}
R(\lambda_0 + z)=-z^{-1} \Pi_H + \frac{\i}{\sqrt{z}} \sum_{j=1}^{n_{\mathrm{res}}} \w{u_j, \cdot} u_j + \vO(|z|^{-\f 1 2 +\ep}).
\end{equation} Here
$\set{u_1, \dots, u_{n_{\mathrm{res}}}}\subset H^1_{(-1/2)^-}$ is a basis of   resonance states  of
$H$ being independent
 of the choice of the  sign of  $\vZ_{\pm}$.
  \end{thm}
 Among the appearing conditions \eqref{eq:dec2000}  is the
 `unpleasant one'. It is an implicit (possibly redundant) condition appearing as  an artifact of
 our methods. If $\lambda_0$ is exceptional
point of $1$st kind, \eqref{eq:dec2000}  is obviously fulfilled since
then $\Pi_H=0$. Our  Theorem \ref{thm:resolv-asympt-physReg0} as well
the above Theorem \ref{thm:resolv-asympt-phys1st3rd00} require
explicitly $\vA^{\textrm{cd}}_3=\emptyset$.

Under a spectral condition
for certain   elements of  $\vA^{\textrm{cd}}_3$ (those for  which
\eqref{eq:hard0} in Subsection \ref{Appplications to
  scattering theory} is violated) oscillatory
behaviour of the resolvent near the  two-cluster  threshold is
expected. This  is  expected thanks to arguments of  \cite{SW}.

\subsection{Applications to   threshold scattering}\label{Appplications to
  scattering theory}
Here we briefly outline our main applications in  Chapter
\ref{Applications}. One of our results concerns the following  generalization
of  a  result from \cite{DS1,DS2} (see also \cite{Fr}).
\begin{thm}
  \label{thm:non-elast-scatt000} Suppose that   $\lambda_0=\Sigma_2$
  is a
    two-cluster threshold, the technical conditions
   \eqref{eq:dirCo1} and \eqref{eq:dirCo2},
\begin{equation}
    \label{eq:2plus3000}
    \widetilde{\vA}=\vA_1,
  \end{equation} and suppose the Regular Case  (i.e. that $\lambda_0$
  is neither  an eigenvalue  nor  a resonance  of $H$). Let $\vC$
  denote the set of  scattering channels
  $\alpha=(a,\lambda_0, \varphi_\alpha)$ with $a\in \widetilde\vA$ (note that
  $\lambda_0$ is a simple eigenvalue of $H^a$). For $\alpha,\beta\in\vC$ the element of the
  scattering matrix $S_{\beta\alpha}(\lambda)$ (modelled after
  \cite{DS1,DS2}) is  well-defined for $\lambda$ slightly above
  $\lambda_0$ and possessing a strong  limit  as $\lambda\to (\lambda_0)_+$. Moreover the
  singular support of the limiting element $S_{\beta\alpha}(\lambda_0)$ fulfills
\begin{align*}
{ \mathrm
{sing \,supp\,}}S_{\beta\alpha}(\lambda_0)\quad
  \begin{cases}
   &\subset\{(\omega,\omega')\mid \omega\cdot \omega'= -1\}\quad \text{for}\quad \alpha=\beta,\\
&=\emptyset\quad \text{for}\quad \alpha\neq\beta.
  \end{cases}
  \end{align*}
\end{thm} The proof of Theorem \ref{thm:non-elast-scatt000}  may be
considered as an extension of the one used in \cite{DS1,DS2} to obtain
a similar `semi-classical' result on the
scattering matrix $S_{\rm
  cou}(E)$ for the one-body problem with an attractive Coulomb
potential. See also \cite{Va1,Va2} for $N$-body scattering matrices in the short-range case.

Another result under \eqref{eq:2plus3000} concerns the
difference
\begin{equation*}
  S_{\alpha\alpha}(\lambda)- S_{\rm
  cou}(\lambda-\lambda_0);\quad \lambda\in [\lambda_0,\lambda_0+\delta].
\end{equation*} Under conditions, in particular including  the
non-multiple property  $\#\widetilde\vA=1$ (primarily
used to simplify
the  presentation)  however covering the case where
  the
    two-cluster threshold $\lambda_0>\Sigma_2$, we show that this
    difference is a
`partial  smoothing operator'  (see Theorem \ref{thm:ScatN} and Remark
\ref{remarks:inelastic-scattering-n} \ref{item:multCase3}). Yet another
result is a  characterization  of the  limiting element
$S_{\alpha\alpha}(\lambda_0)$ given by asymptotics in terms of
appropriate `channel quasi-modes' (see Theorem \ref{thm:repEigen} \ref{it:uniq2}).

This leads  to another subject of interest, more precisely
\emph{non-transmission} at  $\lambda_0$. This  is a `geometric
concept' amounting to the feature
\begin{equation}\label{eq:asymU00}
  \norm{1-S_{\alpha\alpha}(\lambda)^*S_{\alpha\alpha}(\lambda)}\to 0 \text{ for
  }\lambda\to (\lambda_0)_+.
\end{equation} We derive a formula   under \eqref{eq:2plus3000}
(see Corollary \ref{cor:norm_less22})  indicating that  transmission does
occur in this case if $\lambda_0>\Sigma_2$ (see also Remarks
\ref{remark:elast-scattA} and \ref{remarks:non-transmission-bem} \ref{item:tran4}).

In contrast to  the attractive slowly  decaying case  we do prove non-transmission
in the following three  cases (assuming as above in all cases the non-multiple
property  $\#\widetilde\vA=1$):
\begin{enumerate}[\bf I)]
\item
 Effective  repulsive Coulombic
 case, i.e.  $\widetilde{\vA}=\vA_2$.
\item $I_a(x^a=0)=0$, `above the
 Hardy limit' and $\lambda_0$  be regular.
\item $I_a(x^a=0)=0$ and  `fastly decaying case',
  i.e. $\widetilde{\vA}=\vA ^{\mathrm
{fd}}_3$,  and $\lambda_0$  be
  `maximally exceptional of $1$st kind'.
\end{enumerate} A  special case of II) is that $\widetilde{\vA}=\vA ^{\mathrm
{fd}}_3$ and   $\lambda_0$ neither be an eigenvalue nor a
resonance. The notions in  II) and III) are in general given as follows
(see also Section \ref{sec:Transmission problem at threshold}).  The
phrase  `above the
 Hardy limit' refers to a spectral property  of the
vector-valued Schr\"odinger operator on the unit-sphere $\S^2$ with
the matrix-valued potential $Q_a$  (writing $\widetilde\vA=\set{a}$), more
precisely
\begin{equation}\label{eq:hard0}
  \inf \sigma\parb{-\Delta_\theta+Q_a(\theta)}>-1/4.
\end{equation}
 For III) the potential $Q_a=0$ and `maximally
exceptional of $1$st kind' refers to Exceptional Case 1 and the condition
$n_{\mathrm{res}}=m_a=\dim \ker(H^a-\lambda_0)$. We note for
comparison  that if
$m_a>1$ and
$n_{\mathrm{res}}=\set{1,\dots, m_a-1}$ then indeed transmission can
occur for $\lambda_0$  be
  exceptional of $1$st kind  (see
Subsection \ref{subsec: An example of transmission}).

The last  subject of interest concerns    total cross-sections for atom-ion
scattering. It is an observed phenomenon at the very beginning of the quantum mechanics that when there is no dipole moment for the atom, the total cross-sections are finite.  A mathematical proof for this  physics folklore  is given in \cite{JKW}.
The operator under consideration is a special case of the  dynamical nuclei physics
model from Subsection \ref{First principal example} with the particle
dimension $n=3$. Assume  $\lambda_0=\Sigma_2$
  is a two-cluster threshold, the technical conditions
   \eqref{eq:dirCo1} and \eqref{eq:dirCo2}, \eqref{eq:dec2000} and  $ \widetilde{\vA}=\vA ^{\mathrm
{fd}}_3$. It is known  from  \cite{JKW} that for any channel  $\alpha =(a,\lambda_0,\varphi_\alpha )$,
$a\in \widetilde{\vA}$,  and any incident direction $\omega\in \S^2$,
\begin{equation*}
   \text{the  total
cross-section }
\sigma_\alpha(\lambda, \omega)
\end{equation*}
is finite for  non-threshold $\lambda$'s
above $\lambda_0$. In the present work we
derive bounds and asymptotics of this quantity as $\lambda\to
(\lambda_0)_+$  (see Section \ref{total cross-sections}). The result depends on whether $\lambda_0$ is regular
or exceptional of $2$nd kind (yielding bounded asymptotics, see Theorem \ref{cs-caseReg+3}) or if  $\lambda_0$ is
of $1$st or $3$rd kind (yielding $(\lambda-\lambda_0)^{-1}$ type
unbounded asymptotics, see Theorem \ref{cs-caseReg+40}). Our proof
relies on the derivation  of Theorem \ref{thm:resolv-asympt-phys1st3rd00}.

\section{Many-body Schr\"odinger operators}\label{$N$-body
  Schr\"odinger operators}

Let $H$ denote the many-body Schr\"odinger operator obtained
by  the removal of the center of mass from the total
Hamiltonian
\begin{equation}
\label{0.1}
\widetilde{H} =-\sum_{j = 1}^N \frac{1}{2m_j}\Delta_{x_j} + \sum_{1 \le i<j \le N} V_{ij}(x_i
-x_j), \quad x_j\in\R^n,
\end{equation}
where $x_j$ and $m_j$ denote the position and mass of the $j$'th
particle. The pair potentials
$V_{ij}$ are  assumed to be real and relatively compact with respect to
$-\Delta$ in $L^2(\R^n)$,  and they satisfy for some $\rho>0$ the condition
\begin{equation*}
|V_{ij}(y) |  \le C_{ij} |y|^{-\rho}\text { for }y \in \R^n\text {
  with }|y| > R,
\end{equation*}
for some $R>0$.
 However we shall need  some extra regularity. It is convenient to use the following condition.
\begin{cond}\label{cond:smoothg}
  There exists $\rho >0$ such that for all pair potentials $V_{ij}$
  there is a splitting $V_{ij}=V_{ij}^{(1)}+V_{ij}^{(2)}$, where
  \begin{enumerate}[label=(\arabic*)]
  \item \label{item:cond1g} $V_{ij}^{(1)}$ is smooth and
    \begin{equation}
      \label{eq:1g}
      \partial ^\alpha_yV_{ij}^{(1)}(y)=\vO\big(|y|^{-\rho-|\alpha|}\big ).
    \end{equation}
  \item \label{item:cond12g} $V_{ij}^{(2)}$ is compactly supported and
    \begin{equation}
      \label{eq:2g}
      V_{ij}^{(2)}(-\Delta+1)^{-1}\text{ is compact on
      }L^2(\R^n_y).
    \end{equation}
  \end{enumerate}
\end{cond}

 The Hamiltonian $H$ is regarded as a self-adjoint operator
on $L^2(\bX)$, where $\bX$ is the $n(N-1)$ dimensional real vector
space  $ \bX:= \set[\big] { \sum_{j=1}^{N} m_j x_j = 0}$.
Let $\vA$  denote the set of all cluster
decompositions of the $N$-particle system. The notation $a_{\max}$ and
$a_{\min}$ refers to the $1$-cluster and $N$-cluster decompositions,
respectively.
  Let for $a\in\vA$ the notation  $\# a$ denote the number of
clusters in $a$.
For $i,j \in\{1, \dots, N\}$, $i< j$, we denote by $(ij) $ the
$(N-1)$-cluster decomposition given by letting $C=\{i,j\}$ form a
cluster and all other particles $l\notin C$ form $1$-particle clusters. We write $(ij) \subset a$ if $i$ and $j$ belong to the same cluster
in $a$.   More general, we write $b\subset a$ if each cluster of $b$
is a subset of a cluster of $a$. If $a$ is a $k$-cluster decomposition, $a= (C_1, \dots, C_k)$,
we let
\[
\bX^a = \set[\big] { x\in\bX\mid  \sum_{l\in C_j } m_l x_l = 0,  j = 1, \dots,
k}=\bX^{C_1}\oplus\cdots \oplus\bX^{C_k},
\]
and
\[
\bX_a  = \set[\big] { x\in\bX\mid  x_i = x_j \mbox{ if } i,j \in C_m  \mbox{ for some }
m \in \{ 1, \dots, k\}  }.
\]
 Note that $a\subset b\Leftrightarrow \bX^a\subset\bX^b$. Moreover $\bX^a$ and $\bX_a$  give an orthogonal decomposition for $\bX$
equipped  with
the quadratic form
\[
q(x) = \sum_j 2m_j|x_j|^2,  \1 x\in {\bX}.
\]
 For $x\in \bX$, we have the corresponding orthogonal decomposition:
 $x =x^{a} + x_{a}$ with $x^a =\pi^a x\in\bX^a$ and $x_a =\pi_a x\in \bX_a$.

With this notation, the many-body Schr\"odinger operator $H$ introduced above
can be written in the form
\[
 H = H_0 + V
\]
where  $H_0=p^2$ is (minus)  the Laplace-Beltrami operator on   the
Euclidean space  $(\bX, q)$ and
$V=V(x) =  \sum_{a=(ij)\in\vA} V_{a}(x^{a}) $ with $ V_a (x^a) = V_{ij} (x_i - x_j)$ for the  $(N-1)$-cluster decomposition $a=(ij)$. More precisely, for example,
\begin{align*}
  x^{(12)}=\parb{\tfrac{m_2}{m_1+m_2}(x_1-x_2),-\tfrac{m_1}{m_1+m_2}(x_1-x_2),0,\dots,0}.
\end{align*}

We note the following geometric properties for $N\geq 3$: For all
$a,b\in \vA$ with
  $\#a=2$,  $\#b=N-1$  and $b\not\subset a$
\begin{subequations}
  \begin{align}
    \label{eq:22A}
    &\ran\parb{\pi^b\pi^{a}}= \ran\,\pi^b,\\
    \label{eq:89A}
    &\pi^b:\bX_{a}\to \bX^b\text{  is
bijective}.
  \end{align}
   \end{subequations}
\subsection{Principal example, dynamical nuclei}\label{First principal example}
Consider a system of  $N$ particles interacting by Coulomb forces. The
Hamiltonian  then reads
\begin{equation}
\label{0.1Cou}
H=-\sum_{j = 1}^N \frac{1}{2m_j}\Delta_{x_j} + \sum_{1 \le i<j \le N} q_iq_j|x_i
-x_j|^{-1}, \quad x_j\in\R^n,\,n\geq 3,
\end{equation}
where $x_j$,  $m_j$ and $q_j$ denote the position, mass and charge of
the $j$'th particle, respectively. $H$ is regarded as a self-adjoint operator in $L^2(\bX)$ (with mass center removed).

Let us consider a two-cluster decomposition $a=(C_1,C_2)$. For
convenience assume $C_1=\{1,\dots, J\}$ and $C_2=\{J+1,\dots, N\}$. We
can write
\begin{equation*}
 H=H^1\otimes 1\otimes 1+1\otimes H^2\otimes 1+1\otimes 1 \otimes p_a^2 +I_a
\end{equation*}  where  $H^k$, $k=1,2$,  are
cluster-Hamiltonians (defined similarly in their center of mass
frames) and
\begin{equation*}
  I_a=\sum_{i\in C_1,\,j\in C_2} q_iq_j|x_i
-x_j|^{-1}.
\end{equation*}
 To  expand $I_a$  we let for $k=1,2$
 \begin{align*}
   Q_k&=\sum_{j\in C_k} q_j,\,M_k=\sum_{j\in C_k} m_j,\\R_k&=R_k(x)=\sum_{j\in
     C_k} \tfrac{m_j}{M_k}x_j,\,\widetilde Q_k=\widetilde Q_k(x^{C_k})=\sum_{j\in C_k} q_j(x_j-R_k),\\M&=M_1+M_2,\,R=R_1-R_2,
 \end{align*} and we decompose for all $x\in \bX$
 \begin{align*}
   x&=x^{C_1}+x^{C_2}+x_a,\\
x^{C_1}&=(x_1-R_1,\dots,x_J-R_1, 0,\dots,0)\in \bX^{C_1},\\
x^{C_2}&=(0,\dots,0,x_{J+1}-R_2,\dots,x_N-R_2)\in \bX^{C_2},\\
x_a&=\parb{\tfrac{M_2}{M}R,\dots,\tfrac{M_2}{M}R,-\tfrac{M_1}{M}R,\dots,-\tfrac{M_1}{M}R}\in \bX_a.
\end{align*}  Note  that indeed the center of charge $\widetilde Q_k$
is a function of $x^{C_k}$.

Consequently we can expand for $i\in C_1$ and $j\in C_2$
\begin{align*}
  |x_i-x_j|^{-1}= |R|^{-1}-\tfrac{R}{|R|^{3}}\cdot \parb{(x_i-R_1)-(x_j-R_2)}+\vO\parb{|R|^{-3}}|x^a|^2.
\end{align*} This is in the regime $|R|\to \infty$ and
$|x_i-R_1|+|x_j-R_2|\leq \tfrac 12 |R|$.

Whence in turn we obtain  for $|R|\to \infty$
\begin{align}\label{eq:expans1}
  I_a=Q_1Q_2|R|^{-1}+\tfrac{R}{|R|^{3}}\cdot \parb{Q_1\widetilde Q_2(x^{C_2})-Q_2\widetilde Q_1(x^{C_1})}+\vO\parb{|R|^{-3}}|x^a|^2,
\end{align} which leads to various cases. We use the notation
$\varphi^k$, $k=1,2$, to denote a  cluster bound state (for
the cluster Hamiltonian $H^k$) and $\inp{\cdot,\cdot}_k$ to denote  the
corresponding cluster inner product. The \emph{effective
potential}
\begin{equation*}
  V(R):=\inp{\varphi^1\otimes\varphi^2,I_a\varphi^1\otimes\varphi^2}_{L^2(\bX^a)}.
\end{equation*}
\begin{description}
                                   \item[Case 1] $V\approx |R|^{-1}$:\quad  \quad  \quad $Q_1Q_2\neq 0$.
                                   \item[Case 2]  $V\approx |R|^{-2}$:
                                   \item[Subcase 2a] \quad
                                     $Q_1\neq 0$, $Q_2=0$ and
                                     $\inp{\varphi^2,\widetilde
                                       Q_2\varphi^2}_2 \neq 0$.
                                   \item[Subcase 2b] \quad
                                     $Q_2\neq 0$, $Q_1=0$ and
                                     $\inp{\varphi^1,\widetilde
                                       Q_1\varphi^1}_1 \neq 0$.
                                   \item[Case 3]
                                     $V= \vO\parb{|R|^{-3}}$:\quad
                                     \quad \quad
                                     \item[Subcase 3a] \quad $Q_1=Q_2=0$.
                                     \item[Subcase 3b] \quad
                                       $Q_1\neq 0$, $Q_2=0$ and
                                       $\inp{\varphi^2,\widetilde
                                         Q_2\varphi^2}_2= 0$.
                                     \item[Subcase 3c] \quad
                                       $Q_2\neq 0$, $Q_1=0$ and
                                       $\inp{\varphi^1,\widetilde
                                         Q_1\varphi^1}_1 = 0$.
                                   \end{description}

                                   Note that  for Subcases 2a, 3a and 3b,
 assuming sufficient decay of the  cluster bound states,   the {effective
potential}
\begin{equation}
  \label{eq:effe_pot1}
  V(R)=\inp{\varphi^1\otimes\varphi^2,I_a\varphi^1\otimes\varphi^2}_{L^2(\bX^a)}=Q_1\tfrac{R}{|R|^{3}}\cdot \inp{\varphi^2,\widetilde
  Q_2\varphi^2}_2+\vO\parb{|R|^{-3}}.
\end{equation} Whence indeed $V\approx |R|^{-2}$ at infinity in Subcase
2a,  while indeed  $V= \vO\parb{|R|^{-3}}$ for Subcases  3a and
3b. We can argue similarly for Subcases 2b and  3c. Note also that $|R|^{-2}$  is the  \emph{critical decay rate}
for threshold analysis, cf. \cite{SW}. Case 1 is the slowly decaying
case. In  Case 1 the potential $V\approx |R|^{-1}$, and $V$  is said to be
\emph{slowly decaying}. For $Q_1Q_2<0$ and $Q_1Q_2>0$ the one-body
results of \cite{FS} and \cite{Na, Ya2}  will be useful, respectively. In  Case 3  the effective
potential is said to be
\emph{fastly  decaying} and  other one-body results/techniques   will be
useful, cf. for
example \cite{JK}.
Case  2 (the critical case) is different and rather `rich'.

A detailed analysis of  the structure of a  class
of generalized eigenfunctions at a two-cluster threshold, possibly  a
multiple and/or a
non-simple two-cluster threshold, will be carried out  for
physical models in Section \ref{sec:CoulRellich}.  (See \eqref{ass0} for
the  definition of a `two-cluster threshold'.)

From the derivation
it follows that it could happen that the second  term $\vO\parb{|R|^{-3}}$ of
\eqref{eq:effe_pot1} actual has  homogeneity $-3$ at infinity. For
example this happens for Subcase 3a exactly when the moments $\widetilde R_1:=\inp{\varphi^1,\widetilde
  Q_1\varphi^1}_1\neq 0$ and $\widetilde R_2:=\inp{\varphi^2,\widetilde
  Q_2\varphi^2}_2\neq 0$ due to the computation for this case
\begin{align*}
  \vO\parb{|R|^{-3}}=\abs{R}^{-5}\parb{\abs{R}^2 \widetilde R_1\cdot
  \widetilde R_2-3 (R\cdot \widetilde R_1) (R\cdot \widetilde R_2)} +\vO\parb{|R|^{-4}}
\end{align*}
  If certain    `moments'
  vanish for Subcases 3a and  3b the order of the  second  term  of
\eqref{eq:effe_pot1} is  of  the form  $\vO\parb{|R|^{-4}}$, cf. \cite [Appendix A]{JKW}. In Chapter
\ref{chap:resolv-asympt-near} we shall obtain leading order resolvent expansions for
Case 3 without distinguishing between whether  the  homogeneous
$-3$ term vanishes or not.
In Section  \ref{total cross-sections} we shall study a case, where in
fact the effective potential is (at least) of order
$\vO\parb{|R|^{-4}}$. In the same section  an  explicit
calculation of the Hamiltonian is  given
in terms of so-called clustered atomic coordinates.

Strictly speaking
the distinction between Cases 2 and 3 as defined above makes best
sense for a   simple two-cluster threshold  and we will not use this
classification in the non-simple case. Rather in  the general possibly non-simple case one
needs the following  (slightly) different definition, see Section
\ref{sec:CoulRellich} for further details. Let $\lambda_a$ be a
non-threshold eigenvalue of the sub-Hamiltonian $H^a=H^1\otimes
1+1\otimes H^2$ (more precisely,  we will need $\lambda_a\in
\vT_2$, see  \eqref{ass0}). Picking  an
orthonormal basis  $\varphi^a_1,\dots\varphi^a_{m}\in L^2(\bX^a)$,
$m=m_a$ (being one or possibly bigger),
in the range of the corresponding eigenprojection,  the \emph{effective
    potential} is the $m\times m$-matrix-valued function in the
  variable $R=R_1-R_2$
\begin{align}
  \label{eq:effe_pot1ee}
  V(R)_{kl}:=\inp{\varphi^a_k,I_a\varphi^a_l}_{L^2(\bX^a)}=Q_1Q_2\delta_{kl}|R|^{-1}+Q_{kl}(\widehat
  R)|R|^{-2}+\vO\parb{|R|^{-3}}.
\end{align} Here $\delta_{kl}$ is the Kronecker symbol and $\widehat R
=R/\abs{R}$. In terms of \eqref{eq:effe_pot1ee} the more general
(and correct) classification reads:
\begin{description}
                                   \item[Case 1] \quad  \quad  \quad
                                     $Q_1Q_2\neq 0$.
                                   \item[Case 2] \quad \quad \quad
                                     $Q_1Q_2= 0$\, and the
                                     matrix-valued function\,
                                     $Q_a=\parb{Q_{kl}}\neq 0$.
                                   \item[Case 3] \quad \quad \quad
                                     $Q_1Q_2= 0$\, and the
                                     matrix-valued function\,
                                     $Q_a=\parb{Q_{kl}}= 0$.

 \end{description}

\section{$N$-body Schr\"odinger operators with infinite mass
  nuclei}\label{$N$-body Schr\"odinger operators with infinite mass
  nuclei} In the case of $M\geq 1$ infinite mass
  nuclei located at $R_m\in \R^n$, $m=1,\dots,M$,  the Hamiltonian reads
\begin{equation}
\label{0.1b}
H=-\sum_{j = 1}^N \frac{1}{2m_j}\Delta_{x_j} + \sum_{1 \le i<j \le N} V_{ij}(x_i
-x_j)+ \sum_{1 \le j \le N,\;1 \le m \le M} V^{\rm ncl}_{jm}(x_j
-R_m),
\end{equation} where we impose similar conditions on $ V^{\rm
  ncl}_{jm}$ as  for $ V_{ij}$ in Condition \ref{cond:smoothg}. The
one-body problem $N=1$ is included in
\eqref{0.1b} (the middle term is absent in that case).  The configuration space reads $\bX=\R^{nN}$, and we
 use  the metric $q$ as before. The `electron-electron'
interaction $V_{ij}(x_i
-x_j)$ takes as before the form $V_a(x^a)$ where $x^a=\pi^ax$,
$a=(ij)$,  is the
orthogonal projection of $x$ onto an  $n$-dimensional
subspace. Similarly the `electron-nuclei'
interaction   $\sum_{1 \le m \le M} V^{\rm ncl}_{jm}(x_j
-R_m)$ takes  the form $V_a(x^a)$ where again $x^a=\pi^ax$,
$a=a(j)$,  is the
orthogonal projection of $x$ onto an $n$-dimensional
subspace (let $x^a=(0,\dots,0,x_j,0,\dots,0)$, i.e. all other
coordinates than the $j$'th are put equal to zero). Rather than using the cluster decompositions to label a family
of  `subspaces of internal motion'  $\{\bX^a\}$ similar to those
considered in Section \ref{$N$-body
    Schr\"odinger operators}  we prefer henceforth to appeal to abstract
labeling.  Precisely we consider the smallest finite family
$\{\bX^a\mid a\in\vA\}$ of
subspaces of $\bX$  which is stable under addition and which contains
$\{0\}$ and  the $n$-dimensional
subspaces discussed above. See Section \ref{Generalized $N$-body
  Schr\"odinger operators},  and see \cite[Section 5.1]{DG} for a discussion
of the abstract $N$-body problem. On the other hand  there is a
concrete description of the index set $\vA$ and
this  family  $\{\bX^a\mid a\in\vA\}$ which can   be useful to have in mind: Consider $a=(C_1,\dots, C_p)$ where
the sets
$C_q$ are disjoint subsets of $\{1,\dots,N\}$. For $p\geq 2$ and $q<p$ we have $\#C_q\geq 2$  and we let  $\bX^{C_q}=\{x\in
\bX\mid x_j=0\text{ if }j\notin C_q\mand \sum_{i\in
  C_q}m_ix_i=0\}$. Either similarly  $\bX^{C_p}=\{x\in
\bX\mid x_j=0\text{ if }j\notin C_p\mand \sum_{i\in
  C_p}m_ix_i=0\}$ (in that case we have $\#C_p\geq 2$) or $\bX^{C_p}=\{x\in
\bX\mid x_j=0\text{ if }j\notin C_p\}$.
In  both cases let
 correspondingly  $\bX^a=\bX^{C_1}\oplus\dots \oplus
 \bX^{C_p}$. Moreover we
 supplement  by writing $\bX^{a_{\min}}=\{0\}$ where, for example, $a_{\min}:=
\emptyset$. This is a concrete labeling of the family of subspaces of internal motion.

The ordering of subspaces yields an  ordering
of the abstract set of indices $\vA$, by definition $a\subset b\Leftrightarrow \bX^a\subset\bX^b$. We denote
$\bX=\bX^{a_{\max}}$ and $\bX^a+\bX^b=\bX^{a\cup b}$. The orthogonal
complement of $\bX^a$ is denoted by $\bX_a$. To have a uniform
language we refer to the  indices
$a\in \vA$ as `cluster decompositions'.  The length of a chain of
cluster decompositions $a_1\subsetneq \cdots   \subsetneq a_k$ is the
number $k$. This  chain is said to connect $a=a_1$ and $b=a_k$. The
maximal length of all chains connecting a given $a\in \vA\setminus\{a_{\max}\}$  and
$a_{\max}$ is denoted by $\# a$. We define $\# a_{\max}=1$ and note
that $\# a_{\min}=N+1$.  We say $a\in \vA$  is $k$-cluster if $\# a=k$.

We note the following geometric properties for $N\geq 2$: For all
$a,b\in \vA$ with
  $\#a=2$,  $\#b=N$  and $b\not\subset a$
\begin{subequations}
  \begin{align}
    \label{eq:22AA}
    &\ran\parb{\pi^b\pi^{a}}= \{0\}\text{ or }\ran\parb{\pi^b\pi^{a}}= \ran\,\pi^b,\\
    \label{eq:89AA}
    &\pi^b:\bX_{a}\to \bX^b\text{  is
bijective}.
  \end{align}
   \end{subequations}

\subsection{Principal example, fixed nuclei}\label{Second  principal example}
Consider a system of  $N$ $n$-dimensional particles, $n\geq 3$,  interacting by Coulomb forces. The
Hamiltonian \eqref{0.1b}  then reads
\begin{equation}
\label{0.1Cou2}
H=-\sum_{j = 1}^N \frac{1}{2m_j}\Delta_{x_j} + \sum_{1 \le i<j \le N} q_iq_j|x_i
-x_j|^{-1}+ \sum_{1 \le j \le N,\,1 \le m \le M,} q_j q_m^{\rm ncl}|x_j
-R_m|^{-1},
\end{equation}
where $x_j$,  $m_j$ and $q_j$ denote the position, mass and charge of
the $j$'th `electron', and $R_m$ and $q_m^{\rm ncl}$ are the
position and charge of
the $m$'th `nucleus'.

Consider the two-cluster decomposition
$a=(C)$, $C=\{1,\dots,N-1\}$, meaning $\bX^a=\{x=(x_1,\dots,x_N)\in
\bX=\R^{nN}\mid x_N=0\}$. Letting $R=x_N$ we write $H=H^1\otimes 1+1
\otimes p_{R}^2+I_a$
where  $H^1$ is the
cluster-Hamiltonian (i.e. the Hamiltonian for the first $N-1$ electrons) and
\begin{align*}
  I_a=\sum_{1 \le i\le N-1} q_iq_N|x_i
-R|^{-1}+ \sum_{1 \le m \le M,} q_N q_m^{\rm ncl}|R
-R_m|^{-1}.
\end{align*} Introducing
\begin{align*}
  Q&=\sum_{1\leq j\leq N-1} q_j+\sum_{1\leq m\leq M}q_m^{\rm ncl},\\\widetilde Q&=\widetilde Q(x^a)=\sum_{1\leq j\leq N-1}q_jx_j,\\\widetilde Q^{\rm ncl}&=\sum_{1\leq m\leq M}q_m^{\rm ncl}R_m,
\end{align*}
  the asymptotics  of $I_a$
for $|R|\to \infty$ reads
\begin{align}\label{eq:expans2}
  I_a=q_NQ|R|^{-1}+q_N\tfrac{R}{|R|^{3}}\cdot \parb{\widetilde Q(x^a)+\widetilde Q^{\rm ncl}}+\vO\parb{|R|^{-3}}\parb{1+\abs{x^a}^2}.
\end{align} For the expectation in a cluster  bound state $\varphi=\varphi^a(x^a)$ with sufficient
decay we consequently obtain the asymptotics for $|R|\to \infty$
\begin{align}
  \label{eq:effe_pot2}
 \inp{\varphi,I_a\varphi}_{L^2(\bX^a)}=q_NQ|R|^{-1}+q_N\tfrac{R}{|R|^{3}}\cdot \parb{\inp{\varphi,\widetilde
    Q\varphi}_{L^2(\bX^a)}+\widetilde Q^{\rm ncl}}+\vO\parb{|R|^{-3}}.
\end{align} This leads to various cases.
\begin{description}
                                   \item[Case 1] \quad  \quad  \quad $q_NQ\neq 0$.
                                   \item[Case 2]  \quad  \quad \quad  $q_N\inp{\varphi,\widetilde
    Q\varphi}_{L^2(\bX^a)}\neq -q_N\widetilde Q^{\mathrm {ncl}}$ and $Q=0$.

\item[Case 3]   \quad  \quad  \quad  $q_N =0$, or $q_N \neq 0$, $Q=
  0$ and $\inp{\varphi,\widetilde
    Q\varphi}_{L^2(\bX^a)}=-\widetilde Q^{\mathrm {ncl}}$.

                                   \end{description}

 Case 1 is the
  slowly decaying  case, Case 2
  is the critical case and Case 3 is the fastly  decaying  case. Strictly speaking
 this classification makes best
sense for $\varphi$ being unique, i.e. for the    simple case; in the non-simple case one
needs a slightly different terminology, see Subsection \ref{First
  principal example} and Section \ref{sec:CoulRellich}.

\section{Generalized $N$-body Schr\"odinger
  operators}\label{Generalized $N$-body Schr\"odinger operators}
Motivated by Sections \ref{$N$-body Schr\"odinger operators} and \ref{$N$-body Schr\"odinger operators with infinite mass
  nuclei}  we discuss the abstract $N$-body problem,
  cf.  \cite[Section 5.1]{DG}.
Let $\bX\neq \{0\}$ be a real finite dimensional vector space with an inner
product $q$. We consider a finite family $\{\bX^a\mid a\in\vA\}$ of
subspaces of $\bX^a\subset\bX$  which is stable under addition and which contains
$\{0\}$ and  $\bX$. The ordering of subspaces yields an ordering
of the abstract set of indices $\vA$, $a\subset b\Leftrightarrow \bX^a\subset\bX^b$. We denote $\{0\}=\bX^{a_{\min}}$,
$\bX=\bX^{a_{\max}}$ and $\bX^a+\bX^b=\bX^{a\cup b}$. The orthogonal
complement of $\bX^a$ is denoted by $\bX_a$. We refer to the  indices
$a\in \vA$ as `cluster decompositions'. The \emph{length} of a chain of
cluster decompositions $a_1\subsetneq \cdots   \subsetneq a_k$ is the
number $k$. This  chain is said to connect $a=a_1$ and $b=a_k$. The
maximal length of all chains connecting a given $a\in \vA\setminus\{a_{\max}\}$  and
$a_{\max}$ is denoted by $\# a$. We define $\# a_{\max}=1$ and denoting
 $\# a_{\min}=N+1$ we say the  family $\{\bX^a\mid a\in\vA\}$ is of $N$-body
 type. Note that  for the setup of Sections  \ref{$N$-body Schr\"odinger
   operators} and \ref{$N$-body Schr\"odinger operators with infinite mass
  nuclei} these examples are  of $(N-1)$-body
 type  and of $N$-body
 type, respectively. This terminology might appear  slightly misleading for
 Section  \ref{$N$-body Schr\"odinger
   operators}. Henceforth we shall treat  the generalized $N$-body framework
 only. This would  consequently  apply to the   many-body
 framework of Section  \ref{$N$-body Schr\"odinger
   operators} with $N$ there replaced by $N+1$. A cluster
 decomposition  $a$  is said to be \emph{$k$-cluster} if $\# a=k$.

Given the above uniform setup of structure  of `internal subspaces'
we can introduce corresponding generalized  Schr\"odinger
  operators. Let  $ -\Delta^{a}=(p^a)^2$  and $-\Delta_{a}=p_a^2$ denote (minus) the
Laplacians on  $L^2(\bX^a)$ and  $L^2(\bX_a)$, respectively. Here
$p^a=\pi^ap$ and $p_a=\pi_ap$ denote the internal (i.e. `within clusters') and the
inter-cluster components of the momentum operator $p=-\i\nabla$, respectively.
For all  $a\in \vA':=\vA\setminus\{a_{\min}\}$, we introduce
\[  H^{a} = -\Delta^{a} + V^a(x^a), \, V^a(x^a)=\sum_{b\subset a}   V_{b}(x^{b}), \, H_{a}= H^{a}- \Delta_{a} ,\,
I_{a}(x) = \sum_{b\not\subset a} V_{b}(x^{b}),\]
where the potentials fulfill  the condition below.
We define $H^{a_{\min}}=0$ on $L^2(\bX^{a_{\min}})=\C$ and
$H=H^{a_{\max}}$  on $L^2(\bX)$.

\begin{cond}\label{cond:smooth}
  There exists $\rho >0$ such that for all  $a\in
  \vA'$ there is given a function $V_a:\bX^a\to \R$
    with a splitting $V_a=V_a^{(1)}+V_a^{(2)}$, where
  \begin{enumerate}[label=(\arabic*)]
  \item \label{item:cond1} $V_a^{(1)}$ is smooth and
    \begin{equation}
      \label{eq:1}
      \partial ^\alpha_yV_a^{(1)}(y)=\vO\big(|y|^{-\rho-|\alpha|}\big ).
    \end{equation}
  \item \label{item:cond12} $V_a^{(2)}$ is compactly supported and
    \begin{equation}
      \label{eq:2}
      V_a^{(2)}(-\Delta_y+1)^{-1}\text{ is compact on
      }L^2(\R^{\dim \bX^a}_y).
    \end{equation}
  \end{enumerate}
\end{cond}

\emph{Condition \ref{cond:smooth} will be imposed throughout this work}. To
 treat local singularities we shall impose an additional condition,
depending on an $a\in\vA$ from a given  context. The condition  is fulfilled for the
  models of Sections \ref{$N$-body Schr\"odinger operators} and ~\ref{$N$-body Schr\"odinger operators with infinite mass
  nuclei}, cf. \eqref{eq:22A},
\eqref{eq:89A}, \eqref{eq:22AA} and \eqref{eq:89AA}.

Consider  for a given  $a\in \vA$ with
  $\#a=2$  the following properties  for   $b\not\subset a$:
\begin{subequations}
  \begin{align}
    \label{eq:22}
    &\ran\parb{\pi^b\pi^{a}}= \{0\}\text{ or }\ran\parb{\pi^b\pi^{a}}= \ran\,\pi^b,\\
    \label{eq:89}
    &\pi^b:\bX_{a}\to \bX^b \text{ is onto}.
  \end{align}
\end{subequations} We note that  the map in
\eqref{eq:89}  is necessarily injective, and hence bijective if
\eqref{eq:89} is fulfilled, see  \eqref{eq:26}.

\begin{cond}\label{cond:geom_singl}   For all
  $b\not\subset a$ for which  the conditions  \eqref{eq:22} and
  \eqref{eq:89}  are   not
  fulfilled,   the singular part $V^{(2)}_b=0$ and hence $V_b=V^{(1)}_b\in
  C^\infty(\bX^b)$.
  \end{cond}

 If
$T$ is a self-adjoint operator on a Hilbert space the subsets of $\C$,
\begin{align*}
  \sigma(T),
\sigma_\d(T), \sigma_\ess (T)\mand  \sigma_\pp(T)
\end{align*}
 refer to the spectrum, the discrete spectrum, the essential spectrum
 and the set of eigenvalues of $T$, respectively.

The operator $H^a$ is the sub-Hamiltonian associated with the cluster decomposition
$a$ and $I_a$ is the sum of all inter-cluster interactions. The detailed
expression of $H^a$ depends on the choice of coordinates on
$\bX^a$.
 Let
\[
\vT =\vT (H)= \cup_{a\in\vA, \# a\ge 2} \;\sigma_{\pupo}( H^a)
\]
be the set of thresholds of $H$.
 The HVZ theorem \cite[Theorem XIII.17]{RS} gives
the bottom of the essential spectrum $\Sigma_2 :=\inf \sigma_{\ess}(H) $ of $H$ by the formula
\begin{equation}
  \label{eq:3}
  \Sigma_2 = \min_{a\in\vA\setminus\{a_{\max}\}} \inf\sigma( H^a) =
\min_{a\in\vA, \# a = 2} \inf\sigma( H^a).
\end{equation}

If $N\geq 2$ we also introduce
\begin{equation}
  \label{eq:3i}
  \Sigma_3 := \min_{a\in\vA, \# a\ge 3} \inf\sigma( H^a).
\end{equation}

Under Condition \ref{cond:smooth} it is known that non-threshold
 bound
states decay exponentially \cite{FH}. It is also well known
\cite{FH} and \cite{Pe} that under rather general conditions
generalized Schr\"odinger
  operators do
not have
 positive eigenvalues  and  that the  negative
eigenvalues can at most
accumulate at the thresholds from below, see also \cite{AIIS}.

The goal  of the present  work
is to obtain   spectral and  scattering properties of $H$ near a general \emph{two-cluster
threshold} $\lambda_0\leq 0$, i.e.
\begin{equation}
\label{ass0}
\lambda_0 \in \vT_2:=\vT \setminus  \cup_{a\in\vA, \# a\ge 3}
\;\sigma_{\pupo}( H^a).
\end{equation}
Note that $\lambda_0=0$ for  $N=1$, while $\lambda_0<0$ for
$N\geq2$. Note also that $\lambda_0\neq \Sigma_3 $ for $N\geq2$ and
that $\lambda_0> \Sigma_3 $ could occur  for $N\geq3$. Since the
case $N=1$ has been treated  extensively in the literature we assume
throughout our  work that $N\geq 2$ and therefore $\lambda_0<0$. The special case for
$N\geq2$ when $\lambda_0$ is equal to $\Sigma_2$ and is given as
the eigenvalue of a  unique two-cluster sub-Hamiltonian   is studied in \cite{Wa2} for fastly
decaying  potentials, meaning $\rho >2$ or bigger.

\subsection{Spaces and notation}\label{Spaces}
For given Banach (or Fr\'echet) spaces $X$ and $Y$,  the space of linear continuous  operators $T:X\to
Y$ is denoted by $\vL(X,Y)$, and we abbreviate $\vL(X)=\vL(X,X)$.

Let $H_s^{k}$, $k, s\in\R $,  be the weighted Sobolev space on $\bX$
 (or possibly  $\bX_a$ for any $a\in \vA'$) equipped with the norm
 \begin{align*}
\norm{u}_{k,s}:= \parbb{\int \,\abs[\big]{\inp{x}^s( 1-\Delta)^{k/2}u}^2
\d x}^{1/2},\quad\text{where } \inp{x}:=(1+|x|^2)^{1/2}.
   \end{align*}
 The space $ H^{-k}_{-s }$ can be identified as the dual space of  $ H^{k}_s$
 with the usual $L^2$ inner product used as  `pairing'.  Let
 ${\mathcal L}(k,s;k',s')={\mathcal L}(H^{k}_{s},H^{k'}_{s'})$.
It will be convenient to regard $H$ as an operator in ${\mathcal
  L}(2,0; 0,0)$ or as an operator in ${\mathcal
  L}(1,0; -1,0)$. We abbreviate
 $H^{k}=H_0^{k}$, $L_s^{2}=H^{0}_s$,  and $H^{k}_{\infty}=\cap_{s\in \R}H^{k}_{s}$,
$H^{k}_{-\infty}=\cup_{s\in \R}H^{k}_{s}$,  $
 L_\infty^{2}=\cap_{s\in\R}\,L_s^{2}$ and $
 L_{-\infty}^{2}=\cup_{s\in\R}\,L_s^{2}$. Introduce also $
 H^k_{t^+}=\cup_{s>t}H^k_{s}$, $ H^k_{t^-}=\cap_{s<t}H^k_{s}$, $
 L^2_{t^+}=\cup_{s>t}L^2_{s}$ and $ L^2_{t^-}=\cap_{s<t}L^2_{s}$ for
any $k, t\in\R $.

 We shall also use weighted
  Sobolev spaces of $\C^m$-valued functions/distributions on $\bX$
  indicated by similar notation, for example  $H^{k}=H^{k}(\bX;\C^m)$.
For any complex Hilbert spaces $\vH_1$ and $\vH_1$ we use the notation
$\vC\parb{\vH_1,\vH_2}\subset \vL\parb{\vH_1,\vH_2}$ for  the space of
compact operators  $T:\vH_1 \to \vH_2$.  By standard Sobolev
embedding theory, for all $k_1>k_2$ and $s_1>s_2$ the operator
$H^{k_1}_{s_1}\ni f\to f \in H^{k_2}_{s_2}$ is
compact. Let  $H^1_{\rm loc}$ denote the set of locally $H^1$
functions, more precisely the set of
$v\in L^2_{-\infty}$ for which $\varphi v\in H^1$ for all  $\varphi\in C_\c^\infty$.

Consider balls $B(R)=\{ x\in \bX|\,\abs{x}<R\}$, $R\geq 1$, and
the characteristic functions
\begin{align*}
F_0=F\parb{B({R_0})}\mand F_{\nu+1}=F\parb{B({R_{\nu+1}})\setminus B({R_\nu})},\ R_\nu=2^{\nu},
\,\nu\in\N_0:=\N\cup\{0\},
\end{align*}
 where $F(M)$ denotes the sharp characteristic
 function of a subset $M\subset \bX$.
We introduce the Besov spaces $\vB_s=\vB_s(\bX)$ and
$\vB_s^*=\vB^*_s(\bX)=\parb{\vB_s(\bX)}^*$, $s>0$,  as follows.
\begin{align*}
\vB_s&=\{\psi\in  L^2_{\mathrm{loc}}\mid  \norm{\psi}_{\vB_s}<\infty\},\quad
\|\psi\|_{\vB_s}=\sum_{\nu=0}^\infty R_\nu^s
\|F_\nu\psi\|_{{L^2}},\\
\vB_s^*&=\{\psi\in  L^2_{\mathrm{loc}}\mid  \norm{\psi}_{\vB_s^*}<\infty\},\quad
\|\psi\|_{\vB_s^*}=\sup_{\nu\ge 0} \,R_\nu^{-s}\|F_\nu\psi\|_{{L^2}},
\end{align*}
respectively. (Note that indeed $\vB_s^*$ is the dual space of $\vB_s$.)
We  define $\vB_{s,0}^*$ to be the closure of $L^2$ in
$\vB_s^*$. Note that
\begin{align*}
  u\in \vB_{s}^*\Leftrightarrow u\in L^2_{\rm loc}
  \mand
  \sup_{R\geq 1}\,R^{-s}\|F\parb{B(R)} u\|_{L^2}<\infty,
\end{align*} and that
\begin{align*}
  u\in \vB_{s,0}^*\Leftrightarrow u\in L^2_{\rm loc}
  \mand
  \lim_{R\to \infty}\,R^{-s}\|F\parb{B(R)}u\|_{L^2}=0.
\end{align*}

Note the following relations between the
standard weighted $L^2$ spaces and the Besov spaces:
\begin{align}\label{eq:nest}
 \forall s_1>s>0:\quad L^2_{s_1}\subsetneq \vB_s\subsetneq  L^2_{s}
\subsetneq L^2
\subsetneq L^2_{-s}\subsetneq \vB_{s,0}^*\subsetneq \vB_s^*\subsetneq  L^2_{-s_1}.
\end{align}

We introduce for any $R\geq 1$
\begin{align}
\chi_R(t)=\chi(t/R) \text{ for a  real-valued }\chi\in
C^\infty(\R):  \quad\chi(t)
=\left\{\begin{array}{ll}
1 &\mbox{ for } t \le 4/3 \\
0 &\mbox{ for } t \ge 5/3
\end{array},
\right.
\label{eq:14.1.7.23.24}
\end{align}
 and $\bar\chi_R=1-\chi_R$. We assume  $-\chi'\geq 0$ and that
$\sqrt{-\chi'}\in C^\infty$.


\chapter{Reduction to a one-body problem} \label{Reduction to a two-body problem}

In this chapter  we show that the spectral analysis of a generalized  $N$-body
operator near a   two-cluster threshold can be reduced to the analysis
of a one-body  operator with a non-linear spectral parameter. This is under
Condition \ref{cond:smooth}. Of course such reduction would not be
needed for the one-body problem, however we recall that  throughout our  work
we impose the condition  $N\geq 2$.

The idea of the reduction
goes as  follows.  For a given two-cluster threshold $\lambda_0< 0$, denote by $\vF$
the closed subspace in $L^2(\bX)$ `spanned by'  bound states of all
possible two-cluster sub-Hamiltonians with eigenvalue $\lambda_0$. (For a
precise definition/representation in a special case see \eqref{eq:4}
and the discussion there.) Let $\Pi$ be the orthogonal projection from
$L^2(\bX)$ onto $\vF$, $\Pi' = 1-\Pi$  and $H' = \Pi'H \Pi'$. We study
a   Grushin problem for an operator of the form
\begin{equation}
{ H}(z)=\left( \begin{array}{cc}
 H-z  &  S\\[.1in]
       S^* & 0
\end{array} \right) \,   \text{ on } L^2(\bX) \oplus \vH,
\end{equation}
 where the space $\vH$ is some auxiliary
one-body type space,
$S : \vH \to L^2(\bX)$ is an  appropriately defined  operator whose
range coincides with $\vF$, i.e. the range of the projection $\Pi$,  and $S^*$ is the adjoint of
$S$.  The choice of $\vH$ and $S$ may vary according to various
situations and is to make $H(z)$  invertible for $z$ near $\lambda_0$ and $\Im z \neq 0$.  If this is realized,  we set
\begin{equation}
{ H}(z)^{-1}=\left(
\begin{array}{ll} E(z)  & E_+(z)\\[.1in]
 E_-(z)  & E_{\vH}(z) \end{array}
\right).
\end{equation}
 Then the resolvent $R(z)=(H-z)^{-1}$ is represented as
\begin{equation} \label{rep}
R(z)= E(z)- E_{+}(z) E_{\vH}(z)^{-1} E_{-}(z).
\end{equation}
 The operator ${ H}(z)^{-1}$ can be computed in terms of the reduced
 resolvent $R'(z) =(H'-z)^{-1}\Pi'$. If $R'(z)$ has some good
 properties near $\lambda_0$, the spectral analysis of $H$ near $\lambda_0$ is
 then reduced to the analysis of $E_{\vH}(z)^{-1}$ on $\vH$.
 $E_{\vH}(z)$ is a one-body  operator with a non-linear dependence
 on the spectral parameter $z$. For
 example  to obtain a resolvent expansion,  one can then try to use  known methods for one-body operators to study the asymptotics of  $E_{\vH}(z)^{-1}$ as $z \to \lambda_0$, $\Im z >0$.

\fancybreak{}

Below we shall only give detailed analysis in some  situations. In
fact, the reductions given below are only useful in the case  $\lambda_0
\not\in \sigma_{\pupo}(H')$. The case $\lambda_0 \in \sigma_{\pupo}(H')$ can
be treated by adding an  additional finite dimensional space, as done
in \cite[Section 4]{Wa2}  for the lowest threshold. Before we go
into details on various concrete cases we study the  reduction
scheme  from a more general point of view. As we will see afterwards, the abstract
scheme  can be  applied to these cases.

\section{An abstract reduction scheme}
\label{An abstract reduction scheme} In this section we shall describe an
abstract reduction scheme that will be used many times.

\paragraph{\textbf Abstract framework} Suppose $\vH$ and $\vG$ are given
Hilbert spaces and that $\vF$ is a closed subspace of $\vG$. Suppose
$S: \vH \to\vF\subset \vG$ is a bi-continuous isomorphism (i.e. $S$ is
linear, one-to-one, onto $\vF$ and as a map from $\vH$ to $\vF$
bi-continuous). Let $\Pi$ denote the orthogonal projection in $\vG$
onto $\vF$, and let $\Pi'=1-\Pi$. Suppose $H$ is a self-adjoint
operator on $\vG$ with $\Pi: \vD (H)\to \vD (H)$, and that $\Pi'H\Pi$
and $\Pi H\Pi'$ (initially defined on the domain $\vD (H)$ of $H$) extend to bounded
operators on $\vG$. Define $H' = \Pi' H \Pi'$ with domain $\vD
(H')=\vD (H)\cap \ran (\Pi')=\Pi'\,\vD (H)$ and $H_{\Pi} = \Pi H \Pi$
with domain $\vD (H_\Pi)=\vD (H)\cap \ran (\Pi)=\Pi\,\vD (H)$. Then
$H' $ and $H_\Pi$ are self-adjoint on $\Pi'\vG$ and $\vF$,
respectively (see the proof of Lemma \ref{Lemma:basic}
\ref{item:3m}).  Let $R'(z) =(H'-z)^{-1}\Pi'$ for $\Im z\neq 0$.

Introduce for all $z\in \C$ with $\Im z\neq
  0$
\begin{subequations}
 \begin{align}
 \label{eq:5pm}E(z) &= R'(z),   \\
\label{eq:6pm} E_+(z)&= S -R'(z)H S,   \\
\label{eq:7pm} E_-(z) &= S^*- S^*H R'(z),\\
\label{eq:8pm} E_{\vH}(z) &=  S^* \big (z -  H +  H R'(z)H \big ) S.
\end{align}
\end{subequations} Obviously $E(z)\in \vL(\vG)$, $E_+(z)\in \vL(\vH,
\vG)$ and $E_-(z)\in \vL(\vG,\vH)$. Also note
that  $S^*HS=S^*H_{\Pi}S$ is self-adjoint and
consequently that the operator $E_{\vH}(z)$ of \eqref{eq:8pm} is a closed
operator on $\vH$ with domain
  given by
  $S^{-1}\vD (H_\Pi)$.

\begin{prop}
  \label{prop:reduc_form_abstract} Under the above conditions, for $\Im z\neq 0$:
  \begin{enumerate}[label=\textnormal{\arabic*)}]
  \item\label{item:11}
 $E_{\vH}(z)^*=E_{\vH}(\bar z)$.
  \item \label{item:13} $E_{\vH}(z)$ is an invertible operator on $\vH$
    obeying
    \begin{equation}
      \label{eq:86}
      \frac {\Im E_{\vH}(z)}{\Im z }\geq S^* S.
    \end{equation}
\item \label{item:14}
\begin{equation} \label{rep3ipm}
R(z)= E(z)-  E_{+}(z)  E_{\vH}(z)^{-1}  E_{-}(z).
\end{equation}
  \end{enumerate}
\end{prop}
\begin{proof} The identity \ref{item:11} is trivial due to the
   self-adjointness property stated before the
  proposition.

The inequality \eqref{eq:86} follows from the identity
\begin{equation}
    \label{eq:14}
    \Im  E_{\vH}(z)=\Im z \,S^* \big (1+  H R'(z)^*R'(z)H \big ) S.
  \end{equation} The invertibility property of \ref{item:13} follows
  from a standard numerical range argument
   combining \ref{item:11} and  \eqref{eq:86}.

 Finally the identity \eqref{rep3ipm}
  follows by an elementary calculation outlined here: Apply $H-z$ from
  the
  left on  the right-hand side of the identity  and
  write for the first  term of \eqref{rep3ipm}
  \begin{equation}
    \label{eq:87}
    H\Pi'=
  \Pi'H'\Pi'+\Pi H \Pi'.
  \end{equation}  The result is the expression
  $\Pi'+\Pi HR'(z)$.  For the second term of \eqref{rep3ipm} we substitute the
  expression \eqref{eq:6pm} and use  again \eqref{eq:87} (to deal with
  the second term of \eqref{eq:6pm}). We see that by applying $H-z$  to the second term of
  \eqref{rep3ipm} we obtain an operator taking values in
  $\vF$. Consequently we may insert $\Pi=(S^*)^{-1} S^*$ to the left
  and then combine the factor $S^*$ with the given calculated expression. This gives  a factor $E_{\vH}(z)E_{\vH}(z)^{-1} $
  which of course can be omitted on $\vH= \ran E_{-}(z)$. Then
  by substituting  the
  expression \eqref{eq:7pm}  we conclude that the contribution from
  the second term of \eqref{rep3ipm} is $\Pi-\Pi HR'(z)$. Hence the
result of applying $H-z$ from  the
  left on the right-hand side of \eqref{rep3ipm} is
  \begin{equation*}
    \big (\Pi'+\Pi HR'(z)\big)+\big(\Pi-\Pi HR'(z)\big)=1,
  \end{equation*} which coincides with the
  result of applying $H-z$ from  the
  left on the left-hand side of the identity.
\end{proof}

  \begin{remark}\label{remark:Grushin} There is an alternative
    approach for  deriving
    the formula \eqref{rep3ipm} based on a
certain  abstract Grushin problem. For a review of
this method we may refer to \cite{SZ}, but for  sake of completeness of presentation
we  outline  this   alternative
  approach here, cf. the beginning of this chapter:

We
consider  a   Grushin problem for an operator of the form
\begin{equation} \label{g2o}
{H}(z)=\left( \begin{array}{cc}
 H-z  & S\\[.1in]
 S^* & 0
\end{array} \right)  \text{ on }\vG \oplus \vH.
\end{equation}
To simplify the notation, we denote still by
$(S S^*)^{-1}= (S S^*)^{-1}\Pi$
the extension  of $(S S^*)^{-1}$ to $\vG$ by setting $(S S^*)^{-1}=0$
on $\vF^\perp$.

The invertibility of $SS^*$ on $\vF$ allows us to show that ${ H}(z)$ is
invertible on $\vG \oplus \vH$. In fact, one can explicitly compute its inverse.
Let $T:=(S S^*)^{-1}S$, and let ${ Q}(z)$ and  ${ B}(z)$  be defined by
\[{ Q}(z)=\left(\begin{array}{cc} R'(z)  & T \\ T^* & T^*(z- H )T
\end{array} \right)\mand B(z) = \left(
                \begin{array}{cc}
                  \Pi H R'(z) & \Pi' H T\\[.1in]
                   0 & 0
                \end{array}
                \right).
\]
Then one has, at least formally,
\begin{equation*}
{ H}(z){ Q}(z) =
   1      + { B}(z).
\end{equation*}
  Since $B(z)^3 = 0$ the operator  $1 + B(z)$ is invertible and therefore
${H}(z)$ has a right inverse. Similarly, one can show that
${ H}(z)$ has a left inverse. Consequently ${ H}(z)$ should be  invertible
 with inverse
\begin{equation}
{H}(z)^{-1} = {Q}(z) ( 1 - B(z) + B(z)^2).
\end{equation}
Write  $ { H }(z)^{-1}$ in the form
\begin{equation} \label{g3o}
H(z)^{-1}=\left(
\begin{array}{ll} \widetilde E(z)  & \widetilde E_+(z)\\[.1in]
 \widetilde E_-(z)  & \widetilde E_{\vH}(z) \end{array}
\right).
\end{equation}
We have the formulas
\begin{subequations}
 \begin{align}
 \label{eq:5o}\widetilde E(z) &= R'(z),   \\
\label{eq:6o} \widetilde E_+(z)&= T -R'(z)H T,   \\
\label{eq:7o} \widetilde E_-(z) &= T^*- T^*H R'(z),\\
\label{eq:8o} \widetilde E_{\vH}(z) &=  T^* \big (z -  H +  H R'(z)H \big ) T,
\end{align}
\end{subequations}
and
\begin{equation} \label{rep3io}
R(z)= \widetilde E(z)-  \widetilde E_{+}(z)  \widetilde E_{\vH}(z)^{-1}  \widetilde E_{-}(z).
\end{equation}
 Now \eqref{rep3ipm} follows from \eqref{rep3io} and
the identity
\begin{align}\label{eq:commIdent}
T=(S S^*)^{-1}S=S(S^* S)^{-1}.
\end{align}
\end{remark}

\begin{remark}\label{remark:GrushinBB}
Suppose $\lambda_0\notin \sigma(H')$ so that the operators in \eqref{eq:5o}--\eqref{eq:8o} have limits
when taking $z\to \lambda =\lambda_0\in \R$. Then we learn from the identities
\begin{align*}
  \widetilde E_-(\lambda)(H-\lambda)&\subset- \widetilde E_{\vH}(\lambda) S^*,\\
S\widetilde E_{\vH}(\lambda)&\subset -(H-\lambda)\widetilde E_+(\lambda),
\end{align*} that
\begin{align*}
  S^*: &\ker (H-\lambda)\to \ker \widetilde E_{\vH}(\lambda),\\
\widetilde E_+(\lambda): &\ker \widetilde E_{\vH}(\lambda) \to \ker (H-\lambda),
\end{align*} respectively.

Combined with the identities
\begin{align*}
  \widetilde E_+(\lambda) S^*&\supset1-\widetilde E(\lambda)(H-\lambda),\\
S^*\widetilde E_+(\lambda)&=1,
\end{align*} we then conclude that in fact
\begin{align*}
  \widetilde E_+(\lambda)S^*&=1\text{ on }\ker (H-\lambda),\\
S^*\widetilde E_+(\lambda)&=1\text{ on }\ker \widetilde E_{\vH}(\lambda),
\end{align*} i.e. $S^*: \ker (H-\lambda)\to \ker \widetilde
E_{\vH}(\lambda)$ is a linear  isomorphism.

Whence, cf. \eqref{eq:commIdent}, $T^*: \ker (H-\lambda)\to \ker
E_{\vH}(\lambda)$ is a linear isomorphism with inverse $E_+(\lambda):
\ker E_{\vH}(\lambda) \to \ker (H-\lambda)$. We shall refer to the
 vector  $f=T^*\phi\in\ker E_{\vH}(\lambda)$ as the
\emph{eigentransform} of a given $\phi\in \ker (H-\lambda)$, and the
equation $\phi=E_+(\lambda)f$ as the \emph{inversion formula} for the
eigentransform of $\phi$. The Hilbert space setting discussed here is
in our application an $L^2$ setting which will have
extensions to Besov space  settings  and settings where $\lambda_0\in \sigma(H')$,
that are  incompatible with the framework discussed
here. These concrete extensions will be  studied in Chapter \ref{chap:lowest thr}.
\end{remark}

\section{Non-multiple  two-cluster threshold case}
\label{Reduction near a simple two-cluster threshold}

 Assume in this section that $\lambda_0<0$ is a \emph{non-multiple  two-cluster
 threshold} in the following sense:
\begin{cond} \label{cond:uniq}
 There exists a unique $a_0 \in\vA\setminus\{a_{\max}\}$ such
 that $\lambda_0 \in\sigma_{\pupo}(H^{a_0})$. This cluster decomposition is
 a two-cluster decomposition, i.e.
   $ \# a_0 =2$.  Condition \ref{cond:geom_singl} is fulfilled for  $a=a_0$.
\end{cond}
Note that we do not here impose that  $\lambda_0 \in
\sigma_{\d}(H^{a_0})$. Nevertheless $\lambda_0$ is not a threshold for the
 Hamiltonian $H^{a_0}$, and consequently the corresponding bound states
 have exponential decay, cf. \cite{FH, AIIS}.
Note that the multiplicity of $\lambda_0$  as an eigenvalue of $H^{a_0}$, say  $m$,   can  be arbitrary.
The simplest case is $\lambda_0 = \Sigma_2$ (under Condition \ref{cond:uniq}
necessarily $m=1$ in this case), and it is studied in \cite{Wa2} for
fastly  decaying potentials.
For $\lambda_0 > \Sigma_2$, we can  use the same idea to reduce $H$
to a one-body type operator with a non-linear spectral parameter,
although additional complications arise, in particular for  $\lambda_0>\Sigma_3$.

\fancybreak{}

Let $\{\varphi_1, \dots, \varphi_m\} $ be an orthonormal basis of the eigenspace of  $H^{a_0}$ associated with $\lambda_0$.
Let $\Pi $ be the projection in $\vG:=L^2(\bX)$  defined by
\begin{equation}
  \label{eq:4}
  \Pi g = \sum_{j=1}^m \varphi_j\otimes \inp{\varphi_j,g}_{L^2(\bX^{a_0})},  \;  g\in \vG.
\end{equation}

Let $\Pi' = 1 - \Pi$ and $H' = \Pi' H \Pi'$. Note  that $H' $ is
self-adjoint with  domain $\vD (H')=\vD (H)\cap \ran
\,\Pi'=\Pi'\,\vD
(H)$. In fact $H_{a_0} $ is reduced by $\Pi'$ (`reduced' in the sense of
\cite[Subsection V.3.9]{Ka})  implying in particular that $H_{a_0}'= \Pi' H_{a_0} \Pi'$ is
self-adjoint.  Since $I_{a_0}'= \Pi' I_{a_0} \Pi'$ is infinitesimally small
relatively to $H_{a_0}' $ it follows from \cite[Theorem X.12]{RS} that
indeed $H' $ is
self-adjoint. For
$\Im z \neq 0$ we  set
\[
R'(z) = (H'-z)^{-1} \Pi'.
\]

 Let $I_0=I_{a_0}$ and abbreviate similarly $p_0=p_{a_0}$. Since we
  have imposed Condition \ref{cond:geom_singl}  with $a=a_0$ we can
  use  a `free factor' $(|p^{a_0}|^2+1)^{-1}$ (from $\Pi$) to conclude
  that
\begin{align}
    \label{eq:assump_singa}
    \Pi'I_{0}\Pi\in \vL \parb{\vG}.
  \end{align}
 Let $\vH=\vH_{a_0}:=L^2(\bX_{a_0}; \C^m) $, $H^2_{a_0}:=H^2(\bX_{a_0};
 \C^m) \subset \vH_{a_0}$ and the operator  $S : \vH\to
 \vG$ be  defined by
\begin{equation}\label{eq:S}
  S : f= (f_1, \dots, f_m) \to  Sf =\sum_{j=1}^m S_jf_j=\sum_{j=1}^m \varphi_j(x^{a_0}) f_j(x_{a_0}).
\end{equation} Obviously $S : H^2_{a_0}\to
 H^2(\bX)$ and the ($L^2)$ adjoint $S^*:H^2\to
H^2_{a_0}$. Since $S$ in this case is isometric the formulas
\eqref{eq:5pm}--\eqref{eq:8pm} and \eqref{eq:5o}--\eqref{eq:8o}
coincide. They  read (in terms of the identity matrix ${\textbf 1}_m$
of size  $m$)
\begin{subequations}
  \begin{align}
\label{eq:9}E(z) &= R'(z),   \\
\label{eq:10}E_+(z)&=(1-R'(z)I_{0})S,   \\
 \label{eq:11}E_-(z) &= S^*(1- I_{0}R'(z)),\\
\label{eq:12}E_{\vH}(z)&= (z-\lambda_0)-(|p_{0}|^2{\textbf 1}_m+ S^*{ I_{0}}S -
S^* I_{0}R'(z)I_{0} S),
\end{align}
\end{subequations} which can be used in combination with \eqref{rep3ipm}, i.e.
\begin{equation} \label{rep1}
R(z)= E(z)- E_{+}(z) E_{\vH}(z)^{-1} E_{-}(z).
\end{equation}

   We examine
  in the following
  some basic properties of $E_{\vH}(z)$.
Note that $S^*{ I_{a_0}}S $ is a matrix-valued potential. Let us
  first split $ I_{0}=  I^{(1)}_{0}+ I^{(2)}_{0}=I^{(1)}_{a_0}+ I^{(2)}_{a_0}$ in agreement with the
  splitting of Condition \ref{cond:smooth} and look at the
  contribution from $I^{(1)}_{0}$.  Recalling  the elementary geometric property
\begin{equation}
  \label{eq:26}
  X_a\cap X_b=\{0\} \text { if }\#a=2\mand b\not\subset a,
\end{equation} guaranteeing that
$\pi^b:\bX_{a_0}\to \bX^b$ is injective for all $b\not\subset
a_0$, we obtain by  a  Taylor
expansion  that the leading term
has a  scalar leading form. More precisely we obtain by a zero'th
order Taylor
expansion, by using \eqref{eq:elementEst} (stated below)  and by using the
polynomial decay of the bound states $\varphi_j$ that
\begin{subequations}
\begin{equation}\label{eq:22asym}
S_i^*{ I^{(1)}_{0}}S_j  =  I^{(1)}_{0}(0+ y)\delta_{ij} +
\vO(|y|^{-\rho-1})\text{ for }|y|  \to \infty;\quad i,j\leq m.
\end{equation} Here  the variable   $y=x_{a_0}$ can be thought of as a vector in $\R^n$,  abbreviating here and henceforth $\dim \bX_{a_0}=n$.
  We
 claim the following analogue of \eqref{eq:22asym} for $I^{(2)}_{0}$:
\begin{equation}\label{eq:22asymBa}
S_i^*{ I^{(2)}_{0}}S_j  -  I^{(2)}_{0}(0+ y)\delta_{ij} \in
\vC\parb{H^{2}_{s}(\bX_{a_0}),L^{2}_{t}(\bX_{a_0})}\mforall s,t\in\R.
\end{equation}
 \end{subequations}

   To prove \eqref{eq:22asymBa}  it suffices to
consider the contribution from $V_b^{(2)} $ for any $b\not\subset
a_0$ under the conditions
\eqref{eq:22} and \eqref{eq:89}, in fact with the second condition of
\eqref{eq:22} fulfilled.  Due to
\eqref{eq:89} the restriction
$V_b^{(2)}(=V_b^{(2)}(y))\in\vC\parb{H_s^{2}(\bX_{a_0}),L_t^{2}(\bX_{a_0})}$.
 It is compactly supported,  possibly with singularities. On the other
 hand
$S_i^*{ V_b^{(2)}}S_j $ is a bounded potential due to the
second  condition  of
\eqref{eq:22}, and it remains to show that this potential decays
faster than any negative power $\inp{y}^{-s}$.   By the compact
support property we can pick  $R>1$ such that
\begin{align*}
  F(|\pi^b\cdot|>R)S_i^*{ V_b^{(2)}}S_j =S_i^*{ F\parb{|\pi^by|>R}F\parb{|\pi^bx^{a_0}|\geq\tfrac12|\pi^by|}V_b^{(2)}}S_j.
\end{align*}
By the polynomial decay of the cluster bound
states the right-hand side
has arbitrary power decay too, showing the desired decay and therefore \eqref{eq:22asymBa}.

We conclude from \eqref{eq:22asym} and \eqref{eq:22asymBa} that
\begin{align}\label{eq:22asymB}
  \begin{split}
   & S_i^*I_{0}S_j  -  I^{(1)}_{0}(0+ y)\delta_{ij} \in
\vC\parb{H^{2}_{s}(\bX_{a_0}),L^{2}_{t}(\bX_{a_0})}\mfor
s,t\in\R,\,t<\rho+1+s,\\
&S_i^*I_{0}S_j  -  I^{(1)}_{0}(0+ y)\delta_{ij} \in
\vL\parb{H^{2}_{s}(\bX_{a_0}),L^{2}_{\rho+1+s}(\bX_{a_0})}\mfor
s\in\R.
\end{split}
\end{align}

It remains to examine the  term
$S^* I_{0}R'(z)I_{0} S$ in
\eqref{eq:12}.  Note that the two factors of $I_{0}$ freely can be
changed to $\Pi I_{0}\Pi'$ and $\Pi' I_{0}\Pi$,
respectively.  We can then conveniently
implement the following improvements of  \eqref{eq:assump_singa}:
\begin{align}
  \label{eq:90}
  \Pi I_{0}\Pi',\,\Pi' I_{0}\Pi\in \vL\parb{L^{2}_{s}(\bX),L^{2}_{\rho+1+s}(\bX)}\mforall s\in\R.
\end{align} For these bounds it suffices to consider  $\Pi' I_{0}\Pi$. The contribution
from  $I_0^{(1)}$ is treated by a Taylor expansion as in
\eqref{eq:22asym} (and by \eqref{eq:elementEst}).  The contribution from  $I_0^{(2)}$ is treated by
  using a `free factor' $(|p^{a_0}|^2+1)^{-1}$  (as for
  \eqref{eq:assump_singa}) as well as using the above proof of
  \eqref{eq:22asymBa}.

We conclude  the following representation of the operator
$E_{\vH}(z)$:
\begin{align*}
  E_{\vH}(z)&+\parb{p^2_{0}+I^{(1)}_{0}(\pi_{a_0}\cdot)+\lambda_0-z}{\textbf 1}_m=
              \widetilde{V}-\inp{y}^{-\rho-1}\widetilde K(z)\inp{y}^{-\rho-1};\\
&\widetilde{V}=\vO\parb{\inp{y}^{-\rho-1}},\quad \widetilde K(z)=\vO\parb{\inp{y}^0};
\end{align*} here the meaning of the $\vO(\cdot)$ notation is given more
precisely by \eqref{eq:22asymB} and \eqref{eq:90}, respectively. Note
that $\widetilde{V}$ is multiplicative matrix-valued, while
$\widetilde K\in\vL(\vH)$ is
non-multiplicative with a $z$-dependence through the appearance of the
factor $R'(z)$. Note also that $\inp{y}^{\rho}I_{0}^{(1)}(y)$ is
bounded. Although the above discussion is based on a natural operator
interpretation of $E_{\vH}(z)$ we shall prefer to use the
corresponding form
interpretation (see  Remark \ref{remark:formbnd} for details in a
different but similar context).

  The representation formula \eqref{rep1} is valid for any $z$ with
  $\Im z \neq 0$, but for our purposes it is only useful if one has good properties of
  $R'(z)$ for $z$ near $\lambda_0$.  When $\lambda_0$ is the lowest
  threshold, the essential spectrum of $H'$ is shifted to the right
  and the representation \eqref{rep1} can be used in the case
  $\lambda_0$ is not an eigenvalue of $H'$, cf. \cite[Lemma 2.1]{Wa2}
  (see Lemma \ref{lemma2.1} for an extension).    With
  this assumption, the reduced resolvent $R'(z)$ is holomorphic in a small neighbourhood of
  $\lambda_0$, and so are $E_{+}(z), E_{\vH}(z)$ and $ E_{-}(z)$. Moreover
  in this case  $R'(z)$  has uniform bounds in weighted $L^2$ spaces
   (see  \eqref{eq:91BB} for a similar assertion for the multiple case). This is helpful in
   the study of asymptotic expansions of $E_{\vH}(z)^{-1}$ at
   $\lambda_0$  and therefore in turn (by \eqref{rep1}) for expansions
   of $R(z)$ near $\lambda_0$ (see Section
   \ref{sec:resolv-asympt-nearLOWEST} for the multiple case).  In the case $\lambda_0$ happens to be
  an eigenvalue of $H'$ another reduction is needed.

If $\lambda_0 >\Sigma_2$, $\lambda_0$ is in the essential
spectrum of $H'$. Under Conditions \ref{cond:smooth} and
\ref{cond:uniq} we shall show in Chapter \ref{Spectral analysis of H'
  near E_0} that a limiting absorption principle and microlocal
resolvent estimates hold for $R'(z)$ with $z$ near $\lambda_0$
provided $\lambda_0$ is not an eigenvalue of $H'$. This leads in this
case to an extension of the `eigentransform' discussed in Remark
\ref{remark:GrushinBB} (see Chapter
\ref{chap:lowest thr}), and then in
turn \eqref{rep1} can again be used to   analyse  the resolvent of
 $H$ near the threshold $\lambda_0$ (see Section \ref{sec:resolv-asympt-nearHIGHER}). The analysis for $\lambda_0 >\Sigma_2$
is more complicated  than for $\lambda_0 =\Sigma_2$ in that, some  more
refined mapping properties of  $R'(z)$ near $\lambda_0$ are needed,
causing in particular some `loss of weight'.

\section{Multiple two-cluster threshold,\,
  $\vF_1\cap\vF_2=\{0\}$}\label{sec:Reduction near a multiple two-cluster threshold}

Consider now the case where  there exist at least two   two-cluster
decompositions such that the two-cluster threshold $\lambda_0(<0)$ is an eigenvalue of
the corresponding sub-Hamiltonians. To simplify the presentation, let us only consider in detail the
case where $\lambda_0$ is  double occurring without  eigenvalue multiplicity:
\begin{cond} \label{cond:uniqdd}
 There exist  unique $a_1, a_2 \in\vA\setminus\{a_{\max}\}$, $a_1\neq  a_2 $, such
 that $\lambda_0 \in\sigma_{\pupo}(H^{a_j})$, $j=1,2$. The cluster
 decompositions $a_1$ and $  a_2 $ are
  two-cluster decompositions, i.e.
   $ \# a_1 =\# a_2 =2$, and   Condition \ref{cond:geom_singl} is
   fulfilled for  both of them. The number $\lambda_0$ is a simple eigenvalue
for both of the operators $H^{a_1}$ and $H^{a_2}$.
\end{cond}

Denote
\[
\bX^j = \bX^{a_j}, \quad \bX_j = \bX_{a_j}\quad \text{and}\quad n_j =
\dim\bX_{_j}\quad \text{for}\quad j =1,2,
\] and similarly for elements $x$ of $\bX^{a_j}$ and $\bX_{a_j}$,
viz. $x=x^j\oplus x_j$.
  Let
$\varphi_j$ denote the corresponding (real) normalized eigenvector of
$H^{a_j}$, $j=1,2$. Let
\[
\vF_j =\{ g = \varphi_{j}(x^j)f_j(x_j)\mid  f_j \in L^2(\R^{n_j}_{x_j})\},  \quad j =1, 2.
\]
Note that $\vF_j$ is a closed subspace of $\vG:=L^2(\bX)$. Let  $\Pi $ be the orthogonal
projection in  $\vG$ onto  $ \vF$, which by definition is the closure of $ \vF_1 +
\vF_2$ in $\vG$. Let  $\vH = L^2(\bX_1) \oplus L^2(\bX_2)$.

To construct an associated Grushin problem  we need to distinguish
between two cases:
a) $\vF_1 \cap \vF_2 =\{0\}$ and b) $\vF_1 \cap \vF_2 \neq\{0\}$.
Here we shall  impose the condition a), i.e.
\begin{equation} \label{ass2}
\vF_1 + \vF_2 = \vF_1 \oplus \vF_2.
\end{equation}
  We believe that this condition is always fulfilled  for the finite mass
  many-body operators of  Section \ref{$N$-body Schr\"odinger
    operators}, however we do not have a proof.
  A case where it  fails for the infinite
 mass many-body operators of  Section \ref{$N$-body Schr\"odinger operators with infinite mass
  nuclei}  will be discussed in   Section \ref{sec:The case when the
  condition {ass2}}. Moreover   the general case of b) will be studied there.

Let
$S= (S_1, S_2) : \vH  \to \vG$ be  defined by $Sf=S_1f_1+S_2f_2$
for $f=(f_1,f_2)$ and with
\begin{align}\label{eq:19}
  \begin{split}
  S_j&:  L^2(\bX_j) \to \vG,\quad  f_j \to S_jf_j =
\varphi_j(x^j) \otimes f_j(x_j),\\
S_j^*&:  \vG\to  L^2(\bX_j),
\quad f\to  S_j^*f =
\inp{\varphi_j,f}_j;
\quad j =1,2.
  \end{split}
\end{align}
  Here $\inp{.,.}_j$ denotes the scalar product in
$L^2(\bX^j)$;   the notation $\w{\cdot, \cdot}$ will be used to denote  the scalar product in $\vG$.
One has
\[
S^*S =  1 + \left( \begin{array}{cc}
                  0 & s_{12}\\[.1in]
                   s_{21} & 0
                \end{array}
                \right) \quad \mbox{ on } \vH =L^2(\bX_1) \oplus L^2(\bX_2),
\]
where $s_{ij} \in \vL\parb{ L^2(\bX_j),L^2(\bX_i)}, \;i\neq j,$ are given by
\begin{equation*}
  s_{ij}f_j= \inp{\varphi_i, \varphi_j \otimes f_j}_i.
\end{equation*}
  In the proof of
 Proposition \ref{prop2.2} stated below we shall use that $s_{ij}
\in \vC\parb{ L^2(\bX_j),L^2(\bX_i)}$, $i\neq j$.
This property
is a consequence of the following lemma.
\begin{lemma} \label{Lemma:basic2A} For all $r,t\in \R$
\begin{equation*}
  s_{ij} \in \vL
  \big ( L^{2}_{r}(\bX_j), H^{2}_{t}(\bX_i)\big )\cap \vL
  \big ( H^{-2}_{r}(\bX_j), L^{2}_{t}(\bX_i)\big ).
  \end{equation*}
  \end{lemma}
  \begin{proof}
      Note that \eqref{eq:26} implies that
\begin{align}
  \label{eq:21}
  \ran(\pi_b\pi^a)=\ran (\pi_b) \text { if }\#b=2\mand a\not\subset
 b.
\end{align}
 We apply \eqref{eq:21} to $a=a_i$ and $b=a_j$ and  do  integration by
parts  (or alternatively change variables) obtaining  that
\begin{equation*}
  s_{ij} \in \vL
  \big ( L^{2}(\bX_j), H^{2}(\bX_i)\big ).
  \end{equation*}  Whence for  $|\alpha|\leq 2$ we have  a formula for
  $\partial^\alpha_{x_i}s_{ij}$ involving partial derivatives
  $\partial^\beta\varphi_k$, $k=1,2$ and $|\beta|\leq 2$. By using this
  formula, polynomial decay of  $\partial^\beta\varphi_k$ and the bound
  \begin{align*}
    |x^i|+|x^j|\geq c|x|
  \end{align*} (the latter is   a consequence of    \eqref{eq:26}) we
  obtain that indeed
\begin{equation*}
  s_{ij} \in \vL
  \big ( L^{2}_{r}(\bX_j), H^{2}_{t}(\bX_i)\big ).
  \end{equation*}  Since $ s_{ij}= s_{ji}^*$ also
\begin{equation*}
  s_{ij} \in \vL
  \big ( H^{-2}_{r}(\bX_j), L^{2}_{t}(\bX_i)\big ).
  \end{equation*}
\end{proof}

\begin{prop} \label{prop2.2}
 Suppose  \eqref{ass2}. Then
  \begin{enumerate}[\normalfont 1)]
\item  \label{item:1}$\vF= \vF_1
  \oplus \vF_2$ (i.e. $\vF_1
  + \vF_2$  is closed).

\item \label{item:2}$S^*=0$ on $\vF^\perp$ and $S^*: \vF \to \vH$ is
  a bi-continuous  isomorphism. Similarly $S: \vH \to \vF$ is
  a bi-continuous  isomorphism, and therefore
 $S S^* :  \vF \to \vF$  is invertible on $\vF$. One has
 \begin{equation}\label{eq:Pisone}
  S^*(S S^*)^{-1}S = 1 \text{ on } \vH.
\end{equation}

\item  \label{item:3}$S^*S$ is invertible on $\vH$ and
\begin{equation}\label{eq:Qispro}
 S (S^* S)^{-1} S^* = \Pi\text{ on } \vG.
\end{equation}
\end{enumerate}
\end{prop}
\begin{proof}

 \subStep{\ref{item:1}}  For $f \in \vF_1 + \vF_2$, one has for some $f_j \in L^2(\bX_j)$,
\[
f = \varphi_1\otimes f_1 + \varphi_2\otimes f_2.
\]
Since $\vF_1\cap \vF_2 =\{0\}$, this decomposition is unique. We claim that $N(f) = \|f_1\| + \|f_2\|$ defines a norm on $\vF_1 + \vF_2$ which is equivalent with the norm of $\vG$. Clearly, one has
\[
\|f\| \le  N(f).
\]
Conversely, we want to show the existence of a constant $C>0$ such that $N(f) \le C\|f\|$ on $\vF_1 + \vF_2$. If this is not true, there would be a sequence $g_n \in \vF_1 + \vF_2$ such that $\|g_n\| <1/n$ and $N(g_n) =1$.
Write $g_n$ as
\[
g_n = \varphi_1\otimes f_{1,n} + \varphi_2\otimes f_{2,n}
\]
with $\|f_{1,n}\| + \|f_{2,n}\|=1$. Without loss of generality, we can assume that $\{f_{k,n}\}$ converges weakly to $f_k$ in $L^2(\bX_k)$ as $n \to \infty$. Then, one has
\begin{align*}
f_{1,n}&= \inp{\varphi_1,g_n}_1 - s_{12}f_{2,n} =  \vO(n^{-1}) - s_{12}f_{2,n},\\
f_{2,n}&= \inp{\varphi_2,g_n}_2 - s_{21}f_{1,n} =  \vO(n^{-1}) - s_{21}f_{1,n}.
\end{align*}
Substituting the second equation into the first one, we obtain that
\[
f_{1,n}=   \vO(n^{-1}) + K_{1} f_{1,n},\;K_1= s_{12}s_{21}.
\]
 Since $f_{1,n}$ is weakly convergent and $K_1$ is compact, $K_1f_{1,n}$ is strongly convergent.
This implies that $f_{1,n}$  converges to  some $f_1$ in
$L^2(\bX_1)$. Similarly one  shows that $f_{2,n}$ converges to
some $f_2$ in $L^2(\bX_2)$. Since $g_n \to 0$, this implies
\[
 \varphi_1\otimes f_{1} + \varphi_2\otimes f_{2}=0, \quad \|f_1\| + \|f_2\|=1,
\]
which is impossible since $\vF_1\cap \vF_2 =\{0\}$. This proves the equivalence of the norms. It follows that $ \vF_1+ \vF_2$ is closed.

\subStep{\ref{item:2}} Let  $\widetilde S:\vH\to
\vF$ act as $S$, i.e. $ \widetilde S f=S f$ for  $f\in \vH$. By the equivalence
of norms shown above  we see that $\widetilde S$ is a bi-continuous
isomorphism. Clearly $S^*=0$ on $\vF^\perp$. On the other hand we can
identify $(S^*)_{|\vF}$ as the adjoint of the map $\widetilde S$, i.e. $(\widetilde
S)^*$, which is a bi-continuous
isomorphism (since $\widetilde S$ is). In particular  $S S^*$ is
invertible on $\vF$. As for  the last   part of \ref{item:2} we note  that $ P:=S^*
(S S^*)^{-1} S$ is bijective on $\vH$ and that it is also a projection, $P^2=P$.
Therefore  $P=1$, showing \eqref{eq:Pisone}.

\subStep{\ref{item:3}} The invertibility of $S^*S$ on $\vH$
follows from \ref{item:2}.  Letting  $Q = S (S^* S)^{-1} S^*$, we note
\[
Q^2 = Q, \quad Q^*=Q, \quad \ran (Q)=  \ran (S) = \vF,
\]
 showing $Q=\Pi$.

\end{proof}

Put $\Pi' = 1 - \Pi$ and $H' = \Pi' H \Pi'$ with domain $\vD (H')=\vD (H)\cap \ran (\Pi')=\Pi'\,\vD
(H)$. We  show below that indeed $\Pi'$ preserves $\vD
  (H)$ and that $H'$ is self-adjoint. Let $\Pi_j=S_jS_j^*$ and $\Pi'_j=1-\Pi_j$; $j=1,2$. We recall
 that Condition \ref{cond:geom_singl} is imposed for $a=a_1$ and
 $a=a_2$. As in   \eqref{eq:assump_singa} we  record  that
\begin{align}
    \label{eq:assump_sing}
    \Pi_j'I_{j}\Pi_j\in \vL ({\vG});\, j=1,2.
  \end{align} Here and henceforth we abbreviate $I_j(x) = I_{a_j}(x)$; $j=1,2$.

\begin{lemma} \label{Lemma:basic} Suppose
  \eqref{ass2}.
 Then
  \begin{enumerate}[\normalfont 1)]
\item  \label{item:1m} The operators
  \begin{equation}
    \label{eq:24}
   \Pi-S S^*\in \vL
  \big ( \vG, H^2(\bX)\big ).
  \end{equation}
\item  \label{item:4m} $\Pi$ preserves $\vD
  (H)$, and $\Pi'H\Pi$   and  $\Pi H\Pi'$ (initially defined on
  $\vD (H)$) extend to  bounded operators on $\vG$.
\item  \label{item:3m} $H'$ is self-adjoint.
\end{enumerate}
\end{lemma}
\begin{proof}
\subStep{\ref{item:1m}}  We shall use
 the identities
 \begin{subequations}
\begin{align}
\Pi&=S (S^* S)^{-1} S^*=S S^*+S \parbb{S^* S)^{-1} -1}S^*,\label{eq:23}\\
(S^*S)^{-1} -1&=   -\left( \begin{array}{cc}
                  0 & s_{12}\\[.1in]
                   s_{21} & 0
                \end{array}
                \right) + \left( \begin{array}{cc}
                  0 & s_{12}\\[.1in]
                   s_{21} & 0
                \end{array}
                \right)
                 (S^*S)^{-1} \left( \begin{array}{cc}
                  0 & s_{12}\\[.1in]
                   s_{21} & 0
                \end{array}
                \right), \label{eq:25}
\end{align}
 \end{subequations} which follow from Proposition  \ref{prop2.2} and the general identity
 \begin{equation*}
   (A+1)^{-1}-1=-A(A+1)^{-1}=-A+A(A+1)^{-1}A,
 \end{equation*} respectively.
The statement  \eqref{eq:24}  follows from Lemma \ref{Lemma:basic2A},
\eqref{eq:23} and \eqref{eq:25}.

\subStep{\ref{item:4m}}
Writing $S S^*=\Pi_1+\Pi_2$ it suffices (due to \ref{item:1m})  to show that $\Pi_j$ preserves $\vD
  (H)$   and that $\Pi'H\Pi_j$
   extends  to a  bounded operator on
 $\vG$ (note that boundedness of $\Pi H\Pi'$ follows from
 boundedness of $\Pi ' H\Pi$). Since $\Pi_j$ reduces $H_{a_j}$ clearly $\Pi_j$ preserves $\vD
  (H)=\vD
  (H_{a_j})$,  and since $\Pi'H\Pi_j=\Pi'\Pi'_jI_{a_j}\Pi_j\in
  \vL (\vG)$,  cf. \eqref{eq:assump_sing}, we are done.

 \subStep{\ref{item:3m}}  Take $g\in \vD\big
 ((H')^*\big)\subset \Pi'\vG$.  The functional $\Pi'\vD(H)\ni
 f\to \inp{g,H'\Pi'f}$ extends to a bounded functional, so  by
 \eqref{item:4m} also $
 \vD(H)\ni f\to \inp{g,Hf}$ extends to a bounded functional. This
 shows that $g\in \vD
 (H^*)=\vD (H)$. Whence  $g\in \vD
 (H')$.
\end{proof}

 Some parts of the previous lemma can be generalized as follows.

\begin{lemma} \label{Lemma:basic2} Suppose \eqref{ass2}. Then
  \begin{enumerate}[\normalfont 1)]
\item  \label{item:1mq} For all $r,t\in \R$
  \begin{equation*}
  \Pi-S S^*\in \vL
  \big ( L^{2}_{r}(\bX), H^{2}_{t}(\bX)\big )\cap \vL
  \big ( H^{-2}_{r}(\bX), L^{2}_{t}(\bX)\big ).
  \end{equation*}
\item  \label{item:3mq} For all $t\in \R$ there exist extensions
  \begin{equation*}
   \w{x}^{\rho+1} \Pi'H\Pi,\, \w{x}^{\rho+1} \Pi H\Pi'\in \vL
  \big (L^{2}_{t}(\bX)\big ).
  \end{equation*}
\end{enumerate}
\end{lemma}
\begin{proof}
  For \ref{item:1mq} we mimic the proof of Lemma \ref{Lemma:basic}.
  For \ref{item:3mq} it suffices (seen by using  \ref{item:1mq}) to show that
\begin{equation*}
   \w{x}^{\rho+1} \parb{\Pi'_j-\Pi_i}I_{j}\Pi_j,\, \w{x}^{\rho+1} \Pi_j I_{j}\parb{\Pi'_j-\Pi_i}\in \vL
  \big (L^{2}_{t}(\bX)\big );\, i\neq j.
\end{equation*} Next we implement  \eqref{eq:22} of Condition
\ref{cond:geom_singl}.  We obtain, cf. Lemma \ref{Lemma:basic2A} and \eqref{eq:90},
  \begin{align}\label{eq:2proj}
    \begin{split}
     \forall r,t&\in \R:\quad\Pi_iI_{j}\Pi_j,\,\Pi_j I_{j}\Pi_i\in \vL
  \big ( L^{2}_{r}(\bX), L^{2}_{t}(\bX)\big )\, (i\neq j),\\
 &\w{x}^{\rho+1} \Pi'_jI_{j}\Pi_j,\, \w{x}^{\rho+1} \Pi_j I_{j}\Pi'_j\in \vL
  \big (L^{2}_{t}(\bX)\big ).
    \end{split}
  \end{align}
\end{proof}

\subsection{ $\vF_1\cap\vF_2=\{0\}$;  the case $\lambda_0=\Sigma_2$   and
  $\lambda_0\notin\sigma_{\pp}(H')$}\label{subsec:The case bottom}
 The case  $\lambda_0=\Sigma_2$ is simpler and there are better mapping properties of
 various operators compared to the case
 $\lambda_0>\Sigma_2$. Consequently we pay special attention to the
 former case.  We shall show the following  extension of \cite[Lemma 2.1]{Wa2}.
\begin{lemma}
\label{lemma2.1}
Suppose Condition \ref{cond:uniqdd} with  $\lambda_0=\Sigma_2$ and
\eqref{ass2}.  Then there exists $\epsilon_0 >0$ such that the essential spectrum of
$H'=\Pi' H\Pi' $ satisfies
\begin{equation}
  \label{eq:29}
  \sigma_\mathrm {ess }(H') \subset
 [ \lambda_0 + \epsilon_0, \infty ).
\end{equation}
\end{lemma}

\begin{proof} Note
\begin{equation} \label{c2}
\lambda_0 =\Sigma_2< \inf
\sigma({H^b}), \text{ for }b\not \in \{a_1, a_2, a_{\max} \}.
\end{equation}
In particular we have
\begin{equation*}
  \lambda_0 = \min \sigma_{\d}(H^{a_j}).
\end{equation*}

\begin{subequations}

We introduce a
 auxiliary Hamiltonians (with inter-cluster  momenta $p_j=p_{a_j}$)
\begin{align}\label{eq:27}
  \begin{split}
 &\breve H=H'+ p_1^2\Pi_1+p_2^2\Pi_2,\\
\;&\widetilde
H=H-\lambda_0\parb{\Pi_1+\Pi_2};\\\,& \vD (\breve H)=\vD (\widetilde  H)=\vD ( H).
  \end{split}
\end{align} These  operators  differ by  $H$-compact  terms:
\begin{align}
  \label{eq:28}\breve H&=\widetilde
H -K_1+K_2;\\
K_1&=\Pi H\Pi'+\Pi' H\Pi+\Pi_1 I_{1}\Pi_1 +\Pi_2I_{2}\Pi_2,\nonumber \\
K_2&=\Pi_1 H\Pi_1+\Pi_2 H\Pi_2-\Pi H\Pi.\nonumber
\end{align}
 \end{subequations}
 Note that  the first and second terms of  $K_1$ are   $H$-compact due to Lemma
\ref{Lemma:basic2}. The relative compactness of the third and fourth
terms of  $K_1$ follows  from
rewriting $\Pi_jI_{j}\Pi_j=\varphi_j \otimes S_j^*I_jS_j\bra{\varphi_j}$
and then invoking the complete analogue of \eqref{eq:22asymB}. As for $K_2$ we can also use Lemma
\ref{Lemma:basic2} first by replacing the factors of $\Pi$ in the
third term  by $SS^*=\Pi_1+\Pi_2$ and then expanding into four
terms. We are left with considering the sum of cross terms
\begin{equation*}
 \Pi_1 H\Pi_2+\Pi_2 H\Pi_1.
\end{equation*} By writing
\begin{align}\label{eq:procP}
  \begin{split}
   \Pi_1 H\Pi_2&=\Pi_1 I_{1}\Pi_2+H_{a_1}\Pi_1\Pi_2,\\
\Pi_1\Pi_2&=\varphi_1\otimes s_{12}\langle \varphi_2|,
  \end{split}
\end{align} and  then using \eqref{eq:2proj} and Lemma
\ref{Lemma:basic2A},
we conclude that the term $\Pi_1 H\Pi_2$ is
$H$-compact. We
can argue similarly for the term $\Pi_2 H\Pi_1$ and then conclude that
  also $K_2$ is
$H$-compact.

Next, it follows from the very definition \eqref{eq:27} that $\vF$ reduces $\breve
H$. Whence $\sigma_\textrm {ess }(H') \subset \sigma_\textrm {ess }(\breve
H)$,
and consequently it suffices to show \eqref{eq:29} with    $H'$ replaced
by $\breve
H$
and hence in turn  with    $H'$  replaced
by $
\widetilde  H$ (by the compactness shown above).

Consider  a family of smooth non-negative  functions $ \{j_b| b\in \vA, \# b =2 \} $ on $\bX$
obeying that for some $c>0$:
\begin{subequations}
\begin{align}
\sum_{\# b =2 } j_b (x)^2 &= 1,   \label{eq:33}\\
\label{eq:34}|x^a|j_b (x)&\geq c |x|j_b (x) \text{ for }
a\not\subset b \mand |x|\geq 1,\\
|\partial^\alpha j_b(x) | &\le C_\alpha \w{x}^{-|\alpha|}\text{ for all } \alpha \in \N_0^{\dim \bX}.\label{eq:36}
\end{align}
\end{subequations}
We have, using this family of functions,
\begin{equation}
H =\sum_{\#a=2}j_a  H_a j_a+\sum_{\#a=2} I_{a} j_a^2
 - \sum_{\#a=2}|\nabla j_a|^2=\sum_{\#a=2}j_a  H_a j_a+ K,\label{eq:32}
\end{equation} where $K$ is $H$-compact.

For $a \neq a_j$,  $j=1,2$,  it follows from  \eqref{c2} that  $\inf \sigma( H_a)  = \inf \sigma (H^a)
\ge \lambda_0 + \epsilon_0'$ for some $\epsilon_0'>0$.
Therefore,
\begin{equation}
\label{es1}
\sum_{\#a=2, a\ne a_j, j = 1,2}j_a  H_a j_a
\ge
(\lambda_0 + \epsilon_0') \sum_{\#a=2, a\ne a_j, j=1,2} j_a ^2.
\end{equation}

Since $\lambda_0$ is the lowest eigenvalue of $H^{a_j}$,  $j=1,2$, there exists  $\epsilon_1>0$ such that
\[
\Pi_j'H_{a_j} \Pi_j' \ge  (\lambda_0 + \epsilon_1) \Pi'_j.
\]
Hence for $a = a_j$,  $j=1,2$,
we obtain (assuming in the last step that $\lambda_0 + \epsilon_1\leq 0$ and
using that $p_j^2\geq 0$)
\begin{align}
\label{es2}
j_{a_j}  H_{a_j} j_{a_j}
 & \ge  (\lambda_0 + \epsilon_1) j_{a_j}  \Pi'_j  j_{a_j} + j_{a_j}  H_{a_j}\Pi_j j_{a_j}  \nonumber \\
& \ge   (\lambda_0 + \epsilon_1) j_{a_j}^2   +\lambda_0j_{a_j}  \Pi_j j_{a_j}.
\end{align}
It is a  consequence of  \eqref{eq:33} and \eqref{eq:34}  that
the operator $\Pi_j j_{a_j}^2-  \Pi_j $ is $H$-compact. Using \eqref{eq:36} it follows by a Taylor
expansion that $[j_{a_j},  \Pi_j]\inp{x}$ is bounded (cf.
  \eqref{eq:hminus}). We conclude that also
\begin{equation}\label{eq:31}
   K'_j:=\lambda_0 \parb{j_{a_j}  \Pi_j j_{a_j}-  \Pi_j} \text { is }H\text{-compact}.
\end{equation}  We write \eqref{es2} as
\begin{equation}\label{eq:31z}
  j_{a_j}  H_{a_j} j_{a_j}  -\lambda_0\Pi_j \geq (\lambda_0 + \epsilon_1)
  j_{a_j}^2 + K'_j,\;j=1,2.
\end{equation}

Let $\epsilon_0 =\min\{ \epsilon_0', \epsilon_1\}$.
We then  deduce  from \eqref{eq:33}, \eqref{eq:32}, \eqref{es1} and \eqref{eq:31z}
that
\begin{equation}\label{eq:30}
 \widetilde H \ge  (\lambda_0 + \epsilon_0)    + \widetilde K,
\end{equation}
where $\widetilde K = K +  K'_1+  K'_2$.

Since  $\widetilde K$ is  $\widetilde H$-compact it follows from \eqref{eq:30} and Weyl's theorem   \cite[Theorem
XIII.14]{RS}  that
$$
\sigma_{\ess}(\widetilde H)
\subset [ \lambda_0+\epsilon_0, \infty).
$$
  \end{proof}

  \begin{remark}\label{remark:Sigma3vf_1c-case-lambd} By appropriately enlarging the projection $\Pi$ as
    to include the span of all threshold eigenstates corresponding
    to thresholds $\lambda^a\in [\Sigma_2,\Sigma_3-\epsilon_0)$ for
    any given $ \epsilon_0>0$, one can make sure that $ \sigma_{\ess}(H')
    \subset [\Sigma_3-\epsilon_0,\infty)$. Here $H'=\Pi'H\Pi'$ as
    before, and the proof is essentially the same as the one of Lemma
    \ref{lemma2.1}. This trick of `subtracting' all low-energy
    $2$-cluster channels is limited to energies below $\Sigma_3$. We
    prefer for simplicity of presentation to treat two-cluster
    thresholds $\lambda_0> \Sigma_2$ on an equal footing not
    distinguishing between the cases $\lambda_0\in (\Sigma_2,\Sigma_3)$
    and $\lambda_0>\Sigma_3$. In fact our results would be the same
    anyway for the two cases.
\end{remark}

Assume from this point that
\begin{equation}
  \label{eq:20}
  \lambda_0\text{ is not an eigenvalue of }H'.
\end{equation}
By Lemma \ref{lemma2.1}  we can then chose $\delta >0$
such that $\sigma(H') \cap \{ |z-\lambda_0| \leq \delta\} = \emptyset$.
Due to  Proposition \ref{prop:reduc_form_abstract}
\begin{equation} \label{rep3ip}
R(z)= E(z)-  E_{+}(z)  E_{\vH}(z)^{-1}  E_{-}(z),
\end{equation} where the quantities \eqref{eq:5pm}--\eqref{eq:8pm} in
this case are analytic in $\{ |z-\lambda_0| <\delta\} $.

Let us analyse the form of $ E_{\vH}(z) $. For $f \in \vF$, $f
=\varphi_1 \otimes f_1 + \varphi_2 \otimes f_2$ and we can write
\[
H f = \lambda_0 f + ( p_1^2 + I_1) (\varphi_1 \otimes f_1) + ( p_2^2 + I_2)(\varphi_2 \otimes f_2),
\]
recalling  $p_j = p_{a_j}$ and $I_j = I_{a_j}$, $j=1,2$.
Therefore, for $i,j,k\in \{1,2\}$
\begin{subequations}
\begin{align}
S_k^* (H-z) S_k & =   p_k^2 +\lambda_0-z+ W_k(x_k),\\
S_i^* (H-z)  S_j &=  \check s_{ij}(z)+ W_{ij},\\
S_i^*  H R'(z)H S_j & =  -K_{ij}(z),
\end{align}
\end{subequations}
where
\begin{subequations}
\begin{align} \check s_{ij}(z)&=  \inp{\varphi_i, \varphi_j \otimes (p_j^2+\lambda_0-z)\cdot}_i,\\
W_{ij}&= \inp{\varphi_i,I_j\varphi_j\otimes \cdot}_i,  \\
W_k (\cdot) =W_{kk}&=  \inp{\varphi_k,I_k\varphi_k}_k(\cdot),\\
K_{ij}(z)&=  -\inp{\varphi_i,I_i R'(z) I_j (\varphi_j \otimes \cdot)}_i.
\end{align}
\end{subequations}

 Introducing $\vH^2=H^2(\bX_1) \oplus H^2(\bX_2)$ this yields  the
 following expression for the operator $E_{\vH}(z) : \vH^2\to   \vH$.
\begin{align}\label{eq:13}
E_{\vH}(z)  ={}  &z- \lambda_0 \nonumber\\&-  \left( \begin{array}{cc}
                   - \Delta_{x_1} + W_1(x_1) + K_{11}(z) &  \check  s_{12}(z)+W_{12}+K_{12}(z) \\
                    \check  s_{21}(z)+ W_{21}+K_{21}(z) &   - \Delta_{x_2} + W_2(x_2) + K_{22}(z)
                \end{array}\right).
\end{align}

Clearly $ \check s_{ij}(z)=s_{ij}(\lambda_0-z+p_j^2)$ for
$i\neq j$, and therefore   Lemma \ref{Lemma:basic2A} yields that $ \check s_{ij}(z)$ for
$i\neq j$ is bounded
and in fact polynomially
decreasing (uniformly  in $z$ near $\lambda_0$).  Here and henceforth an operator $b_{ij}: L^2(\bX_j) \to L^2(\bX_i)$ is
said to be  \emph{polynomially decreasing}  if for all $r,t\in \R $
\begin{equation*}
b_{ij} \in\vC\parb{ H^{2}_{r}(\bX_j), L^{2}_{t}(\bX_i)}\cap \vC\parb{ L^{2}_{r}(\bX_j), H^{-2}_{t}(\bX_i)}.
\end{equation*} We note  that also the operators  $W_{12}$ and  $W_{21}$  are polynomially
decreasing. An operator $B$ on $\vH$ is
said to be  polynomially decreasing if its  entries $b_{ij}: L^2(\bX_j) \to L^2(\bX_i)$ are polynomially
decreasing.

We claim that  also $K_{12}(z)$ and $K_{21}(z)$ are polynomially
decreasing, in fact uniformly  in $z$ near $\lambda_0$. To see this it suffices (by symmetry) to  consider
$K_{12}(z)$, and it suffices to show that $K_{12}(z)\in\vL\parb{
  L^{2}_{-s}(\bX_2), L^{2}_{s}(\bX_1)}$ for any $s\geq 0$.  So let us
fix $s\geq 0$. A small consideration using the argument for
\eqref{eq:assump_singa} and the polynomial decay of the cluster
 bound states shows,  that it suffices to check that
\begin{align}
  \label{eq:enough} \inp{ x^{1}}^{-2s}  \inp{ x}^sR'(z)\inp{ x}^s\inp{
   x^2}^{-2s} \text{ is
   uniformly bounded near }\lambda_0.
\end{align}
Recall that $\{z| |z-\lambda_0| \leq \delta\}$ is included in the resolvent
set of $H'$  for a small $\delta>0$,  cf.   \eqref{eq:20}.  We also
record  the
following
elementary estimates,
\begin{subequations}
  \begin{align}\label{eq:elementEst}
  \inp{x+y}^t&\leq 2^{|t|/2}\inp{x}^{|t|}\inp{y}^t, \\
\inp{x}&\leq C(\inp{x^1}+\inp{x^2})\leq 2C\inp{x^1}\inp{x^2}. \label{eq:elementEst2}
\end{align}
\end{subequations}
In turn, by interpolation and by using \eqref{eq:elementEst2}, the
assertion  \eqref{eq:enough} is a consequence of
\begin{align}\label{eq:91}
 \inp{ x^{2}}^t  R'(z)\inp{
   x^2}^{-t} \mand  \inp{ x^1}^{-t} R'(z)\inp{ x^{1}}^t
 \text{ are
   uniformly bounded};\quad t=2s.
\end{align} We bound  the first expression only.
 Representing, using notation of the proof of Lemma \ref{lemma2.1},
\begin{align}\label{eq:92}
  R'(z)=\breve R(z)\Pi',\,\,\breve R(z)=(\breve H-z)^{-1},
\end{align} we first
 bound $\inp{ x^{2}}^t  \breve R(z)\inp{
   x^2}^{-t}$.
 Whence we want  to bound $\inp{ \kappa x^{2}}^t  \breve R(z) \inp{
   \kappa x^2}^{-t}$  with
$\kappa=1$,
 for which it suffices to   bound this quantity for any small
$\kappa>0$. We show such bound uniformly  in $z$ near
$\lambda_0$. As in the proof of Lemma \ref{lemma2.1}
\begin{align}\label{eq:93}
  \breve H=H-\lambda_0(\Pi_1+\Pi_2)-K_1+K_2.
\end{align} Letting $\breve H_{\kappa,t} =\inp{\kappa x^{2}}^t  \breve H \inp{
   \kappa x^2}^{-t}$ we obtain from \eqref{eq:93} that
\[
\breve H_{\kappa,t} = \breve H + \vO(\kappa) (\breve H-\i).
\]  For example, seen by using  Taylor
expansion and \eqref{eq:elementEst},
\begin{align}\label{eq:comkap}
  \begin{split}
 \inp{\kappa x^{2}}^t \Pi_1 \inp{\kappa x^2}^{-t}- \Pi_1
  &=[\inp{\kappa \pi^2x}^t -\inp{\kappa \pi^2x_{1}}^t,\Pi_1
  ]\inp{\kappa x^2}^{-t}\\&=\kappa \vO\parb{\inp{\kappa\pi^2
      x_1}^{t-1}}\inp{\kappa x_1}^{-t}\\&=\vO(\kappa).
  \end{split}
\end{align} We can treat the terms $K_1$ and $K_2$ similarly.
Therefore
\begin{align*}
  \parb{\breve H_{\kappa,t}-z}\breve R(z)=1 + \vO(\kappa) (\breve H-\i)\breve R(z)
\end{align*}
  is invertible for
$|z- \lambda_0| \leq \delta$ and for small $\kappa>0$. This shows that $
(\breve H_{\kappa,t} -z)^{-1}$ is uniformly bounded accomplishing our
first goal. The second goal, cf. \eqref{eq:92}, is to
bound $\inp{ x^{2}}^t \Pi'\inp{ x^2}^{-t}$, but indeed $\inp{ x^{2}}^t
\Pi\inp{ x^2}^{-t}$  is bounded due to Lemma \ref{Lemma:basic2}
\ref{item:1mq} and \eqref{eq:comkap}.  Consequently we have justified
\eqref{eq:91}  and therefore shown that
$K_{12}(z) $ is polynomially
decreasing uniformly  near $\lambda_0$.

One can somewhat similarly show that
\begin{equation}\label{eq:35}
\w{x_j}^{\rho_1} K_{jj}(z) \w{x_j}^{\rho_2} \text{ is uniformly bounded for on }L^2(\bX_j) \text{ for }\rho_1 + \rho_2 \le 2
\rho+2.
\end{equation}
In fact we can use a  refinement of \eqref{eq:assump_singa} related to
\eqref{eq:90} and Lemma \ref{Lemma:basic2} \ref{item:3mq}. Note here  the
Taylor expansion
$I^{(1)}_j(x) =I^{(1)}_j(x_j)+\vO(\inp{x^j}^{\rho+2})\inp{x_j}^{-\rho-1}$,
cf. \eqref{eq:elementEst}, which in turn leads to the following bounds  for
all $s\in\R$ and for $\rho'\in\set{\rho_1,\rho_2}$,
\begin{align*}
  \Pi_j \w{x_j}^{\rho'}I_{j}\Pi',\quad\Pi'
  I_{j}\w{x_j}^{\rho'}\Pi_j\in
  \vL\parb{L^{2}_{s}(\bX),L^{2}_{\rho-\rho'+1+s}(\bX)}.
\end{align*}
 We deduce \eqref{eq:35} by combining these bounds with  the following consequence of
\eqref{eq:elementEst}
and \eqref{eq:91},
\begin{align}\label{eq:91BB}
 \inp{ x}^s  \breve R(z)\Pi'\inp{
   x}^{-s} \text{ is bounded for all }s\in \R \,\,(\text{uniformly near }\lambda_0).
\end{align}

Due to the above discussion and     \eqref{eq:22asymB} we finally
obtain a simplified version of \eqref{eq:13}:

\begin{prop}\label{prop2.3}  Suppose Condition \ref{cond:uniqdd} with
  $\lambda_0=\Sigma_2$,
\eqref{ass2} and \eqref{eq:20}. As a bounded  operator $ E_{\vH}(z) :
  \vH^2 \to  \vH $  one then has
\begin{equation}\label{eq:efbNd}
  E_{\vH}(z)  \equiv z- \lambda_0 -   \left( \begin{array}{c !{\kern-7mm} c}
      - \Delta_{x_1} + W_1(x_1) + K_{11}(z) & 0 \\[2mm]
      0 &   - \Delta_{x_2} + W_2(x_2) + K_{22}(z)
                \end{array}\right),
\end{equation}
where ``$\equiv$'' means the equality modulo a polynomially decreasing
term which  depends holomorphically on $z$ near $\lambda_0$. We have
\begin{subequations}
\begin{align}\label{eq:22asymBC}
W_j(x_j) -  I^{(1)}_{j}( x_j) &\in
\vL\parb{H^{2}_{s}(\bX_j),L^{2}_{\rho+1+s}(\bX_j)}\mfor s\in\R,\\
W_j(x_j) -  I^{(1)}_{j}( x_j) &\in
\vC\parb{H^{2}_{s}(\bX_j),L^{2}_{t}(\bX_j)}\mfor s\in\R,\,t<\rho+1+s,\label{eq:22asymBC2}\\
K_{jj}(z) &\in \vL\parb{L^{2}_{s}(\bX_j),L^{2}_{2\rho+2+s}(\bX_j)}\mfor s\in\R.\label{eq:second_bnd}
\end{align}
 \end{subequations} Moreover  the operator $K_{jj}(z)$ depends holomorphically on $z$,
 and the potential $W_j$  is $p_j^2$-compact with  the singularities of
 located in a bounded set.
\end{prop}
\begin{remark}\label{remark:formbnd} We shall prefer a version of
  Proposition \ref{prop2.3}  based on forms rather than
  operators, although this is not essential. Thus we  consider $
  E_{\vH}(z) $ as an operator $ E_{\vH}(z) :
  H^1(\bX_1) \oplus H^1(\bX_2) \to   H^{-1}(\bX_1) \oplus
  H^{-1}(\bX_2)$.
 Now an operator $b_{ij}: L^2(\bX_j) \to L^2(\bX_i)$ is
said to be  \emph{polynomially decreasing}  if for all $r,t\in \R $
\begin{equation*}
b_{ij} \in\vC\parb{ H^{1}_{r}(\bX_j), H^{-1}_{t}(\bX_i)}.
\end{equation*} Also we note that the expansion \eqref{eq:22asymBC}
 should be  given the interpretation  of being in the space
$\vL\parb{H^{1}_{s}(\bX_j),H^{-1}_{\rho+1+s}(\bX_j)}$,  $s\in\R$, rather than in the
stated space (and similarly for \eqref{eq:22asymBC2}).  Note for comparison that
$I_{j}^{(1)}(x_j)\in\vL\parb{L^2_{s}(\bX_j),L^2_{\rho+s}(\bX_j)}$,
$s\in\R$. With these modifications \eqref{eq:efbNd} still holds.
 Moreover it  is easy to show by  a resolvent equation that the null space  $\ker
E_{\vH}(\lambda_0) $ is independent on whether the operator or the
form interpretation of $E_{\vH}(\lambda_0) $ is used.
  \end{remark}

\section{The case  $\lambda_0\in\sigma_{\pp}(H')$}\label{sec:The
  case where  lambda0insigma}
We discuss briefly the modifications needed in the previous two
sections to treat  the case $\lambda_0\in\sigma_{\pp}(H')$.

\subsection{ $\lambda_0\in\sigma_{\pp}(H')$; non-multiple  case}\label{subsec:The case
  lambda0insigma, non-degenerate case} For simplicity we assume
$m=1$ in the setting of Section \ref{Reduction near a simple
  two-cluster threshold}. The corresponding eigenfunction is denoted
by $\varphi$ (rather than  $\varphi_1$), and we shall use the notation $\langle \varphi|\cdot\rangle_0=\langle \varphi,\cdot\rangle_0=\langle \varphi,\cdot\rangle_{L^2(\bX^{a_0})}$. We follow \cite{Wa2} assuming for simplicity
that $\lambda_0$ is a simple eigenvalue of $H'$ and   introduce a
corresponding normalized eigenfunction $\psi$,
$(H'-\lambda_0)\psi=0$. Let $I_0=I_{a_0}$.

Now
$\vH=L^2(\bX_{a_0})\oplus \C$, and $S:\vH\to \vG=L^2(\bX)$ is given by
\begin{align*}
  S(f,c)=\varphi\otimes f+c \psi.
\end{align*} We let the orthogonal
projection $SS^*$ onto the range of this  new $S$  be  denoted by
$1-\Pi''$ (for the  $S$ in Section \ref{Reduction near a simple
  two-cluster threshold} the projection  is denoted
by  $1-\Pi'$), and we introduce
$H''=\Pi''H\Pi''$ on the Hilbert space $\Pi''\vG$. By construction
$\lambda_0$ is not an eigenvalue of $H''$. Using
\eqref{eq:5pm}--\eqref{eq:8pm} as before we obtain \eqref{rep1} now with
\begin{align*}
  -E_{\vH}(z) &=  -S^* \big (z -  H +  H R''(z)H \big ) S\\
 &= \left( \begin{array}{cc}
 \parbb{\lambda_0-z+p^2_{a_0}+ \langle \varphi|{ I_{0}}|\varphi\otimes
 \cdot\rangle_0 -
\langle \varphi| I_{0}R''(z)I_{0} |\varphi\otimes
 \cdot\rangle_0} &  \langle \varphi, I_{0}\psi\rangle_0 \\[.1in]
 \langle I_{0}\psi,\varphi\otimes
 \cdot\rangle  & \lambda_0-z
\end{array} \right).
\end{align*}
Let
\begin{subequations}
\begin{align}
  \label{eq:38pt}
  \begin{split}
 \widetilde H&=H-\lambda_0\Pi,\quad \Pi=\varphi\otimes\langle \varphi|\cdot\rangle_0,\\
 \breve H&=\widetilde{H}-\breve K;\\
&\quad \breve K=\Pi I_{0}+
  I_{0}\Pi -\Pi I_{0}\Pi+\lambda_0|\psi\rangle\langle \psi|.
  \end{split}
\end{align} Note that $\breve K$ is  $H$-compact. This
construction is motivated by the following direct sum decomposition.
\begin{equation}
  \label{eq:73t}
  \breve H= H''+p^2_{a_0}\Pi;\;H''=\Pi'' H\Pi''.
\end{equation}
 The basic structure of  resolvents is  given as follows.
  \begin{align}\label{eq:basresbreve0}
  \begin{split}
  R''(z)&= \breve R(z)-(p_{a_0}^2-z)^{-1}\Pi-z^{-1}|\psi\rangle\langle \psi|,\\
 R''(z)&=\Pi''\breve R(z)\Pi''=\Pi''\breve R(z)=\breve R(z)\Pi'';
\quad\breve R(z)=(\breve H-z)^{-1}.
  \end{split}
\end{align}

We shall prove that $\psi\in H^2_\infty(\bX)$ (see Theorem
\ref{thm:priori-decay-b_0}), which in turn implies `good' properties
of $\breve R(z)$ and $\Pi''$ near $\lambda_0$ (see Remark \ref{remark:The case
  lambda0insigma}).

\end{subequations}
\subsection{ $\lambda_0\in\sigma_{\pp}(H')$; multiple case}\label{subsec:The case
  lambda0insigma, degenerate case}
We adapt the setting of Section \ref{sec:Reduction near a multiple
  two-cluster threshold} in the case \eqref{eq:20}  of Subsection
\ref{subsec:The case bottom} is not
fulfilled and without imposing the condition
$\lambda_0=\Sigma_2$ as in that subsection. In particular \eqref{ass2} is  satisfied.

Assume  for simplicity
that $\lambda_0$ is a simple eigenvalue of $H'$. We introduce a
corresponding normalized eigenfunction $(H'-\lambda_0)\psi=0$. Now
$\vH=L^2(\bX_{1})\oplus L^2(\bX_{2})\oplus \C$, and $S:\vH\to \vG=L^2(\bX)$ is given by
\begin{align*}
  S(f_1,f_2,c)=\varphi_1\otimes f_1+\varphi_2\otimes f_2+c \psi.
\end{align*}

We introduce $\Pi''$  in terms of this $S$ as in
the previous subsection and let again   $H''=\Pi''H\Pi''$. Note that
$\Pi''=1-\Pi-|\psi\rangle\langle \psi|=\Pi'-|\psi\rangle\langle \psi|$
where $\Pi$ is given as in Section \ref{sec:Reduction near a multiple
  two-cluster threshold}. We   obtain \eqref{rep1} now with
\begin{align*}
  E_{\vH}(z) =  S^* \big (z -  H +  H R''(z)H \big ) S),
\end{align*} and this operator has a similar representation as in the previous
subsection, now by a $3\times 3$-block  representation
$(e_{ij})_{i,j\leq 3}$ rather than a
$2\times 2$-block representation as given there. Here
$(e_{ij})_{i,j\leq 2}$ is given as in Proposition \ref{prop2.3} (with
$R'(z)$ replaced by $R''(z)$), $e_{33}=z-\lambda_0$, $e_{i3}=-\langle
\varphi_i, I_{i}\psi\rangle_i $ and $e_{3i}=e_{i3}^*=\langle e_{3i}|$; $i=1,2$.
\begin{subequations} The analogue of \eqref{eq:27}  reads
\begin{align}\label{eq:27p}
  \begin{split}
 &\breve H=H''+ p_1^2\Pi_1+p_2^2\Pi_2,\\
\;&\widetilde
H=H-\lambda_0\parb{\Pi_1+\Pi_2};\\\,& \vD (\breve H)=\vD (\widetilde  H)=\vD ( H).
  \end{split}
\end{align} These  operators  differ by  $H$-compact  terms, cf.   \eqref{eq:28}:
\begin{align}
  \label{eq:28p}\breve H&=\widetilde
H -K_1+K_2;\\
K_1&=\Pi H\Pi'+\Pi' H\Pi+\Pi_1 I_{1}\Pi_1 +\Pi_2I_{2}\Pi_2+\lambda_0|\psi\rangle\langle \psi|,\nonumber \\
K_2&=\Pi_1 H\Pi_1+\Pi_2 H\Pi_2-\Pi H\Pi.\nonumber
\end{align}

 The basic structure of  resolvents is  given as follows.
\begin{align}\label{eq:basresbreveo}
  \begin{split}
  R''(z)&= \breve R(z)-(p_1^2\Pi_1+p_2^2\Pi_2-z)^{-1}\Pi-z^{-1}|\psi\rangle\langle \psi|,\\
 R''(z)&=\Pi''\breve R(z)\Pi''=\Pi''\breve R(z)=\breve R(z)\Pi'';
\quad\breve R(z)=(\breve H-z)^{-1}.
  \end{split}
\end{align}
\end{subequations}

Again we have good decay properties of $\psi$, cf. Theorem
\ref{thm:priori-decay-b_0} and Remark
\ref{remark:microlocal-boundsGEN}, which in turn  yields good
properties of $\breve R(z)$ near $\lambda_0$, cf. Remark \ref{remark:The case
  lambda0insigma}.

\section{Multiple two-cluster case, \,$\vF_1\cap\vF_2\neq\{0\}$}\label{sec:The case when the condition {ass2}}

It may happen that the condition \eqref{ass2} is not satisfied. We
start out by  presenting an
example (the only one we know). It is  given by  an atomic type $2$-body Schr\"odinger
operator with a third particle of infinite mass fixed at the origin
(fitting into the framework of  $N$-body case of Section \ref{$N$-body Schr\"odinger operators with infinite mass
  nuclei}):
\begin{equation}
H = \sum_{j=1}^2 ( - \Delta_{x^j} + V_j(x^j)) + V_{12}(x^1-x^2),
\end{equation}
where $x^j \in \R^n$.  Here $x=(x^1, x^2)$ is used as global coordinates on $\R^{2n}$.
Let $\lambda_0$ be the lowest threshold of $H$. Assume that this threshold is double and is attained by the lowest eigenvalue of $H_j$
\[
H_j =   - \Delta_{x^j} + V_j(x^j), \quad j =1,2.
\]
In this case, the condition \eqref{ass2} is not satisfied. In fact, let $\varphi_j(x^j)$ be the eigenfunction of $H_j$ associated with $\lambda_0$.
Then   clearly one has
\[ \vF_1 \cap \vF_2 = \mbox{span} \{\varphi_{1}(x^1)\varphi_{2}(x^2) \}.
\]

\fancybreak{}

\subsection{$\vF_1\cap\vF_2\neq\{0\}$; a general
  approach}\label{sec:The case when the condition {ass2}, second
  approach}
We discuss a   method which is easy to generalize to the case of
an arbitrary multiplicity of the two-cluster threshold $\lambda_0$.

We define
\begin{align}
  \label{eq:Hmod}
  \vH=\set[\big]{f\in L^2(\bX_1)\oplus  L^2(\bX_2)\mid  f\perp \ker {S^*S}},
\end{align} where the components of this $S=(S_i)$ are given by
\eqref{eq:19}. Due to Lemma \ref{Lemma:basic2A} the space
$\ker {S^*S}$ is finite dimensional consisting of vectors with
components in $H^2_\infty(\bX_j)$. By assumption
$\dim(\ker {S^*S})\geq1 $. Clearly $S:\vH\to
\vF=\vF_1+\vF_2\subset \vG=L^2(\bX)$ is a continuous
isomorphism. Arguing by contradiction as in the proof of Proposition
\ref{prop2.2} we deduce that in fact $S:\vH\to \vF$ is
bi-continuous. In particular $\vF$ is closed in $\vG$, and we let
correspondingly $\Pi$, $\Pi'$, $H'$ and $R'(z)$ be given (as in
Section \ref{An abstract reduction scheme}). If we consider $S$ as a
map from the bigger space $ L^2(\bX_1)\oplus L^2(\bX_2)$ (as in
\eqref{eq:Hmod}) the notation $S^*g$, $g\in \vG$, may seem
ambiguous. However this is in fact not the case since then $S^*g\perp
\ker {S^*S}$, so $S^*g\in \vH$. For this reason the conclusion of
Lemmas \ref{Lemma:basic} and \ref{Lemma:basic2} are still valid and
the formula \eqref{eq:13} applies again.

As before we would like to use \eqref{eq:13} in a neighbourhood of
$\lambda_0$. Let us here assume that $\lambda_0=\Sigma_2$  and
\eqref{eq:20}. Then of course we  can
let $z=\lambda_0$ in \eqref{eq:13} and obtain formulas as in
Proposition \ref{prop2.3} and Remark \ref{remark:formbnd}. We  consider
correspondingly
the operator $E_{\vH}(\lambda_0)$ either as an operator mapping
$ \parb{H^2(\bX_1)\oplus  H^2(\bX_2)}\cap \vH\to \vH$ or as an operator  mapping
\begin{align*}
  \parb{H^1(\bX_1)\oplus  H^1(\bX_2)}\cap \vH\to \set{f\in
  H^{-1}(\bX_1)\oplus  H^{-1}(\bX_2)\mid  f\perp \ker {S^*S}}.
\end{align*}

If
 $\lambda_0\in\sigma_{\pp}(H')$ we need to modify the
construction of $S$ and $E_{\vH}(\lambda_0)$. This is doable along the lines of Subsection
\ref{subsec:The case lambda0insigma, degenerate case}.
 Finally letting $P_{0}$ denote the orthogonal projection onto
$\ker S^*S$ in $L^2(\bX_1)\oplus L^2(\bX_2)$ it is convenient
to study $E_{\vH}(\lambda_0)+P_0$ on $L^2(\bX_1)\oplus L^2(\bX_2)$
(rather than $E_{\vH}(\lambda_0)$ on $\vH$). This is because we have
a `good
parametrix'  of $\diag(h_1,h_2)$ on
$L^2(\bX_1)\oplus L^2(\bX_2)$.

In  Section \ref{sec:CoulRellich} we
shall see what the above  ideas amount  to in
 the  setting of the  models of physics introduced in Chapter
 \ref{Notation}. This will be a general treatment not assuming $\lambda_0=\Sigma_2$.


 \chapter{Spectral analysis of $H'$ near $\lambda_0$}\label{Spectral analysis of H' near E_0}

 In the bulk of this chapter we impose Conditions \ref{cond:smooth}
 and \ref{cond:uniq}. We shall prove various $N$-body resolvent
 estimates  for the operator $H'$
 appearing in the Grushin method (or more prescisely for the operator
 $\breve H$ defined below). The analysis overlaps \cite{AIIS}, in
 particular it is based on an appropriate Mourre estimate. Sharing the
 spirit of \cite{AIIS} our procedure avoids in any other sense the
 `classical Mourre theory' 
 \cite{Mo, Je, JMP, Wa1}. The multiple two-cluster case can be treated in
 a similar way,  although  it is notationally more complicated, see
 Remarks \ref{remark:microlocal-boundsGEN}--\ref{remark:excase}.

 Recall from Section \ref{Reduction near a simple two-cluster
   threshold} the notation $I_{0}=I_{a_0}$ and $p_{0}=p_{a_0}$. We introduce the following modifications of $H$ (note the
 similarity with \eqref{eq:27} and \eqref{eq:28}),
\begin{subequations}
\begin{align}
  \label{eq:38p}
  \begin{split}
 \widetilde H&=H-\lambda_0\Pi,\\
 \breve H&=\widetilde{H}-\breve K;\\
&\quad \breve K=\Pi I_{0}+
  I_{0}\Pi -\Pi I_{0}\Pi.
  \end{split}
\end{align} Note that $\breve K$ is  $H$-compact. This
construction is motivated by the following direct sum decomposition.
\begin{equation}
  \label{eq:73}
  \breve H= H'+p^2_{0}\Pi;\;H'=\Pi' H\Pi'.
\end{equation}
 The basic structure for resolvents is  given as follows.

\end{subequations}

\begin{subequations}
\begin{align}\label{eq:basresbreve00}
  \begin{split}
  R'(z)&= \breve R(z)-(p_{0}^2-z)^{-1}\Pi,\\
 R'(z)&=\Pi'\breve R(z)\Pi'=\Pi'\breve R(z)=\breve R(z)\Pi';
\quad\breve R(z)=(\breve H-z)^{-1}.
  \end{split}
\end{align}

Note that \eqref{eq:basresbreve00} yields `good estimates' of $R'(z)$
near $\lambda_0$ provided we can show `good estimates' of $\breve
R(z)$ near $\lambda_0$. The goal of this chapter is to prove the
latter, which more or less correspond to \eqref{eq:bRESSIM} and Corollary
\ref{cor:microlocal-bounds},  stated as follows: Suppose
$\lambda_0\notin \sigma_{\pp}(\breve H)$. Then there exists the strong
weak-star limits
\begin{align}\label{eq:bRES0}
  \swslim_{\epsilon\to 0_+} \breve R(\lambda_0\pm \i \epsilon)=\breve R(\lambda_0\pm \i 0)\in \vL(\vB_{1/2}(\bX),\vB_{1/2}^*(\bX)),
\end{align} and if $\breve R(\lambda_0+ \i 0)f=\breve R(\lambda_0- \i
0)f$ for  a given $f\in L^2_s$ for some  $s>1/2$, then
 \begin{align}\label{eq:2bnd0}
  \breve R(\lambda_0+ \i 0)f=\breve R(\lambda_0- \i 0)f\in  L^2_{s-1}.
\end{align}

\end{subequations}

We shall prove a Mourre estimate for $\breve H$ near $\lambda_0$ and
show that $\lambda_0$ cannot be an accumulating point of eigenvalues
of $\breve H$ (which are the same as those of $H'$). If $\lambda_0$ is
not an eigenvalue, the limiting absorption principle and some
microlocal estimates hold for the resolvent of $\breve H$ near
$\lambda_0$. If $\lambda_0$ is an eigenvalue of $\breve H$, the
associated eigenfunctions are polynomially decaying. We don't assume
that $\lambda_0<\Sigma_3$.

The case of a multiple two-cluster
threshold can be treated similarly and will not be discussed in
detail. See Remark \ref{remark:microlocal-boundsGEN} for a discussion
and see \eqref{eq:basresbreve}--\eqref{eq:2bnd} for results similar to
\eqref{eq:basresbreve00}--\eqref{eq:2bnd0}.

\section{Mourre estimate}\label{Mourre estimates  and multiple commutators}
  We shall use the vector
  field constructed by Graf \cite{Gr} and the associated family
  of conjugate operators, cf. \cite{Sk1}--\cite{Sk3} and \cite{IS1}. This vector field satisfies the following
  properties, cf. \cite[Lemma 4.3]{Sk3}.
  \begin{lemma}
    \label{lemma:vector field} There exist on $\bX$ a smooth vector
    field $\tilde\omega$ with symmetric derivative $\tilde\omega_*$ and a
    partition of unity $\{\tilde q_a\}$ indexed by $a\in \vA$ and
    consisting of smooth functions, $0\leq \tilde q_a\leq 1$, such
    that for some positive constants $r_1$ and $r_2$
    \begin{enumerate}[\normalfont (1)]
    \item \label{item:4}$\tilde\omega_*(x)\geq\sum _a\pi_a \tilde q_a.$
    \item \label{item:5}$\tilde\omega^a(x)=0\text{ if }|x^a|<r_1.$
    \item\label{item:6} $|x^b|>r_1\mon \supp(\tilde q_a)\mif b\not
      \subset a.$
    \item\label{item:7} $|x^a|<r_2\mon \supp(\tilde q_a).$
    \item\label{item:7i} For all $\alpha\in \N_0^{\dim \bX}$ and
      $k\in \N_0$ there exist $C\in \R$:
      \begin{equation*}
        |\partial ^\alpha_x\tilde q_a|+|\partial ^\alpha_x(x\cdot
        \nabla)^k\parb{\tilde\omega(x)-x}|\leq C.
      \end{equation*}
    \end{enumerate}
  \end{lemma}

  For each $a\in \vA$ there is a similar vector field, denoted by
  $\tilde\omega^a$,  and from the
  construction of these vector fields  there is a relationship we are
  going to use (see \cite[Appendix
  A]{Sk2} for a proof for  the model  of Section \ref{First principal
    example}, see also \cite[Section 5]{Sk6}).

  For any $\delta>0$ there exists $\widetilde R =\widetilde R(\delta)>1$ such that for all
  $a\in \vA$
  \begin{align}
    \label{eq:37}
    \tilde\omega(x)&=x_a+\tilde\omega^a(x^a)\mforall x\in \bY,\\
    \bY=\bY_{a,\delta,\widetilde R}&:=\{x\in \bX \mid|x|>\widetilde
    R,\,|x^{b}|>\delta |x|\mif b\not \subset a\}.\nonumber
  \end{align} Considering  rescaled vector fields $\tilde\omega^a_R(x):=R\tilde\omega(x^a/R)$,
   $R\geq 1$, obviously a
  consequence of \eqref{eq:37} is the analogous result for the
  rescaled fields,
  \begin{equation}
    \label{eq:39}
    \tilde\omega_R(x)=x_a+\tilde\omega^a_R(x^a)\mforall x\in
    \bY_{a,\delta,R\widetilde R},\;R \geq 1.
  \end{equation}

  Now, proceeding as in \cite{Sk3}, we introduce
  \begin{equation*}
    A_R=\tfrac12 \parb{\tilde\omega_R(x)\cdot p+p\cdot \tilde\omega_R(x)},\;R\geq 1,
  \end{equation*} and a function $d:\R\to\R$  by
  \begin{equation}\label{eq:44}
    d(\lambda)=\begin{cases}\inf _{\tau\in \vT (\lambda)}(\lambda-\tau),\;\vT
      (\lambda):=\vT\cap \,]-\infty,\lambda]\neq \emptyset,\\
      1,\;\vT
      (\lambda)=\emptyset.
    \end{cases}
  \end{equation} These devices enter into the following Mourre estimate
  (we refer to \cite{PSS} for another $N$-body Mourre estimate). We refer to
  \cite[Corollary 4.5]{Sk3} noting  that  all  inputs needed for the proof are
  stated in Lemma \ref{lemma:vector field}.
  \begin{lemma}
    \label{lemma:Mourre1} For all $\lambda\in \R$ and $\epsilon>0$ the
    exists $R_0\geq 1$ such that for all $R\geq R_0$ there is a neighbourhood
    $\vV$ of $\lambda$ and a compact operator $K$ on $L^2$ such that
    \begin{equation}\label{eq:40}
      f(H)\i [H,A_R]f(H)\geq
      f(H)\{2d(\lambda)-\epsilon-K\}f(H)\mforall \text{ real }
      f\in C^\infty_\c(\vV).
    \end{equation}
  \end{lemma}

We can write $\tilde\omega=\nabla r^2/2$ for some positive smooth
function $r$ with  $r^2-x^2$  bounded, cf. \cite{De}.

A basic  ingredient of our procedure  is the operator
\begin{align}\label{eq:Bdefined}
  B=\tfrac12\parb{\omega(x)\cdot p+p\cdot \omega(x)}\in \vL(H^1,L^2),
\end{align}
where $\omega=\omega_R=\tilde \omega_R/r_R$, $r_R(x)=Rr(\tfrac x
R)$. Note that $\tilde\omega _R=\nabla r_R^2/2$,  and whence that
${\omega_R}=\nabla r_R$.
We shall suppress the dependence of the parameter $R$ (which
eventually is taken
as a  large number, depending on $\lambda$). In particular we shall slightly abuse the
notation  writing for example $r$ rather than the rescaled
version $r_R$. Using the notation ${\textbf D}$ for the Heisenberg
derivative $\i[H, \cdot ]$ we note  the
 computations $2B={\textbf D}
r$, $A=r^{1/2}Br^{1/2}$ and (formally)
 \begin{align}\label{eq:formulasy}
 {\textbf D}B=r^{-1/2}\parb{{\bD} A -2B^2}r^{-1/2} +\vO(r^{-3}).
\end{align}
Here the function
\begin{align*}
  \vO(r^{-3})=
\tfrac12 \omega\cdot (\nabla^2r)\omega/r^2=r^{-3}v(x),
\end{align*}
  where $v$ belongs to the algebra $\vF=\vF(\bX)$ of smooth functions   on $\bX$
obeying
\begin{equation*}
    \forall \alpha\in \N_0^{\dim \bX}\quad\forall
    k\in \N_0: |\partial_x^\alpha(x\cdot
    \nabla)^k v(x)|\leq C_{\alpha,k}.
  \end{equation*}
   Note also that  the function
  $r^2-x^2\in
  \vF$. Obviously $\vF(\bX)\supset \vF(\bX^a)$ for any $a\neq a_{\min}$. The exact computation of  ${\bD} A$ reads
\begin{align}
    \label{eq:c2com}
     {\bD} A &=2p\tilde\omega_*\big (\tfrac xR\big
)p-\parb{4R^2}^{-1}\parb{\triangle(\nabla\cdot \tilde\omega)}\big (\tfrac xR\big
)-R\tilde\omega\big (\tfrac xR\big
)\cdot \nabla V.
  \end{align}
 In particular
  \begin{align*}
    {\bD} A=\i[H,A]=\sum_{|\alpha\leq 2}v_\alpha p^\alpha;\quad v_\alpha\in \vF,
  \end{align*} which make sense as a bounded form on
  $H^1=Q(H)$. Although we  define $\bD A$ by \eqref{eq:c2com}, it can be computed as a strong limit,
  \begin{align}
    \label{eq:limit2A0}
    \begin{split}
     \i [  H,A]&=\slim_{t\to 0} t^{-1}\parb{ H\e^{\i t A}-\e^{\i
    t A} H}\in\vL(H^1, H^{-1}).
    \end{split}
  \end{align}    Similarly the formal
 commutator in  \eqref{eq:formulasy}  can be computed as a strong
 limit in $\vL(H^1, H^{-1})$,
  \begin{align}
    \label{eq:limit}
    \begin{split}
     {\textbf D}B&=\slim_{t\to 0} t^{-1}\parb{H\e^{\i tB}-\e^{\i
    tB}H}\\&=r^{-1/2}\parb{{\bD} A -2B^2}r^{-1/2}
    +r^{-3}v\\
&=\sum_{|\alpha|\leq 2}r^{-1}v_\alpha p^\alpha;\quad v_\alpha\in \vF.
    \end{split}
  \end{align} The assertions \eqref{eq:limit2A0} and \eqref{eq:limit}
  are standard results, which follow from mapping properties of the
  involved  groups and the fact that the formal commutators are
  bounded forms on $H^1$.

We will  need  modifications  of $A$ and $B$ in terms of a
parameter $\kappa$. This  parameter is needed to control certain multiple commutators. Let
\begin{align}
  \label{eq:41} B_{\kappa,R}=B(\kappa^2 B^2+1)^{-1},\quad
  A=A_{\kappa,R}=r^{1/2}B_{\kappa,R}r^{1/2};\quad \kappa\in [0,1].
\end{align}

We may consider $\widetilde H$ defined in \eqref{eq:38p} as a `generalized'
$N$-body Schr{\" o}dinger operator. The set of thresholds of $\widetilde
H$, say denoted by $\widetilde \vT$, coincides with the set $\vT$ of thresholds of $
H$ except for having one less point. This exception follows from the
identity
\begin{equation*}
\sigma_{\pupo}\parb{\widetilde H^{a_0}}=\Big (\sigma_{\pupo}\parb{H^{a_0}}
\setminus\{\lambda_0\}\Big )\cup\{0\};\quad \widetilde H^{a_0}:=H^{a_0}-\lambda_0\Pi.
\end{equation*} Note in particular that it follows that $\lambda_0 \not\in
\widetilde \vT$. On the other hand  there is no simple relationship between the
eigenvalues of $\widetilde H$ and those of $H$. A similar `subtraction'
of a genuine  eigenprojection $P$ of $H$ was employed in \cite{AHS} in
which case
indeed a
similar relation between $\sigma_{\pupo}\parb{H-\lambda_0P}$ and
$\sigma_{\pupo}\parb{ H}$ hold. Let in the following $\tilde
d:\R\to\R$ refer to the (lower) distance function defined by replacing
$\vT\to \widetilde \vT$ in \eqref{eq:44}. Note that $\tilde
d(\lambda_0)>0$.
\begin{prop} \label{prop:mour2} Suppose Conditions \ref{cond:smooth} and
  \ref{cond:uniq}. Let $A_{\kappa,R}$ and  $\breve  H$  be given by
  \eqref{eq:41} and \eqref{eq:38p}, respectively.

  For all $\lambda\in \R$ and
  $\epsilon>0$ there  exist
  $R_0\geq 1$  and $\kappa_0 \in
      (0,1]$ such for  all $R\geq R_0$,  there exist  a neighbourhood $\vU$ of
  $\lambda $ and a compact operator $K$ on $L^2$:
  \begin{align} \label{eq:43breve}
    \begin{split}
    f(\breve H)&
 \i [\breve  H,A_{\kappa,R}]f(\breve  H)\\&\geq f(\breve
    H)\{2\tilde d(\lambda)-\epsilon-K\}f(\breve  H)\text{ for all
    }\kappa\in [0,\kappa_0] \text{ and real }
      f\in C^\infty_\c(\vU).
    \end{split}
  \end{align}
 \end{prop}

The meaning of the appearing commutator will be explained in
Subsection \ref{subsec:Proof of Proposition mour}. However for  $\kappa=0$ there
is an alternative interpretation to be elaborated on now.
 We claim that we can use \eqref{eq:43breve}  at $\lambda=\lambda_0$ and for $\kappa=0$ to
 conclude,  that there are at most a finite number of eigenvalues for
 $H'$ in a neighbourhood of $\lambda_0$. To see this  it suffices to
 show that the
 commutator in  \eqref{eq:43breve}, interpretated as  the formal
 commutator, can be computed as a strong limit
  \begin{align}
    \label{eq:limit2A}
    \begin{split}
     \i [\breve  H,A_{R}]&=\slim_{t\to 0} t^{-1}\parb{\breve H\e^{\i t A_{R}}-\e^{\i
    t A_{R}}\breve H}\in\vL(H^2, H^{-2}).
    \end{split}
  \end{align} Writing $\breve H=H-T$,  $T:=\lambda_0\Pi +\breve
    K$,  the part of \eqref{eq:limit2A} related
  to $H$ is justified by \eqref{eq:limit2A0}. For the  part  related
  to   the second term  it
  suffices to show that the form, say a priori defined on
    $ C_{\c}^\infty(\bX)$,  extends  as follows. \begin{lemma}\label{lem:Tcont} The form $\i\ad _{A_{R}}(T):=\i [T,A_{R}]$   extends to a bounded form on $H^2$. More generally $\i\ad _{A_{R}}(T)\in\vL(H^{2}, H^{-1})\cap \vL(H^{1}, H^{-2})$.
  \end{lemma}
  \begin{proof}
    Write the form as
\begin{align*}
    \i [\lambda_0 \Pi +\Pi I_{0}+
  I_{0}\Pi -\Pi I_{0}\Pi,A_{R}]=T_1+\cdots +T_4.
  \end{align*} We know that $ [I_{0},A_{R}]\in
  \vF\subset\vL(L^2)$. Let us only consider the contribution   $
  [\Pi,A_{R}]I_{0}$ (from $T_2$) and only show that  $
  [\Pi,A_{R}]I_{0}\in \vL(H^2, H^{-1}$). For that it suffices to show that $
  [\Pi,A_{R}]\in \vL(L^2, H^{-1}$).

Let $1=\chi_1+\chi_2$ be  a partition of unity on $\bX$ given as
follows. We demand that for
      some (small) $\sigma>0$ and all $\widecheck R\geq 1$:

\begin{enumerate}[\normalfont (1)]
\item \label{item:82} $\chi_1\in C^\infty(\bX)$ and for all $\alpha\in  \N_0^{\dim \bX}$ there exists $C\in \R$:
      \begin{equation*}
        |\w{x}^{|\alpha|}\partial ^\alpha_x\chi_1(x)|\leq C.
      \end{equation*}

    \item \label{item:12} $\chi_1(x)=1$ for
$x\in \{|x|\geq 2\widecheck R,\;|x^{a_0}|\leq \sigma |x|\}$ and
      \begin{equation*}
        \supp \chi_1\subset \{|x|> \widecheck R,\;|x^{a_0}|<
 2\sigma
        |x|\}\subset\bX\setminus\cup_{b\not\subset a_0}\bX_b.
      \end{equation*}
\end{enumerate} Note that \ref{item:12} implies, referring here to
notation of \eqref{eq:37} where $\delta=\delta(\sigma)>0$ is taken
sufficiently small (independently of $\widecheck R$),
\begin{equation}
  \label{eq:43}
  \supp \chi_1\subset \bY_{a_0,\delta,\widecheck R}.
\end{equation}
 We write
 \begin{align}
   \label{eq:46q}
  \i[\Pi,A_R]=(\chi_1+\chi_2)\i[\Pi,A_R](\chi_1+\chi_2)=\chi_1\i[\Pi,A_R]\chi_1+S.
 \end{align}

Now for all sufficiently big values of $\widecheck R$  (viz. $ \widecheck R\geq R
\widetilde R=R\widetilde R(\delta)$) we obtain by combining  \eqref{eq:39}
and \eqref{eq:43} that
\begin{equation}
    \label{eq:37i}
    \tilde\omega_R(x)=x_{a_0}+\tilde\omega^{a_0}_R(x^{a_0})\mforall x\in \supp \chi_1.
  \end{equation} The above construction can be done in an
  explicit way (including an explicit dependence of
  parameter $R$): Let us here and henceforth fix  $\widecheck R=R\widetilde
  R(\delta)$ and choose
  $\chi_1(x)=\bar\chi_{\widecheck R}(|x|)\Theta(x/|x|)$ with
  $\bar\chi_{\widecheck R}$ specified in \eqref{eq:14.1.7.23.24} and for a
  suitable real-valued smooth function
  $\Theta$. We record the following improvement of \ref{item:82}.
\begin{enumerate}[\normalfont (1)']
\item \label{item:8} For all $\alpha\in \N_0^{\dim \bX}$ and $k\in \N$
  there exists $C\in \R$
  such that for all $R\geq 1$:
      \begin{equation*}
        |\w{x}^{|\alpha|}\partial   ^\alpha_x\chi_1(x)|+\big|\w{x}^{|\alpha|}\partial
        ^\alpha_x\parbb{\tfrac{\w{x}^k}{\w{x^{a_0}}^k}\bar\chi_{2\widecheck R}({|x|})\chi_2(x)}\big|\leq C.
      \end{equation*}
\end{enumerate}

Since the operator $A_R$ is a local operator the first term in
\eqref{eq:46q} simplifies due to \eqref{eq:37i} as
\begin{equation}
  \label{eq:76}
  \chi_1\i[\Pi,A_R]\chi_1=\chi_1\i[\Pi,A_R^{a_0}]\chi_1=\i\chi_1\parb{\Pi
    A_R^{a_0}-A_R^{a_0}\Pi}\chi_1\in \vL(L^2).
\end{equation}

Now let us look at the term $S$ in \eqref{eq:46q}. It is given by
\begin{equation}\label{eq:77}
  S=\chi_2\i[\Pi,A_R]+\chi_1\i[\Pi,A_R]\chi_2.
\end{equation} The two terms are treated similarly, so let us only
elaborate on the first term. We write
\begin{align*}
  A_R=r_RB_R-\tfrac \i 2\abs{(\nabla r)(x/R)}^2,
\end{align*} and then
\begin{align}\label{eq:79}
  \begin{split}
   \chi_2[\Pi,A_R]&= \chi_2\Pi A_R-\chi_2A_R\Pi\\
&=K_1 B_R-B_R K_2-\i K_3;\\
&K_1=\chi_2\Pi
r_R,\\&K_2=r_R\chi_2\Pi,\\ &K_3=
   \tfrac 12\chi_2 \Pi
                            \,\abs{(\nabla r)(x/R)}^2+ \tfrac 12 \chi_2
                            \abs{(\nabla r)(x/R)}^2 \,\Pi+ \parb{\tilde\omega_R\cdot\nabla \chi_2}\Pi.
  \end{split}
\end{align}
 Noting, cf. \ref{item:8},  that
 \begin{align}\label{eq:Kbnds}
    K_1, K_2\in\vL(H^{-1})\cap \vL(L^2),\,K_3\in\vL(L^2) \mand  B_R\in\vL(L^2, H^{-1}),
 \end{align}
 we are done.
\end{proof}

\begin{remark} \label{remark:2comm}
  It follows from the technique of  the proof that the second order
  commutator $\ad^2_{A_{R}}(T)\in
  \vL(H^2, H^{-2})$ (and possibly no better for singular potentials). This
  suffices for the limiting absorption principle at $\lambda_0$,
  cf. \cite{PSS}. However higher commutators do not exist and we need
   refined micro-local estimates, which usually require multiple
   commutators. (In particular we need bounds  with weights in position
  space.)  Using the
  $\kappa$-distortions of $A_R$  and $B_R$ will allow us to treat multiple
  commutators. Note for example that for $\kappa>0$
  the above proof yields,  at least formally, the improved result
   $ \ad _{A_{\kappa,R}}(T)\in\vL(H^1, H^{-1})$. We introduce a  calculus of the
 $\kappa$-distortion of  $B$ which will be a
 major object  to study. This is done in Section
 \ref{subsec:Higher commutators}, and we give a number of applications of
 Proposition \ref{prop:mour2} and this calculus in Section \ref{subsec:Positive
  commutator estimates}.
\end{remark}

\subsection{Proof of Proposition \ref{prop:mour2}}\label{subsec:Proof of Proposition mour}

Introduce  for $\varepsilon>0$ the operators
\begin{align}
  \label{eq:41delta} X_{\varepsilon}=X_{\varepsilon,R}=r/(1+\varepsilon r ),\quad
  A_{\varepsilon,\kappa,R}=X^{1/2}_\varepsilon B_{\kappa,R}X_\varepsilon^{1/2};\quad \kappa\in [0,1].
\end{align} Note that
$A_{\varepsilon,\kappa,R}\in\vL(H^2,H^1)$. Whence  $\i [  H,A_{\varepsilon,
  \kappa,R}]$ is a well-defined
form on $H^2$. We
\emph{define} the commutators    $\i [
H,A_{\kappa,R}]$, $\i [ \widetilde{H},A_{\kappa,R}]$ and $\i [ \breve
H,A_{\kappa,R}]$  as  strong weak-star limits in $\vL(H^2,H^{-2})$ in
the following way. For technical reasons to be explained after the
definitions,  this works  for small enough values of $\kappa$ only. Let for
each of the operators $H^{\#}=H,\,\widetilde{H}$ or  $\breve H$
\begin{align}\label{eq:limFod}
  \i [  H^{\#},A_{\kappa,R}]=\swslim_{\varepsilon\to 0_+} \i [ H^{\#},A_{\varepsilon,
    \kappa,R}].
\end{align} It is easy to see that for $\kappa=0$, these limiting
forms exist and coincide with our previous computations, cf.
\eqref{eq:limit2A} and Lemma \ref{lem:Tcont}. The case of $\kappa>0$
requires an elaboration to be given in Lemma \ref{lemma:Mourre2}.

An immediate virtue of
\eqref{eq:limFod} is the validity of the virial theorem. Thus for
example,
\eqref{eq:limFod} combined  with Proposition \ref{prop:mour2} yields
that eigenvalues of $\breve H$ cannot accumulate at $\lambda_0$. Note
that this assertion does not require $\lambda_0=\Sigma_2$ for which
case  in
fact $\lambda_0\notin\sigma_\mathrm {ess }(\breve H)$, cf. \cite[Lemma
2.1]{Wa2} or the proof of Lemma \ref{lemma2.1}. Proposition
\ref{prop:mour2} will be a crucial tool for us only for
$\lambda_0>\Sigma_2$. At the level of proofs, the
  reader will see  similarities of the proofs of Lemma
  \ref{lemma2.1} and Proposition
   \ref{prop:mour2}.

We are
going to use that for $\kappa>0$
      \begin{align}\label{eq:hkappa}
        B_{\kappa,R}=\tfrac1{2\kappa} \parb{(\kappa B+\i)^{-1}+(\kappa B-\i)^{-1}}.
      \end{align}
 Note that  it follows from Mourre theory, \cite{Mo}, that $(\kappa
 B\pm \i)^{-1}$
 preserve  $H^2$ provided $\kappa$ is small enough.   We need a
 uniform bound in $\kappa$ and   $R$. Thus, more precisely,  we claim that
 there exists $\kappa'_0\in (0,1]$ such that for all
$\kappa\in[0,\kappa'_0]$ and all $R\geq 1$ the space $H^2$ is preserved,
and in fact
\begin{subequations}
\begin{align}
   \label{eq:stroI2}
   \sup_{R\geq  1,\,\kappa\leq \kappa'_ 0}\,\norm{(\kappa
 B\pm \i)^{-1}}_{\vL(H^2)}<\infty.
 \end{align} We may take this constant  as $\kappa_0'=\min\set{(2C)^{-1},1}$, where
\begin{align}
   \label{eq:stroI3}
   C=\sup_{R\geq  1}\,\norm{[H_0,B](H_0+1)^{-1}}_{\vL(L^2)}<\infty.
 \end{align} Note that  the $R$-dependence of $B$ is through $\omega_R(x)=(\nabla r)(x/R)$, which is bounded along
 with all derivatives. Whence the above finiteness claims
 hold. We also note that for any  $R\geq 1$
 \begin{align}
   \label{eq:stroI}
   \slim_{\kappa\to 0}(\kappa
 B\pm \i)^{-1}=\mp \i  \text{ in }\vL(H^2).
\end{align}
\end{subequations}

Clearly there are similar properties as \eqref{eq:stroI2} and
\eqref{eq:stroI} with  $H^2$ replaced by  $L^2_s$ for any
 $s\geq 0$. For \eqref{eq:stroI2} in that case,   $\kappa'_0$  can be
 an  arbitrary positive number (it does
 not need to be small); see
 the beginning of the proof of Lemma \ref{lem:orderP}.


\begin{lemma}
  \label{lemma:Mourre2} The limits  $\i [
   H^{\#},A_{\kappa,R}]$, with $H^{\#}=H,\,\widetilde{H}$ or  $\breve H$, are
  well-defined bounded form on $H^2$ for $\kappa\in [0,\kappa'_0]$. More generally $\i [
   H^{\#},A_{\kappa,R}] \in\vL(H^{2}, H^{-1})\cap \vL(H^{1}, H^{-2})$.

Moreover for each of the above three forms
\begin{align}\label{eq:45exp}
  \exists C>0\,\forall R\geq 1\,\forall\kappa\in [0,\kappa'_0]:\quad  \norm[\big]{\i [
   H^{\#},A_{\kappa,R}]-\i [
   H^{\#},A_{R}]}_{\vL(H^{2},H^{-2})}\leq C\kappa.
\end{align}

  \end{lemma}
  \begin{proof} We need to examine the case $\kappa>0$. We start by  examining  the quantity $\i
    [H,A_{\kappa,R}]$.

  We  calculate using \eqref{eq:limit}, \eqref{eq:hkappa}  and \eqref{eq:stroI2}
    \begin{align*}
&\i\comm[\Big]{H,A_{\varepsilon,\kappa,R}}\\
&=\i X_\varepsilon^{1/2}\comm[\Big]{H,B_{\kappa,R}}X_\varepsilon^{1/2}+2\Re\parbb{\i \comm[\Big]{H,X_\varepsilon^{1/2}}B_{\kappa,R}X_\varepsilon^{1/2}}\\
      &=-\tfrac12\sum_{\pm}X_\varepsilon^{1/2}(\kappa B\pm\i)^{-1}\i
      \comm[\Big]{H,B}(\kappa B\pm\i)^{-1}X_\varepsilon^{1/2}+2\Re\parbb{ \Re{\parb{X_\varepsilon^{-1/2}(\nabla X_\varepsilon)\cdot
      p}}B_{\kappa,R}X_\varepsilon^{1/2}}\\
&=-\tfrac12\sum_{|\alpha|\leq 2,\,\pm}X_\varepsilon^{1/2}(\kappa B\pm\i)^{-1}
      )r^{-1}v_\alpha p^\alpha (\kappa
  B\pm\i)^{-1}X_\varepsilon^{1/2}\\
&+2\Re\parbb{ r^{-1/4}(1+\varepsilon r)^{-3/4}B_R(1+\varepsilon r)^{-3/4}r^{-1/4}B_{\kappa,R}X_\varepsilon^{1/2}};\quad v_\alpha\in \vF.
\end{align*}  We can let $\varepsilon\to 0$, yielding the
existence of the desired limit in the  $H^2$-form sense. The result is
\begin{align*}
\i\comm[\Big]{H,A_{\kappa,R}}
&=-\tfrac12\sum_{|\alpha|\leq 2,\,\pm}r^{1/2}(\kappa B\pm\i)^{-1}
      )r^{-1}v_\alpha p^\alpha (\kappa
  B\pm\i)^{-1}r^{1/2}\\
&+2\Re\parbb{ r^{-1/4}B_Rr^{-1/4}B_{\kappa,R}r^{1/2}};\quad v_\alpha\in \vF.
\end{align*}
 This expression is obviously in $\vL(H^{2}, H^{-1})\cap \vL(H^{1},
H^{-2})$ (in fact in $\vL(H^{1}, H^{-1})$ in this case).

 Next we  move multiplication operators to the middle as to
become sandwiched by $(\kappa B\pm\i)^{-1}$. Thereby  we pick up errors
of order $\vO(\kappa)$ in $\vL(H^2,H^{-2})$ uniformly in $R\geq 1$,
cf. \eqref{eq:stroI2} and the remark after \eqref{eq:stroI}.  This
means that we only have to consider
 \begin{align}\label{eq:45CO}
   -\tfrac12\sum_{\pm}(\kappa B\pm\i)^{-1}\i
      \comm[\big]{H,A}(\kappa B\pm\i)^{-1},
 \end{align} where $\i
      \comm[\big]{H,A}$ is given by \eqref{eq:c2com}.
 But
\begin{align}\label{eq:regB}
   (\kappa B\pm \i)^{-1}\pm \i=\pm \i\kappa  (\kappa B\pm \i)^{-1}B,
 \end{align} which give an extra error $\vO(\kappa)\text{ in
   }\vL(H^2, H^{-2})$ when removing the factors $(\kappa B\pm
   \i)^{-1}$ in \eqref{eq:45CO} one by one. This proves  \eqref{eq:45exp} in the case  $H^{\#}=H$.

For  the  commutators  $\i [ \breve
    H,A_{\kappa,R}]$ and $\i [ \breve
    H,A_{\kappa,R}]$ we  proceed similarly and by using in  addition  the proof of Lemma
    \ref{lem:Tcont}. Note that in the  case of $H^{\#}=\breve H$ the
    topology in \eqref{eq:45exp} is the strongest Sobolev space
    topology   we can use  for singular
    potentials by this method, cf. Remark \ref{remark:2comm}.
\end{proof}

\begin{lemma}
  \label{lemma:Mourre3}
 For all $\lambda\in \R$ and $\epsilon>0$ there
  exist $R_0\geq 1$  such for all $R\geq R_0$,
  there exist a neighbourhood $\vV$ of $\lambda $ and a compact
  operator $\widetilde{ K}$ on $L^2$:
  \begin{align} \label{eq:42}
    \begin{split}
    f( \widetilde{H})&
 \i [  \widetilde{H},A_{R}]f(  \widetilde{H})\\&\geq f(
    \widetilde{ H})\{2 \tilde{d}(\lambda)-\epsilon-\widetilde{ K}\}f(  \widetilde{H})\text{ for all
     real }
      f\in C^\infty_\c({\vV}).
    \end{split}
  \end{align}
  \end{lemma}
  \begin{proof}
    We divide the proof into several steps.
 \subStep{I}  We note that for any
 product
$P_1\cdots P_m$
 of factors $P_l$ given either as $P_l=\inp{x^{a_0}}$ or $P_l=A_R^{a_0}$ we have
 \begin{equation}
   \label{eq:72}
   P_1\cdots P_m\Pi\in \vL \parb{L^2(\bX^{a_0})}.
 \end{equation} In fact since $\inp{x^{a_0}}^m(H^{a_0}+\i)^m\Pi\in
 \vL \parb{L^2(\bX^{a_0})}$ the result follows from the property
\begin{equation*}
   P_1\cdots P_m(H^{a_0}+\i)^{-m}\inp{x^{a_0}}^{-m}\in \vL \parb{L^2(\bX^{a_0})},
 \end{equation*} which in turn follows from repeated commutation  and
 the property $A_R^{a_0}\inp{x^{a_0}}^{-1}(H^{a_0}+\i)^{-1}\in \vL \parb{L^2(\bX^{a_0})}$.

 \subStep{II}
There is a complete analogue of
Lemma \ref{lemma:Mourre1} with the triple  $\parb{H, \i[H,A_R],d}$ replaced
by $\parb{\widetilde H, \i[H,A_R],\tilde d}$. It reads conveniently  as follows.

For all $\lambda\in \R$ and $\epsilon>0$ there
  exists $R_0\geq 1$  such for all $R\geq R_0$,
  there exist a bounded neighbourhood $\vV$ of $\lambda $ and a compact
  operator $\widetilde{K}_1 $ on $L^2$:
  \begin{align} \label{eq:40tilde}
    \begin{split}
    f( \widetilde{H})&
 \i [  {H},A_{R}]f(  \widetilde{H})\\&\geq f(
    \widetilde{ {H}})\{2 \tilde{d}(\lambda)-\epsilon/2-\widetilde{K}_1\}f(  \widetilde{H})\text{ for all
    real } f\in C^\infty_\c(\vV).
    \end{split}
  \end{align}

This  statement follows by mimicking the proofs  of \cite[Corollaries
4.2 and 4.5]{Sk3}. One comment  is due since the `potential' $\Pi$
is not local: Expanding the operator $T(\delta')$ appearing in
\cite[(4.16)]{Sk3} one `new term'  is given by $T:=[\Pi,j_b]$; here
$j_b$ is a partition function with similar properties as the ones in \eqref{eq:33}--\eqref{eq:36}. We need
to show that $T$ is $\widetilde{H}$-compact. If $a_0\not \subset b$ indeed
$\Pi j_b$ and $j_b\Pi$, and hence also $T$,  are $H$-compact, cf. Step \textit{I}. If on the other
hand $a_0\subset b$ we have $a_0= b$ and we can write $T=-\sum_{d\neq
  b}[\Pi,j_d]$. Since for any such term we have $a_0\not \subset d $
the previous argument yields that $[\Pi,j_d]$ is $H$-compact.

\subStep{III}  Recalling the definition \eqref{eq:38p} it remains  to
show the  following
estimate in terms of the set $\vV=\vV(R)$ from \eqref{eq:40tilde}.

 There exists a compact operator $\widetilde K_2=\widetilde K_2(R)$
such that
\begin{equation}
  \label{eq:80}
  1_{\vV}(\widetilde H)\i [-\lambda_0\Pi,A_{R}] 1_{\vV}(\widetilde
  H)\geq- \epsilon/ 2 -\widetilde K_2.
\end{equation}

Clearly the combination of \eqref{eq:40tilde} and \eqref{eq:80} yields
\eqref{eq:42}  with
this  neighbourhood $\vV $  and with
$\widetilde{K}=\widetilde  K_1 +\widetilde K_2$.
 To show \eqref{eq:80} we decompose as in the proof of Lemma
      \ref{lem:Tcont}  writing
\begin{align*}
  \i
      \comm[\Big]{\Pi,A}&=(\chi_1+\chi_2)\i
      \comm[\Big]{\Pi,A}(\chi_1+\chi_2)=\chi_1\i
      \comm[\Big]{\Pi,A_R^{a_0}}\chi_1+S.
 \end{align*} Since
 \begin{align*}
   \norm{\Pi A_R^{a_0}-A_R^{a_0} \Pi}\to 0\text{ for }R\to \infty,
 \end{align*} cf. Lemma \ref{lemma:vector field}\ref{item:5} and Step
 \textit{I}, the first
 term conforms with \eqref{eq:80}.

It remains to consider $S$, which  is given by \eqref{eq:77}. We
decompose as in \eqref{eq:79} (treating again only the first term of
$S$).  The operators $K_1$,  $K_2$ and  $K_3$ are not only bounded (as
stated in \eqref{eq:Kbnds}) but also
$\inp{p}$-compact. Since $\vV $ is bounded \eqref{eq:80} follows.
\end{proof}

\begin{proof}[Proof of Proposition \ref{prop:mour2}] Let
  $\lambda\in\R$ and $\epsilon>0$ be given. We need to specify $R_0\geq 1$, $\kappa_0 \in
      (0,\kappa_0']$ and for $R\geq R_0$  a neighbourhood $\vU(R)$ of $\lambda$
      and a compact $K(R)$ such that \eqref{eq:43breve} holds.

\subStep {I}
      First we use
    Lemma \ref{lemma:Mourre3} to    obtain the statement
   \eqref{eq:42}   with the factors of $f(\widetilde{H})$ replaced by the factors of $f(\breve
  H$). More precisely we use the statement with $\epsilon$ replaced by
  $\epsilon/2$ allowing us to pick corresponding $R_0\geq 1$  and  for $R\geq R_0$ a  neighbourhood $\vV(R)$ of $\lambda$
      and a compact $\widetilde{K}(R)$ such that  \eqref{eq:42}
      holds. Now take for any such set  $\vV(R)$ any smaller
      neighbourhood of $\lambda$, denoted by  $\vU(R)$,  with compact
      closure contained in the interior of
      $\vV(R)$. We can pick a  real function  $\tilde{f}\in
      C^\infty_\c(\vV (R))$  such that $\tilde{f}=1$ on $\vU(R)$
      and then estimate as follows. Note that $\breve K_0=\breve K_0(R):=\parb{H_0+1}\parb{\tilde f(\breve
        H)-\tilde f(\widetilde{H} )}$  is compact
      (cf. \eqref{eq:almostAnal} given  below) and that $\i [
          \widetilde{H},A_{R}]\in{\vL(H^{2},H^{-2})}$  (cf. Lemma \ref{lem:Tcont}).
Writing then
\begin{align*}
    \breve K_1:=&\tilde f( \breve{H})
 \i [  \widetilde{H},A_{R}]\tilde{f}( \breve{H})-\tilde f(\widetilde{ H})
 \i [  \widetilde{H},A_{R}]\tilde{f}( \widetilde{H})\\
=&\breve K_0^* \breve T+\widetilde{T}^*\breve K_0;\\
& \quad\breve T=\parb{H_0+1}^{-1}\i [
          \widetilde{H},A_{R}]\tilde{f}( \breve{H}),\\
& \quad\widetilde{T}=\parb{H_0+1}^{-1}\i [
          \widetilde{H},A_{R}]\tilde{f}( \widetilde{H}),
     \end{align*} we  conclude that $\breve K_1$  is
     compact. Note also that
\begin{align*}
    \breve K_2(R):=\parb{2 \tilde{d}(\lambda)-\epsilon/2}\parb{\tilde
  f( \breve{H})^2-\tilde{f}( \widetilde {H})^2}
     \end{align*} is compact.

In conclusion we take $R_0\geq 1$,   $\vU(R)$ as described above  and
      \begin{align*}
        \breve K_3(R):=\tilde{f}( \widetilde{H})\widetilde{K}\tilde{f}( \widetilde{H})-\breve K_1+\breve K_2;\quad R\geq R_0,
      \end{align*} yielding
\begin{align*}
   f( \breve {H})&
 \i [  \widetilde{H},A_{R}]f(  \breve {H})\\&\geq f(
    \breve{ H})\{2 \tilde{d}(\lambda)- \epsilon/2-\breve{ K_3}\}f(  \breve {H})\text{ for all
      real }
      f\in C^\infty_\c({\vU(R)}).
  \end{align*}

\subStep {II}
We note that
\begin{align*}
  \breve K_4(R):=1_{U(R)}(\breve H)\i [\breve K,A_{R}]
  1_{U(R)}(\breve H) \text{ is compact}.
\end{align*} We subtract   these terms in the previous estimate,
yielding
\begin{align*}
   f( \breve {H})&
 \i [  \breve{H},A_{R}]f(  \breve {H})\\&\geq f(
    \breve{ H})\{2 \tilde{d}(\lambda)- \epsilon/2-{ K}\}f(  \breve {H})\text{ for all
     real }
      f\in C^\infty_\c({\vU(R)}),
  \end{align*} where $K=K(R)= \breve
K_3+\breve K_4$.

\subStep {III}
We invoke \eqref{eq:45exp}, estimating
\begin{align*}
  1_{U(R)}(\breve H)\parb{\i [\breve H,A_{\kappa,R}]-\i [\breve H,A_{R}]
  }1_{U(R)}(\breve H)\geq -\epsilon/2,
\end{align*} for all small enough $\kappa$. In combination with the
previous estimate this estimate yields \eqref{eq:43breve}.
\end{proof}

\section{Multiple commutators and calculus}\label{subsec:Higher commutators}

It is convenient to introduce a concept of `order' for certain classes
of linear operators  $T$ on $L^2$, cf. \cite{GIS}. It is based on the
following elementary result, cf.  Lemma \ref{lem:orderP}.
\begin{align}\label{eq:basCom}
  \forall \kappa \in (0,1],\;\forall R\geq 1:\quad
  {B_{\kappa,R}}\,L^2_\infty\subset L^2_\infty.
\end{align}
For operators $S,T$ on $L^2$   we define  (formally)  multiple commutators by
 \begin{align*}
  \ad^0_S(T)=T\mand    \ad^{k}_S(T)=[\ad^{k-1}_S(T),S]\text{ for } k\in \N.
 \end{align*} Let $X$ be  multiplication by  $r=r_R$ on  $L^2$, and
 recall that $\kappa_0'\in (0,1]$ is introduced in the discussion of
 \eqref{eq:stroI2} and \eqref{eq:stroI3}.

\begin{defn}\label{defn:order} Let $R\geq 1$ be given.
  An operator $T$ on $L^2$
  is \emph{of $\kappa$-order $t\in \R$} if
  \begin{enumerate}[(1)]
  \item $\vD(T),\vD(T^*)\supset
  L^2_\infty$,
\item $T$ and $T^*$ leave $L^2_\infty$  invariant,
\item \label{item:23}
$\forall \kappa \in (0,\kappa_0']\;\forall s\in \R\,\;\forall k\in \N_0:\quad
  X^{s-k-t}\ad_{B_{\kappa,R}}^k(T) X^{-s}\in \vL(L^2).
$
  \end{enumerate}
\end{defn}

For  any  operator $T$  of  $\kappa$-order $t$ the adjoint
$T^*$  also has   $\kappa$-order $t$, and  we write $T,
T^*=\vO_{\kappa}(X^{t})$. The class of such operators is denoted by
$\gO_\kappa(X^{t})$.  It is readily seen that
$X^{t}=\vO_{\kappa}(X^{t})$, cf. \eqref{eq:X} and \eqref{eq:Ts's}
stated below.
This is related to the fact that $\ad_B^{k}( X^{k})\in \vF$ for
$k\in\N$, where the algebra $\vF$ is introduced in the discussion of \eqref{eq:formulasy}.  By the Leipniz rule, if $S$ and $T$ are of $\kappa$-order $s$ and
$t$, respectively, then $ST$ is of order
$s+t$. We do not   keep
track of neither the $\kappa$- nor the
$R$-dependence, however it  is important for our applications of the
above
concept of `order' that $\kappa=0$
is \emph{not included} in \ref{item:23}. We abbreviate
$\breve R(z)=(\breve H -z)^{-1}$ and
$R_\kappa(z)=(B_\kappa-z)^{-1}=(B_{\kappa,R}-z)^{-1}$.

\begin{lemma}
  \label{lem:orderP}   Let $R\geq 1$ be given big enough. Then for any $f\in C_\c^\infty(\R)$, $g\in
  C^\infty(\R)$, $v\in\vF$, $|\alpha|\leq 2$,  $z\in \C\setminus\R$ and polynomials $P_1,P_2$ in
 $A^{a_0}$ and the components of  $x^{a_0}$ and $p^{a_0}$, such that  the total
 number of components of  $p^{a_0}$  is at most two  for $P_1$ as well as
 for $P_2$, the operators
  \begin{align}
    \label{eq:classP}
    g(B_{\kappa}), \,v,\,P_1\Pi P_2,\,I_{0}(p^2+1)^{-1}, \,p^\alpha \breve
    R(z),\,
  I_{0} \Pi \breve R(z),\,f(\breve H)\in \gO_{\kappa}(X^{0}).
  \end{align}
\end{lemma}
  \begin{proof} We prove the bounds one by one. For simplicity of presentation the
    above list  is not stated as complete as possible. In the proof we will see that a
    certain algebra $\breve \vB\subset \gO_{\kappa}(X^{0})$; this
    algebra contains  a few
 operators not listed above  for which  the assertion that
    they  are  contained in $\gO_{\kappa}(X^{0})$ will also be useful
    later.

By repeated commutation we  obtain for the
    regularization  $X_\varepsilon$ of \eqref{eq:41delta} that for any $s\geq 0$, $\varepsilon>0$
and $z\in \C\setminus\R$
\begin{align}
  \begin{split}\label{eq:rescomm}
  X_\varepsilon^{s} (B-z)^{-1}=\sum^{[s]}_{k=0}&
  (-1)^k (B-z)^{-k-1}\ad_B^k( X_\varepsilon^{s})\\&+(-1)^{[s]+1} (B-z)^{-[s]-1}\ad_B^{[s]+1}( X_\varepsilon^{s})(B-z)^{-1}.
  \end{split}
    \end{align}
    We multiply by $X^{-s}$ from the right and take $\varepsilon\to 0$ observing   that   the
    resulting  right-hand side is explicitly
    in $ \vL(L^2)$. Note that the operator
    $\ad_B^{[s]+1}( X^{s})\in \vF$, and therefore it is a bounded
    multiplication operator. We can
    argue similarly for   $s< 0$. In particular we conclude
    \eqref{eq:basCom} (using \eqref{eq:hkappa}) as well as a boundedness result remarked  before
    Lemma \ref{lemma:Mourre2}. In fact it follows that
\begin{align}\label{eq:bascom}
  (\kappa
  B\pm \i)^{-1}\in \gO_{\kappa}(X^{0}).
\end{align}

\subStep{$g(B_{\kappa})$; 1st proof}
    In terms of an almost analytic
    extension $\tilde{ g}$ of a given $g\in C_\c^\infty(\R)$ (`absorbed'  in a  measure $\mu=\mu_g$ and writing $z=u+
\i v$)
    \begin{align}\label{eq:almostAnal}
      g(B)={\tfrac {1}{\pi }}\int _{{\C}
    }{\left(\bar{\partial} \tilde{g}\right)}{\left(
        z\right)}{\left(B -z\right)}^{-1}\d u\d v =\int_{\C} (B-z)^{-1}\d \mu(z).
    \end{align}

Let for any $t\in\R$
    \begin{align}
      \label{eq:clash} \vG_t=\set{g\in
C^\infty(\R)|\,\forall k\in\N_0:\,
      \abs{
g^{(k)}(x)}\leq  C_k \inp{x}^{t-k}}.
    \end{align}
 It is well-known that \eqref{eq:almostAnal} is valid in fact for  $g\in
    \vG_t$ with  $t<0$ (for any self-adjoint operator $B$).

Now for  any given $g\in
  C^\infty(\R)$ we write  $g(B_{\kappa})=g_\kappa (B)+g(0)$ with  $ g_\kappa\in\vG_{-1}$.
  By  using \eqref{eq:almostAnal}
for $ g_\kappa$
  we   conclude by using \eqref{eq:rescomm}  with $\varepsilon=0$
   (and its adjoint version)
 that
    $g(B_{\kappa})=g_\kappa (B)+g(0)\in \gO_{\kappa}(X^{0})$.

\subStep{$g(B_{\kappa})$; 2nd proof}  Since $B_{\kappa,R}$ is
    bounded we can write
    $g(B_{\kappa,R})=f(B_{\kappa,R})$ for some $f\in
    C_\c^\infty(\R)$  and  then apply \eqref{eq:almostAnal}. For
    $s\geq 0$ (treating only this case in \ref{item:23})
we expand as in \eqref{eq:rescomm}
\begin{align*}
      X^{s} R_\kappa(z)=\sum^{[s]}_{k=0}
  (-1)^k R_\kappa(z)^{k+1}\ad_{B_\kappa}^k( X^{s})+ (-1)^{[s]+1}R_\kappa(z)^{[s]+1}\ad_{B_\kappa}^{[s]+1}( X^{s})R_\kappa(z).
    \end{align*} It  suffices to show that
    \begin{align}\label{eq:X}
      \ad_{B_\kappa}^k( X^{s}) X^{k-s},k=0,\dots,[s],\mand
      \ad_{B_\kappa}^{[s]+1}( X^{s}) \text{ are bounded.}
    \end{align}
Using \eqref{eq:hkappa} and \eqref{eq:bascom} we compute for $k=1,\dots, [s]+1$
\begin{align}\label{eq:Ts's}
  \begin{split}
   \ad_{B_{\kappa}}^k(X^{s})&=\sum_{\sigma=(\sigma_1,\dots, \sigma_k)\in\set{-1,1}^k}T_{\sigma}v_{\sigma} X^{s-k}T_{\sigma};\\ T_{\sigma}&=\prod^{k}_{j=1} (\kappa
  B+\sigma_j \i)^{-1},\quad v_{\sigma}\in\vF.
  \end{split}
\end{align} Clearly this  shows that $ \ad_{B_\kappa}^{[s]+1}( X^{s}) $ is
bounded.
For $k\leq [s]$ we obtain the
first part of \eqref{eq:X} by using that $X^{s-k}T_{\sigma}X^{k-s}$ is
bounded, cf.  \eqref{eq:bascom}. Consequently we are done.

\subStep{$v$} We compute similarly (using the same notation)
\begin{align*}
  \ad_{B_{\kappa}}^k(v)=\sum_{\sigma\in\set{-1,1}^k}T_{\sigma}v_{\sigma} X^{-k}T_{\sigma}.
\end{align*} We conclude
that  $v\in \gO_{\kappa}(X^{0})$ using again \eqref{eq:bascom}.

\subStep{$P_1\Pi P_2$}  We let $T$ be of the form $T=P_1\Pi P_2$ and  compute
\begin{align*}
  \ad_{B_\kappa}(T)&=-\sum_{\sigma\in\set{-1,1}}(\kappa B+\sigma \i)^{-1}
      [T,B](\kappa B+\sigma\i)^{-1},\\
 \i[T,B]&= r^{-1/2}  \i[T,A]r^{-1/2}+2\Re\parbb{\i
  \comm[\Big]{T,r^{-1/2}}Ar^{-1/2} }\\
&= \i r^{-1/2}  \chi_1\parb{TA^{a_0}-A^{a_0}T}\chi_1r^{-1/2}+T_1+T_2+T_3;\\
&T_1=r^{-1/2}  \chi_2\i[T,r^{1/2}Br^{1/2}]r^{-1/2},\\
&T_2=r^{-1/2}  \chi_1\i[T,r^{1/2}Br^{1/2}]\chi_2r^{-1/2},\\
&T_3=2\Re\parbb{\i
  \comm[\Big]{T,r^{-1/2}}r^{1/2}B }.
\end{align*}

We need to compute  higher order commutators as well, so we need to
iterate this computation. We will demonstrate a self-similar structure
which will allow us to control higher order commutators. To this end
we introduce the graded
algebra $\vB=\sum_{k\in \N_0} \vB_{k}$, where $\vB_k$ is
given by  linear combinations   of any products of
\begin{align*}
  X^{-s}  \,\,(s\in \R), v\in \vF, (\kappa B\pm
\i)^{-1}, T, P_\sigma^\alpha, (P_\sigma^\alpha)^*;\quad  P^{\alpha}_\sigma:= p^\alpha
T_{\sigma}(p^2+1)^{-1}\,\,(|\alpha|\leq 2).
\end{align*}
  Here $T$ is any operator of the form $T=P_1\Pi P_2$ and $T_{\sigma}$
  is given by the last equation of \eqref{eq:Ts's} for   $k\in \N$  while $T_{\sigma}:=I$
  for $k=0$,  and for any  such  product the total sum  of  appearing
exponents $s\in \R$
is given by   $k$. We claim that
\begin{align}
  \label{eq:alg_k} \forall S\in  \vB_{k}:\quad \ad_{B_\kappa}(S)\in  \vB_{k+1}.
\end{align}  From this property  it follows that
$\ad_{B_{\kappa}}^k(S)\in \vB_{k}$  for any $S\in
\vB_{0}$ (in particular for $T$), and we readily see that the
powers of $X$ can be redistributed as we want, i.e. we obtain that
$\vB\subset \gO_{\kappa}(X^{0})$, in particular $T\in \gO_{\kappa}(X^{0})$ follows.

To show \eqref{eq:alg_k} it suffices to consider $k=0$ and then in
turn   only the operator $T=P_1\Pi P_2\in \vB_{0}$. So
we consider the above computation of the commutator with   $B_\kappa$.
The first term is clearly in $\vB_{1}$ since
$TA^{a_0}$  and $A^{a_0}T$ have the same form as $T$ and the two outer
factors of $r^{-1/2}$ together provides us with the extra factor
$X^{-1}$ required for the class $\vB_{1}$.

For the term
  $T_1$ and $T_2$ we also `undo' the commutation. Whence
  \begin{align*}
   -\i T_1=r^{-1/2}  \chi_2Tr^{1/2}B-B\chi_2 r^{1/2}Tr^{-1/2}+[B,\chi_2 ]r^{1/2}Tr^{-1/2}.
  \end{align*}
 We use \eqref{eq:regB} to combine the factors of $B$ with the factors
 of $(\kappa B+\sigma\i)^{-1}$. The remaining terms can   for any
 $s\in \R$ be written as
 $X^{-s}v_1 \widecheck Tv_2X^{-s}$, where $v_1 ,v_2 \in\vF$ and
 $\widecheck T$ has the same form as $T$  (cf. \ref{item:8} in the
proof of Lemma \ref{lem:Tcont}).  In particular we get the extra
 factor $X^{-1}$ by choosing $s=1/2$. We can argue similarly for $T_2$.

 For the term
  $T_3$ we note the following general formula. Let $g\in \vG_t$ for
  any $t\in \R$ and
  consider the commutation formula for the composition $g(r)$,
\begin{align}\label{eq:hminus}
  \begin{split}
\i\comm[\big]{{\Pi,g(r)}}&= \Pi h(x)-h(x)\Pi,\\
  h(x)&=\i  \int_0^1 \,\d t\,g' (r(tx^a+x_a))\int_0^t x^a\cdot (\nabla^2
  r)\parb{(sx^a+x_a)} x^a\,\d s.
  \end{split}
\end{align} Here we used two zero'th  order  Taylor expansions, the
fact that  $\comm{{\Pi,g(r(x_a))}}=0$ and the property $\nabla
  r\parb{x_a} \cdot x^a=0$, cf.  Lemma \ref{lemma:vector field}\ref{item:5}.

Applied to $ g(r)=r^{-1/2}$ we conclude that
\begin{align}\label{eq:45}
   \i\comm[\big]{{T,r^{-1/2}}}= T h(x)-h(x)T,
\end{align} with  $h$  given by \eqref{eq:hminus}. We can write
$h=\inp{x^{a_0}}^{7/2}r^{-3/2}v$, where $v\in \vF$.

This representation can
be refined using the formula
\begin{align}
  \label{eq:refi1} r\nabla^2 r= \nabla^2 r^2/2-\ket{\nabla
  r}\bra{\nabla r}\in \vF.
\end{align} In fact
it follows  from \eqref{eq:hminus}  and \eqref{eq:refi1} that
\eqref{eq:45} holds with
\begin{align}\label{eq:refi2}
  h =\inp{x^{a_0}}^{9/2}r^{-5/2}v,\text{ where }v\in \vF.
\end{align}
 Thus intuitively the commutator is two inverse powers  of $X$ better.
 Obviously this
 is a general property for commutation by $g(r)$ for $g\in \vG_t$,  and
 this   comes in handy
 later, see the proof
of Lemma \ref{lem:orderComm}.
In any case by the appearance of at least one extra power of $X^{-1}$ and
\eqref{eq:regB} we see that also $T_3$ contributes by a term in $ \vB_{1}$.

For an application in the next step of the proof  let us note that
\begin{align}
  \label{eq:pcom}
  S_1:=(p^2+1)\Pi (p^2+1)^{-1}\in
\vB_0\mand S_2:=(p^2+1)[\Pi,B_\kappa](p^2+1)^{-1}\in \vB_1.
\end{align} Obviously  $S_1\in
\vB_0$, and  therefore the  second assertion follows from
\eqref{eq:alg_k} and
\begin{align*}
 S_2=
 [S_1,B_\kappa]-[p^2,B_\kappa](p^2+1)^{-1}S_1+S_1[p^2,B_\kappa](p^2+1)^{-1}\in \vB_1.
\end{align*}

\subStep{$I_{0}(p^2+1)^{-1}$ and $p^\alpha T_{\sigma_1,\dots,
    \sigma_j} \breve
  R(z),\,p^\alpha T_{\sigma_1,\dots, \sigma_j}\Pi\breve
  R(z);\,|\alpha|\leq 2,\,j\in\N_0$} We add these operators as well
as  the adjoint expressions to the  list of generators of algebra $\vB$  and denote
the corresponding   algebra  by $\breve \vB=\sum_{k\in \N_0} \breve \vB_{k}$, where as
before the total sum  of  appearing
exponents $s\in\R$ in any term of an element of  $\breve \vB_{k}$
is given by   $k$. We claim that
\begin{align}
  \label{eq:alg_k2} \forall S\in  \breve\vB_{k}:\quad \ad_{B_\kappa}(S)\in  \breve\vB_{k+1}.
\end{align}  From this property it follows that
$\ad_{B_{\kappa}}^k(S)\in \breve\vB_{k}$  for any $S\in
\breve\vB_{0}$ (in particular for the enlisted operators and their  adjoint expressions), and again we can  redistribute  the
powers of $X$ as we want. Whence  we obtain that
$\breve\vB\subset \gO_{\kappa}(X^{0})$, in particular the fourth,
fifth and sixth operators
 listed in \eqref{eq:classP} are all in
$\gO_{\kappa}(X^{0})$ as claimed.
 In conclusion it suffices to show
\eqref{eq:alg_k2}, and therefore in turn \eqref{eq:alg_k2} with $k=0$.
As for the first term we note that  $[I_{0},B] \in
X^{-1}\vF$. Whence it  suffices to consider the other
terms.
For the second  term $S=p^\alpha T_{\sigma_1,\dots,
    \sigma_j}\breve R(z)=p^\alpha T_{\sigma}\breve R(z)$ we compute
\begin{align*}
 [S,B_\kappa]&= [p^\alpha,B_\kappa] T_{\sigma} \breve R(z)
                             - p^\alpha T_{\sigma}\breve R(z)[\breve
                             H,B_\kappa]\breve R(z)\\&=T_1+T_2+T_3+T_4;\\
T_1&=-\sum_{\sigma'\in\set{-1,1}}(\kappa B+\sigma' \i)^{-1}
      [p^\alpha,B]T_{\sigma_1,\dots,
    \sigma_j,\sigma'}\breve R(z),\\
T_2&=- p^\alpha T_{\sigma}\breve R(z)[
                             H,B_\kappa]\breve R(z),\\
T_3&=\lambda_0 p^\alpha T_{\sigma}\breve R(z)[
                             \Pi,B_\kappa]\breve R(z),\\
T_4&= p^\alpha T_{\sigma}\breve R(z)[
                             \breve K,B_\kappa]\breve R(z).
\end{align*}
Since $ [p^\alpha,B]\in X^{-1}\sum_{|\beta|\leq 2}v_\beta p^\beta$
where $v_\beta\in \vF$, the operator $T_1\in\breve\vB_{1}$. Thanks to
\eqref{eq:limit} we conclude  that also $T_2\in\breve\vB_{1}$. Since
$\Pi\in\vB_{0}$ and therefore $[\Pi, B_\kappa]\in\vB_1$ (due to
\eqref{eq:alg_k}) also $T_3\in\breve\vB_{1}$. For $T_4$ we
combine the facts that
\begin{align*}
  I_{0}(p^2+1)^{-1},(p^2+1)\breve
                              R(z)\in \breve \vB_0
\end{align*}
 and \eqref{eq:pcom}. Thus, for
example (with $S_2$ given in \eqref{eq:pcom})
 \begin{align*}
   [I_{0} \Pi, B_\kappa]\breve R(z)-[I_{0} ,B_\kappa]\Pi\breve
                            R(z)
                            =\parb{I_{0}(p^2+1)^{-1}}S_2\parb{(p^2+1)\breve
                              R(z)}\in \breve\vB_1.
 \end{align*}
 Clearly the second term to the left is in $\breve\vB_1$ too,
 completing our treatment of one of the three terms  of $\breve K$. This
 argument also works for the  other terms of $\breve K$.

Finally for the third term $S= p^\alpha T_{\sigma}\Pi\breve
  R(z)$ we write
\begin{align*}
 S= P_\sigma^\alpha S_1 \parb{(p^2+1)\breve R(z)}.
\end{align*} Since $ P_\sigma^\alpha,S_1\in \vB_0$, cf.
\eqref{eq:pcom}, $[P^\alpha_\sigma S_1, B_\kappa]\in \vB_1\subset
\breve\vB_1$, and since  $[(p^2+1)\breve R(z), B_\kappa]\in \breve\vB_1$ (as  proved
above), it follows that
indeed $[S, B_\kappa]\in \breve\vB_1$ as we want.

\subStep{$f(\breve H)$} We use \eqref{eq:almostAnal} with $B$ replaced
by $\breve H$. We have proven that
$R(z)= \vO_{\kappa}(X^{0})$ for any fixed non-real $z$. Whence we can compute any number of commutators with
    $B_\kappa$, redistribute powers of $X$ as we want, and then estimate
    to see that indeed  $f(\breve H)\in \vO_{\kappa}(X^{0})$. Note that
    the bounds $\norm{\breve
    R(z)^m} \leq |\Im z|^{-m}$ come in naturally estimating
    integrals like  \eqref{eq:almostAnal} and in fact appear
    `harmless' in combination with the appearing
    measure $\mu$.
\end{proof}

\subsection{Computing a commutator}\label{subsec:Calculus}
We are interested in computing  commutators $\i [\breve H, P]$ where
$P$  in all relevant cases has the form
\begin{align*}
  P=f(\breve H)h(r)g(B_\kappa)h(r)f(\breve H);\quad f\in
    C_\c^\infty(\R),
\, h,g\in
    C^\infty(\R), \text{ real-valued}.
\end{align*}  (We suppress the dependence of $R$.) Let $\check f(\lambda)=\lambda
    \tilde f(\lambda)$ for any  real-valued $\tilde f\in
    C_\c^\infty(\R)$  chosen such that  $\tilde f =1$ on
    the support of $f$. Let $\widecheck H=\check f(\breve H)$ and  denote the
    corresponding  Heisenberg derivative $\i[\widecheck  H, \cdot ]$   by
    ${\widecheck \bD}$, in particular  $\i [\widecheck H, P]= {\widecheck \bD} P$.
      We are interested in specific choices of $h$ and $g$,
    and for those choices $P$ has a $\kappa$-order, say $2t$.
    More precisely   we impose  the condition $h\in
    \vG_t$, cf. \eqref{eq:clash}, and conclude that
     the composed function  $h(r)$ has  $\kappa$-order $t$, so  that indeed $P$ has a
    $\kappa$-order $2t$ (note  that due to
    Lemma \ref{lem:orderP} the $\kappa$-order of $g(B_\kappa)$ and
    $\widecheck H$ are
    zero). This implies that
    the commutator ${\widecheck \bD} P$ has order  $2t$, possibly (and
    indeed) smaller. We need to
    compute the order more exactly along with  the leading order term.

    \begin{lemma}
      \label{lem:orderComm} Let $R\geq 1$ be given big enough and  suppose $h\in
    \vG_t$. Then
\begin{subequations}
\begin{align}
  \label{eq:lead1}
  \begin{split}
  &{\widecheck \bD} P-L_1-L_2\in
\gO_{\kappa}(X^{2t-2});\\
&L_1=4f(\breve H)\Re \parb{(\kappa^2B^2+1)h'(r)B_\kappa g(B_\kappa)h(r)}f(\breve H),\\
&L_2=f(\breve H)h(r)\Re \parb{g'(B_\kappa){\widecheck \bD B_\kappa}}h(r)f(\breve H).
  \end{split}\end{align} The operators $L_1, L_2\in\gO_{\kappa}(X^{2t-1})$, so in particular ${\widecheck \bD} P\in\gO_{\kappa}(X^{2t-1})$.

 Suppose in addition that $g'\geq 0$ and that
$ (g')^{1/2}\in C^\infty$. Then $L_2$ can be replaced by any of the
following expressions $L_3,L_4,L_5\in\gO_{\kappa}(X^{2t-1})$ for which $\textrm{Com}=f(\breve H)\i[\breve H,
    A_\kappa]f(\breve H)$ is defined in agreement with Lemma \ref{lemma:Mourre2}.
\begin{align} \label{eq:lead2}
  L_3&=h(r)\sqrt{g'(B_\kappa)}f(\breve H)\i[\breve H,
    B_\kappa]f(\breve H)\sqrt{g'(B_\kappa)}h(r),\\
 \label{eq:lead3}
  \begin{split}
  L_4&=h(r)\sqrt{g'(B_\kappa)}X^{-1/2}\textrm{Com}\,X^{-1/2}\sqrt{g'(B_\kappa)}h(r)\\
    &-2h(r)\sqrt{g'(B_\kappa)}X^{-1/2}f(\breve
    H)(\kappa^2B^2+1)B_\kappa^2f(\breve
    H)X^{-1/2}\sqrt{g'(B_\kappa)}h(r),
\end{split}\\
\label{eq:lead4}
\begin{split}
  L_5&=h(r)X^{-1/2}\sqrt{g'(B_\kappa)}\textrm{Com}\,\sqrt{g'(B_\kappa)}X^{-1/2}h(r)\\
    &-2h(r)X^{-1/2}\sqrt{g'(B_\kappa)}f(\breve
    H)(\kappa^2B^2+1)B_\kappa^2f(\breve
    H)\sqrt{g'(B_\kappa)}X^{-1/2}h(r).
\end{split}
\end{align}
\end{subequations}
\end{lemma}
\begin{proof} We write
  \begin{align}\label{eq:firprime}
    \begin{split}
    {\widecheck \bD} P&=L'_1+L'_2;\\
L'_1&=2f(\breve H)\Re \parbb{\parb{\widecheck \bD h(r)}g(B_\kappa)h(r)}f(\breve H),\\
L'_2&=f(\breve H)h(r)\parb{{\widecheck \bD g(B_\kappa)}}h(r)f(\breve H).
    \end{split}
\end{align} By \eqref{eq:almostAnal}
\begin{subequations}
\begin{align}\label{eq:hDEr1}
  \widecheck \bD h(r)&=-\int_{\C} \breve R(z)\i[\breve H,h(r)]\breve
  R(z)\d \mu_{\check f}(z);\\
\i[\breve H,h(r)]&= 2\Re \parb{h'(r)B}-\lambda_0\i[\Pi,h(r)]-\i[\breve K,h(r)].\label{eq:hDEr2}
\end{align}
\end{subequations} We insert \eqref{eq:hDEr2} into \eqref{eq:hDEr1} obtaining
then three terms, say $T_1, T_2, T_3$. We can write  (seen by commutation)
\begin{align*}
   T_1=2\check f'(\breve H)h'(r)B+ O_{\kappa}(X^{t-2})\in  \gO_{\kappa}(X^{t-1}).
\end{align*} The terms $ T_2$ and $ T_3$ are treated using
\eqref{eq:hminus}  and \eqref{eq:refi1}  and Lemma \ref{lem:orderP}
(cf. the proof of discussion after
\eqref{eq:refi1}) to  obtain
\begin{align*}
   T_2, \, T_3\in \gO_{\kappa}(X^{t-2}).
\end{align*}
 Next we insert the resulting formula
 \begin{align}\label{eq:hD1}
   {\widecheck \bD} h=2\check f'(\breve H)h'(r)B+  \vO_{\kappa}(X^{t-2})\in\gO_{\kappa}(X^{t-1})
 \end{align} into the expression $L'_1$, and   we see that
   $L'_1=L_1+\vO_{\kappa}(X^{2t-2})$ (note that $f\check f'=f$). Clearly
$L_1\in \gO_{\kappa}(X^{2t-1})$.

As for $L'_2$ we  can assume that $g\in
    C_\c^\infty(\R)$ (since $B_\kappa$ is bounded we can truncate  $g$ outside
$\sigma(B_\kappa)$)  and use \eqref{eq:almostAnal}  again
\begin{align}\label{eq:hD2}
  \begin{split}
  \widecheck \bD g(B_\kappa)&=-\int_{\C} R_\kappa(z)\parb{\widecheck \bD B_\kappa}
  R_\kappa(z)\d \mu_g(z)\\
&=g'(B_\kappa)\widecheck \bD B_\kappa+\int_{\C} R_\kappa(z)^2[{\widecheck \bD B_\kappa},B_\kappa]
  R_\kappa(z)\d \mu_g(z)\\
&=g'(B_\kappa)\widecheck  \bD B_\kappa+\vO_{\kappa}(X^{-2})\in\gO_{\kappa}(X^{-1});
  \end{split}
\end{align} for the last identity we used  Lemma \ref{lem:orderP} (and its
proof).  Obviously  the first term to the right can be replaced by its
 real part. We
conclude  that $L'_2=L_2+\vO_{\kappa}(X^{2t-2})$. Whence $L'_1+L'_2=L_1+L_2+\vO_{\kappa}(X^{2t-2})$, showing
\eqref{eq:lead1}. Clearly
$L_2\in \gO_{\kappa}(X^{2t-1})$.

Next to show that $L_3-L_2\in\gO_{\kappa}(X^{2t-2})$   under the extra
conditions on $g$ we first note that the    truncation of $g$ outside
$\sigma(B_\kappa)$,  to make
it compactly supported, obviously applies  to $ (g')^{1/2}$ as well making again
\eqref{eq:almostAnal} applicable. Whence we can
symmetrize $ L_2$, so that it suffices  to show that
\begin{align}
  \label{eq:2=3}
  \begin{split}
  \bar L_2-\bar L_3&\in \gO_{\kappa}(X^{2t-2});\\
\bar L_2&=f(\breve H)h(r)\sqrt{g'(B_\kappa)}\parb{{\widecheck \bD
  B_\kappa}}\sqrt{g'(B_\kappa)}h(r)f(\breve H),\\
\bar L_3 &=h(r)\sqrt{g'(B_\kappa)}f(\breve H)\parb{{\widecheck \bD
  B_\kappa}}f(\breve H)\sqrt{g'(B_\kappa)}h(r).
  \end{split}
\end{align} (Note that $L_3=\bar L_3$.)  Using \eqref{eq:hD1} and
\eqref{eq:hD2} with $\check f$ replaced by $f$ we see that
\begin{align*}
  [f(\breve H),h(r)]\in\gO_{\kappa}(X^{t-1}),\,[f(\breve H),\sqrt{g'(B_\kappa)}]\in\gO_{\kappa}(X^{-1}),
\end{align*} which shows \eqref{eq:2=3}.
By this  proof we see
  that $L_3\in \gO_{\kappa}(X^{2t-1})$. Alternatively we may easily  check
  the latter property directly. (We can verify  similarly that  $L_4,
  L_5\in \gO_{\kappa}(X^{2t-1})$.)

To show that  $L_4-L_3\in\gO_{\kappa}(X^{2t-2})$  it suffices to show
that
\begin{align} \label{eq:3=4}
  \begin{split}
   X^{1/2}f(\breve H)&\i[\breve H,
    B_\kappa]f(\breve H) X^{1/2}-f(\breve H)\i[\breve H,
    A_\kappa]f(\breve H)\\&+2f(\breve H)(\kappa^2B^2+1)B_\kappa^2f(\breve
    H)\in \gO_{\kappa}(X^{-1}).
  \end{split}
\end{align} For the  first term we substitute
\begin{align*}
  \i[\breve H,
    B_\kappa]=\i[ H,
    B_\kappa]-\lambda_0\i[\Pi,
    B_\kappa]-\i[\breve K,
    B_\kappa].
  \end{align*} As we saw in the last part of the proof of Lemma
  \ref{lem:orderP} each of these commutators contributes  by an operator in
  $\breve\vB\subset\gO_{\kappa}(X^{0})$. This remains valid if we
  replace the term by
  \begin{align*}
    f(\breve H)X^{1/2}\i[\breve H,
    B_\kappa] X^{1/2}f(\breve H),
  \end{align*} and indeed the error (given by commuting the factors of $X^{1/2}$
  and $f(\breve H)$)  contributes by an operator  in
  $\gO_{\kappa}(X^{-1})$. On the other hand
  \begin{align*}
    f(\breve H)\parbb{X^{1/2}&\i[\breve H,
    B_\kappa] X^{1/2}-\i[\breve H,
    A_\kappa]+2(\kappa^2B^2+1)B_\kappa^2}f(\breve H)\\
&=f(\breve H)\parbb{-2\Re\parb{\i[\breve H,
    X^{1/2}] B_\kappa X^{1/2}}+2(\kappa^2B^2+1)B_\kappa^2}f(\breve H)\\
&=f(\breve H)\parbb{-\Re\parb{
    X^{-1/2}2B B_\kappa X^{1/2}}+2(\kappa^2B^2+1)B_\kappa^2}f(\breve
H)+O_{\kappa}(X^{-1})\\
&\in \gO_{\kappa}(X^{-1}).
\end{align*} Thus \eqref{eq:3=4} is proven.

Noting that
$\textrm{Com}\in \gO_{\kappa}(X^{0})$ and
$\big [\sqrt{g'(B_\kappa)},X^{-1/2}\big ]\in \gO_{\kappa}(X^{-3/2})$ it follows
that $L_5-L_4\in\gO_{\kappa}(X^{2t-2})$, which  finishes  the proof.
\end{proof}
\begin{remark}\label{remark:computing-commutator}
 The formula \eqref{eq:hD2} is an example of the following  general commutator
 expansion formula for operators in $ \gO_{\kappa}(X^{t})$, cf. \cite{GIS}. Thus for
 given  $S\in \gO_{\kappa}(X^{t})$, $t\in\R$, $g\in
C^\infty _\c(\R)$ and
$K\geq 1$
  \begin{align}\label{eq:hD22}
  \begin{split}
 [S,g(&B_\kappa)]\\=&-\int_{\C} R_\kappa(z)\ad_{B_\kappa}(S)R_\kappa(z)\d \mu_g(z)\\
=&\sum^{K}_{k=1}
  \tfrac{(-1)^{k-1}} {k!}g^{(k)}(B_\kappa)\ad_{B_\kappa}^k(
  S)+R(K);\\
&R(K)=(-1)^{K+1}\int_{\C}R_\kappa(z)^{K+1}\ad_{B_\kappa}^{K+1}(
  S)R_\kappa(z)\d \mu_g(z)\in \gO_{\kappa}(X^{t-K-1}).
  \end{split}
\end{align}

We will need  better  control of the error terms of   Lemma
\ref{lem:orderComm}  in  the application in  Subsection
\ref{subsec:Microlocal  bounds and LAP}. Here we explain a procedure for establishing this. As noted before the function $g$
in the first part of the lemma  may be taken  compactly
 supported  as required for
\eqref{eq:hD22}.
A principal  goal is  to  refine   the right-hand side of the statement
\begin{align*}
  {\widecheck \bD} P-L_1-L_5=(L'_1-L_1)+(L_2'-L_5)\in
\gO_{\kappa}(X^{2t-2}).
\end{align*}  Due to  \eqref{eq:hD1} we can
write
\begin{align}\label{eq:L145}
  \begin{split}
 L'_1&=4f(\breve H)\Re {\parb {h'(r)Bg(B_\kappa)h(r)}}f(\breve H)\\
&+ f(\breve H)\Re {\parb{\vO_{\kappa}(X^{t-2})g(B_\kappa)h(r)}}f(\breve
  H)\\
&=L_1+f(\breve H)\Re {\parb{\vO_{\kappa}(X^{t-2})g(B_\kappa)h(r)}}f(\breve
  H).
  \end{split}
\end{align} For $L'_2$ we apply \eqref{eq:hD22} to $S= \check f(\breve H)$,
yielding (for $K\geq 2$)
\begin{align}\label{eq:hD22s}
 \begin{split}
  L'_2&= f(\breve H)h(r)\parb{{\widecheck \bD g(B_\kappa)}}h(r)f(\breve
  H)\\
&=\i\sum^{K}_{k=1} \tfrac{(-1)^{k-1}} {k!}\,f(\breve H)h(r)
 g^{(k)}(B_\kappa)\ad_{B_\kappa}^k(
  \check f(\breve H))h(r)f(\breve
  H)+O_{\kappa}(X^{2t-K-1})\\
&= f(\breve H)h(r)
 g'(B_\kappa)\parb{\widecheck\bD B_\kappa} f(\breve H))h(r)f(\breve
  H)+\sum^{K}_{k=2} S_k+O_{\kappa}(X^{2t-K-1}).
\end{split}
\end{align} The terms $S_k\in \gO_{\kappa}(X^{2t-k})$, and they have an explicit form
suitable for the induction argument of  Subsection \ref{subsec:Microlocal
  bounds and LAP}.
A  consequence of \eqref{eq:hD22s} is the refined formula
\begin{align}\label{eq:L1452}
  \begin{split}
 L'_2=L_2+ \sum^{K}_{k=2} \Re(S_k)+O_{\kappa}(X^{2t-K-1}).
  \end{split}
\end{align}

\end{remark}

\subsubsection{Smooth sign function}\label{subsubsec:Smooth sign
  function} In some of our applications of Lemma \ref{lem:orderComm}
the function $g$ will be a `smooth sign function' $\zeta_\epsilon$, cf.   \cite{AIIS}.  It is
constructed in terms of a cut-off function $\eta_\epsilon\in
C^\infty(\R)$ with special properties: The parameter $\epsilon>0$ is
considered small, and we define $\eta_\epsilon(b)=\tfrac
1\epsilon\eta(\tfrac b\epsilon)$, where $\eta'(b)> 0$ for $|b|<1$,
$\eta(b)=0$ for $b\leq -1$ and $\eta(b)=1$ for $b\geq 1$. We can
choose $\eta$ such that $\eta'$ is even, $\sqrt \eta,\sqrt{\eta'}\in
C^\infty(\R)$ and for some $c>0$
\begin{align}
  \label{eq:constr}
  \eta'(b)\geq c\,\eta(b)\text{ for }b\in(-1,1/2].
\end{align} The optimal  choice of such $c$  is not important for us
since we will only need \eqref{eq:constr}  in the following disguised
form: For any $\tilde c >0$ and all $\epsilon$ small enough
($\epsilon^2\leq \tfrac 23 c\tilde c$ suffices)
\begin{subequations}
\begin{align}
  \label{eq:fundest}
  (\tfrac \epsilon2-b)\eta_\epsilon(b)\leq \tilde c
  \,\eta'_\epsilon(b)\text{ for all }b\in \R.
\end{align}
Note also that since  $\eta_\epsilon'$ is even
  \begin{align}\label{eq:partuni}
    1=\epsilon\eta_\epsilon(b)+\epsilon\eta_\epsilon(-b).
  \end{align}
  Let $\zeta_\epsilon(b)=\eta_\epsilon(b)-\eta_\epsilon(-b)$.
\end{subequations}

\section{Positive commutator estimates}\label{subsec:Positive
  commutator estimates}

We shall prove properties of possibly existing eigenfunctions of
$\breve H$ at $\lambda_0$. In the case where they dont exist we shall
prove various resolvent estimates of $\breve H$ near $\lambda_0$. Our
analysis is based on Lemma \ref{lem:orderComm} and positive commutator
methods of \cite{IS2}, \cite{AIIS} and \cite{GIS}. For convenience we
abbreviate in this section the Besov spaces with index $s=1/2$ as
$\vB:=\vB_{1/2}$, $\vB^*:=\vB^*_{1/2}$
${\vB}_0^*:=\vB^*_{1/2,0}$. (Note however that the same notation is
used with a slightly different meaning in Subsection \ref{subsec:LAP
  bound}.)

\subsection{A Rellich type  theorem}\label{subsec:Rellich's theorem}
 As a first application
 we show the following result, largely mimicking \cite{IS2}.
\begin{thm}\label{thm:priori-decay-b_0} Every
  generalized eigenfunction  in ${\vB}_0^*$ of $\breve H$  at
  $\lambda_0$, or at any sufficiently nearby real $\lambda'_0$, is in $L^2_\infty$.
\end{thm}
\begin{proof}
  We shall use the Mourre estimate \eqref{eq:43breve} with
  $\lambda=\lambda_0$, say with the positive number $\epsilon=\tilde d(\lambda_0)$ and $R=R_0\geq 1$
  fixed sufficiently big. In agreement with the statement $\kappa_0>0$
   is fixed too (small), and we can freely use  the estimate for
   $\kappa\in (0,\kappa_0]$ and for a fixed (small)  neighbourhood $\vU$ of
   $\lambda_0$ and a fixed compact $K$. Whence we have freedom to
   choose  $\kappa>0$ very small whenever conveniently in the proof
   (note that  $\vU$  and $K$ do not depend on $\kappa$). Now suppose
   $\phi\in {\vB}_0^*$ and $(\breve H-\lambda'_0)\phi=0$ with
   $\abs{\lambda_0-\lambda_0'}$ small, then we can write
   $\phi=f(H)\phi$ where $f\in C^\infty_\c(\vU )$, $f(\lambda'_0)=1$
   and  $f$
   is real.

   Let $y_\varepsilon(r)=\tfrac 1 {1+\varepsilon r}$ and
   $x_\varepsilon(r)=ry_\varepsilon(r)$ for $\varepsilon\in [0,1]$,
   and let $X_\varepsilon$ and $Y_\varepsilon$ be the operators of
   multiplication by $x_\varepsilon$ and $y_\varepsilon$, respectively
   (this notation is consistent with \eqref{eq:41delta}). These
   quantities are henceforth used with $R=R_0$ only.  We note that
   $X:=X_0$ agrees with  the notation used in Section \ref{subsec:Higher
     commutators} and that $X_\varepsilon=X Y_\varepsilon$. Note also
   that $\nabla x_\varepsilon(r)= y_\varepsilon^2\omega$, whence for
   example $\i[H,X_\varepsilon]={2Y_\varepsilon}B{Y_\varepsilon}$.

It is convenient to introduce
the following terminology for families
$(T_\varepsilon)_{0<\varepsilon\leq1}\subset \gO_\kappa(X^{t})$ of
operators   on $ L^2$,   $t\in \R$.  We say $
(T_\varepsilon)$ is \emph{uniformly of $\kappa$-order $t$} if
\begin{align*}
\forall \kappa \in (0,\kappa_0']\;\forall s\in \R\,\;\forall k\in
  \N_0:\quad
  \sup_{\varepsilon\in(0,1]}\norm{X^{s-k-t}\ad_{B_{\kappa,R_0}}^k(T_\varepsilon) X^{-s}}<\infty.
\end{align*}  We shall allow ourselves to write
 $T_\varepsilon=\vO_{
\textrm {unf}}(X^{t})$ and  $T_\varepsilon\in \gO_{
\textrm {unf}}(X^{t})$ to symbolize  that the operator $T_\varepsilon$ is member of
a family of operators $(T_\varepsilon)$ which is uniformly of $\kappa$-order $t$. If $(S_\varepsilon)$ and
$(T_\varepsilon)$ are uniformly of $\kappa$-order $s$ and
$t$, respectively, then $(S_\varepsilon T_\varepsilon)$ is uniformly
 of $\kappa$-order $s+t$.

\subStep {I}    We    show
  that $\phi\in L^2_{-1/2}$.
Fix  any $\delta\in(0,1/2)$. We shall consider the `propagation observable'
    \begin{align*}
      P_{\varepsilon}=f(\breve H)
      X_\varepsilon^{\delta}\zeta_\epsilon(B_\kappa)X_\varepsilon^{\delta}f(\breve
      H),\,\varepsilon\in(0,1].
\end{align*} Clearly $P_\varepsilon\in\gO_{
\textrm {unf}}(X^{2\delta})$. The positive parameter   $\epsilon$
used
 here will be
 fixed shortly,    small enough. Note that $X_\varepsilon$ and  $P_\varepsilon$ are  bounded due to the appearance of the
factor $y_\varepsilon$. Eventually this factor will  be removed by
letting $\varepsilon\to 0$.
 More precisely we  shall demonstrate some `essential  positivity'
 of $\i[\breve H, P_\varepsilon]$ persisting in
 the
 $\varepsilon\to 0$ limit. For any  $n\in\N$ the function
 $\phi_n=\chi_n(r)\phi\in H^1$ (cf. the notation \eqref{eq:14.1.7.23.24}), $(\breve H-\lambda'_0)\phi_n=
 [\breve H, \chi_n]\phi$ and whence the expectation (we use in general
 the notation $\inp{T}_\phi=\inp{\phi,T\phi}$)
 \begin{subequations}
 \begin{align}\label{eq:virial3}
   \inp{\i[\breve H, P_\varepsilon]}_{\phi_n}=-2\Re\inp{ \i[\breve H, \chi_n] P_\varepsilon\chi_n}_\phi.
 \end{align}
 Since $\phi\in {\vB}_0^*$ the term to the right vanishes as
 $n\to \infty$. Writing $\breve H=H-\lambda_0\Pi- \breve K $ this
 claim is obvious for the contribution  from $H$. For the other terms
 the commutator is effectively of order $X^{-2}$,
 cf. \eqref{eq:hminus} and the subsequent discussion, whence their contributions vanish too as
 $n\to \infty$.
It remains to study the left-hand side of \eqref{eq:virial3}  in this
 limit.

Let
 $S_\varepsilon={Y^2_\varepsilon}X_\varepsilon^{\delta-1}\left(\in \gO_{
\textrm {unf}}(X^{\delta-1})\right)$ and
$\theta_\epsilon=\sqrt{\eta'_\epsilon}$. We compute using Lemma \ref{lem:orderComm}
 \begin{align*}
    \inp{\i[\breve H, P_\varepsilon]}_{\phi_n}&=\inp{\widecheck \bD  P_\varepsilon}_{\phi_n}=\inp{L_1+L_5+\vO_{
\textrm {unf}}(X^{2\delta-2})}_{\phi_n};\\
L_1&=4\delta f(\breve H)\Re \parb{(\kappa^2B^2+1)S_\varepsilon
  B_\kappa \zeta_\varepsilon(B_\kappa)X_\varepsilon^{\delta}}f(\breve H),\\
X^{1/2}X_\varepsilon^{-\delta}L_5X_\varepsilon^{-\delta}X^{1/2}&=\theta_\epsilon(B_\kappa)f(\breve H)\i[\breve H,
    A_\kappa]f(\breve
    H)\theta_\epsilon(B_\kappa)\\
&+\theta_\epsilon(-B_\kappa)f(\breve H)\i[\breve H,
    A_\kappa]f(\breve H)\theta_\epsilon(-B_\kappa)\\
&-2\theta_\epsilon(B_\kappa)f(\breve
H)(\kappa^2B^2+1)B_\kappa^2f(\breve
H)\theta_\epsilon(B_\kappa)\\
&-2\theta_\epsilon(-B_\kappa)f(\breve H)(\kappa^2B^2+1)B_\kappa^2f(\breve H)\theta_\epsilon(-B_\kappa).
 \end{align*} Noting the essential positivity of the term containing the factor
 $\kappa^2B^2$ we obtain  after symmetrizing
 \begin{align}\label{eq:libnd}
   L_1\geq 4\delta f(\breve
   H)Y_{\varepsilon}X_\varepsilon^{\delta-1/2} B_\kappa
   \zeta_\varepsilon(B_\kappa)X_\varepsilon^{\delta-1/2}Y_{\varepsilon}f(\breve
   H)+O_{ \textrm {unf}}(X^{2\delta-2}).
 \end{align}  Using  the Mourre estimate of the
 form discussed in the beginning of the proof and  factors of $\tilde f(\breve H)$
 (the latter can be inserted for free to bound
 $\kappa^2B^2$), we can estimate for $\kappa>0$ sufficiently small
\begin{align*}
L_5&\geq X_\varepsilon^{\delta}X^{-1/2}\theta_\epsilon(B_\kappa)f(\breve
H)\parb{\tilde d(\lambda_0)-K}f(\breve
    H)\theta_\epsilon(B_\kappa)X^{-1/2}X_\varepsilon^{\delta}\\
&+X_\varepsilon^{\delta}X^{-1/2}\theta_\epsilon(-B_\kappa)f(\breve H)\parb{\tilde d(\lambda_0)-K}f(\breve H)\theta_\epsilon(-B_\kappa)X^{-1/2}X_\varepsilon^{\delta}\\
&-4\epsilon^2f(\breve
H)X_\varepsilon^{\delta}X^{-1/2}\parbb{\eta'_\epsilon(B_\kappa)+\eta'_\epsilon(-B_\kappa)}X^{-1/2}X_\varepsilon^{\delta}f(\breve
H)+\vO_{
\textrm {unf}}(X^{2\delta-2}).
\end{align*}

We shall require  that $\epsilon>0$ is so small  that
  $8\epsilon^2< \tilde d(\lambda_0)$, implying
   $\tilde d(\lambda_0)-4\epsilon^2>  \tilde d(\lambda_0)/2$.
We shall use  \eqref{eq:fundest} with $\tilde c=\tfrac{\tilde d(\lambda_0)}{8\delta}$ and
  any  possibly smaller
 $\epsilon$, henceforth considered fixed.
 Now we fix   a  big $m\in \N $ such that (with this
$\epsilon$)
\begin{align*}
 \tilde d(\lambda_0)-4\epsilon^2-\|K-\chi_{m}K\chi_{m}\|\geq \tilde d(\lambda_0)/2,
\end{align*}
  and  note  that the  contribution from
 the operator $\chi_{m}K\chi_{m}$ is in $\gO_{
\textrm {unf}}(X^{2\delta-2})$.

Then  we  estimate
\begin{align*}
L_5&\geq
2^{-1}\tilde d(\lambda_0)f(\breve
H)X_\varepsilon^{\delta}X^{-1/2}\parbb{\eta'_\epsilon(B_\kappa)+\eta'_\epsilon(-B_\kappa)}X^{-1/2}X_\varepsilon^{\delta}f(\breve
H)+\vO_{\textrm {unf}}(X^{2\delta-2})\\
&\geq 4\delta \tilde{c}f(\breve
H)Y_{\varepsilon}X_\varepsilon^{\delta-1/2}\parbb{\eta'_\epsilon(B_\kappa)+\eta'_\epsilon(-B_\kappa)}X_\varepsilon^{\delta-1/2}Y_{\varepsilon}f(\breve
H)+\vO_{\textrm {unf}}(X^{2\delta-2}).
\end{align*}

We conclude the
following lower
 bound  by combining this bound with  \eqref{eq:libnd} and by  using
 \eqref{eq:fundest} and \eqref{eq:partuni},
\begin{align*}
 L_1+L_5&\geq
 2\delta \epsilon f(\breve
   H)Y_{\varepsilon}X_\varepsilon^{\delta-1/2}\parb{\eta_\epsilon(B_\kappa)+\eta_\epsilon(-B_\kappa)}{X_\varepsilon}^{\delta-1/2}Y_{\varepsilon}f(\breve
   H)+\vO_{\textrm
     {unf}}(X^{2\delta-2})\\
&=2\delta f(\breve
   H)X_\varepsilon^{2\delta-1}Y^2_{\varepsilon}f(\breve
   H)+\vO_{\textrm
     {unf}}(X^{2\delta-2}).
\end{align*}
\end{subequations}

Whence we  obtain from these
arguments  the uniform bound
\begin{align}\label{eq:fundests}
  \|X_\varepsilon^{\delta-1/2}Y_{\varepsilon}\phi\|^2=\lim_{n\to \infty}\|X_\varepsilon^{\delta-1/2}Y_{\varepsilon}f(\breve
   H)\phi_{n}\|^2\leq C\|X^{\delta-1}\phi\|^2.
\end{align}
By letting $\varepsilon \to 0$ in \eqref{eq:fundests} it follows that
$\phi\in L^2_{\delta-1/2}\subset L^2_{-1/2}$.

\subStep {II}   We show that
$\phi\in L_\infty^2$ by a bootstrap argument.  So suppose we have shown that
$\phi\in L^2_{(m-1)/2-1/2}$ for an  $m\in \N$. We did show this for
$m=0$ in Step \textit{I}. Then we come to the
conclusion that $\phi\in L^2_{m/2-1/2}$ by repeating the previous
procedure  using now the observable
    \begin{align*}
      P=P_{\varepsilon}=f(\breve H)
      X_\varepsilon^{m/2}\zeta_\epsilon(B_k)X_\varepsilon^{m/2}f(\breve
      H),\,\varepsilon>0,
    \end{align*} leading to the bound
\begin{align}\label{eq:fundestsB}
  \|Y_{\varepsilon}X_\varepsilon^{m/2-1/2}\phi\|^2=\lim_{n\to
    \infty}\|Y_{\varepsilon}X_\varepsilon^{m/2-1/2}f(\breve
   H)\phi_{n}\|^2\leq C\|X^{(m-1)/2-1/2}\phi\|^2.
\end{align}
By letting $\varepsilon \to 0$ we deduce that $\phi\in
L^2_{m/2-1/2}$.

  \end{proof}

\subsection{LAP bound}\label{subsec:LAP bound}

  We show   a Besov space   limiting
 absorption principle  bound,  largely mimicking  \cite{AIIS}.

\begin{thm}\label{thmlapBnd}  Suppose $\lambda_0$ is not an eigenvalue
  of  $\breve H$. Then there exist a neighbourhood $I\subset \R$ of
  $\lambda_0$ and   $C>0$ such that for all for $z\in \C\setminus\R$ with $\Re
  z\in I$ and all $\psi\in \vB$
\begin{align}\label{eq:lap-besov-spacebnd}
  \norm{\breve R(z)\psi}_{\vB^*}\leq C\norm{\psi}_{\vB}.
\end{align}
\end{thm}

To prove this bound we shall use the  following weight-functions parametrized by
$\nu\in\N_0$ and defined on $\R_+$
\begin{align*}
\Theta=\Theta_\nu(r)=1-\parb{1+r/2^\nu}^{-1}.
\end{align*} Noting the formula for derivatives
\begin{align*}
\Theta^{(k)}=(-1)^{k-1}k!2^{-k\nu}\parb{1+r/2^\nu}^{-1-k};\quad k\geq1,
\end{align*} we obtain the  bounds
\begin{align}\label{eq:elebnd}
  \begin{split}
& 0<\Theta\le \min\{1,r/2^\nu\},\\
&0< (-1)^{k-1}\Theta^{(k)}\le k(-1)^{k}r^{-1}\Theta^{(k-1)}\leq k!r^{-k}\Theta;\quad k\geq2.
  \end{split}
\end{align}

Below we will consider the composition $\Theta=\Theta_\nu(r_R)$ given
in terms of the parameter $R$ appearing in Proposition
\ref{prop:mour2}. In fact we shall apply  Proposition
\ref{prop:mour2} in  essentially the same way as done in the beginning of
the proof of
 Theorem \ref{thm:priori-decay-b_0}, in particular for $R=R_0$ only. Since $\lambda_0\notin
 \sigma_{\pp}(\breve H)$   we can now
 take $K=0$   and
 therefore replace $2\tilde d(\lambda_0)-\epsilon-K$ by $\tilde
 d(\lambda_0)$ in \eqref{eq:43breve}.
 More precisely we  use the Mourre estimate \eqref{eq:43breve}  with $\epsilon=\tilde d(\lambda_0)/2>0$ and $R=R_0\geq 1$
  fixed sufficiently big. In agreement with this  assertion,
  $\kappa_0>0$
   is fixed too (small, in particular $\kappa_0\leq\kappa_0'$), and we can freely use  the estimate for
   $\kappa\in (0,\kappa_0]$,  for a fixed (small)  open neighbourhood $\vU$ of
   $\lambda_0$ and in fact  with $K=0$, leaving us with the lower
   bound $\tilde d(\lambda_0)$ as claimed above.  By the virial
   theorem $\sigma_{\pp}(\breve H)\cap \vU=\emptyset$. We fix a
     compact
   neighbourhood $I\subset \vU$ of $\lambda_0$ such that there are no
   nonzero generalized eigenfunctions of $\breve H$ in
   ${\vB}_0^*$ at any
   point of $I$, which  is doable thanks to   Theorem \ref{thm:priori-decay-b_0}.   As we will
   see this  $I$ works in Theorem \ref{thmlapBnd}. Choose a real-valued
   $f\in C^\infty_\c(\vU )$ such that  $f=1$ on a
   neighbourhood of $I$.

Below we
 use the notation $\vB$ and $\vB^*$ for the  Besov spaces given in terms
 of $r_{R_0}$ (rather than in terms of $\abs{x}$ as before), the latter written for short $r=r_{R_0}$. Of course the
 spaces $\vB= \vB(r)$  and $\vB(\abs{x})$ coincide  and similarly
for the adjoint spaces (allowing us to
 change the meaning of the notation in \eqref{eq:lap-besov-spacebnd}).
 We could use   $X$ for  multiplication by
$r$ as in the previous subsection, however  we find it more convenient to use the notation $r$
only,
even though mostly it will be the operator of multiplication by
$r$. We will assign   the notation $\vO_{ \textrm
  {unf}}(r^{t})$ a  different meaning (although very related) than $\vO_{ \textrm
  {unf}}(X^{t})$ used in the previous subsection.
 More precisely we introduce
the following terminology for families $(T_\nu)_{\nu\in\N_0}$ of
operators    on $ L^2$.

We say $ (T_\nu)\subset \gO_\kappa(X^{t})$ is \emph{uniformly of
  $\kappa$-order $t$} (for $t\in \R$)  if
\begin{align*}
\forall \kappa &\in (0,\kappa_0']\;\forall s\in \R\,\;\forall k\in \N_0:\\&
\sup_{\nu\in\N_0}\norm{r^{s-k-t}\ad_{B_{\kappa,R_0}}^k(T_\nu) r^{-s}}+\sup_{\nu\in\N_0}2^{\nu/2}\norm{r^{s-k-t-1/2}\ad_{B_{\kappa,R_0}}^k(T_\nu) r^{-s}}<\infty.
\end{align*}
 For the elements $T_\nu$ in such a family $ (T_\nu)$  we write
$T_\nu=\vO_{\textrm {unf}}(r^{t})$ and  $ T_\nu\in \gO_{ \textrm {unf}}(r^{t})$. Note that for example
$\Theta^{1/2}$, considered as the  composed function $\Theta_\nu^{1/2}(r(\cdot))$, is uniformly of $\kappa$-order  $0$, cf. \eqref{eq:elebnd}.  If $ T_\nu\in \gO_{ \textrm {unf}}(r^{t})$
 and $S$ is any of the operators listed in
Lemma \ref{lem:orderP} (with $R=R_0$),  then  $(ST_\nu )$  and $(T_\nu S)$
 are also  uniformly of $\kappa$-order
 $t$.  More generally for any
$ T_\nu\in \gO_{ \textrm {unf}}(r^{t})$ and $S\in\gO_\kappa(X^{s})$
the operators   $ST_\nu,T_\nu S\in \gO_{ \textrm {unf}}(r^{s+t})$.
 Similarly,  if  $S_\nu\in\gO_{ \textrm {unf}}(r^{s})$ and $T_\nu\in\vO_{
   \textrm {unf}}(r^{t})$   then $S_\nu T_\nu, T_\nu S_\nu\in\gO_{ \textrm {unf}}(r^{s+t})$.

\begin{lemma}\label{lemma:12.7.2.7.9}
\begin{subequations}
     There exists $ C >0$ and such that for all
$z\in\C\setminus\R$ with $\Re z\in I$  and for
all $\nu\in \N_0$ and  $\psi\in \vB$
\begin{align}
\label{eq:13.8.22.4.59cc22}
\|\Theta'{}^{1/2}\phi\|^2
&\le C\bigl(\|\phi\|_{\vB^*}\|\psi\|_{\vB}
+\norm{{\psi}}_{\vB}^2+2^{-\nu/2}\|r^{-3/4}\phi\|^2\bigr),\\
\norm{\phi}^2_{\vB^*}&\leq C\bigl(\norm{{\psi}}_{\vB}^2+\norm
  {r^{-3/4}\phi}^2\bigr);\label {eq:13.8.22.4.59cc24}
\end{align} here
$\Theta=\Theta_\nu(r(\cdot))$ and  $\phi=\breve R(z)\psi$.
  \end{subequations}
\end{lemma}
\begin{proof}
 The bound
\eqref{eq:13.8.22.4.59cc24} follows from
\eqref{eq:13.8.22.4.59cc22} by taking  supremum  over $\nu\geq0$. So
it suffices  to show  \eqref{eq:13.8.22.4.59cc22}.

We consider $P_\nu=f(\breve H) \Theta^{1/2}\zeta_\epsilon
      (B_\kappa)\Theta^{1/2}f(\breve H)$,
      where   $B_\kappa$ and $\zeta_\epsilon$  are given as in
the previous subsection and    $f$ is the function introduced above.  Note that
$P_\nu=\vO_{\textrm {unf}}(r^{0})$. As before  we have the freedom to choose
$\kappa>0$  sufficiently small,  and again also the parameter   $\epsilon>0$ will be
 fixed sufficiently small.

 Let
$\theta_\epsilon=\sqrt{\eta'_\epsilon}$. We compute using Lemma \ref{lem:orderComm}
 \begin{align*}
    \inp{\i[\breve H, P_\nu]}_{\phi}&=\inp{\widecheck \bD  P_\nu}_{\phi}=\inp{L_1+L_5+\vO_{
\textrm {unf}}(r^{-2})}_{\phi};\\
L_1&=2f(\breve H)\Re \parb{(\kappa^2B^2+1)\tfrac{\Theta'}{\Theta^{1/2}}
  B_\kappa \zeta_\epsilon(B_\kappa)\Theta^{1/2}}f(\breve H),\\
r^{1/2}\Theta^{-1/2}L_5\Theta^{-1/2}r^{1/2}&=\theta_\epsilon(B_\kappa) f(\breve H)\i[\breve H,
    A_\kappa]f(\breve
    H)\theta_\epsilon(B_\kappa)\\
&+\theta_\epsilon(-B_\kappa)f(\breve H)\i[\breve H,
    A_\kappa]f(\breve H)\theta_\epsilon(-B_\kappa)\\
&-2\theta_\epsilon(B_\kappa)f(\breve
H)(\kappa^2B^2+1)B_\kappa^2f(\breve
H)\theta_\epsilon(B_\kappa)\\
&-2\theta_\epsilon(-B_\kappa)f(\breve H)(\kappa^2B^2+1)B_\kappa^2f(\breve H)\theta_\epsilon(-B_\kappa).
 \end{align*}

Noting the essential positivity of the term containing the factor
 $\kappa^2B^2$ we obtain  after symmetrizing
\begin{align}\label{eq:libnd2}
   L_1\geq 2f(\breve
   H)\Theta'^{1/2} B_\kappa
   \zeta_\varepsilon(B_\kappa)\Theta'^{1/2}f(\breve
   H)+O_{ \textrm {unf}}(r^{-2}).
 \end{align} Using  the Mourre estimate of the
 form discussed in the beginning of the proof and  factors of $\tilde f(\breve H)$
 (the latter can be inserted for free to bound
 $\kappa^2B^2$), we can estimate for $\kappa>0$ sufficiently small
\begin{align*}
L_5\geq (\tilde d(\lambda_0)-4\epsilon^2)f(\breve
H)\Theta^{1/2}r^{-1/2}\parbb{\eta'_\epsilon(B_\kappa)+\eta'_\epsilon(-B_\kappa)}r^{-1/2}\Theta^{1/2}f(\breve
H)+\vO_{\textrm {unf}}(r^{-2}).
\end{align*}

We shall require  that $\epsilon>0$ is so small  that
  $8\epsilon^2< \tilde d(\lambda_0)$, implying
   $\tilde d(\lambda_0)-4\epsilon^2\geq   \tilde d(\lambda_0)/2$.
We shall use  \eqref{eq:fundest} with $\tilde c={\tilde d(\lambda_0)}/{4}$ and
  any  possibly smaller
 $\epsilon$, henceforth considered fixed. Thus
 by combining the above  bound with  \eqref{eq:libnd2} and using the bound $r\Theta' \leq  \Theta$ we finally
 obtain the
following lower
 bound
\begin{align*}
 L_1+L_5&\geq
 \epsilon  f(\breve
   H)\Theta'^{1/2}\parb{\eta_\epsilon(B_\kappa)+\eta_\epsilon(-B_\kappa)}\Theta'^{1/2}f(\breve
   H)+\vO_{\textrm
  {unf}}(r^{-2})\\
&= f(\breve H)\Theta'f(\breve H)+\vO_{\textrm
  {unf}}(r^{-2}).
\end{align*}

Note also the trivial bound ${\inp{\Theta'}_{(1- f(\breve
    H))\phi}}\leq C\norm {\psi}_{\vB}^2$ (using that   $\Re z\in I$)
  leading finally  to the bound
\begin{align}
\begin{split}
\|\Theta'^{1/2}\phi\|^2
\le 2\inp{\i[\breve H, P_\nu]}_{\phi}+C\bigl(2^{-\nu/2}\norm{r^{-3/4}\phi}^2+\norm{{\psi}}_{\vB}^2
\bigr).
\end{split}\label{eq:13.8.22.4.59cc2s}
\end{align}
 On the other hand
\begin{align}\label{eq:com22}
\begin{split}
  \inp{\i [\breve H, P_\nu]}_{\phi}&= \i \inp{\psi, P_\nu\phi}-\i \inp{ P_\nu\phi,\psi}
  +2(\Im z) \inp{ P_\nu}_{\phi}\\
&\leq 2
\norm{P_\nu}_{\vL(\vB)}\norm{\psi}_{\vB}\norm{\phi}_{\vB^*}+2\norm{P_\nu}\abs{\Im{\inp{\breve
      H-z}_{\phi}}}\\
&\leq 2 (\norm{P_\nu}_{\vL(\vB)}
+\norm{P_\nu})
  \norm{\psi}_{\vB}\norm{\phi}_{\vB^*}\\
&\leq C\norm{\phi}_{\vB^*}\norm{\psi}_{\vB}.
\end{split}
\end{align}

The combination of
\eqref{eq:13.8.22.4.59cc2s} and \eqref{eq:com22} yields
\eqref{eq:13.8.22.4.59cc22}.
\end{proof}

\begin{proof}[Proof of Theorem \ref{thmlapBnd}] Let $I$ be
  given as in Lemma \ref{lemma:12.7.2.7.9}.
  Suppose by contradiction that $z_n\to \lambda_0'\in I$,
  $\norm{\psi_n}_{\vB}\to 0$ and $\norm{\phi_n}_{\vB^*}=1$ where
$\phi_n=\breve R(z_n)\psi_n$ (here $\vB$ and $\vB^*$ are defined in terms of
$r=r_{R_0}$).
Fix $s\in (1/2,3/4)$. We can assume that there exists $\wslim_{n\to \infty}\phi_n=:\phi\in L^2_{-s}$.
 By local compactness  and an energy estimate we easily see   (using
 the notation \eqref{eq:14.1.7.23.24} and $\chi_m=\chi_m(r)$) that
 \begin{align*}
   \forall m\in \N:
\quad \lim_{n\to \infty}\chi_m\phi_n=\chi_m\phi
\text{ in } H^1.
 \end{align*} Moreover by \eqref{eq:13.8.22.4.59cc22}
\begin{align*}
  \exists C>0\,\forall m\in \N \,\forall \nu\in \N_0: \quad\|\Theta'{}^{1/2}\chi_m\phi\|^2=\lim_{n\to \infty}\|\Theta'{}^{1/2}\chi_m\phi_n\|^2\leq C2^{-\nu/2}.
\end{align*}

From this bound  we learn that
$\phi\in \vB_0^*$, and since also
$(\breve H-\lambda'_0)\phi=0$ we obtain   that
$\phi=0$ (cf.  Theorem \ref{thm:priori-decay-b_0}).
 By local compactness  it then follows that $\lim_{n\to \infty}\phi_n=\phi =0
\text{ in }L^2_{-3/4}$,
  and by \eqref{eq:13.8.22.4.59cc24} we then deduce that $\lim_{n\to \infty}\norm{\phi_n}_{\vB^*}^2=0 $
  contradicting  the  assumption that $\norm{\phi_n}_{\vB^*}=1$
for all $n$.
\end{proof}

\subsection{Microlocal  bounds and LAP}\label{subsec:Microlocal  bounds and LAP}
 We prove microlocal  bounds of the
 resolvent $\breve R(z)$,  largely mimicking  \cite{GIS} and \cite{AIIS}. As in \cite{GIS} we then obtain similar bounds  of powers
 of  the
 resolvent, and this implies LAP (the limiting absorption principle) for $\breve H$ near $\lambda_0$.

We assume $\lambda_0\notin \sigma_{\pp}(\breve H)$. Since we are not going to need the sharpest version of these microlocal
bounds which use the optimal  constant in Proposition
\ref{prop:mour2} we will simplify the presentation and proceed
 exactly as  in
the beginning of the previous subsection, not to be repeated. (The
sharper version is given by replacing the constant $\tilde
d(\lambda_0)/2$ in Lemma \ref{lemma:microLoc} by any positive number
less than $\tilde d(\lambda_0)$.) With the
given
neighbourhood  $I\ni \lambda_0$ we define $I_{\pm}=\{z|\,\Re z \in I,\,\pm\Im z
  >0\}$ and $I_\C=I_{+}\cup I_{-}$.

We start by showing a version of the first step of an induction
procedure of
\cite{GIS}, cf. \cite {AIIS}.  Let for all $\sigma>0$ the notation $\vG_-^\sigma$ and
  $\vG_+^\sigma$ signify
  the classes of real functions $g\in C^\infty (\R)$
   with support in $(-\infty,
  \sigma)$  and  $(-\sigma,\infty)$, respectively.
\begin{lemma}\label{lemma:microLoc} Under the above conditions let
  $\tilde\sigma>0$ be given by
  $\tilde\sigma^2=\tilde d(\lambda_0)/2$  and let $\sigma \in (0,{\tilde\sigma})$.
   There exists $\check\kappa_0\in(0,\kappa_0]$ such that the
   following bounds  hold   for any
  $\chi(\pm\cdot<\sigma)\in \vG_\pm^\sigma$  and $  t\in(0,1/2)$.
  \begin{align}\label{eq:45lem}
     \forall\kappa\in (0,\check \kappa_0]\,\exists C>0\,\forall z\in
I_{\pm} :\quad \norm{\chi(\pm B_\kappa<\sigma)\breve R(z)}_
    {\vL(L^2_{1-t}, L^2_{-t})}\leq C.
  \end{align}
\end{lemma}
\begin{proof}    Let
  $t\in(0,1/2)$ and  $\delta=1/2-t$. Let $z\in I_+$, $\psi\in L^2_{1-t}$ and $ \phi=\breve
  R(z)\psi$ (we only  consider $I_+$).    We introduce  for $\epsilon>0$
  (and with $R=R_0$) the following operator  $P$ of $\kappa$-order
$2\delta$,
    \begin{align}\label{eq:Pminus}
      P=f(\breve H)      X^{\delta}g_\epsilon(B_\kappa)X^{\delta}f(\breve H);
\quad g_\epsilon(b)=-(\sigma+2\epsilon-b)^{2\delta}\eta^2_\epsilon(\sigma-b).
\end{align}
We compute
\begin{align}\label{eq:dergFunc}
  \tfrac12(\sigma+2\epsilon-b)^{2t }g_\epsilon'(b)=\delta\eta^2_\epsilon(\sigma-b)+(\sigma+2\epsilon-b)(\eta_\epsilon\eta'_\epsilon)(\sigma-b),
\end{align}  and noting  that $g_\epsilon'\geq 0$ and $ (g_\epsilon')^{1/2}\in C^\infty$ we
see  that all of
Lemma \ref{lem:orderComm} applies.
  Letting
 $\theta_\epsilon=\sqrt{g_\epsilon'}$ we
 read off
 \begin{align*}
    \inp{\i[\breve H, P]}_{\phi}&=\inp{\widecheck \bD  P}_{\phi}=\inp{L_1+L_5+\vO_{\kappa}(X^{2\delta-2})}_{\phi};\\
L_1&=4\delta f(\breve H)\Re \parb{(\kappa^2B^2+1))X^{\delta-1}
  B_\kappa g(B_\kappa)X^{\delta}}f(\breve H),\\
X^{t}L_5X^{t}&=\theta_\epsilon(B_\kappa)f(\breve H)\i[\breve H,
    A_\kappa]f(\breve
    H)\theta_\epsilon(B_\kappa)\\
&-2\theta_\epsilon(B_\kappa)f(\breve
H)(\kappa^2B^2+1)B_\kappa^2f(\breve
H)\theta_\epsilon(B_\kappa).
 \end{align*}

Letting
 $T=(\sigma+2\epsilon-B_\kappa)^{-t}\eta_\epsilon(\sigma-B_\kappa)X^{-t}f(\breve
 H)$ we can estimate
 \begin{align*}
   L_1&\geq 4\delta f(\breve H)X^{-t}
  B_\kappa g(B_\kappa)X^{-t}f(\breve
  H)+\vO_{\kappa}(X^{2\delta-2})\\
&\geq 4\delta T^*
  \parb{B^2_\kappa -(\sigma+\epsilon)(\sigma+2\epsilon)}T+\vO_{\kappa}(X^{2\delta-2}).
 \end{align*}

To bound $L_5$ we  first fix $\check\kappa_0\in(0,\kappa_0]$  such with
$C'=\norm{\tilde f(\breve H)B^4 \tilde f(\breve H)}$ and with an arbitrary
sufficiently small $\epsilon'>0$, the constant
 \begin{align}\label{eq:cBasic}
c' =\tilde d(\lambda_0)-2(\sigma+\epsilon')(\sigma+2\epsilon')-2{\check
  \kappa}_0^2C'\text{ is positive}.
\end{align} Then for all $\kappa\in (0,\check \kappa_0]$ and
$\epsilon\in (0,\epsilon']$ the second  term on the right-hand side of
\eqref{eq:dergFunc} contributes by a non-negative term and we obtain
\begin{align*}
L_5&\geq X^{-t}f(\breve
H)\parb{\tilde d(\lambda_0)-2\kappa^2C'  -2B_\kappa ^2}g'(B_\kappa )f(\breve
H)X^{-t}+\vO_{\kappa}(X^{2\delta-2})\\
&\geq 2\delta T^*\parb{\tilde d(\lambda_0)-2\kappa^2C'  -2B_\kappa
  ^2}T+\vO_{\kappa}(X^{2\delta-2}).
\end{align*}

 Observing the cancellation of the terms containing $B_\kappa^2$ these bounds  lead to the lower bounds
 \begin{align}\label{eq:Tbnd}
  \begin{split}
  \widecheck \bD P&\geq 2\delta T^*\parb{-2(\sigma+\epsilon)(\sigma+2\epsilon)+\tilde d(\lambda_0)-2\kappa^2C'}T+\vO_{\kappa}(X^{2\delta-2})\\
&\geq  2\delta c'T^*T+\vO_{\kappa}(X^{2\delta-2});\quad \kappa\in (0,\check \kappa_0],\,\epsilon\in (0,\epsilon'].
\end{split}
\end{align}

 Next,  introducing
$S=\eta_\epsilon(\sigma-B_\kappa )X^{-t}$  and using the fact
that
$B_\kappa$ is bounded we obtain
 \begin{align*}
   \norm{Sf(\breve H)\phi}\leq C_1\norm{T\phi}.
 \end{align*} Using  the notation $\|\cdot\|_s=\|\cdot\|_{L^2_s}$ we
  also note that
\begin{align*}
  \norm{S(1-f(\breve H)\phi}\leq C_2\norm {\psi}_{{\delta-1/2}}.
\end{align*} We conclude that
\begin{align*}
   \norm{S\phi}^2\leq
  C_3\parb{\norm{T\phi}^2+\norm
  {\psi}_{{\delta-1/2}}^2}.
 \end{align*}
 By combining this bound   with \eqref{eq:Tbnd} we  obtain
\begin{align*}
 c\|S\phi\|^2
\le \inp{\widecheck \bD P}_{\phi}+C_4\bigl(\norm{\phi}_{{\delta-1}}^2+
\norm {\psi}_{{\delta-1/2}}^2
\bigr);\quad {c}=2\delta c'/C_3.
\end{align*}
 On the other hand for any $\varepsilon>0$
\begin{align*}
 \inp{\widecheck \bD P}_{\phi}&=
\,\i \inp{\psi, P\phi}-\i \inp{ P\phi,\psi}
  +2(\Im z) \inp{ P}_{\phi}\\
\leq& C \parb{\norm {S\phi}+\norm{\phi}_{{\delta-3/2}}}\, \norm {\psi}_{{\delta+1/2}}\\
\leq& C\varepsilon\parb{\norm {S\phi}^2+\norm{\phi}_{{\delta-3/2}}^2}
+C\varepsilon^{-1}\norm {\psi}^2_{{\delta+1/2}}.
\end{align*} We choose $ \varepsilon=
c/(2C)$, yielding
\begin{align*}
\tfrac{ c}2\|S\phi\|^2
\le  C_5\bigl(\norm{\phi}_{{\delta-1}}^2+ \norm {\psi}_{{\delta+1/2}}^2
\bigr).
\end{align*} Finally we invoke Theorem \ref{thmlapBnd} and conclude
(after a commutation)
that for all $\kappa\in (0,\check \kappa_0]$ and $\epsilon\in (0,\epsilon']$
\begin{align}\label{eq:etaEst}
\|\eta_\epsilon (\sigma-B_\kappa)\phi\|_{{\delta-1/2}}^2
\le  C_6\norm {\psi}_{{\delta+1/2}}^2.
\end{align} This finishes the proof since for any given function
$\chi=\chi(\cdot<\sigma)$, we can write $\chi=\epsilon
\chi \eta_\epsilon
(\sigma-\cdot)$ for a small enough $\epsilon$ and then for any $
\kappa\in (0,\check \kappa_0]$ bound
\begin{align*}
\|\chi(B_\kappa<\sigma)\phi\|_{{\delta-1/2}}\leq C_7\|\eta_\epsilon (\sigma-B_\kappa)\phi\|_{{\delta-1/2}}
\le  C_8\norm {\psi}_{{\delta+1/2}}.
\end{align*}
\end{proof}

\begin{prop}\label{prop:microLoc2} Under the  conditions of Lemma
  \ref{lemma:microLoc}  the same assertion holds for  arbitrary
  $t<1/2$  (i.e. the constraint $t>0$ of the lemma is removed).
\end{prop}
\begin{proof}  Let $J_m=\tfrac 14\parb{[2,3)-m};\,m\in \N$. We saw in
  the previous proof that
  $\check\kappa_0=\check\kappa_0(\sigma)\in(0,\kappa_0]$  chosen in agreement with \eqref{eq:cBasic}
  works for any given $\sigma\in (0,  \tilde{ \sigma})$, and in particular we
  obtained  \eqref{eq:45lem}  with $t\in
  J_1$. We fix $\check\kappa_0=\check\kappa_0(\sigma)$ this way along
  with the other  positive constant $\epsilon'=\epsilon'(\sigma)$ of
  \eqref{eq:cBasic} and
  proceed to show by induction  the  assertion $\underline {q(m)}$:
  \begin{align*}
    &\forall \sigma\in (0,  \tilde{ \sigma})\,\forall t\in J_m\,
    \forall g\in \vG_\pm^\sigma \,\forall\kappa\in (0,\check
      \kappa_0(\sigma)]:\\&\quad \sup_{z\in I_{\pm}}\quad \norm{g(B_\kappa)\breve R(z)}_
    {\vL(L^2_{1-t}, L^2_{-t})}<\infty.
  \end{align*}
Since we have shown  $q(1)$ we can assume
$q(m-1)$ for a given $m\geq 2$  is known, and  then it remains to verify $q(m)$.

 Let $t\in J_m$ and $ \kappa\in (0,\check \kappa_0]$ be
given, and introduce again $\delta=1/2-t$.  Let $z\in I_+$, $\psi\in
L^2_{1-t}$ and $ \phi=\breve R(z)\psi$ (we consider only $I_+$). We
consider again \eqref{eq:Pminus} for $\epsilon\in (0,\epsilon']$. It
suffices to show \eqref{eq:etaEst} for any such $\epsilon$ by the
argument at the end of the proof of Lemma \ref{lemma:microLoc}. For
that we use the same scheme of proof as before, however since now
possibly $\delta-1\geq -1/2$ we can not use Theorem \ref{thmlapBnd} in
the same way. Rather we need to combine the commutator expansion formula
\eqref{eq:hD22} with Lemma \ref{lem:orderComm} which will allow us to
use the induction hypothesis in combination with Theorem
\ref{thmlapBnd}. This is already discussed in Remark \ref{remark:computing-commutator}.

 First we look at the contribution to $\widecheck \bD  P$ from
  $L'_1-L_1\in \gO_\kappa(X^{2\delta-2})$ of \eqref{eq:L145}, i.e. we
  look at
  \begin{align*}
    Q_{2\delta-2}=f(\breve H)\Re {\parbb{\vO_{\kappa}(X^{\delta-2})g_\epsilon(B_\kappa)X^\delta}}f(\breve
  H).
  \end{align*} By commutator expansion we see that with $
\check g_\epsilon=\eta_\epsilon (\sigma+2\epsilon-\cdot)$
  \begin{align}\label{eq:tgod}
    \begin{split}
    Q_{2\delta-2}&=\check g_\epsilon(B_\kappa)Q_{2\delta-2}\check g_\epsilon(B_\kappa)
    +\vO_{\kappa}(X^{-2})\\
    &=\check g_\epsilon(B_\kappa)X^{-t-1/2}\vO_{\kappa}(X^0)X^{-t-1/2}\check g_\epsilon(B_\kappa)
    +\vO_{\kappa}(X^{-2}).
    \end{split}
  \end{align} Whence the induction hypothesis combined with Theorem
  \ref{thmlapBnd} works upon taking
  $\epsilon>0$ sufficiently small.

  Next we look at the contribution to $\widecheck \bD P$ from
  $L'_2-L_5\in \gO_\kappa(X^{2\delta-2})$. As in the proof of Lemma
  \ref{lem:orderComm} we split
  \begin{align*}
    L'_2-L_5=(L'_2-L_2)-(L_3-L_2)-(L_4-L_3)-(L_5-L_4)=\sum^4_{j=1}T_j\in \gO_\kappa(X^{2\delta-2}).
  \end{align*} By \eqref{eq:L1452} the term $T_1$ can be represented  as  $Q_{2\delta-2}$ in \eqref{eq:tgod}, so we can argue in
  the same way for this term. By inspection of the proof of  Lemma
  \ref{lem:orderComm} we see that there are similar expansions for
  $T_j\in \gO_{\kappa}(X^{2\delta-2})$, $j=2,3,4$,  which allows us to
  treat these terms in the same way.

Finally using the  auxiliary operators $S$ and $T$ of  the proof of Lemma
\ref{lemma:microLoc}  we can mimic  the last part  of the proof using again the localization argument above to
treat   lower order terms, in particular various terms  in $ \gO_{\kappa}(X^{2\delta-2})$. We obtain $q(m)$.
\end{proof}

By the same method we can prove a two-sided estimate.
\begin{prop}\label{prop:microLoc2TWO} Let $\tilde\sigma>0$ be given by
  $\tilde\sigma^2=\tilde d(\lambda_0)/2$ and let $\sigma \in (0,{\tilde\sigma})$.
   There exists $\check\kappa_0>0$ such that the
   following bounds  hold   for any $  s>0$ and for any pair
  $g_\pm\in \vG_\pm^\sigma$ such that
    $\sup \supp g_-<\inf \supp g_+$.
  \begin{align}\label{eq:45lemTWO}
     \forall\kappa\in (0,\check \kappa_0]\,\exists C>0\,\forall z\in
I_{+} :\quad \norm{g_- (B_\kappa)\breve R(z)g_+ (B_\kappa)}_
    {\vL(L^2_{-s}, L^2_{s})}\leq C.
  \end{align}
\end{prop}

As in \cite{GIS} the assertions   Theorem \ref{thmlapBnd} and Propositions  \ref{prop:microLoc2} and
\ref{prop:microLoc2TWO} combine algebraically yielding  bounds (including
microlocal ones) of
powers of the resolvent. In particular the following bounds  follow,
which in turn implies LAP.
\begin{thm}\label{thm:powers} Suppose $\lambda_0$ is not an eigenvalue
  of  $\breve H$.
   Then
  \begin{align}\label{eq:45lemPower}
     \forall k\in\N\,\forall s<1/2\,\exists C>0\,\forall z\in I_\C
 :\quad \norm {\breve R(z)^k}_
    {\vL(L^2_{k-s}, L^2_{s-k})}\leq C.
  \end{align} In particular for any   $t>1/2$ the limits $\breve R(\lambda\pm \i 0)=
\lim_{\epsilon\to 0_+}\breve R(\lambda\pm \i \epsilon)$
  exist in $\vL (L^2_t,L^2_{-t})$ uniformly in  $\lambda\in I$. Moreover
there exists the strong weak-star limits
\begin{align}\label{eq:bRESSIM}
  \swslim_{\epsilon\to 0_+} \breve R(\lambda \pm \i \epsilon)=\breve
  R(\lambda \pm \i 0)\in \vL(\vB,\vB^*);\quad \lambda\in I.
\end{align}
\end{thm}

We will need the following application, see  \eqref{eq:2bnd}.

\begin{cor}\label{cor:microlocal-bounds}  Suppose $\lambda_0$ is not an eigenvalue
  of  $\breve H$, $s>1/2$  and that  $f\in L^2_s$  is given
  such that $\breve R(\lambda_0+ \i 0)f=\breve R(\lambda_0- \i
0)f$. Then
\begin{align}\label{eq:2bndSI}
  \breve R(\lambda_0+ \i 0)f=\breve R(\lambda_0- \i 0)f\in  L^2_{s-1}.
\end{align}
\end{cor}
\begin{proof}
  The result follows from   Proposition \ref{prop:microLoc2} and
  Theorem \ref{thm:powers}  by using a suitable decomposition $1=g_-
  (B_\kappa)+g_+ (B_\kappa)$, $g_\pm\in \vG_\pm^\sigma$.
\end{proof}
\begin{remark}\label{remark:microlocal-boundsGEN}
  In the multiple case, here discussed with  the assumptions of
  Section \ref{sec:Reduction near a multiple two-cluster threshold}
  only (in particular with \eqref{ass2} imposed), we need the analogues of
  \eqref{eq:bRESSIM} and Corollary \ref{cor:microlocal-bounds},
  cf. \eqref{eq:basresbreve}--\eqref{eq:2bnd}. We argue by commenting
  on the necessary modifications that the same methods of proof work
  with a minimum of extra complication.

  First note that the operator $\breve K$ of \eqref{eq:38p} needs to
  be replaced $\breve K=K_1-K_2$ with $K_1$ and $K_2$ given by
  \eqref{eq:28}. To simplify the   form of this operator $\breve K$
  we introduce the notation
  \begin{align}
    \label{eq:notSim}
    \vL_{-\infty,\infty}=\cap_{r,t\in\R} \,\vL
  \big ( L^{2}_{r}(\bX), L^{2}_{t}(\bX)\big ).
\end{align} Now using \eqref{eq:2proj}, Lemmas \ref{Lemma:basic2A} and  \ref{Lemma:basic2}
  as in the proof of Lemma \ref{lemma2.1} we see
that
\begin{align}
  \label{eq:dec2}
  \begin{split}
  &\breve K-\breve K_1-\breve K_2\in \vL_{-\infty,\infty};\\
  &\breve K_j=\Pi_j I_{j}+
  I_{j}\Pi_j -\Pi_j I_{j}\Pi_j,\quad j=1,2.
  \end{split}
\end{align} Note for example that
\begin{align*}
  \Pi-\Pi_1-\Pi_2\in (1-\Delta)^{-1}\vL^2_{-\infty,\infty}\subset\vL^2_{-\infty,\infty}.
\end{align*} Using \eqref{eq:dec2} we can prove Proposition
\ref{prop:mour2} by mimicking the proof for the non-multiple
case. In the first step of the proof of Lemma \ref{lemma:Mourre3} we
have the given properties for $a_j$ and $\Pi_j$, $j=1,2$, rather than
for $a_0$ and $\Pi$. This is all we need to repeat the proof. In Lemma
\ref{lem:orderP}  we need a similar replacement (in particular
the polynomial factors in $P_1\Pi_jP_2$ are polynomials in quantities
defined for
$a_j$ rather than for $a_0$). In the proof of Lemma \ref{lem:orderP}
we can obviously add the class of operators  in the intersection of
all operators of finite $\kappa$-order to the classes $\breve
\vB_k$. Note that $\vL_{-\infty,\infty}$ is included in this way, and
we can repeat the proof. The applications  of  Subsection
\ref{subsec:Calculus}  and  the present section  are  the same  as before.

\end{remark}

\begin{remarks}\label{remark:The case
  lambda0insigma}
\begin{enumerate}[1)]
\item \label{item:24}
We discuss the extension to the case
$\lambda_0\in\sigma_{\pp}(H')$. We shall use Section \ref{sec:The case
  where lambda0insigma} and Remark
\ref{remark:microlocal-boundsGEN}. Note that
$\lambda_0\notin\sigma_{\pp}(H'')$ and that Theorem
\ref{thm:priori-decay-b_0} implies that the $L^2$-eigenfunctions of $H'$ at
$\lambda_0$ are in $L^2_\infty$ (note also that the dimension of the
corresponding
 eigenspace is finite, cf. Proposition
\ref{prop:mour2}). Under the assumptions of Section
\ref{sec:The case where lambda0insigma} the 'extra term'
$\lambda_0|\psi\rangle\langle \psi|$ of \eqref{eq:38pt} and
\eqref{eq:28p} is consequently in $\vL_{-\infty,\infty}$. We can
therefore prove Proposition \ref{prop:mour2} with the operator $\breve
H$ of Section \ref{sec:The case where lambda0insigma} (cf. Remark
\ref{remark:microlocal-boundsGEN}) and obtain the corresponding
resolvent bounds of $\breve R$ from Subsections \ref{subsec:LAP bound}
and \ref{subsec:Microlocal bounds and LAP}. In particular Corollary
\ref{cor:microlocal-bounds} also holds  for the operators
$\breve
H$ and $\breve
R$ of Section \ref{sec:The case
  where lambda0insigma}  in the case
$\lambda_0\in\sigma_{\pp}(H')$.
\item \label{item:25} We shall in Chapter
  \ref{chap:resolv-asympt-near} use variations of this
  construction. These are given by applying the theory of the present
  chapter to $H$ replaced by $H_\sigma:=H-\sigma\Pi_H$, $\sigma> 0$
  small, where $\Pi_H$ is the orthogonal projection onto the
  eigenspace of $H$ corresponding to $\lambda_0$ (assuming that
  $\lambda_0$ is an eigenvalue). In the context of Chapter
  \ref{chap:resolv-asympt-near} this projection has finite rank, and
  if all $L^2$-eigenfunctions are in $L^2_\infty$ the results of the present
  chapter generalize easily. However we do not have this good decay of
  the eigenfunctions for the cases considered in Chapter
  \ref{chap:resolv-asympt-near}. Assuming `sufficient decay', see for
  example \eqref{eq:dec1} and \eqref{eq:dec2}, the results of the
  present chapter applies to some degree (sufficiently well for
  Chapter \ref{chap:resolv-asympt-near}). In particular, to be
  concrete, there is the following version of (parts of) Theorem \ref{thm:powers}
  and Corollary \ref{cor:microlocal-bounds}. Suppose
\begin{align}
    \label{eq:dec10}
    \ran \Pi_H\subset H^2_t\text{  for some }t\in (1,3/2),
  \end{align} and suppose $\lambda_0$ is not an eigenvalue
  of  $\breve H_\sigma$  and that  $f\in L^2_s$ for some  $s\in (1/2,t-1/2)$. Then there
  exist limits  $\breve R_\sigma(\lambda_0+ \i 0)f,\breve R_\sigma(\lambda_0- \i
0)f\in L^2_{(-1/2)^-}$. If in addition $\breve R_\sigma(\lambda_0+ \i 0)f=\breve R_\sigma(\lambda_0- \i
0)f$, then
\begin{align}\label{eq:2bndSI10}
  \begin{split}
   \breve R_\sigma(\lambda_0+ \i 0)f=\breve R_\sigma(\lambda_0- \i
  0)f\in  L^2_{s-1}.
  \end{split}
\end{align} These results  follow by  straightforward  modifications  of the previous proofs. Note that all  terms
in the calculus involving the perturbation $-\sigma\Pi_H$ are regarded
as `errors', in fact we can use Lemma
\ref{lem:orderComm} and freely interchange $f(\brH)$ and $\check
f(\brH)$ by  $f(\brH_\sigma)$ and $\check
f(\brH_\sigma)$, respectively, and similarly for the Mourre estimate
Proposition \ref{prop:mour2}. To obtain \eqref{eq:2bndSI10} under the
given hypothesis  and given limitations on $s$ and
$t$  we need  Lemma \ref{lemma:microLoc} for $\breve R_\sigma$, and
the interested reader may check that the errors from the indicated
substitutions are harmless when mimicking the proof of the lemma.

In Section \ref{sec:resolv-asympt-physics models near} a version of
\eqref{eq:2bndSI10} is needed, but only for a $s$ arbitrarily close  $1/2$.

\end{enumerate}
\end{remarks}

\begin{remark}\label{remark:excase}
  The case when \eqref{ass2} fails was discussed in Section
  \ref{sec:The case when the condition {ass2}}. One can  then
  obtain the same results as the ones of Remarks
  \ref{remark:microlocal-boundsGEN} and \ref{remark:The case
    lambda0insigma} (depending on whether $\lambda_0$ is an eigenvalue
  of $H'$ or not, respectively). Using Subsection \ref{sec:The case when the condition {ass2}, second
  approach}  we can again invoke  the formulas \eqref{eq:28} and
  \eqref{eq:28p} for $\brH$ (in the respectively cases) and in fact
  argue as before, obtaining the  results.
\end{remark}


\chapter{Rellich type theorems} \label{chap:lowest thr}

We investigate the structure of eigenfunctions and resonances
states at a two-cluster threshold under Condition \ref{cond:smooth}
and (relevant part of) Condition \ref{cond:geom_singl}. This structure
depends strongly on decay properties of effective potentials. In
Section \ref{sec:slow1} we treat the lowest threshold case which
allows for somewhat stronger assertions than above this threshold. In
Section \ref{sec:slow13} we study the latter case which demands a more
careful examination of the 'eigentransform' of Remark
\ref{remark:GrushinBB}. In Section \ref{sec:CoulRellich} we study the
physical models from  Sections \ref{$N$-body Schr\"odinger operators}
and \ref{$N$-body Schr\"odinger operators with infinite mass nuclei} by applying  previously discussed methods, in
particular we shall apply most  of Section \ref{sec:slow13}.

We are going to treat the multiple two-cluster case in detail  using
for the lowest threshold
Proposition \ref{prop2.3}. The non-multiple  two-cluster case can be treated
similarly, it is notationally simpler of course, and we do not pay
special attention to this case. The non-simple cases can be treated similarly. Only in
Section \ref{sec:CoulRellich} we do a complete study covering all
cases, however for the
physical models only.

Under the conditions of Proposition \ref{prop2.3} (used in Section
\ref{sec:slow13} without
the  lowest threshold condition) it is  convenient to consider operators as quadratic
forms defined on the first order Sobolev spaces, cf. Remark
\ref{remark:formbnd}. For $ k,s \in \R$, let $H^{k}_{s}$
denote the weighted Sobolev space of order $k$ with position-space
weight $\w{x}^s$ (as defined in Subsection \ref{Spaces}) and
\begin{align}\label{eq:spaceDEF}
  \begin{split}
   \vH_{s,t}^{k, l} &= H^{k}_{s}(\bX_1) \oplus H^{l}_{t}(\bX_2), \quad \vH^{k}_{s}
   = \vH^{k,k}_{s,s}, \quad \vH^{k} = \vH^{k}_{0}, \quad \vH_{s} =
   \vH^{0}_{s}, \quad \vH = \vH^{0}_0,\\
& \vH^{k}_{\infty}=\cap_{s\in \R}\vH^{k}_{s}, \quad \vH^{k}_{-\infty}=\cup_{s\in \R}\vH^{k}_{s}, \quad \vH^{k}_{s^+}=\cup_{t>s}\vH^{k}_{t}, \quad \vH^{k}_{s^-}=\cap_{t<s}\vH^{k}_{t}.
  \end{split}
\end{align}

  We may  write the operator $E_{\vH}(\lambda_0)$ of
Proposition \ref{prop2.3} as
\begin{align}
  \label{eq:fundec}
  E_{\vH}(\lambda_0)=-(p_1^2+W_1)\oplus (p_2^2+W_2) -V.
\end{align} This is for  the lowest threshold only,  and the operator $V$ is a non-local
symmetric potential of order $\vO(|x|^{-2-2\rho})$, see Proposition
\ref{prop2.3} and Remark \ref{remark:formbnd}. A modification of
\eqref{eq:fundec}  for higher   thresholds will be discussed in  Section
\ref{sec:slow13}, and the 'exceptional
cases' of Sections \ref{sec:The case where lambda0insigma} and
\ref{sec:The case when the condition {ass2}} will be discussed in
Remarks \ref{remark:except1} and
\ref{remark:except12}. We will shortly introduce a slightly different
decomposition (with corresponding different notation).

\section{The case $\lambda_0=\Sigma_2$}\label{sec:slow1}
For $\rho<2$ there are no results for the one-body problem in
general. Thus for example it is not known if zero can be an eigenvalue
of infinite multiplicity. However if either the potential is negative
or positive at infinity like $-|x|^{-\rho}$ or $|x|^{-\rho}$,
respectively, then it is known that zero is at most an eigenvalue of
finite multiplicity (in fact it is not an eigenvalue in the negative
case). In this section we impose the conditions of Proposition
\ref{prop2.3}.  In Subsections
\ref{subsec:negat-effect-potent} and
\ref{subsec:positive-effect-potent}  we are going to use  the
proposition under additional sign conditions
on the effective potentials $W_1$ and $W_2$. In Subsection
\ref{subsec:Homogeneous-effect-potent}  the effective potentials  are assumed to be homogeneous of degree $-2$ to leading
order, which is a border line case with a rich structure. The case
$\lambda_0>\Sigma_2$ is divided  into similar cases, to be treated in Section
\ref{sec:slow13}. In both cases we need extensions of Remark
\ref{remark:GrushinBB}, to be treated in Subsections  \ref{subsec: Extended
  eigentransform1} and  \ref{subsec: Extended
  eigentransform2}, respectively.

\subsection{Extended eigentransform for  $\lambda_0=\Sigma_2$}\label{subsec: Extended
  eigentransform1}
Recall from Remark \ref{remark:GrushinBB}  the
   eigentransform   and the
   corresponding inversion formula:
   \begin{align}\label{eq:invtra}
     f=T^*u\mand u={E}_+(\lambda_0)f=Sf- \breve R(\lambda_0) \Pi'H Sf.
   \end{align} We can   apply these relations to  $L^2$-eigenfunctions as well as to generalized
   eigenfunctions. In this extended sense  the following result holds.
 \begin{lemma}
   \label{lem:4.1a} For any $s\in\R$ the relations
   \eqref{eq:invtra} obey
   \begin{align}\label{eq:corsE}
     u\in H^1_s(=H^1_s(\bX))\mand (H-\lambda_0)u=0\Leftrightarrow f\in \vH^1_s\mand
     {E}_{\vH}(\lambda_0)f=0.
   \end{align}
 \end{lemma}
 \begin{proof}  Let $s\in\R$ be given.
   First we extend the boundedness result  \eqref{eq:91BB}, now
   claiming that
   \begin{align*}
     \breve R(z)\Pi'\in\vL(L^2_s,H^2_s)\subset \vL(L^2_s,H^1_s)\quad(\text{in particular for }z=\lambda_0).
   \end{align*} Noting that $\inp{x}^s(-\Delta-\breve H)\inp{x}^{-s}$
   is $\epsilon$-bounded relatively to $-\Delta$ the result follows by
   a standard computation using \eqref{eq:91BB}.
 By \eqref{eq:25}, Lemmas \ref{Lemma:basic2A}, \ref{Lemma:basic2}~\ref{item:3mq} and   the polynomial decay of the threshold bound
   states  it then follows  that
   \begin{align}\label{eq:sMap}
     S^*, T^*\in \vL(H^1_s,\vH^1_s)\mand \parb{1- \breve R(\lambda_0) \Pi'H
       }S\in\vL(\vH^1_s,H^1_s).
     \end{align} In particular \eqref{eq:corsE} holds.
\end{proof}
 The above proof is rather simple due to the fact that
   $\lambda_0\notin \sigma(H')$. In Subsection \ref{subsec: Extended
     eigentransform2}  we derive a rather detailed version of the lemma when $\lambda_0> \Sigma_2$. In this case the
   eigentransform and the corresponding inversion formula need a more
   subtle treatment. Moreover we remark that by ellipticity of the
   equation on the left-hand side of \eqref{eq:corsE} we can replace
   the requirement $u\in H^1_s$ by $u\in L^2_s$ or $u\in H^2_s$
   without changing the content of this side of  \eqref{eq:corsE}. A similar comment is
   due for the  right-hand side of \eqref{eq:corsE}, cf. the last part
   of Remark \ref{remark:formbnd}.

\subsection{Negative slowly decaying effective
  potentials}\label{subsec:negat-effect-potent}

Suppose $\rho<2$ and that $I_j(x_j)=I_j(x)_{x^j=0}$, $j=1,2$, fulfill the
following condition, cf. \cite{FS}:
\begin{align}\label{eq:virial0}
  \begin{split}
    \exists R\geq 0&\quad \exists \epsilon>0\quad\forall y\in \bX_j\text{ with }|y|\geq R: \\
    &I_j(y)\leq-\epsilon \inp{y}^{-\rho}\text{ and }-2I_j(y)-y\cdot
    \nabla I_j(y)\geq \epsilon \w{y}^{-\rho}.
  \end{split}
\end{align}

We can for each $j$ write the potential
$ W_j(x_j)=\widehat W_j(x_j)+B_j(x_j)$, where $B_j$ is polynomially
decreasing and $\widehat W_j$ is smooth with
$\partial^\alpha_y\widehat W_j=\vO(\w{y}^{-\rho-|\alpha|})$ and fulfills
\begin{align}\label{eq:virial}
  \begin{split}
    \exists \epsilon>0\quad\forall y\in \bX_j: \\
    &\widehat W_j(y)\leq-\epsilon \inp{y}^{-\rho}\text{ and }-2\widehat
    W_j(y)-y\cdot \nabla \widehat W_j(y)\geq \epsilon \w{y}^{-\rho}.
  \end{split}
\end{align} To see this we split each $V_b$ in the definition of
$I_{j}$ as $V_b=V_b^{(1)}+V_b^{(2)}$, where the second term is
compactly supported. We note that
$\inp{V_b^{(2)}}_{\varphi_j}(x_j)=\inp{\varphi_j,V_b^{(2)}(\cdot+ x_{j})\varphi_j}_j$
is polynomially decreasing. This can be proven in the same way as we
proved \eqref{eq:22asymB}.  It remains to consider the contribution
from $V_b^{(1)}$. Now the potential
$A(x_j)=\inp{V_b^{(1)}}_{\varphi_j}(x_j)-V_b^{(1)}(x_j)$ is $C^\infty $. We note that $A(y)=\vO(\w{y}^{-\rho-1})$ follows
in the same way as  we proved \eqref{eq:22asym}. In a similar fashion one  checks that
also $y\cdot \nabla
 A(y)=\vO(\w{y}^{-\rho-1})$, and due to these properties  we are lead
to consider $V_b^{(1)}( x_j)$, or rather $I^{(1)}_j( x_j)=\sum_b V_b^{(1)}( x_j)$. Since
we know the property \eqref{eq:virial0} for this  sum  we conclude a
similar property for $\sum_b \inp{V_b^{(1)}}_{\varphi_j}$, and we can add a suitable compactly supported potential
to have the bounds fulfilled for $R=0$ (and some $\epsilon>0$), cf. \cite{FS}. This leads to
a $\widehat W_j$ fulfilling \eqref{eq:virial}.

For convenience we use below the notation $w_j$ for the potential  $\widehat
W_j$. Introducing  the operators $h_j=p_j^2+w_j$, $j=1,2$ we write
\eqref{eq:fundec} as
\begin{align}
  \label{eq:fundec0}
  -E_{\vH}(\lambda_0)=h_1\oplus h_2 +v.
\end{align} Note that $v-V$ is  polynomially
decreasing. The operators $h_1$ and $h_2$
   have a number of microlocal properties as
stated in \cite[Theorems 4.1 and 4.2]{FS}. In particular there is a
limiting absorption principle at zero: Let $s_0= \tfrac 12 +\tfrac
\rho 4$. Then for any $s>s_0$ there exist the norm-limits
\begin{align*}
  r_j(0\pm \i 0) =\lim_{\epsilon
  \to 0_+} r_j(\pm \i \epsilon) \in \vL(H^{-1}_{s}(\bX_j), H^{1}_{-s}(\bX_j)),
\end{align*} where $r_j(z)=(h_j
-z)^{-1}$. Moreover  it is known from \cite{Sk4} that
\begin{subequations}
\begin{align}\label{eq:BesR}
  r_j(0\pm \i 0)\in \vL(\vB_{s_0},\vB_{s_0}^*),
\end{align} stated in terms of the  Besov spaces introduced in  Subsection \ref{Spaces}. We are going to use yet another  property.

Suppose $r_j(0+ \i 0)f=r_j(0- \i 0)f$ for  a given $f\in L^2_t$ for
some  $t>s_0$, then
 \begin{align*}
  r_j(0+ \i 0)f=r_j(0- \i 0)f\in  L^2_{s}\text  { for any }s< t-2s_0,
\end{align*} cf.  \cite[Theorems 4.1 (ii) and (iii)]{FS}.
 In fact one can strenghten the proof  in \cite{FS} and obtain the
 following statement:

Suppose
 $r_j(0+ \i 0)f=r_j(0- \i 0)f$ for  a given $f\in L^2_t$ for
some  $t>s_0$, then
 \begin{align}\label{eq:micr1}
  r_j(0+ \i 0)f=r_j(0- \i 0)f\in  L^2_{t-2s_0}.
\end{align} This statement is similar to
Corollary \ref{cor:microlocal-bounds}, and it can be proved by
 improved versions of some resolvent bounds from \cite{FS} which are
 one-body threshold
analogues of the bounds of
Proposition \ref{prop:microLoc2}.
 \end{subequations}

\begin{thm}\label{thm:negat-effect-potent} Under the conditions of
  Proposition \ref{prop2.3} and \eqref{eq:virial0} for $j=1,2$,
  \begin{align*}
    \dim \ker(H-\lambda_0)<\infty.
  \end{align*}
  If $\phi\in \vB_{s_0,0}^*(\bX)$ solves the distributional equation
  $(H-\lambda_0) \phi=0$, then $\phi\in L_\infty^{2}(\bX)$.
\end{thm}
\begin{proof} \subStep{I} By Remarks \ref{remark:GrushinBB} and \ref{remark:formbnd}
  \begin{align*}
    \dim \ker(H-\lambda_0)=\dim \ker E_{\vH}(\lambda_0).
  \end{align*} Note that we have chosen the form interpretation of
  the appearing  operators in agreement with Remark \ref{remark:formbnd}, but
  that the operator interpretation of  Remark \ref{remark:GrushinBB}
  amount to the same spaces, cf. Remark \ref{remark:formbnd}. Let $R^{\pm}_{\diag}=r_1(0\pm \i 0)\oplus
  (r_2(0\pm \i 0)$ and suppose $f\in \ker E_{\vH}(\lambda_0)$.
  We write the equations $R^{\pm}_{\diag}E_{\vH}(\lambda_0)f=0$ as
  \begin{align}\label{eq:LipSchW}
    (1+K^{\pm})f=0, \quad K^{\pm}=R^{\pm}_{\diag}v,
  \end{align} where $1f$ arises by writing $f=(f_1,f_2)\in
  \vH^{1}$ and using (to be shown below) that
  \begin{align}\label{eq:leftinverse}
    r_j(0\pm \i 0)h_j f_j=f_j.
  \end{align} We note that
  \begin{align}\label{eq:comneg}
    K^{\pm}\in \vC(\vH^{1}_{-s}) \text{
    for }s\in (s_0, 2\rho+2-s_0).
  \end{align}
  Since $K^{\pm}$ is compact (on any such space) it follows from
  Fredholm theory that $\dim \ker E_{\vH}(\lambda_0)< \infty.$

  \subStep{II}  We prove the following more general version of
  \eqref{eq:leftinverse}: Suppose $f_j\in \vB_{s_0,0}^*(\bX_j)$ and that
  the  distribution
  $h_jf_j\in \vB_{s_0}(\bX_j)$, then \eqref{eq:leftinverse} holds and
    the function $f_j\in H^{1}_{-s}$  for any   $s>s_0$ (for
  the problem above $h_jf_j\in L^{2}_s (\bX_j)\subset \vB_{s_0}(\bX_j)$
  for any  $s$ as in
  \eqref{eq:comneg}). To see this  we  let  $\chi_R(x)=\chi(|x|/R)$, $R\geq 1$, be
   given in
  agreement with \eqref{eq:14.1.7.23.24}. We consider $\chi_R$ as a multiplication operator on
   $\bX_j$.  Now
  $r_j(0\pm \i 0)h_j \chi_Rf_j=\chi_Rf_j$ for all $R>1$. On the other
  hand
  \begin{align*}
    h_j\chi_Rf_j =\chi_Rh_jf_j -(\Delta \chi_R)f_j-2(\nabla
    \chi_R)\cdot \nabla f_j.
  \end{align*} For the last term
  \begin{align}\label{eq:compCom}
    \begin{split}
     &R^2\|(\nabla\chi_R)\cdot \nabla f_j\|^2_{L^2(\bX_j)}
    \\
    &\leq C_1\inp{p^2}_{\chi_{2R}\bar\chi_{R/2}f_j}
    \\
    &\leq
    C_2\parb{\Re\inp{\chi^2_{2R}\bar\chi^2_{R/2}f_j,h_jf_j}+R^{-\rho}\|\chi_{4R}f_j\|^2}
    \\
    &= C_2\parb{\vO(R^0)+R^{-\rho}o(R^{2s_0})}
    \\
    &=o(R^{1-\rho/2}),
    \end{split}
  \end{align}
  yielding
  \begin{align*}
    \lim_{R\to \infty}\,R^{s_0}\|(\nabla
    \chi_R)\cdot \nabla f_j\|_{L^2(\bX_j)}=0.
  \end{align*}
  For the middle term
  \begin{align*}
    R^{s_0}\|(\Delta \chi_R)f_j\|_{L^2(\bX_j)}\leq CR^{2s_0-2}\parb{R^{-s_0}\|\chi_{2R}f_j\|_{L^2(\bX_j)}},
  \end{align*} yielding
  \begin{align*}
    \lim_{R\to \infty}\,R^{s_0}\|(\Delta \chi_R)f_j\|_{L^2(\bX_j)}=0.
  \end{align*} These computations lead to
  \begin{align*}
    f_j=\wslim_{R\to \infty}\,\chi_Rf_j=\wslim_{R\to \infty} \,r_j(0\pm \i 0)h_j\chi_Rf_j=r_j(0\pm \i
      0)h_jf_j\text{ in }\vB^*_{s_0}(\bX_j),
  \end{align*} and therefore \eqref{eq:leftinverse} holds for  $f_j\in \vB_{s_0,0}^*(\bX_j)$ obeying
  $h_jf_j\in \vB_{s_0}(\bX_j)$.

  \subStep{III} Using again the 'eigentransform' of Remark
  \ref{remark:GrushinBB} it is readily seen that  any $\phi\in \vB_{s_0,0}^*(\bX)$ solving the
  distributional equation $(H-\lambda_0) \phi=0$ corresponds to a
  (distributional) solution $f=(f_1,f_2)$, $E_{\vH}(\lambda_0)f=0$
  with $f_j\in \vB_{s_0,0}^*(\bX_j)$ obeying
  $h_jf_j\in \vB_{s_0}(\bX_j)$, cf. Lemma \ref{lem:4.1a}. By the previous steps we then conclude
  that \eqref{eq:LipSchW} is fulfilled. In particular
  $K^+f=K^-f$. Conversely using that $v$ is symmetric
  we
  obtain that
  \begin{align}\label{eq:Imvanissh}
    0=\Im \w{f,vf}=-\Im \w{R^+_{\diag}vf,vf}
  \end{align} for any solution to $(1+K^+)f=0$, where  $f=(f_1,f_2)$
  obeys  $f_j\in
  \vB_{s_0}^*(\bX_j)$, $j=1,2$. Whence also  $(1+K^-)f=0$ for any such
  $f$. (Similarly a solution to $(1+K^-)f=0$ is also a solution to
  $(I+K^+)f=0$.) Moreover  $E_{\vH}(\lambda_0)f=0$ in the
  distributional sense.
  If   $(1+K^+)f=0$ and  $(1+K^-)f=0$ where $f=(f_1,f_2)$
  and  $f_j\in
  \vB_{s_0}^*(\bX_j)$  better decay is obtained by invoking
  \eqref{eq:micr1}. Explicitly we may fix $s$ as in \eqref{eq:comneg}
  and  starting with the input $f\in \vH^0_{-s}$
   conclude that  $f\in \vH^0_{-s+2\kappa}$, where
              $\kappa=\rho+1-s_0$. By iterating  this argument we conclude that $f_j\in
  L_\infty^{2}(\bX_j)$.   Next, by the `inversion formula'  of
  Remark \ref{remark:GrushinBB} any  such
  $f$ corresponds to   a
  \begin{align*}
    \phi\in  \ker(H-\lambda_0)\cap L_\infty^{2}(\bX),
  \end{align*} see Lemma \ref{lem:4.1a}. In particular, cf.  the remark at the beginning of the paragraph,  this is valid for any $\phi\in \vB_{s_0,0}^*(\bX)$ solving
  $(H-\lambda_0) \phi=0$.

\end{proof}

\subsection{Positive slowly decaying effective
   potentials}\label{subsec:positive-effect-potent} Suppose $\rho<2$ and
 that $I_j(x_j)=I_j(x)|_{x^j=0}$, $j=1,2$, fulfill the following condition:
 \begin{align}\label{eq:posLowbnd0}
   \begin{split}
     \exists R\geq0\quad\exists \epsilon> 0\quad\exists \bar \rho\in [\rho, \tfrac 23(1+\rho))&\quad \forall y\in \bX_j\text{ with }|y|\geq R: \\
     &I_j(y)\geq\epsilon \inp{y}^{-\bar\rho}.
   \end{split}
 \end{align}

 We write for each $j$ the potential
 $ W_j(x_j)=\widehat W_j(x_j)+B_j(x_j)$, where $B_j$ is polynomially
 decreasing and $\widehat W_j$ is smooth with
 $\partial^\alpha_y\widehat W_j=\vO(\w{y}^{-\rho-|\alpha|})$ and fulfills
 \begin{equation}\label{eq:posLowbnd}
     \exists \epsilon> 0\quad\exists \bar\rho\in [\rho, \tfrac
     23(1+\rho))\quad \forall y\in \bX_j: \quad \widehat W_j(y)\geq\epsilon
     \inp{y}^{-\bar\rho}.
 \end{equation}
 This is seen by arguments as in Subsection
 \ref{subsec:negat-effect-potent} (using that $\bar\rho< 1+\rho$). Again
 we shall use
the notation $w_j$ for the potential  $\widehat
W_j$, introduce  the operators $h_j=p_j^2+w_j$, $j=1,2$ and  write
\eqref{eq:fundec} as \eqref{eq:fundec0}.

 We introduce the symbols $s_j=s_j(y,\xi)=(\xi_j^2+w_j(y))^{-1}$. Note
 that $s_j$ belongs to the H\"ormander class
 $S\parb{\inp{y}^{\bar\rho}\inp{\xi}^{-2}, g}$,
 where the metric is given by
 \begin{align*}
   g(v)=\inp{y}^{2(\bar\rho-\rho-1)}\,v_y^2+\inp{y}^{\bar\rho}\,v_\xi^2,\quad
   v=(v_y,v_\xi),
 \end{align*}  and note that  the corresponding Weyl calculus has `Planck
 constant' of size
 \begin{align*}
   \inp{y}^{(\bar\rho-\rho-1)}\inp{y}^{\bar\rho/2}=\inp{y}^{-t};\,t=1+\rho-3\bar\rho/2.
 \end{align*} Next we may construct  a parametrix to infinite order for
 $h_j=\Opw(\xi^2_j+w_j(y))$ following a standard
 procedure. However for the theorem below we only need  the first step which amounts to letting
 the parametrix be given  by $r^0_{j}=\Opw(s_j)$.  We compute $r^0_{j}h_j =1+A_j$ for
 some $A_j=\Opw(a_j)$ with  $a_j\in  S(\inp{y}^{-t}\inp{\xi}^{-1}, g)$.

\begin{thm}\label{thm:pos-effect-potent} Under the conditions of
  Proposition \ref{prop2.3} and \eqref{eq:posLowbnd0} for $j=1,2$,
  \begin{align*}
    \dim \ker(H-\lambda_0)<\infty.
  \end{align*}
  If  $\phi\in  L_{-\infty}^{2}(\bX)$ (or alternatively $\phi\in H^{1}_{-\infty}(\bX)$) solves the distributional equation
  $(H-\lambda_0) \phi=0$, then $\phi\in L_{\infty}^{2}(\bX)$.
\end{thm}
\begin{proof} \subStep{I} By Remark \ref{remark:GrushinBB}
  \begin{align*}
    \dim \ker(H-\lambda_0)=\dim \ker E_{\vH}(\lambda_0).
  \end{align*}
   Suppose $(H-\lambda_0)\phi=0$ for $\phi\in
  H^{1}_{s}(\bX)$ for some  $s\in\R$ (possibly nonzero). Due to Lemma
  \ref{lem:4.1a} the eigentransform $f=T^* \phi \in
  \vH^{1}_{s}$ and   $ E_{\vH}(\lambda_0)f=0$.
  Let $R=r^0_{1}\oplus
  r^0_{2}$ and $A=A_{1}\oplus
  A_{2}$. Then we  write the equation  $RE_{\vH}(\lambda_0)f=0$ for
  any $f\in
  L^2_s $ as
  \begin{align}\label{eq:LipSchW2}
    (1+A+K)f=0, \quad K=Rv,
  \end{align} where we use \eqref{eq:fundec0} and that for the components
  of $f=(f_1,f_2)$
  \begin{align*}
    r^0_j(0)h_j f_j=(1+A_j)f_j.
  \end{align*} We note that $A+K\in \vC(\vH^1_s)$.
  It then  follows from Fredholm theory that
  $f$ belongs to a finite-dimensional subspace of $\vH^1_s$. This
  dimension is an upper bound  of the dimension of the
  set of
  functions $\phi\in
  H^{1}_{s}(\bX)$ solving $(H-\lambda_0)\phi=0$ (seen by the inversion
  formula). In particular the
  latter space is finite-dimensional. The first statement of the theorem follows by this
  argument  for $s=0$.

  \subStep{II} By an iteration procedure using \eqref{eq:LipSchW2},
  cf. the proof of Theorem \ref{thm:negat-effect-potent}, it follows
  that any (zero-energy) generalized eigenfunction  $f\in \vH^1_{-\infty}$
  must belong to $\vH^1_{\infty}$. Note that we can deduce that
  $f\in\vH^1_{s+t}$ given that $f\in\vH^1_{s+(k-1)t}$, where
  $t=1+\rho-3\bar\rho/2$. Since there is no limit on $k(\in \N)$ used
  for this argument, indeed this conclusion comes out by
  iteration. Consequently, thanks to Lemma \ref{lem:4.1a},  the
  generalized eigenfunctions of $H$ at $\lambda_0$
  in $H^{1}_{-\infty}(\bX)$ are all in $L^{2}_{\infty}(\bX)$.
\end{proof}

 The equation \eqref{eq:LipSchW2} might have `spurious'
  solutions,  i.e. solutions not corresponding to eigenfunctions of
  $H$. We can cure this `deficiency' by showing that the exact inverse
  $h_j^{-1}$, $j=1,2$,  in fact is a pseudodifferential operator. We do this
  below. In particular  $h_j^{-1}$ has nice commutation proporties
  with power-type weights in the configuration space. Such properties
  were proved by Yafaev \cite {Ya2},  and  we are in fact going to use
  \cite {Ya2}. More precisely we are going to use the assertion
  \begin{align}\label{eq:yafa}
    \forall a,b\geq 0, a-b\geq \bar \rho:\quad B=\slim_{\epsilon\to 0_+} \,\inp{y}^{-a}(h_j+\epsilon)^{-1}\inp{y}^{b}\text{ exists}.
  \end{align} Here we note that
  $B_\epsilon=\inp{y}^{-a}(h_j+\epsilon)^{-1}\inp{y}^{b}$ is bounded uniformly in
  small positive
  $\epsilon$, which follows by keeping track of constants in  the
  proof of \cite [Theorem 1]{Ya2}. Since   $\lim_{\epsilon\to 0_+}B_\epsilon u$ exists if
  $u=\inp{y}^{-b}h_jv$ with  $v\in \vD(h_j)$ and the set of such $u$'s
  is  dense in
  $L^2(\bX_j)$,  indeed \eqref{eq:yafa} follows (in
  fact $Bu=\inp{y}^{-a}h_j^{-1}\inp{y}^{b}u$ for any   $u$ of this form).

  \begin{lemma}\label{lemma:posit-slowly-decay} The operators
    $\slim_{\epsilon\to 0_+} \,\inp{y}^{-\bar \rho}(h_j+\epsilon)^{-1}$, $j=1,2$, (extending $\inp{y}^{-\bar \rho} h_j^{-1}$) are  pseudodifferential
  operators  with symbol
    in  $S\parb{\inp{\xi}^{-2}, g}$.
    \end{lemma}
    \begin{proof} Recall $r^0_{j}h_j =1+A_j$,  where $A_j$ has symbol $a_j\in  S(\inp{y}^{-t}\inp{\xi}^{-1}, g)$.
      If $-1\notin \sigma (A_j)$
      \begin{align*}
        (h_j+\epsilon)^{-1}=(1+A_j)^{-1}r^0_{j}-\epsilon (1+A_j)^{-1}r^0_{j}(h_j+\epsilon)^{-1},
      \end{align*} yielding  $\slim_{\epsilon\searrow
    0} \,\inp{y}^{-\bar \rho}(h_j+\epsilon)^{-1}=\inp{y}^{-\bar \rho}(1+A_j)^{-1}r^0_{j}$.
(Whence formally
      $h_j^{-1}=(1+A_j)^{-1}r^0_{j}$.)  Using the Neumann series to
        expand $(1+A_j)^{-1}$  and Beal's criterion,
      cf. \cite[Subsection 4.2]{FS},
      we conclude that indeed $\inp{y}^{-\bar \rho}(1+A_j)^{-1}r^0_{j}$ has symbol in
      $S\parb{\inp{\xi}^{-2}, g}$ completing the proof in this case.

To treat the general case we let   $\chi_1(r)=\chi(r<1)$ and $\chi_2(r)=\chi(r>1)$ form a quadratic
 partition of unity of smooth non-negative functions on $\R$,
 $\chi_1(r)^2+\chi_2(r)^2=1$, such that $\chi_1$ is supported
 in $(-\infty, 2)$ and $\chi_1(r)=1$ for
 $r<1$.  We introduce for $l>1$ the functions
 $\chi_{i,l}(r):=\chi_{i}(r/l)$; $i=1,2$. Let
 \begin{align*}
   R_l=\oplus_{j=1}^2\parb{\chi_{1,l}(r_j)h_j^{-1}\chi_{1,l}(r_j)+\chi_{2,l}(r_j)r^0_{j}\chi_{2,l}(r_j)}=\oplus_{j=1}^2r_{j,l}.
 \end{align*}
 Here $\chi_{1,l}(r_j)h_j^{-1}\chi_{1,l}(r_j):=\slim_{\epsilon\searrow
    0}
  \,\chi_{1,l}(r_j)(h_j+\epsilon)^{-1}\chi_{1,l}(r_j)$. Using
  \eqref{eq:yafa}  and Beal's criterion  we deduce  that
 $\chi_{1,l}(r_j)h_j^{-1}\chi_{1,l}(r_j)$ and therefore also $r_{j,l}$ are  pseudodifferential operators with symbol
    in  $S\parb{\inp{y}^{\bar\rho}\inp{\xi}^{-2}, g}$. Next we write
     $r_{j,l}h_j =1+A_{j,l}$ and note (as above) that
 $A_{j,l}$ has symbol $a_{j,l}\in
 S(\inp{y}^{-t}\inp{\xi}^{-1}, g)$.
We also observe that
\begin{align*}
  \norm {\chi_{2,l}(r_j)A_j\chi_{2,l}(r_j)}\to 0\text{ for }l \to \infty.
\end{align*} Using this and a little computation it follows that
\begin{align*}
   \norm {A_{j,l}}\to 0\text{ for }l \to \infty.
\end{align*}
 In particular it
 follows that $-1\notin \sigma (A_{j,l})$  for  $l$ taken big
 enough. Consequently for any such big number
  $\inp{y}^{-\bar\rho}h_j^{-1}=\inp{y}^{-\bar\rho}(1+A_{j,l})^{-1}r_{j,l}$ has symbol
    in  $S\parb{\inp{\xi}^{-2}, g}$.
\end{proof}

We note that Lemma \ref{lemma:posit-slowly-decay} can be used to
replace \eqref{eq:LipSchW2} by the cleaner assertion
\begin{align}\label{eq:LipSchW2bb}
    (1+\widetilde{K})f=0; \quad \widetilde{K}=\widetilde{R}v, \quad \widetilde{R}=h^{-1}_{1}\oplus
  h^{-1}_{1},
  \end{align} of course given an appropriate interpretation.

\subsection{Homogeneous degree $-2$ effective
   potentials}\label{subsec:Homogeneous-effect-potent} Suppose
 $\rho\geq 1/2 $, the conditions of
  Proposition \ref{prop2.3} and that $W_j=W_j(x_j)$, $j=1,2$, fulfill the
 following condition (recall that in general $W_j$ is bounded outside
 a bounded set):

 For  $j=1,2$ there exists a real continuous
 function $q_j=q_j(y)$ on the unit sphere $\S_j$ of $\bX_j$ and a
 bounded potential $B_j(x_j)=\vO(\w {x_j}^{-3})$ on $\bX_j$ such that
 \begin{equation}\label{eq:posLowbnd0b}
     \exists R> 0\quad \forall y=r\theta\in X_j\text{ with }|y|=r\geq
     R: W_j(y)= \tfrac{q_j(\theta)}{r^2}+B_j(y).
   \end{equation}
  This condition may arise by a Taylor expansion of $I_j(x)$, $j=1,2$,
 as follows: Decompose
 \begin{align*}
   \sum _{b\not\subset
   a_j}V_b^{(1)}(x^b)=\sum _{b\not\subset
   a_j}V_b^{(1)}((x_j)^b)+\sum _{b\not\subset
   a_j}(x^j)^b\cdot \nabla V_b^{(1)}((x_j)^b)+\vO(\w {x_j}^{-\rho-2}).
 \end{align*} Assume  now   that the first term vanishes
 identically and that the last term contributes to   $W_j$ by a term
 of order $\vO(\w {x_j}^{-3})$ (valid for
 $\rho\geq 1 $ of course). The middle term contributes to $W_j$ by
 $\sum _{b\not\subset a_j}k^b_j\cdot \nabla V_b^{(1)}((x_j)^b)$, where
 $k_j=\inp{\varphi_j,x^j\varphi_j}_j$, so the  content of
 \eqref{eq:posLowbnd0b} would in this case be the condition
 \begin{align*}
   \sum _{b\not\subset a_j}k^b_j\cdot \nabla V_b^{(1)}((x_j)^b)=\tfrac{q_j(\theta)}{r^2}+\vO(r^{-3}).
 \end{align*} This has relevance with nonzero $q_j$'s in a certain
 case for systems with Coulomb interactions, see
 Subsections \ref{First principal example} and \ref{Second  principal
   example}. The case $q_j=0$, $j=1,2$ is of course included in
 \eqref{eq:posLowbnd0b}, and  the relevance of this case for the physics
 models is also explained there.

 For simplicity of presentation let us in the following assume
 $R=1$. For $n_j=\dim \bX_j\geq 2$ the spectrum of the (minus)
 Laplace-Beltrami operator $-\Delta_{\theta_j}$ on the unit sphere
 ${ {\S_j}}\subset \bX_j$ is known to be $\{l(l+n_j-2)\mid l\in \N_0\}$. For
 our problem it is relevant to study the spectrum of
 $-\Delta_{\theta_j}+q_j$. For $n_j=1$ we define
 $-\Delta_{\theta_j}=0$, so in this case the spectrum of
 $-\Delta_{\theta_j}+q_j$ is $\{q_j(-1), q_j(1)\}$. In any dimension
 it is convenient to use the following parametrization: Write each
 $\mu\in \sigma(-\Delta_{\theta_j}+q_j)$ as
 \begin{align}\label{eq:relmunu}
   \mu=\nu^2-\tfrac{(n_j -2)^2}{4},
 \end{align} where by convention $\nu\geq 0$ if $\mu\geq-\tfrac{(n_j
   -2)^2}{4}$, and $\i\nu>0$ if  $\mu<-\tfrac{(n_j
   -2)^2}{4}$. Then the  collection of such numbers $\nu$ is denoted by
 $\sigma_j$. Let $P_{j,\nu}$, $\nu\in \sigma_j$, denote
 the orthogonal projection onto the corresponding eigenspace.  For
 convenience  we omit in the following
 the subscript $j$. Letting $\zeta$ denote a corresponding eigenvector,  the
 eigenvalue problem
 \begin{align*}
   (-\Delta +\tfrac{q_j(\theta)}{r^2})\parb{r^{\tfrac{1-n_j}2}u(r)}\otimes
   \zeta(\theta)=0
 \end{align*}
 reduces to the Euler equation
 \begin{align}\label{eq:ode}
   -u''(r) +\tfrac {\nu^2-1/4}{r^2}u(r)=0.
 \end{align}

 Consider now for any $\nu\geq 0$ or $\i \nu >0$ the Dirichlet
 problem
 \begin{align}\label{eq:compD}
   u \text{ fulfills }\eqref{eq:ode}\text{ for }
   r\geq 1,\mand u(1)=0.
 \end{align}
 The \emph{regular solution} is
 \begin{align*}
   \phi_\nu(r)=
   \begin{cases}
     \tfrac{r^{1/2+\nu}-r^{1/2-\nu}}{2\nu}&\text{ for } \nu\neq 0,\\
     r^{1/2}\ln r&\text{ for } \nu= 0
   \end{cases};
 \end{align*} note that indeed $\phi_\nu(1)=0$. The \emph{outgoing
   solution} to \eqref{eq:ode} (also defined for $r\geq1$)  is
 \begin{align*}
   \psi_\nu(r)=r^{1/2-\nu};
 \end{align*} note that indeed this is `outgoing' for $\nu$ complex. These two solutions form a fundamental system, and we can define a Green's function by
 \begin{align*}
   R_\nu(r,r')=\phi_\nu(r_<)\psi_\nu(r_>);\quad r_<=\min\{r,r'\},
   \, r_>=\max\{r,r'\}.
 \end{align*} Its formal adjoint $R^*_\nu(r,r')=\phi_\nu(r_<)\overline   {\psi_\nu(r_>)}$
 is also a Green's function.
 This  means that for all sufficiently decaying  functions $v=v(r)$
 \begin{align*}
   -u''(r) +\tfrac {\nu^2-1/4}{r^2}u(r)=v(r),\text{ where }
   u(r)=\int_1^\infty R_\nu(r,r')v(r')
   \,\d r',
 \end{align*} and similarly for $R^*_\nu$.

 Let $\chi_1(r)=\chi(r<8)$ and $\chi_2(r)=\chi(r>8)$ form a quadratic
 partition of unity of smooth non-negative functions on $\R$,
 $\chi_1(r)^2+\chi_2(r)^2=1$, such that $\chi_1$ is supported
 in $(-\infty, 8)$ and $\chi_1(r)=1$ for
 $r<4$.

Let $w_j=W_j$,
 $h_j=p_j^2+w_j$ and $v=V$ (this notation conforms with the previous
 subsections). We easily check that
 $v\in \vL\parb{\vH^0_{s},\vH^0_{3+s}}\mfor s\in\R$, cf. \eqref{eq:second_bnd}.
  Let
 \begin{align}\label{eq:ansG}
   \begin{split}
   &G_+=\oplus_{j=1}^2\parbb{\chi_1(r_j)(h_j-\i)^{-1}\chi_1(r_j)+\sum_{\nu\in\sigma_j}\parb{\chi_2(r_j) r_j^{\tfrac{1-n_j}2}R_{j,\nu} r_j^{\tfrac{n_j-1}2}\chi_2(r_j)}\otimes
     P_{j,\nu}},\\
   &G^*_+=\oplus_{j=1}^2\parbb{\chi_1(r_j)(h_j+\i)^{-1}\chi_1(r_j)+\sum_{\nu\in\sigma_j}\parb{\chi_2(r_j) r^{\tfrac{1-n_j}2}R^*_{j,\nu} r^{\tfrac{n_j-1}2}\chi_2(r_j)}\otimes
     P_{j,\nu}},
   \end{split}
 \end{align} where $r_j=|x_j|$ and  $R_{j,\nu}$ has a formal interpretation as an
 (unbounded) operator on $L_I^2:=L^2(I;\d r)$, $I:=[1,\infty)$. Phrased
 differently the powers of $r$ are isometrically identifying the
 spaces   $L_I^2$
 and $L^2(I;r^{n_j-1}\d r)$; note that the latter
 space appears naturally for tensor decompositions in spherical
 coordinates. The operators $G_+$ and  $G^*_+ $ are parametrices for
   $-E_{\vH}(\lambda_0)$ in the sense of the following lemma. Of
   course these operators are not the only options for a parametrix
   construction. Under stronger conditions, we consider in fact   a simpler
   parametrix in Section \ref{sec:resolv-asympt-nearLOWEST}.

 Introduce also the numbers
 \begin{align}
   \label{eq:crrit}
   \nu_0 =\min_{j\in\{1,2\}}\min_{\nu\in \sigma_j}\Re \nu,\quad \quad s_0=1+\nu_0.
 \end{align} In the fastly decaying case, $q_j=0$ for $j=1,2$, these numbers are
 explicit dimensional depending constants, cf.  \eqref{add110}. For
 example $\nu_0=\tfrac 12$ for $n_1=n_2=3$ in this case (to be
 considered in Section \ref{sec:resolv-asympt-nearLOWEST}).
 \begin{lemma}\label{lemHS}
   \begin{enumerate} [1)]
   \item\label{itemG1} The operators
     \begin{align}\label{HS0}
       &G_{+},G^*_{+} \in \vL(\vH^{-1}_{s'},\vH^{1}_{-s}) \text{ for }s, s' > 1 -\nu_0 \text{ with }  s+ s'>2,\\
       & K^*_+:=-E^*_{\vH}(\lambda_0)G^*_{+}-1 \in \vC(\vH^{-1}_{s}) \text{ for
         any   }s \in (1-\nu_0,2+\nu_0);\label{HS1}
     \end{align} here  $E^*_{\vH}(\lambda_0)=E_{\vH}(\lambda_0)$ is given  the
     distributional meaning, and we  define
     $K^+=(K^*_+)^*\in\vC(\vH^{1}_{-s})$.  Explicitly
     \begin{align}
       \label{eq:Kplus}
       \begin{split}
         K^+&=G_+v+\oplus_{j=1}^2\parbb{\chi_1(h_j-\i)^{-1}\parb{\i\chi_1+[\chi_1,p_j^2]}\\
           &+\sum_{\nu\in\sigma_j}\parbb{\parb{\chi_2
               r^{\tfrac{1-n_j}2}R_{j,\nu}
               r^{\tfrac{n_j-1}2}\chi_2}\otimes
             P_{j,\nu}}B_j\\
           &+\sum_{\nu\in\sigma_j}\parb{\chi_2
             r^{\tfrac{1-n_j}2}R_{j,\nu}
             r^{\tfrac{n_j-1}2}[\chi_2,p_j^2]}\otimes P_{j,\nu}}.
       \end{split}
     \end{align}

   \item \label{itemG2} $K^+f=-G_+E_{\vH}(\lambda_0)f-f$ for every
      $f\in \vH^{1}_{-s_0}$ such that
     $E_{\vH}(\lambda_0)f\in \vH^{-1}_{s}$  for some $s > 1 -\nu_0$.

   \item\label{itemG3} For any $s > 1$ the operator
     $G_{+} \in\vL(\vH^{-1}_{s},\vH^{1}_{-s})$ is injective, and
     \begin{align}
       \label{eq:posI}
       \Im G_{+}=\tfrac{G_{+}-G^*_{+}}{2\i}\geq \oplus_{j=1}^2\parb{\chi_1(h_j+\i)^{-1}(h_j-\i)^{-1}}\chi_1\geq 0.
     \end{align}
   \end{enumerate}
 \end{lemma}
 \begin{proof}
   \subStep{I} Let $s,s'$ be given as in \eqref{HS0}. We show that
   $r^{-s}R_{j,\nu}r^{-s'}$ is bounded on $L_I^2$ with a bound
   independent of $\nu\in \sigma_j$. The Hilbert-Schmidt criterion
   works. We note the bounds (here for $\nu\neq 0$ only; the case
   $\nu= 0$ is simpler)
   \begin{align*}
     &\int_1^\infty \d r  |\nu|^{-2} r^{1+ 2\Re \nu-2s}\int_r^\infty r'^{\,(1- 2\Re
       \nu-2s')}\d r'\leq C |\nu|^{-2}(s'+\Re \nu-1)^{-1},\\
     &\int_1^\infty \d r  |\nu|^{-2} r^{1- 2\Re \nu-2s}\int_1^r r'^{\,(1+
       2\Re \nu-2s')}\d r'\\
     &\leq  C |\nu|^{-2}
       \begin{cases}
         (s+\Re \nu-1)^{-1} (s'-\Re \nu-1)^{-1}\quad\text{
           for }s'>\Re \nu+1,\\
         (\Re \nu+1-s')^{-1} \quad\text{
           for }s'< \Re \nu+1,\\
         \int_1^\infty  r^{1- 2\Re \nu-2s}\ln r \, \d r
         \quad\text{ for }s'= \Re \nu+1,
       \end{cases}.
   \end{align*}
   This gives in particular the operator bound $\vO(|\nu|^{-3/2} )$ for
   $\Re\nu\to \infty$. Similarly $\partial_r r^{1-s}R_{j,\nu}r^{-s'}$
   is bounded with the bound $\vO(|\nu|^{-1/2} )$. This leads to
   \begin{align*}
     &\parb{-\Delta_j +\tfrac{q_j(\theta)}{r^2}}\parbb{\parb{\chi_2 r^{\tfrac{1-n_j}2}R_{j,\nu} r^{\tfrac{n_j-1}2}\chi_2}\otimes
     P_{j,\nu}}\\
&=\parb{[p_j^2,\chi_2]r^{\tfrac{1-n_j}2}R_{j,\nu} r^{\tfrac{n_j-1}2}\chi_2}\otimes
     P_{j,\nu} +\chi_2^2\otimes
     P_{j,\nu}
     \in \vL(L^2_{s'},L^2_{-s}),
   \end{align*}  where the first term is bounded by
   $\vO(|\nu|^{-1/2})$. By using the
  ellipticity of
   $-\Delta_j +\tfrac{q_j(\theta)}{r^2}$ and  interpolation we then
   obtain  that
   \begin{align*}
     \parb{\chi_2 r^{\tfrac{1-n_j}2}R_{j,\nu} r^{\tfrac{n_j-1}2}\chi_2}\otimes
     P_{j,\nu}
     \in \vL(H^{-1}_{s'},H^{1}_{-s})
   \end{align*}  with the  bound $\vO(|\nu|^{0} )$.

   Since the infinite sum defining $G_+$ is an orthogonal direct sum
   decomposition (and
   $\chi_1(h_j-\i)^{-1}\chi_1\in \vL(H^{-1}_{s'},H^{1}_{-s})
   $, obviously) we conclude \eqref{HS0}.

   \subStep{II}  The compactness assertion \eqref{HS1} follows by
   first computing $K^*_+$, estimating as in Step \textit{I} (for a good choice
   of parameters) and then invoking compactness on `each
   $\nu$-sector' (obtained   by the Hilbert-Schmidt criterion)  and the $\vO(|\nu|^{-1/2} )$ bound (also
   from I). The requirement $s \in (1-\nu_0,2+\nu_0)$ is dictated by
   the Hilbert-Schmidt criterion. More precisely fixing any such
   $s$  we need
   to find $s' \in\R $ such that, most importantly,
   $r^{-s'}R^*_{j,\nu}r^{-s}$ and $r^{s-3}R^*_{j,\nu}r^{-s}$ are
   bounded on $L_I^2$. Due to \eqref{HS0} it suffices to have
   \begin{align*}
     s' > 1 -\nu_0, \,  s+ s'>2\text{ and } s-3+s'\leq 0.
   \end{align*} These requirements are fulfilled with $s'=3-s$, and we
   can
   use that $\norm{r^{s-3}R^*_{j,\nu}r^{-s}}\to 0$ for $\nu\to\infty$. We
   have shown \ref{itemG1}.

   \subStep{III}  We show \ref{itemG2} by mimicking the proof of
   Theorem \ref{thm:negat-effect-potent}. For any given
   $f\in \vH^{1}_{-s_0}$ with $E_{\vH}(\lambda_0)f\in \vH^{-1}_{s}$ for
   some $s > 1 -\nu_0$ we apply $G_+$ to $E_{\vH}(\lambda_0)f$ and do
   an integration by parts. More precisely we compute as follows for
   any $g\in \vH^0_\infty$ using the smooth cut-off argument
   of Step \textit{II} of the proof of Theorem \ref{thm:negat-effect-potent}
   (the integration by parts) in the third step below:
   \begin{align*}
     \inp{g, G_+E_{\vH}(\lambda_0)f}&=\inp{G^*_+g,
                                      E_{\vH}(\lambda_0)f}=\lim_{R\to \infty}\inp{G^*_+g,
                                      \chi_RE_{\vH}(\lambda_0)f}\\&=\inp{-K^*_+g-g, f}=\inp{g, -K^+f-f},
   \end{align*} yielding \ref{itemG2} by a density argument. This
   argument uses conveniently that $vf\in \vH_{(3-s_0)^-}^{-1}$,
   which in turn follows from \eqref{eq:91BB}.

   \subStep{IV} For the assertion \eqref{eq:posI} it suffices to show
   that $\Im R_{j,\nu}\geq 0$. For $\nu\geq 0$ the kernel of
   $R_{j,\nu}$ is real and symmetric, whence $\Im R_{j,\nu}= 0$ in
   that case. If $\nu=-\i \sigma$ where $\sigma>0$ the kernel is
   \begin{align*}
     \parb{\Im R_{j,\nu}}(r,r')= \sigma\phi_\nu(r_<)\phi_\nu(r_>),
   \end{align*} and since $\phi_\nu$ is real,  it then follows that $\Im
   R_{j,\nu}\geq 0$ also in that case.

   If $G_+f=0$, $f=(f_1,f_2)$, then \eqref{eq:posI} yields
   $\chi_1f=0$. Whence
   \begin{align*}
     \parbb{\chi_2 r^{\tfrac{1-n_j}2}R_{j,\nu} r^{\tfrac{n_j-1}2}}\otimes
     P_{j,\nu}\chi_2 f_j=0\text{ for all }\nu\in\sigma_j.
   \end{align*} Since $R_{j,\nu}(r,r')$ is a Green's function it follows
   from these formulas that also $\chi_2f=0$. Therefore in turn $f=0$, and
   \ref{itemG3} is shown.
 \end{proof}

 \begin{remark*} The space $\vH^{1}_{-s_0}$ in \ref{itemG2} is not
   optimal for the assertion. It suffices to require that $f_j\in
   \vB^*_{s_0,0}(\bX_j)$, $j=1,2$. On the other hand if one uses a
   slightly smaller space than below  in the case  $s_0=1$
   to define
   the notion of a resonance, viz. $\vH^1_{(-1)^+}=\cup_{s>-1}\vH^{1}_{s}$
   rather than  $\vH^{1}_{-1}$, one avoids   the special treatment
   of  the case $s_0=1$ in
   Theorem \ref{thm:short-effect-potent} \ref{item:T3} and \ref{item:T4}.
 \end{remark*}

 \begin{defn}\label{defn:hom}
   \begin{enumerate}[(1)]
   \item \label{item:10a}$\lambda_0$ is called a resonance of $H$ if
     the equation $(H-\lambda_0)u = 0$ admits a solution
     $u \in H^{1}_{-s_0} \setminus H^1$. Such solution is a resonance
      state  of $H$. The
     multiplicity of the resonance $\lambda_0$ is defined as the
     dimension, say denoted by  $n_\mathrm {res}$,   of the quotient space
     \begin{align*}
       \ker (H-\lambda_0)_{|H^{1}_{-s_0}} / \ker (H-\lambda_0)_{|H^{1}}.
     \end{align*}
   \item\label{item:9a}$0$ is called a resonance of
     $ E_{\vH}(\lambda_0)$ if the equation $ E_{\vH}(\lambda_0)f = 0$
     has a solution $f \in \vH^{1}_{-s_0}\setminus \vH^1$. Such solution is a resonance
     state of $ E_{\vH}(\lambda_0)$. The multiplicity of the
     resonance $0$ is defined as the dimension of the quotient space
     \begin{align*}
       \ker E_{\vH}(\lambda_0)_{|\vH^{1}_{-s_0}} / \ker E_{\vH}(\lambda_0)_{|\vH^{1}}.
     \end{align*}
   \end{enumerate}
 \end{defn}

It follows from Lemma \ref{lem:4.1a} (applied with  $s=s_0$ and $s=0$)
that $\lambda_0$ is a resonance
   (resp., an eigenvalue) of $H$ if and only if $0$ is a resonance
   (resp., an eigenvalue) of $ E_{\vH}(\lambda_0)$ and their
   multiplicities are the same.

 We introduce for $j=1,2$,
 \[
 \sigma_{j,k} = \sigma_{j} \setminus (k,\infty)\mand
 \sigma^+_{j,k}=\sigma_{j} \cap (0,k]; \quad k \in \N.
 \]
 We take an orthonormal basis in $\ran P_{j,\nu}$ for each
 $\nu\in \sigma_j$, say
 $\zeta_{j, \nu}^{(1)}, \dots, \zeta_{j, \nu}^{(n_{j,\nu})}$. The set
 of resonances is partially determined by the set $\sigma^+_{j,1}$, as
 the following result shows.

\begin{thm}\label{thm:short-effect-potent} Suppose the conditions of
Proposition \ref{prop2.3} and \eqref{eq:posLowbnd0b} for $j=1,2$.
  \begin{enumerate}[1)]
  \item\label{item:T1} The dimension of the space of vectors
    $u\in H^{1}_{-s_0}$ solving $(H-\lambda_0)u=0$ is finite.

    If a vector $u\in H^{1}_{-s_0}$ obeys $(H-\lambda_0)u=0$
    then $f=T^*u\in\vH^{1}_{-s_0}$ and $ E_{\vH}(\lambda_0)f=0$,
    and conversely, if $f\in\vH^{1}_{-s_0}$ obeys
    $ E_{\vH}(\lambda_0)f=0$ then
    $u={E}_+(\lambda_0)f\in H^{1}_{-s_0}$ and
    $(H-\lambda_0)u=0$.

  \item\label{item:T2} A vector $f\in\vH^{1}_{-s_0}$ obeys
    $ E_{\vH}(\lambda_0)f=0$ if and only if $f\in
    \vH^{1}_{(\nu_0-1)^-}$ and  $f\in \ker (1+K^+)$; here
    $K^+$  is given by \eqref{eq:Kplus} (as an operator on
    $\vH^{1}_{-s}$ for any $s \in (1-\nu_0,2+\nu_0)$).
  \item\label{item:T3} Suppose $f \in \vH^{1}_{-s_0}$ and
    $E_{\vH}(\lambda_0) f =0$. In the  case $s_0=1$  suppose in addition
    that  $f \in \vH^{1}_{-s}$ for some $s<1$. Then the components of $f=(f_1,f_2) $
    can be decomposed as
    \begin{align} \label{eq4.2abb}
      \begin{split}
        -f_j(r_j\theta_j) &=\sum_{\nu\in\sigma_{j,1}} \sum_{k =
          1}^{n_{j, \nu}} l_{j,\nu,k}(f) { r_j^{\frac{2-n_j}{2} -\nu}
        }\chi_2(r_j){\zeta_{j,
            \nu}^{(k)}(\theta_j)} + g_j,\\&\text{ where  } g_j\in L^2\text{ and  }\\
        l_{j, \nu,k}(f)&= \inp[\big] {
          \phi_\nu(|y_j|)|y_j|^{-\frac{n_j-1}{2} }\otimes\zeta_{j,
            \nu}^{(k)}, \tilde f_j(y_j)}_{L^2(\d y_j)};\\
        \tilde f_j&=\chi_2\parb{vf}_j+\chi_2 B_jf_j+
        [\chi_2,p_j^2]f_j\,(\in L^2_{3-s_0}).
      \end{split}
    \end{align} Here
    \begin{align}
      \label{eq:imaog0}
      l_{j, \nu,k}(f)=0\text { for }\i \nu \geq0.
    \end{align}

    In particular
    \begin{align} \label{nonresonant1b}
      \begin{split}
        f \in \vH \quad \Longleftrightarrow \quad \forall &j=1,2,\,
        \nu \in \sigma^+_{j,1}, \, k=1,\dots, n_{j, \nu}:\\& \quad
        l_{j, \nu,k}(f) =0.
      \end{split}
    \end{align}

  \item\label{item:T4} Suppose $u \in H^{1}_{-s_0}$ and
    $(H-\lambda_0) u=0$.  In the  case $s_0=1$ suppose in addition
    that  $u \in H^{1}_{-s}$ for some $s<1$. Then
    \begin{align} \label{nonresonant2b}
      \begin{split}
        u \in L^2(\bX) \quad \Longleftrightarrow \quad \forall
        &j=1,2,\, \nu \in \sigma^+_{j,1}, \, k=1,\dots, n_{j, \nu}:\\&
        \quad l_{j, \nu,k}(T^*u) =0.
      \end{split}
    \end{align}
    In particular
  if $s_0\neq 1$, then  the multiplicity $n_\mathrm {res}$ of the resonance of
    $H$ at $\lambda_0$ (if existing) is bounded by
    \begin{align}
      \label{eq:dimres}
      n_\mathrm {res}\leq\sum^2_{ j=1}\,\,\sum_{\nu\in \sigma^+_{j,1}}
      \,n_{j, \nu}.
    \end{align}
 For  $s_0= 1$  the bound \eqref{eq:dimres} is valid  provided  the left-hand side is
 replaced by the dimension of the quotient space $\ker (H-\lambda_0)_{|H^{1}_{(-1)^+}} / \ker
 (H-\lambda_0)_{|H^{1}}$.

  \end{enumerate}
\end{thm}
\begin{proof}
  \subStep{I} For the second part of \ref{item:T1} we refer to
 Lemma  \ref{lem:4.1a}. To show the first part of \ref{item:T1}
  we use this correspondance and study the equation
  $E_{\vH}(\lambda_0)f=0$ with $f\in\vH^{1}_{-s_0}$. By Lemma
  \ref{lemHS} \ref{itemG2} we can write
  $-f=G_{+}E_{\vH}(\lambda_0)f+K^+f= K^+f$, and we recall that
  $K^{+}\in \vC(\vH^{1}_{-s})$ for any $s \in (1-\nu_0,2+\nu_0)$. In
  particular $f\in\vH^{1}_{-s'}$,  $K^{+}\in \vC(\vH^{1}_{-s'})$  and
  $f+K^+f=0$ for $s'=3/2+\nu_0$.  Whence it follows from Fredholm theory
  that the dimension of the space of functions $f\in \vH^{1}_{-s_0}$ solving
  $E_{\vH}(\lambda_0)f=0$ is finite. We have shown \ref{item:T1}.

  \subStep{II}  To show the `only if part' of \ref{item:T2} we
  consider any  $f\in\vH^{1}_{-s_0}\cap \ker (1+K^+)$. Due to Step
  \textit{I}  it suffices  to show that $f\in
    \vH^{1}_{(\nu_0-1)^-}$. This property is  trivially fulfilled  for $s_0=1$. If $s_0>1$ we
    argue as follows. By using \eqref{eq:Kplus} we conclude that
  $f\in\vH^{1}_{-s_1}$ for any $s_1>\max\{s_0-1,1-\nu_0\}$, and we are
  done if $s_{0}\leq 2-\nu_0$.  In general we pick the smallest
  $k\in \N$ such $s_{0}-k\leq 1-\nu_0$ and note that after $k-1$
  iterations of the above argument we get $f\in\vH^{1,}_{-s_k}$ for any
  $s_k> \max\{s_0-k,1-\nu_0\}=1-\nu_0$.

To show the `if part' of \ref{item:T2} suppose
  $f\in\vH^{1}_{(\nu_0-1)^-}\cap \ker (1+K^+)$. We need to show that
  $E_{\vH}(\lambda_0)f=0$.  By  an
  explicit calculation using \eqref{eq:Kplus} it follows that
  $E_{\vH}(\lambda_0)f=-E_{\vH}(\lambda_0)K^+f\in\vH^{-1}_{s}$ for all
  $s<2+\nu_0$, in particular for some $s>1-\nu_0$.  Whence by Lemma
  \ref{lemHS} \ref{itemG2} we have
  $-G_{+}E_{\vH}(\lambda_0)f=f+K^+f=0$. Invoking  Lemma
  \ref{lemHS} \ref{itemG3} we then conclude that $E_{\vH}(\lambda_0)f=0$.

  \subStep{III} We show \eqref{eq4.2abb} for any given
  $f \in \vH^{1}_{-s_0}$ obeying $E_{\vH}(\lambda_0) f =0$.  If
  $\nu_0=0$ by assumption
  $f\in \vH^{1}_{-s}$ for some $s<1$. If $\nu_0>0$ we know from Step \textit{II} that
  $f\in \vH^{1}_{-s}$ for all $s>1-\nu_0$. Whence
  \begin{align}
    \label{eq:uppsbnd1}
     \text{for some } s<1:\quad f\in \vH^{1}_{-s}.
  \end{align}
  Since $v$ and $B_j$ appearing in the definition of $l_{j, \nu,k}(f)$
  in \eqref{eq4.2abb}  are of order $r^{-3}$, we conclude
  that
  \begin{align}
    \label{eq:prbnd}
   \text{for some } s<1\mand j=1,2:\quad \tilde f_j\in  L^2_{3-s}.
  \end{align}
  On the other hand  for $\nu\in\sigma_{j,1}$ the vector
  $\chi_2(|y_j|)\phi_\nu(|y_j|)|y_j|^{-\frac{n-1}{2} }\otimes\zeta_{j,
    \nu}^{(k)}\in H^{1}_{(-2)^-}$, so   $l_{j,
    \nu,k}(f)$ is well-defined.
  As we saw in the proof of Lemma \ref{lemHS} the operator
  $r^{-s}R_{j,\nu}r^{-s'}$ is bounded on $L_I^2$ with a bound
  independent of $\nu\in \sigma_j$ provided
  $s, s' > 1 -\Re\nu \text{ and } s+ s'>2$. In particular if $\nu>1$ we
  can take $s=0$ and conclude boundedness of $R_{j,\nu}r^{-s'}$ for
  any $s'>2$. In combination with \eqref{eq:prbnd} we conclude that
  the  terms in  the expansion of $-K_+f$ corresponding to $\nu>1$ sum up
  to a vector in $L^2$.

  It remains to consider the contributions from
  $\nu\in\sigma_{j,1}$. We examine the term in $K_+f$ corresponding to
  any fixed $j$, $\nu\in \sigma_{j,1}$ and $ k=1,\dots, n_{j,
    \nu}$. Write it as
  \begin{align*}
    &\inp[\Big]{\zeta_{j,
      \nu}^{(k)},\parbb{\chi_2 r^{\tfrac{1-n_j}2}R_{j,\nu}
      r^{\tfrac{n_j-1}2}\tilde f_j}(r_j,\cdot)}_{L^2(\S_j)}\otimes{\zeta_{j,
      \nu}^{(k)}}\\& =r_j^{\tfrac{1-n_j}2}\chi_2(r_j)\parbb{l_{j,\nu,k}(f)
                     \psi_\nu (r_j)+f_{j,\nu,k}(r_j)}\otimes{\zeta_{j,
                     \nu}^{(k)}},
  \end{align*} where
  \begin{align*}
    f_{j,\nu,k}(r)&=   \phi_\nu (r)\int _r^\infty
                    \psi_\nu (r')\inp[\big]{\zeta_{j,
                    \nu}^{(k)},r'^{\tfrac{n_j-1}2}\tilde
                    f_j(r',\cdot)}_{L^2(\S_j)}\,\d r'\\
                  &-\psi_\nu (r)\int _r^\infty
                    \phi_\nu (r')\inp[\big]{\zeta_{j,
                    \nu}^{(k)},r'^{\tfrac{n_j-1}2}\tilde
                    f_j(r',\cdot)}_{L^2(\S_j)}\,\d r'.
  \end{align*} Using the Cauchy Schwarz inequality and \eqref{eq:prbnd}
  we obtain  that
  \begin{align*}
    \text{for } r\geq 1 \text{ and for some } s<1:|f_{j,\nu,k}(r)|\leq Cr^{s-3/2}.
  \end{align*} In particular we see that  the function $g$ given by   $
    g(r\theta)=r^{\tfrac{1-n_j}2}\chi_2(r)f_{j,\nu,k}(r){\zeta_{j,
    \nu}^{(k)}}(\theta)$  is in $L^2$.

  We have shown \eqref{eq4.2abb}.  Obviously \eqref{eq:imaog0} and
  \eqref{nonresonant1b} are consequences of \eqref{eq4.2abb}  and
  \eqref{eq:uppsbnd1}.

  \subStep{IV}  As for \ref{item:T4} we note that
  \eqref{nonresonant2b} follows from \ref{item:T1}, \ref{item:T3} and Lemma  \ref{lem:4.1a},
  and that  the remaining statements of \ref{item:T4}  follow
  from \eqref{nonresonant2b}.

\end{proof}

\begin{remark}\label{remark:mixed1} In this section we imposed the conditions of
  Proposition \ref{prop2.3} and considered three different cases of
  asymptotics of the effective inter-cluster interaction $W_j$,
  $j=1,2$. For symplicity we did not consider cases where these
  asymptotics mix. Due to the diagonal structure exhibited in
  Proposition \ref{prop2.3} it is easy to see that the diagonal
  parametrix construction works in such cases too. Thus for example it
  could be the case where $W_1$ is determined by a $I_1$ fulfilling
  \eqref{eq:virial0} while $W_2$ fulfills \eqref{eq:posLowbnd0b}. Then
  the first diagonal element of the parametrix should be the one in
  \eqref{eq:BesR} (taken with plus) and the second element should be
  the operator $G_+$ appearing in Lemma \ref{lemHS}. By using an
  analogue of \eqref{eq:Imvanissh} it follows by a similar iteration
  scheme as the ones used before that the resonance states of
  $E_{\vH}(\lambda_0)$  (defined similarly in this case) have a
  similar structure as in Theorem \ref{thm:short-effect-potent}
  \ref{item:T3}  and now constitute  a subspace of
  $\vH^{1,1}_{\infty,-s_0}$,  where $s_0$ is defined  as in
  \eqref{eq:crrit} in terms  of the (assumed) asymptotics
  of  $W_2$.
\end{remark}
\begin{remark}\label{remark:except1} Since we imposed the conditions of
  Proposition \ref{prop2.3} we did not treat the exceptional cases
  discussed in Sections \ref{sec:The case where lambda0insigma} and
  \ref{sec:The case when the condition {ass2}}. This is an additional
  technical issue that can be treated essentially by the same
  parametric construction as the one discussed above. We omit the
  details, referring the reader to Section \ref{sec:CoulRellich} for a
  treatment of   models of physics.
  \end{remark}

\section{The case $\lambda_0>\Sigma_2$}\label{sec:slow13}

Here we investigate the case where $\lambda_0>\Sigma_2$. We need a replacement
$R'(\lambda_0)\to R'(\lambda_0\pm \i 0)$ in the definition of
$E_{\vH}(\lambda_0)$ in Proposition \ref{prop2.3} (or rather in Remark
\ref{remark:formbnd}), giving rise to the notation
$E^\pm_{\vH}(\lambda_0)$. For that (and for other purposes)  we need various statements  from  Chapter \ref{Spectral
  analysis of H' near E_0}. The  basic  structure of the resolvent is
\begin{subequations}
\begin{align}\label{eq:basresbreve}
  \begin{split}
  R'(z)&= \breve R(z)-(p_1^2\Pi_1+p_2^2\Pi_2-z)^{-1}\Pi,\\
 R'(z)&=\Pi'\breve R(z)\Pi'=\Pi'\breve R(z)=\breve R(z)\Pi';
\quad\breve R(z)=(\breve H-z)^{-1}.
  \end{split}
\end{align}

We shall use the following properties that are parallel to
\eqref{eq:BesR} and \eqref{eq:micr1}. See Corollary
\ref{cor:microlocal-bounds} for the non-multiple   two-cluster case and
Remark \ref{remark:microlocal-boundsGEN} for the present multiple  case. There exists the strong
weak-star limits
\begin{align}\label{eq:bRES}
  \swslim_{\epsilon\to 0_+} \breve R(\lambda_0\pm \i \epsilon)=\breve R(\lambda_0\pm \i 0)\in \vL(\vB_{1/2}(\bX),\vB_{1/2}^*(\bX)),
\end{align} and if $\breve R(\lambda_0+ \i 0)f=\breve R(\lambda_0- \i
0)f$ for  a given $f\in L^2_t$ for some  $t>1/2$, then
 \begin{align}\label{eq:2bnd}
  \breve R(\lambda_0+ \i 0)f=\breve R(\lambda_0- \i 0)f\in  L^2_{t-1}.
\end{align}

\end{subequations}
We shall use the operators $h_j=p_j^2+w_j$, $j=1,2$, from Subsections
\ref{subsec:negat-effect-potent}--\ref{subsec:Homogeneous-effect-potent}. Recall
that $w_j-W_j$ are  real polynomially
decreasing potentials with varying meaning
reflecting the different hypotheses \eqref{eq:virial0},
\eqref{eq:posLowbnd0} and \eqref{eq:posLowbnd0b}, respectively. In this section
we modify each of the three subsections to cover the case $\lambda_0>\Sigma_2$.
For each case we write \eqref{eq:fundec} as
\begin{align}
  \label{eq:fundecs}
  -E^\pm_{\vH}(\lambda_0)=h_1 \oplus h_2 +v^\pm\in \vL( \vH^1, \vH^{-1}),
\end{align} where the only new feature is that the previous $v$ is
replaced by  operators $v^\pm$  defined in terms of the two
limits  $R'(\lambda_0\pm \i 0)$, respectively.  As before they are
non-local, but now no longer  symmetric. Due to
to \eqref{eq:bRES} they are roughly of order $\vO(|x|^{-1-2\rho})$, and
a priori  they do
not
have  `good commutation properties' when commuting with
multiplication operators (which is in contrast to the lowest threshold
setup of  Proposition
\ref{prop2.3}). More
precisely we need the following version of \eqref{eq:13} and \eqref{eq:efbNd}
\begin{equation}\label{eq:Eeffs}
  -E^\pm_{\vH}(\lambda_0)  \equiv  h_1\oplus h_2
  +(K^\pm_{ij}(\lambda_0))_{i,j\leq 2},
\end{equation} where the difference is not only polynomially
decreasing, but in fact also symmetric (to be used in \eqref{eq:posIm}).
Due to \eqref{eq:bRES} there are bounds
\begin{align}\label{eq:repla}
  K^\pm_{ij}(\lambda_0) &\in
                          \vL\parb{L^{2}_{s}(\bX_j),L^{2}_{t}(\bX_i)}\mforall s>-1/2-\rho\mand t<
                          1/2+\rho,
\end{align}  and we note that
\begin{align}
  \begin{split}\label{eq:posIm}
    \mp \Im v^\pm&=\mp \Im (K^\pm_{ij}(\lambda_0))_{i,j\leq 2}\geq
    0,\\
    \mp\Im \sum_{i,j\leq 2}\w{f_i,K^\pm_{ij}(\lambda_0)f_j}&=\pm
    \w{F,\Im R'(\lambda_0\pm\i 0)F};\\ \quad &F:=\sum_{i\leq 2}
\,F_i, \quad
    F_i=I_i\varphi_i\otimes f_i.
  \end{split}
\end{align}

\subsection{Extended eigentransform for  $\lambda_0>\Sigma_2$}\label{subsec: Extended
  eigentransform2}

Since the operator $\breve R(\lambda_0\pm \i 0)\Pi'$ cannot be expected to
preserve weighted spaces for $\lambda_0>\Sigma_2$  (as in  \eqref{eq:91BB})
we need to examine carefully  the `eigentransform' of Remark
\ref{remark:GrushinBB}, cf. Lemma \ref{lem:4.1a}.

First we  note the following modification of
\eqref{eq:sMap},
\begin{align}\label{eq:plusMap}
  \begin{split}
  E^{\pm}_+(\lambda_0)&=\parb{1- \breve R
(\lambda_0\pm\i 0)) \Pi'(H-\lambda_0)
       }S\in\vL(\vH^1
_s,H^1_t)\\
&\text{ provided }s>-1/2-\rho,\,t<
                          -1/2\mand t\leq s,
  \end{split}
\end{align} cf. the proof of Lemma \ref{lem:4.1a}. We shall use a
Besov space version of \eqref{eq:plusMap}, in fact we shall need (cf. the proof of Lemma \ref{lem:4.1a}) that
\begin{align}\label{eq:plusMap2}
  \begin{split}
  \inp{-\Delta}^{1/2}\breve R
&(\lambda_0\pm\i 0) \Pi'(H-\lambda_0)
       S\in\vL(\vB ,\vB^*_{1/2}(\bX));\\
&\quad \vB:=\vB_{1/2}(\bX_1)\oplus
       \vB_{1/2}(\bX_2).
  \end{split}
\end{align}

The following Lemma \ref{lem:eigentransform}  holds for  the cases
treated in  Subsections
\ref{subsec:negat-effect-potent}--\ref{subsec:Homogeneous-effect-potent},
in fact the setting is that of Proposition \ref{prop2.3} (now with
$\lambda_0>\Sigma_2$). We introduce $s'_0=1/2+\rho'/4$, where
 $0<\rho'< 4\rho$,  $\rho'\leq 2$  and
\begin{align}\label{eq:newcond}
  \abs{W_j(y)}\leq C \inp{y}^{-\rho'}\text{ for }|y|\geq R;\quad R \text{
    large}.
\end{align}

If $\rho=1$ (our main interest)  the condition
simplify as \eqref{eq:newcond} for some $\rho'\leq 2$, and in that
case we would choose $\rho'=1$ and $\rho'=2$ in Subsections
\ref{subsec:negat-effect-potents} and \ref{subsec:hd2ep},
respectively. On the other hand  Subsection \ref{subsec:psdep} does  not involve
\eqref{eq:newcond} and Lemma \ref{lem:eigentransform} directly, but rather  a
 version of Lemma \ref{lem:eigentransform} stated below as Lemma
\ref{lem:eigentransform2}. We regard Lemmas \ref{lem:eigentransform} and \ref{lem:eigentransform2} as
 substitutes for Lemma \ref{lem:4.1a}. Recall the notation $\vH_s=
L^2_{s}(\bX_1)\oplus L^2_{s}(\bX_2)$.

\begin{lemma}\label{lem:eigentransform} Suppose the conditions of
  Proposition \ref{prop2.3} (except for assuming now
  $\lambda_0>\Sigma_2$).  Let
  $s\leq s'_0$ and $t\in (0,s'_0]$.
  \begin{enumerate}[1)]
  \item \label{item:1ba}  The map $T^*: \vE^\vG _{-s}\to
    \vE^\vH_{-s}$, where
    \begin{align*}
       \vE^\vG _{-s}:=&\{u\in L^2_{-s}(\bX)\mid
      (H-\lambda_0)u=0,\,\,\Pi' u\in \vB^*_{1/2,0}(\bX)\}, \\
      \vE^\vH_{-s}:=&\{f\in  \vH_{-s}\mid  E^{+}_{\vH}(\lambda_0)f=0\},
    \end{align*} is a well-defined linear isomorphism with inverse  $E^{+}_+(\lambda_0):
    \vE^\vH _{-s}\to \vE^\vG _{-s}$.
  \item \label{item:2ba} The spaces $\vE^\vH _{-s}$ and $\{f\in \vH_{-s}\mid  E^{-}_{\vH}(\lambda_0)f=0\}$
     coincide, and
    $E^{-}_+(\lambda_0)f=E^+_+(\lambda_0)f$ for all $f\in \vE^\vH _{-s}$.
  \item \label{item:3ba} The map
      $T^*: \ker (H-\lambda_0)\to \ker  E^{+}_{\vH}(\lambda_0)_{|\vH}$ is a
      well-defined linear isomorphism with inverse  $E^{+}_+(\lambda_0)$.
\item \label{item:1baB} The map $T^*: \vE^\vG _{-t,0}\to
  \vE^\vH_{-t,0}$, where
    \begin{align*}
      \vE^\vG _{-t,0}:=&\{u\in \vB^*_{t,0}(\bX) \mid
      (H-\lambda_0)u=0,\,\,\Pi' u\in \vB^*_{1/2,0}(\bX)\},\\ \vE^\vH_{-t,0}:=&\{f\in  \vB^*_{t,0}(\bX_1)\oplus
         \vB^*_{t,0}(\bX_2)\mid  E^{+}_{\vH}(\lambda_0)f=0\},
    \end{align*} is a well-defined linear isomorphism with inverse  $E^{+}_+(\lambda_0):
    \vE^\vH_{-t,0}\to \vE^\vG _{-t,0}$.
  \item \label{item:2baB} The spaces $
    \vE^\vH_{-t,0}$ and $\{f\in \vB^*_{t,0}(\bX_1)\oplus
         \vB^*_{t,0}(\bX_2)\mid  E^{-}_{\vH}(\lambda_0)f=0\}
    $ coincide,  and
    $E^{-}_+(\lambda_0)f=E^+_+(\lambda_0)f$ for all $f\in
    \vE^\vH_{-t,0}$.
\item \label{item:1baBC} There exists $\sigma<0$ such that  for any real function
  $\chi_\sigma \in C^\infty(\R)$
    which is   supported in $(-\infty,\sigma/2)$ and which is $1$ on
    $(-\infty,\sigma)$,  $\vE^\vG
    _{-s'_0,0}$ and the two spaces
    \begin{align*}
       \vE^{\vG, \pm} _{-s'_0,\pm\sigma}:=\{u\in \vB^*_{s'_0,0}(\bX) \mid
      (H-\lambda_0)u=0,\,\,\Pi' u&\in \vB^*_{1/2}(\bX),\\\,\,\chi_\sigma(\pm B_{R_0})\Pi' u&\in \vB^*_{1/2,0}(\bX)\}
    \end{align*} all coincide. (Here $B=B_{R_0}$ is given by \eqref{eq:Bdefined}.)
  \end{enumerate}
\end{lemma}

\begin{proof} \subStep{I}
 Let $u\in \vE^\vG _{-s}$ be given. Then
  $f=T^*u\in \vH_{-s}$
  follows easily, so to show that $f\in \vE^\vH _{-s}$ it remains to show
  that $E^{\pm}_{\vH}(\lambda_0)f=0$. Note that
  $\Pi'(H-\lambda_0)\Pi u\in \vB_{1/2}(\bX)$, cf.
  Lemma \ref{Lemma:basic2}, so that
  $ R'(\lambda_0\pm \i 0)(H-\lambda_0)\Pi u$ is a well-defined
  element of $\vB_{1/2}^*(\bX)$. We calculate using this fact,
  \eqref{eq:commIdent} and \eqref{eq:23}
  \begin{align*}
    E^{\pm}_{\vH}(\lambda_0)f&=S^*(\lambda_0-H)ST^*u+S^*H R'(\lambda_0\pm
                               \i 0)(H-\lambda_0)ST^*u\\
                             &=S^*(\lambda_0-H)\Pi u+S^*H R'(\lambda_0\pm
                               \i 0)(H-\lambda_0)\Pi u\\
                             &=S^*(\lambda_0-H)\Pi u+S^*(H -\lambda_0)\Pi'\breve R(\lambda_0\pm
                               \i 0) (H-\lambda_0)\Pi u\\
                             &=S^*(\lambda_0-H)\Pi u+\lim_{R\to \infty}S^*(H -\lambda_0)\Pi'\breve R(\lambda_0\pm
                               \i 0) (H-\lambda_0)\Pi\chi_Ru;
  \end{align*}  here $\chi_R(x)=\chi(|x|/R)$, $R\geq 1$, is  given in
  agreement with \eqref{eq:14.1.7.23.24}. We calculate the second term as follows using the
  weak-star topology on  $\vB^*$ where $\vB=\vB_{1/2}(\bX_1)\oplus
  \vB_{1/2}(\bX_2)$
  \begin{align*}
    &\lim_{R\to \infty}S^*(H -\lambda_0)\Pi'\breve R(\lambda_0\pm
      \i 0) (H-\lambda_0)\Pi\chi_Ru\\
    &=-\wslim_{R\to \infty}S^*(H -\lambda_0)\Pi'\breve R(\lambda_0\pm
      \i 0) (H-\lambda_0)\Pi'\chi_Ru\\&\quad+\wslim_{R\to \infty}S^*(H -\lambda_0)\Pi'\breve R(\lambda_0\pm
                                           \i 0) (H-\lambda_0)\chi_Ru\\
    &=\wslim_{R\to \infty}S^*(\lambda_0-H)\Pi'\chi_Ru\\&\quad+\wslim_{R\to \infty}S^*(H -\lambda_0)\Pi'\breve R(\lambda_0\pm
                                                                \i 0) [H-\lambda_0,\chi_R]u\\
    &=S^*(\lambda_0-H)\Pi'u\\&\quad+\wslim_{R\to \infty}S^*(H -\lambda_0)\Pi'\breve R(\lambda_0\pm
                                  \i 0) [H-\lambda_0,\chi_R]u.
  \end{align*}
  Suppose we can show that the second term
  \begin{align}
    \label{eq:olimit}
    \wslim_{R\to \infty}S^*(H -\lambda_0)\Pi'\breve R(\lambda_0\pm
    \i 0) [H-\lambda_0,\chi_R]u=0.
  \end{align} Then we obtain from  the above computations the desired
  result
  \begin{align*}
    E^{\pm}_{\vH}(\lambda_0)f&=S^*(\lambda_0-H)\Pi u+S^*(\lambda_0-H)\Pi'u\\
                             &=S^*(\lambda_0-H)u \\ &=0.
  \end{align*}

  So it remains to show \eqref{eq:olimit}. We note that
  \begin{align*}
    -[(H-\lambda_0),\chi_R] u= 2\nabla\cdot (\nabla
    \chi_R) \,u-(\Delta \chi_R)u
  \end{align*} and that for all the components $\nabla_k$ of $\nabla$,
  the operator $ \nabla_k$ combines with the operator to the left
  in \eqref{eq:olimit},  and for this combination
  we
may use \eqref{eq:plusMap2}. Keeping this `effective boundedness' in mind
we  treat the first term by   substituting
\begin{align*}
  u=\Pi u+\Pi'u=\Pi_1 u+\Pi_2 u+(\Pi-SS^*)u+\Pi'u.
\end{align*}
 For $j=1,2$
  \begin{align}\label{eq:cross}
    \Pi'\nabla\cdot (\nabla
    \chi_R) \Pi_j u=\Pi'\nabla_{x_{j}}\cdot (\nabla
    \chi_R)_{|x^{j}=0} \Pi_j u+\vO(R^{-2})\Pi u=0+\vO(R^{-2}) u,
  \end{align} and therefore the  first and second terms  do not contribute in the
  limit. The third term is treated by Lemma \ref{Lemma:basic2}.
Similarly
\begin{align*}
    \Pi'\nabla\cdot (\nabla
    \chi_R) \Pi' u=\vO(R^{-1})\Pi' u
  \end{align*} does not contribute in the limit,
and obviously the term $(\Delta \chi_R)u=\vO(R^{-2})
 u$ does not neither. We conclude   \eqref{eq:olimit}.

  \subStep{II}  First we note that $E^{\pm}_{\vH}(\lambda_0)f$ has a
  well-defined distributional meaning for
  $f\in \vH_{-s}$ so that
  $\vE^\vH _{-s}$ is well-defined. We show that
  $E^{+}_{\vH}(\lambda_0)f=0$ if and only if $E^{-}_{\vH}(\lambda_0)f=0$.
  Suppose $E^{+}_{\vH}(\lambda_0)f=0$. First we show that
  \begin{align}
    \label{eq:Im0}
    \Im\w{f_j,\check h_jf_j}=0;\quad \check h_j=-\Delta_j +W_j,\,
 j=1,2.
  \end{align} Note that $\check h_jf_j\in L^2_s$ for any $s<1/2+\rho$,
    $f_j\in L^2_{-s'_0}$  and that $s'_0< 1/2+\rho$, so that
  $\w{f_j,\check h_jf_j}$ is well-defined. We compute
\begin{align*}
     \chi_R \check h_jf_j = \check h_j\chi_Rf_j+(\Delta \chi_R)f_j+2(\nabla
     \chi_R)\cdot \nabla f_j,
   \end{align*} and then
   \begin{align*}
     \Im\w{f_j,\check h_jf_j}=\lim_{R\to \infty}\Im\w{\chi_Rf_j,\chi_R\check h_jf_j}=2\lim_{R\to \infty}\Im\w{\chi_Rf_j,(\nabla
     \chi_R)\cdot \nabla f_j}
   \end{align*}
Next we estimate as in \eqref{eq:compCom}
   \begin{align*}
     &R^2\|(\nabla\chi_R)\cdot \nabla f_j\|^2_{L^2(\bX_j)}\\&\leq
                                                              C_1\inp{p^2}_{\chi_{2R}\bar\chi_{R/2}f_j}
     \\&\leq
         C_2\parb{\Re\inp{\chi^2_{2R}\bar\chi^2_{R/2}f_j,\check h_jf_j}+R^{-\rho'}\|\chi_{4R}f_j\|^2}
     \\&= C_2\parb{o(R^0)+R^{-\rho'}o(R^{2s'_0})}\\&=o(R^{1-\rho'/2}),
   \end{align*}
   yielding
   \begin{align*}
     \lim_{R\to \infty}\,R^{s'_0}\|(\nabla
     \chi_R)\cdot \nabla f_j\|_{L^2(\bX_j)}=0.
   \end{align*}
Moreover
\begin{align*}
  \lim_{R\to \infty}\,R^{-s'_0}\|\chi_Rf_j\|_{L^2(\bX_j)}=0.
\end{align*} Whence we obtain \eqref{eq:Im0} by \caS.

   Next we compute by using \eqref{eq:posIm} and \eqref{eq:Im0} that
  \begin{align}\label{eq:im0}
    \begin{split}
      0=\Im\w{f,E^{+}_{\vH}(\lambda_0)f}&=\Im\w{F, R'(\lambda_0+\i 0)F};\\
      \quad &F=\sum_{i\leq 2}F_i, \quad F_i=I_i\varphi_i\otimes f_i.
    \end{split}
  \end{align} Whence also $\Im\w{F, R'(\lambda_0-\i
    0)F}=0$ and we learn that $R'(\lambda_0+\i
  0)F=R'(\lambda_0-\i
  0)F$, and therefore that also $E^{-}_{\vH}(\lambda_0)f=0$. We can
  argue similarly for the other implication.

   \subStep{III}  Suppose $f\in \vE^\vH _{-s}$. We need to show that
  $u:=E^+_+(\lambda_0)f\in \vE^\vG _{-s}$ and that $T^*u=f$. Clearly
  $Sf\in L^2_{-s}(\bX)$. We noted in Step \textit{II} that
  $R'(\lambda_0+\i 0)F=R'(\lambda_0-\i 0)F$ where $F$ is specified in
  \eqref{eq:im0}. This means that $u=E^-_+(\lambda_0)f$, and using
  \eqref{eq:basresbreve}--\eqref{eq:2bnd} we then conclude that
  $u\in L^2_{-s}(\bX)$ and $\Pi' u\in \vB^*_{1/2,0}(\bX)$. To show that $u\in\vE^\vG _{-s}$ it
  remains to show that $(H-\lambda_0)u=0$. We calculate
  \begin{align*}
    (H-\lambda_0)u&= (H-\lambda_0)Sf-(H-\lambda_0)R'(\lambda_0+\i
                       0)HSf\\
                     &= (H-\lambda_0)Sf-\Pi'(H-\lambda_0) R'(\lambda_0+\i
                       0)(H-\lambda_0)Sf\\&\quad-\Pi(H-\lambda_0) R'(\lambda_0+\i
                                                0)(H-\lambda_0)Sf
    \\&=(H-\lambda_0)Sf-\Pi'(H-\lambda_0)Sf-\Pi H R'(\lambda_0+\i
        0)HSf\\
                     &=\Pi(H-\lambda_0)Sf-\Pi H R'(\lambda_0+\i
                       0)HSf.
  \end{align*}
  We conclude that $(H-\lambda_0)u\in\Pi L^2_{-s}(\bX)$ (since
  $\Delta f\in  \vH_{-s}$) and that
  \begin{align*}
    S^*(H-\lambda_0)u=-E^{+}_{\vH}(\lambda_0)f=0,
  \end{align*} and therefore (since
  $S^* $ maps injectively on the space $\Pi L^2_{-s}(\bX)$) that
  indeed $(H-\lambda_0)u=0$.

  Next we compute
  \begin{align*}
    T^*E^+_+(\lambda_0)f&=T^*Sf-T^*R'(\lambda_0+
                          \i 0)(H-\lambda_0)Sf\\
                        &=f-0\quad(\text{since }T^*\Pi'=0)\\
                        &=f.
  \end{align*}

  \subStep{IV} Let $u\in\vE^\vG _{-s}$ and note that we showed in Step
  \textit{I} that $f:=T^*u\in \vE^\vH _{-s}$. Using parts of the proof we show
  that $u=E^+_+(\lambda_0)f$ as follows.
  \begin{align*}
    E^+_+(\lambda_0)T^*u&=ST^*u-R'(\lambda_0+
                             \i 0)(H-\lambda_0)ST^*u\\
                           &=\Pi u-R'(\lambda_0+
                             \i 0)(H-\lambda_0)\Pi u\\
                           &=\Pi u-\lim_{R\to \infty}\Pi'\breve R(\lambda_0+
                             \i 0) (H-\lambda_0)\Pi\chi_Ru\\
                           &=\Pi u+\Pi'u\\
                           &=u.
  \end{align*} This finishes the proof of \ref{item:1ba} and
  \ref{item:2ba}, and
\ref{item:3ba} is a special case of \ref{item:1ba}. The statements
\ref{item:1baB} and \ref{item:2baB} are proved verbatim  as \ref{item:1ba} and
  \ref{item:2ba}.

\subStep{V} Obviously $\vE^\vG _{-s'_0,0}\subset\vE^{\vG, \pm} _{-s'_0,\pm\sigma}$.  We choose $\sigma< 0$ such that
\begin{align*}
  \bar\chi_{\sigma}( B_{R_0})\breve R(\lambda_0-\i 0)\in
  \vL(\vB_{1/2}(\bX), \vB^*_{1/2,0}(\bX));
\quad\bar\chi_{\sigma}( B_{R_0})=\parb{1-\chi_{\sigma}( B_{R_0})},
\end{align*} cf. Lemma \ref{lemma:microLoc}.
 We show the opposite inclusion for $\vE^{\vG, +} _{-s'_0,\sigma} $ only. So let $u\in \vE^{\vG, +} _{-s'_0,\sigma} $ be
given. Then we can mimic Step \textit{I} and see that also in this case
$f=T^*u\in \vE^\vH_{-s'_0,0}$. Note that for showing \eqref{eq:olimit}
we can use the above bound  to estimate
\begin{align*}
 & \wslim_{R\to \infty}S^*(H -\lambda_0)\Pi'\breve R(\lambda_0+
    \i 0) [H-\lambda_0,\chi_R]u\\
&=-2\wslim_{R\to \infty}S^*(H -\lambda_0)\Pi'\breve R(\lambda_0+
    \i 0)  \Pi'\nabla\cdot (\nabla
    \chi_R) \parbb{\chi_{\sigma}( B_{R_0})+\bar\chi_{\sigma}(
      B_{R_0})}\Pi' u \\
&= 0+0\\&=0.
\end{align*} We use the same argument mimicking Step \textit{IV},
cf. \ref{item:1baB}. So indeed $u\in \vE^\vG _{-s'_0, 0}$, and we have
shown that $\vE^{\vG, +} _{-s'_0,\sigma}=\vE^\vG _{-s'_0,0}$.
\end{proof}

In the regime $s\in (s'_0, 1/2+\rho)$  we can define $\vE^{\vH}_{-s}$
as in Lemma \ref{lem:eigentransform} \ref{item:1ba}, however we are
not able to conclude Lemma \ref{lem:eigentransform} \ref{item:2ba} is
this case. This leads to the following definition
\begin{align*}
  \vE^\vH_{-s}         :=\{f\in  \vH_{-s}|\quad
  E^{+}_{\vH}(\lambda_0)f=E^{-}_{\vH}(\lambda_0)f=0\}  \quad \text{for} \quad s<1/2+\rho,
\end{align*} which is consistent with Lemma
\ref{lem:eigentransform}. Obviously we can also extend the definition
of
$\vE^\vG_{-s}  $ given in Lemma
\ref{lem:eigentransform} \ref{item:1ba} to any real $s$.

\begin{lemma}\label{lem:eigentransform2} Suppose the conditions of
  Proposition \ref{prop2.3} (except for assuming now
  $\lambda_0>\Sigma_2$).  Suppose
  $s\leq 3/2$ and $s<1/2+\rho$.
   Then the  map
    \begin{align*}
      T^*: \vE^\vG _{-s}\to  \vE^\vH_{-s}
    \end{align*} is a well-defined linear isomorphism with inverse  $E^{+}_+(\lambda_0):
    \vE^\vH _{-s}\to \vE^\vG _{-s}$ (defined by \eqref{eq:plusMap}). Moreover  $E^{-}_+(\lambda_0)f=E^+_+(\lambda_0)f$ for all $f\in \vE^\vH
    _{-s}$.

\end {lemma}
\begin{proof}
  We can assume \eqref{eq:newcond} and that $s\in (s'_0, 1/2+\rho)$,
  in particular that $s>1/2$.
   It follows from Step \textit{I}
  of the proof of Lemma
\ref{lem:eigentransform} that $T^*$ maps into
$\vE^\vH_{-s}$. Obviously $E^{+}_+(\lambda_0)$ maps into
$L^2_{-s}(\bX)$, but to show
that $E^{+}_+(\lambda_0)$ in fact maps into $\vE^\vG _{-s}$ we cannot use
\eqref{eq:im0}, however the following substitute works:
\begin{align}\label{eq:im02}
    \begin{split}
      0=\Im\w{f,E^{+}_{\vH}(\lambda_0)f-E^{-}_{\vH}(\lambda_0)f}&=2\Im\w{F, R'(\lambda_0+\i 0)F};\\
      \quad &F=\sum_{i\leq 2}F_i, \quad F_i=I_i\varphi_i\otimes f_i.
    \end{split}
  \end{align} This shows that $R'(\lambda_0+\i 0)F=R'(\lambda_0-\i
  0)F$ for all $f\in \vE^\vH
    _{-s}$, and therefore that $\Pi'E^{+}_+(\lambda_0)f\in
    \vB^*_{1/2,0}(\bX)$  and $E^{-}_+(\lambda_0)f=E^+_+(\lambda_0)f$ for all $f\in \vE^\vH
    _{-s}$. Next we use Step \textit{III}
  to see that also
$(H-\lambda_0)E^{+}_+(\lambda_0)f=0$ showing that indeed
$E^{+}_+(\lambda_0)$  maps into $\vE^\vG _{-s}$.

Finally it follows from the last part of Step \textit{III} and of Step \textit{IV} that
$T^*$ and $E^{+}_+(\lambda_0)$ are mutually inverses, as we want.
\end{proof}

 \subsection{Negative slowly decaying effective
  potentials}\label{subsec:negat-effect-potents} We aim at proving a version of
 Theorem
\ref{thm:negat-effect-potent} in the  setting of  Subsection
\ref{subsec:negat-effect-potent} (except that now
$\lambda_0>\Sigma_2$). In particular $\rho<2$ and we can use  Lemma
   \ref{lem:eigentransform}  with $ s_0'=s_0=1/2+\rho/4$ (corresponding to
   taking $\rho=\rho'$).

 \begin{thm}\label{thm:negat-effect-potents} Under the conditions of
   Proposition \ref{prop2.3} (except for assuming now
   $\lambda_0>\Sigma_2$) and \eqref{eq:virial0}  for $j=1,2$
   \begin{align*}
     \dim \ker(H-\lambda_0)<\infty.
   \end{align*}
   Moreover   $ \vE^\vG _{-s_0,0}\subset H^1_{\infty}$.
 \end{thm}
 \begin{proof} We basically mimic the proof of Theorem
   \ref{thm:negat-effect-potent} using  Lemma \ref{lem:eigentransform} as
 a substitute for Lemma \ref{lem:4.1a}.

 Note that Lemma
   \ref{lem:eigentransform} \ref{item:3ba} implies that
   \begin{align}\label{eq:dimeq}
     \dim \ker(H-\lambda_0)=\dim \ker E^\pm_{\vH}(\lambda_0),
   \end{align} and that the map
   \begin{align}\label{eq:mapprop}
     E^\pm_+(\lambda_0): \ker
     E^\pm_{\vH}(\lambda_0)\to \ker(H-\lambda_0)
   \end{align} is an isomorphism.

   We shall use the notation $h_1$, $h_2$ and
   $R^{\pm}_{\diag}=r_1(0\pm \i 0)\oplus (r_2(0\pm \i 0)$ from the proof of Theorem
   \ref{thm:negat-effect-potent}.  Suppose $f\in \ker
   E^\pm_{\vH}(\lambda_0)$.  Then we write the equations
   $R^{\pm}_{\diag}E^\pm_{\vH}(\lambda_0)f=0$ as
   \begin{align}\label{eq:LipSchWs}
     (1+K^{\pm})f=0, \quad K^{\pm}=R^{\pm}_{\diag}v^\pm,
   \end{align} where $1f$ arises by writing $f=(f_1,f_2)\in
   \vH^{1}$ and using  that
   \begin{align}\label{eq:leftinverses}
     r_j(0\pm \i 0)h_j f_j=f_j.
   \end{align} Note that it is shown in Step \textit{II} of the proof of
   Theorem \ref{thm:negat-effect-potent} that generally,  if $f_j\in
   \vB^*_{s_0,0}(\bX_j)$
   obeys $h_jf_j\in \vB_{s_0}(\bX_j)$, then \eqref{eq:leftinverses}
   holds. In particular \eqref{eq:leftinverses} holds for $f\in \ker
   E^\pm_{\vH}(\lambda_0)$.

 Next  \eqref{eq:repla}
   and \eqref{eq:BesR} imply that
   \begin{align}\label{eq:comnegs}
     K^{\pm}\in \vC(\vH^{1}_{-s}) \text{
     for }s\in (s_0, 1/2 +\rho).
   \end{align}
   Since $K^{\pm}$ is compact (on any such space) it follows from
   Fredholm theory that $\dim \ker E^\pm_{\vH}(\lambda_0)< \infty$,
   proving the first assertion of the theorem.

   By Lemma \ref{lem:eigentransform} any given
   $\phi\in \vE^\vG _{-s_0,0}$ corresponds to the vector
   $f=(f_1,f_2)=T^*\phi\in \vE^\vH _{-s_0,0}$ (obeying $\phi=E^+_+(\lambda_0)f=E^-_+(\lambda_0)f$). Note also that    $h_jf_j\in L_s^{2}(\bX_j)\subset
   \vB_{s_0}(\bX_j)$ for any $s$ given as
   in \eqref{eq:comnegs}, so by the  assertion in  Step \textit{II} of the proof of
   Theorem \ref{thm:negat-effect-potent}
    we can conclude that \eqref{eq:LipSchWs} is
   fulfilled. In particular $K^+f=K^-f$.

   A converse assertion is true. We first observe  using \eqref{eq:posIm},  that
   $\Im v^+\leq 0$.  Since $\Im R^+_{\diag}\geq0$ the equation
   \begin{align}\label{eq:zeroTrac}
     0=\Im \w{f,v^+f}+\Im \w{R^+_{\diag}v^+f,v^+f}
   \end{align} for any solution to $(1+K^+)f=0$, where
   $f=(f_1,f_2)\in \vB^*_{s_0}(\bX_1)\oplus \vB^*_{s_0}(\bX_2)$, then yields that
   $v^-f=v^+f$ and that also $(1+K^-)f=0$ for any such
   $f$. (Similarly a solution to $(1+K^-)f=0$ is also a solution to
   $(I+K^+)f=0$.)
Moreover $E^\pm_{\vH}(\lambda_0)f=0$ in the
   distributional sense indeed yielding a converse statement.

More
   importantly we can improve the decay of any  $f=(f_1,f_2)\in
   \vB^*_{s_0}(\bX_1)\oplus \vB^*_{s_0}(\bX_2)$  solving  $(1+K^+)f=0$
   (and therefore also $(1+K^-)f=0$)  by invoking repeatedly \eqref{eq:micr1} and \eqref{eq:2bnd}:
 Since $f_j\in  \vB^*_{s_0}(\bX_j)$, $j=1,2$, and $v^-f=v^+f$ (as
 noted above) we
   deduce from \eqref{eq:2bnd}
    that $v^+f\in \vH^{-1}_t$ for any
   $t<1+2\rho-s_0 $. Next using \eqref{eq:micr1}  (in combination with
   \eqref{eq:zeroTrac}) we deduce
   that $f=-K^+f\in \vH^{1}_{ \epsilon-s_0}$ for any
   $\epsilon<\epsilon_0$, where
   $\epsilon_0=1+2\rho-2s_0 =\tfrac 32 \rho$.  Repeating this
   argument, say $k$ times, yields $f\in \vH^{1}_{t-s_0}$ for any
   $t< \epsilon_0k$, and therefore that
   $f\in \vH^{1}_{\infty}$.

   We have shown that for  any given $\phi\in \vE^\vG _{-s_0,0}$  the
   vector
   $f=(f_1,f_2)=T^*\phi\in \vH^{1}_{\infty}$ and therefore,  due to Lemma
   \ref{lem:eigentransform} \ref{item:1ba},
   indeed $\phi=E^+_+(\lambda_0)f\in H^{1}_{\infty}$.
   \end{proof}

\subsection{Positive slowly decaying effective
   potentials}\label{subsec:psdep}  We aim at proving a version of Theorem
 \ref{thm:pos-effect-potent} in the  setting of Subsection
 \ref{subsec:positive-effect-potent}
 (except that now $\lambda_0>\Sigma_2$).

\begin{thm}\label{thm:pos-effect-potent2} Under the conditions of
  Proposition \ref{prop2.3} and \eqref{eq:posLowbnd0} for $j=1,2$,
  \begin{align*}
    \dim \ker(H-\lambda_0)<\infty.
  \end{align*}
  Moreover $\vE^\vG _{-s}\subset H^1_{\infty}$  for and real $s $  such
  that $s\leq 3/2$ and  $s<1/2+\rho$.
\end{thm}
\begin{proof} Due to  Lemma  \ref{lem:eigentransform2} we need to study
     the space $\vE^\vH_{-s}$ for any such $s$.
Using the notation $\widetilde{R}=h_1^{-1}\oplus
  h_2^{-1}$ of \eqref{eq:LipSchW2bb}   as well as the notation \eqref{eq:fundecs},  we  write the equation  $\widetilde{R}E^+_{\vH}(\lambda_0)f=0$ for
  any $f\in \vE^\vH_{-s}
 $ as
  \begin{align}\label{eq:LipSchW3}
    (1+\widetilde{K})f=0; \quad \widetilde K=\widetilde{R}v^+.
  \end{align} Due to Lemma
  \ref{lemma:posit-slowly-decay} we can  easily check  that $\widetilde{K}\in \vC(\vH^1_{-s})$.
  It then  follows from Fredholm theory that
  $f$ belongs to a finite-dimensional subspace of $\vH^1_{-s}$,
    yielding in particular the first statement of the theorem by
    taking
    $s=0$.

     Finally we show
\begin{align}
     \label{eq:incl}
     \vE^\vH_{-s}\subset \vH^1_\infty,
   \end{align} using   an
    iteration procedure as in the proof of Theorem
    \ref{thm:negat-effect-potents}.  Note that due to  Lemma
    \ref{lem:eigentransform2} we can  use \eqref{eq:2bnd}  and
    \eqref{eq:LipSchW3} to improve the decay. We deduce that
    $f\in\vH^1_{-s+\bar t}$ given that $f\in\vH^1_{-s+(k-1)\bar t}$
    where
    $\bar t:=1+2\rho-\bar\rho$. Since there is no limit on $k(\in \N)$ used for
    this argument, indeed \eqref{eq:incl}   follows by iteration. By Lemma
    \ref{lem:eigentransform2} and  \eqref{eq:incl} it follows that  $\vE^\vG
    _{-s}\subset H^1_{\infty}$.
   \end{proof}

\subsection{Homogeneous degree $-2$ effective
   potentials}\label{subsec:hd2ep}  We aim at proving a version of
 Theorem \ref{thm:short-effect-potent} in a setting similar to that of
 Subsection \ref{subsec:Homogeneous-effect-potent}. Now of course
 $\lambda_0>\Sigma_2$, and having the physics examples in mind we demand
 $\rho=1$ rather than $\rho\geq 1/2$ as before. It is easy to see that our
 theory is void if $\rho=1/2$ is kept. On the other hand there is something to say in
 the case $\rho\in (1/2,1)$, but for simplicity of presentation we
 leave this case out.  It turns out that the proof of
 Theorem \ref{thm:short-effect-potent} works again with  only minor
 modifications, in particular  we will need  Lemma
 \ref{lem:eigentransform} \ref{item:2ba}  with $s'_0=1$.

 First we examine how  Lemma \ref{lemHS} can be modified for
 $\lambda_0>\Sigma_2$. Due to  restricted mapping proporties of $v^+$
 we can  only obtain the following weakened result (by the same
 proof).  Let $s_1=\max\set{1-\nu_0, 1/2}$, and recall that
 $s_0=1+\nu_0$ in agreement with \eqref{eq:crrit}.
 \begin{lemma}\label{lemHS2}
   \begin{enumerate} [1)]
   \item\label{itemG12} The operators
     \begin{align}\label{HS02}
       &G_{+},G^*_{+} \in \vL(\vH^{-1}_{s'},\vH^{1}_{-s}) \text{ for }s, s' > 1 -\nu_0 \text{ with }  s+ s'>2,\\
       & K^*_+:=-E^+_{\vH}(\lambda_0)^*G^*_{+}-1 \in \vC(\vH^{-1}_{s}) \text{ for
            }s \in (s_1, 3/2);\label{HS12}
     \end{align} here  $E^+_{\vH}(\lambda_0)^*$ is given  the
     distributional meaning, and we  define
     $K^+=(K^*_+)^*\in\vC(\vH^{1}_{-s})$.  Explicitly
     \begin{align}
       \label{eq:Kplus2}
       \begin{split}
         K^+&=G_+v^++\oplus_{j=1}^2\parbb{\chi_1(h_j-\i)^{-1}\parb{\i\chi_1+[\chi_1,p_j^2]}\\
           &+\sum_{\nu\in\sigma_j}\parbb{\parb{\chi_2
               r^{\tfrac{1-n_j}2}R_{j,\nu}
               r^{\tfrac{n_j-1}2}\chi_2}\otimes
             P_{j,\nu}}B_j\\
           &+\sum_{\nu\in\sigma_j}\parb{\chi_2
             r^{\tfrac{1-n_j}2}R_{j,\nu}
             r^{\tfrac{n_j-1}2}[\chi_2,p_j^2]}\otimes P_{j,\nu}}.
       \end{split}
     \end{align}

   \item \label{itemG22} $K^+f=-G_+E^{+}_{\vH}(\lambda_0)f-f$ for every
      $f\in \cup_{s<3/2}\,\vH^{1}_{-\min \set{s,s_0}}$ with
     $E^{+}_{\vH}(\lambda_0)f\in \vH^{-1}_{(1 -\nu_0)^+}$.

   \item\label{itemG32} For any $s > 1$ the operator
     $G_{+} \in\vL(\vH^{-1}_{s},\vH^{1}_{-s})$ is injective.
   \end{enumerate}
 \end{lemma}

 Using Lemmas \ref{lem:eigentransform} and \ref{lemHS2} (as
 substitutes for Lemma \ref{lemHS}) and Lemma
 \ref{lem:eigentransform2} (as a substitute for Lemma \ref{lem:4.1a})
 we can show the following weakened version of Theorem
 \ref{thm:short-effect-potent}. We mimic the proof, yielding
 immediately the second assertion of Theorem
 \ref{thm:short-effect-potent2} \ref{item:T12} (stated below) from
 Lemma \ref{lem:eigentransform2} and then in turn the first one from
 Lemma \ref{lemHS2} \ref{itemG22}. The results \ref{item:T22} and
 \ref{item:T32} follow partly as before. Note however that Lemma
 \ref{lem:eigentransform} \ref{item:2ba} is used for the second
 implication of \ref{item:T22}. Clearly \ref{item:T42} follows from
 \ref{item:T12}, \ref{item:T32} and Lemma \ref{lem:eigentransform2}.
\begin{thm}\label{thm:short-effect-potent2} Suppose the conditions of
  Proposition \ref{prop2.3} and \eqref{eq:posLowbnd0b} for
  $j=1,2$. Suppose $s<3/2$ and $s\leq s_0$.
  \begin{enumerate}[1)]
  \item\label{item:T12} The dimension of the space
    $\vE^{\vG}_{-s}\subset H^{1}_{-s}$ is finite.

    If a vector $u\in \vE^{\vG}_{-s}$,
    then $f=T^*u\in \vE^{\vH}_{-s} \subset \vH^{1}_{-s}$.
     Conversely if $f\in \vE^{\vH}_{-s}$,  then
    $u={E}^+_+(\lambda_0)f\in \vE^{\vG}_{-s}$ and $v^+f=v^-f$.

  \item\label{item:T22} If $f\in
    \vE^{\vH}_{-s}$, then  $f\in
    \vH^{1}_{(\nu_0-1)^-}\mand f\in \ker (1+K^+)$; here
 $K^+\in\vC(\vH^{1}_{-t})$   is given by \eqref{eq:Kplus2} and
    $t\in  (s_1,
        3/2)$  is arbitrary.

Conversely if $f\in
    \vH^{1}_{-1}\mand f\in \ker (1+K^+)$, then $f\in
    \vE^{\vH}_{-1}$ (in particular $E^-_+(\lambda_0)f=E^+_+(\lambda_0)f)$.
  \item\label{item:T32} Suppose $f \in \vE^{\vH}_{-s}$. In the case
    $s_0=1$  suppose in addition
    that  $f \in \vH^{1}_{-t}$ for some $t<1$. Then the components of $f=(f_1,f_2) $
    can be decomposed as
    \begin{align} \label{eq4.2abb2}
      \begin{split}
        -f_j(r_j\theta_j) &=\sum_{\nu\in\sigma_{j,1}} \sum_{k =
          1}^{n_{j, \nu}} l_{j,\nu,k}(f) { r_j^{\frac{2-n_j}{2} -\nu}
        }\chi_2(r_j){\zeta_{j,
            \nu}^{(k)}(\theta_j)} + g_j,\\&\text{ where  } g_j\in L^2\text{ and  }\\
        l_{j, \nu,k}(f)&= \inp[\big] {
          \phi_\nu(|y_j|)|y_j|^{-\frac{n_j-1}{2} }\otimes\zeta_{j,
            \nu}^{(k)}, \tilde f_j(y_j)}_{L^2(\d y_j)};\\
        \tilde f_j&=\chi_2\parb{v^+f}_j+\chi_2 B_jf_j+
        [\chi_2,p_j^2]f_j.
      \end{split}    \end{align} Here
    \begin{align}
      \label{eq:imaog02}
      l_{j, \nu,k}(f)=0\text { for }\i \nu \geq0.
    \end{align}

    In particular
    \begin{align} \label{nonresonant1b2}
      \begin{split}
        f \in \vH \quad \Longleftrightarrow \quad \forall &j=1,2,\,
        \nu \in \sigma^+_{j,1}, \, k=1,\dots, n_{j, \nu}:\\& \quad
        l_{j, \nu,k}(f) =0.
      \end{split}
    \end{align}

  \item\label{item:T42} Suppose $u \in \vE^{\vG}_{-s}$.  In the  case
    $s_0=1$  suppose in addition
    that  $u \in H^{1}_{-t}$ for some $t<1$. Then
    \begin{align} \label{nonresonant2b2}
      \begin{split}
        u \in L^2(\bX) \quad \Longleftrightarrow \quad \forall
        &j=1,2,\, \nu \in \sigma^+_{j,1}, \, k=1,\dots, n_{j, \nu}:\\&
        \quad l_{j, \nu,k}(T^*u) =0.
      \end{split}
    \end{align}
    In particular
  if $s_0\neq 1$, then
    \begin{align}
      \label{eq:dimres2}
      \dim\parbb{\vE^{\vG}_{-1}
      / \ker
       (H-\lambda_0)_{|{H^{1}}}} \leq\sum^2_{ j=1}\,\,\sum_{\nu\in \sigma^+_{j,1}}
      \,n_{j, \nu}.
    \end{align}
 For  $s_0= 1$  the bound \eqref{eq:dimres2} is valid   provided  that
 $\vE^{\vG}_{-1}$ to the left is
 replaced by  $\vE^{\vG}_{-1}\cap H^{1}_{(-1)^+}$.
\end{enumerate}
\end{thm}

\begin{remark}\label{remark:except12}
 We note that Remarks \ref{remark:mixed1} and \ref{remark:except1}
 apply equally well for the case $\lambda_0>\Sigma_2$.
\end{remark}

\section{Models of physics}\label{sec:CoulRellich} In this section
we treat the Coulombic potential models of Subsections \ref{First principal example}
and \ref{Second  principal example}
(with $N\geq 3$ and $N\geq 2$, respectively) and demonstrate
consequences of the previous sections for these models for any given
two-cluster threshold $\lambda_0<0$. For convenience
we shall not distinguish between the cases $\lambda_0=\Sigma_2$ and
$\lambda_0>\Sigma_2$ as done before. Whence we will be concerned with
generalizing Section \ref{sec:slow13} only (note that the case
$\lambda_0=\Sigma_2$ can be considered as a special case of Section
\ref{sec:slow13} although
 stronger results are presented in Section
\ref{sec:slow1}). As in Sections \ref{$N$-body
  Schr\"odinger operators} and \ref{$N$-body Schr\"odinger operators
  with infinite mass nuclei} we shall not consider  models
with spin included (see however Remark \ref{remark:physical-models-spin}).

 We  consider a two-cluster threshold  $\lambda_0$ for the models of Sections \ref{$N$-body
  Schr\"odinger operators} and \ref{$N$-body Schr\"odinger operators with infinite mass
  nuclei}, grouping the set of thresholds $a$ for which $\lambda_0\in
\sigma_{\pp}(H^a)$ into $\vA_1$,  $\vA_2$  and $\vA_3$ (all depending
on $\lambda_0$)  for  which
\begin{description}
\item [$\vA_1$:] the effective
inter-cluster interaction   is to leading order attractive Coulombic,
\item [$\vA_2$:] the effective
inter-cluster interaction  is to leading order repulsive  Coulombic,
\item [$\vA_3$:] the effective
inter-cluster interaction  is $\vO(|x_a|^{-2})$.
\end{description}
  Note that this distinction
does not depend on choices of corresponding sub-Hamiltonian
bound states $\varphi^a$ (i.e. channels); it is determined by  charges only
(cf. Case  1 introduced independently  in both of   the above
sections).

Let for $a\in\widetilde{\vA}:=\vA_1\cup\vA_2\cup\vA_3$  (again we
suppress the dependence of $\lambda_0$) the operator $P^a$ be the corresponding orthogonal projection onto
   $\ker(H^a-\lambda_0)$ in
   $L^2(\bX^a)$ and let $m_a$ be the dimension of this
   space. Obviously  $\Pi^a:=P^a\otimes 1$ projects  onto
   the span of functions of the form $\varphi^a\otimes f_a$,
   $\varphi^a\in\ker(H^a-\lambda_0)$, in $L^2(\bX)$.  We identify  $\ran P^a$, say spanned by an
   orthonormal basis
   $\varphi^a_1,\dots\varphi^a_{m_a}$, with $\C^{m_a}$ (using the
   basis), and similarly
   \begin{align*}
     L^2(\bX_a,\C^{m_a}) &\simeq \oplus_{m\leq
     m_a}\,L^2(\bX_a)\ni \oplus_{m\leq m_a}
    f_{a,m}=f_a\\&\simeq S_af_a:=\sum_{ m\leq m_a}\varphi^a_m\otimes
    f_{a,m}\in \ran \Pi^a.
   \end{align*}

 Considering for the moment only  $a\in\vA _3$ we write
   $W_a:=S^*_aI_aS_a= Q_a|x_a|^{-2}+B_a$, where $Q_a$ is a $m_a\times
   m_a$ matrix-valued function depending only on
   $\theta=\hat x_a=|x_a|^{-1}x_a$ while
   $B_a=B_a(x_a)=\vO(|x_a|^{-3})$. The `right generalization' to the case $m_a>1$ of the
   distinction between Cases 2) and 3) (discussed primarily for
   $m_a=1$ in both of Sections \ref{$N$-body
  Schr\"odinger operators} and \ref{$N$-body Schr\"odinger operators
  with infinite mass nuclei}) is to let Case 2) correspond to $Q_a\neq
0$ and let Case 3) correspond to $Q_a=
0$, respectively.

   Let
   $\S_a=\S_a^{n-1}$, $n=\dim \bX_a$,  denote
 the unit sphere in $\bX_a$. We use spherical coordinates on $\bX_a$ and the
   Laplace-Beltrami operator $\Delta_\theta$ to
   express the Laplacian $p_a^2$ on $L^2(\bX_a,\C^{m_a})$.  The
   operator $-\Delta_\theta+Q_a$ on  $L^2(\S_a,\C^{m_a})$  has discrete
   spectrum.
 As in  \eqref{eq:relmunu} we  parametrize its  eigenvalues
   in terms of a  parameter $\nu$ the collection of which is denoted
   be $\sigma_a$. More precisely consider for each $\mu\in
   \sigma(-\Delta_\theta+Q_a)$ the equation $\mu=\nu^2-\tfrac{(n -2)^2}{4}$,
  where by convention $\nu\geq 0$ if $\mu\geq-\tfrac{(n
   -2)^2}{4}$, and $\i\nu>0$ if  $\mu<-\tfrac{(n
   -2)^2}{4}$. Then the  collection of such numbers $\nu$ is denoted by
 $\sigma_a$. The orthogonal projection onto  the corresponding
 eigenspace is denoted by $ P_{a,\nu}$.
 We take an orthonormal basis in $\ran P_{a,\nu}$ for each
 $\nu\in \sigma_a$, say
 $\zeta_{a, \nu}^{(1)}, \dots, \zeta_{a, \nu}^{(n_{a,\nu})}$.
Let
\begin{align*}
  \sigma_{a,1} &= \sigma_{a} \setminus (1,\infty),\quad
 \sigma^+_{a,1}=\sigma_{a} \cap (0,1],\\
\nu_a&=\min_{\nu\in \sigma_a}\Re \nu,\quad s_a=1+\nu_a,\quad d_a=\sum_{\nu\in \sigma^+_{a,1}}n_{a,\nu}.
\end{align*}

 Now for the physics models we have the following result,
 recalling the
 notation  $ L^2_{t^+}=L^2_{t^+}(\bX)=\cup_{s>t}L^2_{s}(\bX)$ for
any real $t$.
\begin{thm}\label{thm:physical-modelsRell}
\begin{enumerate}[1)]
  \item \label{item:PRel1} The space  of locally $H^1$ solutions  to
    $(H-\lambda_0)u=0$ in
    \begin{align*}
      \sum_{a\in \vA_1} \Pi^a L^2_{-3/4}+ \sum_{a\in \vA_2} \Pi^a
      L^2_{(-3/2)^+}+\sum_{a\in \vA_3} \Pi^a
      L^2_{(-\min\set{3/2,\,s_a})^+}+L^2_{-1/2},
    \end{align*} say  denoted by $\mathcal{E}$,
 has finite dimension.

If $\vA _3=\emptyset$, then  $\vE\subset H^1_{\infty}$.
  \item \label{item:PRel2} The number
    \begin{align*}
      \dim\parb{\vE /\ker
       (H-\lambda_0)_{|{H^{1}}}}\leq \sum_{a\in \vA_3} d_a.
   \end{align*}

\item \label{item:asynoCoul}  There exist  linear functionals  $\check
  l_{a,\nu,k}:\vE\to \C$ defined for   $a\in\vA _3$, $\nu\in \sigma^+_{a,1}$ and $k=1,\dots,
   n_{a,\nu}$,  such that for any $u\in\vE
  $
   \begin{align*}
     u\in L^2 \Leftrightarrow \check l_{a,\nu,k}(u)=0 \text{ for
       all such }a,\nu\text{ and }k.
   \end{align*}

Any  $u\in \vE$  fulfills the asymptotics
\begin{align}\label{eq:asymCoulres}
  u-\sum _{a,\nu,k,m}\check l_{a,\nu,k}(u)\varphi^a_m\otimes\zeta_{a,
    \nu,m}^{(k)}\otimes{ \abs{x_a}^{\frac{2-n}{2} -\nu}
  }F(\abs{x_a}>1)\in L^2.
\end{align} (Here $\zeta_{a,
    \nu,\cdot}^{(k)}$ labels  the $m_a$ coordinates  of  $\zeta_{a,
    \nu}^{(k)}$.)

\end{enumerate}
\end{thm}

 For those  $a\in\vA _3$ for which  $s_a\in (1,3/2)$   we can replace
$L^2_{(-\min\set{3/2,\,s_a})^+}$ by $L^2_{-s_a}$ in the definition of
$\vE$. If $\lambda_0=\Sigma_2$ we
can to some extent use a bigger space of generalized eigenfunctions
than the one considered in \ref{item:PRel1}, cf. Theorems
\ref{thm:pos-effect-potent} and
\ref{thm:short-effect-potent} (not to be elaborated on).

Before giving
details of proofs we  discuss the structure of
\eqref{eq:asymCoulres} in  the case the angular-depending potential
  $Q_a$ vanishes (for
a given  $a\in \vA_3$). Note
that for $m_a=1$ this corresponds  to   Case 3 of  Sections \ref{$N$-body Schr\"odinger operators}
 and \ref{$N$-body Schr\"odinger operators
  with infinite mass nuclei}. If $Q_a=0$ the set $\sigma^+_{a,1}$ is given
explicitly as
\begin{equation}\label{add110}
   \sigma^+_{a,1} = \left\{ \begin{array}{cc}
                            \{ \frac 12 \}, \quad & n =3, \\
                            \{1\}, \quad  & n=4, \\
                            \emptyset, \quad &   n\ge 5. \\
                          \end{array} \right.
                      \end{equation}
                      Moreover  for  $n=3$ and  $n= 4$  the total multiplicity of the single
                      point in $\sigma^+_{a,1}$ is $m_a$
                     (corresponding to constant
                      eigenfunctions $\zeta_{a,
    \nu}^{(k)}$, $k=1,\dots, m_a$). In particular we see  that if
  $Q_a=0$ for all $a\in\vA_3$,  then the space  $\vE$ simplifies as the
  space
  of generalized eigenfunctions in
\begin{align*}
      \sum_{a\in \vA_1} \Pi^a L^2_{-3/4}+ \sum_{a\in \vA_2\cup \vA_3} \Pi^a
      L^2_{(-3/2)^+}+L^2_{-1/2},
    \end{align*}  and we see that $\lambda_0$ can be a  `resonance' only
    for   $n=3$ or  $n= 4$ and then with multiplicity at most $\Sigma_{a\in\vA_3}\,m_a$.

The proof Theorem
\ref{thm:physical-modelsRell}  relies strongly on  Section
\ref{sec:slow13}. Clearly  the assertion  \ref{item:PRel2} is a
consequence of
\ref{item:asynoCoul}.  A
complication of the proof of the  assertions \ref{item:PRel1} and
\eqref{eq:asymCoulres} is the possible existence of exceptional cases
discussed previously, partly in remarks. In fact the definition of the
functionals $\check l_{a,\nu,k}$ depends on the such  cases. We shall
first give the proof of  Theorem \ref{thm:physical-modelsRell} in  the  `generic'
    case and then outline the proof for   exceptional cases.
  \begin{proof}[Proof of Theorem \ref{thm:physical-modelsRell} in the  `generic'
    case]
    We consider  the
    non-exceptional case. This amounts to proving the theorem under two
    separate additional conditions. We
    introduce
    $\vG=L^2(\bX)$ and $\vF_a= \Pi^a\vG$,  $a\in
    \widetilde{\vA}$.
    \begin{subequations}
The first assumption is
    \begin{align}
      \label{eq:dirCo1}
      \vF:=\sum_{a\in \widetilde{\vA}} \vF_a=\sum_{a\in \widetilde{\vA}}\oplus \vF_a.
    \end{align}
    Let $\Pi$  be the orthogonal projection in $\vG$ onto
    $\vF$ (with \eqref{eq:dirCo1} $\vF$ is closed, cf. Proposition \ref{prop2.2}). The second  assumption is
    \begin{align} \label{eq:dirCo2} \lambda_0\notin
      \sigma_{\pp}(H');\quad H'=\Pi'H\Pi', \,\Pi'=1-\Pi.
      \end{align}
    \end{subequations}

    Under \eqref{eq:dirCo1} and \eqref{eq:dirCo2} we let
    $\vH=\sum_{a\in \widetilde{\vA}} \oplus \vH_a$ where
    $\vH_a=L^2(\bX_a,\C^{m_a})=\oplus_{m \leq m_a}\,L^2(\bX_a)$, and  we let
    $S=(S_a):\vH\to \vF\subset\vG$ be given by
    \begin{align*}
      \oplus_{a\in
      \widetilde{\vA}}\,
    f_{a}\to \sum_{a\in
      \widetilde{\vA}}\,
    S_af_{a};\quad f_{a}=\sum_{m\leq m_a}\oplus
    f_{a,m}, \quad S_af_{a}=\sum_{m \leq m_a}\varphi^a_m\otimes
    f_{a,m}.
    \end{align*}
   Letting $T=(SS^*)^{-1}S$ we note that $\Pi=ST^*$.

    In agreement with \eqref{eq:8pm} and \eqref{eq:Eeffs} we
    consider two operators $E^\pm_{\vH}(\lambda_0)$, specified by the entries
\begin{align}\label{eq:EeffsCo}
  \begin{split}
 -E^\pm_{\vH}(\lambda_0)_{ab}
  =S^*_a(p_b^2+I_b)S_b+K^\pm_{ab}(\lambda_0)\equiv \delta_{ab }\check h_a
  +K^\pm_{ab}(\lambda_0); \quad  a,b\in \widetilde{\vA},
  \end{split}
\end{align} as operators from $\vH^1=\sum_{a\in \widetilde{\vA}}
\oplus H^1(\bX_a,\C^{m_a})$ to  $\vH^{-1}=\sum_{a\in \vA_1\cup\vA_2}
\oplus H^{-1}(\bX_a,\C^{m_a})$.   Here  $\check h_a=-\Delta_a+S^*_aI_aS_a$ and
$K^\pm_{ab}(\lambda_0)=-S^*_aI_aR'(\lambda_0\pm \i 0)I_bS_b$. As in \eqref{eq:Eeffs} the difference
is symmetric and  polynomially
decreasing. The latter notion is defined as in Remark
\ref{remark:formbnd}, i.e. $(b_{ab})$ is \emph{polynomially
decreasing} if  $b_{ab}\in \vC
(H_{r}^{1}(\bX_b,\C^{m_b}),H_{t}^{-1}(\bX_a,\C^{m_a}))$ for all
$r,t\in\R$.

\begin{subequations}
We note that
\begin{align}\label{eq:KCou}
  K^\pm_{ab}(\lambda_0)\in\vL(L^2_{-s}(\bX_b,\C^{m_b}),L^2_{s}(\bX_a,\C^{m_a}))\text{
  for all }s<3/2.
\end{align}
For $a\in \widetilde{\vA} $ the effective
inter-cluster interaction $W_a=S^*_aI_aS_a $ obeys
\begin{align}\label{eq:EffCou}
  \begin{split}
W_a&=C_a|x_a|^{-1} 1_{\C^{m_a}}+ \vO(|x_a|^{-2}),\\
&\text{ where }C_a\neq 0 \text{ if and only if  }a\in \vA_1\cup \vA_2.
  \end{split}
\end{align}
 Recall also
\begin{align}\label{eq:Effminus2}
  \forall a\in \vA_3:\quad W_a=
  |x_a|^{-2}Q_a(\hat x_a)+B_a(x_a),\quad B_a=\vO(|x_a|^{-3}).
\end{align}

  \end{subequations}

  Write
    $\vA_3=\set{a_1,\dots,a_{\bar l}}$ (assuming $\vA_3\neq
    \emptyset$) and let  $\vT$ denote the set of  vectors $\bar t=(t_1,\dots,t_{\bar l})\in \R^{\bar l}$
    such that $t_l>-\min\set{3/2,\,s_{a_l}}$ for all  $l\leq \bar l$. We
    introduce  for $r\geq -3/4$,   $s>-3/2$ and $\bar t\in
    \vT $ the spaces
    \begin{align*}
\vH_{a,s}^k&=H_{s}^k(\bX_a,\C^{m_a}),\quad  k\in \R,\quad a\in \widetilde{\vA},\\
      \vH_{r,s,\bar t}^k&=\oplus_{b\in \vA_1}\,  \vH_{b,r}^k \bigoplus
      \oplus_{b\in \vA_2}\, \vH_{b,s}^k \bigoplus \oplus_{l\leq \bar l}
      \vH_{a_l,t_{l}}^k,\quad k\in \R,\\
\vE_{r,s,\bar t}^\vH &=\set{f\in \vH_{r,s,\bar t}^1\, |\quad
  E^{+}_{\vH}(\lambda_0)f=E^{-}_{\vH}(\lambda_0)f=0},
    \end{align*} where  $E^{\pm}_{\vH}(\lambda_0)$
 is given by \eqref{eq:EeffsCo}. A similar notation may be used if
 $\vA_3= \emptyset$, for example $\vE_{r,s}^\vH $, however we
 prefer to keep the uniform notation $\vE_{r,s,\bar t}^\vH $  not
 distinguising between whether one (or two) of the sets $\vA_1$,
 $\vA_2$ and $\vA_3$  is (are) empty or, possibly, all three sets are non-empty.  Let
    \begin{align}\label{eq:eh}
      \vE^\vH = \cup_{s>-3/2,\,\bar t\in \vT}\quad\vE^\vH_{-3/4,
        s,\bar t}.
    \end{align}

Let $E^{\pm}_{+}(\lambda_0)\in
  \vL(\vE^\vH,\vE)$ be given by
 \begin{align}\label{eq:enoch}
   E^{\pm}_{+}(\lambda_0)f=Sf-\sum_ {a\in \widetilde{\vA}}R'(\lambda_0\pm
   \i 0)I_aS_af_a.
 \end{align}
 By mimicking the proof of Lemma \ref{lem:eigentransform2} we can show
that   indeed  $E^{+}_{+}(\lambda_0):\vE^\vH \to \vE$  and that, considered as such
a map,  it is a linear isomorphism with inverse $T^*:\vE\to \vE^\vH
$. Note that  as before we need the property  that $E^{+}_{+}(\lambda_0)f=E^{-}_{+}(\lambda_0)f$
for any $f\in\vE^\vH $.

Next we adapt  the parametrix construction of  Section
\ref{sec:slow13} (combining the  three different cases
 treated there)   and convert the equation $E^{+}_{\vH}(\lambda_0)f=0$, $f\in \vE^\vH$,
to an equation of the form $(1+K^+)f=0$. We can then show by iteration that
$f$  belongs to a suitable space (independent of $f$) on which $ K^+$ is
compact,
yielding  \ref{item:PRel1}.

To carry this out in more details, let us define
$R^+=\oplus_{a\in \widetilde \vA} \,r_a$ where $r_a$ is the following parametrix (depending on the three cases): For $a\in  \vA_1\cup \vA_2$
we take  $r_a$  diagonal in $\vH_a$, $r_a=\oplus_{m\leq m_a} \,r_{a,m}$.  For $a\in
\vA_1$ and $m\leq m_a$ we let
\begin{equation*}
 r _{a,m}=\lim_{\epsilon\to
  0_+} (h_a-\i \epsilon)^{-1} \in \vL(H^{-1}_{s}(\bX_a), H^{1}_{-s}(\bX_a)),\,s>3/4,
\end{equation*} where $h_a=p_a^2+w_a$ is constructed as in Subsection
\ref{subsec:negat-effect-potent}.  For $a\in \vA_2$ we let
$r_a=\oplus_{m\leq m_a} \,r_{a,m}$, where $r_{a,m}=h_a^{-1}$ with
$h_a=p_a^2+w_a$ been  constructed as in Subsection
\ref{subsec:positive-effect-potent}. Note that in this case
$\inp{\cdot}^{-1}h_a^{-1}$ is a bounded pseudodifferential operator,
cf. Lemma \ref{lemma:posit-slowly-decay}.  Finally we let  for $a\in \vA_3$
the operator $r_a=G_{a+}$, where $G_{a+}$ is constructed as in Subsection
\ref{subsec:Homogeneous-effect-potent}, i.e.
\begin{align}\label{eq:gaplus}
  G_{a+}=\chi_1(\check h_a-\i)^{-1}\chi_1+\sum_{\nu\in\sigma_a}\parb{\chi_2\,\abs{\cdot}^{\tfrac{1-n}2}R_{a,\nu}
    \abs{\cdot}^{\tfrac{n-1}2}\,\chi_2}\otimes P_{a,\nu}.
   \end{align} Here $\check h_a=p_a^2+W_a$, where $W_a=S^*_aI_aS_a=
   |\cdot|^{-2}Q_a+\vO(|\cdot|^{-3})$ acting as a (matrix-valued)
   multiplication operator on $\vH_a$. The quantities $Q_a$,
   $\sigma_a$ and $ P_{a,\nu}$ are introduced before Theorem
   \ref{thm:physical-modelsRell}, and the operator $R_{a,\nu} $ is the
   one-dimensional Green's function studied in  Subsection
\ref{subsec:Homogeneous-effect-potent}. While the parametrix $r_a$ for
$a\in  \vA_1\cup \vA_2$ is in some sense taken as an exact inverse,  the
operator $r_a=G_{a+}$ chosen for  $a\in \vA_3$ is less accurate. This is
explicitly seen in \eqref{eq:Kplus2} in which the second up to the
fourth terms represent  `errors' not appearing for $a\in  \vA_1\cup
\vA_2$.

Next we write the equation $-R^+E^{+}_{\vH}(\lambda_0)f=0$ for $f\in
\vE^\vH$ as $(1+K^+)f=0$, where integration by parts is used to
produce the term $1f$. This is done separately for the components of
$f$ by the arguments of Subsections
\ref{subsec:negat-effect-potents}--\ref{subsec:hd2ep}.  By
\eqref{eq:EeffsCo} there are terms in $K^+$ involving polynomially
decreasing factors and there are also terms of the form
$r_aK^\pm_{ab}(\lambda_0)$. Effectively $K^\pm_{ab}(\lambda_0)$ is of
order $-3$ when applied to an $f$ given as above. For $a\in \vA_3$
there are additional terms as discussed above,
cf. \eqref{eq:Kplus2}. In any case it follows that the weighted orders
of the components of $K^+f$ tend to be at least $1/2$ power better (as
determined by \eqref{eq:KCou}--\eqref{eq:Effminus2}) than the apriory
orders of the components of $f (=-K^+f)$ (at this point, note also
\eqref{eq:leftinverses} and its application \eqref{eq:micr1}). However
there is an exception for this general rule due to the restriction $s>
1- \nu_a$ imposed by mapping properties of $G_{a+}$,
cf. \eqref{HS02}. Note for example that to prove $\sum_b
r_aK^\pm_{ab}(\lambda_0)f_b\in \vH_{a,-s}^1$ for $a\in \vA_3$ we need
at least $s> 1- \nu_a$. This is the only restriction for an iteration
argument. Actually we infer from one iteration (i.e. just from the
properties of $K^+f$) that for $a\in \vA_1$ the component $f_a\in
\vH_{a,-1/4}^1$, that for $a\in \vA_2$ the component $f_a\in
\vH_{a,-1/2}^1$, while $f_a\in \vH_{a,\epsilon-1}^1$ for some positive
$\epsilon$ for $a\in \vA_3$ (noting though that the latter property is
an assumption if $\nu_a=0$).  With one more iteration we then obtain
that $f_a\in \vH_{a,0}^1$ for $a\in \vA_1\cup \vA_2$.  If
$\vA_3=\emptyset$ continued iteration leads to $f\in \vH^1_\infty$
(with Lemma \ref{lem:eigentransform2} this proves the very last part
of \ref{item:PRel1}). In particular we can
conclude that
\begin{align*}
  f\in \vE^\vH\Rightarrow f\in \vH_{*}^1:=\vH_{-1,-1,\bar t}^1,
\end{align*}  where all coordinates of
$\bar t $  are  taken  to be $-4/3$. Since  $K^+\in\vC
(\vH_{*}^1)$ we have shown  \ref{item:PRel1}.

It remains to prove \ref{item:asynoCoul}. We only need
to construct linear functionals $\check l_{a,\nu,k}:\vE\to \C$ such
that \eqref{eq:asymCoulres} is fulfilled for any $u\in\vE$. For given
$u\in\vE$ we let $f=T^*u$ and need to examine the components $f_a$ of
$f$ that might not be in $\vH_a$. As we have seen this amounts to
looking at $a\in\vA_3$ only in which case we know the analogue of
\eqref{eq:uppsbnd1} that $f_a\in \vH_{a,-s}^1$ for some $s<1$. Looking
at $-f_a=(K^+f)_a$ for such $a$ we need   the first
equation of \eqref{eq:EeffsCo} written as
\begin{align}\label{eq:EeffsCo23}
  \begin{split}
 -E^\pm_{\vH}(\lambda_0)_{ab}
  =S^*_a(p_b^2+I_b)S_b+K^\pm_{ab}(\lambda_0)=\delta_{ab}\check
  h_a+v_{ab}^+;\quad \check  h_a=p_a^2+W_a.
  \end{split}
\end{align} Now, by using \eqref{eq:gaplus} and the analogue of
\eqref{eq:Kplus2}, we obtain  by mimicking  Step \textit{III} of the proof of
Theorem \ref{thm:short-effect-potent} that
\begin{align} \label{eq4.2abb22}
      \begin{split}
        -f_a(\abs{x_a}\theta_a) &=\sum_{\nu\in\sigma_{a,1}} \sum_{k =
          1}^{n_{a, \nu}} l_{a,\nu,k}(f) { \abs{x_a}^{\frac{2-n}{2} -\nu}
        }\chi_2(\abs{x_a}){\zeta_{a,
            \nu}^{(k)}(\theta_a)} + g_a,\\&\text{ where  } g_a\in \vH_a\text{ and  }\\
        l_{a, \nu,k}(f)&= \inp[\big] {
          \phi_\nu(|y_a|)|y_a|^{-\frac{n-1}{2} }\otimes\zeta_{a,
            \nu}^{(k)}, \tilde f_a(y_a)}_{\vH_a};\\
        \tilde f_a&=\chi_2\sum _{b\in \widetilde \vA} v^+_{ab}f_b+\chi_2 B_af_a+
        [\chi_2,p_a^2]f_a.
      \end{split}
    \end{align} Note that $\tilde f_a\in \vH_{a,3-s}^{0}$ for some
    $s<1$, cf. \eqref{eq:prbnd}, yielding in particular that $l_{j,
      \nu,k}(f)$ is well-defined. Next we deduce from
    \eqref{eq4.2abb22} that
    \begin{align}
      \label{eq:imaog0222}
      l_{j, \nu,k}(f)=0\text { for }\i \nu \geq0.
    \end{align}

    Clearly \eqref{eq4.2abb22} and \eqref{eq:imaog0222}, in
    combination with \eqref{eq:enoch} and the  adapted version of
    Lemma \ref{lem:eigentransform2} (stated after \eqref{eq:enoch}), lead to
    \eqref{eq:asymCoulres} with $\check l_{a,\nu,k}$ given by
    \begin{align*}
      \check
    l_{a,\nu,k}(u)= -l_{a, \nu,k}(T^*u);\quad u\in\vE,\, a\in\vA _3, \,\nu\in \sigma^+_{a,1},\,k=1,\dots,
   n_{a,\nu}.
    \end{align*}
\end{proof}
\begin{proof}[Proof of Theorem \ref{thm:physical-modelsRell} in the  `general'
    case, a sketch]
We  outline  a proof of the theorem without the  two
    additional conditions \eqref{eq:dirCo1} and \eqref{eq:dirCo2} imposed. We shall implement Remarks
   \ref{remark:except1} and \ref{remark:except12} using  the previous
    notation as much as possible.

    First, if \eqref{eq:dirCo1} is not fulfilled, we infer again that
    the space
    $\sum_{a\in \widetilde{\vA}} \vF_a$ is closed by considering as
    in Subsection \ref{sec:The case when the condition {ass2}} the
    'restriction' of $S$ to
\begin{align}
  \label{eq:Hmod2}
  \vH_0:=\set[\Big]{f\in \sum_{a\in \widetilde{\vA}} \oplus \vH_a\mid  f\perp
    \ker \parb{S_{ab}}_{a,b\in\widetilde A}},\,S_{ab}=S_{a}^*S_{b},
\end{align} where as before $\vH_a=L^2(\bX_a,\C^{m_a})$. Then
$S=(S_a):\vH_0\to \vF_0:=\sum _{a\in\widetilde\vA} \vF_a\subset \vG=L^2(\bX)$ is a continuous
isomorphism. In particular indeed $\vF_0$ is closed, and we let $\Pi$
denote the orthogonal projection onto $\vF_0$ in $\vG$. The quantities
$\Pi'$, $H'$ and $R'(z)$ are introduced as in previous sections (with
$\vF_0$ playing the role of $\vF$).

Next we would like to consider limits $R'(\lambda_0\pm \i 0)$ which is
doable  under the condition \eqref{eq:dirCo2}. However  we shall
allow, and here consider, the exceptional case where \eqref{eq:dirCo2}
is not fulfilled, and we need a version of Section \ref{sec:The case where
  lambda0insigma} to treat this case. By a version of Theorem
\ref{thm:priori-decay-b_0}, cf.  Remarks
\ref{remark:microlocal-boundsGEN}, \ref{remark:The case
  lambda0insigma} and \ref{remark:excase}, $\lambda_0$ has
finite multiplicity as an eigenvalue of $H'$ with eigenfunctions in
$H^2_\infty$.

To simplify the presentation let us make the same
assumption as in Subsection \ref{sec:The case where lambda0insigma}
that this  eigenspace is given as ${\spann}\set{\psi}$,
$\norm{\psi}=1$. We introduce then $\vH=\vH_0\oplus \C$ and $S:\vH\to
\vF:=\vF_0\oplus {\spann}\set{\psi}$ by
\begin{align*}
  S\tilde f=\sum_a\,S_af_a+c\psi,\quad \tilde f=\parb{(f_a),c}.
\end{align*}

We let  $\Pi''$  denote the orthogonal projection onto the orthogonal
complement
of $\vF$ in
$\vG$, and let $H'':=\Pi''H\Pi''$ and $R''(z)=(H''-z)^{-1}$. Note that
$\Pi''=1-\Pi-|\psi\rangle\langle \psi|=\Pi'-|\psi\rangle\langle
\psi|$. Recall  then  \eqref{eq:8pm}
\begin{align}\label{eq:ezfirst}
  E_{\vH}(z) = S^* \big (z - H + H R''(z)H \big ) S.
\end{align} Letting  $M=\# \widetilde \vA +1$ this operator has a $M\times M$-block  representation
$(e_{ij})_{i,j\leq M}$. Here   (listing the elements $a\in\widetilde \vA$
as   $a_1,\dots,a_{M-1}$) $e_{MM}=z-\lambda_0$, $e_{iM}=-
S^*_{a_i}H\psi=-
S^*_{a_i}I_{a_i}\psi$ and $e_{Mi}=e_{iM}^*=\langle e_{Mi}|$ for
$i\leq M-1$, while for $i,j\leq M-1$ we set $e_{ij}=S_{a_i}^* \big (z
- H + H R''(z)H \big ) S_{a_j}$.
\begin{subequations} The analogue  of \eqref{eq:27p}   reads
\begin{align}\label{eq:27p2}
  \begin{split}
 &\breve H=H''+ \sum_{a\in\widetilde \vA }\,p_a^2\Pi^a,\\
\;&\widetilde
H=H-\lambda_0\sum_{a\in\widetilde \vA }\Pi^a;\\\,& \vD (\breve H)=\vD (\widetilde  H)=\vD ( H).
  \end{split}
\end{align} These  operators  differ by  $H$-compact  terms, cf. \eqref{eq:28p},
\begin{align}
  \label{eq:28p2}\breve H&=\widetilde
H -K_1+K_2;\\
K_1&=\Pi H\Pi'+\Pi' H\Pi+\sum_{a\in\widetilde \vA }\Pi^aI_{a}\Pi^a+\lambda_0|\psi\rangle\langle \psi|,\nonumber \\
K_2&=\sum_{a\in\widetilde \vA }\Pi^aH\Pi^a-\Pi H\Pi,\nonumber
\end{align}
  and the basic structure of  resolvents is  given as
\begin{align}\label{eq:basresbreveo2}
  \begin{split}
  R''(z)&= \breve R(z)-\parbb{\sum_{a\in\widetilde \vA }\,\,p_a^2\Pi^a-z}^{-1}\Pi-z^{-1}|\psi\rangle\langle \psi|,\\
 R''(z)&=\Pi''\breve R(z)\Pi''=\Pi''\breve R(z)=\breve R(z)\Pi'';
\quad\breve R(z)=(\breve H-z)^{-1}.
  \end{split}
\end{align}
\end{subequations}

Due to the good properties of $\psi$ we can make sense to
$\breve R(\lambda_0\pm\i 0)$ and then in turn to $R''(\lambda_0\pm\i
0)$. In particular we can consider \eqref{eq:enoch} with
$f$ replaced by $\tilde f:=\parb{(f_a),c}\in \vH$ and $R'(\lambda_0\pm
\i 0)$  by $R''((\lambda_0\pm \i 0))$, i.e.
\begin{subequations}
\begin{align}\label{eq:enoch3}
   E^{\pm}_{+}(\lambda_0)\tilde f=S\tilde f-\sum_ {a\in \widetilde{\vA}}R''(\lambda_0\pm
   \i 0)I_aS_af_a.
 \end{align}  Similarly, by  taking limits in \eqref{eq:ezfirst},  we obtain
   \begin{align}\label{eq:enoch32}
  E^{\pm}_{\vH}(\lambda_0) =  S^* \big (\lambda_0 -  H +  H
  R''(\lambda_0\pm \i 0)H \big ) S=(e_{ij})_{i,j\leq M,\,z=\lambda_0\pm \i 0}.
\end{align}
\end{subequations}
 The previous definition of $\vE^\vH$ given in
\eqref{eq:eh} needs to be modified as
 follows. We introduce (for parameters given as before)
 \begin{align*}
   \vH_{r,s,\bar t}^{k,0}&=\set[\Big]{f\in \vH_{r,s,\bar t}^k\,\mid f\perp
    \ker \parb{S_{ab}}_{a,b\in\widetilde A}},\\
\vE_{r,s,\bar t}^\vH &=\set{\tilde f=(f,c)\, \mid
  f\in \vH_{r,s,\bar t}^{1,0},\,c\in \C,\,E^{+}_{\vH}(\lambda_0)\tilde
  f=E^{-}_{\vH}(\lambda_0)\tilde f=0},\\
\vE^\vH &= \cup_{s>-3/2,\,\bar t\in \vT}\quad\vE^\vH_{-3/4,
        s,\bar t}.
\end{align*} As
    before we can now  check that     $E^{+}_{+}(\lambda_0):\vE^\vH
    \to \vE$  and that, considered as such
a map,  it is a linear isomorphism with inverse $T^*=S^*(SS^*)^{-1}:\vE\to \vE^\vH$.

Next we consider the space
\begin{align*}
  \widetilde \vH=\parbb{\sum_{a\in \widetilde{\vA}} \oplus \vH_a}\oplus
\C.
\end{align*} Clearly $\widetilde \vH=\vH \oplus
\ker \parb{S_{ab}}_{a,b\in\widetilde A}$, and $E^{\pm}_{\vH}(\lambda_0)$
acts only on the first component. Let us introduce  $\widetilde
E^{\pm}_{\vH}(\lambda_0)= E^{\pm}_{\vH}(\lambda_0)+P_0$, where  $P_{0}$ denotes the orthogonal projection onto
$\ker \parb{S_{ab}}_{a,b\in\widetilde A}$ in $\sum_{a\in \widetilde{\vA}}
\oplus \vH_a$. We can consider $P_0$ as acting on the second
component of $\widetilde \vH$ and hence write $\widetilde
E^{\pm}_{\vH}(\lambda_0)=E^{\pm}_{\vH}(\lambda_0)\oplus 1$. Clearly
$\ker \widetilde E^{\pm}_{\vH}(\lambda_0)=\ker E^{\pm}_{\vH}(\lambda_0)$, and
similar relations hold in weighted Sobolov spaces. We are lead to
consider $\widetilde E^{\pm}_{\vH}(\lambda_0)=E^{\pm}_{\vH}(\lambda_0)\oplus 1$
on spaces of the form
\begin{align*}
  \widetilde \vH^k_{r,s,\bar t}= \parbb{\vH_{r,s,\bar t}^{k,0}\oplus
    \C}\oplus \ker \parb{S_{ab}}_{a,b\in\widetilde A}=\vH_{r,s,\bar t}^{k}\oplus
    \C.
\end{align*} With the same restrictions on the parameters we can prove
a
version of Lemma \ref{lem:eigentransform2} (using the same proof), and
by applying
the parametrix  construction $\oplus_{a\in \widetilde \vA} \,r_a$ of the previous proof of the theorem (the
regular case)  extended as $R_{\diag}^+=\parb{\oplus_{a\in \widetilde \vA}
  \,r_a}\oplus 1_\C$ to the equation
\begin{align*}
 \widetilde E^{\pm}_{\vH}(\lambda_0)\tilde
f=0, \quad  \tilde f:=\parb{f,c}=\parbb{\parb{P'_0f,c},P_0f}\in \widetilde \vH^k_{r,s,\bar t},
\,\,P_0'=1-P_0,
\end{align*}
   we arrive again at Fredholm theory. More precisely, using that the
   operator  $E^{\pm}_{\vH}(\lambda_0)\oplus 0$ is realized concretely
   by the expression on  the right-hand side of \eqref{eq:enoch32}, we arrive at the
   equation
   \begin{align*}
     0=R_{\diag}^+\parb{P_0f+(\widetilde E^{\pm}_{\vH}(\lambda_0)\tilde
f-P_0f)}=-(1+\widetilde K^+)\tilde f
   \end{align*}
 for a `nice' operator $\widetilde K^+$. Indeed
   using the specific form of this
 slightly modified new $K^+$ we can  proceed as
before. We omit the details.
\end{proof}

\begin{remark}\label{remark:restric0}  If $\vA_3=\emptyset$ we can
  write the above operator as $\widetilde K^+=R_{\diag}^+\tilde v^+$,
  where $\Im \tilde v^+\leq
  0$, however
 in general the form of $\widetilde K^+$ is more
complicated. Moreover   the structure of the set of solutions to the equations
  $(1+K^+) f=0$ or $(1+\widetilde K^+) \tilde f$=0  studied in the above proofs
   appears `cleanest' in this case in the sense that  `spurious'  solutions,  i.e. resonances
   states,
   do not  occur.  On the other hand if $\vA_3\neq \emptyset$ resonances
   states could occur, and for this  reason
  our theory of resolvent expansion at the two-cluster threshold is
  more complicated for  $\vA_3\neq \emptyset$. See Chapter
  \ref{chap:resolv-asympt-near} for a systematic study applicable to   the fastly decaying
  case  for which $\vA= \vA_3$ and  the effective
inter-cluster interaction  is $\vO(|x_a|^{-3})$ for $a\in \vA$.
Resolvent expansion at the two-cluster threshold for some non-fastly decaying
  cases  is treated
  in Chapter \ref{Applications}, see  Remark  \ref{remark:formRes}.
\end{remark}
  \begin{remark}\label{remark:physical-models-spin}
    We did not consider particles with spin in Theorem
    \ref{thm:physical-modelsRell}. However the theorem applies to
    cases with spin because this basically amount to restricting the
    operator $H$ to suitable subspaces. However this is under the
    assumption that $\lambda_0$ is an (un-restricted) two-cluster
    threshold. It could be that $\lambda_0$ is a threshold not of this
    type but in the restricted sense is a  two-cluster
    threshold. To treat such situation one  would need a different
    procedure, although  it could be that  a somewhat similar
    treatment would work. Thus  for example, to indicate a possible
    preliminary step,    an obvious procedure for modifying $S$
    for fermions would be to use an anti-symmetrization of
    $S_af_a=\varphi^a\otimes f_a$ to bring this vector into  the
    fermionic subspace of $L^2(\bX)$.
  \end{remark}



\chapter{Resolvent asymptotics near a two-cluster
  threshold}\label{chap:resolv-asympt-near}

This chapter is  devoted to the  asymptotics of the resolvent near an
arbitrary  two-cluster threshold $\lambda_0$ for the $N$-body
Schr\"odinger operator $H$. As discussed in Chapter \ref{chap:lowest thr}, for physics models, the effective potential may decay like
$\vO(\abs{x_a}^{-1})$ (slowly decaying case),  $\vO(\abs{x_a}^{-2})$ (critically decaying case) or  $\vO(\abs{x_a}^{-3})$ (fastly decaying case).
We only study two cases:   1)   the effective potential is fastly decaying;  2) the effective potential is slowly decaying and positive outside a compact set.
 The main difference between these two cases is that threshold resonance may appear in the first one and  it is absent in the second one.  For fastly decay effective potentials, we only study in full details two situations:
 i)  $\lambda_0 =\Sigma_2$ is a double
two-cluster threshold (the case $\lambda_0=\Sigma_2$ is a non-multiple two-cluster
threshold is easier and already studied in \cite{Wa2}); ii) $\lambda_0> \Sigma_2$
is a non-multiple two-cluster threshold of $H$.
We  calculate the leading term of resolvent expansions according to cases the threshold is an eigenvalue or/and a resonance. For positively slowly decreasing effective potentials we only study the lowest threshold and prove Gevrey estimates in exponentially weighted spaces for the remainder in the resolvent expansion (cf. (\cite{AW, Wa6}) for one-body operators). Our study will have applications for  the physics models  of  Sections \ref{$N$-body
  Schr\"odinger operators} and \ref{$N$-body Schr\"odinger operators with infinite mass
  nuclei}.  Parts of our study easily
generalize to cover interesting cases for the physics models
 of Sections \ref{$N$-body Schr\"odinger operators}
and \ref{$N$-body Schr\"odinger operators with infinite mass
  nuclei}; our results will be stated with sketched proofs.

This chapter is organized as follows. In Section \ref{sect5.1} we study   the asymptotic expansion of the
resolvent $R(z)=(H-z)^{-1}$ for rapidly decreasing effective potentials. For simplicity we assume that the intercluster spaces are of dimension three and analyze the notion of two-cluster threshold resonances in Subsection \ref{sec:two-clust-threshres}. We give some normalization conditions for resonance
states according to spectral nature of the two-cluster threshold. These conditions are useful, they allow to explicitly compute  some constants in physically interesting models (cf. Chapter 6).
 The resolvent expansions are given in Subsection \ref{sec:resolv-asympt-nearLOWEST} for $\lambda_0=\Sigma_2$ (when $\Sigma_2$ is a double two-cluster threshold) and  in Subsection \ref{sec:resolv-asympt-nearHIGHER} for $\lambda_0 > \Sigma_2$.  In Section \ref{sect5.2}
 we study the case where the effective potential is positive  outside
 a compact set and slowly decaying at infinity. In this case  threshold resonance is absent and we prove one-term resolvent expansions with Gevrey estimates on the remainder at the lowest threshold $\Sigma_2$. These results may be used to show sub-exponential time-decay of local energies.  In
Section \ref{sec:resolv-asympt-physics models near} we study the
combination of the previous sections for the  physics models
 of Sections \ref{$N$-body Schr\"odinger operators}
and \ref{$N$-body Schr\"odinger operators with infinite mass nuclei}.

\section{Rapidly decaying effective potentials} \label{sect5.1}

\subsection{Two-cluster threshold resonances}\label{sec:two-clust-threshres}

Let $\lambda_0 \in \vT$ be a two-cluster threshold of $H$. We analyze
in this subsection spectral properties of $H$ at $\lambda_0$, assuming
the effective potential obtained through the Grushin method decays
sufficiently fastly. In addition, we assume the intercluster configuration $\bX_a$ is of dimension three for  two-cluster decompositions $a$ such that $\lambda_0$ is an eigenvalue of $H^a$. \\

\subsubsection{The case $\lambda_0 =
  \Sigma_2$}\label{sec:case-lambda_0=sub}

Consider first the case
\be \label{ass5.1bff} \mbox{
  $\lambda_0 =\Sigma_2$ is a double  two-cluster threshold in the sense
  of Condition \ref{cond:uniqdd}.}
  \ee
Whence the lowest threshold $\Sigma_2$ is a  two-cluster threshold of $H$ and an   eigenvalue of exactly  two sub-Hamiltonians $
H^{a_j}$, $j=1,2$. Here of course  $a_1$ and $ a_2$ are two-cluster
decompositions of the $N$-body system.
The above condition implies  that $\Sigma_2<0$, however the
requirement of Condition \ref{cond:uniqdd}
  that $\lambda_0(=\Sigma_2)$ be  a simple eigenvalue
for both $H^{a_1}$ and $H^{a_2}$ is automatically fulfilled. It is
also a part of Condition \ref{cond:uniqdd} that that Condition
\ref{cond:geom_singl} are fulfilled for $a=a_1$ and $a=a_2$.  In agreement with Section
\ref{sec:Reduction near a multiple two-cluster threshold},  let
$\varphi_j$ denote a  corresponding (real) normalized eigenvector of
$H^{a_j}$, $j=1,2$.
 Denote
\[
\bX^j = \bX^{a_j}, \quad \bX_j = \bX_{a_j}, \quad j =1,2,
\]
and $x^j$ (resp., $x_j$) variables in $\bX^j$ (resp., in $\bX_j$).
Let
\[
\vF_j =\{ g = \varphi_{j}(x^j)f_j(x_j)\mid   f_j \in L^2(\bX_j)\},  \quad j =1, 2.
\]
$\vF_j$ is a closed subspace of $L^2(\bX)$.
Assume, cf. \eqref{ass2},
\begin{align*}
  \vF_1 \cap\vF_2  =\{0\} \quad \mbox{ and } n_j:=\dim \bX_j =3\,\text{ for } \, j =1,2.
\end{align*}

Recall from  Proposition \ref{prop2.2} that $\vF = \vF_1 + \vF_2$ is a closed subspace in  $L^2(\bX)$.
Let  $\Pi $ be the orthogonal projection from  $L^2(\bX)$ onto  $ \vF$
and $\Pi'= 1-\Pi$.  We borrow the notation \eqref{eq:spaceDEF} and
denote, in particular,
\[
\vH^{k}_s = H^{k}_s(\bX_1) \oplus H^{k}_s(\bX_2), \quad \vH^k =\vH^{k}_0, \quad \vH_s =\vH^{0}_s, \quad \vH =\vH^{0}_0.
\] Let ${\mathcal L}(k,s;k',s')={\mathcal L}(\vH^{k}_{s},\vH^{k'}_{s'})$,
cf. Subsection \ref{Spaces}.

Let
$S= (S_1, S_2) : \vH  \to L^2(\bX)$ be  defined by $Sf=S_1f_1+S_2f_2$
for $f=(f_1,f_2)$ and with
\begin{equation*}
  S_j:  L^2(\bX_j) \to L^2(\bX),\quad  f_j \to S_jf_j =
\varphi_j(x^j) \otimes f_j(x_j).
\end{equation*}
Then $
S_j^*:  L^2(\bX)\to  L^2(\bX_j)$ is given by $S_j^*:
\, f\to  S_j^*f = \inp{\varphi_j,f}_j$, for $j =1,2.$ Here $\inp{.,.}_j$ denotes the scalar product in
$L^2(\bX^j)$; the notation $\w{\cdot, \cdot}$ will be used to denote
the scalar product in $L^2(\bX)$ (or in $L^2(\bX_j)$).
One has
\[
S^*S =  1 + \left( \begin{array}{cc}
                  0 & s_{12}\\[.1in]
                   s_{21} & 0
                \end{array}
                \right) \quad \mbox{ on } \vH =L^2(\bX_1) \oplus L^2(\bX_2),
\]
where $s_{ij} \in \vL\parb{ L^2(\bX_j),L^2(\bX_i)}, \;i\neq j,$ are
given by $s_{ij}f_j= \inp{\varphi_i, \varphi_j \otimes f_j}_i$.

Introduce for  $z\in \C$ with $\Im z\neq0$
 \begin{align*}
E(z) &= R'(z),   \\
 E_+(z)&= S -R'(z)H S,   \\
 E_-(z) &= S^*- S^*H R'(z),\\
 E_{\vH}(z) &=  S^* \big (z -  H +  H R'(z)H \big ) S.
\end{align*}
Then (recalling from  Section
\ref{sec:Reduction near a multiple two-cluster threshold})
\begin{equation} \label{rep5.9}
R(z)= E(z)-  E_{+}(z)  E_{\vH}(z)^{-1}  E_{-}(z).
\end{equation}
$ E_{\vH}(z) : \vH \to   \vH$ is computed in (\ref{eq:13}):
\begin{align}\label{eq:5.13}
E_{\vH}(z)  ={}  &z- \lambda_0 \nonumber\\&-  \left( \begin{array}{cc}
                   - \Delta_{x_1} + W_1(x_1) + K_{11}(z) &  \check  s_{12}(z)+W_{12}+K_{12}(z) \\
                    \check  s_{21}(z)+ W_{21}+K_{21}(z) &   - \Delta_{x_2} + W_2(x_2) + K_{22}(z)
                \end{array}\right),
\end{align}
where
\begin{subequations}
\begin{align*} \check s_{ij}(z)&=  \inp{\varphi_i, \varphi_j \otimes (\lambda_0-z-\Delta_{x_j})\cdot}_i,\\
W_{ij}&= \inp{\varphi_i,I_j\varphi_j\otimes \cdot}_i,  \\
W_k (\cdot) &=W_{kk}=  \inp{\varphi_k,I_k\varphi_k}_k(\cdot),\\
K_{ij}(z)&=  -\inp{\varphi_i,I_i R'(z) I_j (\varphi_j \otimes \cdot)}_i.
\end{align*}
\end{subequations}

Assume now $\lambda_0 \not\in \sigma_\d(H')$. Then
by Lemma \ref{lemma2.1} $R'(z)$ is holomorphic in $z$ for $z$ near
$\lambda_0$ and $   E_{\vH}(\lambda_0)$ is well-defined. To simplify
notation, denote $P=  -  E_{\vH}(\lambda_0)$ and decompose this
operator  as $P = P_0 + U$, where
\begin{subequations}
\begin{alignat}{2} \label{P0} P_0 &= \left(%
     \begin{array}{cc}
       P_{1,0}& 0\\
       0 & P_{2,0} \\
\end{array}\right)
       : = \left(\begin{array}{cc}-\Delta_{x_1}  & 0\\
       0 & -\Delta_{x_2 } \\
     \end{array}
   \right), \\ \label{U} U &= \left(%
     \begin{array}{cc}
       U_{11}   &  U_{12}\\
       U_{21}  & U_{22}  \\
     \end{array}%
   \right):= \left(%
     \begin{array}{cc}
       W_1  & W_{12} + \check s_{12}(\lambda_0)\\
       W_{21} +\check s_{21}(\lambda_0) & W_2  \\
     \end{array}%
   \right) + K(\lambda_0).
 \end{alignat}
\end{subequations}

We impose the condition $\rho> \tfrac 12$ (where  $\rho$ is given in  (\ref{eq:1})) and
\eqref{eq:posLowbnd0b}  with $q_j(\theta)=0$ for  $j=1,2$, and hence
 (cf. Proposition \ref{prop2.3})  that at least for  $\rho_0= 3$
 \begin{align}\label{eq:Uassp}
   U : \vH^{1}_{ -s} \to \vH^{-1}_{ \rho_0-s} \quad \mbox{ is
     continuous for any $ s\in [0,\rho_0]$}.
 \end{align}

This leads us finally to impose the following  set of conditions
(including the previous ones):
 \begin{subequations}
 \begin{align} \label{ass5.1}
&\dim \bX_1 =\dim \bX_2 =3,\quad\vF_1 \cap\vF_2  =\{0\}\,\mand
\,\lambda_0 \not\in \sigma (H'),\\
&\rho> \tfrac 12\,\mand \,\diag(W_1,W_2) \in \vL(\vH^1,\vH^{-1}_{\rho_0} )
  \mbox{ for some   } \,\rho_0\geq 3,\label{ass5.3}\\
&2 +2 \rho  \ge \rho_0.\label{ass5.4}
 \end{align}
  \end{subequations}
 Actually \eqref{ass5.4} is imposed  for Section
 \ref{sec:resolv-asympt-nearLOWEST} only; it is a  convenient
 assumption relevant for the case $\rho_0>3$. Clearly \eqref{eq:Uassp}
 is fulfilled (in fact the operator $U(z)$
 of Section
 \ref{sec:resolv-asympt-nearLOWEST} fulfills
 $U(z) \in \vL(1, -s; -1,\rho_0-s)$ for $0 \le s \le \rho_0$ uniformly
 for $z$ near zero, generalizing \eqref{eq:Uassp}).  The conditions \eqref{ass5.1}  and
                                     \eqref{ass5.3} will be imposed
                                     throughout the section.

Recalling  the  notation $H^1_{t^-}
  =\cap_{s<t}H^1_{s}$ and  $\vH^1_{t^-}
  =\cap_{s<t}\vH^1_{s}$  for any $t\in \R$, \eqref{ass5.1} and
  \eqref{ass5.3} lead to the
  following effective version of  Definition
  \ref{defn:hom}, see Theorems \ref{thm:short-effect-potent} and
  \ref{thm:short-effect-potent2} (and \eqref{add110}).
\begin{defn}
\begin{enumerate}[(1)]
\item \label{item:10}$\lambda_0$ is  a \emph{resonance of $H$} if the
  equation $Hu = \lambda_0u$ admits a solution
  $u \in H^1_{(-1/2)^-} \setminus H^1$. The multiplicity
  of the resonance $\lambda_0$ is defined as the dimension of the quotient
  space
  $\ker (H-\lambda_0)_{|H^1_{(-1/2)^-}} / \ker
  (H-\lambda_0)_{|H^{1}}$.
\item\label{item:9} $0$ is  a \emph{resonance of $P$}  if the equation
  $ Pv = 0$ has a solution $v \in \vH^1_{(-1/2)^-} \setminus \vH^1$. The multiplicity of the resonance zero is defined as
  the dimension of the quotient space
  $\ker P_{|\vH^1_{(-1/2)^-} } / \ker P_{|\vH^{1}}$.
\end{enumerate}
\end{defn}

Similarly  we can define (although this will not be needed) a notion of zero-resonance for $ \widetilde P  = (S^*S)^{-1} P (S^*S)^{-1}$. Since
$(S^*S)^{-1}$ and  $S^*S$ are
 continuous on $\vH^{1}_{ s}$ for any $s$, zero is a
resonance of $\widetilde P $ if and only if it is a resonance of $ P $ and
their multiplicities are the same.  Following \cite{Ne,JK}  we
distinguish between four cases for the  threshold  $\lambda_0$ according to its spectral
nature as follows.

\begin{description}
\item[Case 0] \textit {A regular point}: $\lambda_0$ ($0$, resp.) is
  neither an eigenvalue nor a resonance of $H$ ($P$, resp.).
\item[Case 1] \textit {An exceptional point of the first kind}:
  $\lambda_0$ ($0$, resp.) is a resonance but not an eigenvalue of $H$
  ($P$, resp.).
\item[Case 2] \textit {An exceptional point of the second kind}:
  $\lambda_0$ ($0$, resp.) is an eigenvalue but not a resonance of $H$
  ($P$, resp.).
\item[Case 3] \textit {An exceptional point of the third kind}:
  $\lambda_0$ ($0$, resp.) is simultaneously  an eigenvalue and a resonance of
  $H$ ($P$, resp.).
\end{description}

From  Lemma \ref{lem:4.1a} we deduce the following result.
\begin{lemma}
\label{prop5.1}
$\lambda_0$ is a resonance (resp., an eigenvalue) of $H$ if and only
if zero is a resonance (resp., an eigenvalue) of $ P$ and their
multiplicities are the same.
\end{lemma}

                                   We need the following simplified version of
                                   Theorem
                                   \ref{thm:short-effect-potent}.

                                   \begin{thm} \label{thm5.2} Assume
                                     the conditions \eqref{ass5.1}  and
                                     \eqref{ass5.3}.
                                     \begin{enumerate}[\rm (a)]
                                     \item \label{item:first} Suppose  $v =
                                       (v_1, v_2)\in
                                       \vH^{1}_{(-1/2)^-} $ and $P v
                                       =0$. Then one has
                                     \begin{equation} \label{eq5.2a}
                                       v_j(x_j)
                                       = - \f{\inp{ 1,
                                         U_{j1}v_1 +
                                         U_{j2}v_2}} {
                                         4 \pi|x_j| } +
                                       w_j\text{ for  } |x_j|\geq 1,
                                     \end{equation}
                                     where
                                     $ w_j\in L^{2}(\bX_j)$;
                                     $j=1,2$.
                                     Moreover
                                     \begin{equation} \label{nonresonant1}
                                       v \in \vH \Longleftrightarrow
                                       \inp{ 1,
                                         U_{j1}v_1 + U_{j2}v_2}
                                       =0;\quad  j=1,2.
                                     \end{equation}
                                   \item \label{item:second}  Suppose  $u \in H^1_{(-1/2)^-}$
                                     is  a solution to  the equation
                                     $(H-\lambda_0) u=0$. Let $v =
                                     T^*u \in \vH^{1}_{(-1/2)^-}$.
                                     Then $u$ is an eigenfunction of
                                     $H$ at $\lambda_0$ if and only if
                                     \begin{equation}\label{nonresonant2}
                                       \inp{ 1, U_{j1}v_1 + U_{j2}v_2} =0;\quad  j=1,2.
                                     \end{equation}
\end{enumerate}
\end{thm}
 \begin{proof}
Since $n_j=3$ and we assume that
$q_j(\theta)=0$ for  $j=1,2$, in Theorem
                                   \ref{thm:short-effect-potent}
                                    the sets $\sigma_{j,1}=\set{\tfrac
                                      12}$. Whence by this theorem we need to
                                    consider the functionals $l_{j,\tfrac
                                   12, 1}$ applied to $v\in
                                       \vH^{1}_{(-1/2)^-} $ solving
                                       $Pv=0$.  Since $q_j=0$,

        \begin{align*}
          \sqrt{4 \pi}\,l_{j,\tfrac
                                   12, 1}(v)= \inp[\big] {
          \phi_\frac 12(|y_j|)|y_j|^{-1}, \Delta_{y_j}
          (\chi_2v_j)(y_j)}_{L^2(\d y_j)}.
        \end{align*}
 By
                                  integrating  by parts in the
                                 integral to
                                 the right of   this formula we
                                 verify that
                                 \begin{align*}
                                 \sqrt{4 \pi}\,l_{j,\tfrac
                                   12, 1}(v)=\inp[\big] {
          1, \Delta_{y_j}
          (\chi_2v_j)(y_j)}_{L^2(\d y_j)}=\inp[\big] {
          1, \Delta_{y_j}
          v_j(y_j)}_{L^2(\d y_j)}= \inp{ 1,
                                         U_{j1}v_1 +
                                         U_{j2}v_2},
                                 \end{align*} showing
                                 \eqref{eq5.2a}. Trivially
                                 \eqref{nonresonant1} is a consequence
                                 of \eqref{eq5.2a}. The above
                                 computation and Theorem
                                   \ref{thm:short-effect-potent}
                                   \ref{item:T4} yield \ref{item:second}.
\end{proof}

                                   \begin{remark}\label{remark:altFred}
      Although this was not used above let us note that the equation
        $P v=0$  for    $v = (v_1, v_2)\in \vH^1_{(-1/2)^-}$
        (equivalently) reads
        \begin{align}
          \label{eq:alt}v= - G_0Uv,
        \end{align}
where $G_0 := \slim_{z\to 0, \,z\not \in \,[0,\infty)} (P_0 -z)^{-1}$ is
computed as
$G_0 = G_{0,1}\oplus G_{0,2}$ with $G_{0,j}$ specified by  its integral kernel
\[
G_{0,j}(x_j, y_j) = \f{1}{4\pi|x_j-y_j|};\quad  j=1,2.
\] This assertion follows from an integration by parts argument as in
the proof of  Lemma \ref{lemHS} \ref{itemG2}.  We
                                   note that the operators
                                   $\inp{x}^{-s}G_{0,j}\inp{x}^{-s'}$
                                   are Hilbert-Schmidt operators for
                                   $s,s'>\tfrac 12$ and $s+s'>2$, cf. \cite[Lemmas 2.2]{JK}. The
                                   asymptotics  \eqref{eq5.2a} may alternatively be
derived from \eqref{eq:alt}.

                                   \end{remark}

The condition $\w{1, U_{j1}v_1 + U_{j2}v_2}=0$ for $j=1,2$ is
equivalent to the condition $\w{U^*c, v}=0$ for all $c \in \C^2$. Since $Hu=\la_0 u$ if and only if $Pv=0$ with $v =T^*u$, the following result is an immediate consequence of Theorem \ref{thm5.2}.
\\

\begin{cor} \label{cor5.3} Assume conditions of Theorem
  \ref{thm5.2}. Let $u \in H^{1}_{(-1/2)^-}$  be a solution to the equation $H u =\la_0 u$. Then $u$ is an $L^2$-eigenfunction of $H$ if and only if
\be \label{nonresonant3}
\w{TU^*c,u} =0 \quad \mbox{ for any } c \in \C^2.
\ee
\end{cor}

 The condition \eqref{nonresonant3} can be rewritten as the system of
 equations
 \begin{subequations}
 \begin{align}
  \int_\bX \varphi_1 I_1 u dx  + \inp{1, \check s_{12} (T^*u)_2} &=  0,  \label{nonresonant3a}\\
  \int_\bX \varphi_2 I_2u dx  +  \inp{1, \check s_{21} (T^*u)_1} &=  0. \label{nonresonant3b}
\end{align}
 \end{subequations}
  In fact, using $u = S v - R'(\lambda_0) H Sv$ with   $v = T^*u$  and
  writing
$Sv = \varphi_1 v_1 + \varphi_2 v_2$, one has
\begin{align*}
  \int_\bX  \varphi_1 I_1u dx & =  \inp{I_1\varphi_1,   S v - R'(\lambda_0) H Sv} \\
                              &= \inp{1, (W_1 + K_{11}(\lambda_0))v_1} + \inp{1, (W_{12} + K_{12}(\lambda_0))v_2},
                             \end{align*} showing, by using
                             \eqref{nonresonant3} with $c=(1,0)$, that
\begin{align*}
  \int_\bX \varphi_1 I_1 u dx  + \inp{1, \check s_{12} v_2}=\inp{1,(UT^*u)_1}=0.
\end{align*} Similarly \eqref{nonresonant3b} follows from the computation
\begin{align*}
  \int_\bX \varphi_2 I_2 u dx  + \inp{1, \check s_{21} v_1}=\inp{1,(UT^*u)_2}=0.
\end{align*}

By the same computations we see that
(\ref{nonresonant3a}) and  (\ref{nonresonant3b}) imply \eqref{nonresonant3}.
\\

For  $ s\in(\f 1 2, \f 3 2)$ we
 set
\begin{equation}
  \vK=\ker ( 1+ G_0U) \subset \vH^{1}_{-s}.
\end{equation}
This space $\vK$
coincides with the subspace
$\{ v \in \vH^{1}_{(-1/2)^-}|\, P v=0\}$ and is of
finite dimension. By Theorem \ref{thm5.2}, if zero is a resonance of
$P$, its multiplicity is at most two.
Next we introduce  the constant functions
\begin{subequations}
  \begin{align}\label{eq:gi}
     g_j(x_i) =\f{1}{2\sqrt{\pi}} \text{ for  } x_i\in{}\bX_i;\quad i=1, 2.
  \end{align}
 Set $\mathfrak{g}_1 = (g_1,0)$ and  $\mathfrak{g}_2 = (0,g_2)$.
For $v \in \vK$, define $c(v) =(c_1(v), c_2(v)) \in \C^2$ by
\begin{align}
  \label{eq:ci} c_i(v) = \w{\mathfrak{g_i}, Uv}=\w{g_i, U_{i1}v_1 + U_{i2}v_2}; \quad
i =1,2.
\end{align}
\end{subequations}

Theorem \ref{thm5.2} implies that  $v\in\vK$ is an eigenfunction of $P$ if and only if $c(v)=0$.   In case zero is a resonance of $P$, we use the following
 normalization of resonance states:

\begin{itemize}
\item  If zero resonance  is  simple, we denote by $\psi_1 \in \vK$ a resonance state
  such that
  \begin{equation} \label{normalization1} |c(\psi_1)| =1.
  \end{equation}
\item  If zero resonance is  double, we denote by $\psi_1$ and $
  \psi_2$ two
   resonance  states in $ \vK$ such that
  \begin{equation} \label{normalization2} c(\psi_1) =(1,0)\mand
    c(\psi_2) =(0,1),
  \end{equation} respectively.
\end{itemize}
If $\lambda_0$ is a resonance  of $H$, let $\kappa$ denote its multiplicity.  Then $\kappa\in\{1,2\}$ and
the corresponding normalization condition for resonance  states $u_j= (1-R'(\lambda_0)H)S\psi_j$, $1 \le j\le \kappa$,  of $H$ reads
\bea \label{normalization1H}
|\w{TU^*\mathfrak{g}_1, u_1}|^2 +  |\w{TU^*\mathfrak{g}_2, u_1}|^2 =1, & &\quad  \mbox{ if $\kappa=1$},\\
\label{normalization2H}
\w{TU^*\mathfrak{g}_i, u_j} = \delta_{ij}  \quad \mbox{for $1\le i, j \le 2$}, & &\quad \mbox{ if $\kappa=2$}.
\eea
With this choice of  resonance  functions  we  show in
Section \ref{sec:resolv-asympt-nearLOWEST}, see Theorem
\ref{thm5.10},  that if $\Sigma_2$ is a resonance of multiplicity $\kappa$ but not an eigenvalue of $H$, then
\be
R(z) = \f{\i}{\sqrt{z-\lambda_0}} \left ( \sum_{j=1}^\kappa{}\w{u_j, \cdot}u_j + \vO(|z-\lambda_0|^\ep)\right)
\ee
in $\vL(-1,s; 1,-s)$, $s>1$, for $z$ near $\lambda_0$ and $z-\lambda_0
\not\in [0,\infty)$.

The case where $\Sigma_2$ is an eigenvalue (and  possibly  a resonance
as well) is  also
studied,  see Theorem
\ref{thm5.10} when  $\rho_0>3$. For $\rho_0=3$ this requires an a priori weak decay
property of the corresponding eigenfunctions, see Theorem
\ref{thm:rhoThree}. On the other hand with this additional assumption in fact
we obtain the  expansion of the resolvent up to second
order. Expansions to higher orders usually require strong decay of the
potential, see \cite{JK, Wa2}. In particular in our setting this would
mean that $\rho_0$ should be sufficiently big. We shall not pursue
this direction, but  rather mainly restrict  our study to the case
$\rho_0=3$  which  indeed has
relevance  for the physics models, see discussions in Sections \ref{$N$-body
  Schr\"odinger operators} and \ref{$N$-body Schr\"odinger operators with infinite mass
  nuclei}.
\medskip

\subsubsection{The case $\lambda_0 > \Sigma_2$}\label{sec:case-lambda_0sub}

In this Subsection we consider the case
\be \label{ass5.1b} \mbox{
  $\lambda_0 >\Sigma_2$ is a non-multiple two-cluster threshold in the sense
  of Condition \ref{cond:uniq}.} \ee
This means there exists a unique
$a_0 \in\vA\setminus\{a_{\max}\}$ such that $\lambda_0
\in\sigma_{\pupo}(H^{a_0})\setminus\set{\Sigma_2}$, and this cluster
decomposition is a two-cluster decomposition obeying Condition
\ref{cond:geom_singl}.  As in Section \ref{Reduction near a
  simple two-cluster threshold} we let $m$ denote the multiplicity of
eigenvalue $\lambda_0$ of $H^{a_0}$ and $\{\varphi_1, \dots,
\varphi_m\} $ be an orthonormal basis of the associated eigenspace.
Let $\Pi $ be the projection in $L^2(\bX)$ defined by
\begin{equation}
  \label{eq5.1.40}
  \Pi g = \sum_{j=1}^m \varphi_j\otimes \inp{\varphi_j,g}_{L^2(\bX^{a_0})},  \;  g\in L^2(\bX).
\end{equation}
Put $\Pi' = 1 - \Pi$, $H' = \Pi' H \Pi'$ and $R'(z)
=(H'-z)^{-1}\Pi'$. Let  in this Subsection
\[
\vH = L^2(\bX_{a_0}; \C^m), \quad \vH^{k}_s = H^{k}_s(\bX_{a_0}; \C^m),
\quad \vH^k_{s^-}
  =\cap_{t<s}\vH^k_{t}; \quad k, s \in \R.
\]
  The scalar product in $\vH$ concerns only the $x_{a_0}$-variable and
  will be denoted as $\w{\cdot, \cdot}_0$, but we shall allow
  ourselves (slightly abusively) to use the same notation for the inner product on $L^2(\bX_{a_0})$. Define $S : \vH \to L^2(\bX)$ by
\be \label{eq5.1.41}
 S : f= (f_1, \dots, f_m) \to  Sf =\sum_{j=1}^m \varphi_j(x^{a_0}) f_j(x_{a_0}).
\ee
 From  Section \ref{Reduction near a
  simple two-cluster threshold} we have the representation formula for the resolvent,
\begin{equation} \label{rep5.1.2}
R(z)= E(z)- E_{+}(z) E_{\vH}(z)^{-1} E_{-}(z)\,\text{ for }\,  \Im z \neq 0,
\end{equation}
where
  \begin{align*}
E(z) &= R'(z),   \\
E_+(z)&=(1-R'(z)I_0)S,   \\
E_-(z) &= S^*\parb{1- I_0R'(z)},\\
E_{\vH}(z)&= (z-\lambda_0)-(P_0 + S^*{ I_0}S -
S^* I_0R'(z)I_0 S),
\end{align*}
where $P_0 =-\Delta_{x_{a_0}}{\textbf 1}_m $ with ${\textbf 1}_m$
being the identity matrix of size  $m$ and $I_0 = \sum_{b\not\subset a} V_b(x^b)$. We  assume
\eqref{ass5.2b} stated below. Then one has the following limiting absorption principle for $H'$,
\be\label{eq5.1.42}
\forall{ s > \tfrac 1 2}:\quad\|\w{x}^{-s} R'(\lambda\pm \i0)\w{x}^{-s}\|
\le C_s\text{ for } \lambda\text{  near }\lambda_0.
\ee
 This follows from Theorem \ref{thmlapBnd}   for $\breve{R}(z)$, recalling $R'(z) = \breve{R}(z) \Pi'$. Decompose $P^\pm = - E_{\vH}(\lambda_0 \pm \i0)$  as
 \begin{subequations}
  \be \label{eq5.1.43}
P^\pm = P_0 + U^\pm,
\ee
where
\be
U^\pm = S^*{ I_0}S - S^* I_0R'(\lambda_0\pm \i0)I_0 S.
\ee
 \end{subequations}
 In addition to (\ref{ass5.1b}) we assume that
 \begin{subequations}
  \bea
 & & \dim \bX_{a_0} =3\,\mand \,\lambda_0 \not\in \sigma_\pp(H'), \label{ass5.2b} \\
  & &\rho \geq{ 1} \,\mand \,S^*{ I_0}S\in \vL(\vH_{-\rho{_1}}^1,\vH_{\rho{_1}}^{-1})
  \mbox{ with  } \,\rho_1\geq\tfrac 32, \label{ass5.3b}\\
 & &\rho \ge \rho_1 -\tfrac 1 2.\label{ass5.4b}
\eea
 \end{subequations}

Recall that $\rho>0$ is the rate of decay of the pair potentials as
specified by (\ref{eq:1}).
The condition (\ref{ass5.4b})  will be convenient  (although not
essential)  in Section
\ref{sec:resolv-asympt-nearHIGHER}; it is not used in the present
section. (Note that with (\ref{ass5.4b}) the operator $U(z)$ of Section
\ref{sec:resolv-asympt-nearHIGHER} fulfills  $U(z) \in \vL(1, -s;
-1,s)$ for $ s < \rho_1$ uniformly in  $z$ near zero  with $\Im z \neq
0$.) Imposing only  (\ref{ass5.3b}) obviously the property
(\ref{ass5.4b}) is valid  for
$\rho_1=\tfrac 32$. This weak version of  (\ref{ass5.4b})  will  be very useful.

According to discussions  in Section 2.2, (\ref{eq5.1.42}) and    the condition (\ref{ass5.3b}) imply
\be
U^\pm \in \vL(\vH^{1}_{-s}, \vH^{-1}_{s})\,\text{ for }\,s< \tfrac 32.
\ee

 The free resolvent $r_0(z) =(P_0-z)^{-1}$ can be expanded in appropriately weighted spaces as
\be
r_0(z) = G_0 + z^{\f 1 2} G_1 + z G_2 + \cdots,
\ee see \eqref{eq:freeres},
and $K^\pm = G_0 U^\pm$ are compact operators on $\vH^{1}_{-s}$ for
$s\in (\tfrac 12, \tfrac 32)$.
Using the integral kernel of $G_0$, one deduces that
$\ker_{\vH^{1}_{-s}}(1+K^\pm)\subset \vH^{1}_{(-1/2)^-}$. Set $\vK = \ker_{\vH^1_{-s}}(1+K^+)$ for
$s\in (\tfrac 12, \tfrac 32)$. Let for any $s<\tfrac 32$
\begin{align*}
  \vE^\vG_{-s}=\set{u\in
    L^2_{-s}(\bX)|\, (H-\lambda_0)u=0,\, \Pi'u\in \vB^*_{1/2,0}(\bX)}.
\end{align*}

\begin{thm}\label{prop5.5} Assume conditions (\ref{ass5.1b}),
  (\ref{ass5.2b}) and (\ref{ass5.3b}), and   let $s\in  (\tfrac 12, \tfrac 32)$.
\begin{enumerate}[\rm (a)]
\item \label{item:plusminuss} For any $v \in \vK $ \be
  R'(\lambda_0+\i0)I_0Sv = R'(\lambda_0 -\i0)I_0Sv \quad \mbox{ and }
  \quad U^+v = U^-v\in{\vH^{-1}_{(2\rho{_1}-1/2)^-}}.  \ee Moreover \be \vK =
  \ker_{\vH^1_{-s}}(1+K^-)=\ker_{\vH^1_{-s}}(P^+)=
  \ker_{\vH^1_{-s}}(P^-).  \ee

\item  \label{item:2res}  For  any $v=(v_1, \cdots,v_m)\in \vK$ and $k=1,\dots,m$
\be
v_k(x_{a_0}) = -\f{\w{1, U^\pm_{k1}v_1 +\cdots +
    U^\pm_{km}v_m}_0}{4\pi|x_{a_0}|} + w_k(x_{a_0})\text{ for }|x_{a_0}|\geq{1},
\ee
 where $w=(w_1, \cdots, w_m) \in \vH$.

\item \label{item:2resDim}
 Let $v\in\vK$. Then $v\in \vH$ if and only if
\be
\forall c\in \C^m:\quad\w{c,  U^\pm v}_{0}=0.
\ee
Here $c \in \C^m$ is considered as an element of $\vH^1_{-s}$, $s > \f 3 2$.

\item
 $u \in \vE^\vG_{-s}$ if and only if $u = (S-R'(\lambda_0 \pm \i0) I_0S)v$
 for some $v\in \vK$. In this case $v$ is  uniquely given by  $v= S^*u\in \vK$.

\item $u\in \vE^\vG_{-s}\cap L^2{(\bX)}$ if and only if $v=S^*u  \in \vK\cap
 \vH$.
\end{enumerate}
\end{thm}
\pf We may argue essentially as in Section \ref{subsec:hd2ep},  in
particular using Lemmas
\ref{lem:eigentransform} and  \ref{lem:eigentransform2}, although the
setup is different there (being a multiple  two-cluster case). The
parametrix $G_0$ is of course simpler to use than the more general construction
used in  Section  \ref{subsec:hd2ep}, but the resulting null space
$\ker_{\vH^1_{-s}}(1+K^+)$ is
 independent of the specific choice of
parametrix to define $K^+$, cf. Remark \ref{remark:altFred}. Accepting
 a detour to  Section  \ref{subsec:hd2ep} the identification of
functionals would follow  like in the proof of  Theorem
\ref{thm5.2}; alternatively \ref{item:2res} follows from the
asymptotics of  $G_0$, cf. Remark \ref{remark:altFred}. \ef

We are lead to define the following effective version of the notion of threshold resonances.

\begin{defn}
\begin{enumerate}[(1)]
\item \label{item:10e}$\lambda_0$ is a \emph{resonance of $H$}  if the
  equation $Hu = \lambda_0u$ admits a solution in
  $\vE^\vG_{(-1/2)^-}:=\cap_{s>1/2} \,\vE^\vG_{-s}$
  which is not in $H^1$. The multiplicity
  of the resonance $\lambda_0$ is defined as the dimension of the quotient
  space
  $\vE^\vG_{(-1/2)^-}/ \ker
  (H-\lambda_0)_{|H^{1}}$.
\item\label{item:9e} $0$ is  a \emph{resonance of $ P^\pm$} if the equation
  $ P^+ v = 0$ (and therefore $ P^- v = 0$,  and vice versa) has a solution $v \in \vH^1_{(-1/2)^-} \setminus \vH^1$. The multiplicity of the resonance zero is defined as
  the dimension of the quotient space
  $\ker P^+_{|\vH^1_{(-1/2)^-} } / \ker P^+_{|\vH^{1}}$.
\end{enumerate}
\end{defn}
Note  the consequence of  Theorem
\ref{prop5.5} that if zero is a resonance of $P^\pm$ then  its multiplicity
does not exceed $m$. As in Subsection \ref{sec:case-lambda_0=sub} we
can, based on those definitions, introduce Cases $0-3$, and
 Theorem
\ref{prop5.5} then shows that the threshold spectral properties of $H$ at $\lambda_0$ are determined by those of $P^\pm$ at zero.
In fact,  completely parallel  to the case of the lowest threshold
$\Sigma_2$,  $\lambda_0$ is  a regular point  (resp.,  an exceptional
point  of the first kind, the second kind, the third kind) of $H$ if
and only if zero is a regular point  (resp.,  an exceptional point  of the first kind, the second kind, the third kind) of $P^\pm$.
\\

Assume zero is a resonance of $P^\pm$. Then  the quotient space
$\vK/\vH:=\vK/(\vK\cap\vH)$ has dimension $\kappa = \dim \vK/\vH \geq
1$. We call $\kappa$ the multiplicity of the zero resonance of
$P^\pm$. Let $\mu = \dim \vK$.
For $\phi =(v_1, \cdots, v_m)\in \vK$, define $c(\phi) \in \C^m$ by
\begin{align}
  \label{eq:constants}
  c(\phi) = \f{1}{2 \sqrt{\pi}}\parb{\w{1, U^\pm_{11}v_1 +\cdots + U^\pm_{1m}v_m}_0, \cdots, \w{1, U^\pm_{m1}v_1 +\cdots + U^\pm_{mm}v_m}_0}.
\end{align}

Theorem \ref{prop5.5} \ref{item:2res}  shows that $\phi\in\vK$ is
a resonance state of $P^\pm$ if and only if $c(\phi) \neq 0$. Clearly
$c(\cdot)$ is a linear action on $\vK$. It follows
that a family of resonance  states $\{\psi_1, \dots, \psi_k\}$ of $P^\pm$  is linearly independent in  $\vK/\vH$ if and only if the family of vectors $\{c(\psi_1), \dots, c(\psi_k)\}$ is linearly independent in $\C^m$.
\\

\begin{prop}\label{prop5.7b}
  \begin{enumerate}[\rm(a)]
  \item \label{item:21s}
   Assume zero is a resonance of $P^\pm$ with multiplicity $\kappa$. Then
there exists a basis
$\{\psi_1, \cdots, \psi_\kappa \}$ of $\vK/ \vH$ such that
\begin{align}
  \label{normalization5.1.2}
(c(\psi_i), c(\psi_j)) =\delta_{ij},  \quad i, j =1, \cdots, \kappa,
\end{align}
where $(\cdot, \cdot)$ is the  scalar product of $\C^m$.\\

\item \label{item:22s} Assume  zero is a resonance but not an
  eigenvalue of $P^\pm$ (i.e. that $\kappa =\mu$). Then the operator $Q$ defined by
\be
Q =\sum_{j=1}^\kappa\w{\psi_j, \cdot} \psi_j : \vH^{-1}_{ s} \to \vH^1_{-s}, s> \tfrac 1 2,
\ee
is independent of the choice of basis $\{\psi_1, \dots, \psi_\kappa \}$ of $\vK$ verifying (\ref{normalization5.1.2}): If   $\{\psi'_1, \dots, \psi'_\kappa \}$ is another basis of $\vK$ verifying (\ref{normalization5.1.2}), then one has
\[
\sum_{j=1}^\kappa\w{\psi_j, \cdot} \psi_j  = \sum_{j=1}^\kappa\w{\psi'_j, \cdot} \psi'_j.
\]
\end{enumerate}
\end{prop}
\begin{proof}
\subStep {\ref{item:21s}} Let $\Phi=\{\phi_1, \dots, \phi_\kappa\}$ be a basis of $\vK/\vH$.
Then the rank of the matrix $C(\Phi) =( c(\phi_1), \dots,
c(\phi_\kappa)) \in \vM_{m\times \kappa}$ is equal to $\kappa$, where
$c(\phi_j)$ is considered as the $j$'th column of $C(\Phi)$.
Consequently $C(\Phi)^*C(\Phi)$ is positive definite.
Let  $M_0\in \vM_{\kappa\times \kappa}(\C)$ be the  positive definite
Hermitian matrix obeying  $M_0^2 = (C(\Phi)^*C(\Phi))^{-1}$.
Set $M_0 =(m_{ij})_{1\le i, j \le \kappa}$ and define
\be
\psi_i =\sum_{k=1}^\kappa m_{ki} \phi_k, \quad i =1, \cdots, \kappa.
\ee
 Then  $\{\psi_1, \dots, \psi_\kappa\}$ is also a basis of  $\vK/\vH$.

 Let $C(\Psi)$ be the matrix defined in the same way as $C(\Phi)$ with $\Phi =\{\phi_1, \dots, \phi_\kappa\}$ replaced by
 $\Psi=\{\psi_1, \dots, \psi_\kappa\}$.  Then
 \[
 c(\psi_j) = \sum_{k=1}^\kappa m_{kj} c(\phi_k);\quad j=1, \cdots, \kappa.
 \]
  This shows $C(\Psi) = C(\Phi) M_0$ and
\[
C(\Psi)^*C(\Psi)= M_0 C(\Phi)^* C(\Phi) M_0 =1\text{ in }\vM_{\kappa \times \kappa}(\C).
\]
  It follows that  $\Psi=\{\psi_1, \dots, \psi_\kappa\}$ is a basis of $\vK/ \vH $ verifying the normalization condition (\ref{normalization5.1.2}).
\\

\subStep {\ref{item:22s}}
 Let $t_{ij} \in \C$, $1\le i, j \le \kappa$, be  such that
\[
\psi'_i = \sum_{j=1}^\kappa t_{ij} \psi_j, \quad  i =1, \cdots, \kappa.
\]
Then $c(\psi'_i) =  \sum_{j=1}^\kappa t_{ij} c(\psi_j)$ in $\C^m$.
The condition $(c(\psi'_i), c(\psi'_j)) =\delta_{ij}$ becomes
\begin{align*}
(c(\psi'_i), c(\psi'_j)) &= \sum_{l,m =1}^\kappa \overline{t_{il}} t_{jm} (c(\psi_l), c(\psi_m))  \\
 & =   \sum_{l =1}^\kappa \overline{t_{il}} t_{jl} = \delta_{ij}
\end{align*}
for $i, j =1, \dots, \kappa$. This means  that if both $\{\psi_1, \dots, \psi_\kappa\}$ and $\{\psi'_1, \dots, \psi'_\kappa\}$ satisfy (\ref{normalization5.1.2}), the matrix $T =(t_{ij})_{1\le i, j\le \kappa}$ is unitary.  We obtain
\begin{align*}
 \sum_{i=1}^\kappa\w{\psi'_i, \cdot} \psi'_i
=     \sum_{l,m =1}^\kappa \parbb{\sum_{i =1}^\kappa \overline{t_{il}}
  t_{im}} \w{\psi_l, \cdot} \psi_m
=     \sum_{l,m =1}^\kappa \delta_{lm}\w{\psi_l, \cdot} \psi_m
=      \sum_{l=1}^\kappa\w{\psi_l, \cdot} \psi_l.
\end{align*}
  \end{proof}

The normalization condition of the resonance
states (\ref{normalization5.1.2}) will be  used to compute the leading term of the
resolvent $R(z) =(H-z)^{-1}$ for $z$ near $\lambda_0$ and $\Im z \neq
0$ in the case $\lambda_0$ is a resonance but not an eigenvalue of
$H$. The normalization condition (\ref{normalization1})  can be regarded as a special case of (\ref{normalization5.1.2}),
and Proposition \ref{prop5.7b} \ref{item:22s} is also valid for the
case $\lambda_0 =\Sigma_2$ is a double two-cluster threshold.

A main result on the resolvent expansion in Subsection
\ref{sec:resolv-asympt-nearHIGHER} is for the case where $\lambda_0$ is
not an eigenvalue of $H$, see Theorem \ref{thm5.20}, Case 0 and Case 1.
However the case where $\lambda_0$ is an eigenvalue (and possibly a resonance
as well) is also  studied, cf. Subsection
\ref{sec:case-lambda_0=sub}. For the physical models we then need an
additional  weak decay property of the corresponding $L^2$-eigenfunctions, see
Theorem \ref{thm:rhoThree2}. On the other hand with this additional
assumption   we obtain the expansion of the resolvent up to
second order.  \medskip

\subsection{Resolvent asymptotics near the lowest
  threshold}\label{sec:resolv-asympt-nearLOWEST}


 In this subsection, we keep the conditions and the notation of
 Subsection \ref{sec:case-lambda_0=sub}, in particular, $\lambda_0 = \Sigma_2$ is a double two-cluster threshold.
  We want to study the asymptotics of the resolvent $R(z)
 =(H-z)^{-1}$ near  $ \Sigma_2$, using
 the formula (\ref{rep5.9}) for $R(z)$.
 We let $P(z) = - E_{\vH}(\lambda_0 + z)$ and decompose $P(z)$ as
\begin{subequations}
  \begin{equation} \label{Ulambda}
P(z) = P_0 +
 U(z) -z
\mand U(z) = U + z U_1 + z ^2 U_2(z)
\end{equation}
with $P_0$ and $U$ defined by  \eqref{P0} and \eqref{U}, respectively, and
\begin{align} \label{U1}
U_1&= -\left( \begin{array}{cc}
0 & \inp{\varphi_1, \varphi_2 \otimes \cdot}_1, \\
\w{\varphi_2, \varphi_1 \otimes \cdot}_2 & 0
 \end{array} \right) - K_1(\lambda_0), \\
  U_2(z) &= - \sum_{j=2}^\infty z^{j-2}  K_j(\lambda_0),\\\label{kj}
K_j(\lambda_0) &= (\w{\varphi_l, I_l(R'(\lambda_0))^{j+1} I_m \varphi_m
  \cdot})_{ 1\le l, m\le2}; \quad j\in\N.
\end{align}
 \end{subequations}
Since $K(\lambda_0 + z)$ is holomorphic in $z$ near zero and continuous from $\vH^{1}_{-1-\rho}$ to $\vH^{-1}_{1+\rho}$  with $\rho>0$ given by \eqref{eq:1}, the
above power series converges in the space  $\vL(1,
-1-\rho; -1, 1+\rho)$ for example. Note that
\be
P(z) = P-  z (1-U_1) +  z^2 U_2(z);\quad P=P_0 + U.
\ee
Differently from one-body Schr\"odinger operators, cf.  \cite{JK}, $P(z)$ is an operator
pencil  depending on the  spectral parameter in a non-linear way.  The
following result is important for the existence of  an asymptotic
expansion  of $P(z)^{-1}$ as $z \to 0$ in the case $0$ is an
eigenvalue of $P$.\\

\begin{lemma}\label{lem5.5}
 $1-U_1$ is positive definite  on $\vH$.
\end{lemma}
\begin{proof} Note firstly that $K_1(\lambda_0) \ge 0$. In fact, for $f=(f_1, f_2)\in \vH$, since $R'(\lambda_0)$ is a bounded self-adjoint operator,
\[
\w{f, K_1(\lambda_0)f} = \sum_{i,j}\w{\varphi_i \otimes  f_i, I_i(R'(\lambda_0))^2 I_j \varphi_j \otimes f_j} = \|R'(\lambda_0) F\|^2 \ge 0,
\]
where $F = \sum_{j=1}^2 I_j \varphi_j \otimes  f_j$.   One can check that
\[
 1 + \left( \begin{array}{cc}
0 & \inp{\varphi_1, \varphi_2 \otimes \cdot}_1 \\
\w{\varphi_2, \varphi_1 \otimes \cdot}_2 & 0
 \end{array} \right) \ge 0.
\]
In fact, cf. the proof of Proposition \ref{prop2.2}, its expectation
value on $f$ is given by
\[
\|f_1\|_1^2 + \|f_2\|_2^2 + 2 \Re \w{\varphi_1\otimes f_1, \varphi_2\otimes f_2} =\|\varphi_1 \otimes f_1 + \varphi_2 \otimes f_2\|^2.
\] It follows that
\begin{equation}\label{eq:formU}
  \w{(1-U_1)f, f} = \|\varphi_1 \otimes f_1 + \varphi_2 \otimes f_2\|^2 +  \|R'(\lambda_0) F\|^2 \ge 0.
\end{equation}
Therefore,  $1-U_1$ is non-negative  on $\vH$ and $f \in \ker (1-U_1)$ if and only if
\[
 \varphi_1 \otimes f_1 + \varphi_2 \otimes f_2 =0  \mbox{ and } R'(\lambda_0)F=0.
\]
 In particular, since we assume \eqref{ass5.1}, $\ker (1-U_1)$ is reduced to $\{0\}$ and $1-U_1$ is positive definite on $\vH$.
\end{proof}

Set
\be
r_0(z) = (P_0 -z)^{-1}; \quad z\not\in [0,\infty).
\ee
To obtain the asymptotic expansion of $P(z)^{-1}$ for $z\not\in
[0,\infty)$ and $z$ near zero, we first use the resolvent equation
\be
P(z)^{-1} = W(z)^{-1}r_0(z)\text{ where } W(z):= 1 + r_0(z) U(z).
\ee
 Next we shall  apply the Grushin method to study the factor
 $W(z)^{-1}$.

Let $N \in \N$ and $s > N + \f 1 2$. Since $\dim \bX_j= 3$,  the free resolvent $r_0(z)$ can be expanded in
$\vL(-1,s; 1,-s)$  as
\begin{equation}\label{eq:resEXp}
r_0(z) = G_0 + \sqrt{z} G_1 + \cdots +z^{\f N 2} G_N  + \vO(|z|^{\f N 2 +\ep}),
\end{equation}
for some $\ep >0$ depending on $N$ and $s$, where $G_i =
\textrm{diag}(G_{i,1}, G_{i,2})$ for $i\leq N$ are  diagonal matrices  which can be
calculated explicitly, cf. \cite[Lemma 2.3]{JK}. In particular the integral kernel of $G_{i,j}$ for $i=0,1$ and $j=1,2$ are given by
\bea
G_{0,j}(x_j, y_j) = \f{1}{4\pi|x_j -y_j|}, \quad  G_{1,j}(x_j, y_j) =
\f{\i}{4\pi};\quad j=1,2.
\eea

We also recall, cf. \cite[Lemmas 2.1--2.3]{JK}, that $r_0(z)\in\vL(-1,s';
1,-s)$ if $s,s'>1/2$ and  $s+s'>2$
 with a H\"older continuous dependence in $z$ at $z=0$. Hence
\begin{align}\label{eq:basicHo}
  r_0(z)-G_0 =\vO(|z|^{\ep}) \in\vL(-1,s'; 1,-s)\text{ for
  }s,s'>\tfrac 12\mand s+s'>2.
\end{align}

\begin{subequations}
  We consider now expansions of the operator $W(z)$ using the
  expansions \eqref{Ulambda}, \eqref{eq:resEXp} and
  \eqref{eq:basicHo}.
\begin{lemma}\label{lemma:W-asympt-near} Under the conditions
  \eqref{ass5.1}--\eqref{ass5.4}  the following expansions in $z$ (with  $z\not\in
[0,\infty)$) hold in terms of the quantities
  \begin{equation*}
W_0  =  1+ G_0 U, \quad  W_1 =  G_1 U, \quad W_2  =G_2 U +G_0U_1,
\end{equation*}   and  for some (small) positive number
$\epsilon$ (depending on given parameters $s$ and $s'$).
\begin{enumerate}[\rm(a)]
\item  For  $s>\f{1}{2}$, $s\geq s'$ and $\rho_0
-s' >\max\set{\f{1}{2}, 2-s}$
\begin{equation} \label{Wz00}
W(z) = W_0 +  \vO(|z|^{\epsilon})\in \vL(1, -s'; 1, -s).
\end{equation}
\item  For  $s>\f{3}{2}$, $s \geq s'$  and $\rho_0
-s' >\f{3}{2}$
\begin{equation} \label{Wz1}
W(z) = W_0 + \sqrt{z} W_1 +  \vO(|z|^{\f 1 2 +\epsilon})\in \vL(1, -s'; 1, -s).
\end{equation}

\item For  $s>\f 5 2$, $s\geq s'$  and $\rho_0 -s'
>\f 5 2$
\begin{equation} \label{Wz2}
W(z) =  W_0 + \sqrt{z} W_1 + z W_2 +  \vO(|z|^{1+\epsilon})\in \vL(1, -s'; 1, -s).
\end{equation}
\end{enumerate}
\end{lemma}
\end{subequations}

 The assertion  \ref{item:22} follows from \eqref{eq:basicHo}, while \ref{item:20} and \ref{item:21} follow from
the bounds
 \eqref{eq:resEXp} with $N=1$ and $N=2$, respectively. In all cases we
use \eqref{Ulambda} as well.

Higher order asymptotic expansions can also be established under
stronger decay assumptions on the effective potentials.
\\

 From the identity   $\w{-U \cdot, \cdot} = \w{P_0 \cdot, \cdot}$ on  $\vK=\ker ( 1+ G_0U) $ and an
integration by parts (cf. \eqref{eq:compCom}) it follows that $\w{-U\cdot, \cdot}$  is a
positive  quadratic form on  $\vK$.  Let
\[
\mu = \dim \vK,
\]
and let $\{\phi_1, \dots, \phi_\mu\}$  be a basis of $\vK$ orthonormalized with respect to $\w{-U \cdot, \cdot}$:
\[
\w{-U \phi_i, \phi_j} =\delta_{ij}.
\]
 This normalization is a
technical tool from \cite{Wa2}
which  in general does not conform with (\ref{normalization1}) and
(\ref{normalization2}).
We make the convention that if zero is a resonance of $P$ with
multiplicity $\kappa$,  $\phi_j$ for  $1\le j\le \kappa$ are resonance
states and $\phi_j$  for  $\kappa< j\le \mu$
(for  $\kappa<\mu$ only of course) are eigenstates of $P$.

In order to obtain the expansion of $W(z)^{-1}$, consider the Grushin problem
\[
\vW(z) = \left(
                \begin{array}{cc}
                  W(z) & \vS\\[.1in]
                   \vS^* & 0
                \end{array}
                \right) \;  :\; \vH^1_{-s}\times \C^\mu \to \vH^1_{-s}\times \C^\mu,
\]
where $s\in(\f{1}{2},\rho_0-\f{1}{2})$, $\vS: \C^\mu \to \vH^1_{-s}$ is  defined by
\[
\vS c
= \sum_{j=1}^\mu c_j\phi_j,\;\;  c = ( c_1, \dots,  c_\mu) \in \C^\mu,
\]
and
\[
\vS^*f=( \w{-U\phi_1,f} \dots, \w{-U\phi_\mu, f}),\quad  f\in
\vH^1_{-s}.
\]
 Note that $\vS^*$ can be regarded as the formal adjoint of $\vS$ with
 respect to the form $\w{-U\cdot, \cdot}$ on $\vH^1_{-s}$ (it is not
 the Hilbert space adjoint). Define for $s$ as above the  map $Q:  \vH^1_{-s}\to \vH^1_{-s}$ by
\[
Qf=\sum_{j=1}^\mu \w{-U \phi_j,f}\phi_j.
\]
Then,
\[
\vS\vS^* = Q  \mbox{ on $ \vH^1_{-s} $ and } \vS^*\vS = 1  \mbox{ on } \C^\mu,
\]
in particular $Q$ is  a projection in  $ \vH^1_{-s} $.

One can then prove  (see \cite{JK,Wa2}) that
\begin{equation}\label{eq:directSum}
\vH^1_{-s}=\vK\oplus  \ran (1+G_0U),
\end{equation}
and that $Q$ projects   onto $\vK$
relatively to the direct sum decomposition \eqref{eq:directSum}, in particular   $\ker Q = \ran(1+G_0U)$.
Then of course  $Q'=1-Q$ is  the   projection from $\vH^1_{-s}$ onto $
\ran(1+G_0U)$ relatively to \eqref{eq:directSum}. It follows
readily that
$Q'(1+ G_0U)Q'$ is bijective on $ \ran(1+G_0U)$. Since $ \ran(1+G_0U)$
is closed the operator
 \begin{equation}
 D_0 := (Q'(1+ G_0U)Q')^{-1}Q'
\end{equation}
exists and is continuous on $\vH^1_{-s}$. By an argument
of perturbation based on  \eqref{Wz00} it follows that for $|z|$ small
enough   $Q'W(z)Q'$ is invertible on $ \ran (1+G_0U)$ with continuous inverse. Let
 \[
  D(z) = (Q'W(z)Q')^{-1} Q'.
\]
 The following expansions hold in $\vL(1,-s; 1,-s)$ under the specified
         conditions and with
\begin{align*}
D_0 = (Q'W_0Q')^{-1}Q', \quad D_1  = -D_0W_1D_0, \quad  D_2 =  -
                                    D_0W_2D_0
  {+D_0W_1D_0W_1D_0}.
\end{align*}
 \begin{subequations}
 \begin{align}
\label{term1}
D(z)   &=    D_0 +  \vO(|z|^{\epsilon}),  \mbox{ if }  \tfrac{1}{2} < s<\rho_0 -\tfrac{1}{2}. \\
D(z)   &=    D_0 + \sqrt{z} D_1 +  \vO(|z|^{\f{1}{2} +\epsilon}),  \mbox{ if } \tfrac{3}{2}<s<\rho_0- \tfrac{3}{2}.  \label{term2}\\
D(z)   &=    D_0 + \sqrt{z} D_1 +  z D_2 +  \vO(|z|^{1+\epsilon}),  \mbox{ if } \tfrac 5 2<s<\rho_0-\tfrac{5}{2}.
\label{term3}
\end{align}
 \end{subequations}

Note that if $\vK =\{0\}$, we have $Q'=1$ and $ D(z) = W(z)^{-1}$ and
the asymptotic expansion of $W(z)^{-1}$ is given by the one  of $D(z)$.
In the following we treat  the case $\vK \neq \{0\}$. The assertions
\eqref{term2} and \eqref{term3} are  not needed for leading term
expansions which is our main interest.

Using the operator $D(z)$, we can compute the inverse of $\vW(z)$ as
\begin{equation} \label{vWz}
\vW(z)^{-1}=\left(
                           \begin{array}{cc}
              \vE(z) & \vE_+(z)
                            \\[.1in]
             \vE_-(z)  & \vE_{-+}(z)
                           \end{array}
                          \right),
\end{equation}
where
\begin{gather*}
\vE(z)  =  D(z),\quad  \vE_+(z) = \vS - D(z)W(z)\vS, \\
\vE_-(z )  =  \vS^* -\vS^*W(z)D(z),\quad   \vE_{-+}(z) =  \vS^*(-W(z) +  W(z)D(z) W(z))\vS.
\end{gather*}
One obtains from (\ref{vWz}) a representation of  the inverse of $W(z)$,
\begin{equation} \label{repW}
W(z)^{-1} = \vE(z) - \vE_+(z) \vE_{-+}(z)^{-1} \vE_-(z).
\end{equation}
 $\vE_{\pm}(z)$ and $\vE_{-+}(z)$ can be expanded similarly as $D(z)$,
 that is
\begin{align*}
\vE_\pm(z)& = \vE_{\pm,0} + \sqrt{z} \vE_{\pm,1}  + z \vE_{\pm,2} + \dots \\
\vE_{-+}(z)& = \vE_{-+,0} + \sqrt{z} \vE_{-+,1}   + z \vE_{-+,2} + \cdots
\end{align*}
More precisely one has the following result.

\begin{lemma}\label{lem5.6}  Assume \eqref{ass5.1}--\eqref{ass5.4}.
\begin{enumerate}[\rm (a)]
\item\label{item:17}  One has in $\vL(\C^\mu, \vH^1_{-s})$ (for the $ + $ case) or $\vL( \vH^1_{-s}, \C^\mu)$ (for the $ - $ case):
  \begin{subequations}
  \begin{align}
\vE_{\pm}(z)   &=    \vE_{\pm,0}  +  \vO(|z|^\ep),  \mbox{ if }
                 s>\tfrac{1}{2}.  \label{term2pm}\\
\vE_{\pm}(z)   &=    \vE_{\pm,0} + \sqrt{z} \vE_{\pm, 1 } +  \vO(|{z}|^{\f 1 2 +\ep}),  \mbox{ if } s>\tfrac{3}{2}.  \label{term2pmbis}\\
\vE_{\pm}(z)   &=    \vE_{\pm,0} + \sqrt{z} \vE_{\pm, 1 } + z \vE_{\pm, 2 } +  \vO(|z|^{1+\ep}), \mbox{ if } s>\tfrac 5 2\mand \rho_0>4.
\label{term3pm}
\end{align}
  \end{subequations}
 Here
\begin{align*}
\vE_{+, 0} = \vS, \quad & \vE_{-,0} = \vS^*, \\
\vE_{+,1}=  - D_0W_1 \vS, \quad & \vE_{-,1}=  - \vS^*W_1D_0,  \\
\vE_{+,2} = - (D_0 W_2+ D_1W_1)\vS,  \quad &  \vE_{-,2} =  -\vS^*(W_2D_0 + W_1D_1).
\end{align*}

\item \label{item:19} One has in $\vM_{\mu\times \mu}(\C)$:
  \begin{subequations}
   \begin{align}
\label{vEpm0}
\vE_{-+}(z)& =  \sqrt{z}  \vE_{-+, 1} + \vO(|{z}|^{\f 1 2 +\ep}).\\
 \vE_{-+}(z)& =  \sqrt{z}  \vE_{-+, 1} + z \vE_{-+, 2} +  \vO(|z|^{1+\ep}), \mbox{ if } \rho_0 >3. \label{vEpm1}
\end{align}
  \end{subequations}

Here
\begin{subequations}
\begin{align*}
\vE_{-+, 1} & =  - \vS^*W_1\vS,
\\
\vE_{-+, 2}  &= -\vS^*(W_2  - W_1D_0 W_1 )\vS.
\end{align*}
 \end{subequations}
\end{enumerate}
\end{lemma}
\begin{proof}
  \subStep{\ref {item:17}} Recall that  $\vS \in
  \vL(\C^\mu, \vH^{1}_{-s'})$ for any $s'>\f{1}{2}$.  Since $W_0\vS
  =0$, using \eqref{Wz00}-\eqref{Wz2}, the following
  expansions in $\vL(\C^\mu;
  \vH^1_{-s})$ hold.
  \begin{subequations}
   \begin{align}\label{eq:for0}
   \forall  s> \tfrac{1}{2}:\quad W(z)\vS& = \vO( |z|^{ \ep})\vS.\\
\label{eq:for1}
    \forall  s> \tfrac{3}{2}:\quad W(z)\vS& = \sqrt{z} W_1 \vS + \vO( |z|^{\f 1 2 + \ep}).\\
   \forall  s> \tfrac{5}{2}:\quad W(z)\vS& = ( \sqrt{z} W_1 + z W_2) \vS + \vO(|z|^{1+\ep}), \quad
            \mbox{ if } \rho_0 >3.\label{eq:for2}
 \end{align}
  \end{subequations}
The expansions for $\vE_+(z)$ follows from \eqref{term1}, \eqref{term2}
and \eqref{eq:for0}-\eqref{eq:for2}. The results for $\vE_-(z)$ can be proved in a similar way.

\medskip

 \subStep{\ref {item:19}}    For (\ref{vEpm0})  we use  (\ref{Wz00}),
 (\ref{Wz1}) and \eqref{term1}. Note that indeed since
 $W_0\vS=\vS^*W_0 =0$ and $D(z) : \vH^1_{-s} \to \vH^1_{-s}$ is
 uniformly bounded for any $s\in(\tfrac12, \rho_0-\tfrac12)$, one obtains
 \begin{subequations}
\begin{align}
\label{eq:as0m11}
  \vS^*W(z)\vS &= \sqrt{z}\vS^*W_1\vS +\vO(|z|^{\f 1 2 +\ep}),\\
  \label{eq:as0m}
  \vS^*W(z)D(z)W(z)\vS &= \vO(|z|^{\f 1 2 +\ep}).
\end{align}
 \end{subequations}
 Clearly (\ref{vEpm0}) follows from
 \eqref{eq:as0m11} and \eqref{eq:as0m}.

For (\ref{vEpm1})  we use  (\ref{Wz1}), (\ref{Wz2}) and
(\ref{term1}).  Note that
\begin{subequations}
  \begin{align}
\label{eq:as0m1111}
  \vS^*W(z)\vS &= \sqrt{z}\vS^*W_1\vS +z\vS^*W_2\vS+\vO(|z|^{1 +\ep}),\\\label{eq:as1m}
  \begin{split}
   \vS^*W(z)D(z)W(z)\vS &= z \vS^*W_1D(z)W_1\vS + \vO(|z|^{1 +\ep})\\&=  z
   \vS^*W_1D_0W_1\vS + \vO(|z|^{1 +\ep}).
 \end{split}
\end{align}
 \end{subequations}
   Clearly
   (\ref{vEpm1}) follows from \eqref{eq:as0m1111} and \eqref{eq:as1m}.
  \end{proof}

 If zero is an eigenvalue but not a resonance, then $\phi_j \in L^2$ for all $j$ and
$\vS$ is continuous from $\C^\mu$ to $\vH$, and by Theorem
\ref{thm5.2}   the composition $W_1 \vS =0$.

In the case $\rho_0>3$ the asymptotics (\ref{vEpm1})
 then amounts to the statement
\begin{align}
  \label{eq:redImp}
  \vE_{-+}(z) =   -z\vS^* W_2 \vS +  \vO(|z|^{1+\ep}).
\end{align}
 For  $\rho_0=3$  the right hand sides of (\ref{vEpm1}) and \eqref{eq:redImp} make
 sense since $\vS$ maps to $\vH$, however we dont know if
these  asymptotics  still hold  in that case. In fact we only know the following
weaker (and too poor) assertion for $\rho_0=3$,
\begin{align}
  \label{eq:redImp2}
  \quad \forall \epsilon>0:\quad \vE_{-+}(z) =\vO(|z|^{1-\ep}).
\end{align}

To show \eqref{eq:redImp2} we apply (\ref{Wz2}) with
$s'=0$ and $s \in ( \f 5 2,3]$ and conclude that
\begin{align*}
  \vS^* W(z) \vS =z \vS^* W_2 \vS +  \vO(|z|^{1+\ep})=\vO(|z|).
\end{align*}

Next by \eqref{Wz1},
 \[
 W(z)\vS = \vO( |{z}|^{\f 1 2 +\epsilon})
 \]
 in $\vL(\C^\mu, \vH^1_{-s})$ for any $s>\f{3}{2}$ and with $\epsilon=\epsilon(s)>0$.  We apply (\ref{term1}) with
  an   $s $ taken close to $\f{3}{2}$. Then we  argue that for
  any small $\delta>0$
  \begin{align}\label{eq:ssbndG}
    \vS^* W(z) = \vO( |{z}|^{\f 1 2-\kappa} ) \text{ in } \vL(\vH^1_{-s}.
    \C^\mu),\, \kappa:=(s-\tfrac 3 2 +\delta)/2,
  \end{align} Given \eqref{eq:ssbndG} it follows that
  \[
  \vS^* W(z) D(z) W(z) \vS = \vO( |{z}|^{1 +\epsilon-\kappa}),
  \]
   showing that
   \begin{align*}
    \vE_{-+}(z) +\vS^* W(z) \vS =  \vO(|z|^{1+\ep-\kappa})=  \vO(|z|^{1-\kappa}).
   \end{align*} In particular \eqref{eq:redImp2} holds.

   The bound \eqref{eq:ssbndG} follows by interpolating the bounds
\begin{align*}
\vS^* W(z) &= \vO( |{z}|^{\epsilon' } ) \text{ in } \vL(\vH^1_{-\tfrac
     52+\delta},
    \C^\mu),\\
\vS^* W(z) &= \vO( |{z}|^{\f 1 2 +\epsilon'} ) \text{ in } \vL(\vH^1_{-\tfrac
     32+\delta},
    \C^\mu),
   \end{align*} in turn valid due to
 (\ref{Wz00}) and  (\ref{Wz1}), respectively. (The interpolation
 requires the bounds with $\epsilon'=0$ only.)

\smallskip

\noindent \textbf {Calculation of $\vE_{-+,j}$, $j=1, 2$.} To compute
explicitly these leading terms, we distinguish between different situations
according to the spectral properties of the threshold zero.

\Step{ Case 1} \textit {Suppose zero is a resonance, but not an eigenvalue of $P$.}
In this case, $\mu =1$ or $2$ and
\[
\vE_{-+,1} = (\w{U \phi_i, G_1 U \phi_j})_{ 1\le i, j \le \mu}.
\]
If $\mu =1$, $\vE_{-+,1} = \w{U \phi_1, G_1 U \phi_1}_0$. Note that
$G_1 = \diag (G_{1,1}, G_{1,2})$ with $G_{1,i}$, $i=1,2$,  given by the rank-one operator
\be
G_{1,i} = \i\w{ g_i, \cdot}g_i,
\ee
where $g_i$  is the constant function in $x_i$ introduced in
\eqref{eq:gi}. Using \eqref{eq:ci} we then obtain
\[
 G_1 U \phi_1 =  \i (c_1(\phi_1) g_1, c_2(\phi_1)g_2)
\]
and
\begin{equation} \label{1a}
\vE_{-+,1} =   \i (|c_1(\phi_1)|^2 + |c_2(\phi_1)|^2) =   \i |c(\phi_1)|^2
\end{equation}
with $c(\phi_1) = (c_1(\phi_1), c_2(\phi_1)) \in \C^2$. If $\mu =2$, a similar computation gives
\begin{equation*}
G_1 U \phi_j =  \i ( c_1(\phi_j)g_1,  c_2(\phi_j)g_2); \quad j=1,2.
\end{equation*}
Therefore
\[
\w{U \phi_i, G_1 U \phi_j} = \i (\overline{ c_1(\phi_i)}  c_1(\phi_j) + \overline{ c_2(\phi_i)}  c_2(\phi_j)).
\]
It follows that
\begin{equation} \label{1b}
\vE_{-+,1} =   (\w{U \phi_i, G_1 U \phi_j} )_{ 1\le i, j \le 2}  =   \i   B_0
\end{equation}
where $B_0 =C_0^* C_0$ and
\be \label{e5.C0}
C_0 =\left(\begin{array}{cc}
c_1(\phi_1) & c_1(\phi_2) \\
c_2(\phi_1) & c_2(\phi_2)
\end{array}
\right)
\ee
Summing up, one obtains in Case 1
\be
\vE_{-+,1} = \i B_0
\ee
where $B_0$ is a $\mu\times \mu$ matrix given by
\be \label{vEpm11}
B_0 = |c(\phi_1)|^2 \quad \mbox{ if } \mu =1,\mand \quad B_0 = C_0^*C_0  \quad  \mbox{ if } \mu =2.
\ee
In both cases $B_0$ is invertible due to Theorem \ref{thm5.2}.  The explicit formula of $\vE_{-+,2}$ is not needed for the leading term of the resolvent expansion in Case 1.


\Step{ Case 2} \textit {Suppose zero is an eigenvalue, but not a resonance of $P$.} In this case, all $\phi_j$'s are eigenfunctions and
by Theorem \ref{thm5.2} on the characterization of resonance  states, one has  $W_1\vS =0$ which implies
\begin{equation}
\vE_{-+,1}= -\vS^*W_1 \vS =0.
\end{equation}
 Assume $\rho_0>3$, so that (\ref{vEpm1}) and \eqref{eq:redImp}
 apply. This means more explicitly that
\[
 \vE_{-+,2} =
  (  \w{ U\phi_i, (G_2U + G_0U_1) \phi_j})_{i, j =1, \dots, \mu}.
\]
Moreover, since $\phi_j \in L^2,  j= 1, \dots, \mu$,
we can check as in \cite{JK}  that
\begin{align}
  \label{eq:comphi}  \w{U\phi_i, G_2U \phi_j} = \w{\phi_i, \phi_j}.
\end{align}
 In fact, writing
 \[ G_2 = z^{-1} ( r_0(z)   - G_0 - \sqrt{z} G_1) + \vO(|z|^{\epsilon})
 \]
 for $\Im z >0$ and $z$ near zero, one obtains using the 1st resolvent
 equation that
 \[
 G_2U \phi_j = - r_0(z) \phi_j + \vO(|z|^{\epsilon}) U \phi_j,
 \]
and
\[
  \w{U\phi_i, G_2U \phi_j} = -\w{ r_0(\overline{z}) U\phi_i, \phi_j}
+ \vO(|z|^{\epsilon}).
\]
Similarly  $ r_0(\overline{z})U\phi_i = -\phi_i - \overline{z}
 r_0(\overline{z})\phi_i$, and by  taking the limit $z = \i \gamma \to 0$ with $\gamma \in\R_+$,
 we indeed obtain \eqref{eq:comphi}.

It is clear that $ \w{U\phi_i, G_0U_1 \phi_j} = - \w{\phi_i, U_1\phi_j}$. Therefore in Case 2, one has
\begin{equation} \label{2}
\vE_{-+,1} =0, \quad   \vE_{-+,2} = (\w{ \phi_i, (1-U_1)\phi_j})_{1 \le i,j \le \mu}.
\end{equation}
By Lemma \ref{lem5.5}, $ (\w{ \phi_i,  (1-U_1)\phi_j})_{1 \le i,j \le \mu}$ is a positive definite matrix.


\Step{ Case 3}  \textit  {Suppose zero is both a resonance and an eigenvalue of $P$.}
Let $\kappa$ be the multiplicity of zero resonance of $P$. Then $\kappa=1$ or $2$ and $\kappa<\mu$.
 For $ \kappa+1\le j \le \mu$, $\phi_j$ is an eigenfunction and therefore
\be \label{e5.70}
W_1 \phi_j = G_1U \phi_j =0, \quad j = \kappa + 1, \cdots, \mu.
\ee
$\vE_{-+,1} $ can be computed as in Case 1. One has in $\vM_{\mu\times \mu}(\C)$
\begin{equation}\label{31}
 \vE_{-+,1} =\left( \begin{array}{cccc}
   \i B_0 & 0 &\cdots &0 \\
  0 &   0  & \cdots & 0 \\
  \vdots & \vdots & \ddots & \vdots \\
  0 & 0 &\cdots  & 0
  \end{array} \right)
\end{equation}
where $B_0$ is the $\kappa\times\kappa$ matrix given by (\ref{vEpm11}) with $\mu$ replaced by $\kappa$.

Using the method of Case 2 and taking notice of (\ref{e5.70}),  we find that under the condition $\rho_0>3$
\begin{equation} \label{32}
\vE_{-+,2} = \left( \begin{array}{cc}
\vE_{11}^{(2)} & \vE_{12}^{(2)}\\
\vE_{21} ^{(2)}&   \vE_{22}^{(2)}
\end{array} \right),
\end{equation}
where
\begin{align*}
\vE_{11}^{(2)} &= \left(\w{ U \phi_i, (W_2 - W_1D_0W_1)\phi_j}\right)_{1\le i, j \le \kappa}, \\
\vE_{12}^{(2)} &= (\w{ U \phi_i, W_2\phi_j})_{1\le i \le \kappa, \;  \kappa+1 \le j \le \mu}, \\
\vE_{21}^{(2)} &=  (\w{ U \phi_i, W_2\phi_j})_{\kappa+1 \le i \le \mu, \;  1 \le j \le \kappa}, \\
\vE_{22}^{(2)} &=  (\w{\phi_i,  (1-U_1)\phi_j})_{\kappa+1 \le i,j \le \mu}.
\end{align*}

\medskip
If $0$ is an eigenvalue of $P$ we let $\Pi_0$ denote the spectral projection
  in $\vH$ onto the zero-eigenspace of this operator. The quantity
  $\Pi_0$ enters for
   Cases 2 and 3 below.

\begin{prop}\label{prop5.8}  The
  following asymptotics  as $z\to 0$ and $z\not\in [0,\infty)$ hold in $\vL(1,-s; 1,-s)$  for $s>1$ and
  close to $1$.

\medskip

\Step{Case 1} Suppose  that zero is an exceptional point
of the first kind. Then
\begin{align} \label{W1}
W(z)^{-1}= \f{\i}{\sqrt z} Q_0 +  \vO(|{z}|^{-\f 1 2 +\ep}).
\end{align}
Here
\be
Q_0 = \sum_{j=1}^\kappa \w{-U\psi_j, \cdot}\psi_j
\ee
 with $\psi_j\in \vK$ such that if $\kappa =1$, $\psi_1$ verifies the normalization condition (\ref{normalization1})
and if $\kappa =2$, $\psi_1$ and $\psi_2$ verify the normalization condition (\ref{normalization2}).
\medskip

\Step{Cases 2 and 3} Suppose  that zero is  an exceptional point of the
second or the third kind, respectively, and suppose   that $\rho_0 > 3$. Then
\begin{equation} \label{W2}
W(z)^{-1}= z^{-1} (\Pi_0 (1-U_1)\Pi_0)^{-1} \Pi_0 U  + \vO(|z|^{-1+\ep}).
\end{equation}

\end{prop}
\pf \subStep {Case 1} For any  $s\in( \tfrac 12, \rho_0-\tfrac 12)$ and  for some $\ep>0$
\begin{align*}
\vE(z) & =  D_0 + \vO(|z|^\ep), \quad \mbox{ in } \vL(1, -s; 1,-s),  \\
\vE_\pm(z) & =  \vE_{\pm, 0}  + \vO(|z|^\ep), \quad \mbox{ in
             }\vL(\C^\mu, \vH^1_{-s})
             \mbox{ or } \vL(
             \vH^1_{-s}, \C^\mu),  \\
\vE_{-+}(z) & =   \i\sqrt{z} B_0 + \vO(|z|^{\f 1 2 +\ep})
\end{align*}
where $B_0$ is given by (\ref{vEpm11}).
Note that $\vE_{+,0} = \vS$, $\vE_{-,0} = \vS^*$ and  $\vS \vS^* = Q$. It follows from \eqref{repW} that
\be
W(z)^{-1} = \f{\i}{\sqrt z} \vS B_0^{-1}\vS^* + \vO(|z|^{-\f 1 2 +\ep}).
\ee
If $\kappa =1$, $B_0 =|c(\phi_1)|^2$. Set
\be
\psi_1 = \f{1}{|c(\phi_1)|}\phi_1.
\ee
Then $\psi_1$ verifies (\ref{normalization1}) and
\[
\vS B_0^{-1}\vS^* = \f{1}{|c(\phi_1)|^2} \w{-U \phi_1, \cdot} \phi_1 = \w{-U\psi_1,\cdot} \psi_1 = Q_0.
\]
If $\kappa =2$, then $B_0 = C_0^*C_0$. Take $\psi_1, \psi_2 \in \vK$ such that
\be
\left( \begin{array}{c}
\psi_1  \\
\psi_2  \end{array} \right) =  \; ^tC_0^{-1} \left( \begin{array}{c}
\phi_1\\
\phi_2 \end{array} \right).
\ee
For $f\in \vH^1_{-s}$, set $(v_1, v_2) = (\w{-U\phi_1, f}, \w{-U\phi_2, f}) \in \C^2$.
One has
\[
\vS B_0^{-1}\vS^*f = (\phi_1, \phi_2) C_0^{-1} C_0^{*-1} \left( \begin{array}{c}
v_1  \\
v_2  \end{array} \right) = u_1 \psi_1 + u_2 \psi_2
\]
where
\[
\left( \begin{array}{c}
u_1\\
u_2 \end{array} \right) =  C_0^{*-1} \left( \begin{array}{c}
v_1  \\
v_2  \end{array} \right)
 =  \left( \begin{array}{c}
 \w{-U\psi_1, f}\\ \w{-U\psi_2, f}
 \end{array} \right).
\]
It follows that
\[
\vS B_0^{-1}\vS^*f =  \w{-U\psi_1, f} \psi_1 +  \w{-U\psi_2, f} \psi_2 = Q_0 f.
\]
To show that   $\psi_1$ and $\psi_2$  verify the normalization
condition (\ref{normalization2}), set $ C_0 = (c_{ij})_{1\le i,
  j \le 2}$, $^t C_0^{-1} = (d_{ij})_{1\le i, j \le 2}$ and $\psi_k =(\psi_{k1}, \psi_{k2})$. Then
\[
c_j(\psi_k) = \w{g_j, U_{j1}\psi_{k1} +  U_{j2}\psi_{k2} }=
\sum_{m=1}^2 d_{km} c_{jm} = \delta_{jk};\quad j,k  =1,2.
\]

Therefore $\psi_1$ and $\psi_2$ are resonance states verifying the
normalization condition (\ref{normalization2}), and
(\ref{W1}) is proved.

\subStep {Case 2}  By (\ref{2}) and Lemmas \ref{lem5.5}  and \ref{lem5.6}
\[
\vE_{-+}(z) = z M_1 + \vO(|z|^{1+\ep}),
\]
where the matrix $M_1 = (\w{ \phi_i, ( 1-U_1)\phi_j})_{1\le i, j\le \mu}$ is positive definite.
We deduce from \eqref{repW} that
\[
W(z)^{-1} =  -z^{-1} \vS M_1^{-1} \vS^* + \vO(|z|^{-1+\ep})).
\]
Let $\vS^\#$ be the adjoint of $\vS$ with respect to the scalar product
$\w{\cdot,\cdot}$ of $\vH$. Then $\vS^* = -\vS^\#U$ and   $M_1 =
\vS^\# (1-U_1) \vS$. The orthogonal projection  $\Pi_0$ of $\vH$ onto
the zero eigenspace of $P$ can be expressed in terms of $\vS$ and
$\vS^\#$ as
\begin{subequations}
 \begin{align}\label{eq:35P1}\Pi_0 = \vS( \vS \vS^\#)^{-1} \vS^\#
  \end{align}
 and  obeys
\begin{align}\label{eq:35P2}\Pi_0 \vS = \vS,\quad  \vS^\# \Pi_0 = \vS^\#
  \end{align}
\end{subequations}

 Letting  $T=\vS( \vS^\# (1-U_1)\vS)^{-1}\vS^\#$ we compute
 using \eqref{eq:35P1} and \eqref{eq:35P2}
 \begin{align*}
   -\vS M^{-1}_1 \vS^*  &=  TU,\\
 (\Pi_0( 1-U_1)\Pi_0) T&=  \vS( \vS \vS^\#)^{-1} \vS^\# = \Pi_0,\\
T(\Pi_0( 1-U_1)\Pi_0) &=  \vS( \vS \vS^\#)^{-1} \vS^\# = \Pi_0.
 \end{align*}
This leads to  the identity $T=  (\Pi_0( 1-U_1)\Pi_0)^{-1} \Pi_0$, and
 consequently that
\[
- \vS M^{-1}_1 \vS^* = TU = (\Pi_0( 1-U_1)\Pi_0)^{-1} \Pi_0 U,
\]
which proves (\ref{W2})

\subStep {Case 3} We use again \eqref{vEpm1}. We want to show $\vE_{-+}(z)^{-1}$
 exists and to calculate its leading term as $z\to 0$ and $z\not\in [0,\infty)$. To do this,  let $M(z) = \sqrt{z} \vE_{-+,1} +  z \vE_{-+,2} \in\vM_{\mu\times\mu}(\C)$.
 Decompose $M(z)$ into blocks:
\[
M(z) =\left( \begin{array}{cc}
\sqrt{z} M_{11}(z)  &  z M_{12}\\
z M_{21} &  z M_{22}  \end{array} \right)
\]
with
\begin{align*}
M_{11}(z) & =  \i B_0 +  \sqrt{z} \vE_{11}^{(2)}  \\
M_{ij} & =  \vE_{ij}^{(2)} \quad \mbox{ for }(ij) \neq (11).
\end{align*}
$M_{22}$ is positive definite by Lemma \ref{lem5.5}. The diagonal part of $M(z)$ is invertible and one has
\[
M(z)\diag ( (\sqrt{z} M_{11}(z))^{-1},  (z M_{22} )^{-1}) = \left( \begin{array}{cc}
1  &  a\\
b  &  1  \end{array} \right)
\]
where
\begin{equation*}
a = M_{12} M_{22}^{-1} ,  \quad b =b(z)= \sqrt{z} M_{21} M_{11}(z)^{-1}.
\end{equation*}
 Since $b=\vO(\sqrt{|z|})$,  the matrix  $\left( \begin{array}{cc}
1  &  a\\
                                                b  &  1  \end{array}
                                            \right)$ is invertible for
                                            $z\not\in [0,\infty)$ and $|z|$
                                            small, and an elementary calculation gives
\[
\left( \begin{array}{cc}
1  &  a\\
b  &  1  \end{array} \right)^{-1} = \left( \begin{array}{cc}
(1-ab)^{-1}  & - a (1-ba)^{-1}\\
-b (1-ab)^{-1}  &  (1 -ba)^{-1}  \end{array} \right).
\]
Consequently we obtain a formula for the inverse of $M(z)$
 for $z\not\in [0,\infty)$ with  $|z|$  small, and  from that we read off that
\begin{equation*}
zM(z)^{-1} =   \left( \begin{array}{cc}
{\sqrt{z}} M_{11}(z)^{-1}  & - {\sqrt{z}} M_{11}(z)^{-1} M_{12} M_{22}^{-1}\\
- {\sqrt{z}} M_{22}^{-1} M_{21} M_{11}(z)^{-1} &  M_{22}^{-1}  \end{array} \right)+\vO(\sqrt {\abs{z}}).
\end{equation*}

Recalling our choice of basis  $\{\phi_1, \dots, \phi_\mu\}$  in $\vK$
the map $\vS$ splits naturally as $\vS= \vS_r \oplus \vS_e : \C^\kappa \oplus \C^{\mu-\kappa} \to \vH^1_{-s}$ with the range of $\vS_r$ ($\vS_e$, respectively) included in the zero-resonance space  (the zero-eigenspace, respectively) of $P$. Similarly,  decompose $\vS^* = \vS_r^* \oplus \vS_e^*$. Then,
\begin{equation*}
\vS_r \vS_r^* = Q_r, \, \vS_r^* \vS_r = I_\kappa, \quad \vS_e \vS_e^* = Q_e,\, \vS_e^* \vS_e = I_{\mu-\kappa},
\end{equation*}
where $Q_r =\sum_{j=1}^\kappa \w{-U\phi_j, \cdot}\phi_j$ and  $Q_e =\sum_{j=\kappa+1}^\mu \w{-U\phi_j, \cdot}\phi_j$.
Since $\vE_{-+}(z) = M(z) + o(|z|)$,  we obtain that
\[
-\vE_+(z) \vE_{-+}(z)^{-1} \vE_-(z) = \frac{1}{z} \parb{ - \vS_e M_{22}^{-1} \vS_e^* + \vO(|z|^\ep)}.
\]
Since $M_{22} = (\w{ (1-U_1)\phi_i, \phi_j})_{\kappa+1 \le i, \; j \le \mu}$,  we can verify as in Case 2 that
\[
- \vS_e M_{22}^{-1} \vS_e^* = (\Pi_0 (1-U_1)\Pi_0)^{-1} \Pi_0 U.
\]
 Whence (\ref{W2}) is proved for Case 3 also.
\ef

Recall that
\begin{align}\label{eq:formE}
  E_{\vH}(\lambda_0 +z)^{-1}  = -W(z)^{-1}(P_0-z)^{-1},\mand G_0 U =
  -1\text{ on }\vK.
\end{align}
  The following result follows immediately from Proposition \ref{prop5.8}.

\begin{subequations}
\begin{prop}\label{prop5.9}
The following asymptotics hold in
  $\vL(-1,s; 1, -s)$  for any $s>1$ and  $z\not\in [0,\infty)$  with $|z|$ small.

\Step{Case 1} Suppose  zero is an exceptional point of the first kind.
Let $\kappa$ be  the multiplicity of zero resonance. If $\kappa=1$,
let $\psi_1$ be a resonance  state normalized according to
\eqref{normalization1}, and if $\kappa=2$, let $\psi_1$ and $\psi_2 $ be
resonance  states normalized according to \eqref{normalization2}. Then   for some $\ep>0$
\begin{equation} \label{e1}
E_{\vH}(\lambda_0 +z)^{-1}= -\frac{\i}{\sqrt{z}} \sum_{j=1}^\kappa \w{\psi_j, \cdot} \psi_j + \vO(|z|^{-\f 1 2 +\ep}).
\end{equation}

\Step{Cases 2 and 3} Suppose  that zero is  an exceptional point of the
second or the third kind, respectively, and suppose   that $\rho_0 >
3$. Then for some $\ep>0$
\begin{equation} \label{e2}
E_{\vH}(\lambda_0 +z)^{-1}= z^{-1} (\Pi_0 (1-U_1)\Pi_0)^{-1} \Pi_0   + \vO(|z|^{-1 +\ep}).
\end{equation}
\end{prop}
\end{subequations}
\pf
 We apply Proposition \ref{prop5.8}, \eqref{eq:formE} and
 the leading order expansion
 \begin{align}\label{eq:leadFree}
   \forall s>1:\quad (P_0-z)^{-1} = G_0 + \vO(|z|^\ep)\in \vL(-1,s; 1, -s).
 \end{align}
\ef

We are now  able to give the asymptotics of the resolvent of $N$-body operator $H$ near its first threshold $\lambda_0 =\Sigma_2$.
Let $R(\lambda_0 + z) =(H-\lambda_0-z)^{-1}$. By (\ref{rep5.9}),
\begin{equation}\label{eq:ResFor}
R(\lambda_0 + z) = R'(\lambda_0 + z) - (1-R'(\lambda_0 +z)H)S E_{\vH}(\lambda_0 +z)^{-1}S^* (1-H R'(\lambda_0+z)).
\end{equation}
Here $R'(\lambda_0 +z)$ is the reduced resolvent which under the condition \eqref{ass5.1} is holomorphic in $z$ in a neighborhood of zero.
The proof of Lemma  \ref{prop5.1} shows that $\psi$ is a zero-resonance
state of $P$ if and only if  $u = (1- R'(\lambda_0) H)S\psi$ is an
$\lambda_0$-resonance state of $H$ and $\phi$ can be recovered from $u$
by $\phi = T^* u$, where $T^* = (S^*S)^{-1}S^*$.

If $\lambda_0$ is a
resonance of $H$ its  multiplicity is denoted by $\kappa$.
 If $\lambda_0$ is an eigenvalue of $H$
 the spectral projection of $H$ onto the corresponding
eigenspace is denoted by $\Pi_H$. Define
\begin{equation*}
 S_I f = I_1\varphi_1 f_1 + I_2 \varphi_2 f_2, \quad f =(f_1,f_2) \in \vH^1_{-\infty}.
 \end{equation*}

 \begin{subequations}
\begin{thm}\label{thm5.10}   Assume \eqref{ass5.1}
and   \eqref{ass5.3}. The following asymptotics hold for $R(\lambda_0
+ z)$ as an operator from $H^{-1}_{s}$ to $H^{1}_{-s}$, $s>1$,  for $z\to
0$  and  $z\not\in [0,\infty)$, and $\epsilon=\epsilon(s)>0$.

\medskip

\Step{Case 0}  Suppose $\lambda_0$ is a regular point of $H$. Then   one has
\begin{equation} \label{asymRz0}
R(\lambda_0 + z) = R'(\lambda_0) +  (S-R'(\lambda_0)S_I)D_0 G_0 (S^*- S_I^*R'(\lambda_0)) + \vO(|z|^\ep).
\end{equation}

\Step{Case 1} Suppose $\lambda_0$ is an exceptional point of the first kind of $H$.  Then
\begin{equation} \label{asymRz1}
R(\lambda_0 + z)= \frac{\i}{\sqrt{z}} \sum_{j=1}^\kappa \w{u_j, \cdot} u_j + \vO(|z|^{-\f 1 2 +\ep}),
\end{equation}
where $u_j = S \psi_j - R'(\lambda_0) S_I\psi_j$
with $\psi_j$ given in Proposition \ref{prop5.9}. The
  quantities  $u_j$ are resonance states of $H$ obeying  \eqref{normalization1H} and
\eqref{normalization2H} for $\kappa=1$ and $\kappa=2$, respectively.
\medskip

\Step{Cases 2 and 3} Suppose  that $\lambda_0$  is  an exceptional point of the
second or the third kind, respectively, and suppose   that $\rho_0 >
3$. Then
\begin{equation} \label{asymRz2}
R(\lambda_0 + z)= - z^{-1} \Pi_H    + \vO(|z|^{-1+\ep}).
\end{equation}

\end{thm}
 \end{subequations}
\pf \subStep {Case 0} One has
\[
(1-R'(\lambda_0 +z)H)S = S-R'(\lambda_0)S_I + \vO(|z|)=E_+(\lambda_0) + \vO(|z|),
\]
in $\vL(\vH^1_{-s}, H^{1}_{-s})$ for $0 \le s  \le \tfrac 32$.
 Similarly, $S^* (1-H R'(\lambda_0+z)) = S^*- S_I^*R'(\lambda_0) + \vO(z)$ in $\vL(H^{-1}_{s}, \vH^{-1}_{s})$.
 One obtains then from \eqref{eq:formE}, (\ref{term1}),
 \eqref{eq:leadFree} and \eqref{eq:ResFor} that
\begin{equation}
R(\lambda_0 + z) = R'(\lambda_0) +  (S-R'(\lambda_0)S_I)D_0 G_0 (S^*- S_I^*R'(\lambda_0)) + \vO(|z|^{\ep})
\end{equation}
in $\vL(-1,s; 1,-s)$ for  $s>1 $,
  proving (\ref{asymRz0}). \\

\subStep {Case 1}  One obtains from \eqref{eq:ResFor} and Proposition
\ref{prop5.9} that
\begin{align*}
 R(\lambda_0 + z) = - E_+(\lambda_0)E_{\vH}(\lambda_0 + z)^{-1} E_-(\lambda_0) + \vO(|z|^{-\f 1 2 +\ep})
\end{align*}
in $\vL(-1, s; 1,-s)$ for $s>1$. Let $u_j =
E_+(\lambda_0)\psi_j$. Since $E_-(\lambda_0)^*= E_+(\lambda_0)$ it
follows from \eqref{e1}  that
\begin{equation*}
E_+(\lambda_0)E_{\vH}(\lambda_0 + z)^{-1} E_-(\lambda_0) = -\frac{\i}{\sqrt{z}} \sum_{j=1}^\kappa \w{u_j, \cdot} u_j +(|{z}|^{-\f 1 2 +\ep}),
\end{equation*}
   proving (\ref{asymRz1}).\\

\subStep {Cases 2 and  3} We can apply  (\ref{e2}) and
\eqref{eq:ResFor}  to obtain
\begin{align*}
  R(\lambda_0 + z) = -z^{-1} B + \vO(|z|^{-1+\ep}),
\end{align*}
where in terms of the spectral projection  $\Pi_0$ of $P$ (for the  eigenvalue zero)
\begin{equation*}
B= E_+(\lambda_0)(\Pi_0 (1-U_1)^{-1} \Pi_0)^{-1} \Pi_0 E_-(\lambda_0).
\end{equation*}

It remains to   check that $B$ is equal to the spectral
projection $\Pi_H$, cf. \cite{Wa2}. Introducing $\tau =
E_+(\lambda_0)\Pi_0$
 it follows that  $\tau^* = \Pi_0 E_-(\lambda_0)$,  and by using the
 properties  $R'(\lambda_0)S =  S^*R'(\lambda_0)=0$ and
 \eqref{eq:formU} we compute
\[
\tau^*\tau = \Pi_0 (  S^*(1-HR'(\lambda_0) )(1-R'(\lambda_0)H)S \Pi_0=  \Pi_0(1-U_1)\Pi_0.
\]
 Consequently
$B$ can be written as
\[
B = \tau (\tau^*\tau)^{-1}\tau^*.
\]
It follows that $B$ is an orthogonal projection with $\ran B \subset \ran
\tau$ and $\ker B = \ker \tau^*$, and therefore in fact $\ran B= \ran
\tau$. Since $\ran \tau$ is equal to  the $\lambda_0$-eigenspace of
$H$, we conclude $B= \Pi_H$ as wanted.
 \ef

For the case where $\rho_0=3$  and $\lambda_0$ is an eigenvalue of $H$
there is no information in Theorem \ref{thm5.10}.
Furthermore we only extracted the leading term asymptotics. Assuming
that all eigenfunctions of $H$ at  $\lambda_0$ have a certain
(weak) power decay we can obtain a resolvent expansion for that case
too, and that expansion will be to  second order without further assumptions,
see Theorem \ref{thm:rhoThree} stated below.

\begin{remark} For  possible higher order expansions the condition
  \eqref{ass5.4} is  needed and give limitations. For example it
  could be that \eqref{ass5.3} is fulfilled for Coulombic systems (for
  which $\rho=1$) with
  a big $\rho_0$ in some cases, but then the condition
  \eqref{ass5.4} would read $\rho_0\leq 4$ and  we dont see how to
  avoid  this restriction. More precisely, for example  we can not improve
  \eqref{eq:Uassp} to  be valid for  any $\rho_0> 4$ for such systems
  (not even with  $s=\rho_0/2$ only).
  This is rooted in the fact that we  do not know how to estimate the diagonal parts of the
  matrix $K(\lambda_0)$ better than $\vO(\inp{x_j}^{-4})$ for $\rho=1$.
 The main reason for not studying higher order resolvent expansions is
that
the conditions would be too strong to be interesting for  Coulombic systems.
  \end{remark}

\begin{thm}\label{thm:rhoThree} Assume \eqref{ass5.1},
   \eqref{ass5.3} and $\rho_0=3$. Suppose $\lambda_0$ is an
  eigenvalue of $H$ and that  $\ran
  \Pi_H\subset L^2_t$ for some $t>3/2$. Then the following
  asymptotics hold for $R(\lambda_0 + z)$ as an operator from
  $H^{-1}_{s}$ to $H^{1}_{-s}$, $s>1$,    for $z\to
0$  and  $z\not\in [0,\infty)$, and $\epsilon=\epsilon(s)>0$.
\begin{equation} \label{asymRz1aa}
R(\lambda_0 + z)=-z^{-1} \Pi_H + \frac{\i}{\sqrt{z}} \sum_{j=1}^\kappa \w{u_j, \cdot} u_j + \vO(|z|^{-\f 1 2 +\ep}).
\end{equation} Here the second  term on the right-hand side appears only if $\kappa\geq1$, that is if $\lambda_0$ is a
resonance of $H$,
 and in that case
the
states $u_j$ constitute a basis of   resonance states of $H$
(recall that either $\kappa=1$ and $\kappa=2$).  If on the other hand
$\lambda_0$ is not a
resonance of $H$, then $R(\lambda_0 + z)+z^{-1} \Pi_H $ has a limit in
norm as  $z\to
0$.
  \end{thm}

  One can show Theorem
  \ref{thm:rhoThree} by mimicking the proof of the assertion
  on Case 1 in  Theorem \ref{thm5.10} with $H$ replaced by $
  H_\sigma:=H-\sigma\Pi_H$, $\sigma> 0$. Note that $\lambda_0$ is not an
  eigenvalue of $ H_\sigma$. Letting
  $R_\sigma(\zeta)=(H_\sigma-\zeta)^{-1}$ we decompose
  \begin{align}\label{eq:sigmeFormR}
    R(\lambda_0 + z)=R_\sigma(\lambda_0 + z)-z^{-1} \Pi_H
    +(z+\sigma)^{-1} \Pi_H,
  \end{align} and conclude that it suffices to show that  for a fixed small $\sigma> 0$
\begin{equation} \label{asymRz1aa2}
R_\sigma(\lambda_0 + z)= \frac{\i}{\sqrt{z}} \sum_{j=1}^{\kappa} \w{u_j, \cdot} u_j + \vO(|z|^{-\f 1 2 +\ep}).
\end{equation}

This can largely  be shown
as before since  the last assumption in \eqref{ass5.1} is valid with
$H'$ replaced by $H_\sigma':=\Pi' H_\sigma\Pi'$ for  $\sigma> 0$
small (by an argument
of perturbation). We can largely   mimic the previous procedure, now for $H_\sigma$
rather than for $H$, cf.  Remark \ref{remark:The case
  lambda0insigma} \ref{item:25} which will be particularly relevant
for similar constructions in the next sections. However we need to argue
that the resonance  states $u_j$ in \eqref{asymRz1aa2}, which a priori are
resonance  states of $H_\sigma$,  in fact also are resonance  states
  of
$H$. For that it suffices in turn to note that $\Pi_Hu_j=0$, which
follows from the fact that
$-\sigma\Pi_Hu_j=\Pi_H(H_\sigma-\lambda_0)u_j=0$. If
$\kappa_\sigma$ denotes the dimension of the space of resonance
states
of $H_\sigma$, we see that $\kappa_\sigma\leq \kappa$. The converse
$\kappa_\sigma\geq \kappa$ follows by modifying any given base of resonance  states
of $H$ by projecting out the corresponding components in $\ran \Pi_H
$, yielding a base of resonance  states
of $H_\sigma$. Whence our use of the notation $\kappa$ in
\eqref{asymRz1aa2}, rather than the initially correct quantity $\kappa_\sigma$,
is justified.

Another comment relates to \eqref{eq5.2a}. The added term $-\sigma
\Pi_H$ is not local, and if we introduce the notation $U_\sigma$ and
$P_\sigma$ for the corresponding reduced quantities, the corresponding
version of \eqref{eq5.2a} does not follow the same way unless $\ran
  \Pi_H\subset L^2_t$ for some $t>2$. (But we only assume
this for  some $t> 3/2$, so if $t\in ( 3/2,2]$ a different
argument is needed.)  Nevertheless we have the following version
of \eqref{nonresonant1} for any  $v =
                                       (v_1, v_2)\in
                                       \vH^{1}_{(-1/2)^-} $ and $P_\sigma v
                                       =0$:
\begin{equation} \label{nonresonant11}
                                       v=0 \Longleftrightarrow
                                       \inp{ 1,
                                         U_{\sigma,j1}v_1 + U_{\sigma,j2}v_2}
                                       =0;\quad  j=1,2.
                                     \end{equation}
                                     This follows by noting that if
                                     the functionals to the right
                                     vanish on $v$, then the proof of
                                     \eqref{eq5.2a} shows that $v\in
                                     \vH^{1}_{(-1/2)^+} $. Then the
                                     corresponding $u\in
                                     H^{1}_{(-1/2)^+} $ (given by the
                                     inverse eigentransform), and this
                                     cannot hold unless $u\in L^2$ since
                                     by the above discussion
                                     $(H-\lambda_0)u=(H_\sigma-\lambda_0)u=0$,
                                     and therefore we can conclude that
                                     $u\in L^2$ using  Theorem
                                     \ref{thm:short-effect-potent}.
                                     Since $\lambda_0$ is not an
  eigenvalue of $ H_\sigma$  it follows that   $u=0$ and then in turn
  that $v=0$. With \eqref{nonresonant11}
                                     in place we can   mimic the proof of the assertion
  on Case 1 in  Theorem \ref{thm5.10} with $H$ replaced by $
  H_\sigma=H-\sigma\Pi_H$ and conclude \eqref{asymRz1aa2}.

\begin{remarks}\label{remark:resolv-asympt-nearSimp}
  \begin{enumerate} [1)]
  \item
We used, among others, the conditions $\lambda_0
  =\Sigma_2 \not\in \sigma_\pp(H')$ and $\vF_1\cap\vF_2 =\{0\}$.  The
  case $\lambda_0 \in \sigma_\pp(H')$ can be treated as in \cite{Wa2}
  when $\Sigma_2$ is a non-multiple two-cluster threshold. The same
  statements as those of Theorem \ref{thm5.10} remain true, but their
  proof requires a more complicated Grushin reduction, cf. Sections
  \ref{sec:The case where lambda0insigma} and  \ref{sec:CoulRellich}. If
  $\vF_1\cap\vF_2 \neq \{0\}$, one can still reduce the spectral
  analysis of $H$ at $\lambda_0$ to a one-body problem, cf.  Sections
  \ref{sec:The case when the condition {ass2}} and
  \ref{sec:CoulRellich}. The conditions $\dim \bX_j =3$ and
  $\rho_0\geq 3$ of \eqref{ass5.1} and \eqref{ass5.3}, as well as
  \eqref{ass5.4}, simplify the resolvent expansion at the threshold
  $\lambda_0$ and fit well with the physics models of Sections
  \ref{$N$-body Schr\"odinger operators} and \ref{$N$-body
    Schr\"odinger operators with infinite mass nuclei}.
\item \label{item:decayEig} The very last assertion of Theorem \ref{thm:rhoThree} only needs  $\ran
  \Pi_H\subset L^2_t$ for some $t> 1$ (rather than this decay  for
  some $t>3/2$, as  required in Theorem \ref{thm:rhoThree}) due to the fact that $\lambda_0$, under the given hypotheses, is neither an eigenvalue
  nor a resonance of   $H_\sigma$.
\item \label{item:decayEig2}
The cases
  where $\dim \bX_j \neq 3$ or the effective potential decays slowly
  or has a critically decaying part were studied in Chapter 4 and we
  obtained some threshold spectral properties of $H$.  One may for these
  cases try to combine the existing results for one-body problems and
  the reduction made in Chapter 2 to establish the asymptotics of the
  resolvent of $H$ near $\lambda_0$. Although leaving this issue
  partly as an open question to
  the interested reader, see however Section
  \ref{sec:resolv-asympt-physics models near}, let us remark that the effectively slowly
  decaying potential cases corresponding to  Subsections
  \ref{subsec:negat-effect-potent},
  \ref{subsec:positive-effect-potent},
  \ref{subsec:negat-effect-potents} and \ref{subsec:psdep} are
  relatively easy to treat due to the fact that the threshold is not a
  resonance and eigenfunctions (if existing) have arbitrary polynomial
  decay (cf. Theorem \ref{thm:physical-modelsRell} \ref{item:PRel1} in the context of
  physics models).
We make use of this simplicity in Chapter
  \ref{Applications}, see  Remark \ref{remark:formRes}. This is  in
  the context of a possible higher two-cluster
  threshold subject to  conditions of Subsections
  \ref{subsec:negat-effect-potent},
  \ref{subsec:positive-effect-potent},
  \ref{subsec:negat-effect-potents} and \ref{subsec:psdep}. However we
  do treat a special case with a critical  decaying part of  the
  effective potential, see Subsection \ref{subsec:Some better
  arguments}. Using the theory of the present chapter we treat in
Subsections   \ref{subsec: caseIII} and \ref{subsec: An example of
  transmission} the fastly decaying
case for the physics models with the occurrence  of a resonance  at a possible higher two-cluster
  threshold.
\end{enumerate}
\end{remarks}

\subsection{Resolvent asymptotics near higher two-cluster thresholds}\label{sec:resolv-asympt-nearHIGHER}

Assume (\ref{ass5.1b}) and  (\ref{ass5.2b})--(\ref{ass5.4b})
 as in Subsection  \ref{sec:case-lambda_0sub}. We shall tacitly use
 the notation from that  subsection. Let  $P(z) = -
 E_{\vH}(\lambda_0 + z)$ for $z$ near zero and $\Im z\neq 0$. Then
$P(z)$ can be written as
\begin{equation}
P(z) = P_0 + U(z) -z,
\end{equation}
where   $P_0 = -\Delta_{x_{a_0}} {\textbf 1}_m $  and $U(z) =  S^* I_0
S - S^*I_0R'(\lambda_0 + z)I_0S$. Here ${\textbf 1}_m$ denotes the
identity matrix in $\vM_{m\times m}(\C)$.
Applying Theorem 3.18, one has  for  $\pm \Im z >0$ and  $j\in \N_0$
with $j<\rho_1$
\begin{align} \label{Upm}
U(z) &=    U^\pm  + z U_1^\pm  + \cdots + z^j U_j^\pm + \vO(|z|^{j+\ep}),
\end{align}
as bounded operators from $\vH^{1}_{-s}$ to $\vH^{-1}_{s}$ for any $  s <
\rho_1-j $ and for some $\ep>0$ depending on $s$ and
$\rho$. Furthermore
\begin{align*}
U^\pm &= S^* I_0 S - S^*I_0R'(\lambda_0 \pm \i0)I_0S, \\
U_j^\pm &= - S^*I_0(R'(\lambda_0 \pm \i0))^{j+1}I_0S; \quad j \ge 1.
\end{align*} See the proof of Lemma \ref{lemma:Wresolv-asympt-near-1}
for an elaboration for  the relevant cases (we shall only need
\eqref{Upm} for $j=0$ and
$j= 1$).

Recall that $\vK = \ker_{\vH^{1}_{ -s}}(1+ G_0U^\pm)$ for
  $s\in(\tfrac 12,  \tfrac 32)$, but that this space is independent of
  such $s$ as well as of the sign $\pm$.
\begin{lemma}\label{lem5.15} The quadratic form $q(\cdot) = \w{ (1-
    U_1^\pm)\cdot, \cdot}_0$ defined on $\vK\cap \vH$ is positive
  definite and independent of sign $\pm$.
\end{lemma}
\begin{proof} For $v \in \vK\cap\vH$ one has $\Pi'I_0 Sv \in
  \vH^{-1}_{(3/2)^+}$, allowing us to conclude that  $R'(\lambda_0 \pm
  \i0)^2 I_0Sv\in \vH^{1}_{(-3/2)^-}$. Hence $q$ is
  well-defined. Now  $R'(\lambda_0 + \i0) I_0
  Sv = R'(\lambda_0 - \i0) I_0 Sv $
  and $R'(\lambda_0 \pm \i0) I_0 Sv \in L^2(\bX)$ (cf. Theorem \ref
  {thm:short-effect-potent2}, Theorem \ref{prop5.5} (a)  and the part of the
  proof of Theorem
\ref{prop5.5} based on microlocal resolvent estimates). These facts
allow us to compute
\begin{align*}
 q(v) &=  \w{ (1+ S^* I_0R'\lambda_0 + \i0)^2 I_0 Sv, v} \\
 & = \|v\|^2 +\w{R'(\lambda_0 + \i0) I_0 Sv, R'(\lambda_0 - \i0) I_0 Sv} \\
  &=  \|v\|^2 + \|R'(\lambda_0 + \i0) I_0 Sv\|^2 \ge \|v\|^2,
\end{align*}
proving that  $q(\cdot)$ is positive definite on $\vK\cap\vH$.
\end{proof}

We shall  study the asymptotics of the resolvent $R(\lambda_0+z)
=(H-\lambda_0-z)^{-1}$ as $z\to 0$ and $\pm \Im z > 0$.    Set \be r_0(z) = (P_0
-z)^{-1}; \quad z\not\in [0,\infty).  \ee To obtain the asymptotic expansion
of $P(z)^{-1}$ for $z\not\in \R$ and $z$ near zero, we use as in
Section \ref{sec:resolv-asympt-nearLOWEST}  the resolvent equation
\be P(z)^{-1} = ( 1 + r_0(z) U(z))^{-1} r_0(z) \ee
and apply the Grushin method to study $( 1 + r_0(z) U(z))^{-1}$. For
simplicity  we shall mostly consider  the case $\Im z>0$ only.\\

Let $N \ge 1$ and $s > N + \f 1 2$. Similarly to the (partly scalar) case
studied in Section \ref{sec:resolv-asympt-nearLOWEST}, the free resolvent $r_0(z)$ can be expanded
$\vL(-1,s; 1,-s)$  as
\begin{equation}\label{eq:freeres}
  r_0(z) = G_0 + \sqrt{z} G_1 + \cdots +z^{\f N 2} G_N  + \vO(|z|^{\f N 2 +\ep}),
\end{equation}
for some $\ep >0$ depending on $N$ and $s$, where $G_j $ is the
$m\times m$ diagonal matrix (with operator-valued entries) whose
integral kernel is given by \bea \tfrac{\i^j}{4\pi} |x_{a_0}
-y_{a_0}|^{j-1}{\textbf 1}_m;\quad j\in{\N_0}.  \eea

In addition to the higher order expansions
\eqref{eq:freeres} we   record the zero'th order expansion \eqref{eq:basicHo} (with an
obvious change of interpretation).

  Since $U^+ = U^-$ on $\vK$ it follows that
$\w{-U^+ \cdot, \cdot}$ is a quadratic form   on $\vK$. By the
identity   $\w{-U^+ \cdot, \cdot} = \w{P_0 \cdot, \cdot}$ and an
integration by parts it follows that $\w{-U^+\cdot, \cdot}$  is a
positive  quadratic form.
Hence there exists a basis  $\{\phi_1, \dots, \phi_\mu\}$ of $\vK$
which is orthonormal with respect to $\w{-U^+ \cdot, \cdot}$. If zero
is a resonance of $P^+$, we denote by $\kappa$ its multiplicity. We
use the convention that  the functions $\phi_j$ for $1\le j \le \kappa$
are  resonance states  of $P^+$ while    the other $\phi_j$'s are eigenfunctions of $P^+$. \\

With the above notation, we can apply the Grushin method  to obtain
asymptotic expansions  of $W(z)^{-1}$ as $z\to 0$, $\Im z >0$, where
\begin{equation*}
W(z) := 1 + r_0(z)U(z).
\end{equation*}

 The leading order expansion  of
$W(z)$ is  given by
\begin{equation} \label{Wz0b}
W(z) = 1+ G_0 U^+ + \vO(|z|^{\epsilon})\in \vL(1, -s'; 1, -s),
\end{equation}
for $\Im z >0$ and $\tfrac 12
 <s, s' <\rho_1$ and for some $\epsilon=\epsilon(s,s')>0$. In particular this holds under the minimum
 conditions $\rho=1$ and  $\rho_1=\tfrac 32$, see  Lemma
 \ref{lemma:Wresolv-asympt-near-1} \ref{item:22} below and the
 outlined proof.

A higher order  expansion of
 $W(z)$  needs a stronger condition than
 $\rho_1=\tfrac 32$ and therefore  also $\rho>1$. With such condition we can apply Theorem
\ref{thm:powers}  on higher  powers of the reduced resolvent and
obtain the following extension of \eqref{Wz0b}. The assertions
\ref{item:20} and \ref{item:21} clearly exclude the interesting case
$\rho=1$. Nevertheless we will later show a higher order  expansion of
 $W(z)\phi$ when applied to vectors $\phi\in \vK$, even when  $\rho=1$,
 see the proof of Lemma \ref{prop5.17}. This will in fact yield a non-trivial
 resolvent expansion in the exceptional  case of the first kind  when  $\rho=1$.

\begin{lemma} \label{lemma:Wresolv-asympt-near-1} Suppose  \eqref{ass5.1b} and
  \eqref{ass5.2b}--\eqref{ass5.4b}. In terms of the quantities
\begin{equation} \label{Wjb}
W_0  =  1+ G_0 U^+, \quad  W_1 =  G_1 U^+, \quad W_2  =G_2 U^+ +G_0U^+_1,
\end{equation}
the following expansions in $z$ (with  $\Im z >0$)  hold for some $\epsilon=\epsilon(s,s')>0$.
\begin{subequations}
\begin{enumerate}[\rm(a)]
\item \label{item:22} For $ s, s' \in (\f{1}{2}, \rho_1)$, $s\geq s'$,
  \begin{align}
    \label{eq:repeat}W(z) = W_0+ \vO(|z|^{\epsilon})\in \vL(1, -s'; 1, -s).
  \end{align}
\item \label{item:20} If $\rho_1 > \f 3 2$, then for $ s, s' \in
  (\f{3}{2}, \rho_1)$ with  $s\geq s'$,
\begin{align} \label{Wz1b}
W(z) = W_0 + \sqrt{z} W_1 +  \vO(|z|^{\f 1 2 +\epsilon})\in \vL(1, -s'; 1, -s).
\end{align}
\item \label{item:21} If  $\rho_1> \f 5 2$, then for  $ s \in
  (\f{5}{2}, \rho_1)$, $s\geq s'$  and $  s' \in  (\f{3}{2}, \rho_1-1)$
\begin{align} \label{Wz2b}
W(z) =  W_0 + \sqrt{z} W_1 + z W_2 +  \vO(|z|^{1+\epsilon})\in \vL(1, -s'; 1, -s).
\end{align}
\end{enumerate}
  \end{subequations}
\end{lemma}
\pf
By Theorem \ref{thm:powers}  there exists  for any  $j\in\N$ and $\ep
>0$  a constant $C>0$ such that
\[
\|\w{x}^{\f 1 2 -j-\ep} R'(\lambda_0 +z)^j\w{x}^{\f 1 2 -j-\ep}\| \le C,
\] and therefore
\[ \|\w{x}^{-\f 3 2 -\ep}   \parb{R'(\lambda_0 +z) - R'(\lambda_0 + \i0) } \w{x}^{-\f 3 2 -\ep}\| \le C |z|
\]
for $z$ near zero with  $\Im z >0$. Using complex
interpolation we then obtain  the following H\"older estimate for the reduced
resolvent for  $0<r'<r\leq 1$,
\begin{subequations}
\be \label{holresol}
 \norm{\w{x}^{-\f 1 2 -r}   \parb{R'(\lambda_0 +z) - R'(\lambda_0 + \i0) } \w{x}^{-\f 1 2 -r}} \le C_{r, r'} |z|^{r'}, \quad \Im z >0.
\ee
 In particular this bound yield the expansion (\ref{Upm}) with
 $j=0$.

By a similar  argument it follows that  for  $0<r'<r\leq 1$ and $\Im z >0$,
\be \label{holresol2}
 \norm{\w{x}^{-\f 3 2 -r}   \parb{R'(\lambda_0 +z) - R'(\lambda_0 + \i0)-zR'(\lambda_0 + \i0)^2 } \w{x}^{-\f 3 2 -r}} \le C_{r, r'} |z|^{1+r'},
\ee implying the expansion (\ref{Upm}) with
 $j=1$.
 \end{subequations}

We only prove \ref{item:21} since the pattern of proof of
\eqref{eq:repeat} and \eqref{Wz2b} is the same. So suppose
 $\rho_1> \f 5 2$.  Applying  \eqref{eq:freeres} (with $N=2$)  we first
obtain  that
\[
W(z) = 1 + (G_0 + \sqrt{z} G_1 + zG_2)U(z) + \vO(|z|^{1+\ep})
\]
in $\vL(1, -s'; 1, -s)$ with $ s, s' \in (\f{5}{2}, \rho_1)$. Note
here
that $U(z)\in \vL(1, -s'; -1, s')$ with a uniform bound, cf. Theorems
\ref{thmlapBnd} and \ref{thm:powers} . Making
(again) use of the boundedness of
$G_2$ and the expansion (\ref{Upm}) with $j=0$ we   obtain
\[
G_2U(z) = G_2 U^+ +  \vO(|z|^{\ep})\text{ in }\vL(1, -s'; 1, -s).
\]

It follows from \eqref{holresol} that for   $ s \in
  (\f{5}{2}, \rho_1)$ and $  s' \in  (2, \rho_1-\f{1}{2})$
\bea
\lefteqn{\| \w{x_{a_0}}^{-s}  G_1 (U(z) -U^+) \w{x_{a_0}}^{s'} \|}  \nonumber \\
& \le & C \| \w{x_{a_0}}^{s-1} S^* I_0(R'(\lambda_0 +z) - R'(\lambda_0 + \i0) ) I_0S  \w{x_{a_0}}^{s'} \| \\
& \le & C'   \| \w{x}^{-\rho_1 -\f 1 2 + s-1} (R'(\lambda_0 +z) - R'(\lambda_0 + \i0) ) \w{x}^{-\rho_1 - \f 1 2  + s'} \| \nonumber  \\
& \le & C'' |z|^{\f 1 2 +\delta}; \nonumber
\eea
 this is for some $\delta>0$.

It remains to show that
\[
G_0U(z) = G_0U^+ + z G_0 U_1^+ +   \vO(|z|^{1+\ep})\text{ in }\vL(1, -s'; 1, -s).
\]
 This  expansion for   $ s \in
  (\f{5}{2}, \rho_1)$ and $  s' \in  (\f{3}{2}, \rho_1-1)$  follows
  from   (\ref{Upm}) with $j=1$ (justified by \eqref{holresol2}). We   conclude
\ref{item:21}.
\ef

Let $\vS: \C^\mu \to \vH^1_{-s}$, $s\in(\f{1}{2},\rho_1)$,  be defined  by
\[
\vS c
= \sum_{j=1}^\mu c_j\phi_j,\;\;  c = ( c_1, \dots,  c_\mu) \in \C^\mu,
\]
and $\vS^* : \vH^1_{-s} \to \C^\mu$ by
\[
\vS^*f=( \w{-U^+\phi_1,f} \dots, \w{-U^+\phi_\mu, f}),\quad  f\in \vH^1_{-s}.
\]
Let $Q:  \vH^1_{-s}\to \vH^1_{-s}$ be given by
\[
Qf=\sum_{j=1}^\mu \w{-U^+ \phi_j,f}\phi_j.
\]
One can then prove, cf.  Section \ref{sec:resolv-asympt-nearLOWEST},  that $Q$ is a projection from $\vH^1_{-s}$ onto $\vK$,
 $Q'=1-Q$ is a projection from $\vH^1_{-s}$ onto $  \ran (1+G_0U^+)$
 and that
 \begin{equation}
 D^+_0 = (Q'(1+ G_0U^+)Q')^{-1}Q'
\end{equation}
exists and is continuous on $\vH^1_{-s}$, $s\in(\f{1}{2},\rho_1)$. By an argument
of perturbation, it follows that for $|z|$ small enough and $\Im z
>0$,  $Q'W(z)Q'$ is invertible on $ \ran (1+G_0U^+)$ with continuous inverse. Let
 \[
  D(z) = (Q'W(z)Q')^{-1} Q'.
\]
Then according to the  varying conditions on $\rho_1$ in Lemma
\ref{lemma:Wresolv-asympt-near-1}  one has the following expansions in
$\vL(1,-s; 1,-s)$:
\begin{subequations}
\begin{align}
\label{term1b}
D(z)   &=    D_0 +  \vO(|z|^{\epsilon}),  \mbox{ if }   s \in (\tfrac{1}{2}, \rho_1). \\
D(z)   &=    D_0 + \sqrt{z} D_1 +  \vO(|z|^{\f{1}{2} +\epsilon_1}),
         \mbox{ if } \rho_1 > \tfrac  3 2
\mand  s \in (\tfrac {3}{2}, \rho_1).\label{term2b}
\end{align}
\end{subequations}

Here
\begin{align*}
D_0 = D_0^+\mand  D_1  = -D_0W_1D_0.
\end{align*}

Using the operator $D(z)$, we can establish the following representation formula for $W(z)^{-1}$ for $\Im z >0$:
\begin{equation} \label{repWb}
W(z)^{-1} = \vE(z) - \vE_+(z) \vE_{-+}(z)^{-1} \vE_-(z),
\end{equation}
where
\begin{gather*}
\vE(z)  =  D(z),\quad  \vE_+(z) = \vS - D(z)W(z)\vS, \\
\vE_-(z )  =  \vS^* -\vS^*W(z)D(z),\quad   \vE_{-+}(z) =  \vS^*(-W(z) +  W(z)D(z) W(z))\vS.
\end{gather*}
The operators $\vE_{\pm}(z)$ and $\vE_{-+}(z)$ can be expanded similarly as $D(z)$,
\begin{align*}
\vE_\pm(z)& = \vE_{\pm,0} + \sqrt{z} \vE_{\pm,1}  + z \vE_{\pm,2} + \dots \\
\vE_{-+}(z)& = \vE_{-+,0} + \sqrt{z} \vE_{-+,1}   + z \vE_{-+,2} + \cdots
\end{align*}
More precisely one has the following two lemmas  on these expansions.

\begin{lemma}\label{lem5.16}  Suppose  \eqref{ass5.1b} and  \eqref{ass5.2b}--\eqref{ass5.4b}.  One has in $\vL(\C^\mu, \vH^1_{-s})$ (for the $ +
  $ case) or $\vL( \vH^1_{-s}, \C^\mu)$ (for the $ - $ case):
  \begin{subequations}
 \begin{align}
\vE_{\pm}(z)   &=    \vE_{\pm,0}  +  \vO(|z|^\ep),  \mbox{ if } s \in (\tfrac{1}{2}, \rho_1).  \label{term1pmb}\\
\vE_{\pm}(z)   &=    \vE_{\pm,0} + \sqrt{z} \vE_{\pm, 1 } +
                 \vO(|{z}|^{\f 1 2 +\ep}),  \mbox{ if } \rho_1 >
                 \tfrac 3 2\mand  s \in (\tfrac {3}{2}, \rho_1).   \label{term2pmb}
\end{align}
  \end{subequations}
Here
\begin{align*}
\vE_{+, 0} = \vS, \quad & \vE_{-,0} = \vS^*, \\
\vE_{+,1}=  - D_0W_1\vS, \quad & \vE_{-,1}=  - \vS^*W_1D_0.
\end{align*}
\end{lemma}

The proof of Lemma \ref{lem5.16} is the same as  that for the
corresponding assertions in Lemma
\ref{lem5.6} and will not be repeated here. For the leading order
resolvent expansion only (\ref{term1pmb}) is needed.

\begin{subequations}
Formally one can expand the matrix
\begin{align*}
  \vE_{-+}(z) = \vS^*(-W(z) + W(z) D(z) W(z))\vS \in \vM_{\mu\times \mu}(\C)
\end{align*} as
\begin{align}\label{eq:E11}
  \vE_{-+}(z) =  \sqrt{z}  \vE_{-+, 1}(z) + z \vE_{-+, 2} +  \vO(|z|^{1+\ep}),
\end{align}
where
  \begin{align}\label{eq:E12}
    \vE_{-+, 1} & =  - \vS^*W_1\vS,
\\
\vE_{-+, 2}  &= -\vS^*(W_2  - W_1D_0 W_1 )\vS.\label{eq:E13}
  \end{align}
  \end{subequations}
The following lemma   gives a precise meaning to this expansion. Note
that  \eqref{ass5.4b} is not imposed. Furthermore none of the
assertions   (\ref {Wz2b}) or   (\ref {term2b}) is used,  and (\ref
{Wz1b}) is  not  used neither for treating the  first case of the lemma
(consequently the stated result  applies for $\rho=1$ and $\rho_1=\f 3
2$). Note that consequently the strong condition $\rho_1>\tfrac 52$ of
(\ref {Wz2b}) is not relevant for the lemma,  although superficially
this condition might
appear necessary for justifying \eqref{eq:E11}--\eqref{eq:E13}; we can
do without (\ref {Wz2b}).

\begin{lemma} \label{prop5.17} Suppose  \eqref{ass5.1b},
  \eqref{ass5.2b}  and
 \eqref{ass5.3b}.

\Step {Case 1} If zero is an exceptional point of the first kind of $P^+$, then
  \begin{subequations}
\be \label{vEpm1b}
\vE_{-+}(z) =  \sqrt{z}  \vE_{-+, 1} + \vO(|{z}|^{\f 1 2 +\ep})
\ee
for $z$ near zero and $\Im z >0$, where the matrix $\vE_{-+, 1}$ is given by
\be
\vE_{-+, 1} = \i C_0^*C_0
\ee
with the  $j$'th column  of $C_0\in\vM_{m\times
  \mu}(\C)$, $1 \le j \le \mu$,  specified as
\be \label{cphi}
c(\phi_j) =\f{1}{2\sqrt{\pi}}(\w{1, \sum_{k=1}^m U^+_{1k}\phi_{j,k}}_0, \cdots,  \w{1, \sum_{k=1}^m U^+_{mk}\phi_{j,k}}_0) \in \C^m.
\ee
 Here $\phi_j$ is denoted as $\phi_j =(\phi_{j,1}, \cdots, \phi_{j,m}) \in \vK$.
   \end{subequations}

\Step {Case 2} Assume zero is an exceptional point of the second kind of
  $P^+$, $\rho >1$ and $\rho_1>\f 3 2$. Then
\begin{subequations}
\be \label{vEpm2b}
\vE_{-+}(z) =  z  \vE_{-+, 2} + \vO(|{z}|^{1 +\ep})
\ee
for $z$ near zero and $\Im z >0$, where
\be
\vE_{-+, 2} =  ( \w{(1- U_1^+)\phi_i, \phi_j}_0)_{1\le i, j\le \mu}.
\ee
\end{subequations}
 \Step {Case 3}  Assume zero is an exceptional point of the third kind of
   $P^+$, $\rho>1$  and $\rho_1>\f 3 2$. Let $1 \le \kappa <\mu$ denote
   the multiplicity of zero resonance of $P^+$.  Then
\begin{subequations}
\be \label{vEpm3b}
\vE_{-+}(z) = \sqrt{z} \vE_{-+,1} +  z  \vE_{-+, 2} + \vO(|{z}|^{1 +\ep})
\ee
for $z$ near zero and $\Im z >0$, where $\vE_{-+,1}$ and $\vE_{-+,2}$ can be written  as block matrices
\be \label{vEpm3b22}
\vE_{-+,1} = \left( \begin{array}{cc}
\i C_0^*C_0 & 0 \\
0 & 0
\end{array}\right)\quad \mand \quad \vE_{-+, 2} =  \left( \begin{array}{cc}
\vE_{11}^{(2)} & \vE_{12}^{(2)} \\
\vE_{21}^{(2)} & \vE_{22}^{(2)}
\end{array}\right)
\ee
with the  $j$'th column  of $C_0\in\vM_{m\times
  \kappa}(\C)$, $1 \le j \le \kappa$,  equal to
$c(\phi_j)$ and with  $\vE_{22}^{(2)} =  ( \w{(1-
  U_1^+)\phi_i, \phi_j}_0)_{\kappa +1\le i, j\le \mu}$.
\end{subequations}
\end{lemma}
\pf
\subStep{Case 1}  We start by estimating the contribution from the
second term of the representation
\begin{align*}
  \vE_{-+}(z) =(\w{U^+\phi_i,W(z)\phi_j - W(z)
   D(z)W(z) \phi_j})_{1 \le i, j \le \mu}.
\end{align*}
 Using the identity
 $\phi_j=-G_0U^+\phi_j$, $j\leq \mu$, one can estimate
$\|W(z) \phi_j\|_{\vH^1_{-s}}$ for any $s \in (1, \rho_1)$ as
\[
\|W(z) \phi_j\|_{\vH^1_{-s}}  \le \|(r_0(z)-G_0) U(z) \phi_j\|_{\vH^1_{-s}} + \|G_0 (U(z)-U^+) \phi_j\|_{\vH^1_{-s}}.
\]
By the (uniform) bound $U(z) \phi_j\in \vH^{-1}_{(3/2)^-}$ and a free
resolvent bound, possibly deduced by interpolating  \eqref{eq:basicHo} and \eqref{eq:freeres} (with $N=1$),
\[
\|(r_0(z)-G_0) U(z) \phi_j\|_{\vH^1_{-s}} \le C |z|^{\f 1 4 +\ep}.
\]
 From this point we  fix any $s\in (1,\tfrac
54)$. Then $1+\rho -s > \tfrac 34$ and by \eqref{holresol}
\begin{align*}
  \|G_0 (U(z)-U^+) \phi_j\|_{\vH^1_{-s}} &\le C_1  \|\w{x}^{-1-\rho
    +s} (R'(\lambda_0 +z) - R'(\lambda_0 +\i0) ) \w{x}^{-1-\rho +s} \|
  \\&\le C_2
  |z|^{\f 1 4 +\ep}.
\end{align*}
  This proves that
\begin{subequations}
\begin{align}
  \label{eq:dircTas}\|W(z) \phi_j\|_{\vH^1_{-s}} \le C |z|^{\f 1 4 +\ep}.
\end{align}
 Similarly we can prove that
\begin{align}
  \label{eq:estadj}\|W(z)^* U^+\phi_i\|_{\vH^{-1}_{s}} \le C |z|^{\f 1 4 +\ep}.
\end{align}
\end{subequations}

Due to \eqref{term1b} the operator  $D(z)$ is uniformly bounded in
$\vL(\vH^{1}_{ -s}, \vH^1_{-s})$, and  we conclude using
\eqref{eq:dircTas} and \eqref{eq:estadj} that
\[
\vE_{-+}(z)  =  (\w{U^+\phi_i, W(z) \phi_j})_{1 \le i, j \le \mu} + \vO(|z|^{\f 1 2 +2\ep}).
\]

Next we simplify the first term to the right.
It follows by  similar arguments that it is expanded as
\begin{align*}
\w{U^+ \phi_i, W(z)\phi_j}_0
 = (\w{U^+ \phi_i, \phi_j+r_0(z)U^+\phi_j})_{1 \le i, j \le \mu}  + \vO(|z|^{\f 1 2 +\ep}),
\end{align*}
Since $U^+ = U^-$ on $\vK$, it follows from Corollary
\ref{cor:microlocal-bounds} that  $U^+ \phi_i,U^+ \phi_j\in \vH^{-1}_{(3/2)^+}$. Then we obtain from \eqref{eq:freeres} (with
$N=1$) that
\begin{align*}
\w{U^+ \phi_i, W(z)\phi_j}_0
 = (\w{U^+ \phi_i, \sqrt z G_1U^+\phi_j})_{1 \le i, j \le \mu}  + \vO(|z|^{\f 1 2 +\ep}),
\end{align*}
Using the explicit formula for the integral kernel of $G_1$, one obtains
\[
\w{U^+\phi_i, G_1 U^+\phi_j}_0 = \i( c(\phi_i),  c(\phi_j))
\]
 where $c(\phi_j) \in \C^m$ is given  by (\ref{cphi}). This shows $
 \vE_{-+,1} = \i C_0^* C_0$, and we have proven the assertion for  Case
 1.

\subStep{Case 2}   If zero is an exceptional point of the second kind, then $\phi_j \in \vH$ and Theorem \ref{prop5.5} shows $W_1\phi_j= G_1U^+\phi_j=0$ for $1 \le j \le \mu$.
Therefore $\vE_{-+,1} = -\vS^* W_1 \vS =0$.

  By possibly making  $\rho_1$ smaller we can
assume that \eqref{ass5.4b} is fulfilled and therefore that
(\ref{term1b}) is valid for some $s> \frac 32$. We invoke the implied
uniform boundedness of $D(z)\in\vL(1,-s; 1,-s)$ for such $s$, apply
 (\ref{Wz1b}) twice and conclude that
 \begin{align*}
   \vS^* W(z) D(z) W(z) \vS =  \vO(|z|^{1 + \ep}).
 \end{align*}

  It remains to show that
\begin{align}\label{eq:asGood}
 (\w{U^+\phi_i,W(z)\phi_j })_{1 \le i, j \le \mu} = z(\w{ (1- U_1^+)\phi_i, \phi_j}_0)_{1\le i, j \le \mu} + \vO(|z|^{1 + \ep}).
 \end{align}
 For that we first invoke (\ref{Upm}) with $j=1$ and write
\begin{align*}
 (\w{U^+\phi_i,W(z)\phi_j })_{1 \le i, j \le \mu} =-\vS^*(1 + r_0(z)( U^+ + { z} U_1^+))\vS  + \vO(|z|^{1 + \ep}).
 \end{align*} Next we note
 that $U^+\phi_j\in \vH^{-1}_{(5/2)^+}$  (cf. Corollary
 \ref{cor:microlocal-bounds}),  so that we can  use the
 expansion \eqref{eq:freeres}
 with $N=2$. Finally we proceed as in the proof of (\ref{2}) to
 rewrite the leading term  obtained this way,
 concluding \eqref{eq:asGood}.

\subStep{Case  3}   If zero is an exceptional point of the third kind, we can
combine methods used for  Case 1 and Case 2 to show \eqref{vEpm3b} and
\eqref{vEpm3b22}. The details are omitted. \ef

\medskip We remark that for Case 1, $\rank C_0 = \kappa =\mu$ and
$C_0^*C_0$ is positive definite. Moreover
Lemma \ref{lem5.15} ensures that the matrices $(\w{ (1- U_1^+)\phi_i,
  \phi_j}_0)_{1\le i, j \le \mu}$ for  { Case 2} and \linebreak
$(\w{ (1- U_1^+)\phi_i, \phi_j}_0)_{\kappa+1\le i, j \le \mu}$
for  { Case 3} of Lemma  \ref{prop5.17} are invertible. These properties
are  used in the following calculation of the leading term for the
inverse of $W(z)$. The proof of the proposition goes largely along the lines of that of  Proposition \ref{prop5.8}.\\

If $0$ is an eigenvalue of $P^+$ (relevant for
   Cases 2 and 3 below) we let $\Pi_0$ denote the spectral projection
  in $\vH$ onto the zero-eigenspace of this operator.
\begin{prop}\label{prop5.18} Assume  Conditions (\ref{ass5.1b}), (\ref{ass5.2b}) and (\ref{ass5.3b}).  The
  following asymptotics hold in $\vL(1,-s; 1,-s)$  for $s>1$
  sufficiently
  close to $1$  and for $|z|$ small with $\Im z >0$.

\medskip

\Step {Case 1}  Suppose  that zero is an exceptional point of the first
kind. Then
\begin{subequations}
\begin{equation} \label{W1b}
W(z)^{-1}= \f{\i}{\sqrt z} Q_0 +  \vO(|{z}|^{-\f 1 2 +\ep}).
\end{equation}
Here
\be\label{cjk0}
Q_0 = \sum_{j=1}^\kappa \w{-U^+\psi_j, \cdot}_0\psi_j
\ee
 with vectors $\psi_j\in \vK$ obeying
 \be \label{cjk}
 (c(\psi_j), c(\psi_k)) =\delta_{jk} \quad  \mbox{ for } j,k =1, \cdots, \kappa.
 \ee
 Here
 $c(\psi_j) \in \C^m$ is defined as $c(\phi_j)$ in Lemma  \ref{prop5.17} with $\phi_j$ replaced by $\psi_j$ and $(\cdot, \cdot)$ denotes the scalar product of $\C^m$.
\end{subequations}
\medskip

\Step {Cases 2 and 3} Suppose  that zero is  an exceptional point of the
second or the third
kind, respectively, $\rho>1$ and $\rho_1 >\f 3 2$.
Then
\begin{equation} \label{W2b}
W(z)^{-1}= z^{-1} (\Pi_0 (1-U^+_1)\Pi_0)^{-1} \Pi_0 U^+  + \vO(|z|^{-1+\ep}).
\end{equation}

\end{prop}
\pf The results  follow from Lemmas \ref{lem5.16} and   \ref{prop5.17} and the
formula (\ref{repWb}).
 We only study {Case 1}. For  $s\in(1, \rho_1)$,  one has for some $\ep>0$
\begin{align*}
\vE(z) & =  D_0 + \vO(|z|^\ep) \, \mbox{ in } \vL(1, -s; 1,-s),  \\
\vE_\pm(z) & =  \vE_{\pm, 0}  + \vO(|z|^\ep) \, \mbox{ in }
             \vL(\C^\mu,\vH^1_{-s})  \mbox{ or in } \vL(\vH^1_{-s}, \C^\mu),  \\
\vE_{-+}(z) & =   \i\sqrt{z} B_0 + \vO(|z|^{\f 1 2 +\ep}),
\end{align*}
where $B_0= C_0^* C_0$ is positive definite. Therefore
\be
W(z)^{-1} = \f{\i}{\sqrt z}\vS B_0^{-1}\vS^* + \vO(|z|^{-\f 1 2 +\ep}).
\ee
Decompose $B_0^{-1}$ as $B_0^{-1}= M_0^2$ where $M_0\in \vM_{\kappa\times \kappa}(\C)$ is a positive definite Hermitian matrix.
Set $M_0 =(m_{ij})_{1\le i, j \le \kappa}$ and
\be
\psi_i =\sum_{k=1}^\mu m_{ki} \phi_k, \quad i =1, \cdots, \kappa.
\ee
 Then  $\{\psi_1, \cdots, \psi_\kappa\}$ is a basis of  resonance
 functions  of $P^+$. Similarly to the proof of
 Proposition \ref{prop5.7b} \ref{item:21s} we  can check that
 $\{\psi_1, \cdots, \psi_\kappa\}$ verifies the normalization condition
 \[
 (c(\psi_i),  c(\psi_j))=\delta_{ij};\quad i,j\leq \kappa.
 \]
   Hence    the expansion (\ref{W1b}), with the
   specification (\ref{cjk0})-(\ref {cjk}), is proved.
  \ef

Since $ E_{\vH}(\lambda_0 +z)^{-1}  = -W(z)^{-1}(P_0-z)^{-1}$ and
$G_0 U = -1$ on $\vK$, the following result  follows immediately
from Proposition \ref{prop5.18} (cf. \eqref{eq:basicHo}).

\begin{subequations}
\begin{prop}\label{prop5.19}  Assume (\ref{ass5.1b}), (\ref{ass5.2b}) and (\ref{ass5.3b}).
The following asymptotics hold in $\vL(-1,s; 1, -s)$  for $s>1$ and for $|z|$ small with $\Im z >0$.

\medskip

\Step {Case 1}   Suppose zero is an exceptional point of the first kind of $P^+$.  Let $\kappa$ be  the multiplicity of zero resonance. Then  for some $\ep>0$
\begin{equation} \label{e1i}
E_{\vH}(\lambda_0 +z)^{-1}= -\frac{\i}{\sqrt{z}} \sum_{j=1}^\kappa \w{\psi_j, \cdot} \psi_j + \vO(|z|^{-\f 1 2 +\ep}),
\end{equation}
where $\psi_j$, $j=1, \cdots, \kappa$, are resonance  states verifying (\ref{cjk}). \\

\Step {Cases 2 and 3}   Suppose zero is  an exceptional point of the
second  or the third
kind, respectively,
 $\rho>1$ and $\rho_1 >\f 3 2$.  Then  for some $\ep>0$
\begin{equation} \label{e2i}
E_{\vH}(\lambda_0 +z)^{-1}= z^{-1} (\Pi_0 (1-U^+_1)\Pi_0)^{-1} \Pi_0   + \vO(|z|^{-1 +\ep}).
\end{equation}
\end{prop}
\end{subequations}

\medskip

A  main result of this section is the following. In agrement with
previous usage, if  $\lambda_0$ is a
resonance of $H$ its  multiplicity is denoted by $\kappa$, and
 if $\lambda_0$ is an eigenvalue of $H$
 the spectral projection of $H$ onto the corresponding
eigenspace is denoted by $\Pi_H$.\\

\begin{thm}\label{thm5.20}   Assume (\ref{ass5.1b}), (\ref{ass5.2b})
  and (\ref{ass5.3b}). The following asymptotics hold for $R(\lambda_0
  + z)$ as a bounded operator from $H^{-1}_{s}$ to $H^{1}_{-s}$ for
  any  $s>1$, as $z\to 0$ and $\pm \Im z >0$, and $\epsilon=\epsilon(s)>0$.

\medskip

\Step {Case 0} Suppose   $\lambda_0$ is a regular point of $H$. Then   one has
\begin{align}\label{asymRz0b}
  \begin{split}
  R(\lambda_0 &+ z)= R'(\lambda_0 \pm \
\i0) \\ &+  (S-R'(\lambda_0\pm \i0) I_0S)D_0 G_0 (S^*- S^*
I_0R'(\lambda_0\pm \i0)) + \vO(|z|^\ep).
  \end{split}
\end{align}
  Here  the boundary value of resolvent $R'(\lambda_0 + \i0)$ (resp., $R'(\lambda_0 - \i0)$ ) is used when $\Im z >0$ (resp. $-\Im z >0$).
\medskip

\Step {Case 1} Suppose  $\lambda_0$ is an exceptional point of the
first kind of $H$.  Then
\begin{subequations}
\begin{equation} \label{asymRz1b}
R(\lambda_0 + z)= \frac{\i}{\sqrt{z}} \sum_{j=1}^\kappa \w{u_j, \cdot} u_j + \vO(|z|^{-\f 1 2 +\ep})
\end{equation}
where $u_j$'s  are resonance  states of $H$ given by
 \be \label{e5.171}
 u_j = (1- R'(\lambda_0\pm \i0) I_0)S\psi_j,\quad  1 \le j \le \kappa,
 \ee
 where  $\psi_j$'s  are the resonance  states of $P^\pm$ given in
 Proposition \ref{prop5.18}.
  \end{subequations}
\medskip

\Step {Cases 2 and 3}  Suppose   $\lambda_0$ is  an exceptional point of the
second  or the third
kind of $H$, respectively, $\rho >1$ and $\rho_1 > \f 3 2$.
 Then
\begin{equation} \label{asymRz2b}
R(\lambda_0 + z)= - z^{-1} \Pi_H    + \vO(|z|^{-1+\ep}).
\end{equation}
\end{thm}

\medskip

Mimicking the proof of Theorem \ref{thm5.10} we see that Theorem \ref{thm5.20} for $\Im z > 0$ is a consequence of
Proposition \ref{prop5.19}  and the formula (\ref{rep5.1.2}).
 The case $\Im z <0$ can be proved in the same way with $P^+$ replaced
 by $P^-$. Note that due to Theorem  \ref{prop5.5}
 \ref{item:plusminuss}, the resonance  states $u_j$ given  by
 (\ref{e5.171}) are  independent
 of the choice of the sign $\pm$.

\medskip

\begin{remark}\label{remark:resolv-asympt-nearNorm}
  For {Case 1} of Theorem \ref{thm5.20} the leading term of the
  resolvent $R(\lambda_0 +z)$ is expressed in terms of a specific
  basis of resonance states $\{u_1, \dots, u_\kappa\}$ of $H$ such
  that $\psi_j = S^*u_j$, $j=1, \dots, \kappa$, are given as in
  Proposition \ref{prop5.18}. The normalization condition (\ref{cjk})
  expressed in terms of the $u_j$'s reads \be \label{ujk}
  (c(S^*u_j),c(S^*u_k)) =\delta_{jk}, \quad j,k =1, \dots, \kappa.
  \ee By Proposition \ref{prop5.7b} (b), (\ref{asymRz1b}) remains
  valid for \textit {any} basis of resonance states $\{u_1, \dots,
  u_\kappa\}$ of $H$ verifying (\ref{ujk}). Note also the formula
\begin{align}\label{eq:norma}
  c(S^*u_j)=\f{1}{2\sqrt{\pi}}\parbb{\int_\bX \varphi_1 I_0  u_j \d x ,
    \dots, \int_\bX \varphi_m I_0 u_j \d x}, \quad j=1, \dots, \kappa.
\end{align}
\end{remark}

For Cases 0  and  1,
  $\rho=1$ is legitimate. However
 Cases 2 and 3
  of Theorem \ref{thm5.20} have the  stronger condition  $\rho>1$ and
  therefore exclude the physics models which would imply  a limitation on the application in Section
  \ref{sec:Transmission problem at threshold}, see Remark \ref{remarks:non-transmission-bem}
  \ref{item:tran1}. However with an a priori weak decay
property of the corresponding $L^2$-eigenfunctions if $\lambda_0$ is an eigenvalue of
  $H$ (as for
  Cases 2 and 3) we can almost verbatim mimic the proof of Theorem \ref{thm:rhoThree}
  and obtain a resolvent expansion  also for  $\rho=1$. Note however
  that the property $\lambda_0\not\in \sigma_\pp(H_\sigma')$ for small $\sigma$
  is based on a perturbation of  Proposition \ref{prop:mour2},
  cf. \cite{AHS}, and note also the relevance of Remark \ref{remark:The case
  lambda0insigma} \ref{item:25}. This
  expansion is
   up to second
order (at least for  Case 3).

\begin{thm}\label{thm:rhoThree2} Assume (\ref{ass5.1b}) and
  (\ref{ass5.2b}), as well as
  (\ref{ass5.3b}) with   $\rho=1$ and $\rho_1= 3/2$. Suppose $\lambda_0$ is an
  eigenvalue of $H$ and that  $\ran
  \Pi_H\subset L^2_t$ for some $t> 3/2$. Then the following
  asymptotics hold for $R(\lambda_0 + z)$ as an operator from
  $H^{-1}_{s}$ to $H^{1}_{-s}$, $s>1$,  for $z\to 0$  and  $\pm \Im z >0$, and $\epsilon=\epsilon(s)>0$.
\begin{equation} \label{asymRz1aa9}
R(\lambda_0 + z)=-z^{-1} \Pi_H + \frac{\i}{\sqrt{z}} \sum_{j=1}^\kappa \w{u_j, \cdot} u_j + \vO(|z|^{-\f 1 2 +\ep}).
\end{equation} Here the second  term on the right-hand side appears
only if $\kappa\geq1$, that is if $\lambda_0$ is a
resonance of $H$,  and in that case
$\set{u_1, \dots, u_\kappa}$ is a basis of   resonance states  of
$H$ being independent
 of the choice of the sign $\pm$.  If on the other hand
$\lambda_0$ is not a
resonance of $H$, then $R(\lambda_0 + z)+z^{-1} \Pi_H $ has  limits in
norm as  $z\to
0$, $\pm\Im z>0$.
  \end{thm}

  \begin{remark}\label{remark:resolv-asympt-near2}
   As in  Remark
   \ref{remark:resolv-asympt-nearSimp} \ref{item:decayEig} the
   condition $\ran
  \Pi_H\subset L^2_t$ for some $t> 1$ suffices for the last assertion
  of Theorem \ref{thm:rhoThree2}.
  \end{remark}

\section{Positive slowly decaying effective potentials}\label{sect5.2}

Low-energy scattering  for one-body Schr\"odinger operators with
positive slowly decaying potentials is studied in \cite{Na,Ya1,Ya2}.
It is shown in \cite{Wa6} that this kind of operators satisfies Gevrey
type resolvent estimates at the threshold and this can be used to
establish large time asymptotics of the quantum dynamics with
sub-exponential time decay estimates on the remainder \cite{AW,Wa6}.
In this section, we want to show similar results for the N-body
Schr\"odinger operator at the lowest threshold $\la_0 = \Sigma_2$,
assuming that the effective potential is positive  outside a compact
set and slowly decaying at infinity.

Recall first from \cite{Wa6} some results for one-body operators.
The model operator $H_0$ in this framework is a closed second order elliptic operator of  the form
 \be \label{5.2.e2.1}
 H_0= -\sum_{i,j =1}^n \partial_{x_i} a^{ij}(x) \partial_{x_j} +  \sum_{j =1}^n  b_j(x) \partial_{x_j}  + v(x),
 \ee
 where $a^{ij}(x)$, $b_j(x)$ and $v(x)$  are  complex-valued functions in $\R^n$, $n \ge 1$.
 \begin{subequations}
  We assume that
  $a^{ij}, b_j $ are bounded $C^1$ functions with bounded derivatives and there exists $c>0$
 such that
 \be \label{5.2.e2.2}
 \Re (a^{ij}(x)) \ge c I_n, \quad \forall x \in \R^n.
 \ee
 Assume also that $v$ is relatively bounded with respect to $-\Delta$ with relative bound zero and  there exist some constants $0<\mu<1$ and $c_0>0$ such that
\bea \label{5.2.ass2}
& &|\w{H_0 u, u}| \ge  c_0 ( \|\nabla u\|^2 + \|\w{x}^{-\mu}u\|^2), \quad \mbox{for all } u \in H^2,  \\
& & \sup_x |\w{x}^\mu b_j (x) | < \infty, \quad j = 1, \cdots, n. \label{5.2.ass2b}
\eea
\end{subequations} The bound 
(\ref{5.2.ass2}) is called a weighted coercive condition and is essential for Gevrey estimates of the resolvent at the threshold zero.

Under the assumptions (\ref{5.2.ass2}) and (\ref{5.2.ass2b}), $H_0$ is bijective from $D(H_0)$ to $R(H_0)$.
 Let $R_0(0) : R(H_0) \to D(H_0)$ be its algebraic inverse.
$R_0(0)$ is a densely defined and closed operator, continuous from $L^{2}_{s}$ to $L^{2}_{ s-2\mu}$  and compact from $L^{2}_{ s} $ to $ L^{2}_{ s-2\mu -\ep}$, $\ep>0$,
 for any $s \in \R$ (\cite{Wa6,Ya2}). Thus $R_0(0)^N : L^{2}_{ s} \to L^{2}_{ s-2\mu N}$ is bounded.
 If $\Re H \ge 0$, one can check that the strong limit
\[
\slim_{z\in \Omega(\delta), z\to 0} \w{x}^{-2N\mu}(R_0(z)^N- R_0(0)^N )=0,
\]
where $\Omega (\delta)=\{z:  \abs{\arg z } > \f \pi 2 + \delta\}$ with $\delta>0$ small. The following Gevrey-type estimates hold for the resolvent at the threshold zero.

\begin{thm}[{\cite[Theorem 2.1]{Wa6}}]\label{5.2.2th1}
  Assume conditions (\ref{5.2.e2.2}), (\ref{5.2.ass2}) and (\ref{5.2.ass2b}). For any $a>0$  there exists $C_a>0$ such that
  \[
 \|\e^{-a\w{x}^{1-\mu}}R_0(0)^N \|+  \|R_0(0)^N \e^{-a\w{x}^{1-\mu}}\|\le C_a^{N +1} N^{\gamma N},
\]
uniformly in  $N\ge 1$.  Here $\gamma = \f{2\mu}{1-\mu}$.
\end{thm}

From  Theorem \ref{5.2.2th1} one deduces the following result.

\begin{cor}[{\cite[Corollary 4.2]{Wa6}}]\label{5.2.2cor2}
Let $H_0 = -\Delta + v(x)$ be self-adjoint and positive and satisfy the assumptions of Theorem \ref{5.2.2th1}.  Then for any $a>0$  there exists $C_a>0$ such that
  \[
 \|\e^{-a\w{x}^{1-\mu}}R_0(z)^N \|+  \|R_0(z)^N \e^{-a\w{x}^{1-\mu}}\|\le C_a^{N +1} N^{\gamma N}, \quad\forall N \ge 1,
\]
uniformly in $N \ge 1$ and $z \in \Omega=\{\zeta\in \C:  |\arg \zeta | > \delta\}$ with $\delta >0$.
\end{cor}

Since $R_0^{(N)}(z) = N! R_0(z)^{N+1}$, Corollary \ref{5.2.2cor2} means that  $\e^{-a\w{x}^{1-\mu}}R_0(z)$ belongs to the Gevrey class $\vG^{(1+ \gamma)}(\Omega)$, where
\[
\vG^{(1+\gamma)}(\Omega) =\{ F : \Omega \to \vL(L^2)\mid \exists C>0 \mbox{ s.t. }  \|F^{(N)}(z)\| \le   C^{N+1} N! N^{\gamma N}, \; \forall z \in \Omega, N\in \N \}
\]
For $n=3$, the repulsive Coulomb Hamiltonian $- \Delta +
\f{c}{|x|}$, $c>0$, satisfies all conditions of Corollary \ref{5.2.2cor2}  with $\mu =\f 1 2$, so its resolvent belongs to $\vG^{(3)}(\Omega)$ in exponentially weighted spaces.
 In \cite{AW,Wa6}, non-selfadjoint perturbations $H$ of $H_0$ are studied and large time expansions are obtained for the quantum dynamics $e^{-it H}$ and $e^{-tH}$  as $t\to +\infty$.
\\

To study the N-body Schr\"odinger operator $H$ at its lowest threshold $\la_0 = \Sigma_2$, we assume
\be \label{ass5.2.21b} \mbox{
  $\lambda_0 =\Sigma_2$ is a non-multiple two-cluster threshold in the sense
  of Condition \ref{cond:uniq}}. \ee
This means there exists a unique
$a_0 \in\vA\setminus\{a_{\max}\}$ such that $\Sigma_2
\in\sigma_{\pupo}(H^{a_0})$, and this cluster
decomposition is a two-cluster decomposition obeying Condition
\ref{cond:geom_singl}.
 We keep the notation of Section \ref{Reduction near a
  simple two-cluster threshold} (in the present case,  $m=1$, $\varphi_1$ is a normalized eigenfunction of $H^{a_0}$ with eigenvalue $\Sigma_2$ and $\vH = L^2(\bX_{a_0})$).
  Denote $x_0 =x_{a_0} \in \bX_{a_0}$.
   From  Section \ref{Reduction near a
  simple two-cluster threshold} we have the representation formula for the resolvent,
\begin{equation*} 
R(z)= E(z)- E_{+}(z) E_{\vH}(z)^{-1} E_{-}(z)\,\text{ for }\,  \Im z \neq 0,
\end{equation*}
where
  \begin{align*}
E(z) &= R'(z) = (H'-z)^{-1}\Pi',   \\
E_+(z)&=(1-R'(z)I_0)S,   \\
E_-(z) &= S^*\parb{1- I_0R'(z)},\\
E_{\vH}(z)&= (z-\lambda_0)-(-\Delta_0 + S^*{ I_0}S -
S^* I_0R'(z)I_0 S),
\end{align*}
and  $\Delta_0 =\Delta_{x_{0}} $,  $I_0(x) = \sum_{b\not\subset a} V_b(x^b)$ and $S f (x)(x) =\varphi_1(x^{a_0})f(x_{0})$. The effective potential here is
\[
S^*I_0S(x_0) =\w{I_0\varphi_1, \varphi_1}_0,
\]
where $\w{\cdot, \; \cdot}_0$ is the scalar product of $L^2(\bX^{a_0})$.

 In addition to (\ref{ass5.2.21b}) we assume that
 \begin{subequations}
  \bea
 & & \lambda_0 \not\in \sigma_\pp(H'),   \quad \rho \in (0, 2), \label{ass5.2.22b} \\
  & & \exists  c, R>0  \mbox{ such that } S^*I_0 S\ge
  \f{c}{|x_{0}|^\rho}\text{ for } |x_{0}|>R.  \label{ass5.2.23b}
\eea
 \end{subequations}
Conditions (\ref{ass5.2.21b}) and  (\ref{ass5.2.22b}) show that
$R'(z)$ is holomorphic for $z$ near $\la_0$.  Condition
(\ref{ass5.2.23b}) implies that in some sense $\la_0$ can not be a
resonance of $H$ (cf. Lemma \ref{5.2.2lem6}). Therefore we need only
to distinguish the cases $\la_0$ be an eigenvalue of $H$ or not. In
the case $\la_0$ is not an eigenvalue of $H$ we prove the following result.

\begin{thm} \label{5.2.2th3}
 In addition to conditions (\ref{ass5.2.21b}),  (\ref{ass5.2.22b}) and (\ref{ass5.2.23b}), we assume    $\la_0 =\Sigma_2\not\in \sigma_\pp(H)$.
Then for  any $a>0$, $\e^{-a\w{x_0}^{1-\mu}}R(z)$ belongs to the Gevrey class $\vG^{(\f{1+\mu}{1-\mu})}(\Omega_{\la_0}(\delta))$, where $\mu = \f{\rho}{2}$  and  $\Omega_{\la_0}(\delta) =\{z\in\C: |z-\la_0| <\delta, |\arg (z-\la_0)|>\delta\}$, $\delta >0$ small.
\end{thm}

The proof of Theorem \ref{5.2.2th3} is divided into several steps. The main task is to prove that $E_{\vH}(z)^{-1} \in \vG^{(\f{1+\mu}{1-\mu})}(\Omega_{\la_0}(\delta))$ if $\la_0$ is not an eigenvalue of $H$.
Set
\[
E_{\vH}(\la_0 + \la )= \la-(-\Delta_0 + S^*{ I_0}S -
S^* I_0R'(\la_0 +\la )I_0 S) := \la -( -\Delta_0 + U(\la))
\]
and
\begin{equation}
W(\lambda) =-S^* I_0R'(\la_0 +\la )I_0 S = - \w{\varphi_1,   I_0R'(\lambda + \la_0)I_0(\varphi_1\otimes\cdot)}_0.
\end{equation}
 $W(\lambda)$ is holomorphic for $\lambda$ near $0$ and satisfies
\[
\|\w{x}^{1+2\mu} W(\lambda) \w{x}^{1 +2\mu}\| \le C
\]
uniformly in $|\lambda|<2\delta$ for some $\delta>0$ small, because $\Pi I_0 \Pi' = \vO(\w{x_0}^{-1-2\mu})$.
Under the assumption (\ref{ass5.2.23b}),  the results of \cite{Wa6} can be applied to $-\Delta_0 + S^*{ I_0}S$. However the non-local term $W(\la)$ can not be treated as a perturbation in the Gevrey setting.  To prove Theorem \ref{5.2.2th3} we follow the approach of \cite{Wa6} from the very beginning and  exploit the holomorphicity of $W(\la)$ in $\la$ to prove some uniform energy estimates.\\

For $s\in \R$, let  $\varphi_s$ be the weight function defined by
\[
\varphi_s(x_0) = \left(1 + \f{|x_0|^2}{R_s^2}\right)^{ s/2 } \mbox{ with }
R_s = M \w{s}^{\f 1{1-\mu}},
\]
where $M>1$ is to be chosen sufficiently large and is independent of $s$.

\begin{lemma} \label{5.2.2lem5}  There exist  constants $M, \delta, C >0$ such that
\[
\|\w{x_0}^r  \varphi_s W(\lambda) \varphi_{-s}\w{x_0}^{r'}\| \le C
\]
for   $r, r' \in \R$ with $r, r' \le 1 + 2\mu$,
 for $\la \in \C$ with $|\lambda| < \delta$ and for $s\in \R$.
\end{lemma}
\pf Since $[\varphi_s, \Pi]=0$ and $\Pi I_0 \Pi' = \vO(\w{x_0}^{-1-2\mu})$, one has
\[
\|\w{x_0}^r  \varphi_s W(\lambda) \varphi_{-s}\w{x_0}^{r'}\| \le C \| \varphi_s R'(\la_0+ \lambda) \varphi_{-s}\|.
\]
Computing the commutator $[-\Delta_0, \varphi_s]$ one obtains
\[
\varphi_s H' \varphi_{-s} = H' + \vO\big(\f{s}{M\w{s}^{\f 1{1-\mu}}}\big) = H' +   \vO\big(\f{1}{M}\big),
\]
where the term $\vO\big(\f{1}{M}\big)$ satisfies the bound
 \[
 \|\vO\big(\f{1}{M}\big) (H'+i)^{-1}\| \le \f{C}{M}
 \]
 uniformly in $s\in \R$. Since $\la_0$ is in the resolvent set of $H'$, there exists $\delta>0$ such that $R'(\la_0 +\la)$  is a well defined holomorphic function for $|\la| <2\delta$  and we can take $M>1$ large so that
 \[
 \| \vO\big(\f{1}{M}\big)R'(\la_0+ \lambda)\| \le \f 1 2
 \]
 for $|\la| <\delta$.  From the equation
\[
 \varphi_s R'(\la_0+ \lambda) \varphi_{-s} =  R'(\la_0+ \lambda) \left(1 +   \vO\big(\f{1}{M}\big)R'(\la_0+ \lambda)\right)^{-1} \Pi',
\]
it follows that $ \varphi_s R'(\la_0+ \lambda) \varphi_{-s}$ is
uniformly bounded for $s\in \R$ and $|\la|<\delta$. Whence Lemma \ref{5.2.2lem5} is proven.
\ef

Let $\chi_R(x_0) =\chi_1(\f{x_0}{R})$, $R\ge 1$, where $\chi_1\in C^\infty$ is a cut-off such that $0\le \chi_1 \le 1$, $\chi_1(x_0) =0$ if $|x_0| \le 1$ and
$\chi_1(x_0) =1$ if $|x_0| \ge 2$. Set
\begin{align*}
F(\la) & = -\Delta_0 + 1- \chi_R + \chi_R U(\lambda)\chi_R \\
\widetilde U(\la) &=  U(\la) -  (1- \chi_R + \chi_R U(\lambda)\chi_R) \\
  h_0 &= -\Delta_0 + 1-\chi_R + \chi_R^2 S^*I_0S.
\end{align*}
 Then $E_\vH(\la_0 + \la)$ can be decomposed as
\begin{align}
E_\vH(\la_0 + \la)& = \la - (F(\la) + \widetilde U(\la)) \\
&= \la - (h_0  + (1-\chi_R) ((1+\chi_R)S^*I_0S -1) + W(\la)) \nonumber
\end{align}
For $R >1$ sufficiently large, $h_0$ is a one-body Schr\"odinger operator with globally positive and slowly decaying potential $v_0= 1-\chi_R + \chi_R^2 S^*I_0S$:
\[
 v_0(x_0) \ge \f{c}{\w{x_0}^{2\mu}}, \quad x_0 \in \bX_0,
\]
for some $c>0$. The operator $F(0) = h_0 + \chi_R W(0)\chi_R $ is a non-local perturbation  of $h_0$  and $F(0) \ge 0$.
Note that
\[
\widetilde U(\la) = (1-\chi_R) (U(\la)-1) + \chi_R U(\la) (1-\chi_R).
\]
Since $1-\chi_R$ has  compact support and $\la_0$ is in the resolvent set of $H'$, making use of the relation
\[
 \e^{a\w{x_0}} H'  \e^{-a\w{x_0}} = H' + \vO(a)
\]
for $a>0$ small,  one sees that $\widetilde U(\la) $ is exponentially decaying in the sense that
\be \label{5.2.2Uexponential}
\| \e^{a\w{x_0}} \widetilde U(\la)  \e^{a\w{x_0}} \| \le C
\ee
uniformly for $\la$ near $0$ and $|a|\le \delta$, $\delta>0$ small.


\begin{lemma} \label{5.2.2lem6} Assume (\ref{ass5.2.23b}) and let $\mu =\f \rho 2 \in (0, 1)$.  One has
\be
  \| \w{x_0}^{-\mu} \varphi_s u\| +   \| \nabla (\varphi_s u)\| \le  C \| \w{x_0}^{\mu} \varphi_s (F(\la)-\la) u \|
\ee
uniformly in  $s\in\R$,  $u\in \vS$ and $ \lambda \in \Omega_0(\delta)=\{\zeta\in \C: |\zeta|<\delta, |\arg \zeta | > \delta\}\cup \{0\}$ .
\end{lemma}
\pf By construction, $h_0 =-\Delta_0 + v_0(x_0) $ satisfies the weighted coercive condition (\ref{5.2.ass2}).  Moreover,
 for $\la =\tau e^{i\phi} \in \C$ with $\tau > 0$ and $\phi\neq 0$,  $|\phi|\ge \delta >0$, $e^{-i\phi}(h_0- \la)$ also satisfied  (\ref{5.2.ass2}) with a lower bounded independent of $\la$ so long as $|\phi|\ge \delta$ and $|\tau| <\delta$ for $\delta >0$ small.  Applying  \cite[Lemma 3.1]{Wa6}  to $h_0$ and $e^{-i\phi}(h_0- \la)$, we deduce
\be \label{5.2.2e61}
  \| \w{x_0}^{-\mu} \varphi_s u\| +   \| \nabla (\varphi_s u)\| \le  C \| \w{x_0}^{\mu} \varphi_s (h_0-\la) u \|
\ee
uniformly in  $s\in \R$ and $\la\in \Omega_0(\delta)$.

 Since
$F(\la) = h_0 + \chi_R W(\la)\chi_R $, Lemma \ref{5.2.2lem5} shows that
\[
\| \w{x_0}^{\mu} \varphi_s (h_0-\la) u \| \le \| \w{x_0}^{\mu} \varphi_s (F(\la)-\la) u \| + C R^{-2-2\mu} \|\w{x_0}^{-\mu} \varphi_s u \|
\]
uniformly in $s\in \R$ and $|\la|$ small. Now  Lemma \ref{5.2.2lem6} follows if $R>1$ is taken appropriately large.
\ef

 Lemma \ref{5.2.2lem6} shows that $G(\la) = (F(\la) -\la)^{-1}$ satisfies the estimate
\[
\|\w{x_0}^{-\mu} \varphi_s G(\la) \varphi_{-s}  \w{x_0}^{-\mu} \| + \|\nabla \varphi_s G(\la) \varphi_{-s} \w{x_0}^{-\mu}  \| \le C
\]
uniformly in $s\in \R$ and  $\la \in \Omega_0(\delta)$. It follows that
\be \label{5.2.2e44}
\|\w{x_0}^{-2\mu} \varphi_s G(\la) \varphi_{-s}   \| \le C'\w{s}^\gamma, \quad \gamma = \tfrac{2\mu}{1-\mu},
\ee
uniformly in $s$ and $\la$. The following technical estimates are the main step in the proof of Gevrey estimates of $G(\la)$, cf. \cite[Theorem 3.4]{Wa6}.

\begin{prop}\label{5.2.2prop7}
   There exist   constants $ C, \delta >0$ such that for any $r\in \R$, $N \in \N$ and $\lambda \in \Omega_0(\delta)$
 \be
\| \w{x_0}^{-2\mu} \w{x_{N,r}}^{-(2N+r)\mu} G^{(N)}(\la)  \w{x_{N,r}}^{ r\mu}  \|
 \le C^{N+1} N! \w{(2N+r)\mu}^{\gamma(N+1)}.
\ee
Here
\[
x_{N,r} = \f {x_0}{R_{N,r}}  \mbox{ with }R_{N,r} = M \w{ (2N +r) \mu }^{\f 1 {1-\mu}}
\]
and $\w{x_{N,r}} = (1 + |x_{N,r}|^2)^{\f 1 2}$.
\end{prop}
\pf The case $N=0$ and $r\ge 0$  follows from (\ref{5.2.2e44}) with $s = r\mu$.
For the general case $N\ge 1$ and $r\ge 0$, we write
\[
G^{(N)}(\la) =  (G(\la)^2 - G(\la)\chi_R W'(\la)\chi_R G(\la))^{(N-1)},
\]
and prove by an induction on $N$ that
\begin{eqnarray*}
\lefteqn{\| \w{x}^{-2\mu} \w{x_{N,r}}^{-(2N+r)\mu} G^{(N)}(\la)  \w{x_{N,r}}^{ r\mu}  \| } \nonumber \\[2mm]
& & \le C_N^{N+1} N! \w{(2N+r)\mu}^{\gamma (N+1)}
 \end{eqnarray*}
 with $ C_{N+1} \le C_N ( 1 + \f{c}{N^{1+\gamma}})$ for some $c>0$ independent of $N$. The details are the same as the proof of  \cite[Theorem 3.4]{Wa6} and are omitted here.
\ef

To convert polynomial weight depending on $N$ into exponential weight independent of $N$, we use the following estimate. 
\be \label{polytoexpo} 
\forall{a>0}\,\exists{A_a>0}:\quad\|\w{x_{N,r}}^{(2N+r)\mu} \e^{-a\w{x_0}^{1-\mu}}\|_{L^\infty} \le A_a^{\max\{2N+r, 1\}},
\ee
uniformly in $N \in \N$ and $r\in \R$.
In fact, if $2N+ r \le C$ for some constant $C>0$,  the left-hand side of (\ref{polytoexpo}) is uniformly bounded by some constant $C_1$. For
$2N + r > C$ with $C>1$ large but fixed, consider the function
\[
f(\tau) = \w{\f{\tau}{R_{N,r}}}^{(2N+r)\mu} e^{-a \tau ^{1-\mu}}, \quad \tau =|x_0|.
\]
One has
\[
f'(\tau) = \f{f(\tau)}{\tau^\mu ( R_{N,r}^2 + \tau^2) } \left( (2N+r)\mu \tau^{1+\mu} - (1-\mu) a (R_{N,r}^2 +\tau^2)\right).
\]
Since $\mu \in]0, 1[$, for each $a>0$, one can find some constant  $A>0$ such that
\[
f'(\tau) <0, \quad \mbox{ for } \tau > A R_{N,r}.
\]
Therefore for $2N+ r > C$,
\[
\|\w{x_{N,r}}^{(2N+r)\mu} e^{-a\w{x_0}^{1-\mu}}\|_{L^\infty} \le \sup_{0 \le \tau \le  A R_{N,r}} f(\tau) \le \w{A}^{(2N+r)\mu}
\]
This proves (\ref{polytoexpo}) for some appropriate constant $A_a>0$.

 Proposition \ref{5.2.2prop7} implies the following Gevrey estimates for $G(\la) =(F(\la)-\la)^{-1}$:
\be \label{GevreyG}
 \|\w{x_0}^{-\tau} \e^{-a\w{x_0}^{1-\mu}} G^{(N)}(\la)  \w{x_0}^{ \tau}  \| \le C_a^{N+1+\tau} N! \w{N+\tau}^{\gamma (N+1) +\f{\tau}{1-\mu} }
\ee
uniformly in $\tau \ge 0$,  $N \in \N$ and $\la \in\Omega_0(\delta)$.

\begin{lemma} \label{5.2.2lem8}  Assume (\ref{ass5.2.21b}), (\ref{ass5.2.22b}) and (\ref{ass5.2.23b}). Let $u $ be a solution to the equation $E_{\vH}(\la_0)u=0$ and $u \in \vH_{\tau}$ for some $\tau\in \R$. Then there exists
some positive constant $b_0$ such that $\e^{b_0\w{x_0}^{1-\mu}}u \in
\vH$. Here $\vH = L^2(\bX_{0})$ and $\vH_\tau = L^2(\bX_0,
\w{x_0}^{2\tau} \d x_0)$.
\end{lemma}
\pf $u$ satisfies the equation $u = G(0)(1-\chi_R) (1- U(0))u$. Since $U(0)$ is continuous in $\vH_s$ for any $s\in \R$, $u\in \vH_\tau$ and $1-\chi_R$ is of compact support, $(1-\chi_R)(1-U(0))u
\in \vH_s$ for any $s>0$.     Proposition \ref{5.2.2prop7} with $N=0$ and $r <0$ shows that $u \in \vH_{\infty}$. See also Subsection 4.1.3.
To show the sub-exponential decay of $u$, we write
\[
E_\vH(\la_0 ) = - (h_0  + (1-\chi_R) (S^*I_0S -1) + W(\la_0)).
\]
Since $h_0 \ge -\Delta_0 + \f{c}{\w{x_0}^{2\mu}}$ for some $c>0$, the following   Agmon energy estimate holds true: $\exists b, C>0$ such that
\be \label{5.2.2e81}
\|\w{x_0}^{-\mu}\e^{b\w{x_0}^{1-\mu}}f \|^2 + \| \nabla_{x_0}(\e^{b\w{x_0}^{1-\mu}}f)\|^2 \le C (|\w{\e^{2b\w{x_0}^{1-\mu}} h_0 f,f}|  +\|f\|^2)
\ee
for $f\in D(h_0)$  with $\e^{2b\w{x_0}^{1-\mu}} h_0 f \in \vH$ (cf. \cite[(5.37)]{Wa6}).
 $R'(z)$ being holomorphic for $z$ near $\la_0$,  one has for $b>0$  small enough
\begin{eqnarray*}
|\w{\e^{2b\w{x_0}^{1-\mu}}W(\la_0) f,f}| &\le &C_1 (\|\w{x_0}^{-1-2\mu}\e^{b\w{x_0}^{1-\mu}} f\|^2 + \|f\|^2)
\\
& \le &\ep \|\w{x_0}^{-\mu}\e^{b\w{x_0}^{1-\mu}} f\|^2 +  C_\ep\|f\|^2
\end{eqnarray*}
for any $\ep >0$.
We  deduce from (\ref{5.2.2e81}) that
\be \label{5.2.2e82}
\|\w{x_0}^{-\mu}\e^{b\w{x_0}^{1-\mu}}f \|^2 + \| \nabla_{x_0}(\e^{b\w{x_0}^{1-\mu}}f)\|^2 \le C (|\w{\e^{2b\w{x_0}^{1-\mu}} E_\vH(\la_0) f,f}|  +\|f\|^2)
\ee
with possibly another constant $C$. Since $u \in \vH_\infty$  and $E_\vH(\la_0) u=0$, the above inequality applied to $u$ shows that $\e^{b_0\w{x_0}^{1-\mu}}u \in \vH$ for $0<b_0 <b$. \ef

{\noindent \bf Proof of Theorem \ref{5.2.2th3}.}
Since $\la_0$ is not an eigenvalue of $H$, Lemma \ref{5.2.2lem8} shows that $E_\vH(\la_0)$ is injective in $\vH_s$ for any $s$.
Writing
\be \label{5.2.2effective}
E_{\vH}(\la_0 + \la) =  -   (F(\la)-\la) (1 + G(\la)\widetilde U(\la))
\ee
for $\la \in \Omega_0(\delta)$, one sees that $1 + G(0)\widetilde U(0)$ is injective. The mapping
 $\la \to G(\la)\widetilde U(\la)\in \vL(\vH_s)$ is a continuous and compact operator-valued. Consequently, $1+ G(0)\widetilde U(0)$ is invertible and  $ (1 + G(0)\widetilde U(0))^{-1} \in \vL(\vH_s)$ which implies,  by continuity,  $(1 + G(\la)\widetilde U(\la))^{-1}$ exists and
\[
 \|(1 + G(\la)\widetilde U(\la))^{-1}\| \le C
 \]
 for $\la \in  \Omega_0(\delta)$,  $\delta >0$ small. This proves that $E_{\vH}(\la_0 + \la)$ is invertible with the inverse given by
 \be \label{5.2.2inverse}
 E_{\vH}(\la_0 + \la)^{-1} =  -   (1 + G(\la)\widetilde U(\la))^{-1}G(\la) = -    G(\la)(1+ \widetilde U(\la)G(\la)))^{-1}.
 \ee
 As operator from $\vH_s$ to $\vH_{s-2\mu}$, $ E_{\vH}(\la_0 + \la)^{-1} $ is uniformly bounded for $\la \in  \Omega_0(\delta)$. Therefore the formula
\[
R(z) =  E(z)- E_+(z) E_{\vH}(z)^{-1} E_-(z)
\]
initially valid for $\Im z \neq 0$ can be extended to $z\in \Omega_{\la_0}(\delta)$. We conclude that $H$ has no eigenvalue in $(\la_0 -\delta, \la_0)$, hence $\sigma_d(H)$ is finite.

$E(z)$ and $E_{\pm}(z)$ are holomorphic for $z$ near $\la_0$.  Since
\[
\e^{-a\w{x_0}^{1-\mu}}H' \e^{a\w{x_0}^{1-\mu}} = H' + O(a),
\]
$
\e^{-a\w{x_0}^{1-\mu}}R'(z) \e^{a\w{x_0}^{1-\mu}}
$
is holomorphic and uniformly bounded for $|z-\la_0|$ small, provided that $a>0$ is small. It follows that
$\e^{-a\w{x_0}^{1-\mu}}E_+(z) \e^{a\w{x_0}^{1-\mu}}$ is holomorphic and uniformly bounded for $z\in \Omega_{\la_0}(\delta)$. Therefore
to prove Theorem \ref{5.2.2th3}, it is sufficient to show that
  $\e^{-a\w{x_0}^{1-\mu}} E_\vH(z)^{-1} $ belongs to $\vG^{(\f{1+\mu}{1-\mu})}(\Omega_{\la_0}(\delta))$.

  For $a>0$ small, $\widetilde U(\la) e^{a\w{x_0}^{1-\mu}}$ is holomorphic and uniformly bounded for
$\la\in \Omega_0(\delta)$. We conclude from (\ref{GevreyG}) (with $\tau =0$) that
$\widetilde U(\la)G(\la) \in \vG^{(\f{1+\mu}{1-\mu})}(\Omega_0(\delta))$ which, together with the uniform bound
\[
\|(1 + \widetilde U(\la) G(\la))^{-1}\| \le C,
\]
shows that $(1 + \widetilde U(\la) G(\la))^{-1} \in  \vG^{(\f{1+\mu}{1-\mu})}(\Omega_0(\delta))$. For $z = \la_0 +\la$, it follows from
(\ref{GevreyG}) that
\[
\e^{-a\w{x_0}^{1-\mu}} E_{\vH}(z)^{-1} =  -    \e^{-a\w{x_0}^{1-\mu}} G(z-\la_0)(1+ \widetilde U(z-\la_0)G(z-\la_0))^{-1}
\]
belongs to $\vG^{(\f{1+\mu}{1-\mu})}(\Omega_{\la_0}(\delta))$.  This finishes the proof of Theorem \ref{5.2.2th3}. \hfill $\Box $ \vskip 3mm

\begin{remarks}
  \begin{enumerate}[1)]
  \item 
  Making use of (\ref{5.2.2e44}) and repeating the proof of Theorem \ref{5.2.2th3}, one can prove that for any $N\in \N$ and $\tau \in \R$,  one has
\be
\|\w{x_0}^{- 2(N+1)\mu-\tau} R^{(N)}(z)  \w{x_0}^{ \tau}  \| \le C_{N, \tau}
\ee
uniformly for $\la \in \Omega_{\la_0}(\delta))$. This implies in particular that the limit
\be
R(\la_0) = \lim_{\la \in \Omega_{\la_0}(\delta)), \la \to \la_0} R(\la)
\ee
exists in $\vL(L^2_\tau, L^2_{\tau-2\mu -\ep})$ for any $\ep>0$ and
$R(\la_0) \in \vL(L^2_\tau, L^2_{\tau-2\mu})$ for any
$\tau\in\R$. Here and in the remaining part of this section, $L^2_r =
L^2(\bX, \w{x_0}^{2r}\d x)$.

\item    Making use of (\ref{GevreyG}), one can show the following improvement of Theorem \ref{5.2.2th3}: For any $a>0$ and $\tau \ge 0$, there exists some constant $C$ such that
\be \label{GevreyH}
 \|\w{x_0}^{-\tau} \e^{-a\w{x}^{1-\mu}} R^{(N)}(z)  \w{x_0}^{ \tau}  \| \le C^{N+1} N! \w{N+1}^{\gamma (N+1)}
\ee
uniformly in  $N \in \N$ and $z \in\Omega_{\la_0}(\delta)$.
\ef
\end{enumerate}

\end{remarks}

\begin{prop}\label{5.2.2propLAP} Assume the conditions of Theorem \ref{5.2.2th3} and let $s> \f{1+\mu} 2$. Then the boundary values of the resolvent
\[
R(\la\pm \i0) =\lim_{\ep\to 0_+} R(\la  \pm \i\ep)
\]
exist in $\vL(L^2_s, L^2_{-s})$ for $\la \in [\la_0, \la_0 +\delta]$, $\delta >0$ and
\be \label{5.2.2eLAP}
\|\w{x_0}^{-s} R(\la\pm \i0)\w{x_0}^{-s}\| \le C
\ee
uniformly in $\la \in [\la_0, \la_0 +\delta]$.
\end{prop}
\pf We keep the notation used before. It is known \cite{Na} that
(\ref{5.2.2eLAP}) holds true for  $r_0(z) = (h_0-z)^{-1}$. Whence 
\be \label{5.2.2eLAP0} 
\forall{s > \tfrac{1+\mu}2}\,\,\exists C>0:\quad \|\w{x_0}^{-s} r_0(\la\pm \i0)\w{x_0}^{-s}\| \le C
\ee
uniformly in $\la \in [0, \delta]$.
For $F(z) = h_0 + \chi_R W(z)\chi_R$, we write
\[
F(z) -z = (h_0-z)(1+ r_0(z) \chi_R W(z)\chi_R), \quad \Im z \neq 0.
\]
For $\la \ge 0$ small, $r_0(\la+i0)\chi_R W(\la)\chi_R$ is compact in $\vH_{-s}$  and continuous for $\la \in[0, \delta]$, $\delta>0$.
For $R >1$ large,  $F(0)\ge -\Delta_0 + \f{c}{\w{x_0}^{2\mu}}$, $c>0$, in sense of selfadjoint operators and Lemma \ref{5.2.2lem6} remains true if $E_\vH(\la_0)$ is replaced by $F(0)$. Consequently,
 $F(0)$ is injective in $\vH_t$ for any $t\in\R$ because $0$ is not an eigenvalue of $F(0)$. This implies that  $1+ r_0(0)\chi_R W(0)\chi_R$ is injective in $\vH_{-s}$, hence $1 + r_0(0)\chi_R W(0)\chi_R$
 is invertible. By the continuity in $\la$, we conclude that $(1+ r_0(\la\pm \i0) \chi_R W(\la)\chi_R)$ is invertible in $\vL(\vH_{-s})$ for $\la>0$ small and
 its inverse is continuous in $\la \in[0, \delta]$.  Consequently the boundary values of $G(z) =(F(z)-z)^{-1}$ exist in $\vL(\vH_s, \vH_{-s})$ and
 \[
 G(\la \pm \i0) = \big(1+ r_0(\la\pm \i0) \chi_R W(\la)\chi_R\big)^{-1}r_0(\la \pm \i0)
 \]
are continuous for $\la\in[0, \delta]$. A similar argument  shows that the boundary values
 \[
 E_{\vH}(\la_0 + \la\pm \i0)^{-1} =  -   \big(1 + G(\la\pm  \i0)\widetilde U(\la)\big)^{-1}G(\la\pm \i0)
 \]
 exist in $\vL(\vH_s, \vH_{-s})$ and are continuous in $\la \in[0, \delta]$. Finally we obtain
 \[
 R(\mu \pm \i0) = E(\mu) - E_+(\mu) E_{\vH}(\mu \pm \i0)^{-1} E_-(\mu)
 \]
  exist in $\vL(L^2_s, L^2_{-s})$ and are continuous in $\mu \in[\la_0, \la_0 +\delta]$.
\ef

\begin{cor}\label{5.2.2corSpecM} Let $ e(\la)$ denote the spectral projector of $H$ on $]-\infty, \la]$. Assume the conditions of  Theorem \ref{5.2.2th3}. Then for any $a>0$ and $s> \f{1+\mu}2$,  there exist some constants $b, B>0$ such that
\be \label{5.2.2 e20a} \|\e^{-a\w{x_0}^{1-\mu}} e'(\lambda) \w{x_0}^{-s}\| \le B e^{-  b |\lambda-\la_0|^{-\f 1 \gamma}}, \quad \lambda \in (\la_0, \la_0 +\delta).
\ee
\end{cor}
\pf Since $ e'(\lambda) = \f 1{2\pi i} ( R(\lambda+\i0) - R(\lambda-\i0))$,
$\| \w{x_0} ^{-s} e'(\la)\w{x_0}^{-s}\| $ is uniformly bounded   for $\lambda \in (\la_0, \la_0 +\delta)$.
Iterating the first resolvent equation, one obtains for any $N\in \N$
\be
e'(\la) = (\la-\la_0)^N R(\la_0)^N e'(\la).
\ee
Applying (\ref{GevreyH}) with $\tau=s$, one deduces  that for any $a>0$, there exist some constants $c, C >0$ such that
\be
\|\e^{-a\w{x_0}^{1-\mu}} e'(\lambda) \w{x_0}^{-s}\| \le C c^N N^{\gamma N} (\la-\la_0)^N
\ee
for all $N\in \N$ and $\lambda \in (\la_0, \la_0 +\delta)$. It remains to minimize  the right-hand side by choosing $N$ in terms of $\la-\la_0$  such that
$ N\approx A |\lambda-\la_0|^{-\f 1 \gamma}$ as $\la \to \la_0$  for some appropriate constant $A>0$. Then
\begin{align*}
c^N N^{\gamma N} |\la-\la_0|^N  & \approx  e^{A |\lambda-\la_0|^{-\f 1 \gamma} ( \gamma \ln A +  \ln c )}  \\
  &\le B e^{-b|\la-\la_0| ^{-\f 1 \gamma}}, \quad  \lambda \in (\la_0, \la_0 +\delta),
\end{align*}
 for some constants $b, B>0$, if $A>0$ is  such that  $\gamma \ln A +
 \ln c<0$. This proves (\ref{5.2.2 e20a}).
 \ef

 When  $\la_0$ is an eigenvalue of $H$ we prove the following
 analogue of Theorem  \ref{5.2.2th3}.

\begin{thm} \label{5.2.2th9}
Assume the conditions (\ref{ass5.2.21b}),  (\ref{ass5.2.22b}) and (\ref{ass5.2.23b}).  Let $\la_0$ be an eigenvalue of $H$ and $\Pi_{\la_0}$ be the eigenprojection of $H$ associated with $\la_0$. Then one has
\be \label{5.2.2e91}
R(z) = -\f{\Pi_{\la_0}}{z-\la_0} + R_1(z)
\ee
where, for any $a>0$,  $e^{-a\w{x_0}^{1-\mu}}R_1(z)$ belongs to the Gevrey class $\vG^{(\f{1+\mu}{1-\mu})}(\Omega_{\la_0}(\delta))$.
\end{thm}
\pf  As in the proof of Theorem \ref{5.2.2th3}, we are led to study $E_\vH(z)$ for $z \in \Omega_{\la_0}(\delta)$.
We use  formula (\ref{5.2.2effective}) and another Grushin reduction to study $E_\vH(\la_0+\la)^{-1}$. The smoothness of $G(\la)$ and $\widetilde U(\la)$ at $\la=0$ implies that they can be expanded in appropriate spaces in powers of $\la$ for $\la $ near $0$. Set
\begin{align}
G(\la) &= G_0 + \la G_1 + o(\la), \\
\widetilde U(\la) &= \widetilde U_0 + \la \widetilde U_1 + o(\la).
\end{align}
One has $G_j \in \vL(\vH_s, \vH_{s-(j+1)\rho})$ and $\widetilde U_j \in \vL(\vH_s, \vH_\infty)$ for any $s \in \R$.

Let $k$ denote the multiplicity of eigenvalue $\la_0$ of $H$. Then $0$ is an eigenvalue of $E_\vH(\la_0)$  with multiplicity $k$. Let
$ K(z) = G(z) \widetilde U(z)$ and   $K_0 = G_0 \widetilde U_0$. Then $K_0$ is a compact operator in $\vH$ and $ \ker (1+ G_0 \widetilde U_0) =\ker E_\vH(\la_0)$.
As in Subsection 5.1.2, we can choose a basis, $\{\phi_1, \dots, \phi_k\}$, of $\ker E_\vH(\la_0)$ verifying
\[
\w{-\widetilde U_0 \phi_i, \phi_j} = \delta_{ij}, \quad 1 \le i, j \le k
\]
and construct a Grushin problem for $1 + K(\la)$. Let
\[
\vS : \C^k \to \vH, c= (c_1, \cdots, c_k) \to \vS c = \sum_{j=1}^k c_j\phi.
\]
Then $Q = \vS \vS^*$ is a projection onto $\ker (1+ K_0)$. Let $Q' = 1- Q$. One has $Q'(1+K_0)Q'$ is invertible on $\ran Q'$ and by continuity,
$Q'(1+K(\la))Q'$ is invertible on $\ran Q'$ with uniformly bounded inverse for $\la \in \Omega_0(\delta)$, $\delta >0$ small.  Let
\[
D(\la) = (Q'(1+K(z)Q')^{-1}Q'
\]
Since $(1+K(\la)) \in \vG^{(1+\gamma)}(\Omega_0(\delta))$ with $\gamma = \f{2\mu}{1-\mu}$, one has $D(\la) \in \vG^{(1+\gamma)}(\Omega_0(\delta))$.
By studying the Grushin problem for $1+ K(z)$ using $\vS$ defined above, we obtain
\be \label{5.2.2e92}
(1+ K(\la))^{-1} =\vE(\la) - \vE_+(\la) \vE_{-+}(\la)^{-1}\vE_-(\la)
\ee
where
\begin{align*}
\vE(\la) &=  D(\la), \\
\vE_+(\la) &= \vS - D(\la) (1+K(\la))\vS \\
\vE_-(\la) &= \vS^* - \vS^*(1+K(\la)) D(\la)\\
\vE_{-+}(\la) & = \vS^* ( -(1+ K(\la))  + (1+ K(\la)) D(\la) (1+K(\la))\vS.
\end{align*}
$\vE_{\pm}(\la)$ and $\vE_{-+}(\la)$ all belong to Gevrey classes of order $1+\gamma$. As in Subsection 5.2.1 Case 2, we can compute the
$k\times k$ matrix $\vE_{-+}(\la)$ et obtain
\be
\vE_{-+}(\la) = \la M_0 + \la^2 r_1(\la)
\ee
where $M_0$ is invertible and $r_1(\la) \in \vG^{(1+\gamma)}(\Omega_0(\delta))$. This leads to
\[
\vE_{-+}(\la)^{-1} = \la^{-1} M_0^{-1} +  r_2(\la)
\]
with $r_2(\la) \in   \vG^{(1+\gamma)}(\Omega_0(\delta))$. By  (\ref{5.2.2inverse}) and (\ref{5.2.2e92}), we obtain
\be
 E_{\vH}(\la_0 + \la)^{-1} = -    G(\la)(1+ K(\la))^{-1} = \f{C_0}{\la} + \widetilde R_1(\la)
\ee
in $\vL(\vH_s, \vH_{s-2\mu-\ep})$, where $C_0 = G_0 \vE_+(0) M_0^{-1}\vE_-(0)$ and for $a>0$, $e^{-a\w{x_0}^{1-\mu}} \widetilde R_1(\la) \in  \vG^{(1+\gamma)}(\Omega_0(\delta))$.
Therefore, $R(z)$ verifies the expansion
\[
R(z) = -\f{B_0}{z-\la_0} + R_1(z), \quad B_0 = E_+(\la_0)C_0 E_-(\la_0)
\]
One can show as in the proof of Theorem \ref{5.2.2th3} that
\[
\e^{-a\w{x_0}^{1-\mu}}R_1(z)  \in  \vG^{(1+\gamma)}(\Omega_{\la_0}(\delta)).
\]
  One has necessarily $B_0 =\Pi_{\la_0}$ by the spectral theorem for
$H$.
\ef

Using Cauchy integral formula to represent $\e^{-tH}$ in terms of the resolvent, we obtain from Theorem \ref{5.2.2th9} a large-time expansion for the heat semi-group $\e^{-tH}$ (cf. \cite[Theorem 2.3]{Wa6}).

\begin{cor}\label{5.2.2cor4}  Assume the conditions of Theorem \ref{5.2.2th3}. For any $a>0$,  there exist some constants $C, c>0$ such that
\be
\|\e^{-a\w{x_0}^{1-\mu}} \big(\e^{-tH} -\sum_{\la \in \sigma_d(H)} \e^{-t\la} \Pi_\la  - \e^{-t\la_0} \Pi_{\la_0} \big)\| \le C \e^{-t\la_0 - ct^{\f{1-\mu}{1+\mu}}}.
\ee
Here $\Pi_\la$ is the eigenprojector of $H$ associated with eigenvalue $\la$.
\end{cor}

\begin{remarks}
  \begin{enumerate}[1)]
  \item 
 For one-body operators  it is shown in \cite{AW,Wa6} that under some
 additional conditions the quantum dynamics $\e^{-\i tH}$, when regarded as operator from $L^2_{\rm comp}$ to $L^2_{\rm loc}$,  can be expanded as $|t| \to \infty$ with the same sub-exponential estimates on the remainder as in
 Corollary \ref{5.2.2cor4}. The conditions used  there exclude  a
 possible accumulation of quantum resonances towards threshold
 zero. It is an interesting and non-trivial open question to see if a
 similar result holds true for $\e^{-\i tH}$ in the $N$-body problem.

\item  If  $\la_0=\Sigma_2$ is a multiple two-cluster threshold one can apply the Grushin reduction of Section 2.3 to  show that  Theorems \ref{5.2.2th3}  and \ref{5.2.2th9} still hold true. Since the proof of Theorems \ref{5.2.2th3}  and \ref{5.2.2th9} relies heavily on the continuity of $R'(z)$ in exponentially weighted spaces, which is deduced from the  holomorphicity of $R'(z)$  for $z$ near $\la_0$, one can not expect these results  to be true for higher thresholds $\la_0 >\Sigma_2$.
 \end{enumerate}
\end{remarks}

\begin{example} Theorems  \ref{5.2.2th3}  and \ref{5.2.2th9} can be applied to physics models with Coulomb interactions given by (\ref{0.1Cou}). Assume that the lowest threshold $\la_0=\Sigma_2$ is non-multiple two-cluster and is equal to the lowest eigenvalue of a two-cluster Hamiltonian $H^a$ with $a=(C_1, C_2)$. Let $Q_j =\sum_{k \in C_j} q_k$ be the total charge of particles in cluster $C_j$, $j=1,2$. Assume that
\[
Q_1Q_2 >0.
\]
Then the effective potential is positive and slowly decreasing outside a compact set and (\ref{ass5.2.23b}) is satisfied with $\rho =1$ (see (\ref{eq:expans1})).  In this case,  Theorems  \ref{5.2.2th3}  and \ref{5.2.2th9} hold true in Gevrey class $\vG^{(3)}(\Omega_{\la_0}(\delta))$.
\end{example}

\section{Resolvent asymptotics for physics models near two-cluster thresholds}\label{sec:resolv-asympt-physics models near}
We will discuss some extensions of Section 5.1 for the physics models
 of Sections \ref{$N$-body Schr\"odinger operators}
and \ref{$N$-body Schr\"odinger operators with infinite mass nuclei}
(with $N\geq 3$ and $N\geq 2$, respectively) using the same notation
as in Section  \ref{sec:CoulRellich}.  As in the previous sections of
this chapter the  single particle space dimension is fixed as
$n=3$. In agreement with the settng  of Section  \ref{sec:CoulRellich} we shall not distinguish between the cases $\lambda_0=\Sigma_2$ and
$\lambda_0>\Sigma_2$.

We recall
 that for a given two-cluster threshold  $\lambda_0$ for the models of Sections \ref{$N$-body
  Schr\"odinger operators} and \ref{$N$-body Schr\"odinger operators with infinite mass
  nuclei} we  group the set of two-cluster decompositions $a$ for which $\lambda_0\in
\sigma_{\pp}(H^a)$ into $\vA_1$,  $\vA_2$  and $\vA_3$ for which
\begin{description}
\item [$\vA_1$:] the effective
inter-cluster interaction   is to leading order attractive Coulombic,
\item [$\vA_2$:] the effective
inter-cluster interaction  is to leading order repulsive  Coulombic,
\item [$\vA_3$:] the effective
inter-cluster interaction  is $\vO(|x_a|^{-2})$.
\end{description}
  This distinction
does not depend on choices of corresponding sub-Hamiltonian
 bound states $\varphi^a$ (i.e. channels); it is determined by  charges only
(cf. Case  1 introduced independently  in each  of  Sections \ref{$N$-body Schr\"odinger operators}
and \ref{$N$-body Schr\"odinger operators with infinite mass nuclei}).

Let for $a\in\widetilde{\vA}:=\vA_1\cup\vA_2\cup\vA_3$  the operator $P^a$ be the corresponding orthogonal projection onto
   $\ker(H^a-\lambda_0)$ in
   $L^2(\bX^a)$ and let $m_a$ be the dimension of this
   space. Obviously  $\Pi^a:=P^a\otimes 1$ projects  onto
   the span of functions of the form $\varphi^a\otimes f_a$,
   $\varphi^a\in\ker(H^a-\lambda_0)$, in $L^2(\bX)$.  We identify  $\ran P^a$, say spanned by an
   orthonormal basis
   $\varphi^a_1,\dots\varphi^a_{m_a}$, with $\C^{m_a}$ (using the
   basis), and similarly
   \begin{align*}
     \vH_a:=L^2(\bX_a,\C^{m_a}) &\simeq \oplus_{m\leq
     m_a}\,L^2(\bX_a)\ni \oplus_{m\leq m_a}
    f_{a,m}=f_a\\&\simeq S_af_a:=\sum_{ m\leq m_a}\varphi^a_m\otimes
    f_{a,m}\in \ran \Pi^a.
   \end{align*}

The effective potential  for $a\in\vA _3$ obeys
   $W_a:=S^*_aI_aS_a= Q_a|x_a|^{-2}+B_a$, where $Q_a$ is a $m_a\times
   m_a$ matrix-valued function depending only on
   $\theta=\hat x_a=|x_a|^{-1}x_a$ while
   $B_a=B_a(x_a)=\vO(|x_a|^{-3})$. We shall here only study the case
   where $Q_a=0$, meaning that the effective potential is \emph{fastly
    decaying}. The `right generalization' to the case $m_a>1$ of the
   distinction between Cases 2) and 3) (discussed primarily for
   $m_a=1$ in both of Sections \ref{$N$-body
  Schr\"odinger operators} and \ref{$N$-body Schr\"odinger operators
  with infinite mass nuclei}) is to let Case 2) correspond to $Q_a\neq
0$ and let Case 3) correspond to $Q_a=
0$, respectively. This means that we shall not consider Case 2)
defined in Section  \ref{sec:CoulRellich}. Whence,  letting  $\vA ^{\mathrm
{fd}}_3=\set{a\in\vA _3|\, Q_a=0}$, we assume that
\begin{equation}
  \label{eq:fastDec}
  \vA _3=\vA ^{\mathrm
{fd}}_3.
\end{equation}

For simplicity we shall also assume \eqref{eq:dirCo1} and
\eqref{eq:dirCo2} leaving us with studying the Grushin resolvent
representation \eqref{rep3ipm} where (as in Section
\ref{sec:CoulRellich}) $\vH=\sum_{a\in \widetilde{\vA}} \oplus \vH_a$,
    $\vH_a=L^2(\bX_a,\C^{m_a})=\oplus_{m \leq m_a}\,L^2(\bX_a)$ and
    $S=(S_a):\vH\to \vF\subset\vG=L^2(\bX)$ is  given by
    \begin{align*}
      f=\oplus_{a\in
      \widetilde{\vA}}\,
    f_{a}\to Sf=\sum_{a\in
      \widetilde{\vA}}\,
    S_af_{a};\quad f_{a}=\sum_{m\leq m_a}\oplus
    f_{a,m}, \quad S_af_{a}=\sum_{m \leq m_a}\varphi^a_m\otimes
    f_{a,m}.
    \end{align*}
    Let (as usual) $T=(SS^*)^{-1}S$. These operators $S$ and $T$ will
    freely be used on weighted  spaces, and we also adapt the following
    notation of Theorem \ref{thm5.10}.
    \begin{align*}
      S_If:= \sum_{a\in
      \widetilde{\vA}}\,
    I_aS_af_{a};\quad f=\oplus_{a\in
      \widetilde{\vA}}\,\,
    f_{a}.
\end{align*}

We write $-E_{\vH}(\lambda_0+z)=P(z)=P_0+U(z)-z$, where
$P_0=\oplus_{a\in \widetilde \vA} \,h_a$ with $h_a$ specified as
follows: If $a\in  \vA_1$ or $a\in  \vA_2$  we take $h_a=p_a^2+w_a$
(acting as a   diagonal operator  if $m_a>1$)   where
$w_a$ is constructed as in Subsection \ref{subsec:negat-effect-potent}
or
\ref{subsec:positive-effect-potent}, respectively. If $a\in\vA ^{\mathrm
{fd}}_3=\vA _3$ we take  $h_a=p_a^2$. We define $r_0(z)= (P_0-z)^{-1}=\oplus_{a\in \widetilde \vA} \,r_a(z)$ for
any $z\in\C\setminus\R$. Next we write
\begin{align}\label{eq:Presolv}
  P (z)^{-1}=\parb{1+r_0(z)U(z)}^{-1}r_0(z)=W(z)^{-1}r_0(z), \quad
  z\in\C\setminus\R,
\end{align} as in the previous sections.

If $\lambda_0$ is regular, meaning  that $\lambda_0$  is neither an
eigenvalue nor a resonance of $H$ (or, more precisely, that the set $\vE$ in
Theorem \ref{thm:physical-modelsRell} is the zero set), then $\ker W^{\pm}(0)=0$ where
\begin{align}\label{eq:limiW}
  W^{\pm}(0)=\lim_{\epsilon\to 0} W^{\pm}(\pm \i\epsilon)
  =1+\lim_{\epsilon\to 0} r_0(\pm \i\epsilon)\lim_{\epsilon\to 0}
  U(\pm \i\epsilon)=1+r^\pm_0 U^\pm_0
\end{align} with limits taken in  appropriate spaces. By  Fredholm
theory we would then obtain limits $\lim_{\epsilon\to 0} (P\mp \i\epsilon)^{-1}$
and then in turn, due to \eqref{rep5.9},  limits $\lim_{\epsilon\to 0}
R(\lambda_0\pm \i\epsilon)$.
 Actually  for $\lambda_0$ being  regular we can take  limits for
$z\to 0$ in the
quadrants
\begin{equation*}
  \vZ_{\pm}=\set{\Re z \geq 0,\, \pm\Im z>0}.
\end{equation*}

To be more precise about these assertions we first look at the   three
types of (diagonal) building blocks  $r_a(z)=(h_a-z)^{-1}$. Recall the
notions  from Section  \ref{sec:CoulRellich},
\begin{align*}
  \vH_{a,s}^k=H_{s}^k(\bX_a,\C^{m_a})\quad \text{for}\quad a\in
      \widetilde{\vA}, \,k\in \R\mand s\in\R.
\end{align*}
If $a\in \vA_1$
there are limits
\begin{align*}
  r_a(0\pm \i 0)=\lim_{z\to 0, z\in \vZ_{\pm}} r_a(z)\text{ in
  }\vL\parb{\vH^{-1}_{a, s},\vH^{1}_{a, -s}},\quad s>3/4,
\end{align*} cf. Subsection
\ref{subsec:negat-effect-potent}. Similarly, if $a\in \vA_2$
there are limits
\begin{align*}
  r_a(0\pm \i 0)=\lim_{z\to 0, z\in \vZ_{\pm}} r_a(z)\text{ in
  }\vL\parb{\vH^{-1}_{a, s},\vH^{1}_{a, -s}},\quad s>3/4,
\end{align*} however in this case $r_a(0+ \i 0)=r_a(0- \i 0)$ and the
common limit coincides with the quantity $h_a^{-1}$ of Lemma
\ref{lemma:posit-slowly-decay}, cf. \cite {Na} and \cite{Ya2}. Finally,
for $a\in\vA ^{\mathrm
{fd}}_3$  we know from the previous sections that
\begin{align*}
  r_a(0\pm \i 0)=\lim_{z\to 0, z\in \vZ_{\pm}} r_a(z)\text{ in
  }\vL\parb{\vH^{-1}_{a, s},\vH^{1}_{a, -s}},\quad s>1,
\end{align*}  again with $r_a(0+ \i 0)=r_a(0- \i 0)$ and in this case  the
common limit $G_0$ is given explicitly (each diagonal entry has the
kernel $(4\pi)^{-1}\abs{x_a-y_a}^{-1}$).

We are lead to consider  for $k\in\R$, $r,s< -3/4$ and $ t<-1$ the spaces
    \begin{align*}
 \vH_{r,s, t}^k =\oplus_{b\in \vA_1}\,  \vH_{b,r}^k \bigoplus
      \oplus_{b\in \vA_2}\, \vH_{b,s}^k \bigoplus \oplus_{b\in
        \vA_3}\,  \vH_{b,t}^k.
    \end{align*}

We could for example fix  $ (r,s,t)=\bar t:= -(1, 1, 4/3)$, and the corresponding spaces $\vH_{\bar
      t}^k $ could then be used by considering $r_0(z)\in \vL(\vH_{-\bar
      t}^{-1} ,\vH_{\bar
      t}^1 )$,  $U(z)\in \vL(\vH_{\bar
      t}^{1} ,\vH_{-\bar
      t}^{-1} )$ and therefore, in turn, $W(z)\in \vL(\vH_{\bar
      t}^{1} )$; $z\in \vZ_{\pm}$.
    With these interpretations we can check that the limits in
\eqref{eq:limiW} exist. Since $\lambda_0$ is regular for $H$ it
follows that $\ker
W^{\pm}(0)=0$. Since $W^{\pm}(0)\in \vC(\vH_{\bar
      t}^1)$ this allows us to take the $z\to 0$ limits in
\eqref{eq:Presolv}  as well.

We made an unnecessary simplifying assumption on the parameters, and using the spaces
$\vH_{r,s, t}^k$ we may similarly deduce the following result.
\begin{prop}
  \label{thm:resolv-asympt-physRegular} Suppose the conditions of
  Theorem \ref{thm:physical-modelsRell} and in addition  \eqref{eq:fastDec}, \eqref{eq:dirCo1} and
\eqref{eq:dirCo2}. Suppose $\lambda_0$ is regular. Then for all sufficiently big $r,s< -3/4$ and $ t<-1$
there is an
$\epsilon>0$ such that the following asymptotics holds in
$\vL(\vH_{r,s, t}^{1} )$    for  $z\to 0$ in $\vZ_{\pm}$.
\begin{subequations}
 \begin{align}
  \label{eq:asW1}
  W(z)^{-1} =W^{\pm}(0)^{-1}+ \vO(|z|^{\ep}).
\end{align}

Similarly, in the space $\vL(\vH_{-r,-s, -t}^{-1},\vH_{r,s, t}^{1})$
\begin{align}
  \label{eq:asW2}
  P(z)^{-1}=(P^{\pm})^{-1} +
  \vO(|z|^{\ep});\quad (P^{\pm})^{-1}:=W^{\pm}(0)^{-1}r^\pm_0.
\end{align}
\end{subequations}
\end{prop}

\begin{thm}[Regular case]
  \label{thm:resolv-asympt-physReg0} Under the conditions of
  Proposition \ref{thm:resolv-asympt-physRegular} the following asymptotics hold for $R(\lambda_0
+ z)$ as an operator from $H^{-1}_{t}$ to $H^{1}_{-t}$, $t>1$,  for
$z\to 0$ in $\vZ_{\pm}$  and for some $\epsilon=\epsilon(t)>0$.
\begin{align} \label{asymRz0P}
  \begin{split}
 R(\lambda_0 &+ z) = R'(\lambda_0\pm \i 0) \\&+  (S-R'(\lambda_0\pm \i
0)S_I)(P^{\pm})^{-1}(S^*- S_I^*R'(\lambda_0\pm \i 0)) + \vO(|z|^\ep).
   \end{split}
\end{align}
 \end{thm}

 This result is immediate from \eqref{rep3ipm} and Proposition
 \ref{thm:resolv-asympt-physRegular}, however the latter  results
 actually  give
 more detailed information on the (anisotropic) behaviour of  the resolvent. For
 example the `most singular part' of $R(\lambda_0 + z)$ is given by the
 term $S(P-z)^{-1}S^*\in \vL(H^{-1}_{t},H^{1}_{-t})$, but
 $S(P-z)^{-1}S^*\approx Sr_0(z)S^*\approx \Sigma_{a\in\vA}  P^a\otimes
 r_a(z)$ needs $t>1$ only for $a\in\vA ^{\mathrm
{fd}}_3$; $r_a(z)$ is `smaller' for $a\in \vA_1\cup\vA_2$.

We can also
 derive a result if $\lambda_0$ is an exceptional point of the second
 kind, meaning that the set $\vE$ in Theorem
 \ref{thm:physical-modelsRell} obeys $0\neq \vE\subset L^2$. For that
 we need the additional condition
  \begin{equation}
    \label{eq:dec1}
    \ran \Pi_H\subset L^2_t\text{  for some }t> 1,
  \end{equation} where $\Pi_H$ is the eigenprojection corresponding to
  the eigenvalue $\lambda_0$ of $H$ (note that $ \Pi_HL^2=\vE$). We
  can argue as for the last assertion of Theorem \ref{thm:rhoThree},
  see also Remark \ref{remark:resolv-asympt-nearSimp}
  \ref{item:decayEig}. This means more precisely that we use the above
  procedure for $ H_\sigma:=H-\sigma\Pi_H$, $\sigma> 0$ small.  Under
  the given hypotheses $\lambda_0$ is neither an eigenvalue nor a
  resonance of $H_\sigma$, and therefore there is an analogous version of
   \eqref{asymRz0P} for $H_\sigma$ and we can then (as in the
  previous sections) invoke \eqref{eq:sigmeFormR}. Hence we obtain the
  following result, where quantities depending on the `potential'
  $-\sigma\Pi_H$ are equipped with the  subscript $\sigma$.

  \begin{thm}[Exceptional
point of $2$nd kind]
    \label{thm:resolv-asympt-phys2nd} Suppose the conditions of
  Theorem \ref{thm:physical-modelsRell} and in addition
  \eqref{eq:fastDec}, \eqref{eq:dirCo1}, \eqref{eq:dirCo2}  and \eqref{eq:dec1}. Suppose $\lambda_0$ is an eigenvalue but not a
resonance of $H$. Then the following asymptotics hold for $R(\lambda_0
+ z)$ as an operator from $H^{-1}_{s}$ to $H^{1}_{-s}$, $s>1$,  for
$z\to 0$ in $\vZ_{\pm}$  and for some $\epsilon=\epsilon(s)>0$.
\begin{align} \label{asymRz0P6}
  \begin{split}
 R(\lambda_0 &+ z) = -z^{-1}\Pi_H+\sigma^{-1}\Pi_H+R_\sigma'(\lambda_0\pm \i 0) \\&+  (S-R'_\sigma(\lambda_0\pm \i
0)S_{I_\sigma})(P_\sigma^{\pm})^{-1}(S^*- S_{I_\sigma}^*R_\sigma'(\lambda_0\pm \i 0)) + \vO(|z|^\ep).
   \end{split}
\end{align}
  \end{thm}

 Note that the $z^{-\f 1 2}$-term is absent in (\ref{asymRz0P6}). This looks as if there is a discrepancy with the known results given in \cite{JK,Wa2} where there are $z^{-\f 1 2}$-terms in the resolvent expansions for exceptional point of the second kind. However there is not the case,  because under the decay assumption on the threshold eigenstates used in Theorem \ref{thm:resolv-asympt-phys2nd},  one can check by an explicit calculation that the $z^{-\f 1 2}$-terms of \cite{JK,Wa2} also disappear.

It remains to examine the cases where  $\lambda_0$ is an exceptional
point of the first or of the third
 kind, meaning (treating them uniformly) that the set $\vE$ in Theorem
 \ref{thm:physical-modelsRell} is not strictly a subset of $L^2$.
  For that
 we need the additional condition
  \begin{equation}
    \label{eq:dec20}
    \ran \Pi_H\subset L^2_t\text{  for some }t> 3/2,
  \end{equation} where $\Pi_H$ is the orthogonal projection onto $\ker
  (H-\lambda_0)$ (i.e. the eigenprojection if $\lambda_0$ is
  an eigenvalue  of $H$).  We know from
  Theorem \ref{thm:physical-modelsRell} that the  condition
  $\vE\not\subset L^2$ needs $\vA_3\neq \emptyset$.  If $\vA_1\neq
  \emptyset$ and $\vA^{\textrm{fd}}_3\neq \emptyset$ there is a technical problem on
  identifying the geometric and algebraic multiplicities for the
  eigenvalue $-1$ of  certain
  operators  $K^\pm$, cf. \eqref{eq:directSum}. The same method as the
  one used in the previous sections does not work (and in fact we dont
  know if the multiplicities  are equal). On the other hand under the additional condition
  \begin{equation}
    \label{eq:2plus3}
    \widetilde{\vA}=\vA_2\cup \vA^{\textrm{fd}}_3,
  \end{equation}
 the analogue of
  \eqref{eq:directSum} holds (for $H_\sigma$), by the same proof  (see
  also   Remark \ref{remark:The case
  lambda0insigma} \ref{item:25}). In
  this case we can mimic the proof of Theorem \ref{thm:rhoThree}
  and obtain the following result, which is very similar to Theorems
  \ref{thm:rhoThree} and \ref{thm:rhoThree2}.
  \begin{thm}[Exceptional
point of $1$st or $3$rd kind]
    \label{thm:resolv-asympt-phys1st3rd} Suppose the conditions of
  Theorem \ref{thm:physical-modelsRell} and in addition
   \eqref{eq:dirCo1}, \eqref{eq:dirCo2},
  \eqref{eq:dec20}   and \eqref{eq:2plus3}. Suppose $\lambda_0$ is a
resonance of $H$. Then the following asymptotics hold for $R(\lambda_0
+ z)$ as an operator from $H^{-1}_{s}$ to $H^{1}_{-s}$, $s>1$,  for
$z\to 0$ in $\vZ_{\pm}$  and for some $\epsilon=\epsilon(s)>0$.
\begin{equation} \label{asymRz1aa9last}
R(\lambda_0 + z)=-z^{-1} \Pi_H + \frac{\i}{\sqrt{z}} \sum_{j=1}^\kappa \w{u_j, \cdot} u_j + \vO(|z|^{-\f 1 2 +\ep}).
\end{equation} Here
$\set{u_1, \dots, u_\kappa}\subset H^1_{(-1/2)^-}$ is a basis of   resonance states  of
$H$ being independent
 of the choice of the  sign of  $\vZ_{\pm}$.
  \end{thm}

  \begin{remarks}\label{remarks:resolv-asympt-phys} The basis of
    resonance states can be specified in a fashion similar to
    normalization procedures in Theorem \ref{thm5.10} and Remark
    \ref{remark:resolv-asympt-nearNorm}, cf.  (\ref{nonresonant3a})
    and (\ref{nonresonant3b}) (and subsequent the computations).

 We
    consider the imposed conditions \eqref{eq:dirCo1} and
    \eqref{eq:dirCo2} as technically convenient but not being
    crucial. The interested reader may trace the outlined proof of
    Theorem \ref{thm:physical-modelsRell} for the general case where
    these conditions are not imposed and see how the above theory
    modifies.

We already discussed our need for \eqref{eq:2plus3},
    however we remark that \eqref{eq:fastDec} is not strictly
    necessary. In fact there are some results for the case where
    $\vA_3\setminus\vA^{\textrm{fd}}_3\neq \emptyset$ (this is an
    interesting case for the physics models). Then the spherical
    potential $Q_a$ should not to be `too negative', more precisely we
    need that $\sigma(-\Delta_\theta+Q_a)\subset (-1/4,\infty)$ for
    all $a$ in this set. With this extension there should be analogous
    results on the resolvent expansion at $\lambda_0$, not to be
    elaborated on here, see \cite{Wa5}.  However we shall later do a
    version of Theorem \ref{thm:resolv-asympt-physReg0} in a special
    case where indeed $\vA_3\setminus\vA^{\textrm{fd}}_3\neq
    \emptyset$, see Subsection \ref{subsec:Some better arguments}. If
    the above spectral condition on
    $a\in\vA_3\setminus\vA^{\textrm{fd}}_3$ is not fulfilled the
    resolvent asymptotics would be expected to be oscillatory, see
    \cite{SW} where oscillatory behaviour is detected for a one-body
    `toy model'. We shall not study this case; it does not seem to be
    an `easy problem'.
  \end{remarks}


 \chapter{Applications}\label{Applications}

  We will give  applications of the
 previous chapters to scattering theory. We shall primarily study  the non-multiple  case imposing Conditions
 \ref{cond:smooth} and \ref{cond:uniq}. With
 additional efforts the
 multiple  is treatable  somewhat similarly, see Subsectons \ref{subsec:Non-elastic
   scattering} and \ref{subsec:total-cross-sectionsMult} for  actual
 accounts of the
 multiple  case. The non-multiple  case appears rather complicated already,
 and we believe that treating only this case  may be considered as  `heart of
 the matter'. In Sections \ref{sec:rtominus2 potentials},
 \ref{sec:Transmission problem at threshold} and \ref{total
   cross-sections} we consider only the physics models of Subsections
 \ref{First principal example} and \ref{Second  principal example}
 with the particle
dimension  $n=3$
 (in Section \ref{total
   cross-sections} even more specialized).

Although most of the  material  presented in this chapter  is new it depends on the
literature, obviously most importantly for example \cite{DS1, JKW}. Since scattering
theory is an old well-studied subject the literature is large, let us
here  mention the   related works \cite{Bo, Do, CT, De1, De2, De, DG,
Is2, Is3, Is4, ITa, Ne, Sk2, Sk4, Sk5, Sk6, Sk7, SW, Wa6, Ya1, Ya3, Ya4,
Ya5}. Obviously this  list is not complete.

\section{Negative slowly decaying effective potentials}\label{sec:Long-range
  negative effective potentials}
We impose the attractiveness  condition \eqref{eq:virial0} for the cluster
decomposition  $a=a_0$ given in
Condition \ref{cond:uniq}.  More precisely suppose
 $\rho<2$ and  that  the inter-cluster potential
$I_a(x_a)=I_{a_0}(x_a)=I_0(x_a)$  fulfills the condition
\begin{align}\label{eq:virial0dd}
  \begin{split}
    \exists R\geq 0&\quad \exists \epsilon>0\quad\forall y\in \bX_a\text{ with }|y|\geq R: \\
    &I_a(y)\leq-\epsilon \inp{y}^{-\rho}\text{ and }-2I_a(y)-y\cdot
    \nabla I_a(y)\geq \epsilon \w{y}^{-\rho}.
  \end{split}
\end{align}

We assume
\begin{equation}\label{eq:15usual}
\lambda_0\notin \sigma_{\pp}( H'),
\end{equation}  or equivalently $\lambda_0\notin
\sigma_{\pp}(\breve H)$. If \eqref{eq:15usual} is not fulfilled  we may modify the theory to be
discussed
in agreement with Subsection \ref{sec:The case where
  lambda0insigma}, see Remark \ref{remark:restric0}.

We know from Chapter \ref{Spectral analysis of H' near E_0} that the
boundary values $\breve R(\lambda\pm \i 0)$ are smooth (in weighted
spaces) in a real neighbourhood $I\ni \lambda_0$. We take $I$,
$R=R_0\geq 1$ and the operator $B= B_{ R}=: \widecheck B$ exactly as done
in Subsection \ref{subsec:LAP bound}  (this $B$ should not to be mixed
up with the operator
$B_\kappa=B_{\kappa, R}=B(\kappa^2 B^2+1)^{-1}$ in the same chapter).

In addition we assume
 \begin{equation}
   \label{eq:nonEg0}
   \lambda_0\notin
\sigma_{\pp}( H).
\end{equation} If \eqref{eq:nonEg0} is not fulfilled and $\Pi_H$ is the
corresponding eigenprojection then $\lambda_0\notin \sigma_{\pp}
(H-\sigma\Pi_H)$ for  $\sigma>0$,  and we can consider the $H'$
construction for $H_\sigma=H-\sigma\Pi_H$, say denoted by $H'_\sigma$,
  which fits well onto the framework of Chapter \ref{Spectral
    analysis of H' near E_0} since the eigenfunctions decay
polynomially, cf. Theorems \ref{thm:negat-effect-potent},
\ref{thm:negat-effect-potents} and \ref{thm:physical-modelsRell}
\ref{item:PRel1}. In fact if  $\sigma>0$ is taken  small enough the condition
$\lambda_0\notin \sigma_{\pp}( H')$ implies the same property with $H'$
replaced by $H'_\sigma$, see the discussion before Theorem
\ref{thm:rhoThree2}. This would lead to a resolvent formula similar to
\eqref{asymRz0P6}. For simplicity we impose \eqref{eq:15usual} as well as \eqref{eq:nonEg0}, which
leads to the following analogue \eqref{eq:ResB2} of the physics models
resolvent
formula \eqref{asymRz0P}.

With these conditions  we know from Chapters \ref{Spectral analysis of
  H' near E_0} and \ref{chap:lowest
  thr} (see also Remark \ref{remark:resolv-asympt-nearSimp}
\ref{item:decayEig2}) that there exist continuous
boundary values $R(\lambda\pm \i 0)$ (in appropriate spaces) in an   interval of the form
$I_\delta^+=[\lambda_0,\lambda_0+\delta]\subset I$ with  $\delta>0$ small. In
fact we have the formulas
\begin{align}
  \label{eq:ResB2}
  R(\lambda\pm \i 0)= E^\pm(\lambda)- E^\pm_{+}(\lambda)
  E^\pm_{\vH}(\lambda)^{-1} E^\pm_{-}(\lambda),
\end{align} where,  abbreviating
$I_0=I_{a_0}$ and  $p_0=p_{a_0}$,
\begin{align*}
E^\pm(\lambda) &= R'(\lambda\pm \i 0),   \\
E^\pm_+(\lambda)&=(1-R'(\lambda\pm \i0)I_{0})S,   \\
E^\pm_-(\lambda) &= S^*(1- I_{0}R'(\lambda\pm \i 0)),\\
E^\pm_{\vH}(\lambda)&= (\lambda-\lambda_0 )-(p_0^2+ S^*{ I_{0}}S -
S^* I_{0}R'(\lambda\pm \i 0)I_{0} S).
\end{align*}  We also have
\begin{align*}
  R'(\lambda\pm \i 0)= \breve R(\lambda\pm \i 0
)-(p_{0}^2-\lambda)^{-1}\Pi=\breve R(\lambda\pm \i 0
)\Pi'.
\end{align*} The matrix-valued operator $W=S^*{ I_{0}}S$ is
well approximated by a diagonal one fulfilling a global virial
condition,
in fact approximated by  a multiple of the
identity  say denoted by $w {\mathbf 1}$, cf. the discussion
in the beginning of Subsection \ref{subsec:negat-effect-potent}. As in
the same subsection  we modify  the potential $V(z)$
correspondingly and denote the result by $v(z)$. Let
$h=(p_0^2 +w){\mathbf 1}$,
$r(z)=(h-z)^{-1}$ and
\begin{align*}
  r^\pm_\lambda=r(\lambda-\lambda_0\pm\i 0) \mand  v^\pm_\lambda =v(\lambda\pm \i
  0),
\end{align*} and note that
\begin{align}\label{eq:eff_v}
  \begin{split}
  -E^\pm_{\vH}(\lambda)&=h + v^\pm_\lambda-
    (\lambda-\lambda_0 ),\\
-E^\pm_{\vH}(\lambda)^{-1}&=r^\pm_\lambda(\lambda)\parb{{\mathbf 1}+v^\pm_\lambda
                            r^\pm_\lambda(\lambda)}^{-1};\\
 v^\pm_\lambda&=-S^* I_{0}R'(\lambda\pm \i 0)I_{0} S +\vO(r^{-1-\rho}).
  \end{split}
\end{align} If $m_a:=\dim \ker (H^{a}-\lambda_0)=1$,  $a=a_0$, we can here
replace $\vO(r^{-1-\rho})$ by  $\vO(r^{-\infty})$ which would refer  to a polynomially
decreasing term. In any case the  term is a   $\lambda$-independent
local potential (i.e. a function).
 The above  expressions
are  substituted into \eqref{eq:ResB2} to obtain  formulas for
$R(\lambda\pm \i 0)$ to be studied.

Note that
\begin{align}\label{eq:fredInverse}
  \parb{{\mathbf 1}+v^\pm_\lambda r^\pm_\lambda}^{-1}\in \vL(L^2_s)\cap \vL(\vB_{s_0});\quad s\in (s_0,1/2+\rho 3/4),
\,
 s_0=1/2+\rho/4.
\end{align}  This is the best we can do when
$m_a>1$. If $\lambda_0>\Sigma_2$ and
$m_a=1$ then   any $s\in (s_0,1/2+\rho)$
works. On
the other hand if  $\lambda_0=\Sigma_2$ and
$m_a=1$  then this inverse exists in
$\vL(L^2_s)$ for any $s\in (s_0,2\rho+2-s_0)$,
cf. \eqref{eq:comneg}. In particular we have  the following formula  for any $\psi\in \vB_{s_0}$
 (we abbreviate
throughout this section $\vB:=\vB_{1/2}$, $\vB^*=\vB^*_{1/2}$ and
${\vB^*_0}:=\vB^*_{1/2,0}$),
\begin{align}\label{eq:resolBAS}
  \begin{split}
   R(\lambda\pm \i 0 )\psi&=S\phi_{a}^\pm(\lambda) +\breve
   R(\lambda\pm \i 0)\psi_a^\pm(\lambda);\\
  \phi_{a}^\pm(\lambda)&=r^\pm_\lambda f_a^\pm (\lambda)\in \vB^*_{s_0}(\bX_a),\\
  f_a^\pm (\lambda)&=\parb{{\mathbf 1}+v^\pm_\lambda r^\pm_\lambda}^{-1}S^*\parb{1- I_{0}\breve R(\lambda\pm \i 0
    )\Pi'}\psi\in \vB_{s_0}(\bX_a),\\
  \psi_a^\pm(\lambda)&=\Pi'\parb{\psi-I_{0}S\phi_{a}^\pm(\lambda)}\in
  \vB.
  \end{split}
\end{align}
It follows that for any $s>s_0$ the $\vL (L^2_s, L^2_{-s})$--valued
functions $R(\cdot\pm \i 0 )$ are  continuous on $I_\delta^+$. For each
$\lambda\in I_\delta^+$ the operators $R(\lambda\pm \i 0 )\in \vL
(\vB_{s_0}, \vB^*_{s_0})$. It will be convenient to isolate the `main
parts` of \eqref{eq:resolBAS} writing
\begin{align}
  \label{eq:mainSplit}
  R(\lambda\pm \i 0 ) =Sr^\pm_\lambda S^*+\breve R(\lambda\pm
  \i0)\Pi'+\widecheck R(\lambda\pm \i 0 ).
\end{align} Here $\widecheck R(\lambda\pm \i 0 )$  are represented as  sums  of  various terms. Note for
example that for any $\psi\in \vB_{s_0}$ (as above) $ f_a^\pm (\lambda)-S^*\psi\in L^2_s$ for
any  $s\in (s_0,1/2+\rho 3/4)$.

\begin{remark}\label{remark:formRes}
  Note that in particular \eqref{eq:ResB2} as well as \eqref{eq:resolBAS} and
  \eqref{eq:mainSplit} are valid for $\lambda=\lambda_0$. This is
  a consequence of the imposed regularity condition
  \eqref{eq:nonEg0}. For other models, to be treated in Section
  \ref{sec:Transmission problem at threshold}, there are similar
  formulas as \eqref{eq:resolBAS} and \eqref{eq:mainSplit}. Again this
  requires regularity, i.e. absence of  bound  and resonance
  states at the threshold. Note also that the trick of replacing $H$ by
  $H_\sigma=H-\sigma\Pi_H$  in the case \eqref{eq:nonEg0} is not fulfilled is
  not restricted to the case of attractive slowly decaying  effective
  potentials, i.e. the condition \eqref{eq:virial0dd}, but can be used
  as well
  for repulsive slowly decaying  effective potentials fulfilling
  the following version of \eqref{eq:posLowbnd0},
\begin{equation*}
     \exists R\geq0\,\exists \epsilon> 0\,\exists \bar \rho\in [\rho, \tfrac 23(1+\rho))\, \forall y\in \bX_a\text{ with }|y|\geq R: \quad
     I_a(y)\geq\epsilon \inp{y}^{-\bar\rho}.
 \end{equation*}
  This is manifestly done already in
  the proof of Theorem \ref{thm:resolv-asympt-phys2nd} and as before
  doable thanks to the polynomial decay of the eigenfunctions. For possible threshold eigenfunctions
  for  non-slowly decaying effective
potentials  (cf. the asymptotics
  condition  \eqref{eq:posLowbnd0b}) the
  polynomial decay is missing making \eqref{eq:nonEg0} a
  non-trivial assumption in such  cases, see however Theorem \ref{thm:resolv-asympt-phys2nd}.
\end{remark}

\subsection{Sommerfeld's theorem}\label{subsec:Sommerfeld theorem}
 We impose the above conditions
 \eqref{eq:virial0dd}--\eqref{eq:nonEg0} on the threshold $\lambda_0$.
First we recall the following version of  the Sommerfeld's  theorem  above
$\lambda_0$, see \cite[Corollary 1.10]{AIIS} which extends the seminal
work \cite{Is4}.
\begin{thm}\label{prop:Sommerfeld1}  For any
   $\lambda\in I_\delta^+\setminus \set{\lambda_0}$  there exist
    $R= R_0\geq1$ and  $\sigma>0$ such that   for any real function
  $\chi(\cdot<\sigma)\in C^\infty(\R)$,
    which is   supported in $(-\infty,\sigma)$ and whose derivative
  has
  compact support, and for any $\psi\in \vB$
  \begin{align}\label{eq:45lem51}
     \chi(\pm
     \widecheck B<\sigma)R(\lambda\pm \i 0)\psi \in  \vB^*_0;\,\widecheck B= B_R.
  \end{align} Moreover (for each sign) $\phi=R(\lambda \pm \i 0
  )\psi\,(\in\vB^*)$ solves  $(H-\lambda)\phi=\psi$.

Conversely if  $\phi\in L_{-\infty}^2$  solves
$(H-\lambda)\phi=\psi$ for a given   $\psi\in \vB$ and (for either
`plus' or `minus') $\chi(\pm
\widecheck B <\sigma)\phi\in \vB^*_0$ for some $\sigma >0$  and for all
  functions   $\chi(\cdot<\sigma)$ of this type, then
  $\phi=R(\lambda\pm \i 0 )\psi$ (with the same sign).
  \end{thm}

  Here (most likely) $R=R(\lambda)\to \infty$ and $ \sigma=\sigma(\lambda) \to 0$ for $\lambda\downarrow
  \lambda_0$. We state  a new version of the Sommerfeld theorem,
  now at $\lambda=\lambda_0$, but otherwise under the same
  conditions. (A one-body version of the theorem at zero energy is given in
  \cite{Sk4}.) Recall for comparison the quantity $ \vE^\vG _{-s_0,0}$ of
  Lemma \ref{lem:eigentransform},
  \begin{align*}
    \vE^\vG _{-s_0,0} =\{u\in \vB^*_{s_0,0}\mid
      (H-\lambda_0)u=0,\quad\Pi' u\in \vB^*_{0}\}.
  \end{align*} Let
  \begin{align*}
    \widecheck B_\rho=r^{\rho/4}\,\widecheck
  B\,r^{\rho/4}; \quad  r=r_{R},\, \widecheck B= B_R,\, R=R_0,
  \end{align*}
 recalling  that
  $R_0\geq 1$ is chosen in agreement with our version of the  Mourre
  estimate  at  $\lambda_0$   (as done
in Subsection \ref{subsec:LAP bound}).

\begin{thm}[Sommerfeld's  theorem at threshold]\label{prop:Sommerfeld}  There exists
   $\sigma>0$ such that   for
    any real function
  $\chi(\cdot<\sigma)\in C^\infty(\R)$,
    which is   supported in $(-\infty,\sigma)$ and whose derivative
  has
  compact support,   and for any $\psi\in \vB_{s_0}$, the function
  $\phi=\phi^\pm=R(\lambda_0\pm \i 0)\psi\in  \vB^*_{s_0}$ obeys
  \begin{align}\label{eq:45lem5}
    \begin{split}
\Pi'\phi &\in  \vB^*,\\
\chi(\pm
\widecheck B<\sigma)\Pi'\phi &\in  \vB^*_{0},\\
\chi(\pm
     \widecheck B_\rho<\sigma)\Pi\phi &\in  \vB^*_{s_0,0},\\
(H-\lambda_0)\phi &=\psi.
    \end{split}
\end{align}

Conversely suppose   that (for either
`plus' or `minus')
    $\phi\in  \vB^*_{s_0}$ fulfills \eqref{eq:45lem5}
 for some $\sigma >0$,  for all
  functions   $\chi(\cdot<\sigma)$ of this type and for  a given
  $\psi\in \vB_{s_0}$, then  $\phi=R(\lambda_0\pm \i 0 )\psi$  (with the same sign).
  \end{thm}

  \begin{proof}  We shall only
    consider the  case of `$+$'.  By Proposition
    \ref{prop:microLoc2}, for some $\sigma>0$
\begin{align}\label{eq:45lem5a}
     \chi( \widecheck  B<\sigma)\breve R(\lambda_0+ \i 0) \Pi'\psi\in  \vB^*_0.
   \end{align} (This is for $\lambda_0> \Sigma_2$; if $\lambda_0=
   \Sigma_2$ the statement is trivial.) Note that $\Pi'\psi\in \vB$ and
   that \eqref{eq:45lem5a} holds with $\widecheck B $ replaced by
   $B_{\kappa}=\widecheck B(\kappa^2 \widecheck B^2+1)^{-1}$ ($\kappa>0$
   small). Then \eqref{eq:45lem5a}  follows by using another function
   of the same type, say $\tilde \chi(\cdot<\sigma)$, { such that }
    \begin{align*}
    \chi(b<\sigma)=\chi(b<\sigma)\tilde \chi(b/(\kappa^2b^2+1)< \sigma).
  \end{align*}

By using \cite[Theorems 4.1 and
 4.2]{FS}  one  can show  (here omitting the argument) that for all small $\sigma>0$
 \begin{align}
   \label{eq:microFS}
   \begin{split}
    \inp{x_{a}}^{t-s_0}\chi\parb{
     \widecheck B_\rho&<\sigma}S r^+_{\lambda_0} \inp{x_{a}}^{-t-s_0-\epsilon}f\in  L^2;\\
\quad &f\in
   L^2,
\, t,\epsilon>0.
   \end{split}
 \end{align}

We use  \eqref{eq:resolBAS}  writing
\begin{align}\label{eq:deuseful}
  \phi^+=R(\lambda_0+ \i 0 )\psi=S\phi_{a}^+(\lambda_0) +\breve
  R(\lambda_0+ \i 0 )\psi_a^+(\lambda_0).
\end{align}
 Now
   \eqref{eq:45lem5} follows by using
  \eqref{eq:45lem5a} and   \eqref{eq:microFS}   to treat the second  and
  the first terms of   \eqref{eq:deuseful}, respectively.

  To show the second assertion (the uniqueness part) note that $\phi^+=R(\lambda_0 + \i 0)\psi$ is a particular solution of  this problem. Whence
    we may assume that $\psi=0$.
     Due to Theorems \ref{thm:negat-effect-potent} and
\ref{thm:negat-effect-potents}  it suffices to
show that $\phi\in \vE^\vG _{-s_0,0} $ (here we use the same notation
if $\lambda_0=
   \Sigma_2$), and whence that
\begin{align}
  \label{eq:2ident}
  \phi\in  \vB^*_{s_0,0}, \quad\Pi' \phi\in \vB^*_{0}.
\end{align}

Introduce a smooth quadratic partition of unity $1=
\chi(\cdot<\sigma)^2+\chi(\cdot>\sigma)^2$ such that
$\chi(b<\sigma)=1$ for $b<\sigma/2$. Let for $R\geq 1$ the function $\chi_R=\chi_R(r)$ be
given by \eqref{eq:14.1.7.23.24}.
Abbreviating
$\theta_R=\sqrt{-\chi_R'}$  we can then estimate
  \begin{align*}
    -&\inp{\chi'_R(r)}_{\Pi'\phi}-\inp{r^{-\rho/2}\chi'_R(r)}_{\Pi\phi}\\
&\leq \tfrac 2{\sigma} \parb{s_R+t_R}+v_R;\\
s_R&=\inp{\widecheck B}_{\chi(\widecheck B>\sigma)\theta_R\Pi'\phi},\\
    t_R&=\inp{\widecheck  B_\rho}_{\chi(\widecheck B_\rho>\sigma)r^{-\rho/4}\theta_R\Pi\phi},
  \end{align*} where $v_R\to  0$ for $R\to \infty $. Next we write
  \begin{align*}
    t_R&=\tilde t_R+\tilde v_R;\\
    \tilde t_R&=\inp{\widecheck  B}_{\chi(\widecheck B_\rho>\sigma)\theta_R\Pi\phi},
\end{align*} and note that  $\tilde v_R=\vO(R^{\rho/2-1})\to 0$
  for $R\to \infty$. After further  commutation we (should) obtain  that
  \begin{align*}
    -&\inp{\chi'_R(r)}_{\Pi'\phi}-\inp{r^{-\rho/2}\chi'_R(r)}_{\Pi\phi}\\
&\leq \tfrac 2{\sigma} \parb{\inp{\widecheck
    B}_{\theta_R\Pi'\phi}+\inp{\widecheck  B}_{\theta_R\Pi\phi}}+o(R^0)\\
&=-{\sigma}^{-1} \inp{\i [H,\chi_R ]}_{\phi}+o(R^0),\\
&=0+o(R^0),
\end{align*} yielding
  \eqref{eq:2ident} and therefore the uniqueness part. However we need to
  argue for the validity of the above estimates. In the second step we
  used that
  \begin{align*}
    \Re{\Pi\theta_R\widecheck B \theta_R \Pi'}=\vO(R^{-2}),
  \end{align*} cf. \eqref{eq:cross}.
In the first step we
  used that
  \begin{align*}
    s_R-\inp{\widecheck B}_{\theta_R\Pi'\phi}&=-\inp{\widecheck B}_{\chi(\check
    B<\sigma)\theta_R\Pi'\phi}=o(R^0),\\
\tilde t_R-\inp{\widecheck  B}_{\theta_R\Pi\phi}&=o(R^0).
  \end{align*} The first bound is easy since $\widecheck B \Pi'\phi\in \vB^*$ (the
  latter seen by an energy bound).

To get the second  bound  it
  suffices to show that  also $\widecheck B_\rho \Pi\phi\in \vB_{s_0}^*$.
 Let us first
  note that due to the assumption $\chi(\widecheck B<\sigma)\Pi'\phi \in
  \vB^*_{0}$   we can  verbatim  use     Step V of  the proof
     of Lemma \ref{lem:eigentransform} to conclude  that  $E^+_\vH
     (\lambda_0)f
     =0$, $f=T^*\phi\in \vB^*_{s_0}(\bX_a)$.  Next  we decompose
\begin{align}\label{eq:redB}
  \begin{split}
    \widecheck B_\rho \Pi\phi&= \widecheck B_\rho ST^*\phi=S\widecheck B_{a,\rho}
  f+\hat \phi,\\
    &\widecheck B_{a,\rho} =
    r_a^{\rho/4-1/2}\Re \parb{x_{a}\cdot p_{a}}
    r_a^{\rho/4-1/2}; \quad r_a=r(x_{a}).
  \end{split}
\end{align}
  By the  arguments
 of the proof of Lemma \ref{lem:Tcont}
  indeed $\hat \phi\in \vB_{s_0}^*$, so it remains only to show that $ \widecheck B_{a,\rho}
 f \in \vB_{s_0}^*(\bX_a)$. Noting  that the operator $\widecheck
 B_{a,\rho}$ has symbol
 \begin{subequations}
 \begin{equation}\label{eq:gs}
   \check
 b_{a,\rho}\in S\parb{s, g};\quad s=\inp{\xi/g_a}^2,\quad g_a=\sqrt{-w},\quad
   g=\inp{x}^{-2} \d x^2+ g_a^{-2}\d \xi^2,
 \end{equation} it suffices in turn to show that $ g^{-1}_a \Opw (s) g_af\in
\vB_{s_0}^*(\bX_a)$. As in \cite[(4.15)]{DS1}
 \begin{equation}\label{eq:FScomp}
   \Opw (s) - g_a^{-1}hg_a^{-1} -2\in  S(\inp{x}^{\rho-2},g)\subset S(1,g).
 \end{equation}
 \end{subequations} By writing $h=(h+v^+_{\lambda_0})-v^+_{\lambda_0}$
  we end up with bounding $-g^{-2}_a v^+_{\lambda_0} f
$, which clearly  belongs to $\vB_{s_0}^*(\bX_a)$. Whence $ \widecheck B_{a,\rho}
 f \in \vB_{s_0}^*(\bX_a)$ is proven.
\end{proof}

\begin{cor}
  \label{cor:energyPbnd} For  any  $\psi\in \vB$ the function $\phi=R(\lambda_0\pm \i 0 )\psi$
  obeys the bounds
  \begin{equation}
    \label{eq:energ2bd}
    \widecheck B_\rho \Pi \phi,\,\,g_a^{-1}\widecheck B_\rho g_a\Pi \phi\in \vB_{s_0}^*.
  \end{equation}
  \begin{proof}
    We substitute \eqref{eq:mainSplit} and use \eqref{eq:redB} and
    \eqref{eq:FScomp}. Note then that $g^{-2}_ahT^*\phi\in  \vB_{s_0}^*$ since $2s_0>\rho$.
\end{proof}

\end{cor}

\subsection{ Elastic  part of the scattering matrix at $\lambda_0$}\label{subsec: Elastic scattering  at lambda0}

We will to a large degree use \cite{DS1}. We  recall from  \eqref{eq:S} that  $S$ is given in terms
of cluster bound states $\varphi_1, \dots, \varphi_{m_a}$, $m_a=\dim \ker
(H^{a}-\lambda_0)$,  $a=a_0$. The quantity
$\alpha=\alpha_j=(a,\lambda_0,\varphi_j)$, $j\leq m_a
$,  is referred to as a
\emph{channel}.

Let us for convenience here assume $m_a=1$ and denote
$\varphi_1$ by $\varphi_\alpha$ (see Subsecton \ref{subsec:Non-elastic
   scattering} for an example where $m_a=2$). To make contact to \cite{DS1} it is
convenient to  change notation:
Recall that up to a polynomially decreasing potential
$w(x_a)\approx\inp{I_a^{(1)}(\cdot +x_a)}_{\varphi^a}$. Let us now assume that
\begin{align}
   w=V_1+ V_2,
\end{align} where $V_1$ and $V_2$ fulfill the following conditions of
\cite{DS1}. (For Coulomb systems one can take $V_1(x)=-\gamma r^{-1}$
for $r:=\abs{x}\geq 1$ and   $V_2=\vO(r^{-2})$.)

 Let $n=\dim \bX_a$.

\begin{cond}
\label{assump:conditions1}
The function $w$ can be written as a sum of two
real-valued measurable functions, $ w= V_1 + V_2$, such that:
For some $\rho \in (0,2)$ we
have

\begin{enumerate}[\quad\normalfont (1)]

   \item \label{it:assumption1} $V_1$ is a smooth negative function that only
     depends on  the
     radial variable $r$ in the region $r\geq 1$ (that is
     $V_1(x)=V_1(r)$ for $r=|x_a|\geq 1$). There exists $\epsilon_1 > 0$ such
     that $$V_1(r) \leq -\epsilon_1
r^{-\rho},\;r\geq 1.$$
   \item \label{it:assumption2}     For all $\gamma\in \mathbb N_0^{n}$
there exists
$C_{\gamma} >0$ such that
$$
\langle x \rangle^{\rho+|\gamma|} |\partial^{\gamma} V_1(x)|
\leq C_{\gamma}.$$
\item \label{it:assumption3}
     There exists $\tilde\epsilon_1 > 0$ such that
     \begin{equation}
       \label{eq:virialb}
       rV_1'(r)
\leq -(2-\tilde\epsilon_1) V_1(r),\;r\geq 1.
     \end{equation}

\item \label{it:assumption4}
     $V_2=V_2(x)$ is smooth and there exists $\epsilon_2 > 0$ such that for all $\gamma\in \mathbb N_0^{n}$
$$
\langle x \rangle^{\rho+\epsilon_2 +|\gamma|} |\partial^{\gamma} V_2(x)|
\leq C_{\gamma}.$$
\end{enumerate}
\end{cond}

The following condition will be  needed (and imposed) only in the case $V_2\neq
 0$.

\begin{cond}
\label{assump:conditions2}
Let $V_1$ be given as in Condition  \ref
{assump:conditions1}  and $\alpha
:=\tfrac{2}{2+\rho}$. There exists $\bar\epsilon_1 > \max (0,
1-\alpha(\rho+2\epsilon_2))$ such that

\begin{align*}
  &\limsup_{r\to \infty} r^{-1}V_1'(r)\Big (\int^r_1 (-V_1(s))^{-\frac
  {1}{2}}\d s\Big )^2 < 2^{-1}(1-\bar\epsilon_1^2),\\
&\limsup_{r\to \infty} V_1'{}'(r)\Big (\int^r_1 (-V_1(s))^{-\frac
  {1}{2}}\d s\Big )^2 < 2^{-1}(1-\bar\epsilon_1^2).
\end{align*}
\end{cond}

\subsubsection{Scattering for the one-body problem at zero energy, \cite{DS1}}\label{subsubsec: Review of DS1}

We review a number of results from \cite{DS1} valid under Conditions
\ref{assump:conditions1} and \ref{assump:conditions2}. (For a
different approach to one-body scattering theory, see \cite{Is}.)  Recall that
for any $\omega\in \S^{n-1}$, $\lambda\in[0,\infty)$ and $x$ from an
appropriate outgoing/incoming region there exists a solution to the
system of equations
    \begin{align}
      \label{eq:mixed conditions222}
      \begin{split}
\ddot y(t) &=-2\nabla w(y(t)),\\
\lambda&=\tfrac 14\dot y(t)^2 +w(y(t)),\\
y(\pm 1)&=x,\\
\omega&={\pm}\lim_{t\to {\pm}\infty}y(t)/|y(t)|.
      \end{split}
    \end{align}

One obtains a family $y^\pm
(t,x,\omega,\lambda)$ of solutions depending regularily (at least
continuously) on
parameters. Moreover all `scattering orbits' are of this form.
Using these solutions one can construct a solution
$\phi^\pm(x,\omega,\lambda)$ to
the eikonal  equation
\begin{equation}\label{eq:eik2}
  \left(\nabla_x\phi^\pm(x,\omega,\lambda)\right)^2+w(x)=\lambda
\end{equation}
satisfying $\nabla_x\phi^\pm(x,\omega,\lambda)=\tfrac 12\dot
y(\pm1,x,\omega,\lambda)$.

 For $R\geq 1$ and $\sigma\in (0,2)$
 \begin{align}\label{eq:Gamma}
   \begin{split}
   \Gamma^+_{R,\sigma}(\omega)&:=\{y\in {\mathbb R}^n\ |\ y\cdot \omega\geq (1-\sigma)|y|,\;|y|\geq R\};\; \omega\in \S^{n-1},\\
   \Gamma^+_{R,\sigma}&:=\{(y,\omega)\in {\mathbb R}^n \times
   \S^{n-1}\ |\ y \in \Gamma^+_{R,\sigma}(\omega)\}.
  \end{split}
 \end{align}

\begin{lemma}
   \label{lemma:mixed_2}
There exist $R_0\geq 1$ and $\sigma_0\in (0,2)$
such that for
all $R\geq R_0 $ and for all  positive
$\sigma\leq \sigma_0$ the system \eqref{eq:mixed conditions222}
is
solved  for all data
$(x,\omega)\in \Gamma^+_{R,\sigma}$ and $\lambda \geq 0$
by a unique
 function $y^+(t,x,\omega,\lambda),\;t\geq 1$, such that
$y^+(t,x,\omega,\lambda)\in \Gamma^+_{R,\sigma}(\omega)$ for all $t\geq 1$.
Define a vector field $F^+(x,\omega,\lambda)$
 on $\Gamma^+_{R_0,\sigma_0}(\omega)$ by
  \begin{equation}
    \label{eq:vector field}
    F^+(x,\omega,\lambda)=\tfrac 12\dot y^+(t=1,x,\omega,\lambda).
  \end{equation}
Then
\[\rot_x F^+(x,\omega,\lambda)=0.\]
\end{lemma}

 We  define
$\phi^+(x,\omega,\lambda)$ at
 $(x,\omega,\lambda)\in\Gamma^+_{R_0,\sigma_0}\times[0,\infty[$  by
 requiring
 $\nabla_x\phi^+= F^+$ and
 $\phi^+(R_0\omega,\omega,\lambda)=\sqrt {\lambda}R_0 $.  We let
\begin{equation*}
 \phi^-(x,\omega,\lambda):=-\phi^+(x,-\omega,\lambda)\;\text{for }x\in
 \Gamma^-_{R_0,\sigma_0}(\omega) :=\Gamma^+_{R_0,\sigma_0}(-\omega).
\end{equation*}

For $\xi\neq 0$  we write
$\xi=\sqrt{\lambda}\omega$,  $\omega\in \S^{n-1}$, and then
\begin{align*}
  \phi^\pm(x,\xi)=\phi^\pm(x,\omega,\lambda); \quad (x,\omega)\in \Gamma^\pm_{R_0,\sigma_0}.
\end{align*} We are motivated to write, slightly abusely, $(x,\xi)\in \Gamma^\pm_{R_0,\sigma_0}$
instead of $(x,\omega,\lambda)\in \Gamma^\pm_{R_0,\sigma_0}\times
(0,\infty)$  (in fact even for the case $\xi=0$).

Fixing  $0<\sigma<\sigma'<\sigma_0$
 we
introduce  a smoothed out characteristic function
\begin{subequations}
\begin{equation}\label{eq:chi^2}
  \chi_{\sigma,\sigma'}(t)=
\begin {cases} 1, & \text{for}\; t\geq 1-\sigma,\\
0, & \text{for}\; t\leq 1-\sigma'.
\end {cases} \;
\end{equation}
 Next define,
 in terms of \eqref{eq:chi^2} and  the function  $\bar \chi_R=1-\chi_R$ of \eqref{eq:14.1.7.23.24},
\[a_0^\pm(x,\xi):=\chi_{\sigma,\sigma'}(\pm\hat
x\cdot\hat \xi)\bar\chi_{R_0}(\abs{x});
\quad \hat z=z/|z|.\]
 We introduce then a Fourier integral operator $J_0^+$ on $L^2(\R^n)$ by
\begin{equation}
  \label{eq:int_ope9}
  (J_0^\pm f)(x)=(2\pi)^{-n/2}\int \e^{\i
    \phi^\pm(x,\xi)} a_0^\pm(x,\xi) \hat f(\xi)
  \d \xi,
\end{equation}
where
\[\hat f(\xi):=(2\pi)^{-n/2}\int \e^{-\i x\cdot \xi}f(x)\d x\]
denotes the  Fourier transform of $f$.

\end{subequations}

\begin{subequations}

The WKB method suggests to approximate the wave operator by
a Fourier integral operator $J^+$ on $L^2(\R^d)$ of the form
\begin{equation}
  \label{eq:int_ope}
  (J^+f)(x)=(2\pi)^{-n/2}\int \e^{\i
    \phi^+(x,\xi)} a^+(x,\xi) \hat f(\xi)
  \d \xi,
\end{equation}
where the symbol $a^+(x,\xi)$ is
 supported  in  $\Gamma^+_{R_0,\sigma_0}$ and
 constructed  by an iterative procedure (partly recalled in Subsection
 \ref{subsubsec: Elastic scattering at lowest threshold}) attempting to make the difference
 $T^+:=\i (hJ^+-J^+p^2)$  small in $\Gamma^+_{ R_0,\sigma}$. We have
\begin{equation}
  \label{eq:int_ope4}
  (T^+f)(x)=(2\pi)^{-n/2}\int \e^{\i
    \phi^+(x,\xi)} t^+(x,\xi) \hat f(\xi)
  \d \xi,
\end{equation} where
\begin{equation}
  \label{eq:sym_dif}
t^+(x,\xi)= \left((2\nabla_x\phi^+(x,\xi)) \cdot
\nabla_x +(\triangle_x\phi^+(x,\xi))\right)a^+(x,\xi)
 -{\i}\triangle_x a^+(x,\xi).
\end{equation} The symbols  $a^+(x,\xi)$ and $a_0^+(x,\xi)$ coincide
to leading order away from $\xi=0$, more precisely
\begin{align} \label{eq:sym_difbb}
  a^+(x,\xi)\approx \left(\det\nabla_\xi\nabla_x
\phi^+(x,\xi)\right)^{1/2} a_0^+(x,\xi)=\e^{\zeta^+(x,\xi)}a_0^+(x,\xi),
\end{align}
and  $a^+(x,\xi)$ should be thought of as an
`improvement' of  the right-hand side.
 For details of construction, see \cite[Section 5]{DS1}. However,
 since we are partly going to  mimic this construction in Subsection
 \ref{subsubsec: Elastic scattering at lowest threshold}, let us
 here recall that the equation
\begin{equation}
  \label{eq:sym_difb2}
\left((2\nabla_x\phi^+(x,\xi)) \cdot
\nabla_x +(\triangle_x\phi^+(x,\xi))\right)\e^{\zeta^+(x,\xi)}=0,
\end{equation} takes \eqref{eq:sym_dif} onto the form
\begin{align}\label{eq:Bfac}
  \begin{split}
 t^+&=\e^{\zeta^+} \parb{(2\nabla_x\phi^+) \cdot
\nabla_x-\i A^+}b^+,\quad b^+=\e^{-\zeta^+}a^+,\\
A^+&= \Delta +2(\nabla \zeta^+)\cdot \nabla +(\Delta \zeta^+)+ (\nabla \zeta^+)^2.
  \end{split}
\end{align}

\end{subequations}

Similar to \eqref{eq:int_ope}--\eqref{eq:sym_dif}  we introduce a   Fourier integral
operator $J^-$ and $T^-$ as
\begin{align}
  \label{eq:int_ope2}
  \begin{split}
  (J^-f)(x)&=(2\pi)^{-n/2}\int \e^{\i\phi^-(x,\xi)}a^-(x,\xi) \hat
  f(\xi) \d \xi,\\
(T^-f)(x)&=(2\pi)^{-n/2}\int \e^{\i
    \phi^-(x,\xi)} t^-(x,\xi) \hat f(\xi)
  \d \xi.
  \end{split}
\end{align}

 For all $\tau\in L^2(\S^{n-1})$ we introduce
\begin{align}
  \begin{split}
 (J^{\pm}(\lambda)\tau)(x)&=
(2\pi)^{-n/2}
\int
  \e^{\i \phi^{\pm}(x,\omega,\lambda)}\tilde a^{\pm}(x,\omega,\lambda)
  \tau(\omega)\d \omega,\label{eq:jdef}\\
(T^{\pm}(\lambda)\tau)(x)&=
(2\pi)^{-n/2}
\int
  \e^{\i \phi^{\pm}(x,\omega,\lambda)}\tilde t^{\pm}(x,\omega,\lambda)
                                                            \tau(\omega)\d \omega,
  \end{split}
\end{align}
where
\begin{eqnarray*} \tilde a^{\pm}(x,\omega,\lambda)&=&
\tfrac{\lambda^{(n-2)/4}}{\sqrt 2}  a^{\pm}(x, \sqrt{\lambda}\omega),\\
 \tilde t^{\pm}(x,\omega,\lambda)&=&\tfrac{\lambda^{(n-2)/4}}{\sqrt 2}
t^{\pm}(x,\sqrt{\lambda}\omega).\end{eqnarray*}

The functions $\tilde a^{\pm}$ and $\tilde t^{\pm}$ are \emph{continuous} in
$(x,\omega,\lambda)\in \R^d\times \S^{n-1}\times [0,\infty)$, and
therefore  we can \textit {define} $J^{\pm}(\lambda)$ and $T^{\pm}(\lambda)$
at $\lambda=0$ by the expressions (\ref{eq:jdef}). These properties
hinge on \cite[Proposition 5.3]{DS1} stating properties of the function
\begin{equation*}
  \tilde{\zeta}^+(x,\omega,\lambda)={\zeta}^+(x,\sqrt{\lambda}\omega)-
  \ln \lambda^{(2-n)/4};\quad \lambda>0.
\end{equation*} In particular it follows that  there exist locally
uniform  limits  (along with derivatives)
\begin{equation}\label{eq:tildez}
  \tilde{\zeta}^+(x,\omega,0)=\lim_{\lambda\to
  0_+}\tilde{\zeta}^+(x,\omega,\lambda).
\end{equation}
It will be
convenient to use a   splitting  $T^{\pm}
(\lambda)=T^{\pm}_\bd(\lambda)+T^{\pm}_\pr (\lambda)$  in agreement
with a  certain decomposition of $\tilde t^{\pm}(x,\sqrt{\lambda}\omega)$,
see Lemma \ref{lem:22a} and references given before the lemma.
 There are \emph{wave operators}
\begin{align}  W^{\pm}f=\lim _{t\to {\pm}\infty}\e^{\i th}J_0^{\pm}\e^{-\i tp^2}f=
\lim _{t\to {\pm}\infty}\e^{\i th}J^{\pm}\e^{-\i tp^2}f;\;
  \hat f\in C^\infty_\c(\R^n\setminus\{0\}). \label{eq:ooo1}
\end{align}
 The  two operators $W^\pm$ extend   isometrically  on
 $L^2(\R^d)$ with extensions satisfying  $hW^\pm=W^\pm p^2$. Moreover,
\begin{align}
  0=\lim _{t\to {\mp}\infty}\e^{\i th}J_0^{\pm}\e^{-\i tp^2}f
=\lim _{t\to {\mp}\infty}\e^{\i th}J^{\pm}\e^{-\i tp^2}f;\;
  \hat f\in C^\infty_\c(\R^n\setminus\{0\}).
\label{eq:ooo}\end{align}
 For $\rho\in (1/2,2)$ we may write $W^{\pm}=W_{\rm dol}^{\pm}\,\e^{\i
   \psi_{\rm \dol}^\pm(p)}$ in terms of the familiar Dollard wave operators
   \cite{Do} (cf. \eqref{modified-WO} in Section \ref{total cross-sections})  and  explicit real  momentum-depending
   phase factors $\psi_{\rm dol}^\pm$, see \cite [Theorem  6.15]{DS1}.

Let $\Delta_\omega$ denote the Laplace-Beltrami operator on the sphere
$\S^{n-1}$. For  $k\in{\mathbb R}$ we define the Sobolev spaces on the sphere
$H^k(\S^{n-1})=(1-\Delta_\omega)^{-k/2}L^{2}(\S^{n-1})$.  Let $\vL_s^k=\vL_s^k(\R^n)={\mathcal
 L}(H^k(\S^{n-1}),L_s^{2}(\R^n))$  and
$\vL^k=\vL_0^k$ for any $k,s\in \R$.

For $\lambda>0$  we introduce the restricted Fourier transform ${\mathcal
F}_0(\lambda)$  as
\begin{equation}\label{eq:resFour}
{\mathcal F}_0(\lambda)f(\omega)=\tfrac{\lambda^{(n-2)/4}}{\sqrt 2} \hat f(\sqrt
{\lambda}\omega).
\end{equation}
Let  $s>\frac12$ and $k\geq0$.
Note that  ${\mathcal
F}_0(\lambda)\in \mathcal L
(L^{2}_{s+k}(\R^n),H^k(\S^{n-1}))$
  with a continuous  dependence   on $\lambda>0$.
Likewise,
  ${\mathcal F}_0(\lambda)^*\in
\mathcal L ^{-k}_{-s-k}$ with a continuous  dependence    on
$\lambda>0$. Note also that the operator
 \begin{equation}\label{decom1}
 \int_{\R_+}^ \oplus {\mathcal F}_0(\lambda)\, \d
\lambda:L^2(\R^n)\to \int_{\R_+}^\oplus L^2(\S^{n-1})\;\d \lambda
 \end{equation}   is unitary, and
  consequently that it  diagonalizes the operator $p^2$.
Formally, we have
$J^\pm(\lambda)=J^\pm{\mathcal F}_0(\lambda)^*$ and
$T^\pm(\lambda)=T^\pm{\mathcal F}_0(\lambda)^*$. The formal identity
 $W^{\pm}(\lambda)=W^{\pm}{\mathcal F}_0(\lambda)^*=(J^{\pm}+\i r(\lambda{\mp}\i0)
T^{\pm}){\mathcal F}_0(\lambda)^*$ leads us then to consider the \emph{wave matrices}
\begin{equation}\label{wave22}
W^{\pm}(\lambda):=J^{\pm}(\lambda)+\i r(\lambda{\mp}\i0)
T^{\pm}(\lambda),\end{equation}
 which in fact belong to $
\mathcal L ^{-k}_{-s}$  for any
$k\geq 0$ for  a suitable $s=s(k)>s_0$. In this space
$W^{\pm}(\lambda)$ have  continuous dependence of $\lambda\geq 0$
(including $\lambda=0$!).

The scattering operator commutes with $p^2$, which is diagonalized by the
direct integral mapping  (\ref{decom1}).
Because of
 that the general theory of decomposable operators yields a
measurable family $\R_+\ni\lambda\mapsto S(\lambda)$  with the
\emph{scattering  matrix} $S(\lambda)$
    being a unitary operator on  $L^2(\S^{n-1})$ for almost all
    $\lambda$, and such
    that  in terms of  the mapping (\ref{decom1})
\begin{equation}
  S\simeq\int_{\R_+}^\oplus S(\lambda)\;\d \lambda.\label{decom2}
\end{equation}
A main result of \cite{DS1} reads, that for $\lambda\geq 0$ the
scattering  matrix
  \begin{align}
  \label{eq:Smatrix}
                        \begin{split}
                          S(\lambda)
&= -2\pi J^+(\lambda)^{*}T^-(\lambda)
   +2\pi \i T^{+}(\lambda)^*r(\lambda+\i 0)T^-(\lambda)\\
&=-2\pi W^{+}(\lambda)^*T^-(\lambda)
                        \end{split},
\end{align}
defining  a unitary operator $L^2(\S^{n-1})$  with a
strongly continuously dependence  on $\lambda \geq0$. Moreover (\ref{decom2}) is true,
 and
\begin{equation}\label{eq:1Sbnd}
\forall  k\in \R \,\forall \epsilon >0:
\quad S(\lambda)\in {\mathcal L}(H^k(\S^{n-1}),
  H^{k-\epsilon}(\S^{n-1})),
\end{equation}
  depending norm-continuously on
$\lambda\geq 0$.
Hence in particular $S(\lambda)$ maps $C^\infty(\S^{n-1})$ into
itself.

Another main  result of \cite{DS1} adopted to the  setting discussed here (in
particular not including    a certain singular term $V_3$) is the following result:

Suppose in addition to Conditions \ref{assump:conditions1} and
\ref{assump:conditions2} that  $V_1(r)= -\gamma r^{-\rho}$
 for $r\geq 1$. Then
 the kernel $S(0)(\omega,\omega')$ is smooth
  outside the set $\{(\omega,\omega')\mid\omega\cdot \omega'= \cos
  \frac \rho {2-\rho}\pi\}$.

 \subsubsection{Elastic scattering for the $N$-body problem at $\lambda_0$}\label{subsubsec:Back to
  the N-body problem}
Using the constructions $J_0^{\pm}$  in the previous subsection  for
the potential $w=V_1+V_2$ in the variable $x_a$ and recalling
$-E_{\vH} (\lambda\pm \i 0) =h+ v^\pm_\lambda +\lambda_0-\lambda$, $h=p^2+w$, we introduce  wave operators
\begin{align}
  \begin{split}
    \lim _{t\to {\pm}\infty}\e^{\i
  tH}\parb{1\otimes  J_0^{\pm}}&\e^{-\i t H_a}Sf=\lim _{t\to {\pm}\infty}\e^{\i
  tH}\parb{1\otimes  J_0^{\pm}}S\e^{-\i t (p_a^2+\lambda_0)}f;\\\quad
  &\hat f\in C^\infty_\c(\bX_a\setminus\{0\})^{m_a}. \label{eq:ooo12}
  \end{split}
\end{align}
We consider a  channel $\alpha=(a,\lambda_0,\varphi_\alpha)$
(recall that  this means that $H^a\varphi_\alpha=\lambda_0\varphi_\alpha$
with
$\norm{\varphi_\alpha}=1$)
assuming  for simplicity from this point that
\begin{equation}\label{eq:multI}
  m_a=\dim \ker
(H^{a}-\lambda_0)=1
\end{equation}
 (making $\varphi_\alpha$ essentially unique). With this assumption  the above limit is nothing but the \emph{channel
wave operator}
\begin{align}
  \begin{split}
    W_\alpha^{\pm}f&=\lim _{t\to {\pm}\infty}\e^{\i
  tH}\parb{1\otimes  J_0^{\pm}}\varphi_\alpha\otimes \e^{-\i t (p_a^2+\lambda_0)}f\\&=\lim _{t\to {\pm}\infty}\e^{\i
  tH}\parb {\varphi_\alpha\otimes J_0^{\pm}\e^{-\i t (p_a^2+\lambda_0)}f} ;\quad
  \hat f\in C^\infty_\c(\bX_a\setminus\{0\}). \label{eq:ooo13}
  \end{split}
\end{align}

This leads us to define
\begin{align*} {\mathcal
    F}_{\lambda_0}(\lambda)f(\omega)&=\tfrac{(\lambda-\lambda_0)^{(n-2)/4}}{\sqrt
    2} \hat f\parb{\sqrt
                                      {\lambda-\lambda_0}\,\omega},\\
  (J_N^{\pm}(\lambda)\tau)(x)&= (2\pi)^{-n/2} \int \e^{\i
                               \phi^{\pm}(x,\omega,\lambda-\lambda_0)}\tilde
                               a^{\pm}(x,\omega,\lambda-\lambda_0)
                               \tau(\omega)\d \omega,\\
  (T_N^{\pm}(\lambda)\tau)(x)&= (2\pi)^{-n/2} \int \e^{\i
                               \phi^{\pm}(x,\omega,\lambda-\lambda_0)}\tilde
                               t^{\pm}(x,\omega,\lambda-\lambda_0)
                               \tau(\omega)\d \omega,\\
  J_\alpha^{\pm}f&=  {\varphi_\alpha}\otimes J^{\pm}f;\quad f=f(x_a),\\
  T_\alpha^{\pm}&=\i\parb{H J_\alpha^{\pm}-J_\alpha^{\pm}(p_a^2+\lambda_0)},\\
  J_\alpha^{\pm}(\lambda)\tau&= {\varphi_\alpha}\otimes
                               J_N^{\pm}(\lambda)\tau,\\
  T_\alpha^{\pm}(\lambda)\tau&= {\varphi_\alpha}\otimes
                               T_N^{\pm}(\lambda)\tau + \widecheck
                               T_\alpha^{\pm}(\lambda)\tau;\\ &
                               \quad\widecheck T_\alpha^{\pm}(\lambda)=\i
                                                                I^{(2)}_aJ_\alpha^{\pm}(\lambda)+
                                                                \i \parb{I^{(1)}_a-w}J_\alpha^{\pm}(\lambda),\\
  W^{\pm}_\alpha(\lambda)&=J^{\pm}_\alpha(\lambda)+\i
                           R(\lambda{\mp}\i0)
                           T^{\pm}_\alpha(\lambda),\\
  S_{\alpha\alpha}(\lambda) &= -2\pi
                              J^+_\alpha(\lambda)^{*}T^-_\alpha(\lambda)
                              +2\pi \i
                              T^{+}_\alpha(\lambda)^*R(\lambda+\i
                              0)T^-_\alpha(\lambda).
\end{align*}
 Here formally
 \begin{align*}
   J_N^{\pm}(\lambda)=J^{\pm}(\lambda-\lambda_0)= J^{\pm}{\mathcal F}_{\lambda_0}(\lambda)^*\mand T_N^{\pm}(\lambda)=T^{\pm}(\lambda-\lambda_0)= T^{\pm}{\mathcal F}_{\lambda_0}(\lambda)^*,
 \end{align*} where $J^{\pm}$ is the  `improvement' of $J_0^{\pm}$ and
 $T^{\pm}=\i(hJ^{\pm}-J^{\pm}p_a^2)$
  as defined by \eqref{eq:int_ope}--\eqref{eq:sym_dif}  and \eqref{eq:int_ope2}.

 The first  term of $\widecheck T_\alpha^{\pm}(\lambda)$
 has arbitrary  polynomial decay, for example stated precisely as
 $\widecheck T_\alpha^{\pm}(\lambda)\in
\vL_s^k=\vL_s^k(\bX):={\mathcal
 L}(H^k(\S^{n-1}),L_s^{2}(\bX))$ for any $k,s\in \R$, and this  is also the case for
 $\Pi\parb{I^{(1)}_a-w}J_\alpha^{\pm}(\lambda)$. But
 $\Pi'\parb{I^{(1)}_a-w}J_\alpha^{\pm}(\lambda)=\Pi'I^{(1)}_a\Pi
 J_\alpha^{\pm}(\lambda)$ is only one power better than $w$,
  more precisely it has the form
  $\vO(\inp{x}^{-1-\rho})J_\alpha^{\pm}(\lambda)$. Thus we can
   record
  \begin{align}
    \label{eq:smallT}
    \Pi\widecheck T_\alpha^{\pm}(\lambda)=\vO\parb{\inp{x}^{-\infty}},\quad  \Pi'\widecheck T_\alpha^{\pm}(\lambda)=\vO\parb{\inp{x}^{-1-\rho}}J_\alpha^{\pm}(\lambda).
  \end{align}
A
  similar remark is due  for the `non-restricted'
  quantity $\widecheck T_\alpha^{\pm}$
  in the formula
  \begin{align*}
    T_\alpha^{\pm}= {\varphi_\alpha}\otimes T^{\pm}
+ \widecheck T_\alpha^{\pm};\quad \widecheck T_\alpha^{\pm}=I^{(2)}_aJ_\alpha^{\pm}+ \parb{I^{(1)}_a-w}J_\alpha^{\pm}.
  \end{align*}
  Moreover, since for any $f$ with $\hat f\in
  C^\infty_\c(\bX_a\setminus\{0\})$ the quantity $\norm {T_\alpha^{\pm}\e^{-\i t
      (p_a^2+\lambda_0)}f}$ is integrable at $\pm \infty$, the Cook
  argument gives the existence of the wave operator
  $W_\alpha^{\pm}$. Note that this integrability may be shown by a
  stationary phase argument (for example by   using
  \eqref{eq:partitionleft} and a version of  Lemma \ref{lem:22a} \ref{it:p40}).

  Using \cite[Appendix A]{DS1} the elastic part of the scattering
  matrix defined by  \eqref{eq:ooo13} may be shown to be  given by the expression
  $S_{\alpha\alpha}(\lambda)$, $\lambda>\lambda_0$, introduced above. We will
  study some properties of this operator, which is an operator on
  $\vL(\S^{n-1})$ with norm at most one.

  We state  some  basic properties of $J^{\pm}(\lambda)$
and $T^{\pm}(\lambda)=T^{\pm}_\bd(\lambda)+T^{\pm}_\pr (\lambda)$, see
\cite[(5.8),  (5.16), (5.19)  and  Lemmas 6.8 and 6.9]{DS1} and the proof of \cite[Theorem 6.11]{DS1} (which we adapt to the present problem). First we recall some
notation. The most basic one is the function
\begin{align*}
  g(r)=g_\lambda=\sqrt{\lambda-\lambda_0-V_1(r)},\quad \lambda\geq\lambda_0,
\end{align*} which roughly controls the momentum. Next  we introduce
the
 symbols
\begin{align}
\label{eq:def_a0_b}
a(x,\xi) = \frac{\xi^2}{g_{\lambda}(|x|)^2},\quad
b(x,\xi)  =  \frac{\xi}{g_{\lambda}(|x|)} \cdot F(x),
\end{align} where $F$  is an arbitrary (henceforth fixed) vector field on $\R^n$
extending $F(x)=\hat x= x/r$ for $r=\abs{x}\geq1$. Of course the symbols $a$ and $b$ also have $\lambda$-dependence, but
for convenience  this is  here and henceforth  omitted  in the notation.
 Let
$\tilde{\chi}_{-},\tilde{\chi}_{+}\in C^\infty$ be  non-negative
functions obeying
$\tilde{\chi}_{-}+\tilde{\chi}_{+}=1$ and
\begin{subequations}
\begin{align}
  \label{eq:supp tilde chi1}
 &\supp \tilde{\chi}_{-}
 \subseteq(-\infty, 1-\bar \sigma],\\&\supp \tilde{\chi}_{+} \subseteq
 [1-2\bar \sigma,\infty),  \label{eq:supp tilde chi2}
\end{align}
where the number $\bar \sigma>0$ needs to be taken sufficiently small,
depending on the parameter $\sigma$ used in the previous subsection
(see \eqref{eq:chi^2}) and properties of
the phase $\phi^{\pm}(x,\xi)$.  Let ${\chi}_{-},{\chi}_{+}\in C^\infty$ be  non-negative
functions obeying
${\chi}_{-}+{\chi}_{+}=1$ and
\begin{align}
  \label{eq:supp tilde chi1x}
 &\supp {\chi}_{-}
 \subseteq(-\infty, 2),\\&\supp {\chi}_{+} \subseteq
 (1,\infty).  \label{eq:supp tilde chi2x}
\end{align}
 Introduce then symbols
\begin{align}\label{eq:supp tilde chi3}
  \begin{split}
\chi_1&= {\chi}_{+}(a),\\
\chi_2^{\pm}&={\chi}_{-}(a) \tilde{\chi}_{-}(\pm b),\\
  \chi_3^{\pm}&={\chi}_{-}(a) \tilde{\chi}_{+}(\pm b). \end{split}
\end{align} These symbols belong to a class of (parameter-depending)
pseudodifferential operators studied in \cite{FS, DS1}. The `Planck
constant' for this class is $\inp{x}^{-1}g_{\lambda}(|x|)^{-1}$, in
particular at most  $\inp{x}^{\rho/2-1}$.
  \end{subequations} Note the partition of unity in terms
  of corresponding (right-quantized) operators
\begin{align}
  \label{eq:partitionleft}
  1 =\Opr(\chi_1)+\Opr(\chi^\pm_2)+\Opr(\chi^\pm_3).
\end{align}
Recall  $\vL^k:={\mathcal
 L}(H^k(\S^{n-1}),L^{2}(\R^n))$, $k\in \R$.

\begin{lemma}\label{lem:22a}  Let  $\chi_1$,
$\chi^\pm_2$ and $\chi^\pm_3$ be  given by
\eqref{eq:supp tilde chi3}.
\begin{subequations}
  \begin{enumerate}[\normalfont 1)]
\item \label{it:p20}For  all  $k\in [0,\infty)$  and  $\epsilon>0$,
\begin{equation}
  \label{eq:x_weights29bc9}
 (\langle x \rangle g_\lambda)^{-k} \langle x \rangle^{-1/2 -\epsilon}
  g_\lambda^{1/ 2}J_N^{\pm}(\lambda)
\end{equation} is a continuous $\vL^{-k}$--valued  function of
$\lambda\in[\lambda_0,\infty)$. With a bounding constant
  independent of $\lambda\geq \lambda_0$,
 \begin{equation*}
    g_\lambda^{1/2}J_N^{\pm}(\lambda)\in \mathcal{L} \parb{L^2(\S^{n-1}),\vB_{ 1/2}^*(\R^n) }.
  \end{equation*}
\item \label{it:p30}For  all $k\in \R$  and  $\epsilon>0$,
\begin{equation}
  \label{eq:x_weights29c}
 (\langle x \rangle g_\lambda)^{-k} \langle x \rangle^{1/2 -\epsilon} g_\lambda^{-1/2}\Opr(\chi^\pm_2)T_{\bd}^{\pm}(\lambda-\lambda_0)
\end{equation} is a continuous $\vL^{-k}$--valued  function of
$\lambda\in[\lambda_0,\infty)$.
\item \label{it:p40}For all $k,m\in \R$,
\begin{equation}
  \label{eq:x_weights29d}
  \langle x \rangle^m \Opr(\chi_1+\chi^\pm_3) T_{\bd}^\pm(\lambda-\lambda_0),
  \,\,\langle x \rangle^m T_{\pr}^\pm(\lambda-\lambda_0)\mand \langle x \rangle^m \Opr(\chi_1) J_N^\pm(\lambda)
\end{equation} are   continuous $\vL^{-k}$--valued function of
$\lambda\in[\lambda_0,\infty)$.
\end{enumerate}
\end{subequations}
\end{lemma}

This lemma will be used in combination with the following excerpts of \cite[Proposition 4.1]{DS1}
(adapted to the present problem). Note that
$\Opl(\chi^\pm_2)=\Opr(\chi^\pm_2)^*$ is given by left-quantization.

\begin{lemma}\label{lemma:reso-1-body} Let $\Lambda$ denote any
  interval of the form $\Lambda=[\lambda_0, \lambda_0']$, and let
  $r^\pm_\lambda=r(\lambda-\lambda_0\pm\i 0)$ for $\lambda\in\Lambda$. Then the
  following bounds hold unifomly in $\lambda\in\Lambda$, and the
  corresponding $\vL\parb{L^2(\R^n)}$--valued functions are continuous.
  \begin{subequations}
  \begin{enumerate}[\normalfont 1)]
\item \label{it:p10a}
For all $\epsilon> 0$
there exists
$C>0$ such that
\begin{align}
  \label{eq:x_weights2}
\|\langle x \rangle^{-\epsilon-1/2} g_\lambda^{1/ 2}r^\pm_\lambda  g_\lambda^{1/ 2}\langle x \rangle^{-\epsilon-1/2}
  \| \leq C.
\end{align}
\item \label{it:p20a} For all
$s \geq   0$ and  $0\leq\epsilon<\epsilon'$  there exists $C > 0$
\begin{align}
  \label{eq:PsDO_part21}
&\|(\langle x \rangle g_\lambda)^{s}\langle x \rangle^{\epsilon-1/2} g_\lambda^{1/ 2}\Opl(\chi^\pm_2)
r^\pm_\lambda g_\lambda^{1/ 2}\langle x \rangle^{-\epsilon'-1/2}
(\langle x \rangle g_\lambda)^{-s}\| \leq C.
\end{align}
\end{enumerate}
\end{subequations}
\end{lemma}

Now taking $\tau\in H^k(\S^{n-1})$ with  a sufficiently big $k$
(actually any $k>0$ suffices), we
will show
 by combining \eqref{eq:resolBAS}  with  Lemmas \ref{lem:22a} and \ref{lemma:reso-1-body}
that $S_{\alpha\alpha}(\lambda)\tau$ is a well-defined element of
$L^2(\S^{n-1})$, in fact with a continuous dependence of $\lambda\in
I_\delta^+$.

We insert \eqref{eq:resolBAS}
into the formula
\begin{align}\label{eq:S1}
  S_{\alpha\alpha}(\lambda)\tau = -2\pi
  J^+_\alpha(\lambda)^{*}T^-_\alpha(\lambda)\tau +2\pi \i
  T^{+}_\alpha(\lambda)^*R(\lambda+\i 0)T^-_\alpha(\lambda)\tau.
\end{align}

Ignoring the
contribution from $\widecheck T^{-}_\alpha(\lambda)$ (its  contribution
is  a `partial smoothing operator'
  as exemplified  in  the beginning of the
  proof of Theorem \ref{thm:ScatN}) we obtain by using   Lemma
  \ref{lem:22a}  that the first term $-2\pi
J^+_\alpha(\lambda)^{*}T^-_\alpha(\lambda)\tau\approx -2\pi
J_N^+(\lambda)^{*}T_N^-(\lambda)\tau \in L^2(\S^{n-1})$.

For
the second term in \eqref{eq:S1} we (again) ignore
  terms containing  $\widecheck T^{\pm}_\alpha(\lambda)$ and consider only
${\varphi_\alpha}\otimes T_N^{\pm}(\lambda)$ using  \eqref{eq:partitionleft}
to
 write
\begin{align*}
T_N^{\pm}(\lambda)=\Opr(\chi_1+\chi^\pm_3)T_N^{\pm}(\lambda)+
\Opr(\chi^\pm_2)T_N^\pm(\lambda)= T_{1,3}^\pm(\lambda)+ T_2^\pm(\lambda).
\end{align*}  By Lemma \ref{lem:22a} \ref{it:p40}  the first term has
strong decay, so let us consider the seemingly worse term given (up to a
constant) by
\begin{align*}
 \parbb{\varphi_\alpha\otimes &T_2^+(\lambda)}^* R(\lambda+\i 0)
 \varphi_\alpha\otimes T_2^-(\lambda)\\
&=T_2^+(\lambda)^*S^* R(\lambda+\i 0)S
 T_2^-(\lambda).
\end{align*}  The contribution from the
'leading  terms', cf. \eqref{eq:mainSplit},  are
\begin{align*}
  T_2^+(\lambda)^*r^+_\lambda
 T_2^-(\lambda)\mand
T_2^+(\lambda)^*S^* \breve R(\lambda+\i 0)\Pi'S
 T_2^-(\lambda),
\end{align*} respectively. The second term vanishes. For the first
term we
use Lemma \ref{lemma:reso-1-body}~\ref{it:p20a} for   the case of `+' and  see that indeed $T_2^+(\lambda)^*r^+_\lambda
 T_2^-(\lambda)\tau\approx T_N^+(\lambda)^*r^+_\lambda
 T_N^-(\lambda)\tau$ is a well-defined element of
$L^2(\S^{n-1})$.

Of course there are other terms to consider, and for some of those
also Lemma \ref{lemma:reso-1-body}~\ref{it:p10a} is needed.  We can
check all other terms (cf. the proof of Theorem \ref{thm:ScatN}) and see
that they are well-defined with a continuous dependence of
$\lambda$. Furthermore we can write the formula for the action by
$S_{\alpha\alpha}(\lambda)$ as
\begin{eqnarray}
  \label{eq:SmatrixN}
 S_{\alpha\alpha}(\lambda)\tau
=-2\pi W_\alpha^{+}(\lambda)^*T_\alpha^-(\lambda)\tau,\label{eq:SgodN}
\end{eqnarray}  and this is a  continuous  $L^2(\S^{n-1})$--valued function of $\lambda\in
I^+_\delta$ for $\tau\in H^k(\S^{n-1})$, $k>0$.

Next we will examine the degree of regularity that is needed on $\tau$,
measured by the size of  $k$. Let
\begin{align*}
  S_w(\lambda-\lambda_0)
&= -2\pi J_N^+(\lambda)^{*}T_N^-(\lambda)
   +2\pi \i T_N^{+}(\lambda)^*r(\lambda+\i 0)T_N^-(\lambda);\quad
  \lambda\geq \lambda_0.
\end{align*}  Note that $S_w(\cdot)$  is the scattering
matrix for the one-body problem given by \eqref{eq:Smatrix}.
We will
examine the quantity
\begin{align*}
  \widecheck S_{\alpha\alpha}(\lambda)\tau=S_{\alpha\alpha}(\lambda)\tau-
  S_w(\lambda-\lambda_0) \tau.
\end{align*} By the above
preliminary investigation the term $S_w(\cdot)$
is the `leading term', and we know
 that $k=0$ works for this term
although $k=\epsilon$ for any $\epsilon>0$ is neeeded for
operator-continuity, see \eqref{eq:1Sbnd} and the discussion there. So
we expect well-definedness of $\widecheck
S_{\alpha\alpha}(\lambda)\tau$ and   continuous dependence of $\lambda$
 with  less regularity imposed on  $\tau$. In fact we can show
that $k=-\epsilon$ for a computable $\epsilon>0$ works for this term,
even for   operator-continuity. More generally  we have the following
result given in terms of any $k\geq 0$ satisfying  one of the options
\begin{align}\label{eq:minC0}
  \begin{cases}
   & k+1/2+ (1/2-k)\rho/2< 1/2+\rho/2,\quad \,k\in [0,1/2],\\
&k+1/2< 1/2+\rho/2,\quad \,k>1/2,
  \end{cases}
  \end{align} or equivalently stated, one of the options
\begin{align}\label{eq:minC02}
  \begin{cases}
   & 0\leq k\leq 1/2\quad\mand \quad k< \rho(4-2\rho)^{-1},\\
&1/2<k<\rho/2,\quad \,\rho>1.
  \end{cases}
  \end{align}

  \begin{thm}\label{thm:ScatN} Suppose \eqref{eq:multI}, i.e. $m_a=1$. For any $k\geq 0$ obeying \eqref{eq:minC02} the operator
  \begin{align*}
    \widecheck S_{\alpha\alpha}(\lambda)=S_{\alpha\alpha}(\lambda)- S_w(\lambda-\lambda_0)\in {\mathcal L}\parb{H^{-k}(\S^{n-1}),
  H^k(\S^{n-1})}
\end{align*} with a continuous dependence of $\lambda\in I_\delta^+$.
\end{thm}
\begin{remarks}\label{remarks:inelastic-scattering-n}
  \begin{enumerate}[1)]
  \item The number $s=1/2+\rho/2$ on the
  right-hand side in \eqref{eq:minC0} is the number appearing in
  \eqref{eq:strongConst}. The left-hand side
  comes from estimating the factor $\langle x \rangle ^{k+1/2}
  g^{k-\frac 12}$; it is  uniformly bounded by the  power $C{\langle x
    \rangle}^c$ with $c$ given as the expression appearing to the left in
  \eqref{eq:minC0}. If we replace $ I_\delta^+$ by  $
  I_\delta^+\setminus\set{\lambda_0}$  in Theorem \ref{thm:ScatN}  any $k< \rho/2$ suffices.
\item \label{item:multCase} We assumed $m_a=1$. However the interested
  reader may check that if $m_a>1$, then Theorem \ref{thm:ScatN}
  remains valid.
 Moreover the off-diagonal elements  of
  the scattering matrix,  denoted by $S_{\beta\alpha}(\lambda)$ with
  $\alpha\neq \beta$, fulfill the same assertion as the one for $
  \widecheck S_{\alpha\alpha}(\lambda)$ in Theorem \ref{thm:ScatN}.
\item \label{item:multCase3} Under the additional condition that $\lambda_0=\Sigma_2$ one
  can show that $\widecheck S_{\alpha\alpha}(\lambda )$ is bounded on any of the
  spaces $H^{l}(\S^{n-1})$, $l\in \R$,  with a continuous dependence of
  $\lambda\in I_\delta^+$, in fact it is partially
  smoothing. More precisely one can show in this case that  for all
  $l\in \R$ the
  operator $ \widecheck S_{\alpha\alpha}(\lambda)\in {\mathcal L}\parb{H^{l-k},
  H^{l+k}}$  with a continuous dependence on
  $\lambda\in I_\delta^+$. This is for any $k\geq 0$ obeying
\eqref{eq:minC02}, in fact for a computable bigger $k$ using
  \eqref{eq:strongConst} below
  for $s=1+\rho$ (recall from   the discussion after
  \eqref{eq:eff_v} that the latter  boundedness   condition is fulfilled for
  $\lambda_0=\Sigma_2$). The proof  consists of combining ideas from Subsection \ref{subsubsec: Elastic
    scattering at lowest threshold} with \cite[Proposition 4.1]{DS1}, however  we shall not
  elaborate. In particular since $S_w(\lambda-\lambda_0)$ has a similar
  property up a loss of an `$\epsilon$-smoothness', cf. \cite[Theorem 7.2]{DS1}, we
  conclude (more precisely stated) that for all $l\in \R$ and
  $\epsilon>0$ the operator $ S_{\alpha\alpha}(\lambda)\in {\mathcal L}\parb{H^{l},
  H^{l-\epsilon}}$  with a continuous dependence of
  $\lambda\in I_\delta^+$. Note that
  this implies that $S_{\alpha\alpha}(\lambda)\tau\in
  C^\infty(\S^{n-1})$ for $\lambda\in I_\delta^+$ and   $\tau\in C^\infty(\S^{n-1})$.
\end{enumerate}
\end{remarks}
\begin{proof} [Proof of Theorem \ref{thm:ScatN}]

  We need to treat the term $-2\pi
  J^+_\alpha(\lambda)^{*}\widecheck T^-_\alpha(\lambda)$ left out when
  discussing $-2\pi J^+_\alpha(\lambda)^{*}T^-_\alpha(\lambda)$
  above. (Note that $-2\pi
  J^+_\alpha(\lambda)^{*}\parb{{\varphi_\alpha}\otimes T_N^{\pm}(\lambda)}$
  already is subtracted in the definition of $\widecheck
  S_{\alpha\alpha}(\lambda)$.) Writing $J^+_\alpha(\lambda)^{*}\widecheck
  T^-_\alpha(\lambda)=J^+_\alpha(\lambda)^{*}\Pi \widecheck
  T^-_\alpha(\lambda)$ we can invoke \eqref{eq:smallT}.

  The remaining  terms of $\widecheck
S_{\alpha\alpha}(\lambda)$
  fall into three  disjoint groups according to whether  there is a dependence of:
\begin{enumerate}[\bf a)]
\item  \label{item:15} $\breve R(\lambda+\i 0))$, and no dependence of $r_\lambda^+\parb{1+v^+_\lambda r^+_\lambda}^{-1}$.
\item  \label{item:16}
  $r_\lambda^+\parb{1+v^+_\lambda r^+_\lambda}^{-1}$, and no dependence of $\breve R(\lambda+\i
  0))$.
\item \label{item:18} $\breve R(\lambda+\i
  0))$  as well as a dependence of  $r_\lambda^+\parb{1+v^+_\lambda r^+_\lambda}^{-1}$.
\end{enumerate}

\subStep{\ref{item:15}} We need to treat
\begin{align*}
  T^+_\alpha(\lambda)^* \breve R(\lambda+\i
  0)\Pi'T^-_\alpha(\lambda)=\widecheck T^+_\alpha(\lambda)^* \Pi'\breve R(\lambda+\i
  0)\Pi'\widecheck T^-_\alpha(\lambda).
\end{align*}
 Due to \eqref{eq:smallT}   we need
 to bound
\begin{align}\label{eq:Jbnd}
   \inp{x}^{-1/2-\rho+\epsilon}J_\alpha^{\pm}(\lambda)\in \vL^{-k}=\vL^{-k}(\bX)
   \quad \text{ for some }\epsilon>0.
 \end{align}
 By  Lemma \ref{lem:22a} \ref{it:p20} this requires
\begin{align}\label{eq:minC03}
  \begin{cases}
   & k+ (1/2-k)\rho/2<\rho,\quad \,k\in(0,1/2],\\
&k< \rho,\quad \,k>1/2,
  \end{cases}
  \end{align} which is weaker than \eqref{eq:minC02}.

\subStep{\ref{item:16}}  We need to treat
\begin{align*}
  &T^+_\alpha(\lambda)^*Sr^+_\lambda\parb{1+v^+_\lambda r^+_\lambda}^{-1}S^*T^-_\alpha(\lambda)-T_N^+(\lambda)^*r^+_\lambda
  T_N^-(\lambda)\\
&=\parbb{T^+_\alpha(\lambda)^*Sr^+_\lambda S^*T^-_\alpha(\lambda)-T_N^+(\lambda)^*r^+_\lambda
  T_N^-(\lambda)}\\
&-T^+_\alpha(\lambda)^*Sr^+_\lambda \parb{1+v^+_\lambda r^+_\lambda}^{-1}v^+_\lambda r^+_\lambda S^*T^-_\alpha(\lambda)
\end{align*} Due to \eqref{eq:smallT}  the first term is a smoothing
operator, i.e. in ${\mathcal L}(H^{-l}(\S^{n-1}),
  H^{l}(\S^{n-1}))$ for any $l\in \R$.   We shall frequently in the
  rest of the proof use
\begin{align}
  \label{eq:parUniy}
  1=\inp{x}^{-s}\inp{x}^{s};\quad s=1/2+\rho/2.
\end{align} For the second term it suffices, due to
\eqref{eq:fredInverse} and \eqref{eq:parUniy}, to bound
\begin{align}\label{eq:cruBnd}
  \inp{x}^{-s}r^{\mp} _\lambda S^* T_\alpha^{\pm}(\lambda)\in
  \vL^{-k}.
 \end{align}
 Here we used that
 \begin{align}
   \label{eq:strongConst}
    \quad v^+_\lambda \in\vL(L^2_{-s},L^2_s ).
  \end{align} (We could   do better using \eqref{eq:multI}, cf. the discussion after \eqref{eq:eff_v},
  but we prefer to use a method that more or less obviously  generalizes to
  the case where $m_a>1$, c.f. Remark
  \ref{remarks:inelastic-scattering-n} \ref{item:multCase}.) To show
  \eqref{eq:cruBnd} it suffices due to \eqref{eq:smallT} and  Lemma
  \ref{lem:22a} \ref{it:p40} to bound
\begin{align}\label{eq:cruBnd2}
  \inp{x}^{-s}r^{\mp} _\lambda  T_2^{\pm}(\lambda)\in  \vL^{-k},
 \end{align} or equivalently that
\begin{align}\label{eq:cruBnd3}
   T_N^{\pm}(\lambda)^*\Opl(\chi_2^\pm) r^{\pm} _\lambda\in
  \vL\parb{L_s^2,H^k(\S^{n-1})}.
 \end{align} Due to  Lemma
  \ref{lem:22a}   it suffices in  turn to show that
  \begin{align}
    \label{eq:next}
    (\langle x \rangle g_\lambda)^{k} \langle x \rangle^{\epsilon-1/2 }
    g_\lambda^{1/2}\Opl(\chi_2^\pm) r^\pm_\lambda \in
  \vL\parb{L_s^2,L^2}\,\text { for some }\epsilon>0.
  \end{align}
  We need to bound
\begin{align*}
    (\langle x \rangle g_\lambda)^{k} \langle x \rangle^{\epsilon-1/2 } g_\lambda^{1/2}\Opl(\chi_2^\pm) r^\pm_\lambda g_\lambda^{1/2}&\langle x \rangle^{-1/2-2\epsilon}
(\langle x \rangle g_\lambda)^{-k}\\
& \parbb{g_\lambda^{-1/ 2}\langle x \rangle^{1/2+2\epsilon}
(\langle x \rangle g_\lambda)^{k}   \inp{x}^{-s}}\in \vL(L^2).
  \end{align*}
 Using \eqref{eq:PsDO_part21} we need to check  that the last factor is
bounded. This amounts to
\begin{align}\label{eq:minC04}
  \begin{cases}
   & k+1/2+2\epsilon+ (1/2-k)\rho/2\leq s,\quad \,k\in(0,1/2],\\
&k+1/2+2\epsilon\leq s,\quad \,k>1/2,
  \end{cases}
  \end{align} which indeed is fulfilled for small  $\epsilon>0$ thanks
   to \eqref{eq:minC0}.

\subStep{~\ref{item:18}} There are four  terms to be considered  given (up to the
common  factor
$2\pi\i$) by:
\begin{align*}
  &S_1=-T^+_\alpha(\lambda)^*Sr^+_\lambda\parb{1+v^+_\lambda  r^+_\lambda}^{-1}S^*I_a\breve R(\lambda+ \i
    0)\Pi'T^-_\alpha(\lambda),\\
&S_2=T^+_\alpha(\lambda)^* \breve R(\lambda+\i
  0)\Pi' I_aSr^+_\lambda\parb{1+v^+_\lambda r^+_\lambda}^{-1}S^* I_{a}\breve R(\lambda+ \i 0
    )\Pi' T^-_\alpha(\lambda),\\
&S_3=-T^+_\alpha(\lambda)^* \breve R(\lambda+\i
  0)\Pi' I_aSr^+_\lambda S^* T^-_\alpha(\lambda),\\
&S_4=T^+_\alpha(\lambda)^* \breve R(\lambda+\i
  0)\Pi' I_aSr^+_\lambda\parb{1+v^+_\lambda  r^+_\lambda}^{-1}
  v^+_\lambda  r^+_\lambda S^*  T^-_\alpha(\lambda).
\end{align*}

   For $S_1$ we insert the decomposition  \eqref{eq:parUniy}   to the right of
the factor $r^+_\lambda$ (the left one). Suppose we can show that
\begin{align*}
  \inp{x}^{s} S^*I_a\breve R(\lambda\pm  \i
    0)\Pi'T^\mp_\alpha(\lambda)\in
  \vL^{-k},
\end{align*} then we are done using \eqref{eq:cruBnd} and
\eqref{eq:minC04}. But due  to \eqref{eq:smallT}  this term has the form
\begin{align}\label{eq:bndGG}
  \vO(\inp{x}^{s-1-\rho})\breve R(\lambda\pm  \i
    0)\Pi'\widecheck T^\mp_\alpha(\lambda)=\vO(\inp{x}^{s-1-\rho})\breve R(\lambda\pm \i
    0)\vO(\inp{x}^{-1-\rho})J^\mp_\alpha(\lambda),
\end{align} for which we can apply
\eqref{eq:Jbnd}. The argument for  $S_3$ is the same.

We treat $S_2$ by  inserting
\eqref{eq:parUniy}  to the right as well as
to the left of the factor $r^+_\lambda$ (the left one) and then use
\eqref{eq:bndGG} and Lemma \ref{lemma:reso-1-body} \ref{it:p10a} (to
bound $\inp{x}^{-s}r^+_\lambda \inp{x}^{-s}$).

We treat $S_4$ by inserting \eqref{eq:parUniy} to the right of the far
left factor $r^+_\lambda$ and to the left of the far right factor
$r^+_\lambda$. Then we invoke \eqref{eq:cruBnd} and
\eqref{eq:strongConst}.

\end{proof}

\subsubsection{Elastic scattering  at
  $\lambda_0$, a `geometric'  approach}\label{subsubsec:geometric approach}
In the spirit of \cite[Section 8]{DS1} we will give a `geometric' description  of  the operator
$S_{\alpha\alpha}(\lambda)$ studied in Subsection \ref{subsubsec:Back to
  the N-body problem}. To keep
the discussion   short we shall  consider the limiting case $\lambda=\lambda_0$
only. For analogue results for (all most all) non-threshold energies in a  general
$N$-body setting we refer to \cite{Sk6}.

We recall
the following construction for  the one-body problem, see
\cite[Proposition 5.6]{DS1}, which could be a basis for discussing
$\lambda>\lambda_0$ also.

\begin{lemma}\label{lemma:1-body-phi} There exist  $R\geq R_0$ and $\tilde
  \sigma\in (0,\sigma_0]$ such that for all $x=|x|\hat x\in \R^n$ with
  $|x|\geq R$ and
  $\lambda\geq \lambda_0$  there exists a
  unique $\omega\in \S^{n-1}$ satisfying $\omega\cdot \hat x\geq
  1-\tilde\sigma$ (equivalently, $x\in \Gamma^+_{R,\tilde
    \sigma}(\omega)$) and
  $\partial_{\omega}\phi^+(x,\omega,\lambda-\lambda_0)=0$.  We introduce
  the notation $\omega_{\crt}^+=\omega_{\crt}^+(x,\lambda)$ for this
  vector. It is smooth in $x$,  and for some $\breve \epsilon =\breve \epsilon(\rho,\bar
  \epsilon_1,\epsilon_2)>0$
\[\partial ^\gamma_x(\omega_{\crt}^+-\hat x)=\vO(|x|^{-\breve \epsilon-|\gamma|}).\]
 Let
 \begin{equation}
   \label{eq:newph}
   \phi(x,\lambda)=\phi^+(x,\omega_\crt^+(x,\lambda),\lambda-\lambda_0).
 \end{equation}
This function  solves  the eikonal equation
\begin{equation*}
(\partial_x \phi(x,\lambda))^2+w(x)=\lambda-\lambda_0;\quad  |x|\geq
R\mand \lambda\geq \lambda_0.
\end{equation*}
In the spherically symmetric case (viz. $V_2(x)=V_2(r)$) we have $\omega_\crt^+=\hat x$ and
\begin{equation}\label{eq:eik222}
  \phi(x,\lambda)=\phi_\sph(x,\lambda):=\int_{R_0}^{|x|}\sqrt{\lambda-\lambda_0-w(r)}\d
r +\sqrt{\lambda-\lambda_0}R_0.
\end{equation}
\end{lemma}

  With reference to Theorem  \ref{prop:Sommerfeld} let $\chi(\cdot<\sigma)$ be one of the
functions described  there  (with  $\sigma>0$ sufficiently  small) but taken
with  the additional property  that  $\chi(\cdot <\sigma)=1$ on
$(-\infty, \sigma/2)$. Let $\widecheck B$ and $\widecheck B_\rho$  be given as in  Theorem
\ref{prop:Sommerfeld}, and let
\begin{align*}
 \mathcal V^+_{-s_0,\sigma}(\lambda_0)&=\{u\in \vB^*_{s_0}(\bX) \mid
      \Pi' u\in \vB^*, \quad \chi(
\widecheck B<\sigma)\Pi'u\in \vB^*_{0},\quad(H-\lambda_0)u=0\},\\
 \vR^+_\sigma
 &=\{u\in \vB^*_{s_0,0}(\bX) \mid \Pi' u\in \vB^*, \quad \chi(
\widecheck B<\sigma)\Pi'u\in \vB^*_{0}\}.
\end{align*} We  use the notation
$\hat y=y/|y|$ for (nonzero) vectors $y\in \bX_a$  as well
as
 \begin{align*}
   u_a^+(y)&=c_n\,g_{\lambda_0}^{-{1}/{2}}(|y|)\,|y|^{-\tfrac{n-1}{2}}
   \,\e^{\i\phi(y,\lambda_0)},\quad u_a^-(y)=\overline{u_a^+(y)};\quad c_n= \e^{\i
  \pi\tfrac{1-n}{4}}(4\pi)^{-{1}/{2}},\\
& v_a^{\pm}(y)=\pm\inp
  {y}^{\rho/2}g_{\lambda_0}(|y|)u_a^\pm(y).
 \end{align*}

\begin{thm}\label{thm:repEigen}
\begin{subequations}
  \begin{enumerate}[\normalfont 1)]
\item\label{it:onto}
 The channel wave matrix
 \begin{align*}
   W_\alpha^-(\lambda_0):\,L^2(\S^{n-1})\to
   \mathcal V^+_{-s_0,\sigma}(\lambda_0)\subset  \vB^*_{s_0}
 \end{align*}
  is a
   well-defined   bicontinuous isomorphism. (In particular the space
   $\mathcal V^+_{-s_0,\sigma}(\lambda_0)$ does not depend on the small $\sigma>0$.)
\item \label{it:uniq2}
 For all  $\tau\in L^2(\S^{n-1})$ the vectors
 $\tau^+=S_{\alpha\alpha}(\lambda_0)\tau $ and
 $\phi^+_\tau=W_\alpha^-(\lambda_0)\tau $ are  the unique
 vectors  in $L^2(\S^{n-1})$ and $\mathcal V^+_{-s_0,\sigma}(\lambda_0)$,
 respectively, fulfilling
\begin{align} \label{eq:asymp}
    \phi^+_\tau(x)
-{\varphi_\alpha}(x^a)\parbb{u_a^-(x_a)\tau(-\widehat{x_a})+u_a^+(x^a)\tau^+(\widehat{x_a})}\in
\vR^+_\sigma.
  \end{align}
\item \label{it:ivBder} For all  $\tau\in
  L^2(\S^{n-1})$, and with  $\tau^+=S_{\alpha\alpha}(\lambda_0)\tau $
  and $\phi^+_\tau=W_\alpha^-(\lambda_0)\tau $ as above,
\begin{align} \label{eq:asympBder}
  \begin{split}
    \parb{\widecheck B_\rho \Pi \phi^+_\tau}(x)
-{\varphi_\alpha}(x^a)\parbb{v_a^-(x_a)\tau(-\widehat{x_a})+v_a^+(x^a)\tau^+(\widehat{x_a})}\in \vB^*_{s_0,0}.
  \end{split}
  \end{align}
\end{enumerate}
\end{subequations}
\end{thm}

\begin{proof}
  \subStep {I} We insert \eqref{eq:resolBAS} into  the definition
  \begin{align*}
    W^{-}_\alpha(\lambda)&=J^{-}_\alpha(\lambda)+\i R(\lambda{+}\i0)
T^{-}_\alpha(\lambda),
  \end{align*} leading to a study of
  \begin{subequations}
  \begin{align}\label{eq:modW}
    \widecheck
    W^{-}_\alpha(\lambda):=W^{-}_\alpha(\lambda)-{\varphi_\alpha}\otimes
    W^-_w(\lambda-\lambda_0),
  \end{align} where $W^-_w(\cdot)$ is the incoming wave matrix for the
  one-body problem as
  discussed in Subsection \ref{subsubsec: Review of DS1}. We
  take $\lambda=\lambda_0$. It is checked as in the proof of Theorem
  \ref{thm:ScatN}  (i.e.  by using Lemmas
  \ref{lem:22a} and \ref{lemma:reso-1-body})  and by using Proposition
  \ref{prop:microLoc2}, that for any $k\geq 0$ fulfilling
  \eqref{eq:minC02}
  \begin{align}\label{eq:phsPro}
    \begin{split}
     &\widecheck W^{-}_\alpha(\lambda_0)H^{-k} \subset \vB^*_{s_0},\quad
    \Pi' \widecheck W^{-}_\alpha(\lambda_0) H^{-k}\subset \vB^*,\\&
    \chi( \widecheck B<\sigma)\Pi' \widecheck W^{-}_\alpha(\lambda_0)
    H^{-k}\subset \vB^*_{0}, \quad \chi(\widecheck
    B_\rho<\sigma)\Pi\widecheck W^{-}_\alpha(\lambda_0)H^{-k} \subset
    \vB^*_{s_0,0}.
    \end{split}
\end{align}
\end{subequations} (For the last property we can use the argument
\eqref{eq:Bgennem} given below.) In particular \eqref{eq:phsPro}
hold for $k= 0$. Combining the first three assertions of
\eqref{eq:phsPro}  with \cite[Theorem
8.2]{DS1} (containing information on $W^-_w(0)$) we conclude that the
range of $ W^{-}_\alpha(\lambda_0)$ is a subset of
$\mathcal V^+_{-s_0,\sigma}(\lambda_0)$, i.e. that the map in \ref{it:onto}
is a well-defined  map. By the same argument it follows that this  map
is continuous.

  \subStep {II} Let any $\tau\in
  L^2(\S^{n-1})$ be given. By using the splitting \eqref{eq:modW}, Proposition
  \ref{prop:microLoc2}, \cite[Theorem 8.2]{DS1} and \cite{Sk5} we
  deduce that \eqref{eq:asymp} is fulfilled  with
  $\phi^+_\tau=W_\alpha^-(\lambda_0)\tau $ and for some  $\tau^+\in
  L^2(\S^{n-1})$.

 Using the formula
\begin{align*}
  S_{\alpha\alpha}(\lambda_0)^*= 2\pi W^{-}_\alpha(\lambda_0)
  ^{*}T^+_\alpha(\lambda_0)=2\pi \i\,
  W^-_\alpha(\lambda_0)^{*}(H-\lambda_0)J^+_\alpha(\lambda_0)
\end{align*}
 and  \cite[Theorem 5.7]{DS1} we
calculate for any $\tilde \tau\in
C^\infty(\S^{n-1})$
\begin{align*}
  \inp{\tilde\tau, S_{\alpha\alpha}(\lambda_0)\tau}&=-2\pi \i \lim_{n\to
    \infty}\inp{J^+_\alpha(\lambda_0) \tilde \tau, [H,\chi_n
    (r)]W^-_\alpha(\lambda_0)\tau}\\
&=-4\pi  \lim_{n\to
    \infty}\inp{J^+_\alpha(\lambda_0) \tilde \tau,
    g_{\lambda_0}\chi_n' W^-_\alpha(\lambda_0)\tau}\\
&=\inp{\tilde \tau, \tau^+}.
\end{align*}
We conclude that $\tau^+=S_{\alpha\alpha}(\lambda_0)\tau $. (Note that with
our normalization there is an extra factor $1/\sqrt 2$ in
\cite[(3.13) and Theorem 5.7]{DS1}.)

\subStep {III} For uniqueness, suppose that  for  a given $\tau\in
L^2(\S^{n-1})$ the formula \eqref{eq:asymp} is fulfilled with
$\tau^+=\tau_1$ and $\phi^+_\tau=\phi^+_1$
as well as for
$\tau^+=\tau_2$ and $\phi^+_\tau=\phi^+_2$, then by Theorem
\ref{prop:Sommerfeld}  (applied with $\psi=0$)
$\phi^+_1=\phi^+_2$,  which in turn implies that $\tau_1= \tau_2$. The injectivity part of the assertion
\ref{it:onto} follows similarly.

\subStep {IV} We   show that the map in \ref{it:onto} maps
onto (this finishes \ref{it:onto} by the open mapping theorem).  So
let $u\in\mathcal V^+_{-s_0,\sigma}(\lambda_0)$ be given, then we
need to find $\tau\in L^2=L^2(\S^{n-1})$ such that  $u=
W^{-}_\alpha(\lambda_0)\tau$. Due to the assumption $\chi(\widecheck B<\sigma)\Pi'u \in
  \vB^*_{0}$   we can  verbatim  use     Step \textit{V} of  the proof
     of Lemma \ref{lem:eigentransform} to conclude  that  $E^+_\vH f
     =0$, $f=T^*u\in \vB_{s_0}^*(=\vB^*_{s_0}(\bX_a) )$. The vector $\tilde
     f:=(1+r_{\lambda_0}^+v_{\lambda_0}^+)f\in  \vB_{s_0}^*$ fulfills $h\tilde f=0$. Now from \cite [Theorem 8.2]{DS1} we know that the map
 $W^-_w: L^2 (\S^{n-1})\to \set{\tilde f\in \vB^*_{s_0} \mid h\tilde f=0}$ is
 bijective. Whence  $\tilde
     f= W^-_w(0) \tau$ for some $\tau\in L^2$. We
 want to show that $\check u:=u-W^-_\alpha(\lambda_0)\tau=0$. To do
 this  we
 check the conditions of Theorem  \ref{prop:Sommerfeld} (with $\psi=0$). Since
 $u\in\mathcal V^+_{-s_0,\sigma}(\lambda_0)$ also
 $\check u\in\mathcal{V}^+_{-s_0,\sigma}(\lambda_0)$,  and it
 suffices to check that $\chi(\widecheck B_\rho<\sigma)\Pi\check u\in \vB^*_{s_0,0}$.
   Using \eqref{eq:modW} and \eqref{eq:phsPro} we calculate modulo $\vB^*_{s_0,0}$
 \begin{align*}
   \chi(\widecheck B_\rho<\sigma)\Pi\check u&\approx \chi(\widecheck
   B_\rho<\sigma)S f-\chi(\widecheck B_\rho<\sigma)\Pi{\varphi_\alpha}\otimes
    W^-_w(0)\tau\\
&= \chi(\widecheck
   B_\rho<\sigma)S \tilde f -\chi(\widecheck B_\rho<\sigma) S
    W^-_w(0)\tau+ \chi(\widecheck
   B_\rho<\sigma)S (f-\tilde f)\\
&= \chi(\widecheck
   B_\rho<\sigma)S (f-\tilde f).
 \end{align*}  We let $\widecheck B_{a,\rho} $ be given as in
 \eqref{eq:redB} and then
   write, substitute  and estimate
 \begin{align*}
   \chi(\widecheck
   B_\rho<\sigma)S -S\chi(\widecheck
   B_{a,\rho}<\sigma) =-\int_{\C} \widecheck R_\rho(z)[\widecheck
   B_\rho S-S\widecheck B_{a,\rho}]
  \widecheck R_{a,\rho}(z)\d \mu_\chi(z)
 \end{align*}
 by the  arguments
 of the proof of Lemma \ref{lem:Tcont},  estimating
   \begin{align}\label{eq:Bgennem}
   \chi(\widecheck
   B_\rho<\sigma)S (f-\tilde f) \approx -S\chi(\widecheck
   B_{a,\rho}<\sigma)r_{\lambda_0}^+v_{\lambda_0}^+f\approx 0;
 \end{align} in the last step we used   Lemma
   \ref{lemma:reso-1-body}.
  Whence $\check u=0$, and we are done.

  \subStep {V} It remains to show \ref{it:ivBder}. First we note that
  $\cB_\rho \Pi W_\alpha^-(\lambda_0)\tau\in \vB^*_{s_0}$, cf. the
  proofs of Theorem  \ref{prop:Sommerfeld} and Corollary
  \ref{cor:energyPbnd}  (and
  \eqref{eq:Bgennem}).  Writing $u=W_\alpha^-(\lambda_0)\tau$ and
  $f=T^*u\in \vB^*_{s_0}$ we know that $E^+_\vH f
     =0$ and $\cB_{a,\rho} f\in \vB^*_{s_0}$,  and we need to show that
\begin{align} \label{eq:asympBder2}
  \begin{split}
 {\widecheck B_{a,\rho} f}-f_2
\in \vB^*_{s_0,0};\quad f_2={v_a^-(x_a)\tau(-\widehat{x_a})+v_a^+(x^a)\tau^+(\widehat{x_a})}.
  \end{split}
  \end{align} From \eqref{eq:asymp} it follows  that
\begin{align} \label{eq:asympBder3}
  \begin{split}
 f-f_1
\in
\vB^*_{s_0,0};\quad f_1=u_a^-(x_a)\tau(-\widehat{x_a})+u_a^+(x^a)\tau^+(\widehat{x_a}).
  \end{split}
  \end{align}
  With  symbols $s$  and  $g_a$ given  as in  \eqref{eq:gs}
 we can use \eqref{eq:FScomp}   to
  write
  \begin{align*}
    \tilde f:=f-Af\in \vB^*_{s_0,0};\quad A=2g^{-1}_a\Opw (s^{-1})g_a.
  \end{align*}
   From the
  proof of Theorem  \ref{prop:Sommerfeld} it follows that
  $\widecheck B^2_{a,\rho} f\in \vB^*_{s_0}$ and therefore we know
  that  $ \widecheck B^2_{a,\rho}
  \tilde f\in \vB^*_{s_0}$,   and using this fact we can show that
  also $\widecheck  B_{a,\rho}\tilde f\in \vB^*_{s_0,0}$.
    Estimating
  \begin{align*}
    R^{-2s_0}\norm {\chi_R \widecheck  B_{a,\rho}\tilde f}^2\leq
    R^{-2s_0}\norm {\chi_R \tilde f}\,\norm {\chi_R \widecheck  B_{a,\rho}^2\tilde f}+o(R^0)=o(R^0),
  \end{align*} the assertion follows.

  Next we substitute $ f=Af+ \tilde f$ into \eqref{eq:asympBder2} and
  calculate modulo $\vB^*_{s_0,0}$ using \eqref{eq:asympBder3}
  \begin{align*}
    \widecheck  B_{a,\rho} f\approx  \widecheck B_{a,\rho}A f\approx \widecheck B_{a,\rho}A
    f_1\approx A \widecheck  B_{a,\rho}f_1\approx A f_2=f_2-\tilde f_2;\quad
    \tilde f_2=f_2-Af_2.
  \end{align*} We calculate using \eqref{eq:FScomp}
  \begin{align*}
    \tilde f_2\approx \tfrac 12 A\parbb{g^{-1}_a\Opw
      (s)g_af_2-2f_2}\approx \tfrac 12 A g^{-2}_ah f_2
  \end{align*}
  It remains to show that $Ag^{-2}_ah f_2\in \vB^*_{s_0,0}$. We approximate $\tau$ and $\tau^+$
  by sequences of smooth functions in $L^2(\S^{n-1})$ giving
  convergence in  $\vB^*_{s_0}$ since
  $Ag^{-2}_ah\in \vL( \vB^*_{s_0})$. Whence we can assume that
  $\tau,\tau^+\in C^\infty$, which allow us to compute
  $g^{-2}_a hf_2\in \vB^*_{s_0,0}$ and we are done.

\end{proof}

The following bound is a consequence of the definition of wave
operators and Theorem \ref{thm:ScatN}, but for completeness of
presentation we give an independent stationary proof.

\begin{cor}\label{cor:norm_less} Under the same  conditions as above $\norm{S_{\alpha\alpha}(\lambda_0)}\leq 1$.
  \end{cor}
  \begin{proof} Following the proof of Theorem  \ref{prop:Sommerfeld}
    we calculate, abbreviating $\phi=W_\alpha^-(\lambda_0)\tau$ for
    any given  $\tau\in
  L^2(\S^{n-1})$,
\begin{align*}
0=-\tfrac12\inp{\i [H,\chi_R ]}_{\phi}&={\inp{\widecheck
    B}_{\theta_R\Pi'\phi}+\inp{\widecheck  B}_{\theta_R\Pi\phi}}+o(R^0)\\
    &\geq -\Re \inp{\chi'_R\inp {x}^{-\rho/2}\Pi\phi, {\widecheck
        B_\rho}\Pi \phi}+o(R^0)
\end{align*} By \eqref{eq:asymp}
\begin{align*}
    \Pi \phi
-{\varphi_\alpha}(x^a)\parbb{u_a^-(x_a)\tau(-\widehat{x_a})+u_a^+(x^a)\tau^+(\widehat{x_a})}\in
\vB^*_{s_0,0}.
\end{align*}  We insert this and \eqref{eq:asympBder} and take $R\to
\infty$, yielding $0\geq -\norm {\tau}^2+\norm
{S_{\alpha\alpha}(\lambda_0)\tau}^2$.
    \end{proof}

Let us for any $\tau\in
  C^\infty(\S^{n-1})$ write, using  \eqref{eq:resolBAS} and \eqref {eq:deuseful},
\begin{align*}
  R(\lambda_0+ \i 0 )\psi&=S\phi_{a}^+(\lambda_0) +
  R'(\lambda_0+ \i 0 )\psi_a^+(\lambda_0) \text{ with input }
  \psi:=T^{-}_\alpha(\lambda_0)\tau\in \vB_{s_0},\\
\delta(H'-\lambda_0)&=
\tfrac1{2\pi \i}\parb{R'(\lambda_0+ \i 0 )-R'(\lambda_0 -\i 0 )}\in\vL(\vB,\vB^*).
\end{align*}
\begin{cor}\label{cor:norm_less22} Under the above  conditions and
  with $\psi_a^+(\lambda_0)$ defined by \eqref{eq:resolBAS} with
  $\psi=T^{-}_\alpha(\lambda_0)\tau$ for any  given $\tau\in
  C^\infty(\S^{n-1})$, or alternatively given by
  \begin{align*}
    \psi_a^+(\lambda_0)&= \Pi'(H-\lambda_0) \Pi' R(\lambda_0+ \i 0
    )T^{-}_\alpha(\lambda_0)\tau,\,\\&
    T^{-}_\alpha(\lambda_0)\tau=\i(H-\lambda_0)J^{-}_\alpha(\lambda_0)\tau\,\text{ and }\,\tau\in
  C^\infty(\S^{n-1}),
  \end{align*}
the following formula holds:
  \begin{align}\label{eq:transAttrac15}
    \norm {\tau}^2-\norm{S_{\alpha\alpha}(\lambda_0)\tau}^2=\pi\inp{\delta(H'-\lambda_0)}_{\psi_a^+(\lambda_0)}.
  \end{align}
 \end{cor}
 \begin{proof}
   Recall
\begin{align*}
    \phi:=W^{-}_\alpha(\lambda_0)\tau&=J^{-}_\alpha(\lambda_0)\tau+\i R(\lambda_0{+}\i0)
\psi.
  \end{align*} We used in the proof of Corollary \ref{cor:norm_less}
  that $\inp{\widecheck B}_{\theta_R\Pi'\phi}$ asymptotically is non-negative. More
  precisely we can  calculate
  \begin{align*}
    \lim_{R\to\infty }\inp{\widecheck
    B}_{\theta_R\Pi'\phi}&=\lim_{R\to\infty }\inp{-\tfrac12\i
                           [H',\chi_R ]}_{\phi}\\&=\lim_{R\to\infty
                                                   }\inp{-\tfrac12\i
                                                   [H',\chi_R
                                                   ]}_{R'(\lambda_0+
                                                   \i 0
                                                   )\psi_a^+(\lambda_0)}
                                                   =\pi\inp{\delta(H'-\lambda_0)}_{\psi_a^+(\lambda_0)}.
  \end{align*}
\end{proof}
\begin{remark}\label{remark:elast-scattA}  If $\lambda_0=\Sigma_2$ the
  right-hand side of \eqref{eq:transAttrac15} vanishes for all $\tau\in
  C^\infty(\S^{n-1})$. However for the non-multiple case above
   $\Sigma_2$ we dont see any  reason this  be the case. Consequently
  the option of `transmission'
is left as a
  conjecture for
   $\lambda_0>\Sigma_2$  (including the multiple case of Subsection \ref{subsec:Non-elastic
  scattering}). We shall study the problem of `non-transmission' for
the physics models at a
two-cluster threshold in  detail in Section \ref{sec:Transmission
  problem at threshold}.

\end{remark}
\subsubsection{Elastic scattering at $\Sigma_2$}\label{subsubsec: Elastic scattering at  lowest
   threshold}

We shall supplement Theorem \ref{thm:ScatN} under the additional
conditions that $\lambda_0=\Sigma_2$ and that   $V_1(r)= -\gamma r^{-\rho}$ for $r\geq1$, by
then proving  that  the kernel of $S_{\alpha\alpha} (\lambda_0)$ is
smooth outside $\{\omega\cdot \omega'= \cos
  \tfrac \rho{ 2-\rho}\pi\}$. With the new more restrictive condition on  $\lambda_0$,
\begin{equation}
  \label{eq:discr}\lambda_0\in\sigma_\d(H^a).
\end{equation} With \eqref{eq:discr} it is possible to `improve' on the
properties of the operator $T^\pm_\alpha(\lambda_0)$ by solving the
transport equations more carefully, cf. \cite {Bo, Sk2}. Relying  only
on \eqref{eq:discr} this does not
 need $\lambda_0$ to be
the lowest threshold. However we will need the
 additional property
 \begin{subequations}
 \begin{equation}
  \label{eq:discrkkk} \lambda_0\notin \sigma_{\ess}(H'),
\end{equation} see  Lemma \ref{lemma:Respectelast-scatt-at}, which
indeed is fulfilled for $\lambda_0=\Sigma_2$. Although this  condition
can be weakened to  $\lambda_0<\Sigma_3$, cf.
Remark \ref{remark:Sigma3vf_1c-case-lambd}, we will in this subsection
assume  $\lambda_0=\Sigma_2$.

 Recall that in this
 section  we impose \eqref{eq:15usual},  so  \eqref{eq:discrkkk} may be stated equivalently as
   \begin{equation}
     \label{eq:conlowsmo}
     \lambda_0\notin \sigma (H').
   \end{equation}
 \end{subequations}

We shall show  the following analogue of \cite[Theorem 9.3]{DS1}.

\begin{thm}  \label{thm:sings}
Suppose in addition to the conditions of Theorem \ref{thm:ScatN} that
$\lambda_0=\Sigma_2$ and that   $V_1(r)= -\gamma r^{-\rho}$
 for $r\geq 1$. Then
 the kernel $S_{\alpha\alpha}(\lambda_0)(\omega,\omega')$ is smooth
  outside the set $\{(\omega,\omega')\mid \omega\cdot \omega'= \cos
  \tfrac \rho{2-\rho}\pi\}$.
\end{thm}

 We need various  preparation partly similarly to \cite[Section 9]{DS1} to prove this
 result. The first result stated as Lemma \ref{lem:conTransp} uses
 only the conditions of Theorem \ref{thm:ScatN} and
 \eqref{eq:discr}.

To see how
  \eqref{eq:discr}  allows us to `improve'
  $T^\pm_\alpha(\lambda_0)=(H-\lambda_0)J^\pm_\alpha(\lambda_0)$ we  look at the term
\begin{align*}
  \parb{I^{(1)}_a-w}J_\alpha^{\pm}(\lambda_0)=\Pi\parb{I^{(1)}_a-w}J_\alpha^{\pm}(\lambda_0)+\Pi'I^{(1)}_a\Pi
  J_\alpha^{\pm}(\lambda_0).
\end{align*}
 The first term has  polynomial decay, viz. it belongs to $
\vL_s^k={\mathcal
 L}(H^k(\S^{n-1}),L_s^{2}(\bX))$ for any $k,s\in \R$,  and the second term is
 $\vO(\inp{x}^{-1-\rho})J_\alpha^{\pm}(\lambda)$,
 cf. \eqref{eq:smallT}.  We look at the leading
 term after a Taylor expansion
 \begin{align*}
   \Pi'I^{(1)}_a\Pi
  J_\alpha^{\pm}(\lambda_0)\approx\parb{\Pi'x^a\varphi_\alpha}\otimes\parb{\nabla I^{(1)}_a(x_a) J_N^{\pm}(\lambda_0)}.
 \end{align*}  Let $\tilde{r} _\alpha(\lambda_0)$ be the
 reduced resolvent of $H^a$ at $\lambda_0$.
Then we add the term
\begin{align*}
   J_\alpha^{1\pm}(\lambda_0):=-\parb{\tilde{r} _\alpha(\lambda_0)\Pi'x^a\varphi_\alpha}\otimes\parb{\nabla I^{(1)}_a(x_a) J_N^{\pm}(\lambda_0)}
\end{align*} to $J_\alpha^{\pm}(\lambda_0)$  and compute     for the resulting
operator $\breve J_\alpha^{1\pm}(\lambda_0)=J_\alpha^{\pm}(\lambda_0)+J_\alpha^{1\pm}(\lambda_0)$, observing  a
cancellation,
 \begin{subequations}
 \begin{align}\label{eq:comE1}
   \begin{split}
   &(H-\lambda_0)
   \breve J_\alpha^{1\pm}(\lambda_0)\\
&=(p_a^2+I_a)
  J_\alpha^{\pm}(\lambda_0)-\parb{\Pi'x^a\varphi_\alpha}\otimes\parb{\nabla
  I^{(1)}_a(x_a) J_N^{\pm}(\lambda_0)}\\
&-(p_a^2+I_a)\parb{\tilde{r}
  _\alpha(\lambda_0)\Pi'x^a\varphi_\alpha}\otimes\parb{\nabla I^{(1)}_a(x_a)
  J_N^{\pm}(\lambda_0)}\\
&={\varphi_\alpha}\otimes
  T_N^{\pm}(\lambda_0)\tau+\vO(\inp{x}^{-2-\rho})J_\alpha^{\pm}(\lambda_0)\\
&-\parb{\tilde{r}
  _\alpha(\lambda_0)\Pi'x^a\varphi_\alpha}\otimes\parb{h\nabla I^{(1)}_a(x_a)
  J_N^{\pm}(\lambda_0)}+\vO(\inp{x}^{-2-2\rho})J_\alpha^{\pm}(\lambda_0).
   \end{split}
\end{align} The second and the fourth terms have better decay than
 the term $\vO(\inp{x}^{-1-\rho})J_\alpha^{\pm}(\lambda)$ we
 started out with. The third term has a different form and needs
 examination. We compute, using the notation \eqref{eq:tildez} and
 doing for simplicity only the plus case,
 \begin{align}\label{eq:comE2}
\begin{split}
   &h\nabla I^{(1)}_a(x_a)
  \parb{J_N^{+}(\lambda_0)\tau}(x_a)\\&=\parbb{\vO(\inp{x_a}^{-1-\rho})h+\vO(\inp{x_a}^{-2-\rho})p_a+\vO(\inp{x_a}^{-3-\rho})}J_N^{+}(\lambda_0)\tau\\
&=
\int
  \e^{\i \phi^{+}(x_a,\omega,0)}\e^{\tilde\zeta^+
  (x_a,\omega,0)}\vO(\inp{x_a}^{-2-\rho 3/2})
  \tau(\omega)\d \omega.
 \end{split}
 \end{align}
 \end{subequations}
 Whence  the third term is  a power $\vO(\inp{x}^{-1-\rho/2})$ better than
 what we started with.

Partly motivated by the above considerations let  us
 introduce the spaces
 \begin{align}\label{eq:laTensora}
   \begin{split}
    L^{2,a}_{s}&=L^{2,a}_{s}(\bX)=\set{u\in  L^2_{s}(\bX)\mid  \forall m\in
\N: \inp{x^a}^mu\in L^2_{s}(\bX)};\quad s\in \R,\\ L^{2,a}_{-\infty}&=\cup_{s\in\R}
\,L^{2,a}_{s},\\
\vL_s^{k,a}&={\mathcal
 L}(H^k(\S^{n-1}),L_s^{2,a});\quad k,s\in \R.
 \end{split}
\end{align}
 We shall also need   the following subclasses
 $\vC^+_s$, $s\in \R$;  see \eqref{eq:incC} for a relationship.

\begin{defn}\label{defn:Class}
Let $\vF^a$ be the set of functions in
 $L^2_\infty(\bX^a)$
 of the form $f^a=T\varphi_\alpha$,  where $T$ is any multiple
 product of factors of $\Pi'$, $\tilde r_\alpha(\lambda_0)$ and  multiplication
 by  components  of $x^a$. Let $R\geq R_0$ and $\sigma'\in (0,\sigma_0)$
 be given.  We consider  operators of
 the tensor product type
 \begin{align*}
   f^a\otimes J^+_g;\quad (J^+_g\tau)(x_a)=\int_{\S^{n-1}}
  \e^{\i \phi^{+}(y,\omega,0)}\e^{\tilde\zeta^+ (y,\omega,0)}g(y,\omega)
  \tau(\omega)\d \omega,\quad y=x_a,
 \end{align*} where $ f^a\in\vF^a$ and the \emph{symbol} $g\in \vS^+_s$, meaning (with
 reference to \eqref{eq:Gamma}) that
 \begin{align}\label{eq:bndderG}
   |\partial_\omega^\delta\partial_{y}^\gamma
   g(y,\omega)|\leq C_{\delta,\gamma}\inp{y}^{-s-\abs{\gamma}}\mand
   \supp g\subseteq \Gamma^+_{R,\sigma'}.
 \end{align}

 The  class $\vC^+_s$, $s\in \R$,
  is  the set of operators given as a finite  sum  of such
 products, where for each term the symbol $g\in \vS^+_s$.
\end{defn}

Due to properties of $\phi^{+}$ and
 $\tilde\zeta^+$ it follows from the proof of \cite[Theorem
 6.5]{DS1} that
 \begin{align}\label{eq:incC}
   \vC^+_s\subset\vL_{s-(1-\rho/2)k-s_0-\epsilon}^{-k,a};\quad k\geq
   0,\,\epsilon>0.
 \end{align} We also note that $J^+(0)$ is of the form
 $J^+_g$ for some symbol $g\in \vS^+_0$, see
 \eqref{eq:int_ope} and \eqref{eq:jdef}, and therefore
 $J_\alpha^+(\lambda_0)\in \vC^+_0$.

 Ideally  we would like  to solve $
\breve
T^\pm_\alpha(\lambda_0)=(H-\lambda_0)\breve J^\pm_\alpha(\lambda_0)\in
\vL_s^{k,a}$ for any given $k\leq 0$ and $s\geq 0$  by modifying the construction
$J^\pm_\alpha(\lambda_0)$, say denoted by $\breve
J^\pm_\alpha(\lambda_0)\approx J^\pm_\alpha(\lambda_0)$. Of course
like   for the one-body
problem  we cannot do that, but we can  solve transport equations in
a (small) forward cone  (cf. the splitting for the one-body
problem,  $T_N^{\pm}
(\lambda)=T^{\pm}_\bd(\lambda-\lambda_0)+T^{\pm}_\pr (\lambda-\lambda_0)$, corresponding to
 the  decomposition of $t^{\pm}(x,\omega,\lambda)=t^{\pm}(x,\sqrt{\lambda}\omega)$ in  \cite[(5.8) and (5.16)]{DS1}). We shall mimic this procedure using  the
reduced resolvent $\tilde r_\alpha(\lambda_0)$ as an additional
tool. In fact  we computed  above  for $\breve J_\alpha^{1+}(\lambda_0):=J_\alpha^{+}(\lambda_0)+J_\alpha^{1+}(\lambda_0)$
 \begin{align*}
   (H-\lambda_0)\breve J_\alpha^{1+}(\lambda_0)={\varphi_\alpha}\otimes
  T_N^{+}(\lambda_0)+T_\alpha^{1+}(\lambda_0)+T_{\alpha,{\r}}^{1+}(\lambda_0),
 \end{align*} where $T_\alpha^{1+}(\lambda_0)\in \vC^+_{2+\rho 3/2}$   and (for example)
 \begin{align*}
   T_{\alpha,{\textrm
       r}}^{1+}(\lambda_0)\in\vL_{2+\rho -(1-\rho/2)l-s_0-\epsilon}^{-l,a};\quad l\geq 0,\,\epsilon>0.
 \end{align*} The argument relied  on using the first term of a Taylor
 expansion only, however we can do the Taylor
 expansion of $\Pi'I^{(1)}_a\Pi$ to any order, and each term will then
 contribute by a term in $\vC^+_{2+\rho}$ except (possibly) a  `remainder
 term' which can be taken in the   fixed space $\vL_s^{k,a}$ of
 interest. We can argue the same way for the fourth term
 $\vO(\inp{x}^{-2-2\rho})J_\alpha^{\pm}(\lambda_0)$ of \eqref{eq:comE1}. We conclude that
 $(H-\lambda_0)\breve J_\alpha^{1+}(\lambda_0)-{\varphi_\alpha}\otimes T_N^{+}(\lambda_0)+T\in \vL_s^{k,a}$ for
 some $T \in \vC^+_{2+\rho}$.

To improve further, for this  $T \in \vC^+_{2+\rho}$  ideally we would like  to `solve' the equation
$(H-\lambda_0)J_\alpha^{2+}(\lambda_0)\approx T$. This would yield a
better approximation by adding  $J_\alpha^{2+}(\lambda_0)$, viz.
 by considering $\breve J_\alpha^{2+}(\lambda_0):=J_\alpha^{+}(\lambda_0)+J_\alpha^{1+}(\lambda_0)+J_\alpha^{2+}(\lambda_0)$.
 We are lead to considering the following iteration scheme. Suppose
  that for given $m\in\N$ we have constructed $\breve
 J_\alpha^{m+}(\lambda_0):=J_\alpha^{+}(\lambda_0)+J_\alpha^{1+}(\lambda_0)+\cdots+J_\alpha^{m+}(\lambda_0)$
 such that
 \begin{align}
   \begin{split}
    \label{eq:ite} (H-\lambda_0)\breve
   J_\alpha^{m+}(\lambda_0)&-{\varphi_\alpha}\otimes T_N^{+}(\lambda_0)+T_{\pr}^{m+}(\lambda_0)-T_{\bd}^{m+}(\lambda_0)\in
   \vL_s^{k,a}\\
 T_{\pr}^{m+}(\lambda_0)&\in \vC^+_{s_m};\quad s_m=(m-1)\min\set{1-\rho/2,
 \rho/2}+2+\rho,\\
T_{\bd}^{m+}(\lambda_0)&\in \vC^+_{\bd, \,s_1},
   \end{split}
 \end{align} where $\vC^+_{\bd, \,s}$ is given as follows. The
 notation  $M^{\circ}$ refers in general to the interior of a subset $M$
 of a topological space (below taken as $\R^n\times \S^{n-1}$).
\begin{defn}\label{defn:Class2}
Let  $R\geq R_0$ and  $\sigma'\in (0,\sigma_0)$ be given, cf.
Definition \ref{defn:Class}, and  let
 $\sigma\in (0,\sigma')$. Then $\vC^+_{\bd, \,s}$ is the subclass of
 operators
 ${\varphi_\alpha}\otimes J^+_g \in\vC^+_{ s}$ for which the symbol
 $g=g^+_{\bd}\in \vS^+_{s}$ (as required by
 Definition \ref{defn:Class}) but in addition  has the support property
 \begin{align}\label{eq:symBD}
 \supp g^+_{\bd}\subseteq \Gamma^+_{R,\sigma'}\setminus \parb{\Gamma^+_{2R,\sigma}}^{\circ}.
 \end{align}
\end{defn}

So far we have verified \eqref{eq:ite} for $m=1$ only (with
$T_{\bd}^{1+}(\lambda_0)=0$). Now, suppose \eqref{eq:ite} for a given
$m\geq 1$. Then we split
\begin{align*}
  T=T_{\pr}^{m+}(\lambda_0)=\Pi T_{\pr}^{m+}(\lambda_0)+\Pi'T_{\pr}^{m+}(\lambda_0),
\end{align*}
 and both  of the terms to the right contributes to the construction of
$J_\alpha^{(m+1)+}(\lambda_0)$ as follows: For the first term we
mimic  \cite[Propositions  5.4 and 5.5]{DS1} and add correspondingly a
term, say $J_{\alpha,1}^{(m+1)+}(\lambda_0)$. By a Taylor expansion
(as used above) we see that
\begin{align}\label{eq:indStep}
  (H-\lambda_0)J_{\alpha,1}^{(m+1)+}(\lambda_0)-\Pi T_{\pr}^{m+}(\lambda_0)\in \vC^+_{s_{m+1}}+\vC^+_{\bd, \,s_1}+\vL_s^{k,a};
\end{align} we give the details below. Letting
\begin{align*}
  J_{\alpha,2}^{(m+1)+}(\lambda_0)=\tilde{r} _\alpha(\lambda_0)\Pi'T_{\pr}^{m+}(\lambda_0)
\end{align*} we then obtain that indeed \eqref{eq:ite} is fulfilled with
$m$ replaced by $m+1$ and with
\begin{align*}
  \breve J_\alpha^{(m+1)+}(\lambda_0)=\breve
  J_\alpha^{m+}(\lambda_0)+J_\alpha^{(m+1)+}(\lambda_0);\quad J_\alpha^{(m+1)+}(\lambda_0)=J_{\alpha,1}^{(m+1)+}(\lambda_0)+J_{\alpha,2}^{(m+1)+}(\lambda_0).
\end{align*} Note that obviously
$(H-\lambda_0)J_{\alpha,2}^{(m+1)+}(\lambda_0)-\Pi'T_{\pr}^{m+}(\lambda_0)\in
\vC^+_{s_{m+1}}+\vL_s^{k,a}$, cf.  \eqref{eq:comE1} and
\eqref{eq:comE2}.

To complete the recursive construction it  remains to justify
\eqref{eq:indStep}: We use the function in
\eqref{eq:chi^2}, more precisely we consider `cut-offs' $\bar\chi_R(r)$
and  $\chi_{\sigma,\sigma'}(\hat x\cdot \omega)$. Defined in terms of the function $\zeta^+=\zeta^+
  (x_a,\omega,0)$ of \eqref{eq:tildez} we let $A^+$ be the differential
  operator (in $x=x_a$)
  \begin{align*}
    A^+= \Delta +2(\nabla \zeta^+)\cdot \nabla +(\Delta \zeta^+)+ (\nabla \zeta^+)^2,
  \end{align*} cf. \eqref {eq:Bfac},  and we define then for $\Pi
  T_{\pr}^{m+}(\lambda_0)=\varphi_\alpha \otimes
  J^+_{g_{m}}$   correspondingly
  $J_{\alpha,1}^{(m+1)+}(\lambda_0)=\varphi_\alpha\otimes J^+_{g_{m+1}}\in\vC^+_{s_m-1-\rho/2}$, where (with $x=x_a$)
  \begin{align*}
    g_{m+1}:=-\i \chi_{\sigma,\sigma'}(\hat x\cdot \omega)\bar\chi_R(r)\int_1^\infty g_m(y(t,x,\omega,0),\omega)\,\d t;
  \end{align*} here $y(\cdot ,x,\omega,0)$ is the classical
  zero-energy orbit
  starting at $x$ for $t=1$ and with asymptotic normalized velocity
  $\omega=\lim_{t\to +\infty}y/\abs{ y}$, cf. \eqref{eq:mixed conditions222}.  Let
  \begin{align*}
    \breve g_{m+1}=\i \chi_{\sigma,\sigma'}(\hat x\cdot \omega)\bar\chi_R(r)\int_1^\infty A^+g_m(y(t,x,\omega,0),\omega)\,\d t
  \end{align*}
We calculate
  \begin{align*}
    &(H-\lambda_0)\varphi_\alpha\otimes J^+_{g_{m+1}}+\Pi T_{\pr}^{m+}(\lambda_0)+\varphi_\alpha\otimes J^+_{\breve
    g_{m+1}} \\&=\varphi_\alpha \otimes
                                              hJ^+_{g_{m+1}}+(I_a-w)\varphi_\alpha
                 \otimes
                                              J^+_{g_{m+1}}+\Pi T_{\pr}^{m+}(\lambda_0)+\varphi_\alpha\otimes J^+_{\breve g_{m+1}}\\
&\in \vC^+_{s_{m+1}}+\vC^+_{\bd, \,s_1}+\vL_s^{k,a};
  \end{align*} here we used a Taylor expansion to treat the second term
  (as we did in \eqref{eq:comE1}) and the fact that the
  three other terms cancel up to derivatives of the factor
  $\chi_1(r/R)\chi_2(\hat x\cdot \omega)$, cf. \eqref{eq:Bfac}. Using next  that $\breve
  g_{m+1}\in \vS_{s_{m+1}}^+$ indeed \eqref{eq:indStep} follows.

Next, rather than doing the Borel summation as for the one-body
problem, for simplicity we `terminate' the recursive construction at
$m=M$ taken so large that $\vC^+_{s_{m+1}}\subset \vL_s^{k,a}$,
cf. \eqref{eq:incC}. Beforehand we treated for convenience only the
plus case. Leaving it to the reader to figure out how Definitions
\ref{defn:Class} and \ref{defn:Class2} should read in the minus  case, and
how the above procedure correspondingly modifies,
we
consider henceforth  the   `improved' operators
$J^{\pm}_{\alpha,M}(\lambda_0):=\breve J^{M\pm}_{\alpha}(\lambda_0)$
and
\begin{equation*}
  T^\pm_{\alpha,M}(\lambda_0):=(H-\lambda_0)J^{\pm}_{\alpha,M}(\lambda_0)\in
  {\varphi_\alpha}\otimes T_N^{\pm}(\lambda_0)+\vC^\pm_{\bd, \,s_1}+\vL_s^{k,a}.
\end{equation*}
 In turn
 we
 consider  the corresponding  wave matrices
$W^{\pm}_{\alpha,M}(\lambda_0)$ and the scattering operator
$S_{\alpha,M}(\lambda_0)$; see \eqref{eq:Sforma0} below. Like
for the one-body problem the latter  quantites are canonical (this is
stated  more
precisely as the last assertion in the  following conclusion).
\begin{lemma}\label{lem:conTransp}
  Let $k\leq 0$, $s\geq 0$, $R\geq R_0$ and
 $0<\sigma<\sigma'<\sigma_0$. Then the above recursive procedure,
 terminating at any  sufficiently large  $M\in \N$, yields the
 existence of
 $J^{\pm}_{\alpha,M}(\lambda_0)\in
 J^{\pm}_{\alpha}(\lambda_0)+\vC^\pm_{2s_{0}}$, $g_{\bd,M}^\pm \in
 \vS^\pm_{2s_0}$ with $\supp g^\pm_{\bd,M}\subseteq
 \Gamma^\pm_{R,\sigma'}\setminus \parb{\Gamma^\pm_{2R,\sigma}}^{\circ}$ and
 $R_{s,M}^{k,a\pm}\in \vL_s^{k,a}$ such that
 \begin{align*}
   T^\pm_{\alpha,M}(\lambda_0)=(H-\lambda_0)J^{\pm}_{\alpha,M}(\lambda_0)
= (2\pi)^{-n/2}{\varphi_\alpha}\otimes J_{g}^{\pm}+R_{s,M}^{k,a\pm};\quad g=g^\pm_{\bd,M}.
 \end{align*} (In particular $ T^\pm_{\alpha,M}(\lambda_0)\in
\vC^\pm_{\bd, \,2s_0}+\vL_s^{k,a}$.)

 The  wave matrices
$W^{\pm}_{\alpha,M}(\lambda_0)$ and the scattering operator
$S_{\alpha,M}(\lambda_0)$ defined by the operators $J^{\pm}_{\alpha,M}(\lambda_0)$
and $T^\pm_{\alpha,M}(\lambda_0)$ coincide with
$W^{\pm}_{\alpha}(\lambda_0)$ and $S_{\alpha}(\lambda_0)$, respectively.
\end{lemma}

 Note that the parameter $R$, $\sigma$  and $\sigma'$ of  the lemma are used  to
 define the classes $\vC^\pm_{\bd, \,s}$. The first part of the lemma
 is clearly a consequence of the explained
construction. For the second part, note that for the cases of  $W^{-}_{\alpha,M}(\lambda_0)$ and $ S_{\alpha,M}$
 the  assertion  is a
consequence of the uniqueness part  of Theorem
\ref{thm:repEigen}~\ref{it:uniq2}.
The identification $W^{+}_{\alpha,M}(\lambda_0)=W^{+}_{\alpha}(\lambda_0)$
follows from an analogue statement for
$W^{+}_{\alpha}(\lambda_0)$ given in terms of $
S_{\alpha}(\lambda_0)^*$ (not given   in our  presentation).

 As in \eqref{eq:def_a0_b} the notation $g$ means
 $g_{\lambda_0}=\sqrt{-V_1}$.  We introduce  symbols $b$ and $\bar c$ to
 decompose the normalized momentum $\eta:=\xi/g\in \bX_a'$ as
\begin{equation}
   \label{eq:sym_basdef}{\eta}=b{ \hat x}+\bar c,\;b:={\hat x}
\cdot \eta\text{  and }\bar c :=\big(I-\big |\hat x\big\rangle \big\langle {\hat x}\big|\big)\eta.
\end{equation} Of course this decomposition  requires $x\neq 0$ since
$\hat x=x/|x|$ only make sense for such $x$. (In  the  wave front set definition below  this issue
is handled by a cut-off $\bar\chi_1$, cf. \eqref{eq:14.1.7.23.24}. In
\eqref{eq:def_a0_b} a slightly different $b$ `cured' the problem.)

 On the energy shelf $\xi^2 =
\gamma r^{-\rho}$  the quantity $a:=b^2+\bar c^2= \xi^2/g^2=1$ and the
  Hamiltonian orbits solve the ODE on the  `reduced phase space' $ {\mathbb
    T}^*$, consisting of points $(\hat x,\bar c,b)$,
\begin{equation}\label{eq:reduced eqns}
\begin{cases}
\tfrac\d{\d\tau}\hat x=\bar c,\\
\tfrac \d{ \d\tau}\bar c=-(1-\tfrac\rho 2)b\bar c-\bar c^2\hat x,\\
\tfrac\d{ \d\tau}b=(1-\tfrac\rho 2)\bar c^2.
\end{cases}
\end{equation} This is in the
  `new time' $\tau$ given by
  $\tfrac{\d\tau}{
 \d t}=2g/r$.

The maximal solution of \eqref{eq:reduced eqns} that passes
 $z=(\hat x,b,\bar c)\in {\mathbb T}^*$ at $\tau=0$ is denoted by $\gamma(\tau,
z)$.  The quantity $a$ is preserved by the flow, and the equation $\bar c=0$ defines the
fixed points.  Away from those  points
\begin{align*}
b(\tau)=\sqrt a\tanh \sqrt a(1-\tfrac\rho2)(\tau-\tau_0),
\end{align*} showing moreover  that $b$ is monotonely
increasing in $\tau$ from $-\sqrt a$ to $\sqrt a$ (away from fixed points).
We introduce in terms of the variables \eqref{eq:sym_basdef} the `wave front set' $WF^a_s (u)$ of  a distribution $u\in L^{2,a}_{-\infty}$  as  the
subset of
 ${\T}^*$ given by the condition
\begin{align}
   \label{eq:WF^sa}
   &z_1=(\omega_1,\bar c_1,b_1) =(\omega_1,b_1\omega_1+\bar c_1)= (\omega_1,\eta_1) \notin WF^a_s(u) \nonumber\\
&\Leftrightarrow\\
   &\exists\;{{neighbourhoods} }\;\mathcal{N}_{\omega_1}\ni
   {\omega_1},\;\mathcal{N}_{\eta_1}\ni {\eta_1}\ \
\forall \chi_{\omega_1}\in
   C^{\infty}_c(\mathcal{N}_{\omega_1}),\;\chi_{\eta_1}\in
   C^{\infty}_c(\mathcal{N}_{\eta_1}): \nonumber\\
&\Opw\big(\chi_{z_1}\bar\chi_1(r)\big)u\in L_s^{2,a}\textit{ where }\chi_{z_1}=\chi_{z_1}(x,\xi)=\chi_{\omega_1}({\hat x})\chi_{\eta_1}\parb{\xi/g(r)}. \nonumber
 \end{align}
 Obviously this  notion  of wave front set  is a (fibered)  adaption of
 the notion  of `scattering wave front set' $WF^s_{\sc} (v)$ of  a
 distribution $v\in L^{2}_{-\infty}(\R^n)$  of \cite[Subsection 4.2]{DS1} to the
 present problem.

Due to \eqref{eq:conlowsmo} the operators $(H'-\lambda_0)^{-1}$ and $ \breve R(\lambda_0)$
respect the above notion of `fibered scattering wave front set'.
\begin{lemma}
  \label{lemma:Respectelast-scatt-at} For any $ s\in \R$ and $ u\in
  L^{2,a}_{-\infty }$ the
  following  properties hold.
\begin{subequations}
\begin{align}\label{eq:wiii}
  \Pi I_a \Pi' L^{2,a}_{s} &\subset
                             L^{2,a}_{s+1+\rho}, \quad \Pi' I_a \Pi L^{2,a}_{s} \subset
                             L^{2,a}_{s+1+\rho},
  \\
  WF_{s+1+\rho}^a( \Pi
  I_a \Pi'u)&\subset
              WF_s^a( u), \quad \label{eq:wv}
              WF_{s+1+\rho}^a(
              \Pi' I_a \Pi
              u)\subset WF_s^a( u),\\
\label{eq:wi}
  \breve R(\lambda_0)\Pi' L^{2,a}_{s} &\subset L^{2,a}_{s},\\\label{eq:wii}
  WF_s^a(\breve R(\lambda_0)\Pi' u)&\subset WF_s^a( u).
\end{align}
  \end{subequations}
\end{lemma}
\begin{proof}
   For \eqref{eq:wiii} and \eqref{eq:wv}  we may use a simplified
   version of
  \eqref{eq:90}. The results are almost trivial since  $\Pi$ and
  $\Pi'$ are operators in the $x^a$-coordinate while the wave front
  setting is defined in terms of quantization on $\bX_a$.

 The  arguments for \eqref{eq:91} works for \eqref{eq:wi}, so it
 remains to consider \eqref{eq:wii}, actually without the factor
 $\Pi'$. So for any $z_1\notin WF^a_s(u)$ we need to estimate
 $\Opw\big(\chi_{z_1}\bar\chi_1(r)\big)\breve R(\lambda_0)u\in
 L_s^{2,a}$, where we can assume that the localization function
 $\chi_{z_1}$ is
 supported sufficiently close to $z_1$. In particular there is a
 slightly bigger one, say denoted by $\tilde \chi_{z_1}$, i.e. $\tilde
 \chi_{z_1}=1$ on $\supp \chi_{z_1}$, such that  $\Opw\big(\tilde \chi_{z_1}\bar\chi_1(r)\big)u\in
 L_s^{2,a}$. Whence it remains to show that
 \begin{align*}
   \Opw\big(\chi_{z_1}\bar\chi_1(r)\big)\breve R(\lambda_0)\parbb{1-
   \Opw\big(\tilde \chi_{z_1}\bar\chi_1(r)\big)}u\in
 L_s^{2,a}.
 \end{align*} For that it suffice to show that
\begin{align*}
   \forall s,t\in \R\,\exists r\geq 0:\quad \Opw\big(\chi_{z_1}\bar\chi_1(r)\big)\breve R(\lambda_0)\parbb{1-
   \Opw\big(\tilde \chi_{z_1}\bar\chi_1(r)\big)}\inp{x^a}^{-r}\in
 \vL\parb{L_s^{2},L_t^{2}}.
 \end{align*}

Clearly we can assume that $s< t$. We recall that the symbols $\chi_{z_1}\bar\chi_1(r)$, $\tilde
\chi_{z_1}\bar\chi_1(r)$ as well as $\xi^2$ belong to a class for
which the corresponding  `Planck
constant' is $\inp{x}^{\rho/2-1}$. Using this fact we can now repeatedly
commute the quantization of
 localized symbols  through factors
of $\breve R(\lambda_0)$ to extract the desired decrease. The first
step consists in noting that $\Opw\big(\chi_{z_1}\bar\chi_1(r)\big)\parbb{1-
   \Opw\big(\tilde \chi_{z_1}\bar\chi_1(r)\big)}$ has arbitrary
 decrease, writing
\begin{align*}
  [\Opw\big(\chi_{z_1}\bar\chi_1(r)\big),\breve R(\lambda_0)]=\breve
  R(\lambda_0)[\breve H,\Opw\big(\chi_{z_1}\bar\chi_1(r)\big)]\breve R(\lambda_0)
\end{align*} and computing the  commutator to the right to gain at
least a
factor $\inp{x_a}^{\rho/2-1}$. In the next step we make a similar
commutation with the far right factor of $\breve R(\lambda_0)$ of a
similar localized operator thereby obtaining another factor
$\inp{x_a}^{\rho/2-1}$. Repeating the argument we produce in this
fashion  efficiently
any power of $\inp{x_a}^{\rho/2-1}$, and taking $r$ sufficient
big (depending on given $s,t\in \R$) we can  efficiently produce  any
power  $\inp{x}^{-m}$, in particular an $m\geq t-s$  as wanted.

\end{proof}

We also shall need the following lemma  which  is based on  an extension of
 parts of \cite[Propositions 4.1 and 4.8]{DS1}  (cf.  the extension  for
obtaining
\eqref{eq:micr1} and Remark \ref{remark:WFelast-scatt-at} given below) and  the propagation of singularities  result
\cite[Proposition 9.1]{DS1} (in fact \eqref{eq:WF3} follows by
from  \eqref{eq:WF2}  and \cite[Proposition 9.1]{DS1}). The underlying
commutator methods are rather robust. In particular, even though the
lemma as  stated has a more qualitative flavour, one can demonstrate concrete bounds for
the claimed embeddings,  and the stability of those  bounds  can be used to
control  a  parameter dependence, cf. \cite[ Remark
9.2]{DS1}. In our application  below the pair of angles
$(\omega,\omega')$ on which  the kernel of the scattering matrix
depends will play the role of such a parameter.

\begin{lemma}
  \label{lemma:propSing}
  \begin{subequations}
Suppose  $u\in L^{2,a}_{s}=L^{2,a}_{s}(\bX)$ for
  some $s< s_0$, and suppose that for some $t>s_0$ and  $\kappa \in (-1,1)$
  \begin{align}\label{eq:WF1}
     WF^a_t(u)\cap\set{ b< \kappa, a=1}=\emptyset.
  \end{align} Then the following assertions hold:
  \begin{enumerate}[1)]
  \item \label{item:liW1}  There exists the limit
    \begin{align*}
      r_{\lambda_0}^+u:=\lim_{R\to \infty}\parb{1\otimes
      r_{\lambda_0}^+}\parb{\chi_R(r)u}\text{  exists in }
  L^{2,a}_{s-2s_0}.
    \end{align*}
  \item \label{item:liW2} The intersection
\begin{align}\label{eq:WF2}
    WF^a_{t-2s_0}(r_{\lambda_0}^+u)\cap\set{ b< \kappa, a=1}=\emptyset,
  \end{align} and (more generally)
\begin{align}\label{eq:WF3}
  \begin{split}
  WF^a_{t-2s_0}(r_{\lambda_0}^+u)&\cap\set{ a =1}\\&\subset\set{\gamma(\tau,z}\mid  \tau\geq0,\,
  z\in WF^a_t(u)\}\cup\set{ b=1}.
  \end{split}
  \end{align}
   \end{enumerate}
\end{subequations}
\end{lemma}

\begin{remark}\label{remark:WFelast-scatt-at} We note that
  \eqref{eq:WF2} may be  proven as in the proof of \cite[Proposition
  4.8 (iii)]{DS1}, but that \cite[(4.48]{DS1} in fact holds with
  $\epsilon=0$ (this generalization of \cite[Proposition 4.1]{DS1},
  originating from \cite{FS},   was
  noted before, cf. \eqref{eq:micr1}). Another generalization comes
  about noting that the condition on $\chi_+$ in \cite[Proposition 4.1
  (v)]{DS1} is too strong. In fact the relevant property is not that
  $\inf \supp(\chi_+) >C_0$ for some `big' $C_0\geq 1$, but rather
  $\chi_+(a)=0$ in a neighbourhood of $a=1$. This generalization can
  be proven by using a  parametrix construction of
  \cite{Sk4}  (cf. \cite[(3.2)]{Sk4})  replacing the  positivity argument of
  \cite{FS}. Note that for  the same reason we dont need the constant $C_0$ for the other parts of \cite[Proposition 4.1]{DS1} neither.
 \end{remark}
Now, to analyse $S_{\alpha,M}(\lambda_0)(\omega,\omega')$ we write (formally)
\begin{align}\label{eq:Sforma0}
  \begin{split}
S_{\alpha,M}(\lambda_0)(\omega,\omega')=&-2\pi\langle j_{\alpha,M}^+(\cdot,\omega),
t_{\alpha,M}^-(\cdot,\omega')\rangle\\
&+
2\pi\i
\langle t_{\alpha,M}^+(\cdot,\omega),
R(\lambda_0+\i0)t_{\alpha,M}^-(\cdot,\omega')\rangle,
  \end{split}
\end{align}
where  for $x\in \bX$
\begin{align}\label{eq:fundSymb}
  \begin{split}
   j_{\alpha,M}^\pm(x,\omega)&=(2\pi)^{-n/2}\big (\e^{\i \phi^\pm}
\e^{\tilde\zeta^\pm} a_{\alpha,M}^\pm\big )(x,\omega,0),\\
t_{\alpha,M}^\pm(x,\omega)&=t_{\alpha,M}^{1\pm}(x,\omega)+\parb{R_{s,M}^{k,a\pm}\delta_\omega}(x),\\& t_{\alpha,M}^{1\pm}(x,\omega):=(2\pi)^{-n/2}{\varphi_\alpha}(x^a)\,\big (\e^{\i \phi^\pm}
\e^{\tilde\zeta^\pm} g^\pm_{\bd,M}\big )(x_a,\omega,0);
  \end{split}
\end{align} here the `symbol' $a_{\alpha,M}^\pm$ is  taken as a vector-valued function of $y=x_a$ and
$\omega$ in
agreement with Lemma \ref{lem:conTransp}, while $\phi^\pm$,
$\tilde\zeta^\pm$ and  the functions  $ g^\pm_{\bd,M}$ (as introduced in  Lemma
\ref{lem:conTransp}) have   dependence of
the component $x_a$ of $x$ and $\omega$ only (i.e. no dependence of $x^a$). Thus we  write $J^{\pm}_{\alpha,M}(\lambda_0)\in
 J^{\pm}_{\alpha}(\lambda_0)+\vC^\pm_{2s_{0}}$ as
\begin{align*}
  \parb{J^{\pm}_{\alpha,M}(\lambda_0)\tau}(x^a,x_a)&=\int_{\S^{n-1}}
  j_{\alpha,M}^\pm(x,\omega) \tau(\omega)\,\d \omega\\&=(2\pi)^{-n/2}\int_{\S^{n-1}}
  \e^{\i \phi^{+}(y,\omega,0)}\e^{\tilde\zeta^+ (y,\omega,0)}a_{\alpha,M}^\pm(x^a,y,\omega)
  \tau(\omega)\,\d \omega,
\end{align*} which in fact applies to the Dirac delta function
$\tau=\delta_\omega \in \cap_{k<-(n-1)/2} H^k(\S^{n-1})$ thereby defining
the `kernel' $j_{\alpha,M}^\pm$. Similarly $t_{\alpha,M}^\pm$ is the
kernel
\begin{equation*}
 t_{\alpha,M}^\pm(x,\omega)=\parb{(2\pi)^{-n/2}{\varphi_\alpha}\otimes
   J^\pm_g\delta_\omega+R_{s,M}^{k,a\pm}\delta_\omega}(x); \quad g=g^\pm_{{\bd},M}.
\end{equation*}

We note the following bounds,  cf. \cite[Theorem
6.5]{DS1}.
\begin{subequations}
 \begin{align}\label{eq:j_space1}
   &\partial_\omega^\delta j_{\alpha,M}^\pm(\cdot,\omega)\in
 L^{2,a}_{t}\text{ for all } t< -(|\delta|+n/2) (1-\rho/2) -\rho/2,\\
   &\partial_\omega^\delta t_{\alpha,M}^{1\pm}(\cdot,\omega)\in
 L^{2,a}_{t}\text{ for all } t< -(|\delta|+n/2) (1-\rho/2)
     +1+\rho/2,\label{eq:j_space2}\\
&R_{s,M}^{k,a\pm}\partial_\omega^\delta\delta_\omega\in
 L^{2,a}_{s}\text{ provided } {k<-\abs{\delta}-(n-1)/2}.\label{eq:j_space2b}
 \end{align}
\end{subequations}

Let $\phi_\sph^+$ denote the solution of the eikonal equation for the
potential $V_1$ at zero energy, cf. \eqref{eq:eik2}. It is
given by
\begin{equation}\label{eq:tilll}
\phi_{\sph}^+(x_a,\omega)
=\frac{\sqrt{\gamma}}{ 1-\rho/2}\left( |x_a|^{1-\rho/2}\cos (1-\rho/2)
\theta-R_0^{1-\rho/2}\right),
\end{equation}
where $ \cos \theta=\widehat {x_a} \cdot
\omega$. We omit the subscript $a$ writing below (including the
following lemma)  $x$
rather than $x_a$. Using $x^\perp=\frac{\omega-\hat x\cos\theta}{\sin\theta}$ and
$\nabla_x\theta=-\frac{x^\perp}{r}$, we can also compute
\begin{align*}
F_\sph(x,\omega)&=F^+_\sph(x,\omega):=\nabla_x\phi_\sph^+(x,\omega)\\
&=\sqrt{\gamma}r^{-\rho/2}
\left(\hat x\cos(1-\rho/2)\theta+x^\perp\sin(1-\rho/2)\theta\right).
\end{align*}

 The lemma  stated below is a straightforward generalization of
\cite[Lemma 9.4]{DS1},  and as in \cite{DS1} it follows  by
integration by parts. Note  that due to \eqref{eq:j_space2b} only the
contribution from the term
$t_{\alpha,M}^{1\pm}$  of  $t_{\alpha,M}^{\pm}$
in \eqref{eq:fundSymb} matters. Note also that \eqref{eq:j_space1} and
 \eqref{eq:j_space2} provide an  a priori `size' of the quantatives
 $\partial_\omega^\delta j_{\alpha,M}^\pm(\cdot,\omega)$ and
 $\partial_\omega^\delta t_{\alpha,M}^{1\pm}(\cdot,\omega)$.

\begin{lemma}
  \label{lemma:wavefas}
  \begin{subequations}
     Let $k<-(n-1)/2$, $s\geq 0$, $R\geq R_0$ and
 $0<\sigma<\sigma'<\sigma_0$ (cf.  Lemma
\ref{lem:conTransp}). Then for all $M\in \N$ taken large
 enough and
     for all $\omega\in \S^{n-1}$ and
  multiindices  $\delta$ with $|\delta|<-k-(n-1)/2$ the quantities in
  \eqref{eq:fundSymb} obey
  \begin{align}\label{eq:wacom1} WF^a_{s}(\partial_\omega^\delta
j_{\alpha,M}^{\pm} &(\cdot ,\omega)) \subset\Big \{z=(\hat x,
    \bar c, b)\in \T^*\mid  \\& 1-\sigma'\leq
\pm \hat x\cdot \omega
 ,\quad b\hat x+\bar c=\pm {F_{\sph}(\hat x,\pm \omega)}/{\sqrt
                            \gamma}\Big\},\nonumber
    \\
WF^a_{s}(\partial_\omega^\delta
 t_{\alpha,M}^{\pm} &(\cdot ,\omega))\subset \Big\{z=(\hat x,
    \bar c, b)\in \T^*\mid \label{eq:wacom}\\&1-\sigma'\leq
\pm \hat x\cdot \omega \leq 1-\sigma,\quad b\hat x+\bar c=\pm
                                                 {F_{\sph}(\hat
                                                 x,\pm \omega)}/\sqrt
                                                 \gamma\Big\}.\nonumber
  \end{align} Suppose that $\chi_+\in
   C^\infty(\R)$, $\chi_+'\in
   C^\infty_c(\R)$  and $\supp
   \chi_+ \subset(1,\infty)$. Under the same conditions as above the following bounds hold  uniformly in $\omega\in \S^{n-1}$:
   \begin{equation}
     \label{eq:+bound}
     \Opw(\chi_+(a))\partial_\omega^\delta j_{\alpha,M}^{\pm} (\cdot ,\omega),\;\Opw(\chi_+(a))\partial_\omega^\delta t_{\alpha,M}^{\pm} (\cdot ,\omega)\in L^{2,a}_{s}.
 \end{equation}
\end{subequations}
\end{lemma}

Now   we  fix (a big) $\breve k\in\N$ and want to apply Lemma
\ref{lemma:wavefas} under the additional conditions
\begin{align}\label{eq:restri}
  \breve k<-k-(n-1)/2\quad \mand \quad s\geq s_{\breve k}:=(\breve k+\tfrac n2) (1-\tfrac\rho2) +1;
\end{align}
 this
requires $M=M(k,s) $
  taken  large. We consider with  $\chi_R=\chi_R(|x_a|)$ and
$|\delta|,|\delta'|\leq \breve k$
\begin{align}\label{eq:Sforma}
  \begin{split}
  &\partial_\omega^\delta \partial_{\omega'}^{\delta'}S_{\alpha,M}(\lambda_0)(\omega,\omega')\\&=-2\pi\lim_{R\to\infty}\langle \partial_\omega^\delta
  j_{\alpha,M}^+(\cdot,\omega),
\chi_R\,\partial_{\omega'}^{\delta'}t_{\alpha,M}^-(\cdot,\omega')\rangle\\
&+
2\pi\i
\lim_{R\to\infty} \langle \partial_\omega^\delta t_{\alpha,M}^+(\cdot,\omega),
\chi_R R(\lambda_0+\i0)\chi_R\,\partial_{\omega'}^{\delta'}t_{\alpha,M}^-(\cdot,\omega')\rangle
  \end{split}
\end{align} This is still formal, but suppose that we can prove that
the two limits on the right-hand side exist locally uniformly in the
set $\{\omega\cdot \omega'\neq \cos
  \tfrac \rho{ 2-\rho}\pi\}$ for all  multi-indices  with
  $|\delta|,|\delta'|\leq \breve k$, then Theorem \ref{thm:sings} follows.
   \begin{proof}[Proof of Theorem \ref{thm:sings}]
     \subStep{I}  We need for any $\breve k\in\N$ (henceforth
     fixed) to verify the existence and continuity of
     $\partial_\omega^\delta \partial_{\omega'}^{\delta'}S_{\alpha,M}(\lambda_0)(\omega,\omega')$
     in $\{\omega\cdot \omega'\neq \cos \tfrac\rho{ 2-\rho}\pi\}$ for
     $|\delta|,|\delta'|\leq \breve k$. We impose the conditions of
     Lemma \ref{lemma:wavefas} as well as \eqref{eq:restri}. We also  assume that $\sigma'$
      is small (to be tacitly used for
     \eqref{eq:emptyFirst} stated below).

We look at the
     first term to the right in \eqref{eq:Sforma} (before taking the limit), which we claim is
      treatable on $\S^{n-1}\times \S^{n-1}$ (i.e. without restriction on
      $(\omega,\omega')$).  This can be seen by a direct integration by
      parts, but we prefer to give a presentation that conforms with
      our treatment of the second term (see \eqref{eq:strategy}).
 Under the above  conditions   it follows from  Lemma
 \ref{lemma:wavefas} and \cite[(3.5c)]{DS1}
    that for $|\delta|,|\delta'|\leq \breve k$,
     \begin{align}\label{eq:emptyFirst}
       \begin{split}
        WF^a_{s_{\breve k}}(\partial_\omega^\delta
 j_{\alpha,M}^{+} (\cdot ,\omega))\cap WF^a_{s_{\breve k}}(
\partial_{\omega'}^{\delta'} &t_{\alpha,M}^{-} (\cdot
,\omega'))=\emptyset.
       \end{split}
\end{align}
     Note also that
     the  functions $\partial_\omega^\delta
 j_{\alpha,M}^{+} (\cdot ,\omega)$ and $\partial_{\omega'}^{\delta'} t_{\alpha,M}^{-} (\cdot
,\omega')$ are in $L^{2,a}_{-s_{\breve k}}$  due to
     \eqref{eq:j_space1}--\eqref{eq:j_space2b}. Next we insert a suitable
     phase space partition of unity conforming with
     \eqref{eq:emptyFirst} in the expression
     \begin{align*}
       \langle \partial_\omega^\delta
  j_{\alpha,M}^+(\cdot,\omega),
\chi_R\,\partial_{\omega'}^{\delta'}t_{\alpha,M}^-(\cdot,\omega')\rangle,
     \end{align*} and then we remove   the factor $\chi_R$ (by letting
     $R\to \infty)$. The `large $a$-part' is
     controlled by \eqref{eq:+bound} (using here only that   $s\geq
     0$) applying a single
     partition function $\chi_+(a)$. The `small
     $a$-part', treated as indicated above,  needs  many partition
     functions  using
     \eqref{eq:emptyFirst} and a sufficient  `sharpness' of the
     localization  of the
      partition  to make sure that for each
     term  the   partition operator  brings  at least one of
     the above two functions to   $L^{2,a}_{s_{\breve k}}$. Since the other function
     globally is in $L^{2,a}_{-s_{\breve k}}$  \cas applies. This argument
     is first used  with  the factor $\chi_R$ in place. Any commutator with
     this factor gives at least an extra factor $R^{\rho/2-1}$, so in
     the limit  we get (by dominated convergence) the corresponding expression
     without $\chi_R$ which  by the same arguments indeed is well-defined. Finally the convergence is
     uniform in $(\omega,\omega')\in \S^{n-1}\times \S^{n-1}$ since all involved bounds
     can be done   uniformly. This  feature relies on regularity
     of the classical constructions  and  the underlying integration by
     parts arguments (the feature  is only partially stated in Lemma
     \ref{lemma:wavefas}).

\subStep{II} We look at the
     second  term to the right in \eqref{eq:Sforma}, which we claim is
     well-defined with the limit taken
      locally uniformly in $\{\omega\cdot \omega'\neq \cos
  {\tfrac\rho {2-\rho}}\pi\}$. Before taking the limit we insert the expression
  \eqref{eq:resolBAS} (for the plus case) and
  get various   terms (possibly after a further expansion) to treat
  when applied to the function $\partial_{\omega'}^{\delta'}
  t_{\alpha,M}^-(\cdot,\omega')$ to the right.

It is convenient to introduce the following notation.
 \begin{align*}
u_\omega^+&=\partial_\omega^\delta t_{\alpha,M}^{+} (\cdot ,\omega),\\
u_{\omega'}^-&=\partial_{\omega'}^{\delta'} t_{\alpha,M}^{-} (\cdot
               ,\omega'),\\
\vM_\omega^+&=\Big\{z=(\hat x,
    \bar c, b)\in \T^*\mid  1-\sigma'\leq
 \hat x\cdot \omega \leq 1-\sigma,\,\, b\hat x+\bar c=
                                                 {F_{\sph}(\hat
                                                 x, \omega)}/\sqrt
                                                 \gamma\Big\},\\
   \vM_{\omega'}^-&=\Big\{z=(\hat x,
    \bar c, b)\in \T^*\mid 1-\sigma'\leq
- \hat x\cdot \omega' \leq 1-\sigma,\,\, b\hat x+\bar c=-
                                                 {F_{\sph}(\hat
                                                 x,- \omega')}/\sqrt
                                                 \gamma\Big\},\\
\vM^-_{\omega'+}&=\big\{\gamma(\tau, z) \in \T^*\mid  \,\tau\geq 0, \, z\in \vM_{\omega'}^-\big\}\cup\set{a,b=1}.
 \end{align*}

We  follow the indicated scheme, so suppose $\hat R$ is one of the terms of an expansion of
$R(\lambda_0+\i 0)$. Then we need to treat
\begin{align*}
  \lim_{R\to\infty} \langle u_\omega^+,
\chi_R \hat R \chi_R\,u_{\omega'}^-\rangle.
\end{align*} We first study the
expression $\langle u_\omega^+,
\chi_R \hat R \chi_R\,u_{\omega'}^-\rangle$ without the factors
$\chi_R$. We let $t_1=s_{\breve k}$ and intend to  find $t_2\in\R$ such that
\begin{align}\label{eq:strategy}
  \begin{split}
 WF^a_{t_1}( u_\omega^+)&\cap WF^a_{t_2}(\hat R
  u_{\omega'}^-)=\emptyset,\\
u_\omega^+&\in L^{2,a}_{-t_2},\quad
\hat R u_{\omega'}^-\in L^{2,a}_{-t_1}.
  \end{split}
     \end{align}
 More precisely we shall use \eqref{eq:strategy}  for all such terms $\hat R$ except
for a certain `remainder term'  that is treated differently (see the end
of the proof).

  Given  \eqref{eq:strategy}  we  insert a suitable phase space partition of
unity. The `large $a$-part' is treated
by using \eqref{eq:+bound}    to $ u_\omega^+$ (note  that $s\geq t_1$
so that \cas works). For the  `small
     $a$-part' the partition functions there are chosen such that for each
     term  the   partition operator  either brings  $ u_\omega^+$ to
      $L^{2,a}_{t_1}$ or $\hat R u_{\omega'}^-$ to $L^{2,a}_{t_2}$. In
      either case  \cas works to make sense to the expression. These arguments can also
     be done  with  the factors $\chi_R$ in place, and we can also
     control the uniformity in the angles  (cf. the discussion in Step I
     and a remark before Lemma \ref{lemma:propSing}).

The way to prove \eqref{eq:strategy} goes as
follows. Since  $t_1=s_{\breve k}\leq s$ it follows from \eqref{eq:wacom}
that
\begin{align}\label{eq:WF4}
  WF^a_{t_1}(u_\omega^+)\subset\vM_\omega^+.
\end{align} Suppose
\begin{subequations}
  \begin{align}
    \label{eq:hatbeti} \hat R u_{\omega'}^-\in L^{2,a}_{-t_1},
  \end{align} and that
$t_2$ is  chosen such that
\begin{align}\label{eq:WF4b}
WF^a_{t_2}(\hat R
  &u_{\omega'}^-)\cap\set{ a=1}\subset \vM^-_{\omega'+},\\
&u_\omega^+\in L^{2,a}_{-t_2}. \label{eq:WF4c}
  \end{align}
\end{subequations}
 Then
\begin{align*}  WF^a_{t_1}(u_\omega^+)\cap WF^a_{t_2}(\hat R
  u_{\omega'}^-)\subset
\vM_\omega^+\cap \vM^-_{\omega'+},
  \end{align*} so to finish the proof of~\eqref{eq:strategy} we just
  need to check that the right-hand side is empty. It follows from
  \cite[(3.5d)]{DS1} that $\vM_\omega^+\cap \set{a,b=1}=\emptyset$. To complete the proof  we note
  that
  \begin{align*}
    \forall z\in \vM_\omega^+ &:\quad \hat x(\tau,z)\to \omega\text {
    for } \tau\to +\infty,\\
\forall z\in \vM^-_{\omega'+}\setminus \set{a,b=1}&:\quad \hat x(\tau,z)\to -\omega'\text {
    for } \tau\to -\infty.
  \end{align*} But if $\omega\cdot \omega'\neq \cos
  \tfrac\rho {2-\rho}\pi$ we can not find  $z\in \T^*$ obeying these
  asymptotics. So indeed $\vM_\omega^+\cap
  \vM^-_{\omega'+}=\emptyset$  and  \eqref{eq:strategy} follows.

\subStep{III}
It remains to check \eqref{eq:hatbeti}--\eqref{eq:WF4c} for the  terms
$\hat R$ in an expansion of $R(\lambda_0+\i 0)$. Note that the
parameter $t_2$ may depend on the particular term $\hat R$ we consider.

\subStep{$\hat R=Sr^+_{\lambda_0}S^*$} We  write $\hat
R=(1\otimes r^+_{\lambda_0})\Pi=r^+_{\lambda_0}\Pi$ in order to apply
Lemmas \ref{lemma:propSing} and \ref{lemma:wavefas}.
Take  $t_2= s_{\breve  k}-2s_0$ (recall
$s_0=1/2+\rho/4$).
 From \eqref{eq:j_space2} it follows that
  $u_\omega^+\in  L^{2,a}_{-t}$ for $t>t_2-1$, so in particular \eqref{eq:WF4c} is proven.
  To show \eqref{eq:WF4b} we note that $WF^a_{s_{\breve k}}(u_{\omega'}^-)
  \subset \vM^-_{\omega'}$, cf. \eqref{eq:wacom}, and therefore  the condition  \eqref{eq:WF1} is
 fulfilled for
 $u=\Pi u_{\omega'}^-$ with $t= s_{\breve k}$ and
 $\kappa=\kappa_\sigma$, where
\begin{align*}
   \kappa_\sigma=-\cos\parb{(1-\rho/2)\cos^{-1}(1-\sigma)}.
 \end{align*} Then
we learn from
 \eqref{eq:WF3} that
\begin{align*}
  WF^a_{t_2}(\hat R
  u_{\omega'}^-)\cap\set{ a=1}\subset \vM^-_{\omega'+},
  \end{align*} showing \eqref{eq:WF4b}. It follows from \eqref{eq:j_space2}  and
  Lemma \ref{lemma:propSing} \ref{item:liW1}  that
    $\hat R u_{\omega'}^-\in
   L^{2,a}_{-t_1}$, showing \eqref{eq:hatbeti}.

\subStep{$\breve R(\lambda_0)\Pi'$}
 Take $t_2=t_1= s_{\breve k}$. By \eqref{eq:wi} $ u_{\omega'}^-, \breve R(\lambda_0)\Pi' u_{\omega'}^-\in L^{2,a}_{-t_1}$. From the previous case we
  know that  $u_\omega^+\in  L^{2,a}_{-t_2}$ and
  $WF^a_{t_1}(u_{\omega'}^-) \subset \vM^-_{\omega'}$. Hence   by   \eqref{eq:wii}
  \begin{align*}
    WF^a_{t_2}(\breve R(\lambda_0)\Pi'
  u_{\omega'}^-)\cap\set{ a=1}\subset WF^a_{t_2}(
    u_{\omega'}^-)\subset \vM^-_{\omega'}\subset \vM^-_{\omega'+}.
  \end{align*} We have shown \eqref{eq:hatbeti}--\eqref{eq:WF4c} in
  this case.

\subStep{$\widecheck R(\lambda_0+\i 0)$, cf. \eqref{eq:mainSplit}} Take  $t_2=s_{\breve k}-2s_0$ as we did
treating $\hat R=Sr^+_{\lambda_0}S^*$ above, so
  \eqref{eq:WF4c} is known. We also know that the condition  \eqref{eq:WF1} is
 fulfilled for
 $u= u_{\omega'}^-$ with $t=s_{\breve k}$ and
 $\kappa=\kappa_\sigma$, i.e.
\begin{align}\label{eq:WF1bb}
    WF^a_{s_{\breve k}}( u_{\omega'}^-)\cap\set{ b < \kappa, a=1}=\emptyset.
  \end{align}

We also  note that for any  $L\in \N$ (eventually taken large) and with
\begin{align*}
  K=v_{\lambda_0}
r^+_{\lambda_0}=\parb{-S^*I_a\breve R(\lambda_0)\Pi'I_aS+\vO(r^{-1-\rho})}r^+_{\lambda_0},
\end{align*}
\begin{align*}
    Sr^+_{\lambda_0}\parb{1+v_\lambda r^+_\lambda}^{-1}S^*
&=\sum^{2L-1}_{l=0}Sr^+_{\lambda_0} (-K)^lS^*\\
  &+Sr^+_{\lambda_0} (-K)^L\parb{1+K}^{-1} (-K)^LS^*.
\end{align*}
Note  that the present formula  $v^+_\lambda=v^-_\lambda$ gives naturally raise
to the  unambiguous  notation $v_\lambda$. Moreover  the term
$\vO(r^{-1-\rho})$ is actually $\vO(r^{-\infty})$ since $m_a=1$ making
$v_{\lambda_0}=\vO(r^{-2-2\rho})$ (as discussed after  \eqref{eq:eff_v}), however
 only the decay $v_{\lambda_0}=\vO(r^{-1-\rho})$ is used below.

We will treat  the terms in the summation essentially by the
methods used above, while the last term
will  be treated
differently. Let us first examine the  contribution
 to $\widecheck R(\lambda+
  \i0) $ from the terms in the summation.
\subStep{$l=0$}
 We consider
 \begin{align*}
   \widehat R_0:=-Sr^+_{\lambda_0}S^*I_a\breve R(\lambda_0)\Pi' -\breve
   R(\lambda_0)\Pi' I_aS r^+_{\lambda_0}S^*(1-I_a\breve
   R(\lambda_0)\Pi').
 \end{align*}
 By Lemma \ref{lemma:propSing},
 \eqref{eq:WF1bb} and \eqref{eq:wiii}--\eqref{eq:wii}
   it
 follows that
  $\widehat R_0 u_{\omega'}^-\in L^{2,a}_{-t_1}$, i.e.  \eqref{eq:hatbeti}
  holds for this contribution. (This is a
  rough bound due to the  extra factor $I_a$.)

Again we learn from
 \eqref{eq:WF3} (cf. \eqref{eq:WF1bb}) that
\begin{align*}
  WF^a_{t_2}(\widehat R_0
  u_{\omega'}^-)\cap\set{ a=1}\subset \vM^-_{\omega'+},
  \end{align*} showing \eqref{eq:WF4b}.
\subStep{$l\geq 1$}  We consider
 \begin{align*}
   \widehat R_l:=\parb{1-\breve
   R(\lambda_0)\Pi' I_a}S r^+_{\lambda_0}(-K)^lS^*\parb{1-I_a\breve
   R(\lambda_0)\Pi'}.
 \end{align*} Recalling that  $WF^a_{t_1}( u_{\omega'}^-) \subset
 \vM^-_{\omega'}$ we are lead to investigate the action of $S
 r^+_{\lambda_0}(-K)^lS^*$ to a vector $u\in L_{-\infty} ^{2,a}$  with $WF^a_{t_1}(u) \subset
 \vM^-_{\omega'}$. Each factor of $K$ improves the weight by
 $\inp{x}^{2s_0-1-\rho}=\inp{x}^{-q}$, $q:=\rho/2$, more precisely
 we obtain from Lemma \ref{lemma:propSing} that
 \begin{align*}
WF^a_{t_1+l q}(S
 (-K)^lS^*u) \cap\set{ a=1}\subset \vM^-_{\omega'+}&\subset \set{ b\geq \kappa_\sigma, a=1},\\
   WF^a_{t_2+l q}(S
 r^+_{\lambda_0}(-K)^lS^*u) \cap\set{ a=1}&\subset \vM^-_{\omega'+}.
 \end{align*} So we learn that
\begin{align}\label{eq:nBND}
   WF^a_{t_2+l q}(\widehat R_l   u_{\omega'}^-)\cap\set{ a=1}\subset \vM^-_{\omega'+},
 \end{align} in particular \eqref{eq:WF4b} holds.

By Lemma
 \ref{lemma:propSing} \ref{item:liW1} the global weight is improved by a factor
 $\inp{x}^{-q}$  for each action by $K$ as
 long as the condition $s'<-s_0$ of the lemma  is fulfilled
 (note that     $L^{2,a}_{-s_0}$ can
 not  be reached by the   action by  $r^+_{\lambda_0}$). Since $l\geq
 1$ and $
 u_{\omega'}^-\in L^{2}_{-t_2}$ we then  conclude that
 $\widehat R_l
 u_{\omega'}^-\in L^{2}_{-t_1}$, and  hence  also \eqref{eq:hatbeti} is
 shown.

\subStep{Remainder  term} Finally we need to examine
\begin{align*}
   \widehat R:=\parb{1-\breve
   R(\lambda_0)\Pi' I_a}S r^+_{\lambda_0} (-K)^L\parb{1+K}^{-1} (-K)^L     S^*\parb{1-I_a\breve
   R(\lambda_0)\Pi'}.
 \end{align*} The operator  $\parb{1+K}^{-1}=1\otimes \parb{1+K}^{-1}\in \vL (L^2_s(\bX))$,
 $s\in (s_0, s_0+q)$, but we dont  know its microlocal properties. Whence
 we proceed differently by  introducing   the vectors
 \begin{align*}
   v_{\omega'}^-&= S(-K)^L     S^*\parb{1-I_a\breve
   R(\lambda_0)\Pi'}u_{\omega'}^-,\\
v_{\omega}^+&=S\parbb{\parb{1-\breve
   R(\lambda_0)\Pi' I_a}S r^+_{\lambda_0} (-K)^L}^*u_{\omega}^+
 \end{align*} and write
 \begin{align*}
   \langle u_\omega^+,
 \widehat R u_{\omega'}^-\rangle=\langle v_\omega^+,
  \parb{1+K}^{-1}v_{\omega'}^-\rangle .
 \end{align*}
It suffices to show that for some $s\in (s_0,  s_0+q
)$ and $L$ sufficiently large
 \begin{align}\label{eq:suffv}
    v_{\omega'}^- \in L^{2,a}_{s},\quad  v_{\omega}^+\in L^{2,a}_{-s}.
 \end{align} To obtain the result for $ v_{\omega'}^- \in L^{2,a}_{s}$
  we use  Lemma
 \ref{lemma:propSing} repeatedly as explained above    improving the
 global
 weight by  a factor
 $\inp{x}^{-q}$  for each action by
 $K$. It follows that $ v_{\omega'}^- \in
 L^{2,a}_{s}$   holds for any  $s\in (s_0,  s_0+q)$ for    $L$ large  enough. For
 the assertion $v_{\omega}^+\in L^{2,a}_{-s}$ we invoke a parallel
 `incoming' version of Lemma \ref{lemma:propSing} for
 $r^-_{\lambda_0}$ (for simplicity not stated). Note that powers of
 $r^-_{\lambda_0}$ show up when we expand the adjoint to treat
 $v_{\omega}^+$. Since the construction of $u_{\omega}^+$ and
 $u_{\omega'}^-$ appears
 symmetric indeed the incoming version of Lemma \ref{lemma:propSing}
 applies in the same fashion to $u_{\omega}^+$ as we have seen the lemma applies to $u_{\omega'}^-$.

   \end{proof}

\subsection{Scattering for  physics
  models at a two-cluster threshold,  case $\widetilde\vA=\vA_1$}\label{subsec:Non-elastic
  scattering}

In the previous subsections we made several simplifying
assumptions under the condition of an effective attractive  slowly decaying
inter-cluster potential. These were made partly for simplicity of
presentation. Rather than giving a full account on how to remove these
assumptions under the weakest possible conditions we shall here focus
on the models of physics which we studied in  Section
\ref{sec:CoulRellich} (and actually we shall not give  details of proof
below). This means that the effective
inter-cluster potential here is attractive Coulombic, so in agreement
with  Section
\ref{sec:CoulRellich} we consider the case $\widetilde\vA=\vA_1$. (Note that
this condition  includes cases of overall  neutral as well as overall
non-neutral systems of particles.) We shall
focus on treating the multiple two-cluster case, but this will only be
 in the `generic' situation studied in the first part of the proof of Theorem
 \ref{thm:physical-modelsRell}. Whence our main interest here is
 scattering for the  physics
  models with  $\widetilde\vA=\vA_1$  in the simplest possible multiple case.

We consider the setting of Theorem
 \ref{thm:physical-modelsRell} with $\widetilde{\vA}=\vA_1$ under
 the two `generic' conditions \eqref{eq:dirCo1} and
 \eqref{eq:dirCo2}. Let us also  for simplicity  assume that
 $\#\vA_1=2$, i.e. $\vA_1=\set{a_1,a_2}$
 as in Proposition \ref{prop2.3},  and as for  the latter result we
 assume for convenience
 that $\lambda_0$ is a simple eigenvalue for $H^{a_j}$; $j=1,2$. We
 impose \eqref{eq:nonEg0},
 again for simplicity of presentation only.

We are interested in the four parts of the scattering operator
$S_{\beta\alpha}=\parb{W^+_\beta}^*
W^-_\alpha$ defined by \eqref{eq:ooo13}, where
$\alpha,\beta\in\set{(a_j, \lambda_0,\varphi_j)\mid j=1,2}$, and
particularly in the corresponding pieces of the scattering operator
\begin{align*}
 S_{\beta\alpha}(\lambda)
=-2\pi W_\beta^{+}(\lambda)^*T_\alpha^-(\lambda),\quad \lambda\in I^+_\delta=[\lambda_0,\lambda_0+\delta],
\end{align*}  cf.
\eqref{eq:SmatrixN}.

We can prove a complete  analogue of Theorem
\ref{thm:ScatN}. Note that for each channel there is
associated an effective one-body potential denoted by $w_j$ (or
$w_\alpha$), which essentially is attractive Coulombic. We define
correspondingly $\widecheck
S_{\alpha\alpha}(\lambda)$ as the remainder after the one-body
scattering matrix is subtracted, exactly as in Theorem
\ref{thm:ScatN}. Similarly we may define $\widecheck
S_{\beta\alpha}(\lambda)=S_{\beta\alpha}(\lambda)$ for $\alpha\neq
\beta$. With these conventions we obtain the assertion of Theorem
\ref{thm:ScatN}  (with the same value of $k$) for all of the four
entries of $\widecheck S_{\beta\alpha}(\lambda)$. The proof is
essentially the same. In particular the off-diagonal parts of the
scattering matrix are
compact, while the same is the case for the diagonal parts only after
a subtraction of a unitary operator, and this includes the threshold
energy.
 \begin{cor}\label{cor:elast-scattR} Under the above conditions the
    diagonal parts
   $S_{\alpha\alpha}(\lambda_0)$ are unitary up to a compact
   term.  In particular  elastic two-cluster scattering exists in the small
   inter-cluster energy regime.
   \end{cor}

  It is not known, even with the above simplifying assumptions,  if $S_{\alpha\alpha}(\lambda_0)$ is  an
 exact unitary operator (we can in fact not  exclude that $\ker
 S_{\alpha\alpha}(\lambda_0)\neq\set{0}$), cf.  Corollary
 \ref{cor:norm_less22} and a brief discussion in Section \ref{sec:Transmission problem at
  threshold}. In particular this is not known for
$\lambda_0=\Sigma_2$ and this threshold being
 multiple (note that in this case the threshold is
automatically a simple eigenvalue of the involved
 sub-Hamiltonians). On the other hand  obviously the scattering matrix
$S_{\alpha\alpha}(\lambda_0)$ in Theorem
\ref{thm:ScatN}  (where $\lambda_0$ is non-multiple) is an exact unitary operator for
$\lambda_0=\Sigma_2$, cf. Remark \ref{remark:elast-scattA}.

As for the second main result Theorem \ref{thm:sings} we need an
analogue of Theorem \ref{thm:repEigen} so that we can  first `improve'
the construction of $S_{\beta\alpha}(\lambda_0)$ in the spirit of
Lemma \ref{lem:conTransp}, not to be elaborated on. By mimicking Subsection
\ref{subsubsec: Elastic scattering at lowest threshold} we then obtain
the following result.
\begin{thm}
  \label{thm:non-elast-scatt0 } Under the above conditions and with $\lambda_0=\Sigma_2$
\begin{align*}
{ \mathrm
{sing \,supp\,}}S_{\beta\alpha}(\lambda_0)\quad
  \begin{cases}
   &\subset\{(\omega,\omega')\mid \omega\cdot \omega'= -1\}\quad \text{for}\quad \alpha=\beta,\\
&=\emptyset\quad \text{for}\quad \alpha\neq\beta.
  \end{cases}
  \end{align*}
\end{thm}

 \section{Effective $r^{-2}$ potentials,  atom-ion
   case}\label{sec:rtominus2 potentials} One may  consider elastic
 scattering in the non-slowly decaying case, as before  slightly above
 a
given two-cluster threshold $\lambda_0$. How does
 $S_{\alpha\alpha}(\lambda)$
behave as
$\lambda\to \lambda_0$?

Suppose $n=3$ and consider the dynamical nuclei model, see Subsection
\ref{First principal example}. Consider a two-cluster decomposition
$a=a_0=(C_1,C_2)$ of $N$ charged particles. Suppose \eqref{eq:discr},
i.e.
\begin{align*}
  \lambda_0\in \sigma_\d(H^a),
\end{align*} and that the total charge of cluster $C_1$ vanishes:
\begin{align*}
  Q_1=\sum_{j\in C_1}q_j=0.
\end{align*} Let {$\alpha=(a, \lambda_0,\varphi_\alpha)$} be a
{channel}; $\varphi_\alpha=\varphi^{C_1}\otimes
\varphi^{C_2}=\varphi^1\otimes \varphi^2$. Recall from Subsection
\ref{First principal example} the cluster charge moment
\begin{align*}
  \inp{\varphi^1,\widetilde
  Q_1\varphi^1}_2={\inp[\Big]{\varphi^1,\sum_{j\in
    C_1}q_j(x_j-R_1)\varphi^1}}.
\end{align*}

Now we can state a result from \cite{Sk2} (see \cite[Theorem 2.6 (3)]{Sk2}):

\begin{thm}[Elastic two-cluster scattering away from thresholds]
\label{thm:effective-c-n case}
  \begin{enumerate}[1)]
  \item There exist
\begin{align*}
  W_\alpha^\pm f=\lim_{t\to \pm \infty}
\,\e^{\i t H}\e^{-\i t H_a}\parb {\varphi_\alpha\otimes f}.
\end{align*}
\item Let $S_{\alpha\alpha}=(W_\alpha^+)^*W_\alpha^-$ and  {$I_c=(\lambda_0,
    \infty)\setminus \vT$}. Let  {$\theta$ denote  the coordinate vector of
    $\omega-\omega'\in \bX_a$}. Then for   any $\epsilon>0$
  \begin{align*}
    S_{\alpha\alpha}&\simeq \int_{I_c}^\oplus S_{\alpha\alpha}(\lambda)\;\d\lambda,\\
  S_{\alpha\alpha}(\lambda,\cdot )&\in C^\infty(I_c, \parb{\S^{2}\times \S^{2}}\setminus\{\theta=0\}),\\
S_{\alpha\alpha}(\lambda,\omega,\omega')-\delta (\theta)+&\pi^{-1}\tfrac{M_1M_2}{M_1+M_2}Q_2\inp{\varphi^1,\widetilde
  Q_1\varphi^1}_2\cdot
\tfrac{\theta}{|\theta|^2}=\vO_{{\lambda}} (|\theta|^{-\epsilon}).
  \end{align*}
\end{enumerate}
\end{thm}

The latter bound is $\lambda$-dependent with a locally uniform
dependence.
A natural question is, if  (under conditions) there  is a
uniform bound on an interval
$(\lambda_0,\lambda_0+\delta)$, $\delta>0$? This
question
 depends on the resolvent behaviour at $\lambda_0$, and Chapter
  \ref{chap:resolv-asympt-near} provides  some  information of
  relevance for  this
  type of  problem. A related problem is, if  low-energy elastic scattering  \emph{exists},  cf. Corollary
\ref{cor:elast-scattR}. In the next section we address the latter
problem, not only for the above case  (for which  $I_a(x^a=0)=0$) but also
for the effective repulsive Coulombic case.

\section{Non-transmission  at
  a threshold,  physics
  models}\label{sec:Transmission problem at threshold}
We consider a non-transmission problem  for  the  physics
  models of Sections \ref{$N$-body
  Schr\"odinger operators} and \ref{$N$-body Schr\"odinger operators with infinite mass
  nuclei}  with $n=3$. Although the statement of the problem makes sense more
generally we restrict the attention to these models.

Mimicking \eqref{eq:ooo13} we introduce for a given  channel
 $\alpha=(a,\lambda_0, \varphi_\alpha)$ (at a given two-cluster threshold
  $\lambda_0$) the channel wave operators
\begin{align}
  \begin{split}
    W_\alpha^{\pm}f&=\lim _{t\to {\pm}\infty}\e^{\i
  tH}\parb{1\otimes  J_0^{\pm}}\varphi_\alpha\otimes \e^{-\i t (p_a^2+\lambda_0)}f\\&=\lim _{t\to {\pm}\infty}\e^{\i
  tH}\parb {\varphi_\alpha\otimes J_0^{\pm}\e^{-\i t (p_a^2+\lambda_0)}f} ;\quad
  \hat f\in C^\infty_\c(\bX_a\setminus\{0\}). \label{eq:ooo132}
  \end{split}
\end{align}   In the
effective Coulombic cases  the  appearing
stationary modifiers $J_0^{\pm}$  are
chosen   as in
\eqref{eq:int_ope9} and \eqref{eq:ooo13} (i.e. also this way for the
repulsive case), and in the effective non-Coulombic case (like in Theorem
\ref{thm:effective-c-n case}) we  take  $J_0^{\pm}=1$.  Let $\Pi^\alpha=\ket {\varphi_\alpha}\bra {\varphi_\alpha}\otimes 1$ and
\begin{align*}
  \Pi^\alpha_{\pm}=\slim_{t\to \pm \infty} \e^{\i
  tH}\Pi^\alpha\e^{-\i
  tH}1_{\R\setminus\sigma_\pp(H)}(H).
\end{align*}

We are interested in `transmission' (for example  `exchange', cf. \cite{CT}), or rather lack of transmission
between channels  with the energy constraint of localization slightly above $\lambda_0$. Such lack of
transmission may be phrased mathematically
as
\begin{align}\label{eq:proConj}
  \norm[\big]{\parb{1-\Pi^\alpha_+}\Pi^\alpha_-F_\delta(H-\lambda_0)}\to
  0\text{ for }\delta\to
  0_+.
\end{align} Here $F_\delta$ denotes the
characteristic function $F([0,\delta])$ on $\R$.

Note that a state $\psi=1_{\R\setminus\sigma_\pp(H)}(H)\psi$ is in the range
of the projection $\Pi^\alpha_{\pm}$ if and only if $(1-\Pi^\alpha)\psi(t)\to 0$ for $t\to \pm \infty$; $\psi(t):=\e^{-\i tH}\psi$. Intuitively a consequence of
\eqref{eq:proConj} is that if for a state $\psi$ localized in energy slightly
above $\lambda_0$ and effectively $\psi(t)\approx \Pi^\alpha
\psi(t)$ for $t\to -\infty$ then also $\psi(t)\approx \Pi^\alpha
\psi(t)$ for $t\to +\infty$. Whence, asymptotically in the low  energy
regime, only elastic
scattering occurs and transmission to another channel than (the
incoming) $\alpha$ does not occur.

Clearly
$\Pi^\alpha_{\pm}W_\alpha^{\pm}=W_\alpha^{\pm}$, and by asymptotic
completeness (cf. \cite{De}) $\Pi^\alpha_{\pm}\subset W_\alpha^{\pm}\parb{W_\alpha^{\pm}}^*$. Whence
\begin{equation*}
  \Pi^\alpha_{\pm}=W_\alpha^{\pm}\parb{W_\alpha^{\pm}}^*.
\end{equation*} Consequently we can compute
\begin{equation*}
  \parb{1-\Pi^\alpha_+}\Pi^\alpha_-F_\delta(H-\lambda_0)=\parb{W_\alpha^--W_\alpha^+S_{\alpha\alpha}}F_\delta(p_a^2)\parb{W_\alpha^-}^*.
\end{equation*} Now
\begin{equation*}
  \norm{\parb{W_\alpha^--W_\alpha^+S_{\alpha\alpha}}f}^2=\norm{f}^2-\norm{S_{\alpha\alpha}f}^2,
\end{equation*}
 showing with the above identity that  a necessary and sufficient condition for \eqref{eq:proConj} is
that  $S_{\alpha\alpha}(\lambda)$  is asymptotically isometric, that is
\begin{equation}\label{eq:asymU}
  \norm{1-S_{\alpha\alpha}(\lambda)^*S_{\alpha\alpha}(\lambda)}\to 0 \text{ for
  }\lambda\to (\lambda_0)_+.
\end{equation}

Note that in the context of Corollary  \ref{cor:elast-scattR} it
follows from the assertion that  the limiting diagonal parts
$S_{\alpha\alpha}(\lambda_0)\neq 0$, however it
is \emph{not stated} that $S_{\alpha\alpha}(\lambda_0)$ are
isometric, and we can not show this to be the case. See
Remark \ref{remark:elast-scattA}  for  a conjecture on
transmission for the attractive Coulombic case.

We consider  the physical
 models with $n=3$, and for convenience we impose Condition \ref{cond:uniq}
 with $a_0=a$, i.e. we consider only the non-multiple case.
 Then there  are three  cases  for which we can show \eqref{eq:asymU} (or
 equivalently \eqref{eq:proConj}):
 \begin{enumerate}[\bf I)]
\item
 Effective  repulsive Coulombic
 case.
\item $I_a(x^a=0)=0$, `above the
 Hardy limit' and $\lambda_0$  be `regular'.
\item $I_a(x^a=0)=0$, `fastly decaying case' and $\lambda_0$  be
  `maximally exceptional of $1$st kind'.
\end{enumerate} Note that I) corresponds to Case 1 in Sections \ref{$N$-body
  Schr\"odinger operators} and \ref{$N$-body Schr\"odinger operators with infinite mass
  nuclei} (with an additional sign condition for charges).  II)
corresponds to Cases 2 and  3 in the same
sections, while III)  corresponds to Case 3 only and includes the
occurrence of  resonance states (but not bound states). The terminology `above the
 Hardy limit' is explained in the simple case $m_a=1$ in Subsection \ref{subsec:Some better
   arguments}, see \eqref {eq:non-osci},
 while `regular' as in Section \ref{sec:two-clust-threshres} refers to
 $\lambda_0$ not be neither an eigenvalue nor a resonance. In
 general the
 condition `above the
 Hardy limit' is a spectral condition   for the $m_a\times m_a$ matrix-valued
 effective potential $S^*_aI_aS_a=Q_a|x_a|^{-2}+B_a$,  or rather for
 its leading term  $Q_a(\hat x_a)$. By definition it means that the eigenvalues of
 $-\Delta_\theta+Q_a(\theta)$ are all strictly bigger than $-1/4$, cf.
 Section
 \ref{sec:CoulRellich}, and
 of course this condition is  fulfilled for Case 3 (where $Q_a=0$), but it might not
be fulfilled in Case 2.  The condition `exceptional of $1$st kind'
refers to
 $\lambda_0$ be a resonance but not  an eigenvalue, cf.   Section \ref{sec:two-clust-threshres}. The added word
 `maximally' refers to maximal multiplicity of the space of resonance
 states, i.e. $n_{\mathrm{res}}=m_a$. Note that the fastly decaying
 regular case
  is included in Case II).

We summarize our results  as follows.
\begin{thm}\label{thm:non-transmission-at} For the physics
 models with $n=3$ and with Condition \ref{cond:uniq} fulfilled
 at $\lambda_0$ for  $a_0=a$  there is no transmission assuming
 {\rm{I), II)}} or {\rm{III)}}. Whence for  each of these  cases
   there is no transmission in the  small inter-cluster
   energy regime from any   given  (incoming) channel
   $\alpha=(a,\lambda_0, \varphi_\alpha)$  to any different (outgoing)
  channel, i.e. \eqref{eq:proConj} and  \eqref{eq:asymU} are  fulfilled.
\end{thm}
 Our results  are   proven in subsequent subsections under
the following additional (convenient but non-essential) simplifying conditions:

\begin{enumerate} [a)]
\item   $m_a=1$.
\item $\lambda_0\notin
\sigma_{\pp}( H)\cup \sigma_{\pp}( H')$.
\end{enumerate}

\begin{remarks}\label{remarks:non-transmission-bem}
  \begin{enumerate}[1)]
  \item \label{item:tran1} We dont know  if \eqref{eq:asymU} holds in full generality for
    the fastly decaying case and $\lambda_0$ be an exceptional point of $2$nd
    or $3$rd kind, although we have  results with a decay condition for the latter
    cases (for which $\lambda_0\in
\sigma_{\pp}( H)$).  The proof
    for III) is based on Theorem
    \ref{thm:resolv-asympt-phys1st3rd} which applies  under the decay condition   $\ran
  \Pi_H\subset L^2_t$ for some $t> 3/2$, where
 $\Pi_H$ denotes the orthogonal projection onto $\ker
  (H-\lambda_0)$. Hence  under this
  additonal condition,
    III) can be extended to cover the `maximally exceptional case of
    $3$rd kind'. Note that the (local) scattering theory of $H$ and
    $H_\sigma:=H-\sigma\Pi_H$, $\sigma>0$,  are identical, and that
    $\lambda_0$ in this case is a `maximally exceptional point of $1$st kind' of  $H_\sigma$.

\item \label{item:tran2}
Similarly we dont know in full generality if the regularity
condition in II) is necessary for the
    critical decay rate case. However  for  the 'exceptional case  of
    $2$nd kind' we can use Theorem \ref{thm:resolv-asympt-phys2nd} for
    which the decay condition $\ran
  \Pi_H\subset L^2_t$ for some $t>1$ is needed. Note that
    $\lambda_0$ is a regular point of  $H_\sigma$ in this case.

\item \label{item:tran3} If $m_a>1$ and the  word
 `maximally' is omitted in Case {\rm{III)}},
 i.e. $n_{\mathrm{res}}\in\set{1,\dots,m_a-1}$,  we  show in
 Subsection \ref{subsec: An example of transmission} that
 \emph{transmission  does occurs}.
\item \label{item:tran4}  A key ingredient of  our proof of Theorem
  \ref{thm:non-transmission-at}  is various low-energy bounds on the one-body
  spectral density operator. These are absent for the
  effective attractive Coulombic case.
\end{enumerate}
  \end{remarks}

\subsection{Proof of  \eqref{eq:asymU}, Case I)}\label{subsec:Some  arguments}

We consider the Case I), i.e.   the repulsive Coulombic
 case.  We use Subsections \ref{subsec:positive-effect-potent} and
 \ref{subsec:psdep} and
 introduce a globally positive function $\widehat W$ as in
 \eqref{eq:posLowbnd} with $\rho=\bar \rho =1$. More
 precisely  we
 here  consider  the non-multiple case version discussed in Remark \ref{remark:formRes}, imposing the   above
 conditions  a) and b).  As in Subsection \ref{subsec:positive-effect-potent} we then
take $w=\widehat W$. We define $r^\pm_\lambda=(h_a+\lambda_0-\lambda
\mp\i 0)^{-1}$,  $h_a=p_a^2+w$ and $\lambda\geq \lambda_0$, and recall the bound
\begin{align}\label{eq:reLAP}
  \norm{\inp{x}^{-s}r^\pm_\lambda\inp{x}^{-s}}\leq
  C, \quad s>s_0=\tfrac12 +\tfrac  \rho 4,
\end{align} cf. \cite {Na} and \cite{Ya2}. Recall also that $r^+_{\lambda_0}=r^-_{\lambda_0}$.

Next we introduce relative wave operators
\begin{align*}
  \Omega_\alpha^\pm f=\lim_{t\to \pm \infty}
\,\e^{\i t H}\e^{-\i t (H_a+w)}\parb {\varphi_\alpha\otimes f}
\end{align*} and the corresponding scattering operator
 $(\Omega_\alpha^+)^*\Omega_\alpha^-$.

Using one-body Isozaki-Kitada scattering theory we can write
\begin{align}\label{eq:deltaIK}
  \delta(h_a-\lambda+\lambda_0)=W_1^\pm(\lambda-\lambda_0)W_1^
\pm(\lambda-\lambda_0)^*,
\end{align} cf. \cite {DS1}. Here $W_1^\pm(\cdot)$ signify wave
matrices for the pair $(p_a^2, h_a)$ corresponding to  Isozaki-Kitada
wave operators $W_1^\pm$, cf. \cite {DS1}. Clearly  the wave operators $
W_\alpha^\pm=\Omega_\alpha^\pm W_1^\pm$. Let us use $W_1^-$ to
diagonalize $h_a$. Whence we introduce the `identification operator'
$J_\alpha=\varphi_\alpha\otimes (W_1^-\cdot)$ and the corresponding relative wave operators
\begin{align*}
  \widetilde W_\alpha^\pm=\slim_{t\to \pm \infty}
\,\e^{\i t H}J_\alpha \e^{-\i t (p^2_a+\lambda_0)}.
\end{align*} Note that $\widetilde W_\alpha ^-=W_\alpha^-$,  while
$\widetilde W_\alpha^+=W_\alpha^+S_1$ with  $S_1$ being  the scattering
operator for the pair  $(p_a^2, h_a)$. Whence
\begin{align*}
  S_{\alpha\alpha}=\parb{W_\alpha^+}^*W_\alpha^-=S_1\parb{\widetilde
  W_\alpha^+}^*\widetilde W_\alpha^-=S_1\widetilde S_{\alpha\alpha},
\end{align*} yielding upon diagonalizing $p^2_a+\lambda_0$ by the
Fourier transform
\begin{align*}
  S_{\alpha\alpha}(\lambda)=S_1(\lambda-\lambda_0)\widetilde S_{\alpha\alpha}(\lambda).
\end{align*} (A similar factorization formula appears for a one-body
problem in \cite{Ba}.)
We can represent, cf.   \cite [Appendix
 A]{DS1} or \cite [Section 7.3]{Ya3},
\begin{align}\label{eq:tildeS}
  \widetilde S_{\alpha\alpha}(\lambda)=1-2\pi J_\alpha(\lambda)^{*}T_\alpha(\lambda)
   +2\pi \i T_\alpha(\lambda)^*R(\lambda+\i
  0)T_\alpha(\lambda),
\end{align} where
$J_\alpha(\lambda)=\varphi_\alpha\otimes
W_1^-(\lambda-\lambda_0)$ and $T_\alpha(\lambda)=\i
(H-\lambda_0)J_\alpha(\lambda)=\i(I_a-w)J_\alpha(\lambda)$. Note that formally
$J_\alpha(\lambda)=J_\alpha\vF_{\lambda_0}(\lambda)^*$,  where
$\vF_{\lambda_0}(\lambda)=\vF_0(\lambda-\lambda_0)$ with
$\vF_0(\cdot)$ specified by \eqref{eq:resFour}. Note also that Taylor
expansion yields  $T_\alpha(\lambda)\approx\vO(\inp{x}^{-2}\inp{x^a}^3)J_\alpha(\lambda)\approx\vO(\inp{x}^{-2})J_\alpha(\lambda)$.

Now we can verify \eqref{eq:asymU}. By inserting \eqref{eq:tildeS}
into \eqref{eq:asymU} it suffices to show that
\begin{subequations}
\begin{align}\label{eq:De1new}
  \norm{J_\alpha(\lambda)^{*}T_\alpha(\lambda)}&\to 0,\\
\norm{T_\alpha(\lambda)^*R(\lambda+\i
  0)T_\alpha(\lambda)}&\to 0.\label{eq:De2new}
\end{align}
\end{subequations}

To show \eqref{eq:De1new} we note the bound
\begin{align}
  \label{eq:35resRep}
  \norm{\delta(h_a-\lambda+\lambda_0)}_{\vL(L^2_s ,L^2_{-s})}=o((\lambda-\lambda_0)^0)\text
  { for }s>s_0,
\end{align} cf.  \cite {Na}. We use it with $s=1$  in combination  with \eqref{eq:deltaIK},
proving \eqref{eq:De1new}.

To derive  \eqref{eq:De2new} it suffices by the same argument to show the bound
\begin{align*}
\sup_{\lambda\in I_\delta^+}\,\,\norm{\inp{x}^{-1}
R(\lambda+\i
  0) \inp{x}^{-1}}  <\infty,
\end{align*} which  in turn follows from \eqref{eq:resolBAS},
\eqref{eq:mainSplit} (cf. Remark \ref{remark:formRes})   and
\eqref{eq:reLAP}.

\subsection{Proof of  \eqref{eq:asymU}, Case II)}\label{subsec:Some better
  arguments}

For    Case II)  the wave operators are defined   by  \eqref{eq:ooo132}
with  $J_0^{\pm}=1$. Note the following
formula for the scattering matrix for $\lambda$ slightly above  $\lambda_0$,
\begin{align}\label{eq:scMatr}
  S_{\alpha\alpha}(\lambda)= 1-2\pi J^+_\alpha(\lambda)^{*}T^-_\alpha(\lambda)
   +2\pi \i T^{+}_\alpha(\lambda)^*R(\lambda+\i
  0)T^-_\alpha(\lambda),
\end{align} where
$J^+_\alpha(\lambda)=J^-_\alpha(\lambda)=\varphi_\alpha\otimes
\vF_{\lambda_0}(\lambda)^*=:J_\alpha(\lambda)$, $\vF_{\lambda_0}(\lambda)=\vF_0(\lambda-\lambda_0)$  and
similarly $T^+_\alpha(\lambda)=T^-_\alpha(\lambda)=\i
I_aJ_\alpha(\lambda)=:T_\alpha(\lambda)$.

As in Theorem \ref{thm:ScatN} we subtract a one-body scattering
matrix introducing
\begin{align*}
  \widecheck S_{\alpha\alpha}(\lambda)=S_{\alpha\alpha}(\lambda)-
  S_w(\lambda-\lambda_0);
\end{align*} here $w$ is a reel  effective one-body potential to be
determined. Since we know that $S_w(\lambda-\lambda_0) $ is unitary
we aim at showing  that
\begin{align}\label{eq:aim}
  \norm{\widecheck S_{\alpha\alpha}(\lambda)}\to 0\text{ for }\lambda\to (\lambda_0)_+.
\end{align}

To examine this problem we need a formula for $ R(\lambda+\i 0)$ like
the one we used in the proof of Theorem \ref{thm:ScatN}. In agreement
with  the condition of Case II)
$\lambda_0$ is neither  a resonance nor an eigenvalue of $H$  and
the `above the Hardy limit' condition
\begin{equation}\label{eq:non-osci}
  \nu_0:=\min_{\nu\in \sigma} \Re \nu >0
\end{equation} is imposed. Here  the set $\sigma$ is given as in the analogous formula \eqref{eq:crrit}
 for the multiple case, that is computed by the eigenvalues  of  a perturbed
 Laplace-Beltrami operator $-\Delta_{\theta}+q(\theta)$  as
 in  \eqref{eq:relmunu}. Alternatively stated
\begin{equation*}
  \inf \sigma\parb{-\Delta_\theta+q(\theta)}>-1/4.
\end{equation*}

We assume  that
$m_a=1$  and  that $\lambda_0\notin
\sigma_{\pp}( H')$,  cf.  the conditions a) and b) stated after
Theorem \ref{thm:non-transmission-at}.

Under the condition \eqref{eq:non-osci} one could hope  that   the operator
$(p_a^2+W+\lambda_0-\lambda \mp\i 0)^{-1}$ with $W:=\inp{\varphi_\alpha,
  I_a\varphi_\alpha\otimes \cdot}$  would be a good auxillary operator for
stydying asymptotics as $\lambda\to (\lambda_0)_+$, cf.   Subsections
\ref{subsec:Homogeneous-effect-potent} and \ref{subsec:hd2ep}. However zero   could be an
eigenvalue or a
resonance for $p_a^2+W$,  spoiling this idea. To avoid this
problem we are lead to  modifying
$W(x_a)$ by a suitable stronger decaying term (the perturbation
 will be of order $\vO(\abs{x_a}^{-3})$) to assure zero be regular for
 the auxillary operator. More precisely we introduce, using the same quadratic partition of unity as in \eqref{eq:ansG},
\begin{align*}
   h_a&=\chi_1(r)p_a^2\chi_1(r)+\chi_2(r)
\parb{p_a^2+\tfrac{q(\theta)}{r^2}}\chi_2(r) =p_a^2+ w,\\w &=\abs{\nabla\chi_1}^2+\abs{\nabla\chi_1}^2+\tfrac{q(\theta)}{r^2}\chi^2_2=W+\vO(r^{-3}),\\
r_\lambda^\pm&=(h_a-k^2 \mp\i 0)^{-1},\quad k\in (0,1], \quad \lambda=\lambda_0+k^2.
\end{align*}
Due to \eqref{eq:non-osci}  $h_a\geq 0$, and in fact $
h_a$ is strictly
positive (i.e. $h_a\geq 0$ and $0\notin \sigma_{\pp}(h_a)$).

We will
show that $h_a$ does not have a resonance at zero  by an integration
by parts argument, cf. \cite{Ji},
and  we will show a number of properties of $r_\lambda^\pm$ by a
parametrix construction. In addition we will need  properties of
 the restriction of the Fourier transform
$\vF_{0}(k^2)$.

\subStep{1}  Let, in analogy with the zero-energy formulas  \eqref{eq:ansG},
\begin{align*}
  G_k^\pm&=\chi_1(r)( h_a\mp\i)^{-1}\chi_1(r)+\sum_{\nu\in
  \sigma}\chi_2(r) r^{-1} R^\pm_{\nu,k} r \chi_2(r)\otimes
     P_{\nu};\\
  R^\pm_{\nu,k}&=( p^2+\tfrac{\nu^2-1/4}{r^2}-k^2 \mp\i 0)^{-1},\quad
               p=-\i \partial_r,\,k>0.
\end{align*} Here $ R^\pm_{\nu,k}\in \vL_{s',-s}:=\vL\parb{L_{s'}^2(\R_+,\d
  r),L_{-s}^2(\R_+,\d r)}$ for $s,s'>1$, in fact with a bound
independent of $\nu\geq\nu_0$ and $k\in (0,1)$. For simplicity  we
assume here and below that
$\nu_0\in (0,1)$. See the proof of
\cite[Proposition 4.1]{Ca}, and  note also that this  proof shows that the weighted space operator
norm vanishes uniformly in $k$ as $\nu\to \infty$. We define the
adjoint $\parb{ K^\pm_{k}}^*=\parb{h_a-k^2 }\parb{
G_k^\pm}^*-1
$  and conclude that $K^\pm_{0}:=\lim_{k\to 0_+ }
K^\pm_{k}\in \vL\parb{\vH^1_{-1-\epsilon}}$ for any sufficiently  small $\epsilon>0$.

\subStep{2}
 We derive   bounds  of the operator $
R^+_{\nu,k}$ acting on functions on $\R_+$, here assuming $\nu>1$. Let
$\epsilon\in (0,\nu_0)$. The kernel is given
explicitly as
\begin{align*}
R^+_{\nu,k}(r,r')=-2^{-1}\sqrt{rr'}\int _0^\infty\, \exp\parb{\i\rho t+\i
  k^2 r r'/(2t)-\i \pi \nu/2}\,J_\nu(t) t^{-1}\,\d t;\quad \rho=\tfrac{r^2+r'^2}{2rr'}.
\end{align*} (See for example \cite {Ca}.)  We write
$1=\chi_{2kr'}(t)+\bar\chi_{2kr'}(t)$, cf. the notation
\eqref{eq:14.1.7.23.24}, and
\begin{align*}
   \exp(\cdot)=\parb{1+\rho^2(1-y)^2}^{-1}\parb{1-\i\rho
     (1-y)\partial_t}\exp(\cdot);\quad y=\tfrac{k^2(rr')^2}{(r^2+r'^2)t^2}.
\end{align*}
We split the above integral $\int_0^\infty=\int_0^\infty \cdot \chi_{2kr'}(t)\,\d t +\int_0^\infty
\cdot \bar\chi_{2kr'}(t)\,\d t$ and note that for $\nu\geq \nu_0$ a  Bessel
function bound, cf. the proof of \cite[Proposition 4.1]{Ca},  yields
\begin{align*}
  \abs[\Big]{\int_0^{\infty}\cdot \chi_{2kr'}(t)\,\d t}\leq
  o(\nu^0)\min \{(kr')^{\nu_0},1\}\leq o(\nu^0)(kr')^\epsilon,
\end{align*}
\begin{subequations}
leading with the  Hilbert-Schmidt criterion to the operator bound
\begin{align}\label{eq:smal2}
  \norm[\Big]{\sqrt{rr'}\int_0^{\infty}\cdot \chi_{2kr'}(t)\,\d t}_{\vL_{s',-s}}\leq
  o(\nu^0)k^\epsilon;\quad  s>1,\,s'>1+\epsilon.
\end{align} On the other hand, assuming  (for \eqref{eq:smal4}) that $\nu>1$,
\begin{align*}
  \bar\chi_{2kr'}(t)y\leq 1/4,
\end{align*} leading by integration by parts and Bessel function bounds
 to
 \begin{align}\label{eq:smal4}
  \abs[\Big]{\int_0^{\infty}\cdot \bar\chi_{2kr'}(t)\,\d t}&\leq
  o(\nu^0)\rho^{-1},\\
\norm[\Big]{\sqrt{rr'}\int_0^{\infty}\cdot \bar\chi_{2kr'}(t)\,\d t}_{\vL_{s',-s}}&\leq
  o(\nu^0);\quad  s+s'>2,\,s,s'>0.\label{eq:smal5}
\end{align}
 Note that for \eqref{eq:smal5} we used \eqref{eq:smal4},
the  Hilbert-Schmidt criterion, the bound
\begin{align*}
  r^{\kappa}\rho^{-1}r'^{-\kappa}\leq 2;\quad \kappa\in [-1,1],
\end{align*} as well as the identity
\begin{align*}
  2J_\nu'(t)=J_{\nu-1}(t)-J_{\nu-1}(t);
\end{align*} cf. \cite [(3.6.17)]{Ta1}.
 \end{subequations}
\begin{subequations}
Now by combining \eqref{eq:smal2} and \eqref{eq:smal5} we can write
\begin{align}\label{eq:smal6}
  \begin{split}
   R^+_{\nu,k}&=R^+_{\nu,k,1}+R^+_{\nu,k,2};\\
R^+_{\nu,k,1}&\in \vL_{s',-s};\quad s>1,\,s'>1+\epsilon,\\
\norm{R^+_{\nu,k,1}}_{\vL_{s',-s}}&\leq o(\nu^0)k^\epsilon.\\
R^+_{\nu,k,2}&\in \vL_{s',-s};\quad s=1-\epsilon,\,s'>1+\epsilon,\\
\norm{R^+_{\nu,k,2}}_{\vL_{s',-s}}&\leq o(\nu^0).\end{split}
\end{align}

\subStep{3} We show  bounds  of  $
R^+_{\nu,k}$ for the case $\nu\leq 1$.  Let
again $\epsilon\in (0,\nu_0)$.
 Note the representation, cf. \cite [pp. 228--230]{Ta1},
\begin{align*}
R^+_{\nu,k}(r,r')=\tfrac{\pi \i}{2k}(kr_<)^{1/2}J_\nu(kr_<)(kr_>)^{1/2}H^{(1)}_\nu(kr_>),
\end{align*} as well as the  classical  pointwise bounds
\begin{align*}
  \begin{split}
   \abs{t^{1/2}J_\nu(t)}&\leq
                         C_\nu\parb{\tfrac{t}{\inp{t}}}^{\nu+1/2},\\
\abs{t^{1/2}H_\nu^{(1)}(t)}&\leq
C_\nu\parb{\tfrac{t}{\inp{t}}}^{-\nu+1/2}.
 \end{split}
\end{align*}
 Whence
 \begin{align*}
   \forall \nu\in (0,1/2]:\quad \abs{(kr_<)^{1/2}J_\nu(kr_<)(kr_>)^{1/2}H^{(1)}_\nu(kr_>)}\leq C_\nu^2k\,r_<^{\nu+1/2}\,r_>^{-\nu+1/2},
\end{align*}
\begin{align*}
   \forall \nu\in ( 1/2,\infty):\quad
   \abs{(kr_<)^{1/2}J_\nu(kr_<)(kr_>)^{1/2}H^{(1)}_\nu(kr_>)}\leq
    C_\nu^2k\,r_<.
 \end{align*} Using the Hilbert-Schmidt criterion we obtain from
 these bounds:
\begin{align}\label{eq:BesHan10}
   \begin{split}
   \forall \nu\in (0,1/2]:\quad R^+_{\nu,k}=R^+_{\nu,k,2}&\in \vL_{s',-s};\quad
   s=1-\epsilon,\,s'>1+\epsilon,\\&
\norm{R^+_{\nu,k,2}}_{\vL_{s',-s}}\leq C_{\nu,s,s'}.
   \end{split}
\end{align}
\begin{align}\label{eq:smal6220}
  \begin{split}
  \forall \nu\in (1/2,1]:\quad R^+_{\nu,k}=R^+_{\nu,k,2}&\in \vL_{s',-s};\quad
   s,s'>1/2, \,s+s'>2,\\&
\norm{R^+_{\nu,k,2}}_{\vL_{s',-s}}\leq C_{\nu,s,s'}.
 \end{split}
\end{align}

Clearly $R^+_{\nu,k,1}=0$ in \eqref{eq:BesHan10} and in \eqref{eq:smal6220}, and it is
natural to ask if we also can take $R^+_{\nu,k,1}=0$ in
\eqref{eq:smal6}? Well, we dont know. Note that
it is known for the above pointwise bound  of the Bessel function
that   the constant $C_\nu\to \infty$  for
$\nu\to \infty$, see  \cite {La}. This  is the reason  we proceeded  differently
in Step \textit{2} to treat the regime $\nu>1$.
\end{subequations}
\begin{subequations}
  \subStep{4}   With reference to Step \textit{1},  we claim that
\begin{align}\label{eq:notin_specK2}
  -1\notin
\sigma(K^+_{0}).
\end{align} This is equivalent to asserting that $r_\lambda^\pm$ is    regular at
$k=0$, so it remains to show that  $0$ is not a `resonance' for $h_a$. We say
that  $0$ is  a resonance if there exists $0\neq f\in
\vH^1_{-s_0}\setminus \vH^1$, $s_0=1+\nu_0$,
such that $h_af=0$, cf. Definition \ref{defn:hom}. Note that the
notation $\vH_t^1$ here and henceforth is used as an alternative to
$H_t^1$ (to conform with the notation of Chapter \ref{chap:lowest thr}).  So suppose $ f\in
\vH^1_{-s_0} $ obeys $ h_af=0$. Then we  learn from the proof of Theorem
\ref{thm:short-effect-potent}  using
\eqref{eq:smal6}--\eqref{eq:smal6220} with $k=0$ that $f\in \vH^1_{-s'}$ for
any $s'>1-\nu_0$ (note that $K^+_{0} \in \vC\parb{\vH^1_{-s'}}$ for
$s'\in (1-\nu_0, 2-\nu_0)$). In particular $f\in \vH^1_{-1}$. Next we mimic the
proof of Theorem \ref{thm:negat-effect-potent}, that is integrate by part
and deduce that
\begin{align*}
  0=\lim_{R\to \infty}\inp{\chi_Rf, \chi_R h_af}= \lim_{R\to
  \infty}\inp{\chi_Rf,  h_a\chi_Rf}\geq c\norm{|x|^{-1}f}^2.
\end{align*} Whence $f=0$.
\subStep{5}  We can  obtain    bounds  of $r_\lambda^\pm$ using \eqref{eq:notin_specK2}  and
\eqref{eq:smal6}--\eqref{eq:smal6220}.  Note that
thanks to  the
time-reversal property the latter bounds also hold with  $R^+_{\nu,k}$ replaced by
$R^-_{\nu,k}$. Bounds of $r_\lambda^\pm$ can then be derived from  the formulas
\begin{align}\label{eq:perResFo}
   r_\lambda^+=\parb{1+ K^+_{k}}^{-1}
(G_k^-)^*=G_k^+\parb{1+ (K^-_{k})^*}^{-1}.
\end{align} We can here use $G^\pm_k=G^\pm_{k,1}+G^\pm_{k,2}$ in
agreement with the splittings
$R^+_{\nu,k}=R^+_{\nu,k,1}+R^+_{\nu,k,2}$ in
\eqref{eq:smal6}--\eqref{eq:smal6220}. Implemented in
\eqref{eq:perResFo} leads to two formulas for
$r^+_\lambda=r^+_{\lambda,1}+r^+_{\lambda,2}$ exhibiting  mapping properties
directly determined  by \eqref{eq:smal6}--\eqref{eq:smal6220}.

Since we are assuming that
$\lambda_0$ is not a resonance, we are lead to write
\begin{equation*}
  -E_{\vH}^+(\lambda)=p_a^2+w+v^+_\lambda+\lambda_0-\lambda
\end{equation*} and then consider
\begin{equation*}
  \parb{ \widetilde K^+_{k}}^*=\parb{-E_{\vH}^+(\lambda)^*}\parb{
r_\lambda^+}^*-1.
\end{equation*}
 By the regularity condition of  Case II) the limiting operator $\widetilde K^+_{0}:=\lim_{k\to 0_+ }
\widetilde K^+_{k}\in \vL\parb{\vH^1_{-1-\epsilon}}$  (which exists for any sufficiently
small $\epsilon>0$) fulfills   that
\begin{align}\label{eq:notin_specK2e}
  -1\notin
\sigma( \widetilde K^+_{0}).
\end{align} By using
\eqref{eq:notin_specK2}--\eqref{eq:notin_specK2e} in combination with
\eqref{eq:smal6}--\eqref{eq:smal6220}  we  can  obtain bounds of the
inverse of $E_{\vH}^+(\lambda)$. However we also need bounds on the
Fourier transform to treat \eqref{eq:scMatr}.
 \end{subequations}

\subStep{6} Bounds on the restriction of the Fourier transform. We note   the bound
\begin{align*}
  \norm{\vF_{\lambda_0}(\lambda)\inp {x}^{-1-\epsilon}}_{\vL(L^2(\R^3),L^2(\S^2))}\leq
  C_{\epsilon,\epsilon'}(\lambda-\lambda_0)^{\epsilon'},\quad
  0<2\epsilon'<\epsilon\leq 1/2,
\end{align*} which follows by using the
well-known expression for the free three-dimensional resolvent kernel
and  the Hilbert-Schmidt criterion. In terms of the $k$-variable
it corresponds to the first of the following
bounds. For the  second bound we refer to \cite[Theorem 3.2]{Ag}.
\begin{align*}
    \norm{\vF_0(k^2)\inp {x}^{-1-\bar\epsilon}}_{\vL(L^2(\bX),L^2(\S^2))}&\leq
  Ck^{\epsilon'},\quad  0<\epsilon'<\bar\epsilon\leq 1/2,\\
 \norm{\vF_0(k^2)\inp {x}^{-s}}^2_{\vL(L^2(\bX),L^2(\S^2))}&\leq
  Ck^{-1},\quad s>1/2,
  \end{align*}

 Interpolation of these bounds  yields
 \begin{align}\label{eq:interp}
\begin{split}
   \forall\kappa\in[0,1],\,&\epsilon'<\bar \epsilon\leq 1/2,\,s>1/2:\\ &\norm{\vF_0(k^2)\inp {x}^{-\kappa(1+\bar\epsilon)-(1-\kappa)s}}_{\vL(L^2(\bX),L^2(\S^2))}\leq
  Ck^{\kappa\epsilon'-(1-\kappa)/2}.
\end{split}
 \end{align}

\subStep{7}
 We can now treat \eqref{eq:scMatr}. The middle term
of \eqref{eq:scMatr}  may be written
\begin{align*}
  -2\pi J_\alpha(\lambda)^{*}T_\alpha(\lambda)&=-2\pi\i
  \vF_{\lambda_0}(\lambda) \parb{w+\vO(r^{-3})}\vF_{\lambda_0}(\lambda)^*\\&=-2\pi\i
  \vF_{\lambda_0}(\lambda) w\vF_{\lambda_0}(\lambda)^*+o((\lambda-\lambda_0)^0).
\end{align*} The first term corresponds to the `Born term' for the  one-body scattering
matrix with potential $w$,  and similarly in the third term of
\eqref{eq:scMatr}  the contribution $2\pi \i
T_\alpha(\lambda)^*Sr^+_\lambda S^*T_\alpha(\lambda)$,
recalling  $r^\pm_\lambda=(p_a^2+w+\lambda_0-\lambda \mp\i 0)^{-1}$,
may be written up to a term of order $o((\lambda-\lambda_0)^0)$  as
\begin{align*}
 2\pi \i
\vF_{\lambda_0}(\lambda)wr^+_\lambda w\vF_{\lambda_0}(\lambda)^*.
\end{align*}
 The  combination of this expression,   the identity operator $1$ and the  Born
 term is exactly  $S_w(\lambda-\lambda_0)$. Whence it suffices to show that
\begin{align*}
  \norm{T_\alpha(\lambda)^*\parb{R(\lambda+\i
  0)-Sr^+_\lambda
  S^*}T_\alpha(\lambda)}\to 0.
\end{align*}
Using formulas as \eqref{eq:resolBAS} and \eqref{eq:mainSplit} (cf. Remark \ref{remark:formRes}) it  suffices
to show that
\begin{subequations}
\begin{align}\label{eq:De1}
  \norm{T_\alpha(\lambda)^*\breve R(\lambda+\i
  0)\Pi'T_\alpha(\lambda)}&\to 0,\\
\norm{T_\alpha(\lambda)^*\widecheck R(\lambda+\i
  0)T_\alpha(\lambda)}&\to 0.\label{eq:De2}
\end{align}

To prove  these assertions we look at three cases   a)--c) as in  the proof of Theorem \ref{thm:ScatN}.
\subStep{\ref{item:15}} The  asymptotics \eqref{eq:De1} follows by first writing
\begin{align*}
  T_\alpha(\lambda)^*\breve R(\lambda+\i
  0)\Pi'T_\alpha(\lambda)=J_\alpha(\lambda)^*\vO(\inp{x}^{-2})\breve R(\lambda+\i
  0)\Pi'\vO(\inp{x}^{-2})J_\alpha(\lambda)
\end{align*} and then using Theorem \ref{thm:powers} and
\eqref{eq:interp}.
\end{subequations}

 As for \eqref{eq:De2} we expand $\widecheck R(\lambda+\i
0)$ as in the proof of Theorem \ref{thm:ScatN}, cf. Cases b) and
  c) in the proof.

\subStep{\ref{item:16}}  We write
\begin{align*}
  &T_\alpha(\lambda)^*Sr^+_\lambda\parb{1+v^+_\lambda
    r^+_\lambda}^{-1}S^*T_\alpha(\lambda)-T_\alpha(\lambda)^*Sr^+_\lambda
    S^*T_\alpha(\lambda)\\
&=-T^+_\alpha(\lambda)^*Sr^+_\lambda \tilde v^+_\lambda r^+_\lambda S^*T^-_\alpha(\lambda);\\
&\quad \quad \quad \tilde v^+_\lambda = \parb{1+v^+_\lambda
  r^+_\lambda}^{-1}v^+_\lambda.
\end{align*}

In this formula we insert the two representations of
$r^+_\lambda=r^+_{\lambda,1}+r^+_{\lambda,2}$ from the beginning of
Step \textit{5} (replacing the  factor of
$r^+_\lambda$ to the right with the middle expression of
\eqref{eq:perResFo} and replacing the  factor of
$r^+_\lambda$ to the left  with the third  expression of
\eqref{eq:perResFo}). Then   we can show the
desired decay.

When we exand the two sums we obtain four terms. Let us first consider
the contribution when the type $r^+_{\lambda,2}$ appear twice.  We
insert $1=\inp{x}^{s}\inp{x}^{-s}$ with $s=1-\epsilon$, $\epsilon>0$
small, next to the two factors of $r^+_{\lambda,2}$. Thus for the
left factor we write
$r^+_{\lambda,2}=\inp{x}^{s}\inp{x}^{-s}r^+_{\lambda,2}$ and then in
turn
$r^+_{\lambda,2}=r^+_{\lambda,2}\inp{x}^{-1-2\epsilon}\inp{x}^{1+2\epsilon}$
using bounds (essentially given by
\eqref{eq:smal6}--\eqref{eq:smal6220}) to estimate
$\inp{x}^{-s}r^+_{\lambda,2}\inp{x}^{-1-2\epsilon}$. Similary for the
right factor we write
$r^+_{\lambda,2}=r^+_{\lambda,2}\inp{x}^{-s}\inp{x}^{s}$. In
combination with \eqref{eq:interp} (used with a vanishing right-hand
side as $k\to 0_+$) we then obtain the bound
$o((\lambda-\lambda_0)^0)$.

We can obtain the same conclusion for the three other terms  in a
similar way. The only difference lies in the choice of
parameters. Thus the  worse mapping property of $r^+_{\lambda,1}$
leads to an application of \eqref{eq:interp} with a (small)
negative power of $k$. However the bounds \eqref{eq:smal6}  offer positive powers in this case, which in
combination indeed leads to the bound
$o((\lambda-\lambda_0)^0)$.

\subStep{~\ref{item:18}} There are four  terms to be considered  given (up to the
common  factor
$2\pi\i$) by:
\begin{align*}
  &S_1=-T_\alpha(\lambda)^*Sr^+_\lambda\parb{1+v^+_\lambda  r^+_\lambda}^{-1}S^*I_a\breve R(\lambda+ \i
    0)\Pi'T_\alpha(\lambda),\\
&S_2=T_\alpha(\lambda)^* \breve R(\lambda+\i
  0)\Pi' I_aSr^+_\lambda\parb{1+v^+_\lambda r^+_\lambda}^{-1}S^* I_{a}\breve R(\lambda+ \i 0
    )\Pi' T_\alpha(\lambda),\\
&S_3=-T_\alpha(\lambda)^* \breve R(\lambda+\i
  0)\Pi' I_aSr^+_\lambda S^* T_\alpha(\lambda),\\
&S_4=T_\alpha(\lambda)^* \breve R(\lambda+\i
  0)\Pi' I_aSr^+_\lambda\parb{1+v^+_\lambda  r^+_\lambda}^{-1}
  v^+_\lambda  r^+_\lambda S^*  T_\alpha(\lambda).
\end{align*}

We argue similarly as for  Case b) and conclude again the
bound $o((\lambda-\lambda_0)^0)$, skipping  the details. Whence
\eqref{eq:De2} is established.

\subsection{Proof of  \eqref{eq:asymU}, Case III)}\label{subsec: caseIII}
 We consider  the assertion \eqref{eq:asymU} for Case 3  of  Sections \ref{$N$-body
  Schr\"odinger operators} and \ref{$N$-body Schr\"odinger operators with infinite mass
  nuclei} under the conditions a) and b) stated after Theorem \ref{thm:non-transmission-at}.
  By definition of Case III), $\lambda_0$ is a
resonance  but not an eigenvalue of $H$. We can apply \eqref{asymRz1b}  of Theorem \ref{thm5.20}
with $\rho=1$ (or alternatively Theorem \ref{thm:resolv-asympt-phys1st3rd}), used formally with small $z>0$. More
precisely we
consider $R(\lambda +\i 0)$ with $\lambda$ slightly above
$\lambda_0$, given by  \eqref{asymRz1b}. We  insert this expression in
\eqref{eq:scMatr} (again with the wave operators given by \eqref{eq:ooo132}
with  $J_0^{\pm}=1$).

As $\lambda\to (\lambda_0)_+$ the second term on the right-hand side of
\eqref{eq:scMatr}  disappears in the limit, cf. \eqref{eq:interp}. As
for the third term there is a cancellation of powers of
$\lambda-\lambda_0$ exactly as for  the one-body resonance case of
\cite{JK}. In fact we compute as in \cite[Section 5]{JK}, and by using Remark \ref{remark:resolv-asympt-nearNorm},
\begin{align*}
  \lim_{\lambda\to (\lambda_0)_+}\,2\pi \i T^{+}_\alpha(\lambda)^*R(\lambda+\i
  0)T^-_\alpha(\lambda)=-2 \inp{Y_0,\cdot}Y_0;\quad Y_0=(4\pi)^{-1/2}.
\end{align*} Whence
\begin{align}\label{eq:Levison}
  \lim_{\lambda\to (\lambda_0)_+}\,S_{\alpha\alpha}(\lambda)=1-2 \inp{Y_0,\cdot}Y_0,
\end{align} and since the right-hand side is unitary, indeed
\eqref{eq:asymU} is proven.

We remark that \eqref{eq:Levison} agrees with Levison's theorem for
the one-body problem, see for example \cite{Ne}. In the context of
\eqref{asymRz1b} of Theorem \ref{thm5.20} with multiplicity $m=m_a>1$
and upon assuming $\kappa=n_{\mathrm{res}}=m$ (the maximality
condition) we obtain the same result as \eqref{eq:Levison}, see below. We will study
the case  $1\leq \kappa =n_{\mathrm{res}}<m$ in the next subsection,
but  let us here give a formula valid in both situations.

We will use the normalization in Remark
\ref{remark:resolv-asympt-nearNorm}, i.e.  $(c(S^*u_j),c(S^*u_k))
=\delta_{jk}$ for $j,k =1, \dots, \kappa$, where for any $u\in \vE$
(with $\vE$ given in Theorem \ref{thm:physical-modelsRell})
\begin{align*}
  c_i(S^*u)=\f{1}{2\sqrt{\pi}}\int_\bX  \varphi_i I_0u \,\d x;\quad \quad i=1, \dots, m.
\end{align*}

Considering  $\alpha_i=(a,\lambda_0,\varphi_i)$, $i=1,\dots,m$, we compute
the following substitute for \eqref{eq:Levison},
\begin{align}\label{eq:Levison2}
  \lim_{\lambda\to (\lambda_0)_+}\,S_{\alpha_i\alpha_i}(\lambda)=1-2
  \sum_{j\leq \kappa}\abs{c_i(S^*u_j)}^2 \inp{Y_0,\cdot}Y_0.
\end{align} Obviously the result \eqref{eq:asymU} for  Case III) follows from
\eqref{eq:Levison2} (since  $\kappa=m$ in that case),
cf. \eqref{normalization5.1.2}.

\subsection{An example of transmission}\label{subsec: An example of transmission}

For Case  III) treated above, but with the  maximality
condition replaced by $1\leq \kappa=n_{\mathrm{res}}<m$,  the limiting
scattering operator might differ from  \eqref{eq:Levison}.
Thanks  to  \eqref{eq:Levison2},  obviously this
happens for $\alpha_i=(a,\lambda_0,\varphi_i)$ unless $\sum_{j\leq
  \kappa}\abs{c_i(\psi_j)}^2=1$ (i.e. not smaller). The other extreme
is that  $\sum_{j\leq
  \kappa}\abs{c_i(\psi_j)}^2=0$ meaning that $ \lim_{\lambda\to
  (\lambda_0)_+}\,S_{\alpha_i\alpha_i}(\lambda)=1$, which is a unitary
operator. If on the other hand $ 0<\sum_{j\leq
  \kappa}\abs{c_i(\psi_j)}^2<1$, the isometry property of  $ \lim_{\lambda\to
  (\lambda_0)_+}\,S_{\alpha_i\alpha_i}(\lambda)$ is  not fulfilled and
transmission from the channel $\alpha_i$ will occur. We will construct such example by
redefining the basis $\varPhi:=\set{\varphi_1,\dots,\varphi_m}$. Whence we consider a general unitary
transformation $(\varphi'_1,\dots,\varphi'_m)^t=M
(\varphi_1,\dots,\varphi_m)^t$. We introduce the basis
$\varPhi':=\set{\varphi'_1,\dots,\varphi'_m}$
and notation $S_{\varPhi}$, $S_{\varPhi'}$ and similarly for adjoints to indicate the dependence of basis.

We compute
\begin{align}
  \label{eq:tran1}
  S^*_{\varPhi'}= \overline MS^*_{\varPhi}.
\end{align}

Using the notation $\alpha_i=(a,\lambda_0,\varphi_i)$ and $\alpha'_i=(a,\lambda_0,\varphi'_i)$, $i=1,\dots,m$,
we similarly
compute
\begin{align}
  \label{eq:Stran}
  \parb{S_{\alpha'_i\alpha'_j}(\lambda_0)}_{i,j\leq
    m}=\parbb{\lim_{\lambda\to
      (\lambda_0)_+}\,S_{\alpha'_i\alpha'_j}(\lambda)}_{i,j\leq
    m}=\overline M \parb{S_{\alpha_i\alpha_j}(\lambda_0)}_{i,j\leq m}\overline M ^*.
\end{align}

Let us on $\vE$ introduce the $\C^m$-valued function $c_{\varPhi}$,
slightly abusing notation,
\begin{align*}
  c_{\varPhi,i}(u)= \f{1}{2
  \sqrt{\pi}}\w{1, (S^*_{\varPhi}I_0
  u)_i}_0=\f{1}{2\sqrt{\pi}}\int_\bX  \varphi_i I_0u \ \d x;\quad
i=1,\dots, m,\,\, u\in\vE.
\end{align*}

Using \eqref{eq:tran1}  we compute
\begin{align}
  \label{eq:tra}
  c_{\varPhi'}(u)=\overline M c_{\varPhi}(u);\quad u\in\vE.
\end{align}

The normalization of the resonance
functions in Remark
\ref{remark:resolv-asympt-nearNorm}  reads with this notation
\begin{align*}
  (c_{\varPhi}(u_j),c_{\varPhi}(u_k)) =\delta_{jk}; \quad j,k =1, \cdots, \kappa.
\end{align*}

Let $e_1, \dots, e_m$  denote the canonical basis in $\C^m$,
$(e_j)_i=\delta_{ij}$. We choose the unitary matrix $M$ such that
\begin{align*}
  \overline Mc_{\varPhi}(u_j)=e_j; \quad j\leq \kappa.
\end{align*} Thanks to \eqref{eq:tra} this leads to
\begin{align}
  \label{eq:newba}
  c_{\varPhi',i}(u_j)=\delta_{ij}; \quad i\leq m,
 j\leq \kappa.
\end{align}

Using \eqref{eq:newba} we
compute as in Subsection \ref{subsec: caseIII}
\begin{align*}
  S_{\alpha'_k\alpha'_l}(\lambda_0)&=\delta_{kl}1-2
  \sum_{j\leq \kappa}c_{\varPhi',k}(u_j)\overline{c_{\varPhi',l}(u_j)}
  \,\inp{Y_0,\cdot}Y_0\\
  &=\delta_{kl}\begin{cases} 1-2\,\inp{Y_0,\cdot}Y_0,\text{ for
    }k,l\leq \kappa,\\ 1, \text{
      otherwise}.
\end{cases}
\end{align*} In view of \eqref{eq:Stran} this formula is an explicit
diagonalization of  the
operator-valued matrix
$\parb{S_{\alpha_i\alpha_j}(\lambda_0)}_{i,j\leq m}$. Each of the
appearing diagonal elements is  given by either $1-2\,\inp{Y_0,\cdot}Y_0$
or $1$ (depending on its  location). Since both of these options are
unitary we conclude that transmission from an incoming  channel associated with
$(a,\lambda_0)$, if occurring at all,  is limited to outgoing channels
in the family  associated with
$(a,\lambda_0)$. Morover indeed such  transmission does not occur  from any of the
channels $\alpha'_i=(a,\lambda_0,\varphi'_i)$, $i=1,\dots,m$. However
for  `mixtures' of these channels indeed transmission occurs. Thus
for example  considering $\alpha=(a,\lambda_0,\varphi)$ for the
mixture
$\varphi=\cos \theta \,\varphi'_1+\sin \theta
\,\varphi'_m$, $\cos \theta \sin \theta \neq 0$, we compute using  \eqref{eq:Stran}
\begin{align*}
  S_{\alpha\alpha}(\lambda_0)=1-2\cos^2\theta\,\inp{Y_0,\cdot}Y_0,
\end{align*} and  this operator is \emph{not isometric}.

\section{Threshold behaviour of total  cross-sections in atom-ion scattering}
\label{total cross-sections}

The scattering process for multi-particle Coulomb systems with initial  two-cluster data
has been studied in physics literature, both experimentally and theoretically.
In particular, in the collision of a  neutral cluster with a  charged
one (atom-ion scattering),
physical pictures suggest that if the neutral sub-system has no static dipole moment,
the total cross-sections would be finite. Its mathematical proof is subtle. In \cite{es}, V. Enss and B. Simon  put forward
several open questions and the sixth one of them is  the following:
\bigskip

\parbox[t]{15cm}{
``{\textit 6.  Atom-Ion Scattering.} An induced polarization picture suggests that Coulomb cross-sections with one neutral and one charged cluster will be finite if the neutral system has no static dipole moment. We are unable to prove this. Can one obtain explicit bounds in such a case?''}
\\[2mm]

In \cite{JKW}, the authors give an affirmative answer to this question of Enss-Simon. In fact they prove the finiteness of Coulomb total cross-sections  in atom-ion scattering  for non-threshold energies and study
the Born-Oppenheimer approximation. Here we are interested in the threshold behaviour of total cross-sections.
\\

Recall the well-known  fact  in one-body scattering  theory (see
for example \cite{Ya4})
that if a  bounded real potential $V$ on $\R^3$  decays  like
$ \vO( |x|^{-\rho})$
for some  $ \rho >2$  then the  total cross-section for the scattering process
described by the couple of operators ( $-\Delta$, $-\Delta + V(x)$)  is finite, while if $ V(x) \approx \frac{C}{|x|^2}$
as $ |x| \to \infty$  for some $C \ne 0$ then the  total cross-section is infinite.
In the  scattering theory for multi-particle Coulomb systems with initial two-cluster data, the inter-cluster
interaction between the two clusters decays  like $ \vO( |x|^{-1})$ in
the general case,
like  $ \vO( |x|^{-2})$ if one of the clusters is neutral (atom-ion scattering) and
like $ \vO( |x|^{-3})$  if  both clusters are neutral (atom-atom
scattering), see Subsection \ref{First principal example}.
 Here $x\in\R^3$
denotes the relative position between mass-centers of the two clusters.
For atom-ion scattering,
the known results for the  one-body case suggest that without an additional assumption, the total
cross-section  would  be infinite.
On the other hand  with  the assumption that the atom is in
the fundamental state the finiteness
of total cross-section  follows from \cite{JKW},  since by the
symmetry of Coulomb potentials
there is no static dipole moment for the atom in this case,
cf. Subcases  3b and 3c listed in Subsection \ref{First principal example}.  The goal of this section is to study the threshold behaviour
of the total cross-section in atom-ion scattering at the lowest threshold $\Sigma_2$.
\\

Consider the Hamiltonian of a diatomic molecule with $N$ electrons
which can be written  in the form
\begin{eqnarray} \label{Hamilton}
{ H}_{\textrm {phys}}&=&\sum _{k=1}^{2}\frac{1}{2m_{k}}\bigl(-\Delta _{x_{k}}\bigr)
\ +\ \sum _{j=3}^{N+2}\frac{1}{2}\bigl(-\Delta _{x_{j}}\bigr)\ +\
\frac{Z_{1}Z_{2}}{|x_{1}-x_{2}|}\label{diatomic-hamiltonian}\\
& &+\ \sum _{k=1}^{2}\sum _{j=3}^{N+2}\frac{q_{j}Z_{k}}{|x_{j}-x_{k}|}
\ +\ \sum\limits _{3\leq l<j\leq N+2}\frac{q_{l}q_{j}}{|x_{l}-x_{j}|},\nonumber
\end{eqnarray}
where $x_{k}\in \R ^{3}$, $k=1,2$, denote the position of the two nuclei
with mass $m_{k}$ and charge $Z_{k}>0$ and $x_{j}\in \R ^{3}$, $j=3,
\ldots ,N+2$, denote the position of $N$ electrons with mass $1$ and
charge $q_{j} \in \R$ (for the physical case charges
are equal and negative). Planck's constant $\hbar$ is taken to be $1$
in this formula.  Obviously \eqref{diatomic-hamiltonian} is a special
form of \eqref{0.1Cou}.

We are interested in scattering processes where the incoming  channel
is a two-cluster one, while the outgoing  channel is arbitrary.
Let $a=(C_1,C_2)$ be a two-cluster decomposition of $\{1,\ldots ,N+2\}$,
i.e.  a partition $(C_1,C_2)$ of the set of particle labels $\{1,\ldots ,N+2\}$,
where $j\in C_j$ for $j=1,2$. In order to make explicit calculations, we choose so called \emph{ clustered atomic
coordinates } $(x,y)\in \R^3\times \R^{3N}$ adapted to this cluster decomposition:
\begin{equation*}\label{h}
h\ =\ \left(\frac{1} {2M_1}+\frac{1}{2M_2}\right) ^{1/2}, \quad \ M_k
=m_k+|C'_k|,\quad  C'_k=C_k\setminus \{k\}, \quad
k=1, 2,
\end{equation*}
\begin{eqnarray*}
R_k&=&\frac{1}{M_k}\biggl(m_kx_k+\sum _{j\in C'_k}x_j\biggr),
k=1, 2, \nonumber\\
x&=&R_1-R_2, \label{x}\\
y_j&=&x_j-x_k,\quad j\in C'_k,\quad k=1,2,\quad \label{y}\\
l(y)&=&\frac{1}{M_1}\sum _{j\in C'_1}y_j-\frac{1}{M_2}
\sum _{j\in C'_2}y_j. \label{l(y)}
\end{eqnarray*}
In the Born-Oppenheimer approximation studied in \cite{JKW}, $h$ is regarded as a small parameter. Here
$h$ is regarded as  a constant, so we set $h=1$.
Note  that $R_k$ is the center of mass of the cluster $C_k$  for
$k=1,2$, and that $x$ is the relative position of these centers of
mass. These coordinates are well adapted to describe two-cluster
scattering of diatomic molecules.
After removal of  the molecular center of mass motion, the Hamiltonian
${ H}_{\textrm{phys}}$ may be written in this system of coordinates as
\begin{equation}\label{P-Pe}
{ H}\ =\ -\Delta _{x}+{ H}_\e(x),\ \ { H}_\e(x)\ =\
{\ H}^a+{\rm I}_{a}(x), \
\end{equation}
where the sub-Hamiltonian ${ H}^a$ is given by
\begin{equation}
{ H}^a ={ H}^{C_{1}}+{ H}^{C_{2}}, \label{P^a}
\end{equation}
with
\[
H^{C_k} =  \sum _{j\in C_{k}'}
\Bigl(-\frac{1}{2}\Delta _{y_{j}}+\frac{Z_{k}q_{j}}{|y_{j}|}\Bigr)
-\frac{1}{2m_{k}}\biggl(\sum _{j\in C_{k}'}\partial _{y_{j}}\biggr)^{2}
+\sum\limits _{\nfrac{l,j\in
C_{k}'}{l<j}}\frac{q_{l}q_{j}}{|y_{l}-y_{j}|},
\]
and the inter-cluster interaction ${ I}_{a}$ by
\begin{equation*}
{ I}_{a}(x)
=\frac{Z_{1}Z_{2}}{|x-l(y)|}+\sum\limits
_{\nfrac{k\in C_{1}'}
{j\in C_{2}'}}\frac{q_{k}q_{j}}{|y_{k}-y_{j}+x-l(y)|}
+ \sum _{j\in C_{1}'}\frac{Z_{2}q_{j}}{|y_{j}+x-l(y)|}+\sum _{j\in C_{2}'}
\frac{Z_{1}q_{j}}{|x-l(y)-y_{j}|}. \label{I_a}
\end{equation*}
$H_e(x)$ is the electronic Hamiltonian in the Born-Oppenheimer approximation.  Finally, we set
\begin{equation}
  \label{free-motion}
  { H}_a\ =\ -\Delta _{x}+{ H}^a.
\end{equation}
This operator as well as the full Hamiltonian $H$  are  considered as  self-adjoint operators on $L^2_{x,y} = L^2(\R^{3(N+1)}; {\rm d}x{\rm d}y)$.

For an arbitrary cluster decomposition $a=(C_1,\ldots ,C_k)$ of
$\{1,\ldots ,N+2\}$, i.e.   $C_1\cup \cdots \cup C_k=
\{1,\ldots ,N+2\}$ and $C_j\cap C_k=\emptyset$, for $j\neq k$,
we can, as for the case $k=2$ discussed above,  choose adapted coordinates $(x_a,y^a)$. We call
${H}^{a}$ the sub-Hamiltonian, $x_a\in \R^{3(k-1)}$ the inter-cluster
coordinates, $y^a$ the intra-cluster coordinates, and ${ I}_{a}(x_a,y^a)$
the inter-cluster interaction. By ${ D}_{x_a}$
(resp. ${ D}_{y^a}$) and by $-\Delta_{x_a}$ (resp. $-\Delta_{y^a}$),
we denote $-\i$ times the gradient
and the Laplacian in the inter-cluster (resp. intra-cluster) coordinates.
It is well known  \cite{Do, De, Sk6}
that  for this Schr\"odinger operator ${H}$, the
Dollard  wave operators
\begin{equation}\label{modified-WO}
  W^{\pm} _\alpha\ =\ s-\lim _{t\to \pm \infty}\e^{\i t{H}}
\e^{-\i t\bigl(-\Delta _{x_a}+\int _{{{\pm 1}}}^{t}I_{a}({2}s{ D}_{x_a},0)ds+
\lambda_{\alpha}\bigr)}J _{\alpha}
\end{equation}
exist for an  arbitrary  channel $\alpha =(a, \lambda_{\alpha},
\varphi_{\alpha})$, recalling  that (by definition) $a$ is an arbitrary cluster
decomposition ($\neq a_{\rm{max}}$) and
$\varphi_{\alpha}$ is a normalized  eigenfunction of ${H}^{a}$ with
eigenvalue $\lambda_\alpha$,
${H}^{a}\varphi_{\alpha}=\lambda_{\alpha}\varphi_{\alpha}$. The
operator
$J_\alpha$ denotes the identification operator,
which is defined for
any $L^2$-function $f$ of the variable $x_a$ by
\begin{equation}
  \label{identification}
(J _{\alpha}f)(x_a,y_a)
=f(x_a)\varphi_{\alpha}(y_a).
\end{equation}
Furthermore, the family of wave operators $ \{W^{\pm} _\alpha,
\text{ all }\alpha \}$ is asymptotically complete \cite{De}.
If $a=(C_{1},C_{2})$ is a two-cluster decomposition with one
neutral cluster (an atom), say $C_{1}$, i.e.
\begin{equation}\label{neutral}
\sum _{j\in C'_{1}}q_{j}\ =\ -Z_{1},
\end{equation}
then for any channel $\alpha =(a,\lambda_\alpha,\varphi_{\alpha})$
the  wave operators simply (cf. Section \ref{sec:rtominus2
  potentials}) as
%
\begin{equation}\label{WO}
W^{\pm} _\alpha \ =\ s-\lim _{t\to \pm \infty}\e^{\i t{H}}
\e^{-\i t\bigl(-\Delta _{x_a}+\lambda_\alpha\bigr)}J _{\alpha}.
\end{equation}
%
For any two  channels $\alpha$ and $\beta$ we  define
the associated scattering operator from channel $\alpha$
to channel $\beta$ by
\begin{equation}\label{S-T-matrices}
{\rm S}_{\beta \alpha}\ =\ \parb{W^+ _\beta}^{\ast}
W^- _\alpha,\ \ {\rm T}_{\beta \alpha}\ =\
{\rm S}_{\beta \alpha}-\delta _{\beta \alpha},
\end{equation}
where $\delta _{\beta \alpha}=1$ if $\alpha =\beta$ and $0$ otherwise.

Following \cite{es}, we define the total scattering cross-sections as follows.
 For $\lambda > \lambda_\alpha$, we introduce the magnitude of the
momentum associated with the kinetic energy of the relative motion
of the two clusters in the  channel $\alpha$ via
\begin{equation}
  \label{magnitude}
n_\alpha(\lambda):= (\lambda -\lambda_\alpha)^{\frac 1 2}.
\end{equation}
For $g\in C_{\c}^{\infty}(I_\alpha;\C )$,
$I_\alpha=(\lambda_\alpha,\infty ) $, and $\omega\in \S^2$,   we consider
the wave packet
\begin{equation}\label{wave-packet-h}
\R^3 \ni x \ \mapsto \
g_\omega (x)\ =  \tilde{g} (\omega\cdot x)
\end{equation}
where
\[ \tilde{g}(\nu) =
 \ \frac{1}{2\sqrt{\pi }}\int _{\R}
e^{in_\alpha(\mu ) \nu}
\frac{g(\mu )}{n_\alpha(\mu )^{1/2}}\, d\mu.
\]
The normalization is chosen such that
\[
\|g\|_{L^2(\R)} =
\|\tilde{g}\|_{L^2(\R)}.
\]
 Denoting by $\vC$ the set of all channels
we want to apply the operator ${\rm T}_{\beta \alpha}$,  for any
$ \beta \in \vC $,  to the
function
$g_\omega (x)\varphi_\alpha (y)$. Since this function does not belong to
$L^{2}(\R^{3(N+1)})$ - it  decays rapidly only in the direction defined by
$\omega$ -  we regularize it by multiplication by a function $h_{R,\omega}
\in L^{\infty}(\R^3)$, depending only  on the variable
$x-(\omega \cdot x)\omega$ transversal to the direction
$\omega$ of the
incident wave packet $g_\omega(x)$, such that  pointwisely
\begin{equation}\label{limit-pointwise}
\lim _{R\to \infty} h_{R,\omega}\ =\ 1 .
\end{equation}
For the purpose of this paper we shall specify this cut-off function
to be a Gaussian, explicitly  we take
\begin{equation}
  \label{Gaussian}
  h_{R,\omega}(x)= \e^{-(x-(\omega \cdot x)\omega)^2/R}.
\end{equation}

\noindent  {\bf Definition. } (\cite{JKW})
For $\lambda \in I_\alpha$ and $\omega \in \S^2$,  we shall say that
the total cross-section $\sigma_\alpha
(\lambda ,\omega )$ with the incoming channel $\alpha$
exists at the energy $\lambda$ with
the incident direction $\omega$, if the following limit is finite:
 \begin{equation}
 \label{def}
\sigma _{\alpha} (\lambda ,\omega)
 \ := \ \lim_{n \to \infty}\lim _{R\to \infty}
\sum_{\beta \in \vC,\,\lambda_\beta <\lambda}\| {\rm T}_{\beta \alpha}h_{R,\omega}
g_{n,\omega}\varphi_{\alpha}\|^2,
\end{equation}
where $g_{n, \omega}$ is defined as in (\ref{wave-packet-h}) with $g$ replaced by $g_n$:
\[
g_n (\mu ) = n^{-1/2} h( (\mu-\lambda)/n) \]
and $h$ is any $C_\c^\infty (\R)$-function normalized by $\int_\R |h(\mu)|^2 d\mu = 1$.
\\

 Recall that in \cite{es} and \cite{Wa3},
the total cross-section is defined as a distribution in $\mu\in I_\alpha$ by
 \begin{equation}\label{d1}
\int_{I_\alpha}\sigma _{\alpha} (\mu ,\omega)
|g(\mu)|^2\, \d\mu \ = \ \lim _{R\to \infty}
\sum_{\beta \in \vC}\| {\rm T}_{\beta \alpha}h_{R,\omega}
g_{\omega}\varphi _{\alpha}\|^2,
\end{equation}
 for all
$g\in C_{\c}^{\infty}(I_\alpha ; \C )$.
Since $|g_n(\cdot)|^2$ converges to $\beta_\lambda(\cdot)$,
the Dirac measure at $\lambda$,
as $n\to\infty$, the definitions (\ref{def}) and (\ref{d1})
coincide if the distribution defined in (\ref{d1})
 can be identified with a continuous
function  in a neighbourhood of $\lambda$. For fastly decaying pair potentials, total cross-sections can also be defined through scattering amplitudes (see \cite{ITa}).

\subsection{Finiteness of total cross-sections in atom-ion scattering}\label{subsec:finit-total-cross}

\begin{hypothesis}\label{hypo-alpha}
Let $\alpha =(a,\lambda_\alpha ,\varphi_\alpha )$ be a channel with $\lambda_\alpha\in
\sigma_{\rm d}({ H}^{a})$ and
cluster decomposition $a=(C_1,C_2)$ such that the cluster $a_1$ is neutral (an atom), that is
\begin{equation}\label{neutral-bis}
\sum _{j\in C'_{1}}q_{j}\ + Z_{1} = 0.
\end{equation}
Assume that $\lambda_\alpha \ =\ \lambda^1_{\alpha} \, + \, \lambda^2_{\alpha}$,  $\varphi_\alpha=\varphi_\alpha^1\otimes \varphi_\alpha^2$,
$H^{C_j}\varphi_\alpha^j=\lambda_\alpha^j\varphi_\alpha^j$ for $j=1,2$ and that
 $\lambda^1_{\alpha}$  is a simple
 eigenvalue of $H^{C_1}$ (i.e. the eigenvalue for  the neutral cluster
 is non-degenerate).
\end{hypothesis}

The following result of \cite{JKW} shows the finiteness   of $\sigma
_{\alpha}(\lambda, \omega)$ for $\lambda$ outside the set
\begin{equation*}
  \vT_{\rm
  p}=\vT_{\rm
  p}(H)=\vT(H)\cup\sigma_{\pp}(H)
\end{equation*}
 of
thresholds and eigenvalues of ${ H}$,
and it provides an  optical formula which is useful in many problems.

\begin{thm}[\cite{JKW}]\label{sigma-alpha-exists}
Let $\alpha =(a,\lambda_\alpha,\varphi_\alpha)$ be a
scattering channel satisfying
Hypothesis~\ref{hypo-alpha}.  We set
\begin{subequations}
\begin{equation}\label{F}
F(z, \omega) = \Bigl \langle{ I}_ae^\omega_\alpha \, , \, { R}(z)\,\parb{
{ I}_ae^\omega_\alpha }\Bigr\rangle, \qquad \Im z \ne 0,
\end{equation}
where
\begin{align}\label{eq:e15}
  e^\omega_\alpha (x,y)={\rm e}^{i n_\alpha (\lambda)\omega  \cdot x}
\varphi_\alpha (y).
\end{align}
\end{subequations}
Then, for any energy $\lambda \in  I_\alpha \setminus \vT_{\rm p}$
and any incident direction $\omega \in
\S^{2}$, the limit
\begin{subequations}
\begin{equation} \label{lim1}
F(\lambda + \i 0, \omega) = \lim_{\epsilon\to 0_+} F(\lambda + \i \epsilon, \omega)
\end{equation}
exists and defines a continuous function of $\lambda \in  I_\alpha \setminus \vT_{\rm p}$.
The total scattering cross-section
$\sigma_\alpha(\lambda , \omega )$
exists    for any energy $\lambda \in  I_\alpha \setminus \vT_{\rm p}$
and any incident direction $\omega \in
\S^{2}$ and  one has the optical formula
\begin{equation}\label{optical}
\sigma_\alpha(\lambda , \omega )  =\ \frac{1}
{ n_\alpha (\lambda)} \Im F(\lambda + \i 0, \omega).
\end{equation}
\end{subequations}
\end{thm}

Note that under Hypothesis \ref{hypo-alpha},  ${ I}_a
e^\omega_\alpha $ only belongs to  $ L_{( 1 /2)^-}^2 = L_{( 1/
  2)^-}^2(\R^{3 +3N}; {\d}x{\d}y)$. The existence of the limit
(\ref{lim1}) is non-trivial.  Its proof given in \cite{JKW} uses
the limiting absorption principle (cf. Theorem \ref{thm:powers}) and phase space analysis
through appropriate localizations in the relative kinetic
energy of the two clusters. For the same reason, the results
established in Chapter \ref{chap:resolv-asympt-near} on  threshold resolvent
asymptotics can not be applied directly, because they hold as
operators from $L_s^{2}$ to $L_{-s}^{2}$ with  at least $s>1$.
Nevertheless, as the reader will see,
various ingredients from Chapter \ref{chap:resolv-asympt-near}  can be applied to study  the threshold behaviour of $\sigma_\alpha(\lambda , \omega)$.
In this subsection we  consider for convenience only the  simplest  case, where the  channel satisfies
the following condition. (In Subsection
\ref{subsec:total-cross-sectionsMult} we consider a more general setup.)

\subsection{Total cross-sections at $\Sigma_2$,  non-multiple
  two-cluster  case}\label{subsec:total-cross-sectionNon}

\begin{hypothesis}\label{hypo-beta} Let $\alpha = (a, \lambda_\alpha, \varphi_\alpha)$ be a scattering channel with $\lambda_\alpha=\lambda_0 = \Sigma_2$, the lowest threshold  of $H$.
Assume that $\Sigma_2$ is a  non-multiple two-cluster threshold  and $\Sigma_2 \not\in \sigma(H')$, where $H' = \Pi_\alpha'H\Pi_\alpha'$ and
 \[
 \Pi_\alpha' u = u - \w{\varphi_\alpha, u}_y\otimes \varphi_\alpha, \quad u \in L^2(\R^{3+3N}_{x,y}).
 \]
\end{hypothesis}

Under Hypothesis \ref{hypo-beta}, we can apply the Grushin reduction (cf.
\cite{Wa2}, extended in  Chapter \ref{chap:resolv-asympt-near})
to study the behaviour of $\sigma_\alpha(\lambda, \omega)$ for
$\lambda $ near $ \lambda_0= \Sigma_2$. The more  general situations
where $ \lambda_\alpha$ possibly is a multiple two-cluster threshold and/or $\lambda_\alpha \in \sigma(H')$  can also be studied by the
methods developed in the present work (the multiple case will indeed
be studied in Subsection
\ref{subsec:total-cross-sectionsMult}).  Under the Hypothesis \ref{hypo-beta},  $R'(z)= (H'-z)^{-1}\Pi_\alpha'$ is holomorphic for $z$ near $\lambda_0$ and one has
\begin{equation}
R(z) = E(z) - E_+(z) E_{\mathcal H}(z)^{-1} E_-(z)
\end{equation}
where
\begin{eqnarray*}
E(z) &= & R'(z) \\
E_+(z) &=& S - R'(z) H S \\
E_-(z) &=& S^* - S^* HR'(z) \\
E_{\mathcal H}(z) &= & S^* (z- H + H R'(z) H)S
\end{eqnarray*}
and $S : L^2(\R^{3}_x) \to L^2(\R^{3+3N}_{x,y})$, $(x,y) \in \R^3 \times \R^{3N}$, is defined by
\[
Sf(x,y) = f(x)\varphi_\alpha(y), \quad f \in L^2_x.
\]
Under Hypotheses \ref{hypo-alpha} and \ref{hypo-beta},  it follows
from Theorem \ref{sigma-alpha-exists} that for $\lambda >\lambda_0$
close to  $\lambda_0$
\begin{align}
  \begin{split}
  \sigma_\alpha(\lambda , \omega )  &= \frac{1}
{ n_\alpha (\lambda)} \Im \w{ { I}_ae^\omega_\alpha, \,R(\lambda + \i 0) ({ I}_ae^\omega_\alpha) } \\
 &=  -  \frac{1}
{ n_\alpha (\lambda)}  \Im \inp{ f(\lambda, \omega) ,E_{\mathcal H}(\lambda + \i 0)^{-1}f(\lambda, \omega) }_x;\label{cs-threshold}
  \end{split}\\
 f(\lambda, \omega) &=  E_-(\lambda)({ I}_ae^\omega_\alpha).
\end{align}
 Here we used the fact that $R'(\lambda + \i 0) = R'(\lambda)$ is self-adjoint for $\lambda$ near $\lambda_0$.
 One can calculate
 \begin{subequations}
   \begin{align}\label{eq:16i}
     f(\lambda, \omega) &= \inp{ \varphi_\alpha ,{ I}_ae^\omega_\alpha}_y -
 \inp{ \varphi_\alpha, \,I_aR'(\lambda)({ I}_ae^\omega_\alpha)}_y,\\\label{eq:16ii}
\inp{ \varphi_\alpha,\,{ I}_ae^\omega_\alpha}_y  &=
\vO({|x|^{-4}}),\\
\inp{\varphi_\alpha ,\,I_a R'(\lambda)({I}_ae^\omega_\alpha) }_y &\in L_{( 5/ 2)^-}^2(\R^{3}_x),\label{eq:16iii}
   \end{align}
 \end{subequations}
 uniformly in $\lambda$ and $\omega$. For \eqref{eq:16ii} we refer to    \cite[ Lemma A.1]{JKW}. Hence
\begin{align*}
 f (\lambda, \omega) \in L_{(5/ 2)^-}^2=
L_{( 5/ 2)^-}^2(\R^{3}_x)\text {  uniformly in }\lambda \text { and }\omega.
 \end{align*}
Similarly,
\[E_{\mathcal H}(z)  =   z- \lambda_0 - (-\beta_x + U(z)),
\]
where
\[
U(z) =  \inp{\varphi_\alpha, \, I_a\varphi_\alpha }_y - \inp{\varphi_\alpha,\,I_a R'(z) I_a
  (\varphi_\alpha\otimes \cdot) }_y
\]
is an operator-valued holomorphic function for $z$ near $\lambda_0$. One  has
\[
\w{\varphi_\alpha,\,I_a\varphi_\alpha }_y = \vO({|x|^{-4}}), \quad  \inp{ \varphi_\alpha,\,I_a
  R'(z) I_a (\varphi_\alpha\otimes \cdot)}_y \in \mathcal
L(H^1_s,H^{-1}_{s+4});\quad s\in\R.
\]
Set
\begin{equation}
U = U(\lambda_0)  =  \inp{\varphi_\alpha, I_a\varphi_\alpha}_y - \inp{ \varphi_\alpha,\,I_a R'(\lambda_0) I_a (\varphi_\alpha\otimes \cdot)}_y.
\end{equation}
Comparing $f (\lambda_0, \omega)$ with $U$, one observes
\begin{equation}\label{eq:16compU}
f (\lambda_0, \omega) =E_-(\lambda_0)(I_a\varphi_\alpha)= U 1,
\end{equation}
that is, $U$ applied to the constant function  $1\in L_{( -3/
  2)^-}^2(\R^{3}_x)$.

Recall from Section \ref{sec:resolv-asympt-nearLOWEST} that the asymptotic expansion of the resolvent $(-\Delta_x - z')^{-1}$, $z'= z-\lambda_0$ with $\Im z \neq 0$,  in weighted $L^2$-spaces
is given by
\[
(-\Delta_x -z')^{-1} = G_0 + \sqrt{z'} G_1 + \cdots
\]

\begin{thm}[Regular case]\label{cs-case0} Let Hypotheses \ref{hypo-alpha} and \ref{hypo-beta} be satisfied.
 Assume that $\lambda_0=\Sigma_2$ is a regular point of $H$ (see
 Subsection \ref{sec:case-lambda_0=sub}). Then
 \begin{subequations}
  \begin{equation}\label{asymp-cs-case0}
\sigma_\alpha(\lambda, \omega) = 4 \pi |s_{\textrm lnt}|^2 + o(1)
\quad\text{as} \quad  \lambda \to (\lambda_0)_+,
\end{equation}
where
\begin{equation}\label{s-length}
s_{\rm {lnt}} = \frac{1}{4\pi} \int_{\R^3_x} (1+
UG_0)^{-1}E_-(\lambda_0)(I_a\varphi_\alpha) \ \d x.
\end{equation}
 \end{subequations}
\end{thm}
\pf Since  $\lambda_0$ is a regular point of $H$, $(1+ G_0U)^{-1}$ exists and is bounded from $L_s^{2}$ to $L^2_{-s}$ for any $s >1$.
Let $z' = z-\lambda_0$ and $\Im z\neq 0$.  One has
\begin{eqnarray*}
 \lefteqn{E_{\mathcal H}(\lambda_0 +z')^{-1}}\\
 & = &- \parb{1+ (-\Delta_x -z')^{-1}U( z')}^{-1} (-\Delta_x -z')^{-1} \\
&= & -(1+G_0U)^{-1}G_0 + \sqrt{z'}\left\{ (1+G_0U)^{-1}G_1 \parb{U(1+ G_0 U)^{-1}G_0 -1}\right\} + o(\sqrt{z'}) \\
&= &  -(1+G_0U)^{-1}G_0 - \sqrt{z'} (1+G_0U)^{-1}G_1 (1+ UG_0)^{-1} + o(\sqrt{z'})
\end{eqnarray*}
in $L_s^{2} \to L^2_{-s}$, $s>  3/ 2$.
Making use of (\ref{cs-threshold}), we obtain for $\lambda' = \lambda-\lambda_0$
\begin{eqnarray}
\lefteqn{\sigma_\alpha(\lambda , \omega ) } \nonumber \\
 &= &   \frac{1}{ \sqrt{\lambda'}}  \Im \inp[\Big]{  f(\lambda, \omega),\,\left( (1+G_0U)^{-1}G_0 + \sqrt{\lambda'} (1+G_0U)^{-1}G_1 (1+ UG_0)^{-1} + o(\sqrt{\lambda'})\right) f(\lambda, \omega) }_x  \nonumber \\
 &= &
  \Im \inp{   f(\lambda_0, \omega),\,(1+G_0U)^{-1}G_1 (1+ UG_0)^{-1}  f(\lambda_0, \omega)}_x  + o(1)\nonumber
\end{eqnarray}
as $\lambda' \to 0_+$. Since   $G_1$ has the  constant integral kernel $\frac{\i}{4\pi}$,  the leading term in the above equation is equal to
\[
\frac{1}{4\pi}\abs[\Big]{\int_{\R^3_x} (1+
UG_0)^{-1}E_-(\lambda_0)(I_a\varphi_\alpha) \ {\d} x}^2,
\] proving \eqref{asymp-cs-case0} and \eqref{s-length}. \hfill $\Box$
\\

Note that since
\[
 E_-(\lambda_0)(I_a\varphi_\alpha) =f(\lambda_0, \omega) \in L_{( 5/ 2)^-}^2,
\]
indeed the constant $s_{\textrm {lnt}}$ is well defined. In addition it can be written as
\[
s_{\textrm {lnt}}= \frac{1}{4\pi} \int_{\R^3_x} (1+ UG_0)^{-1}U 1 \
\d x.
\]
Recall that in the scattering theory for the pair $(-\Delta, -\Delta +
V(x))$, the scattering length is related to the low-energy limit of
the total cross-section. In dimension three it is equal to
\[
 \frac{1}{4\pi} \int_{\R^3_x} (1+ VG_0)^{-1}V(x)\ {\d}x,
\]
if $V(x)$ decreases sufficiently rapidly  (cf.  \cite[Theorem 5.1]{JK}). Analogously to one-body
scattering, $s_{\textrm {lnt}}$ given in (\ref{s-length}) can be interpreted
as a \textit{scattering length} in atom-ion scattering.
\\

\begin{thm}[Exceptional
point of $1$st kind]\label{cs-case1} Let Hypotheses \ref{hypo-alpha} and \ref{hypo-beta} be satisfied.
 Assume that $\lambda_0=\Sigma_2$ is a resonance but not an eigenvalue  of $H$. Then
\begin{equation} \label{asymp-cs-case1}
\sigma_\alpha(\lambda, \omega) =\tfrac{1}{\lambda-\lambda_0} \parb{4 \pi +
o(1)} \quad\text{as} \quad \lambda \to (\lambda_0)_+.
\end{equation}
\end{thm}
\pf
In general  $\lambda_0$ is a resonance  (an eigenvalue, resp.) of $H$
if and only if $0$ is a resonance (an eigenvalue, resp.) of $-\Delta_x
+ U$, cf. Lemma \ref{prop5.1}.
 Since we assume that  $\lambda_0$ is exceptional
point of $1$st kind  the operator $E_{\mathcal H}(\lambda_0 +z')^{-1}$ admits the  asymptotics
\begin{eqnarray*}
 \lefteqn{E_{\mathcal H}(\lambda_0 +z')^{-1}}\\
 & = &- \parb{ 1+ (-\Delta_x -z')^{-1}U( z')}^{-1} (-\Delta_x -z')^{-1} \\
&= & -\tfrac{\i}{\sqrt{z'}} \parb{ \w{\psi,\, \cdot}_x \psi + o(1)}, \quad z' \to 0,
\end{eqnarray*}
 in $L^2_s \to L^2_{-s}$, $s> 1$, where $\psi \in L_{(- 1/
  2)^-}^2$ is a resonance  state of $-\Delta + U$ satisfying
\begin{align}\label{eq:normreseen}
  \tfrac{1}{2\sqrt \pi} \w{U\psi, 1} =1.
\end{align}
See the proof of \cite[Theorem 3.9]{Wa2} or that of
Proposition \ref{prop5.8}.
We deduce from (\ref{cs-threshold}) that for $\lambda' = \lambda-\lambda_0$,
\[
\sigma_\alpha(\lambda, \omega )
 =   \tfrac{1}{\lambda'}  ( |\w{\psi,  f(\lambda_0, \omega)}|^2 +
 o(1))\,\text{ as }\,\lambda' \to 0_+.
\]
 Noting  that
\begin{eqnarray*}
\w{\psi,  f(\lambda_0, \omega)}& = & \w{\psi, E_-(\lambda_0) (I_a\varphi_\alpha)}  \\
 &=&   \w{U \psi, 1}_x = 2 \sqrt \pi,
\end{eqnarray*}
  (\ref{asymp-cs-case1}) follows.  \hfill $\Box$

It is a remarkable phenomenon that in presence of threshold resonance,  total cross-sections display a universal behavior near the threshold, independent of concrete form of potentials.
If $\lambda_0=\Sigma_2$ is an eigenvalue of $H$, the method used here
does not allow us to give the leading term of $\sigma_\alpha(\lambda,
\omega ) $, due to the lack of decay. However with some  weak decay of the
$L^2$-eigenfunctions at $\lambda_0$ one can obtain similar results as
above. We shall give  elaboration in the next subsection in the more
general setting of $\lambda_0=\Sigma_2$ possibly being  a multiple two-cluster threshold.

\subsection{Total cross-sections at $\Sigma_2$,  multiple
  two-cluster case}\label{subsec:total-cross-sectionsMult}
 In this subsection we consider the atom-ion model with  $\alpha
 =(a,\lambda_\alpha,\varphi_\alpha)$ being a
scattering channel satisfying
Hypothesis~\ref{hypo-alpha}. Rather than Hypothesis \ref{hypo-beta} we
impose the following condition using notation from Sections
\ref{sec:CoulRellich} and \ref{sec:resolv-asympt-physics models
  near}. Note that the atom-ion model is a specific example  covered
by these sections. Although  we can not use the results of Section \ref{sec:resolv-asympt-physics models
  near} directly,  we can use their proofs, cf. Subsection
\ref{subsec:total-cross-sectionNon}. Again we consider only the case
$\lambda_0 = \Sigma_2$ being  a two-cluster  threshold.
We shall impose the condition
\begin{align}
    \label{eq:2plus33}
    \widetilde{\vA}=\vA^{\textrm{fd}}_3,
  \end{align}
 and in addition \eqref{eq:dec20}, that is
\begin{align}
    \label{eq:dec200}
    \ran \Pi_H\subset L^2_t\text{  for some }t> 3/2,
  \end{align} where $\Pi_H$ is the orthogonal projection onto $\ker
  (H-\lambda_0)$.

\begin{hypothesis}\label{hypo-gamma}  Let $\alpha = (a, \lambda_\alpha,
  \varphi_\alpha)$ be a  channel with $\lambda_\alpha=\lambda_0 =
  \Sigma_2$, the lowest threshold  of $H$. Assume
 that   $\Sigma_2$ is
  two-cluster, \eqref{eq:dirCo1}, \eqref{eq:dirCo2}, \eqref{eq:2plus33}
  and \eqref{eq:dec200} (in particular $a\in\vA^{\textrm{fd}}_3$).
\end{hypothesis}

We remark that Hypothesis \ref{hypo-gamma} implies  Hypothesis
\ref{hypo-alpha} for any channel $\alpha$,  and that it is more
general (i.e. weaker)  than the combination of  Hypotheses \ref{hypo-alpha} and
  \ref{hypo-beta} used in the previous subsection.  The content of
  \eqref{eq:2plus33} is that for the
  two-cluster threshold  $\Sigma_2$ it holds for any $a=(C_{1},C_{2})\in
  \widetilde\vA$ (possibly not unique) that one of the
  clusters is neutral, say  \eqref{neutral} is fulfilled for $C_{1}$.

\begin{thm}[Regular or
   exceptional of $2$nd kind case] \label{cs-caseReg+3}  Let  Hypothesis
  \ref{hypo-gamma} be satisfied for
  $\alpha=(a,\lambda_\alpha,\varphi_\alpha)$,
  and assume that $\lambda_\alpha=\lambda_0=\Sigma_2$ is either  regular  or
   exceptional of $2$nd kind. Let $
 H_\sigma:=H-\sigma\Pi_H$ with  $\sigma> 0$ small. Then, using the
 notation \eqref{eq:e15} and that of
 Theorem  \ref{thm:resolv-asympt-phys2nd}, we let
\begin{subequations}
 \begin{align}\label{eq:fcross}
   f_\sigma(\lambda, \omega) &=  \parb{S^*-
                        S_{I_\sigma}^*R_\sigma'(\lambda)} ({
                        I}_{a,\sigma} e^\omega_\alpha);\quad \lambda\geq \lambda_0.
 \end{align} Here
 \begin{align*}
   S_{I_\sigma}^*(\cdot)=\sum_{\beta=(b,E_\beta,\varphi_\beta):\,b\in
   \widetilde{\vA}}^{\oplus}\inp{ \varphi_\beta| I_{b,\sigma}(\cdot)}_{y_b},\quad I_{b,\sigma}=I_b-\sigma\Pi_H,
 \end{align*}
   and we   recall  that $R_\sigma'(\lambda)=(\Pi'H_\sigma\Pi'-\lambda)^{-1}\Pi'$. Let similarly $U_\sigma(\lambda-\lambda_0)$ be
 defined by writing the effective Grushin Hamiltonian for  $
 H_\sigma$ as $-E_{\vH,\sigma}(\lambda_0 +z) = P_0+U_\sigma(z)-z$, $z =\lambda-\lambda_0\geq
 0$,  along the lines of
 \eqref{P0},  \eqref{U}  and \eqref{Ulambda}.  Then
  \begin{equation}\label{asymp-cs-case09}
\sigma_\alpha(\lambda, \omega) = 4 \pi |s_{\textrm lnt}|^2 + o(1)\quad\text{as} \quad   \lambda \to (\lambda_0)_+,
\end{equation}
where
\begin{equation}\label{s-length9}
s_{\rm {lnt}} = \frac{1}{4\pi} \int_{\R^3_x} (1+
U_\sigma G_0)^{-1}f_\sigma(\lambda_0, \omega)\ \d x\,\in
\C^m;\quad U_\sigma=U_\sigma (0),\, m={\#
  \widetilde{\vA}}.
\end{equation}
 \end{subequations}
\end{thm}
\begin{proof}
  We can mimic the proof of Theorem \ref{cs-case0}, combining Theorem
  \ref{sigma-alpha-exists} and \eqref{eq:sigmeFormR},  checking first the
  analogue bound
\begin{align*}
 f _\sigma(\lambda, \omega) \in \vH_{( 3/ 2)^+}=
\vH_{(3/ 2)^+}(\R^{3}_x).
 \end{align*}  The latter assertion follows readily from
 \eqref{eq:dec200} (yielding in fact  that $f _\sigma(\lambda, \omega) \in
 \vH_s$ provided   $s\leq t$ and  $s<5 /2$).
\end{proof}

\begin{thm}[Exceptional
point of $1$st or $3$rd kind] \label{cs-caseReg+40} Let Hypothesis
  \ref{hypo-gamma} be satisfied for
  $\alpha=(a,\lambda_\alpha,\varphi_\alpha)$. Assume that
  $\lambda_\alpha=\lambda_0=\Sigma_2$ is exceptional of $1$st or $3$rd kind. Let $
 H_\sigma=H-\sigma\Pi_H$ with  $\sigma> 0$ small and
 $\set{u_1, \dots, u_\kappa}\subset H^1_{(-1/2)^-}$ be  a basis of   resonance states  of
$H$, more precisely it is taken as a basis of   resonance states  of
$H_\sigma$ fulfilling the normalization
\begin{subequations}
\begin{align}
  \label{normalization5.1.2333}
(c(T^*u_i), c(T^*u_j)) =\beta_{ij},  \quad i, j =1, \cdots, \kappa,
\end{align} cf. \eqref{eq:normreseen} and  Proposition \ref{prop5.7b}. By definition $T$ is given
by \eqref{eq:commIdent} and
\begin{align}\label{eq:15c}   c(v) = \tfrac{1}{2 \sqrt{\pi}}\parb{\w{1, \Sigma_{k\leq
  m}U_{\sigma,{1k}}v_k}_x, \cdots, \w{1, \Sigma_{k\leq
  m}U_{\sigma,{mk}}v_k}_x}\,\in \C^m;\,m={\#
  \widetilde{\vA}}.
\end{align}
\end{subequations}
\begin{subequations}
Here the elements of $
   \widetilde{\vA}$ are labelled   by the numbers $1,\dots,m$.

 Then, with
   the convention that  $1$ labels  $a$,
\begin{align} \label{asymp-cs-case10}
  \begin{split}
 \sigma_\alpha(\lambda, \omega) &=\tfrac{1}{\lambda-\lambda_0} \parbb{\sum_{j=1}^\kappa \,\abs{\w{T^*u_j, f_\sigma(\lambda_0, \omega)}}^2   +
o(1)}\\&=\tfrac{4
\pi}{\lambda-\lambda_0} \parbb{\sum_{j=1}^\kappa \,\abs{c_1(T^*u_j)}^2   +
o(1)} \quad\text{as} \quad \lambda \to (\lambda_0)_+.
  \end{split}
\end{align}

In particular
\begin{align} \label{asymp-cs-case1022}
 \sigma_\alpha(\lambda, \omega) \leq\tfrac{4
\pi}{\lambda-\lambda_0} \parb{1  +
o(1)} \quad\text{as} \quad \lambda \to (\lambda_0)_+,
  \end{align}
while  in the maximally exceptionally case where  $\kappa=m$,
\begin{align} \label{asymp-cs-case102}
 \sigma_\alpha(\lambda, \omega) =\tfrac{4
\pi}{\lambda-\lambda_0} \parb{1  +
o(1)} \quad\text{as} \quad \lambda \to (\lambda_0)_+.
  \end{align}
\end{subequations}
\end{thm}
\begin{proof} We obtain \eqref{asymp-cs-case10} by   mimicking  the proof of Theorem \ref{cs-case1}. The analogue of
  \eqref{eq:16compU} reads
  \begin{align*}
     f_\sigma(\lambda_0, \omega)&= E_- (H_\sigma-\lambda_0)
    \varphi_\alpha=E_- (H_\sigma-\lambda_0) S_1[1]\\&= -E_\vH [\Sigma^{\oplus}_k\,\beta_{1k}]=(P_0+U_\sigma)[\Sigma^{\oplus}_k\,\beta_{1k}]=U_\sigma[\Sigma^{\oplus}_k\,\beta_{1k}].
  \end{align*}

The assertions
\eqref{asymp-cs-case1022} and \eqref{asymp-cs-case102} follow from
\eqref{asymp-cs-case10} and  the
normalization condition \eqref{normalization5.1.2333}.
\end{proof}

\begin{remark*} The quantity $\sum_{j=1}^\kappa \,\abs{c_1(T^*u_j)}^2$
  of \eqref{asymp-cs-case10} appears also in \eqref{eq:Levison2} and
  in Subsection \ref{subsec: An example of transmission}
  although for a different setup. This is not a coincidence, an
  overall conclusion may be viewed as the same: It may be interpreted as a
  probability for resonance
  induced  scattering phenomena at a threshold. This impact appears
  with probability one in the maximally exceptionally case for both models.

\end{remark*}





\begin{thebibliography}{99}

\bibitem[AW]{AW} M. Aafarani, X.P. Wang, \emph{Gevrey estimates of the resolvent and sub-exponential time-decay for the heat and Schrödinger semigroups. II.}
 J. Differential Equations \text{316} (2022), 387-424.

 \bibitem[AIIS]{AIIS} T. Adachi, K. Itakura, K. Ito, E. Skibsted,
\emph{New methods in spectral theory of $N$-body Schr{\"o}dinger
  operators},  Rev. Math. Phys. \textbf{33}  (2021), 48 pp.


\bibitem[Ag]{Ag} S. Agmon, \emph{A representation theorem for solutions of the Helmholtz equation and resolvent
estimates for the Laplacian}, Academic Press, Inc., Boston, MA,
 xii+694 pp.,  1990, 39--76.

\bibitem[ABG]{ABG} W.O. Amrein, A.M. Boutet de Monvel, V. Georgescu, \emph{Lower bounds for zero
    energy eigenfunctions of Schr\"odinger operators},
Helv. Phys. Acta \textbf{57}  no. 3 (1984), 301--306.

\bibitem[AHS]{AHS} A. Agmon, I. Herbst and E. Skibsted,
  \emph{Perturbation of embedded eigenvalues in the generalized
    $N$-body problem},
 Comm. Math. Phys. \textbf{122}  (1989), 411--438.

\bibitem[Ba]{Ba} E. Balslev, \emph{Analyticity properties of
    eigenfunctions and scattering matrix}, Comm. Math. Phys.
\textbf{114}  no. 4
(1988),
599--612.

\bibitem[Bo]{Bo} A. Bommier, \emph{Properties of the
2
-cluster
2
-cluster scattering matrix for long-range
N
-body problems},  Ann. Inst. H. Poincar{\'e}e Phys. Th{\'e}or.
\textbf{59} no. 3
(1993),
 237--267.

\bibitem[Ca]{Ca} G. Carron, \emph{Le saut en zéro de la fonction de
    décalage spectrale}, J. Funct. Analysis, 212(2004), 222-260.



\bibitem[CT]{CT}  J.M. Combes, A. Tip, \emph{Properties of the
    scattering amplitude for electron-atom collisions},
  Ann. Inst. H. Poincarr{\'e}  Phys. Th{\'e}or. \textbf{40} (1984), 117--139.

\bibitem[De1]{De1}
J. Derezi\'nski, \emph{Existence and analyticity of many-body
  scattering amplitudes at low energies}, J. Math. Phys. \textbf{28}   (1987),
  1080--1088.

\bibitem[De2]{De2}
J. Derezi\'nski, \emph{Threshold singularities of
  two-cluster--two-cluster
  scattering amplitudes for dilation analytic potentials}, J. Math. Phys. \textbf{29}   (1988),
  1171--1173.

\bibitem[De3]{De}
J. Derezi\'nski, \emph{Asymptotic completeness for $N$-particle long-range quantum
    systems}, Ann. of Math.  \textbf{38}   (1993),
  427--476.

\bibitem[DG]{DG}
J. Derezi{\'n}ski  and C. G{\'e}rard, \emph{Scattering theory of
  classical and quantum {$N$}-particle systems}, Texts and Monographs in
  Physics,  Berlin, Springer 1997.


\bibitem[DS1]{DS1} J. Derezi{\'n}ski, E. Skibsted, \emph{Quantum
  scattering at low
energies}, J. Funct. Anal.,  {\textbf 257} (2009), 1828--1920.

\bibitem[DS2]{DS2} J. Derezi{\'n}ski, E. Skibsted,
\emph{Scattering at zero energy for attractive homogeneous
  potentials}, Ann. H. Poincar{\'e}  {\textbf 10} no. 3 (2009),
549--571.

\bibitem[Do]{Do} J. Dollard, \emph{Asymptotic convergence and Coulomb interaction}, J. Math. Phys. {\textbf 5} (1964), 729--738.

\bibitem[ES]{es} V.Enss, B.Simon, {\em Finite Total Cross-Section in
Nonrelativistic Quantum Mechanics.} Comm. Math. Phys. \textbf{76}
(1980), 177--209.

\bibitem[Fr]{Fr} R. Frank \emph{A note on low energy scattering for
    homogeneous long-range potentials},  Ann. H. Poincar{\'e}  {\textbf
    10} no.3  (2009), 573--575.

\bibitem[FH]{FH} R. Froese and  I. Herbst,
  \emph{Exponential bounds and absence of positive eigenvalues for $N$-body {S}chr\"odinger operators},
 Comm. Math. Phys. \textbf{87}  (1982), 429--447.



\bibitem[FS]{FS} S. Fournais, E. Skibsted,
\emph{Zero energy asymptotics of the resolvent for a
  class of slowly decaying potentials}, Math. Z. {\textbf 248}
 (2004), 593--633.

\bibitem[Gr]{Gr} G.M. Graf,
  \emph{Asymptotic completeness for $N$-body short-range quantum
    systems: a new proof},
 Comm. Math. Phys. \textbf{132}  (1990), 73--101.



\bibitem[GIS]{GIS} C. G\'erard, H. Isozaki, E. Skibsted,
  \emph{$N$-body resolvent estimates},
 J. Math. Soc. Japan \textbf{48}  no. 1 (1996), 135--160.

\bibitem[H{\"o}1]{Ho1}
L. H{\"o}rmander, \emph{The analysis of linear
partial differential operators. {I}},
Berlin,  Springer 1990.

\bibitem[H{\"o}2]{Ho2}
L. H{\"o}rmander, \emph{The analysis of linear
partial differential operators. {II}-{IV}},
Berlin,  Springer 1983-85.

\bibitem[Is1]{Is}
H. Isozaki, \emph{Differentiability of generalized {F}ourier transforms
  associated with {S}chr\"odinger operators}, J. Math. Kyoto
Univ. \textbf{25}
   no. 4 (1985), 789--806.



\bibitem[Is2]{Is2}
 H. Isozaki, \emph{Structures of the S-matrices for  three-body
   Schr{\"o}dinger operators},   Commun. Math. Phys. \textbf{146}
 (1992), 241--258.



\bibitem[Is3]{Is3}
 H. Isozaki, \emph{Asymptotic properties of solutions to $3$-particle
   Schr{\"o}dinger equations},   Commun. Math. Phys. \textbf{222}  (2001), 371--413.

\bibitem[Is4]{Is4}
 H. Isozaki, \emph{A generalization of the radiation condition of
 Sommerfeld for $N$-body Schr\"odinger operators}, Duke
 Math. J. \textbf{74} no. 2 (1994), 557--584.

\bibitem[IS1]{IS1} K. Ito, E. Skibsted, \emph{Absence of positive
    eigenvalues for hard-core $N$-body systems}, Ann. Inst. Henri Poincar{\'e} {\textbf 15}
  (2014), 2379--2408.

\bibitem[IS2]{IS2} K. Ito, E. Skibsted, \emph{Rellich's theorem and
    $N$-body Schr\"odinger operators}, Reviews Math. Phys. \textbf
  {28}  no. 5 (2016), 12 pp.

\bibitem[IT]{ITa} H. Ito, H. Tamura, \emph{Semi-classical asymptotics for total scattering cross sections of N-body quantum systems}, Osaka J. Math. 32(3) (1995), 753--781.

\bibitem[JKW]{JKW} T. Jecko, M. Klein, X.P. Wang, \emph{Existence and Born-Oppenheimer asymptotics of the total scattering
cross-section in ion-atom collisions}, Long time behaviour of classical and quantum systems (Bologna, 1999),
220--237, Ser. Concr. Appl. Math., \textbf{1}, World Sci. Publ.,
River Edge, NJ, 2001.

\bibitem[Je]{Je}
A. Jensen, \emph{Propagation estimates for {S}chr\"odinger-type operators},
  Trans. Amer. Math. Soc. \textbf{291} no.~1  (1985), 129--144.

\bibitem[Ji]{Ji}  X. Jia, \emph{Some threshold spectral problems of
    Schr{\"o}dinger operators}, thesis, Université de Nantes, 2009.

\bibitem[JK]{JK} A. Jensen and T. Kato, \emph{Spectral
    properties of {S}chr\"odinger operators and time-decay of the wave
  functions}, Duke Math. J. \textbf{46} no.~3  (1979), 583--611.


\bibitem[JMP]{JMP}
A. Jensen, {\'E}. Mourre, P. Perry, \emph{Multiple
 commutator estimates and resolvent smoothness in quantum scattering theory}, Ann. Henri Poincar\'e  \textbf{41} no. 2 (1984), 207--225.

\bibitem[JN1]{JN} A. Jensen, G. Nenciu, \emph{ A unified approach to
    resolvent expansions at thresholds}, Rev. Math. Phys., \emph{13} (2001),
  717-754.

\bibitem[JN2]{JN2} A. Jensen, G. Nenciu, \emph{ Erratum: ``A unified
    approach to resolvent expansions at thresholds''
    [Rev. Math. Phys. 13 (2001),  717--754]}, Rev. Math. Phys., 2014, 675--677.

\bibitem[Ka]{Ka}
T. Kato, \emph{Perturbation theory for linear operators}, Springer-Verlag,
  Berlin, 1995, Reprint of the 1980 edition.


\bibitem[La]{La} L.J. Landau, \emph{Bessel functions: monotoicity and
    bounds}, J. London Math. Soc. (2) \textbf{61} (2000), 197--215.

\bibitem[Mo]{Mo}
{\'E}. Mourre, \emph{Absence of singular continuous spectrum
for certain
  selfadjoint operators}, Comm. Math. Phys. \textbf{78} no. 3 (1980/81),
  391--408.

\bibitem[Na]{Na} S. Nakamura, \emph{Low  energy asymptotics for
    {S}chr\"odinger operators with slowly decreasing potentials}, Comm. Math. Phys. \textbf{161}
  (1994), 63--76.


 \bibitem[Ne]{Ne} R. G.  Newton,  \emph{Noncentral potentials: the
 generalized Levinson theorem and the structure of the spectrum},  J. Math.
Phys., \textbf{18} (1977),  1582-1588.


\bibitem[Pe]{Pe} P. Perry,
  \emph{Exponential bounds and semifiniteness of point spectrum for $N$-body {S}chr\"odinger operators},
 Comm. Math. Phys. \textbf{92}  (1984), 481--483.

\bibitem[PSS]{PSS}
P.~Perry, I.M. Sigal, and B.~Simon, \emph{Spectral analysis of ${N}$-body
  {S}chr\"odinger operators}, Ann. of Math.  \textbf{114} no. 3  (1981),
  519--567.


\bibitem[RS]{RS}
M.~Reed and B.~Simon, \emph{Methods of modern mathematical physics {I}-{IV}},
  Academic Press, 1972-78.

\bibitem[Sk1]{Sk1}
 E. Skibsted, \emph{Propagation estimates for $N$-body Schr{\"o}dinger operators}, Comm. Math. Phys. \textbf{142} (1991),
  67--98.

\bibitem[Sk2]{Sk2}
 E. Skibsted, \emph{Smoothness of $N$-body scattering amplitudes}, Rev. Math. Phys.
 \textbf{4} no. 4
 (1992), 619--658.

\bibitem[Sk3]{Sk3}
 E. Skibsted, \emph{Spectral analysis of $N$-body systems coupled to
   a bosonic field}, Rev. Math. Phys.  \textbf{10}
 (1998), 989--1026.


\bibitem[Sk4]{Sk4} E. Skibsted, \emph{Sommerfeld radiation condition at threshold},
  Comm. PDE \textbf {38} (2013),
  1601--1625.

\bibitem[Sk5]{Sk5} E. Skibsted, \emph{Renormalized two-body low-energy
    scattering}, Journal d'Analyse Math\'ematique \textbf {122}
  (2014), 25--68.

\bibitem[Sk6]{Sk6}
 E. Skibsted, \emph{Stationary scattering theory: the $N$-body
   long-range case},  Comm. Math. Phys. (2023), https://doi.org/10.1007/s00220-023-04689-7 

\bibitem[Sk7]{Sk7}
 E. Skibsted, \emph{Green functions and completeness; the $3$-body problem
  revisited}, preprint   30 May 2022, http://arxiv.org/abs/2205.15028v1.


\bibitem[SW]{SW} E. Skibsted, X. P. Wang, \emph{Two-body threshold spectral analysis, the critical case}, J. Funct. Anal. \textbf {260} (2011), 1766-1794.

\bibitem[SZ]{SZ} J. Sj\"ostrand and M. Zworski, \emph{Elementary
linear algebra for  advanced spectral problems},
  Ann. Inst. Fourier (Grenoble) \textbf{57}  no. 7 (2007), 2095--2141.



\bibitem[Ta1]{Ta1}
  M. Taylor,  \emph{Partial Differential Equations, Basic Theory}, Springer, New York, 1996; corrected second printing 1999.

\bibitem[Ta2]{Ta2}  M. Taylor,  \emph{Partial Differential Equations II, Qualitative Studies of Linear Equations}, Springer, New York, 1996;
corrected second printing 1997.

\bibitem[Va1]{Va1} A. Vasy,  \emph{Propagation of singularities in three-body scattering}. Astérisque No. \textbf{ 262 } (2000), vi+151 pp.

\bibitem[Va2]{Va2} A. Vasy,  \emph{Propagation of singularities in many-body scattering in the presence of bound states}. J. Funct. Anal. \textbf{184}(1) (2001), 177--272.

\bibitem[VW]{VW} A. Vasy and X.P. Wang, \emph{Smoothness and high
    energy asymptotics of the spectral shift function in many-body
    scattering},   Comm. Partial Differential Equations  \textbf{27}   (2002),  no. 11-12, 2139--2186.

\bibitem[Wa1]{Wa1} X.P. Wang, \emph{Microlocal resolvent estimates for $N$-body Schr\"odinger operators}, J. Fac. Sci. Univ. Tokyo, Sec. IA Math. \textbf{40} (1993), 337-385.
\bibitem[Wa2] {Wa2} X.P. Wang,  \emph{Asymptotic behavior of the resolvent of $N$-body Schr\"odinger
operators near a threshold},  Ann. Henri Poincar\'e  \textbf{4} (2003), 553--600.
\bibitem[Wa3]{Wa3} X.P. Wang, {\em Total Cross Sections in $N$-body Problems:
Finiteness and High Energy Asymptotics.} Comm. Math. Phys. \textbf{156} (1993),
 333--354.
\bibitem[Wa4]{Wa4} X.P. Wang, \emph{Threshold energy resonance in geometric scattering},
Proceedings of Symposium ``Scattering and Spectral Theory'',
August 2003, Recife, Brazil. Matem\'atica Contempor\^anea \textbf{26}
(2004), 135--164.
\bibitem[Wa5]{Wa5} X.P. Wang, \emph{Asymptotic expansion in time of the
  Schr\"odinger group on conical manifolds}, Ann. Inst. Fourier
  (Grenoble) \textbf{56} no. 6 (2006), 1903--1945.
\bibitem[Wa6]{Wa6} X.P. Wang,  \emph{Gevrey estimates of the resolvent and sub-exponential time-decay for the heat and Schrödinger semigroups}, J. Math. Pures Appl. \text{135} no. 9 (2020), 284-338

\bibitem[Ya1]{Ya1} D.R. Yafaev, \emph{The low energy scattering for slowly
    decreasing potentials}, Comm. Math. Phys. \textbf{85} (1982), 177--196.
\bibitem[Ya2]{Ya2} D.R. Yafaev, \emph{Spectral properties of the
    {S}chr\"odinger operator with a potential having a slow falloff},
Funct. Anal. and its Appl. (Steklov Math. Inst.) \textbf{16}  (1982),
47--54. (English translation, 1983.)

\bibitem[Ya3]{Ya3} D.R. Yafaev, \emph{Scattering theory: some old and new problems},
Lecture Notes in Mathematics, \textbf{1735}. Springer-Verlag, Berlin,
2000. xvi+169 pp.

\bibitem[Ya4]{Ya4} D.R.Yafaev, {\em The eikonal approximation and the
asymptotics of the total scattering cross-section for the
Schr\"odinger equation.} Ann. Inst. Henri Poincar\'e,  \textbf{40} (4)  (1986),
 397--425.

\bibitem[Ya5]{Ya5} D.R. Yafaev, \emph{Resolvent estimates and scattering
    matrix for $N$-body Hamiltonians},
  Intgr. Equat. Oper. Th. \textbf{21} (1995), 93--126.


\end{thebibliography}
\end{document}